\documentclass{article}

\usepackage[hyphens]{url} 
\usepackage{graphicx}
\usepackage[authoryear]{natbib} 
\usepackage{caption}
\usepackage{geometry}
\usepackage{algorithm}
\usepackage[T1]{fontenc}
\usepackage{wrapfig}

\title{City Sampling for Citizens' Assemblies}
\author{Paul Gölz$^1$, Jan Maly$^{2,3}$, Ulrike Schmidt-Kraepelin$^4$,\\Markus Utke$^4$, Philipp C. Verpoort$^5$}

\date{\small$^1$Cornell University, $^2$WU Vienna University of Economics and Business, $^3$TU Wien,\\ $^4$TU Eindhoven, $^5$Sortition Foundation}

\newcommand{\E}[1]{\mathbb{E}\left[#1\right]}
\usepackage{amsmath,amsfonts}
\usepackage{enumitem}
\usepackage{xfrac}

\usepackage{bm}
\usepackage{xspace}
\usepackage{graphicx}
\usepackage{subcaption}
\usepackage{dblfloatfix}
\usepackage{tikz}
\usetikzlibrary{arrows.meta}
\usetikzlibrary{decorations.pathreplacing}
\usepackage{booktabs}
\usepackage{makecell}

\usepackage{svg}

\newcommand{\draft}{false}

\newif\ifcomments 
\commentsfalse

\allowdisplaybreaks

\usepackage{xcolor}
\usepackage{amsthm,amssymb}
\usepackage{dsfont}
\usepackage{mathtools}
\usepackage{thmtools} 
\usepackage{thm-restate}
\usepackage{lscape}

\usepackage[capitalize]{cleveref}

\usepackage[noend]{algpseudocode}

\newtheorem{example}{Example}
\newtheorem*{claim*}{Claim}
\newtheorem{assumption}{Assumption}

\Crefname{lemma}{Lemma}{Lemmas}
\crefname{enumi}{property}{properties}

\newenvironment{claimproof}[1][Proof of the claim]{%
  \begin{proof}[#1]%
}{%
  \end{proof}
}

\newcommand{\StackRects}[6][]{%
    \pgfmathsetmacro{\y}{#5}%
  \foreach \h [remember=\y as \y (initially #5)] in {#3} {%
    \draw[draw=black,thick,fill=#4,#1] (#6,\y) rectangle (#6 + #2,{\y+\h});
    \pgfmathsetmacro{\y}{\y + \h}%
  }%
}

\tikzset{myfatarrow/.style={draw=black!50,line width=.5mm,-Stealth,shorten >=12pt,shorten <=4pt}}

\newcommand{\decision}{\textsc{FeasibleCities}\xspace}
\newcommand{\optimization}{\textsc{MinFeasibleCities}\xspace}

\newcommand{\greq}{\textsc{GreedyEqual}\xspace}
\newcommand{\buckets}{\textsc{Buckets}\xspace}
\newcommand{\colgen}{\textsc{ColumnGeneration}\xspace}

\usepackage{colortbl}
\usepackage{xcolor}

\usepackage{array}

\usepackage{todonotes}

\definecolor{lightgreen}{RGB}{148, 210, 189}
\definecolor{lightyellow}{RGB}{238, 155, 0}

\usepackage{placeins}

\begin{document}

\maketitle

\begin{abstract}
In \emph{citizens' assemblies}, a group of constituents is randomly selected to weigh in on policy issues.
We study a two-stage sampling problem faced by practitioners in countries such as Germany, in which constituents' contact information is stored at a municipal level.
As a result, practitioners can only select constituents from a bounded number of cities ex post, while ensuring equal selection probability for constituents ex ante. 

We develop several algorithms for this problem.
Although minimizing the number of contacted cities is NP-hard, we provide a pseudo-polynomial time algorithm and an additive~$1$-approximation, both based on separation oracles for a linear programming formulation. Recognizing that practical objectives go beyond minimizing city count, we further introduce a simple and more interpretable greedy algorithm, which additionally satisfies an ex-post monotonicity property and achieves an additive~$2$-approximation. Finally, we explore a notion of ex-post proportionality, for which we propose two practical algorithms: an optimal algorithm based on column generation and integer linear programming and a simple heuristic creating particularly transparent distributions. We evaluate these algorithms on data from Germany, and plan to deploy them in cooperation with a leading nonprofit organization in this space.

\end{abstract}

\section{Introduction}

\emph{Citizens' assemblies} are an emerging form of democratic participation, in which a random sample of constituents formulate policy recommendations.
The random selection of assembly members, called \emph{sortition}, gives each person an equal chance to participate and ensures that the assembly forms a cross section of the population.
Citizens' assemblies have been increasing in frequency~\citep{OECD20}.
National-level examples include assemblies on same-sex marriage, abortion, and gender equality in Ireland~\citep{Courant21} 
and German assemblies on the country's global role~\citep{BurgerratDeutschlandsRolleinderWelt21}, nutrition~\citep{DeutscherBundestag24}, and disinformation~\citep{BertelsmannStiftung24}.

In practice, the sortition proceeds in two stages: first, a large number of random constituents are invited by mail; second, the members of the assembly are selected among those invited who volunteered to participate.
Most algorithmic work on citizens' assemblies focuses on the second stage~\citep{FGG+20a,FGG+21,FLP+24,FKP21,BF24}.

This work, instead, studies a practical problem arising in the first sampling stage in certain countries.
Sampling constituents with equal probability is straight-forward in countries with a central population register such as the Nordic countries~\citep{SMB+17}.
The sampling process is also simple in countries like the UK and US where no register exists and assembly organizers use postal lists to invite random households, though these lists under-represent ``rural areas, \dots, Hispanic households, non-English-speaking households'' among others~\citep{KKS14}.

The first sampling stage is more complex in countries such as Germany and Italy, where population registers are kept by municipalities.
Since these municipalities must be individually petitioned for sampling access in a burdensome process~\citep{SSG+23}, statistical surveys first sample a set of municipalities and then sample participants only from these municipalities' registers~\citep{WBW+17,inapp2022ess}.

Our project was sparked by discussions with German sortition practitioners, who have been following a similar two-level sampling approach~\citep{StabsstelleBuergerraete23}.
Using numbers from the assembly on nutrition for illustration, they were looking for a sampling process that would (1)~send out 20,000 invitation letters, (2)~not send letters to more than 80 distinct municipalities at once, and (3)~give each German resident an equal chance of being invited.\footnote{In fact, assembly organizers break down the sampling into 42 sampling processes of this form, one for each federal state and category of municipality size. For exposition, we focus on an individual such problem, and consider the national level in \Cref{sec:federal}.} 

\begin{wrapfigure}{r}{0.4\textwidth}
    \centering
    \includegraphics[width=0.4\textwidth]{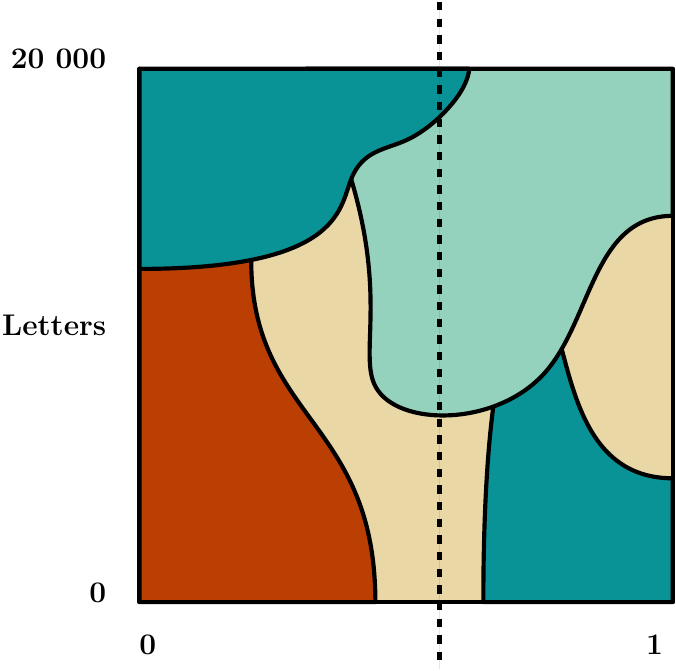}
    \caption{Graphical representation of sampling process.}
    \label{fig:illustratediagram}
\end{wrapfigure}

The output of any sampling process, i.e., any probability distribution determining how many invitations to send to each municipality, can be represented in a graphical form, which we illustrate in \Cref{fig:illustratediagram}.
To sample from this distribution, one draws a number $\rho \in [0, 1)$ uniformly at random, and considers the vertical line at this position (dashed in the figure).
This line intersects the shapes in the diagram, each of which is labeled with a municipality, and the number of letters sent to a municipality is equal to the total height of the municipality's shapes at the vertical line.\footnote{Clearly, the x-axis ordering of the diagram is arbitrary. All that we need is that each color's union of shapes is measurable.}
Without loss of generality, the selection within each municipality is uniform without replacement.
In this representation, the practitioners' requirements are easy to express: (1)~the total height of the figure at each vertical stripe should be 20,000 letters, (2)~no vertical stripe should intersect with more than 80 shapes, and (3)~the total area of a municipality's shapes (i.e., its expected number of received letters) must be proportional to its population.

A final requirement is that (4)~the number of letters received by each municipality (or, the height of the municipality's shapes in any vertical strip) has an upper bound.
Indeed, the municipality's population\,---\,which can be as low as $9$ inhabitants in the case of Germany\,---\,is definitely an upper bound, and many municipalities are moreover reluctant to allow sampling of more than about 10\% of their population due to privacy concerns.
In survey sampling, such upper constraints are not present because it is possible to upweight a resident in the analysis, effectively sampling them more than once.
As a result, the solution used in survey sampling\,---\,sampling municipalities with \emph{probability proportional to size}~\citep{BH83}, so that each vertical stripe consists of 80 equal-height layers\,---\,does not apply to assembly selection.

Whereas practitioners have so far relaxed conditions~(1) and~(2)~\citep{StabsstelleBuergerraete23} due to limitations in available methods, we show that all desiderata can, in fact, be satisfied by moving beyond rectangular shapes to more flexible geometric constructions.

\paragraph{Our Results and Techniques}
We begin by formulating our task as an optimization problem, \optimization, which seeks a probability distribution satisfying the four conditions while minimizing the number of contacted municipalities. Although \optimization is NP-hard, we provide a pseudo-polynomial time algorithm and an additive~$1$-approximation, both based on separation oracles for a linear programming formulation.

Since minimizing municipalities is only one of several practical goals, we introduce additional criteria.
We first propose \emph{ex-post monotonicity}, which states that, among the contacted municipalities, larger ones should receive at least as many letters as smaller ones.
We present \greq, a natural algorithm that achieves ex-post monotonicity and an additive~$2$-approximation under mild assumptions.

Whereas \greq promotes balanced letter allocations, it is natural to strengthen monotonicity to \emph{ex-post proportionality}, which states that a municipality's number of letters received scales with its size. We capture different proportionality goals through \emph{target letter functions} and develop two algorithms to pursue them: an optimal method based on integer linear programming and a simpler heuristic.

Finally, we evaluate all algorithms on data from the German Citizens’ Assembly on Nutrition~\citep{StabsstelleBuergerraete23}. Since the selection is applied independently within 42 subgroups, we show how to lift the notion of target letters from the local to the global level. Our algorithms offer practical solutions that can accommodate a wide range of real-world requirements.

\paragraph{Related Work}
By contributing to the first stage of the assembly selection pipeline in practice, our work is complementary to, but technically independent from, algorithms for selecting the final assembly from those accepting the invitation.
\citet{FGG+21} developed an optimization-based algorithm for this task; subsequent work studied transparent ways of drawing from the algorithm's computed probability distribution~\citep{FKP21}, incentives for misrepresentation~\citep{FLP+24,BF24}, accounting for self-selection bias~\citep{FGG+20a}, and the replacement of assembly members who drop out later~\citep{ABF+25}.

Other works have studied sortition algorithms that directly draw the assembly from the population and resulting theoretical properties.
These works study the variance of representation of features in the assembly~\citep{BGP19}, the social welfare if assembly members participate in a sequence of binary majority votes~\citep{MST21}, axioms and approximation bounds on the proximity of assembly members to the population in a metric space~\citep{EKM+22,EM25,CMP24}, and a proposed hierarchy of interconnected assemblies~\citep{HPS+25}.
\citet{DAL+21} study an online selection problem motivated by citizens' assemblies, in a random-dial methodology resembling the recruitment process in France.

\section{The Theoretical Model}\label{sec:model}

We are given $n$ cities\footnote{For brevity, we use `cities' as a synonym for `municipalities'.} and a fixed number of letters $\ell \in \mathbb{N}$ to allocate. Each city has a population $\pi_i \in \mathbb{R}$ and we assume normalization wlog, i.e., $\sum_{i \in [n]} \pi_i = 1$. We also write $\vec{\pi} = (\pi_1, \dots, \pi_n)$. Every city has a maximum number of letters it can receive, denoted by $\vec{u} = (u_1, \dots, u_n) \in \mathbb{N}^n$. We assume $\pi_1 \leq \dots \leq \pi_n$, $u_1 \leq \dots \leq u_n \leq \ell$. A letter allocation is a vector $a \in \mathbb{R}_{\geq 0}^n$ with the property that $\sum_{i \in [n]} a_i = \ell$ and $0 \leq a_i \leq u_i$ for all $i \in [n]$.\footnote{For $k \in \mathbb{N}$ let $[k]=\{1,\dots, k\}$ and $[k]_0 = \{0,\dots,k\}$.} An allocation is \emph{$t$-bounded} if at most $t$ cities receive a non-zero number of letters; let $A_t$ denote the set of all such allocations.
Given an instance of our problem $(\vec{\pi},\vec{u},t)$, \decision describes the problem of deciding whether there exists a probability distribution $\mathcal{D}$ over $A_t$ such that  
\begin{equation}
\mathbb{E}_{a \sim \mathcal{D}}[a_i] = \pi_i \cdot \ell \text{ for all } i \in [n]. \label{prop:fairshare}    
\end{equation}
We also refer to a probability distribution respecting \cref{prop:fairshare} as \emph{ex-ante fair}. \optimization describes the corresponding optimization problem of finding the minimum $t$ such that the answer to \decision is yes.

Though letter allocations are integral in practice, i.e., $a \in \mathbb{N}^n$, this restriction is wlog for \decision since any distribution over fractional allocations for $t$ can be turned into a distribution over $t$-bounded integral allocations with the same ex-ante properties, through dependent randomized rounding~\citep{GKP+06}.
For convenience, we assume $A_t$ to be integral in \Cref{sec:optimization} and fractional in \Cref{sec:monotone}.
We assume $\pi_i \ell \leq u_i$ for all $i \in [n]$, which is a necessary condition for the existence of an ex-ante fair distribution (for any $t$) due to the upper bounds (see \Cref{lem:trivialLB}). 
Through the paper, we refer to the following running example:
\begin{example} \label{ex:running_example}
    Distribute $\ell = 60$ letters over $n = 8$ cities. The city sizes and upper bounds are $\vec{\pi}~=~\frac{1}{360} \cdot (10, 10, 40,40,40, 50,70,100)$ and $\vec{u} = 180 \cdot \vec{\pi} = (5, 5, 20, 20, 20, 25, 35, 50)$.
\end{example}

While \Cref{sec:optimization} studies city sampling through the lens of the optimization problem defined above, 
\Cref{sec:monotone,sec:proportional} motivate and define additional desirable concepts: \emph{ex-post monotonicity}, \emph{ex-post proportionality}, and \emph{binary outcomes}.

\section{The \optimization Problem} \label{sec:optimization}
In this section, we show that, though \decision{} is NP-hard, it is only barely a hard problem, in the sense that pseudopolynomial time computation, or a slack of a single city suffice to overcome this complexity barrier.
We defer all missing proofs to \cref{app:missingproofs}.

We start by showing a simple lower bound that will be helpful throughout the paper. To this end, we define $w_i = \frac{\pi_i \ell}{u_i}$ for all $i \in [n]$, which yields a lower bound on the \emph{selection probability} of a city (also interpreted as the minimum \emph{width} within our illustrations).

\begin{restatable}{lemma}{lemTrivialLB}
     \label{lem:trivialLB}
    For any instance $(\Vec{\pi},\Vec{u},t)$, and an ex-ante fair probability distribution $\mathcal{D}$ over $A_t$, it holds that 
    \begin{enumerate}[label=(\roman*)]
        \item $\Pr[a_i>0] \geq w_i$ for all $i \in [n]$, and  \label{lem:trivialLB-selProb}
        \item $t \geq \sum_{i \in [n]} w_i$. \label{lem:trivialLB-t}
    \end{enumerate}
\end{restatable}

For \Cref{ex:running_example}, \Cref{lem:trivialLB} shows that $t$ must be at least $3$ since the minimum total width of all cities is $\sum_{i \in [n]} w_i  = \frac{8}{3}$.

\smallskip

In the appendix, we show that \decision is NP-hard via a reduction from \textsc{Partition}. 
In a nutshell, this reduction constructs an instance of our problem, in which all allocations in the support of a $t$-bounded, ex-ante fair distribution must give half of the cities 0 letters and half their upper bound $u_i$.
The question whether any such allocation assigns exactly $\ell$ letters is exactly \textsc{Partition}. 

\begin{restatable}{theorem}{thmNPHard}
    \decision is NP-hard. \label{thm:NPHard}
\end{restatable}

Since we reduce from \textsc{Partition}, which admits a pseudo-polynomial time algorithm, it is natural to ask whether our problem does too. 
To show that this is the case, 
we formulate the problem as a linear program with one variable $x_a$ for each integral allocation $a \in A_t$. 
The LP searches for a fair distribution over these, with $x_a$ representing the probability assigned to allocation $a$. The first constraint ensures that the probabilities sum to at most 1; the second enforces fairness. Both hold with equality in any feasible solution, but are written as inequalities for clarity in the dual.

\begin{align}
\textbf{Primal:}&& \text{minimize}~ & 0 \notag \\
 && \text{subject to}~& \sum_{\mathclap{a \in A_t}} x_a\leq 1, \notag \\
                 &&& \sum_{\mathclap{a \in A_t}} x_a a_i \ge \pi_i \ell & \text{for } i & \in [n], \quad && \notag\\
                 && &x_a \geq 0 & \text{for } a & \in A_t. \notag \\[.5em]
\textbf{Dual:}&& \text{maximize} &\sum_{\mathclap{i \in [n]}} \pi_i \ell y_i - y \notag \\
 && \text{subject to}~& \sum_{i \in [n]} a_i y_i \le y & \text{ for } a & \in A_t, & \label{dualConstraint}\\
                 &&& y, y_i  \ge 0 & \text{for } i &\in [n]. \notag 
\end{align}

We aim to decide whether the primal LP is feasible, which is the case iff the dual LP admits no solution with positive objective value (which could be scaled to show that the dual value is unbounded). 
We add a constraint to the dual requiring a strictly positive objective value. 

Though the resulting system has exponentially many constraints, its feasibility can be decided with the ellipsoid method~\citep{GLS93a} provided we can implement a \emph{separation oracle} for the dual: given a vector $\left((y_i)_{i \in [n]}, y\right)$, we must decide whether it is feasible for the modified dual or return a violated constraint. We show that this separation problem can be solved in pseudo-polynomial time using a knapsack-style dynamic program.

\begin{restatable}{theorem}{thmPseudoAlgorithm}
    There exists a pseudo-polynomial time algorithm for \decision. \label{thm:pseudo}
\end{restatable}

More surprisingly, we can construct a polynomial-time \emph{approximate separation oracle}, in the following, strong sense:
given a vector $\left((y_i)_{i \in [n]}, y\right)$, our oracle either determines that the vector satisfies all dual constraints or identify a violated constraint of type~\eqref{dualConstraint}, but for for some allocation $a \in A_{t+1} \supseteq A_t$ rather than in $A_t$.
As \citet{ScUh13a} show, the ellipsoid method with such an approximate oracle can determine either that the dual above is unbounded (so the primal is infeasible) or that the dual for $t+1$ cities is bounded (hence, the primal for $t+1$ cities is feasible).
By applying this algorithm to increasing values of $t$ until a feasible primal is found, we can find the lowest possible number of contacted cities, up to perhaps one additional city.

\begin{restatable}{theorem}{thmOneapproximation}
   There exists a polynomial-time algorithm that is an additive {$1$-approximation} to \optimization. 
\end{restatable}

While the above algorithms are theoretically tractable, the ellipsoid method is a famously impractical algorithm.\footnote{Although lacking theoretical guarantees, combining our (or similar) separation oracles with the simplex method can still lead to practical algorithms (see \colgen in \Cref{sec:proportional}).
}
Moreover, these algorithm may yield highly unintuitive allocations that would be difficult to justify in practice. For example, a large city might receive significantly fewer letters than a smaller one ex post, or only small cities might be selected while all larger ones are excluded. We now shift our focus from mere feasibility to fair distributions that uphold additional desirable properties, all while keeping $t$ low.

\section{A Simple and Monotone Approximation} \label{sec:monotone}
In this section, we aim for \emph{ex-post monotonicity}.
A fractional allocation $a$ is called \emph{monotone} if $a_i \geq a_j$ whenever $i > j$. A probability distribution is ex-post monotone if its support consists of monotone allocations.
We present a simple $2$-approximation for \optimization that yields ex-post monotone distributions under mild assumptions. The algorithm is inspired by $\pi ps$ sampling~\citep{BH83}: given an instance $(\vec{\pi}, \vec{u}, t)$, one samples $t$ cities with probabilities proportional to $\vec{\pi}$ without replacement and assigns $\ell/t$ letters to each. While this would violate cities' upper bounds, our algorithm can be viewed as a minimal adjustment to $\pi ps$ sampling to ensure feasibility.

Our algorithm, \greq, is best understood through its geometric interpretation. The algorithm processes cities in increasing order of size and starts by attempting to place a $\pi_i \ell$-area rectangle of height $\ell/t$. If this violates the city's upper bound, it instead places a rectangle of height $u_i$. It then proceeds to place the next rectangle to the right. Once the first layer is filled, \greq moves to the next layer, now aiming to keep the height of rectangles at the remaining vertical space divided by $t-1$. This ensures that later (and thus larger) cities can receive at least as many letters as those already placed. We remark that, starting from the second layer, cities may receive a set of rectangles summing to $\pi_i \ell$ instead of a single rectangle, which is due to shifts in lower layers. See \Cref{fig:greqExampleStacked} for an illustration.

To formalize \greq, we introduce a second type of illustration, which is a flattened version of the illustration in \Cref{fig:greqExampleStacked}. This illustration is formalized by functions $\lambda_i$ for each $i \in [n]$ that are defined on the interval $[0,t)$. The value $\lambda_i(x)$ corresponds to the height of the rectangle that the algorithm draws for city $i$ in layer $\lfloor x\rfloor$ (0-indexed) and at position $x - \lfloor x\rfloor$ of the stacked picture. (Note that for any position $x \in [0,t)$ this value will be non-zero for exactly one city 
as the algorithm draws for one city at a time.)
We illustrate these functions in \Cref{fig:greqExampleFlat}. 

When the algorithm draws at position $x$ in the flat picture, it needs to know the height of all rectangles that were placed at some value $y\leq x-1$ with $y \equiv x \;(\text{mod} \;1)$. We define
\[\Lambda(x)  = \sum_{y \leq x, \;y \equiv x \;(\mathrm{mod\ } 1)} \sum_{j \leq i} \lambda_j(y).\]

Last, we define $\mu_i(x)$, describing the height of the rectangle to be drawn, given that we place city $i$ at position $x$, 
$$\mu_i(x) = 
\begin{cases} 
\min \left( u_i, \frac{\ell-\Lambda(x-1)}{t - \lfloor x \rfloor}\right) & \text{ for } x \in [0,t) \\ 
u_i  & \text{ for } x \geq t,   
\end{cases}$$
and are now ready to formalize \greq:
\begin{center}
\hrule\vspace{.2em}
\begin{algorithmic}
\Procedure{GreedyEqual}{$\Vec{\pi},\Vec{u},t$}
  \State $x \gets 0, i \leftarrow 1$
  \While{$i \leq n$}
  \State let $y \geq x$ such that $\int_{x}^y \mu_i(z) dz = \pi_i \ell$
  \State $\lambda_i(z) \leftarrow \mu_i(z) $ for $z \in [x,y)$, $x \leftarrow y$, $i\leftarrow i+1$
  \EndWhile
  \If{$x=t$} \Return $(\lambda_i)_{i \in [n]}$ \textbf{else } ``fail'' \EndIf
\EndProcedure
\end{algorithmic}\vspace{.2em}
\hrule
\end{center}

\begin{figure}[t]
\centering
    \begin{subfigure}{.6\linewidth}
    \centering
        \includegraphics[width=0.7\linewidth]{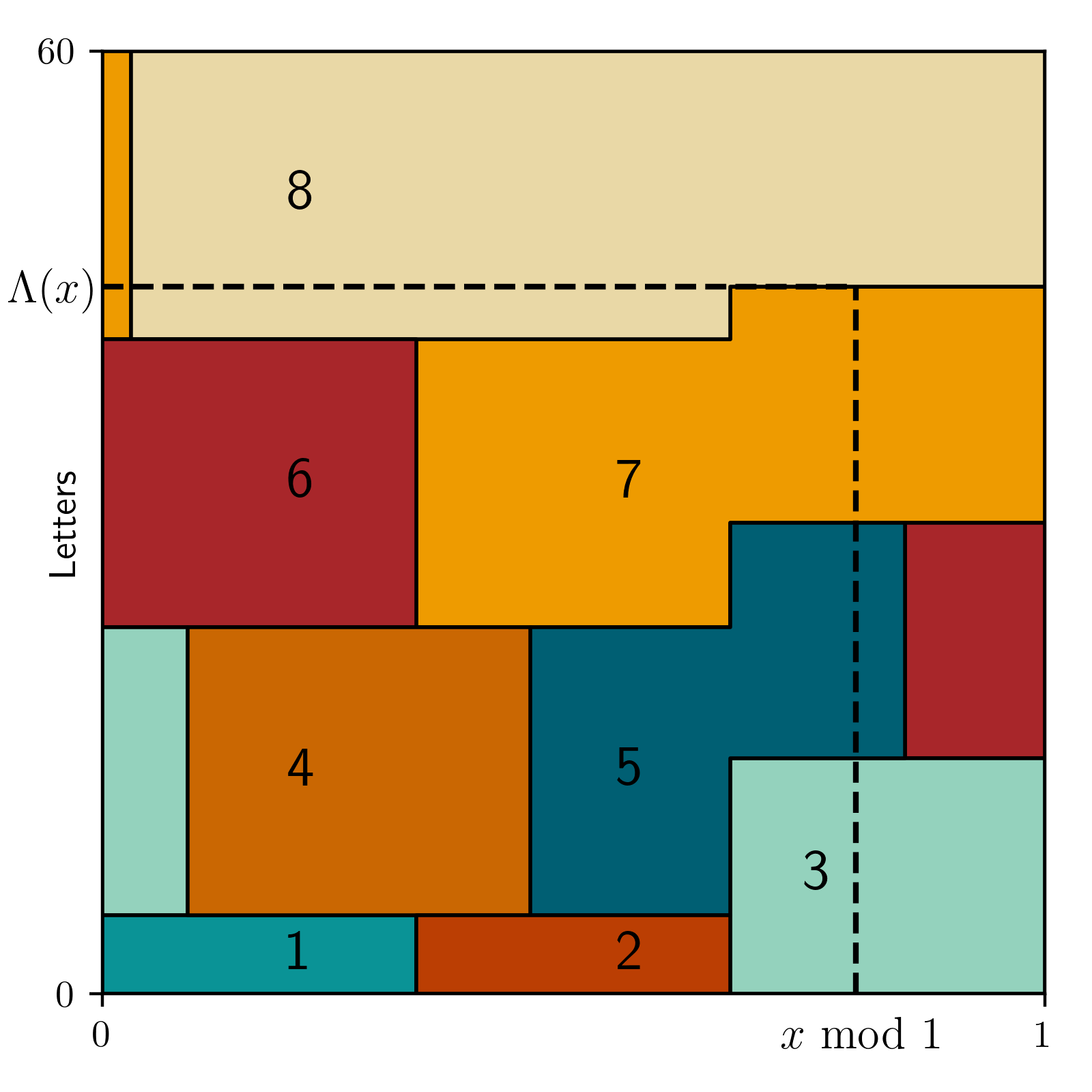}
        \caption{Illustration \greq for \Cref{ex:running_example} for $t=4$.}
        \label{fig:greqExampleStacked}
    \end{subfigure}
    \begin{subfigure}{.6\linewidth}
        \includegraphics[width=1\linewidth]{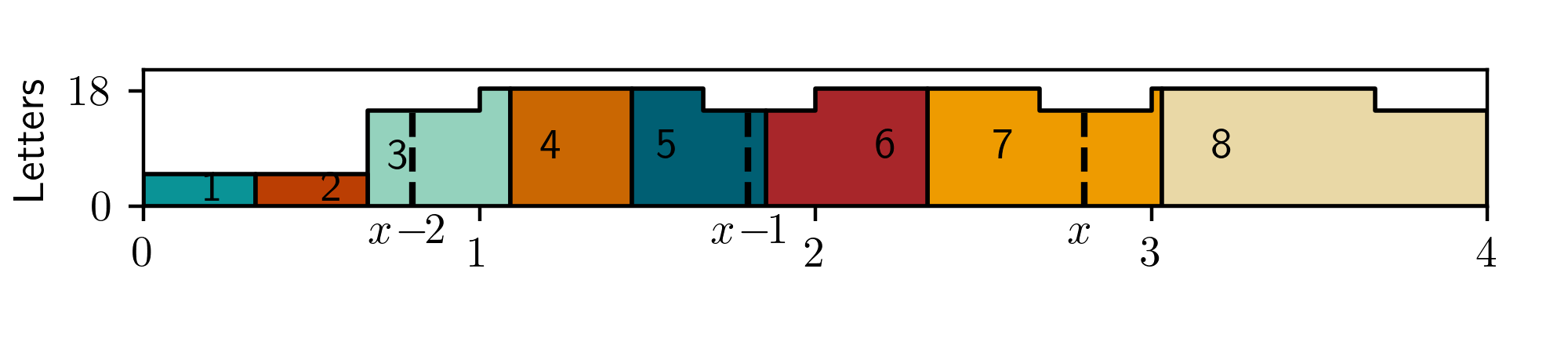}
    \caption{Illustration of the functions $\lambda_i$ for all $i\in [n]$, where each $\lambda_i$ is indicated by a different color.} \label{fig:greqExampleFlat}
    \end{subfigure}
    \caption{\greq applied to \Cref{ex:running_example}.}
    \label{fig:greqExample}
\end{figure}
\footnotetext{For ex-post monotonicity, the relaxation to fractional allocations is not quite wlog, but any fractional monotone allocation can be decomposed into a distribution over integral allocations that are monotone up to one letter.}\greq\ can fail in two ways. First, it may terminate with $x > t$, meaning it requires more than $t$ layers and thus does not yield a $t$-bounded distribution. In \Cref{thm:greq2approx}, we bound the optimum of \optimization in this case. Second, and more subtly, the area assigned to a city may be so wide that it overlaps across layers, leading to an allocation that exceeds the upper bound of the city. This can only happen when a city is \emph{oversized}, i.e., when $\pi_i > 1/t$. While such cities appear in parts of our data, they always have upper bounds well above $\ell$, making this a non-issue in practice. We formalize the following assumption:

\begin{assumption}
    For any oversized city $i$, $u_i \geq \ell$. \label{ass:oversized}
\end{assumption}

In our dataset, \Cref{ass:oversized} is satisfied as long as $t\leq 420$, far above the past choice of $t = 80$. 

\begin{restatable}{theorem}{thmGreqCorrect}\label{thm:GreqCorrectness}
    Under \Cref{ass:oversized}, \greq always returns an ex-ante fair and $t$-bounded probability distribution (if it succeeds).
\end{restatable}

In instances without oversized cities we furthermore guarantee monotonicity. In our data we do not observe any monotonicity violation, even for oversized cities. 

\begin{restatable}{theorem}{thmGreqMonotone} \label{thm:GreqMonotone}
    For instances without oversized cities, \greq is ex-post monotone. 
\end{restatable}

Before proving the additive $2$-approximation, we provide insight into the structure of \greq's solutions. We say that \greq\ \emph{selects the average} at position $x \in [0, t)$ if $\lambda_i(x) = \frac{\ell - \Lambda(x - 1)}{t - \lfloor x \rfloor}$ (rather than $\lambda_i(x) = u_i$), where $i$ is the unique city with $\lambda_i(x) > 0$.

\begin{restatable}{lemma}{LemGreqStructure}
Independent of whether \greq succeeds, the following holds:
\begin{enumerate}[label=(\roman*)]
    \item If \greq selects the average at $x \in [0,t)$, then it selects the average at all $y \in [x, \lceil x\rceil)$ as well as all $y \in [x,t)$ with $y \equiv x \;(\mathrm{mod\ } 1)$. \label{lem:StructureA}
    \item The function $\Lambda(x)$ is non-decreasing on $[0,t)$.\label{lem:StructureB}
\end{enumerate}
\label{lem:greqStructure}
\end{restatable}

\begin{restatable}{theorem}{thmGreqTwoApprox}
    Under \Cref{ass:oversized}, \greq is an additive $2$-approximation for \optimization. \label{thm:greq2approx}
\end{restatable}
\begin{proof}[Proof (first part)]
    Let $(\vec{\pi},\vec{u},t)$ be an instance of our problem such that \greq fails. 
    We show that $\sum_{i \in [n]} w_i > t-2$, which by \Cref{lem:trivialLB} \ref{lem:trivialLB-t}, implies that the optimal budget for \optimization is at least $t-1$. 

    To gain intuition for the proof, consider the following thought experiment: imagine scaling each city's shapes so that it maintains its total area but reaches its maximum height $u_i$, attaining its minimum width $w_i$. How much width do we lose in total? The original sum of widths exceeds~$t$; we show that even after scaling, the total width remains strictly greater than $t - 2$. While it may seem natural to scale each city individually, our analysis instead partitions the stacked picture into ``columns'' and scales each column separately.

    Let $\mathcal{I}$ be a partition of the interval $[0,1)$ with the property that all functions $\lambda_i$ are constant along each interval $I \in \mathcal{I}$.
    Now, consider the interval in $\mathcal{I}$ that starts at $0$. For all $k \in [t-1]_0$ let $j(k)$ be the unique city with  $\lambda_{j(k)}(k) >0$. From now, we drop the position and write $\lambda_{j(k)}$ instead of $\lambda_{j(k)}(k)$. Since \greq failed, we know that $\Lambda(x)<\ell$ for some $x \in [t-1,t)$. Moreover, by the monotonicity of $\Lambda$ (\Cref{lem:greqStructure} \ref{lem:StructureB}) we know that $\Lambda(t-1) < \ell$. By \Cref{lem:greqStructure} \ref{lem:StructureA} it holds that $\lambda_{j(k)} = u_{j(k)}$ for all $k \in [t-1]_0$.
    
    Now consider an arbitrary interval $[\alpha, \beta) \in \mathcal{I}$. For each $k \in [t{-}1]_0$, let $i(k)$ be the unique city with $\lambda_{i(k)}(k + \alpha) > 0$. We write $\lambda_{i(k)}$ for $\lambda_{i(k)}(k + \alpha)$. See \Cref{fig:greq2approx} for an illustration. The original total width in column $[\alpha, \beta)$ is $(\beta - \alpha)t$. We show that scaling each subarea to its maximum height yields a total width greater than $(\beta - \alpha)(t - 2)$. Since the factor $(\beta - \alpha)$ is irrelevant to our argument, we drop it.

\begin{figure}[t]
    \centering
    
    \begin{tikzpicture}[scale=0.17]
        \draw[very thick] (0,0) rectangle (28,32);
        \draw[very thick] (0,0) rectangle (4,26); 
    
        \StackRects{4}{0.5,2,3,3,7,7,7}{lightgreen}{0}{0}
        \StackRects{4}{2,3,3}{lightgreen}{0}{12}
        \StackRects{4}{6,6,6,6}{lightyellow!70}{8}{12}
    
        \node(j1) at (2,1.5) {\small $j(1)$};
        \node(j2) at (2,4) {\small $j(2)$};
        \node(j3) at (2,7) {\small $j(3)$};
        \node(j4) at (2,12) {\small $j(4)$};
        \node(j5) at (2,19) {\small $j(5)$};
        \node(j6) at (2,26) {\small $j(6)$};
    
        \node(i1) at (14,1) {\small $i(0)$};
        \node(i2) at (14,3.5) {\small $i(1)$};
        \node(i3) at (14,6.5) {\small $i(2)$};
        \node(i4) at (14,11) {\small $i(3)$};
        \node(i5) at (14,17) {\small $i(4)$};
        \node(i6) at (14,23) {\small $i(5)$};
    
        \draw [decorate, very thick, decoration={brace,mirror,amplitude=10pt}] (21,8) -- (21,32);
        \node at (24.5,20) {$\ell_i$}; 
    
        \draw [decorate, very thick, decoration={brace,amplitude=8pt}] (-6,8.5) -- (-6,29.5);
        \node at (-9,19) {$\ell_j$}; 
    
        \draw [decorate, very thick, decoration={brace,mirror,amplitude=5pt}] (16.5,20) -- (16.5,26);
        \node at (20,23) {$\lambda_{i(5)}$}; 
    
        \draw [decorate, very thick, decoration={brace,amplitude=5pt}] (-.5,22.5) -- (-.5,29.5);
        \node at (-4,26) {$u_{j(6)}$}; 
        
        \draw [decorate, thick, decoration={brace,amplitude=5pt}] (-0.5,0) -- (-0.5,8.5);
        \node at (-4,4.25) {\footnotesize$\Lambda(t')$};
    
        \draw [decorate, thick, decoration={brace,mirror,amplitude=5pt}] (16.5,0) -- (16.5,8);
        \node at (22.8,4) {\footnotesize$\Lambda(t'\!-\!1\!+\!\alpha)$};

        \foreach \a/\b in {j1/i1,j2/i2,j3/i3,j4/i4,j5/i5,j6/i6}{
          \draw[myfatarrow] (\a) -- (\b);
        }

        \draw[very thick] (12,0) rectangle (16,32); 
    \end{tikzpicture}
    \caption{Situation in the proof of \Cref{thm:greq2approx}. In teal areas, cities receive their upper bounds $u_i$ and in orange areas they receive less than $u_i$. An arc indicates $u_{j(k)} \geq u_{i(k-1)}$.}\label{fig:greq2approx}
    \end{figure}
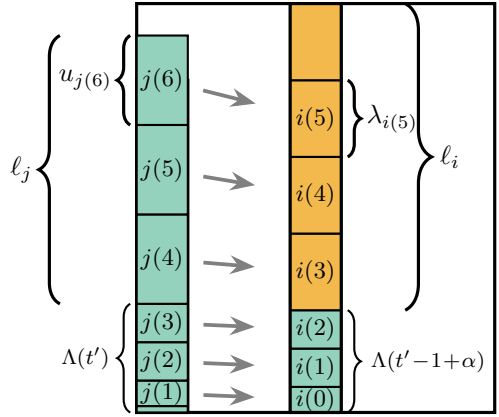

    \begin{claim*}
        It holds that $\sum_{k = 0}^{t-2} \frac{\lambda_{i(k)}}{u_{i(k)}} > t-2$.
    \end{claim*}

    \begin{claimproof}
     Since \greq processes cities with increasing indices, it holds that $i(k) \leq j(k+1)$ for all $k \in [t-2]_0$. Thus: 
     \begin{equation}
         \lambda_{i(k)} \leq u_{i(k)} \leq u_{j(k+1)} = \lambda_{j(k+1)}. \label{eq:relation_two_columns}
     \end{equation}
    Let $t'$ be the first index for which $\lambda_{i(t')} = \frac{\ell - \Lambda(t'-1 + \alpha)}{t-t'}$. If no such index exists, then we know that $\lambda_{i(k)} = u_{i(k)}$ for all $k \in [t-2]_0$ and the claim follows trivially. In the example in \Cref{fig:greq2approx} it holds that $t'=3$. 
    We define $\ell_i =\sum_{k=t'}^{t-1} \lambda_{i(k)}$ and $\ell_j = \sum_{k=t'+1}^{t-1} u_{j(k)}$. Note that $\ell_i > \ell_j$, since $\Lambda(t'-1+\alpha) < \Lambda(t')$. 
    By \Cref{lem:greqStructure} \ref{lem:StructureA}:
    \begin{equation}
      \lambda_{i(k)} = \frac{\ell_i}{t-t'} \text{ for all } k \in \{t',\dots,t-1\} \label{eq:equalLambdas}
    \end{equation}
    We are now ready to prove the claim
     \begin{align*}
         \sum_{k=0}^{t-2}\frac{\lambda_{i(k)}}{u_{i(k)}} & = t' +  \sum_{k=t'}^{t-2}\frac{\lambda_{i(k)}}{u_{i(k)}} \\
         & \stackrel{(\ref{eq:equalLambdas})}{=} t' + \frac{\ell_i}{(t-t')} \sum_{k=t'}^{t-2}\frac{1}{u_{i(k)}} \\ 
         & \stackrel{(\star)}{\geq} t' + \frac{\ell_i}{(t-t')}\frac{(t-t'-1)^2}{\sum_{k=t'}^{t-2}u_{i(k)}} \\ 
         & \stackrel{(\ref{eq:relation_two_columns})}{\geq} t' + \frac{\ell_i}{(t-t')}\frac{(t-t'-1)^2}{\sum_{k=t'+1}^{t-1}u_{j(k)}} \\ 
         & > t' + \frac{\ell_j}{t-t'}\frac{(t-t'-1)^2}{\ell_j} \\
         & = t' + t-t'-2 + \frac{1}{t-t'} > t-2, 
     \end{align*}
    where ($\star$) follows from the fact that the arithmetic mean is at least the harmonic mean (applied to the values $\frac{1}{u_i}$). 
    \end{claimproof}

    It remains to apply the above claim to all columns and conclude $\sum_{i \in [n]} w_i > t-2$. We refer to the appendix. 
\end{proof}

We also show that our upper bound for the approximation guarantee of \greq is tight:

\begin{restatable}{theorem}{thmGreqNotOneApprox} \label{prop:greqNotOneapprox}
    Even under \Cref{ass:oversized}, \greq is not an additive $1$-approximation for \optimization.
\end{restatable}

\section{Ex-post Proportionality} \label{sec:proportional}
\emph{Ex-post monotonicity} ensures that after randomization, larger selected cities receive at least as many letters as smaller ones. However, it does not guarantee that larger cities receive \emph{more} letters. For example, a city with millions of inhabitants might still receive the same number of letters as one with only tens of thousands\,---\,behavior that \greq in fact encourages. We explore how both the selection probability and the number of letters a city receives (if selected) can grow with the population. To achieve this, we introduce the more general concept of \emph{target letters}.

We assign each city $i$ a number $\tau_i$ of target letters, that it should receive if it is selected. This, in term also implies a target selection probability $\frac{\pi_i \ell}{\tau_i}$ for that city. If we let the target letters grow proportionally to the population, then each city gets selected with the same probability. If, on the other hand, the target is equal for all cities, then the target selection  probability grows proportionally to the population, which is close to what \greq achieves. It seems natural to allow for target functions in between those two extremes. 

We define a \textit{target letter function} to be a monotone function $f$ taking as input a population size $\pi_i$ and outputting a target in $\mathbb{R}_{\geq 0}$. However, blindly setting targets without knowing the budget $t$ can lead to infeasibility: For example, small targets will clearly be missed if $t$ is very small. To mitigate this issue, we introduce a scaling factor $\kappa$ and define for each city $i$ the scaled target letters as 
\begin{equation}
    \tau_i^\kappa = \max\left(\pi_i \ell, \min\left(u_i, \kappa f\left(\pi_i\right)\right)\right), \label{eq:targets}
\end{equation}
which makes sure that target letters do not exceed $u_i$ and the target selection probability does not exceed $1$. 
We then determine the value $\kappa$ such that the total target selection probability (or width) satisfies $\sum_{i \in [n]} \frac{\pi_i \ell}{\tau_i^{\kappa}} = t$ and set $\tau_i = \tau_i^{\kappa}$ for all $i \in [n]$.
A total width of at most $t$ is a necessary condition for the targets to be achievable but is far from sufficient due to the more complex structure of the problem. 

We suggest $f(x) = \sqrt{x}$ as a particularly natural target letter function since it allows target letters and target selection probability to scale in equal measure. We introduce two methods that take as input an instance and the target letters and aim to construct a fair distribution meeting the targets. As in \Cref{sec:monotone}, we allow for distributions over fractional allocations for the sake of simplicity, which immediately approximates integral ex-post proportionality up to one letter.

\paragraph{Column generation}
Recall our linear programming approach from \Cref{sec:optimization}, which we used for deciding whether an instance is feasible or not. It is natural to add an objective function to this LP to minimize, in expectation, a measure of deviation from the targets. 
Specifically, we minimize $\sum_{a \in A_t} x_a \varphi(a),$ where $\varphi$ measures the total relative deviation from the targets, i.e., $$\varphi(a) = \sum_{i \in [n], a_i > 0} \frac{|\tau_i - a_i|}{\tau_i}.$$
This objective penalizes the same absolute deviation from the target more heavily for smaller cities than for larger ones.
To optimize the resulting primal LP, we again design a separation oracle for the dual LP (a process also termed \emph{column generation}). This time, the separation problem is more complex, and we formulate a mixed integer linear program to solve it (\Cref{app:proportional}). Though not polynomial-time, state-of-the-art solvers scale to large problems in practice.

While \colgen is optimal with respect to the target letters, the resulting distributions have little visual structure (e.g., see \Cref{fig:column_generation_result}), and the algorithm's reliance on optimization solvers makes them hard to explain to the public. We introduce an alternative approach that is arguably more transparent, while still aiming to meet the target letters.

\paragraph{Bucket Approach}

The idea of \buckets is to partition the cities into $t$ disjoint sets (the \textit{buckets}), such that we can then sample exactly one city from each bucket. Each bucket has a \textit{height}, which determines how many letters the selected city from that bucket receives. Within each bucket, we thus need to sample proportional to size. By ex-ante fairness, the height of a bucket $B \subseteq [n]$ is determined by its elements $h = \sum_{i \in B} \pi_i \ell$. To approximate the target letters $\vec\tau$, we define the buckets such that the target letters of each city are close to the height of the bucket it belongs to.

To achieve this, we fill the buckets iteratively with cities in increasing order of their size. We move on to the next bucket if adding another city would either (i) increase the total target probability of all cities in the bucket above one, or (ii) would increase the height of the bucket above the maximum number of letters of its smallest city. See \Cref{app:proportional}.

The bucket approach has the advantage of producing easily explainable distributions (see, e.g., \Cref{fig:buckets_result}). In particular, it satisfies the \emph{binary outcome} property: each city knows in advance how many letters it will receive if selected. While the method does not guarantee ex-post monotonicity in the worst case, we observe no violations in our data. Moreover, it ensures that selected cities are distributed somewhat evenly across cities of different sizes.
On the downside, the approach lacks worst-case approximation guarantees.

\begin{restatable}{theorem}{thmBucketsNoApprox}
    For any targets and constant $c$, \buckets is not an additive $c$-approximation for \optimization.
\end{restatable}

\section{Towards Practice} \label{sec:federal}
We aim to apply one of our algorithms in future implementations of citizens' assemblies, particularly in Germany. To this end, we tested them on the data used for the assembly on nutrition \citep{StabsstelleBuergerraete23}, where $\ell = 20\text{K}$ letters were sent, the outreach budget was $t = 80$, and there are $n = 100\,755$ cities with a total population of $84$\,M. Following suggestions from practitioners, we define the maximum number of letters a city can receive as follows: $50\%$ of the population for cities under $500$ inhabitants, $10\%$ for those over $2\,500$, and $250$ for populations in between.

In this recent assembly, practitioners divided the country into 42 groups, based on the 16 federal states and on three city size classes ($[0,20\text{K})$, $[20\text{K},100\text{K})$, $[100\text{K},\infty)$),\footnote{Some states consist of only a single or two large cities.} and sampled the letter allocation per group. This stratification ensures sufficient numbers of invitations within each group for forming the assembly in the second stage of selection. Since the same grouping will likely be used for future assemblies, we test our algorithms in this setup.

\begin{figure*}
    \captionsetup{justification=centering}
    \centering
    \begin{subfigure}[b]{.3\linewidth}
         \includegraphics[width=\linewidth, trim={0 10pt 0 0}]{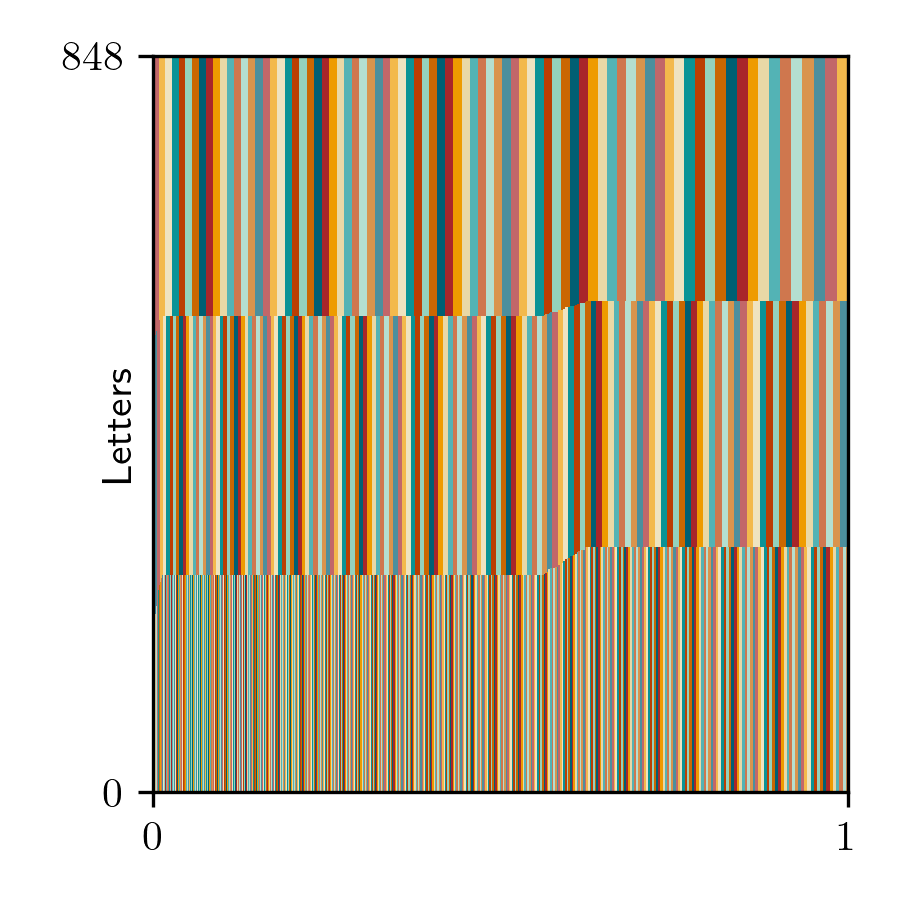}
     \caption{\greq \\($t_G=3$)}
         \label{fig:greq_result}
    \end{subfigure}\hfill
    \begin{subfigure}[b]{.3\linewidth}
        \includegraphics[width=\linewidth,trim={0 10pt 0 0}]{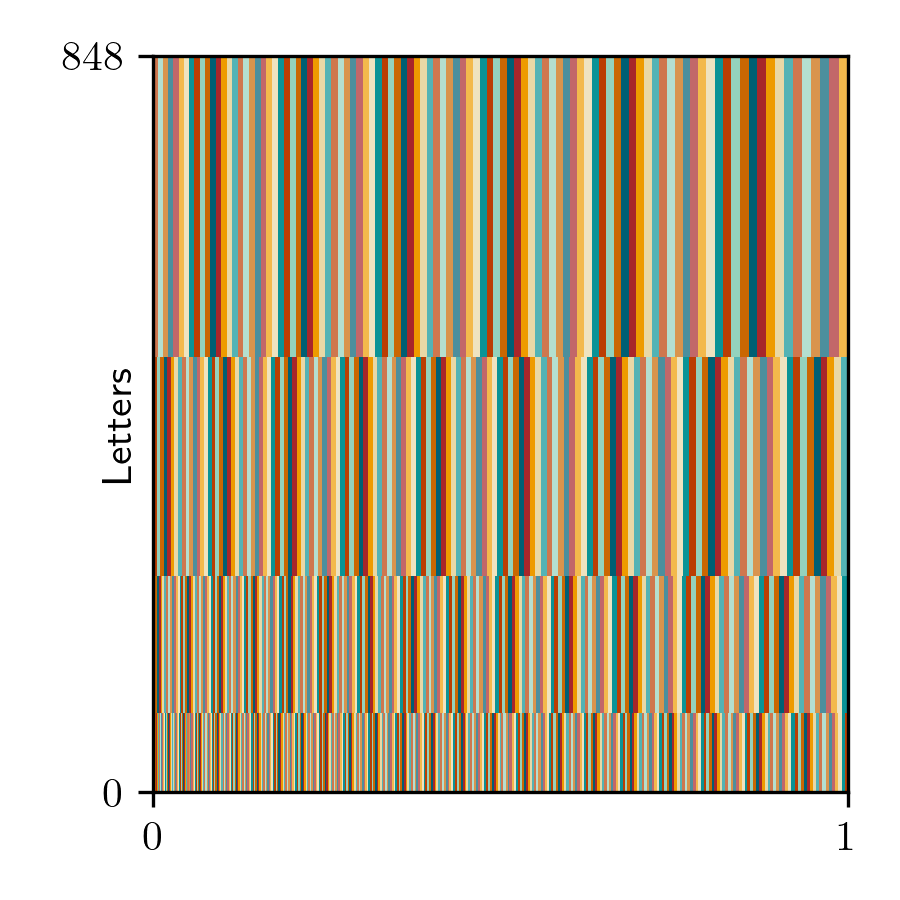}
        \caption{\buckets \\($t_G=4$)}
        \label{fig:buckets_result}
    \end{subfigure}\hfill
    \begin{subfigure}[b]{.3\linewidth}
        \includegraphics[width=\linewidth, trim={0 10pt 0 0}]{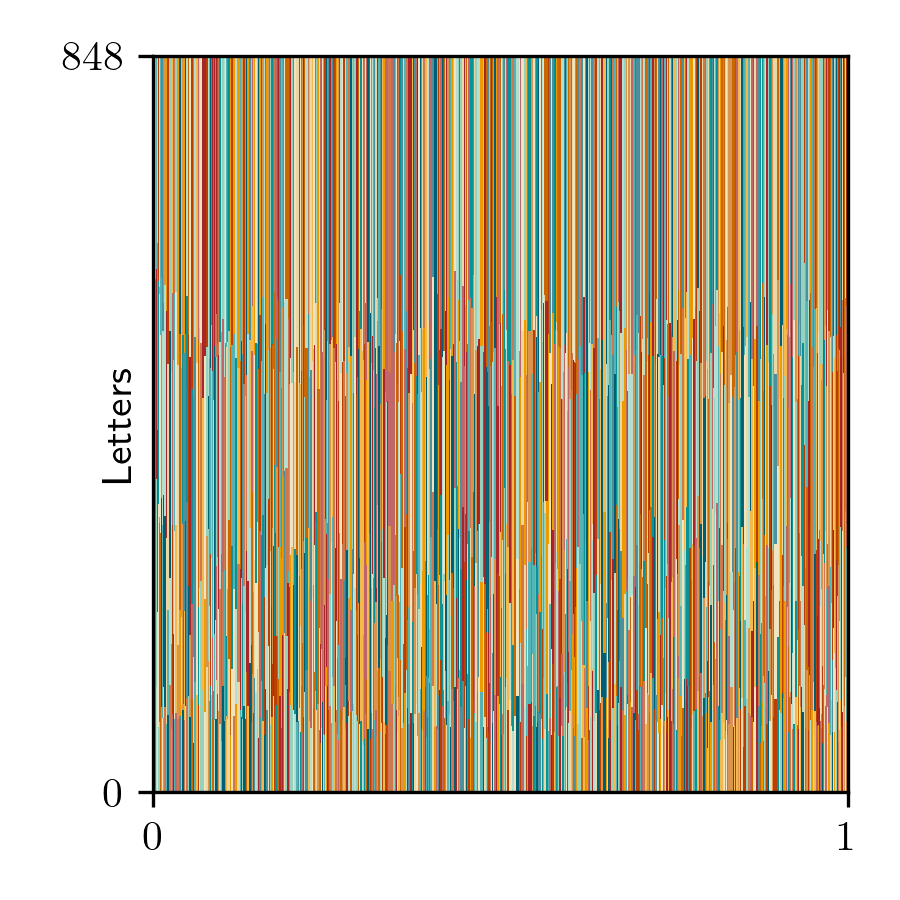}
        \caption{\colgen \hfill\\($t_G=4$)}
        \label{fig:column_generation_result}
    \end{subfigure}
    \begin{subfigure}[b]{0.48\linewidth}
        \includegraphics[width=\linewidth, trim={0 0 1cm 0}]{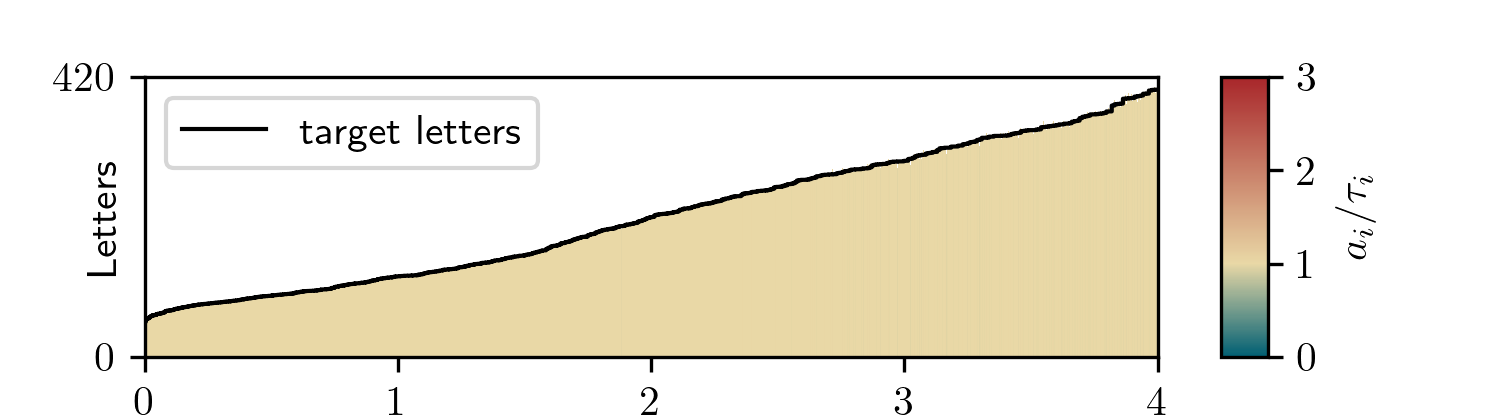}
        \caption{Proportionality of \colgen}
        \label{fig:alg_proportionality_column}    
    \end{subfigure}\hfill
    \begin{subfigure}[b]{0.48\linewidth}
        \includegraphics[width=\linewidth, trim={0 0 1cm 0}]{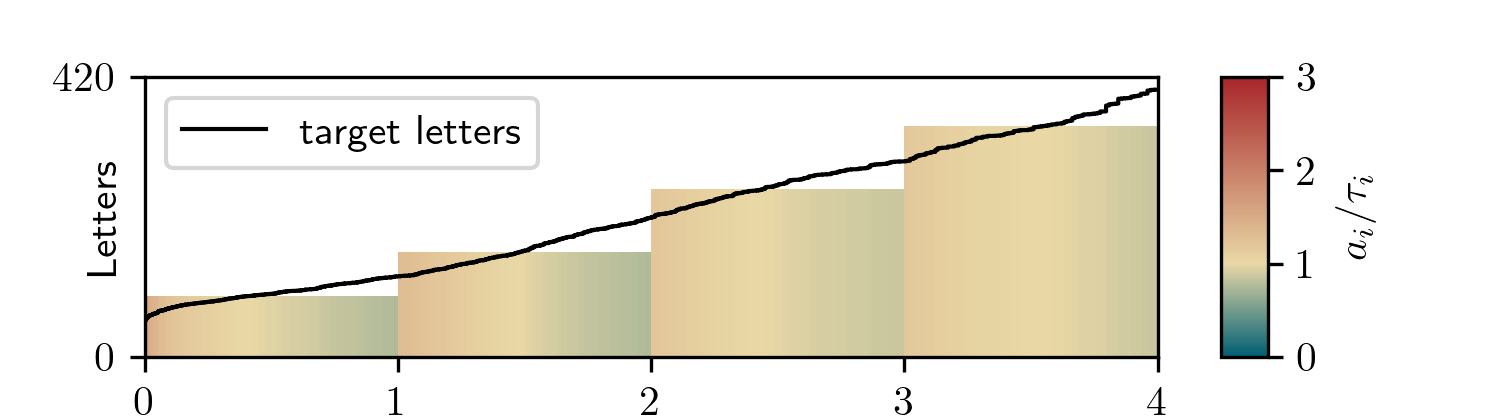}
        \caption{Proportionality of \buckets}
        \label{fig:alg_proportionality_bucket}    
    \end{subfigure}
   \caption{Probability distributions for the group of small cities in the state of Niedersachsen.}
    \label{fig:alg_results}
\end{figure*} 

\vspace*{-5pt}
\paragraph{Apportionment via Global Targets}
While a group's number of letters is just proportional to its population, we must decide how to allocate the outreach budget $t=80$ across groups. 
Let $\mathcal{G}$ be the partition of $[n]$ into groups.
Blindly apportioning the outreach budget $t$ into group budgets $t_G$ for each $G \in \mathcal{G}$
and then applying our algorithms is not ideal for meeting letter targets: Similarly sized cities in different groups may receive vastly different numbers of letters when selected, as this number depends on $t_G$.

We introduce the concept of \emph{global targets}, which help finding an apportionment that keeps letter targets comparable across groups. Given a target letter function (for \colgen and \buckets we use $f(x) = \sqrt{x}$ and for \greq we use a constant function), we compute the global target letters $\tau_i$ by finding a scaling factor $\kappa$ such that the corresponding target widths $\omega_i = \frac{\pi_i \ell}{\tau_i^{\kappa}}$ sum up to $t = 80$ (compare \Cref{sec:proportional}).

However, as argued in \Cref{sec:proportional}, within each group $G$, we need to rescale $\sum_{i \in G} \omega_i$ to a width of $t_G$ to obtain sensible \emph{local targets}. To keep the amount of rescaling required low (and, in turn, local targets close to global targets), we want to assign each group an integer budget $t_G$ close to their fractional target width $\sum_{i \in G} \omega_i$. This is an apportionment problem, for which we use an adjustment of Adam's apportionment method~\citep{BaYo01a}.
For details, see \Cref{app:federal}.

\paragraph{Meeting Local Targets}
After finding the apportionment as described above, we test our algorithms (\greq, \colgen, and \buckets) on these $42$ groups.
All algorithms find distributions for the apportioned $t_G$, and run in a practical amount of time on consumer hardware. This shows that our algorithms scale to practical problems and are plausible contenders for deployment.
We defer results and detailed discussions to \Cref{app:all_results} and display the distributions for one group in \Cref{fig:greq_result,fig:column_generation_result,fig:buckets_result}. 

\Cref{fig:alg_proportionality_column,fig:alg_proportionality_bucket} visualize how well \colgen and \buckets meet their local targets. 
The figure is the result of ordering all rectangles from \Cref{fig:column_generation_result,fig:buckets_result} by the city they represent and lining them all up next to each other in increasing order of city sizes. Each rectangle's color represents how close its height is to the target letters of that city.
We plot the local targets of cities in black.
In this instance, \colgen meets the target letters almost perfectly, which is true for most of the groups (more precisely, $35$ out of $42$ and in particular for all groups with $t_G \geq 3$). \buckets approximates the target letters, with the smaller cities within each bucket receiving slightly too many, and the larger ones slightly too few letters. \buckets struggles when there are very small cities, since the smallest city in a bucket bounds its height and limits the number of letters to the other cities in the bucket. This effect appears in $5$ out of the $42$ groups. Both approaches align more closely with local targets for higher values of $t_G$.

\paragraph{Meeting Global Targets}
We observe that the local targets of groups with $t_G > 1$ never deviate from the global targets by more than a factor of $1.5$. For groups with $t_G = 1$ the local targets are independent of the target function, as every city must receive all letters of this group when selected, which can lead to arbitrarily high deviations from the global targets.
In particular, many of the medium- and large-size groups have a low total share of the population, which leads to a target width significantly below 1. Since each group must have $t_G \ge 1$, these groups are assigned $t_G = 1$, resulting in local targets letters that can be much lower than the global targets. 
We present a visualization of global and local targets in \Cref{fig:global_local_targets_bucket_all,fig:global_local_targets_bucket_small,fig:global_local_targets_colgen_all,fig:global_local_targets_colgen_small} in the appendix.

\smallskip 

Germany holds assemblies nationally and at the state level. \Cref{fig:results_Baden-Württemberg_All} (\Cref{app:all_results}) shows results for Baden-Württemberg, a particularly active state.

\section{Discussion}

We introduced a novel two-stage sampling problem, motivated by the practical demands of selecting citizens' assemblies. Our results offer a solid algorithmic foundation and give rise to two compelling open questions: Does there exist an ex-post monotone additive~$1$-approximation algorithm? And can the representation of city groups\,---\,currently addressed via partitioning\,---\,be integrated more directly into the model? \greq and \buckets already ensure the ex-post representation of cities of different sizes by design, and one might envision a two-dimensional sampling framework, as is often used in survey sampling~\citep{Cox87}.

While these questions offer exciting directions for theory, our focus remains on practical impact.
As Germany’s newly elected government just reaffirmed its commitment to citizens’ assemblies~\citep{CCSa}, our work offers a suite of implemented algorithms, striking distinct, favorable tradeoffs between different practical desiderata.
Based on our discussions with practitioners, we are optimistic that they can soon be used to sample real assemblies.

\section{Acknowledgements}
We would like to thank Federico Fioravanti for valuable conversations in the early stages of the project, Jannik Matuschke for helpful input on configuration LPs and approximate separation oracles, Bettina Speckmann for inspiring discussions on the topic, and Brett Hennig for pointing us to the practical problem during an online talk organized by the European Digital DemocracY (EDDY) network in 2024.

Part of this work was performed while Paul Gölz was at the Simons Institute for the Theory of Computing as a FODSI research fellow, for which he acknowledges the NSF’s support through grant DMS-2023505.
Jan Maly was supported by the Austrian Science Fund (FWF) under the grants 10.55776/PAT7221724 and 10.55776/COE12, by netidee Förderungen (https://www.netidee.at/) and the Vienna Science and Technology Fund (WWTF) (Grant ID: 10.47379/ICT23025). Ulrike Schmidt-Kraepelin was supported by the Dutch Research Council (NWO) under project number VI.Veni.232.254.

\bibliography{arXiv/bib}
\bibliographystyle{abbrvnat}

\appendix
\setcounter{secnumdepth}{1}

\onecolumn

\section{Missing Proofs}
\label{app:missingproofs}

\lemTrivialLB*

\begin{proof}
    Let $\mathcal{D}$ be an ex-ante fair $t$-bounded probability distribution. To prove \Cref{lem:trivialLB-selProb}, rewrite ex-ante fairness as $$\pi_i \ell = \E{a_i}\leq  \Pr[a_i > 0]\cdot u_i,$$ which is equivalent to \begin{equation} \Pr[a_i>0] \geq w_i = \frac{\pi_i \ell}{u_i}. \label{eq:prob-lb} \end{equation}
    To show \Cref{lem:trivialLB-t}, we rewrite the sum of selection probabilities as $$t \geq \sum_{a \in A_t} \Pr[a] \sum_{i \in [n]} \mathds{1}[a_i>0] = \sum_{i \in [n]} \Pr[a_i > 0] \geq \sum_{i \in [n]} w_i,$$ which proves the lemma statement. 
\end{proof}

\thmNPHard*
\begin{proof}
We reduce from the NP-hard problem \textsc{EqualCardinalityPartition} \citep{GaJo79a}, where we are given $2k$ elements with positive integer weights $x_1, \dots, x_{2k}$ and the task is to decide whether there exists a subset $S \subseteq [2k]$ with the property that $$\sum_{i \in S} x_i = \frac{1}{2} \sum_{i \in [2k]}x_i.$$ 

We create an instance of \decision by introducing $n=2k$ cities with populations $\pi_i = \frac{x_i}{\sum_{j \in [n]}x_j}$ and upper bounds $u_i = x_i$ for all $i \in [n]$. Moreover, we set $\ell = \frac{1}{2} \sum_{j \in [n]}x_j$ and $t=k$. We now prove the equivalence of the reduction. 

\medskip

First, assume that the partition instance is a yes-instance. That is, there exists a subset $S \subseteq [2k]$ such that $S$ and $[2k] \setminus S$ have equal weight and are each of cardinality $k$. Now, consider the probability distribution $\mathcal{D}$ that, with probability $\frac{1}{2}$ selects the allocation $a_i = u_i$ for all $i \in S$ (and $a_i = 0$ for all $i \in [n] \setminus S$), and with probability $\frac{1}{2}$ selects the allocation $a_i = u_i$ for all $i \in [n] \setminus S$ (and $a_i = 0$ for all $i \in S$). This allocation is ex-ante fair and $t$-bounded. Hence, our \decision instance is a yes-instance. 

\smallskip

For the other direction, let our \decision instance be a yes-instance. We then derive two properties that have to hold for any ex-ante fair and $t$-bounded probability distribution. First, note that by construction $w_i = \frac{1}{2}$ for all $i \in [n]$. Moreover, as we have argued in \Cref{lem:trivialLB}, it holds that $$t \geq \sum_{i \in [n]} \Pr[a_i > 0 ] \geq \sum_{i \in [n]} w_i =   \frac{n}{2}= t,$$ and therefore $\Pr[a_i > 0 ] = w_i = \frac{1}{2}$ for all cities. In particular, this also implies that for any allocation in the support of $\mathcal{D}$ there exist exactly $t$ cities with $a_i = u_i$ while for all others it holds that $a_i = 0$. Thus, consider any allocation in the support of $\mathcal{D}$ and let $S$ be the set of cities with non-zero letters. Since $S$ is of weight $\sum_{i \in S} u_i = \ell = \frac{1}{2}$ and of cardinality $k$, this proves the existence of an equal-cardinality partition. 
\end{proof}

\thmPseudoAlgorithm* 

\begin{proof}
    Our goal is to decide whether the primal LP is feasible. Since the dual is always feasible (set $y_i = y = 0$ for all $i \in [n]$), this is equivalent to deciding whether the dual LP is bounded. Note that the dual is unbounded if and only if it has any solution with positive objective value (we can always scale the $y_i$'s and $y$). Thus, we can add the constraint $$\sum_{i \in [n]} \pi_i \ell y_i - y = 1$$ and ask whether the dual LP is non-empty. For deciding non-emptiness of an LP it suffices to have access to a separation oracle in order to employ the Ellipsoid method \citep{GLS93a}. 
    
    The dual separation problem is the following: Given a dual solution $((y_i)_{i \in [n]},y)$, decide whether there exist a dual constraint that is violated: $$\sum_{i \in [n]}a_i y_i \leq y \quad a \in A_t.$$ Stated differently, is $\max_{a \in A_t} \sum_{i \in [n]}a_i y_i$ larger than $y$? This problem is reminiscent of a knapsack problem and can be solved with help of a dynamic program. First note that it is without loss of generality that an optimal solution $a$ to this maximization problem gives $a_i < u_i$ to at most one city $i \in [n]$, namely to the one (among the $\leq t$ selected ones) with smallest weight $y_i$. 

    \newcommand{\dyp}{\text{DP}}    

    We start our algorithm by guessing the city with smallest weight that will be included the support of $a$, let's call this city $i^*$. Now, we relabel the cities such that $y_1 \geq \dots \geq y_{i^*} > y_{i^*} \geq \dots y_n$. For deciding which other cities will be included with their upper bound in $a$ (which will be up to $t-1$), we now write an dynamic program, which is essentially the same as for a cardinality-constrained knapsack problem. We give it for the sake of completeness.  
    We write the following dynamic program, where $\dyp(0\leq j <i^*, 0 \leq k \leq t-1, 0 \leq z \leq \ell)$ corresponds to the maximum value that we can create if we select $k$ cities from the set $\{1, \dots, j\}$ and allocate exactly $z$ letter to them. Then, the recursive formulas of the dynamic program can be described as follows: $$\dyp(0,\cdot,\geq 0) = \dyp(\cdot,0,\geq 0) = \dyp(\cdot,\cdot,0) = 0, \dyp(\cdot,\cdot,<0) = - \infty$$ and $$\dyp(j+1,k,z) = \max\left( \dyp(j,k-1,z-u_{j+1}) + y_iu_{j+1}, \dyp(j,k,z)\right).$$
    After filling the table, we choose the entry that maximizes $\dyp(i^*-1,k,z)$ among all $k \in [t-1], z \in [\ell-u_{i^*},\ell)$, let's call this value $\gamma_{i^*}$ and remember the corresponding $z$ as $z_{i^*}$. We then choose $i^*$ as to maximize $\gamma_{i^*} + (\ell - z_{i^*}) \cdot y_{i^*}$. This leads to an algorithm with running time $\mathcal{O}(n^2t\ell)$. 
\end{proof}

\thmOneapproximation*

\begin{proof}
    We start this proof very similarly to the one of \Cref{thm:pseudo}, namely, we add the constraint $$\sum_{i \in [n]} \pi_i \ell y_i - y = 1$$ and ask whether the dual LP is non-empty. This time, we show the existence of an approximate separation oracle: Given a dual solution $((y_i)_{i \in [n]},y)$, decide that the solution is feasible or provide one constraint of the following that is violated: $$\sum_{i \in [n]}a_i y_i \leq y \quad a \in A_{t+1}.$$ Note that these correspond to the dual constraint for the case of $t$, since we replaced $A_t$ by $A_{t+1}$. \citet{ScUh13a} prove in the appendix of their paper (Theorem A.6) that if an approximate separation oracle in combination with the Ellipsoid method yields an approximate algorithm for the deciding the emptiness of a polytope. More precisely, in our case this implies that if there exists a polynomial time algorithm for our approximate separation oracle, then the Ellipsoid method yields an algorithm that either decides that our dual LP for $t$ is non-empty or our dual LP for $t+1$ is empty. In the former case, this implies that the primal LP for $t$ is infeasible, in the latter case the primal LP for $t+1$ is feasible. Thus, this algorithm is an additive $1$-approximation for \optimization.

    We now provide the polynomial-time algorithm for the approximate separation oracle. For a given dual solution $((y_i)_{i \in [n]},y)$ we should either decide that 
    \begin{equation}
        m_t:=\max_{a \in A_{t}}\sum_{i \in [n]}a_i y_i \leq y \label{first_opt}
    \end{equation}
    or prove that 
    \begin{equation}
        m_{t+1} := \max_{a \in A_{t+1}}\sum_{i \in [n]}a_i y_i > y. \label{second_opt}
    \end{equation}
    Note that $m_{t+1} \geq m_t$ holds simply because $A_t \subseteq A_{t+1}$. We can write down the following LP, which serves as a LP relaxation of the former optimization problem: 
    \begin{alignat*}{5}
     & \text{maximize} &\sum_{i \in [n]} z_i y_i u_i & \\
    & \text{subject to} \quad& \sum_{i \in [n]} z_i  u_i & = \ell, &&&& \\
                 && \sum_{i \in [n]} z_i & \leq t \quad & & \quad && \\
                 && 0 \leq z_i &\leq 1 & \text{for } i & \in [n].\\
\end{alignat*} 
The idea is that $z_i$ captures the number of letters allocated to city $i \in [n]$, namely, $a_i = z_i u_i$. Note that in particular, if $a$ is an optimal solution to \Cref{first_opt}, then $z_i = \frac{a_i}{u_i}$ is a feasible solution to the above LP. Hence, $m_{t}$ is smaller or equal than the optimal value of the LP, which we call $m^*$. Now, let $z^*$ be some optimal solution to the LP that also corresponds to a basis. Since $n$ constraints must be tight in such a solution, we know that at least $n-2$ of the constraints $0 \leq z_i \leq 1$ must be tight on one of the two sides. Thus, we have at most two variables $z_i$ with fractional values, thus there are at most $t+1$ variables with non-zero entry. We construct a solution $a^*_i = z^*_i u_i$ for all $i \in [n]$. This solution is almost an element of $A_{t+1}$ with the subtlety that there might be two cities that may receive a non-integral amount of letters (those with fractional $z_i$). Let's call these cities $i(1)$ and $i(2)$ and assume wlog that $y_{i(1)} \geq y_{i(2)}$. Then, we round $a^*_{i(1)}$ up and $a^{*}_{i(2)}$ down. While the corresponding $z$ solution would not be feasible for the LP (it would violate the $t$-bound), $a^*$ remains $(t+1)$-bounded and only increases the objective value of $a^*$ wrt the $y_i$ weights. Hence, we can assume wlog that $a^* \in A_{t+1}$.

If $m^* \leq y$, then we report that the point $((y_i)_{i \in [n]},y)$ is feasible for our dual LP for $t$, which is true since $m_t \leq m^* \leq y$. On the other hand, if $m^* > y$, then we return the violated constraint from the dual constraint for $t+1$ that corresponds to $a^* \in A_{t+1}$, which is true since $$\sum_{i \in [n]} a^*_i y_i = \sum_{i \in [n]} z^*_i u_i y_i = m^* > y,$$ which implies $m_{t+1} \geq m^* > y$. 
\end{proof}

\thmGreqCorrect*

\begin{proof}
    Note that, when \greq succeeds in particular it needs to hold that $\Lambda(x) = \ell$ for all $x\in[t-1,t)$ since both the area of the rectangle and the sum of the areas of the cities sums to $\ell$. 
    
    We define the corresponding probability distribution $\mathcal{D}$ as follows: Sample $\alpha \in [0,1)$ uniformly at random and for any $k \in [t-1]_0$, let $i(k)$ be the unique city with $\lambda_i(k + \alpha)>0$. Then, return $a_{i(k)} = \sum_{k' \in [t-1]_0}\lambda_{i(k)}(k'+\alpha)$ for all $k \in [t-1]_0$ and $a_j = 0$ for all other cities $j$. Note that typically, the sum in the definition of $a_{i(k)}$ only contains one non-zero entry with the exceptions of cities that "wrap" around several layers. 
    
    We now prove that $\mathcal{D}$ is a probability distribution over $A_t$. 
    The fact that $\sum_{i \in [n]} a_i = \ell$ follows since $\Lambda(t-1+\alpha)= \ell$. Moreover, $a$ is $t$-bounded by construction. Lastly, we claim that $a_i \leq u_i$ for all $i \in [n]$. Here, we only have to be concerned about cities that "wrap around", since otherwise the constraint follows directly by the definition of the algorithm. Note that for any $\lambda_i$ at any position, $\min\{u_i, \frac{\ell}{t}\}$ is a global lower bound. This is because the average value that is compare to $u_i$ in the function $\mu_i$ starts out to be $\frac{\ell}{t}$ in the first round and then only grows over time. Thus, assume that $i \in [n]$ wraps around. Then, either $\frac{\pi_i \ell}{u_i} > 1$ or $\frac{\pi_i \ell t}{\ell} > 1$. The first constraint is a direct contradiction to our assumption from \Cref{sec:model} since this would immediately imply that our instance is infeasible for all $t$. The second constraint implies that $i$ is oversized. In this case, \Cref{ass:oversized} implies that $u_i = \ell$ and therefore $a_i \leq \ell$ holds trivially. 
    
    Lastly, $\mathcal{D}$ is ex-ante fair since by construction the area in the picture for each city $i \in [n]$ is $\pi \ell$.
\end{proof}

\thmGreqMonotone*

\begin{proof}
    Consider any interval $[\alpha,\beta)$ such that all $\lambda_i(x)$ are constant for all $x \in [\alpha,\beta)$. For $k \in [t-1]_0$ let $i(k)$ be the unique city with $\lambda_{i(k)}(k + \alpha) > 0$. We show that $$\lambda_{i(0)} \leq \lambda_{i(1)} \leq \dots \leq \lambda_{i(t-1)}.$$ Let $t' \in [t'-1]_0$ be the largest index such that $\lambda_{i(t')} = u_{i(t')}$ (if no such index exists, we set $t'=-1$). For all indices smaller or equal to $t'$, the statement follows from the monotonicity of the letter bounds $u_i$. We now show that also $\lambda_{i(t')} \leq \lambda_{i(t'+1)}$ holds. By definition of \greq it holds that 
    \begin{align*}
        \lambda_{(t'+1)} &= \frac{\ell - \Lambda(t'+\alpha)}{t-(t'+1)} \\ 
        &= \frac{\ell - \Lambda(t' -1 +\alpha) - \lambda_{i(t')}}{t-(t'+1)} \\ 
        & \geq \frac{\ell - \Lambda(t' -1 +\alpha) - \frac{\ell - \Lambda(t' -1 +\alpha)}{t-t'}}{t-(t'+1)} \\ 
        & = \frac{\ell - \Lambda(t' -1 +\alpha)}{t-t'} \geq \lambda_{i(t')},
    \end{align*}
    which holds with strict inequality if and only if $\lambda_{i(t')} = u_{i(t')} < \frac{\ell - \Lambda(t'-1 +\alpha)}{t-t'}$. By analogous argumentation and a straightforward induction we get that $$\lambda_{i(t'+1)} = \dots = \lambda_{i(t-1)}.$$ 

    If no city ``wraps around'', i.e., all $i(k)$ are distinct, then the statement follows directly by the order of the cities. In \Cref{thm:GreqCorrectness} we have argued that in order to wrap around, a city must be oversized, i.e., $\pi_i > \frac{1}{t}$. This concludes the proof. 
\end{proof}

\LemGreqStructure*

\begin{proof}
    \begin{enumerate}
        \item[(ii)] We start by showing that the claim holds for each of the intervals $[t'-1,t')$ for $t' \in [t]$. We do so by induction over $t' \in [t]$. For the induction start, we consider the interval $[-1,0)$ where $\Lambda$ is constant $0$ and therefore the statement holds trivially. 

        Now, consider some $[t'-1,t')$ and assume that $\Lambda$ is non-decreasing for the interval $[t'-2,t'-1)$. Let $\alpha,\beta \in [0,1)$ such that $\alpha < \beta$. Our goal is to show that $\Lambda(t'-1+\alpha) \leq \Lambda(t'-1 + \beta)$. For every $k\in [t'-1]_0$ let $i(k)$ be the unique city with $\lambda_{i(k)}(k+ \alpha)>0$ and let $j(k)$ be the unique city with $\lambda_{j(k)}(k+ \beta)>0$. We write $\lambda_{i(k)}$ and $\lambda_{j(k)}$ instead of $\lambda_{i(k)}(k+ \alpha)$ and $\lambda_{j(k)}(k+ \beta)>0$ from now on. We distinguish two cases: 
        
        \textbf{Case 1:} city $j(t'-1)$ receives their maximum amount of letter, i.e., $\lambda_{j(t'-1)} = u_{j(t'-1)}$. We claim that in this case $i(t'-1)$ also receives their maximum amount of letters. Formally, 
        $$u_{i(t'-1)} \leq u_{j(t'-1)} \leq \frac{\ell-\Lambda(t'-2+\beta)}{t-(t'-1)} \leq \frac{\ell-\Lambda(t'-2+\alpha)}{t-(t'-1)}$$ by induction hypothesis and therefore $\lambda_{i(t'-1)} = u_{i(t'-1)}$. Therefore, it follows directly from the induction hypothesis that 
        \begin{align*}
        \Lambda_{i(t'-1)} &= \Lambda(t'-2+\alpha) + u_{i(t'-1)} \\ &\leq \Lambda(t'-2+\beta) + u_{j(t'-1)} = \Lambda(t'-1+\beta).
        \end{align*}

        \textbf{Case 2:} city $j(t'-1)$ does not receive their maximum amount of letter, i.e., $\lambda_{j(t'-1)} \neq u_{j(t'-1)}$. Then,
        \begin{align*}
          \Lambda(t'-1+\beta) & = \Lambda(t'-2+\beta) + \frac{\ell - \Lambda(t'-2+\beta)}{t-(t'-1)}  \\ 
          & \geq \Lambda(t'-2+\alpha) + \frac{\ell - \Lambda(t'-2+\alpha)}{t-(t'-1)} \\ 
          & \geq \Lambda(t'-1+\alpha), 
        \end{align*}
        where the first inequality follows from the fact that the stated expression is non-decreasing in $\Lambda$ and $\Lambda$ is monotone on $[t'-2,t'-1)$ by induction hypothesis. The second inequality follows from the upper bound on $\lambda_{i(t'-1)}$ in the definition of \greq.

        \medskip

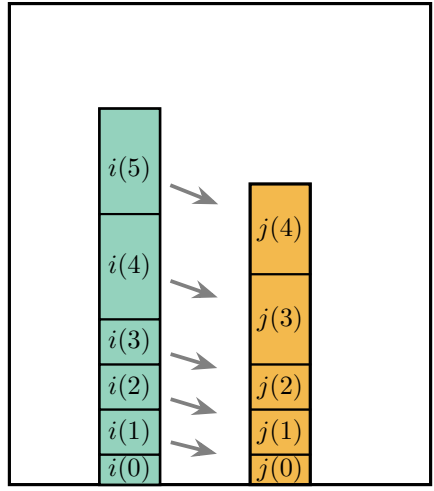
\begin{figure}[t]
    \centering
    
    \begin{tikzpicture}[scale=0.2]
    \draw[very thick] (0,0) rectangle (28,32);

    \draw[very thick] (6,0) rectangle (10,25); 

    \StackRects{4}{2,3,3,3}{lightgreen}{0}{6}
    \StackRects{4}{7,7}{lightgreen}{11}{6}
    \StackRects{4}{2,3,3}{lightyellow!70}{0}{16}
    \StackRects{4}{6,6}{lightyellow!70}{8}{16}

    \node(i0) at (8,1) {$i(0)$};
    \node(i1) at (8,3.5) {$i(1)$};
    \node(i2) at (8,6.5) {$i(2)$};
    \node(i3) at (8,9.5) {$i(3)$};
    \node(i4) at (8,14.5) {$i(4)$};
    \node(i5) at (8,21) {$i(5)$};

    \node(j0) at (18,1) {$j(0)$};
    \node(j1) at (18,3.5) {$j(1)$};
    \node(j2) at (18,6.5) {$j(2)$};
    \node(j3) at (18,11) {$j(3)$};
    \node(j4) at (18,17) {$j(4)$};

    \foreach \a/\b in {i1/j0,i2/j1,i3/j2,i4/j3,i5/j4}{
  \draw[myfatarrow] (\a) -- (\b);
}

    \draw[very thick] (16,0) rectangle (20,20); 
\end{tikzpicture}
    \caption{Situation in the second part of the proof of \Cref{lem:greqStructure}.}\label{fig:greqLambdaMonotone}
    \end{figure}

        We now move on to show that monotonicity also holds across intervals. Let $t' \in [t-1]_0$ and $\alpha,\beta \in [0,1)$ with $\alpha \leq \beta$. We aim to show that $$\Lambda(t'-1+\beta) \leq \Lambda(t'+\alpha).$$ Given this statement, the more general statement for arbitrary $x,y \in [0,t)$ follows immediately by induction.

        For every $k\in [t'-1]_0$ let $i(k)$ be the unique city with $\lambda_{i(k)}(k+ \alpha)>0$ and for every $k\in [t'-1]_0$ let $j(k)$ be the unique city with $\lambda_{j(k)}(k+ \beta)>0$. We refer to \Cref{fig:greqLambdaMonotone} for an illustration of the situation and simply write $\lambda_{i(k)}$ and $\lambda_{j(k)}$ from now on. We prove the statement by showing that $\lambda_{i(k)} \geq \lambda_{j(k-1)}$ for all $k \in [t']$. Note that $i(0) \leq j(0) \leq \dots \leq i(t'-1) \leq j(t'-1) \leq i(t')$. We distinguish two cases:

        \textbf{Case 1:} city $i(k)$ receives their maximum amount of letter, i.e., $\lambda_{i(k)} = u_{j(k)}$. Then, $$\lambda_{j(k-1)} \leq u_{j(k-1)} \leq u_{i(k)} = \lambda_{i(k)}.$$

         \textbf{Case 2:} city $i(k)$ does not receive their maximum amount of letter, i.e., $\lambda_{i(k)} - \frac{\ell - \Lambda(k-1+\alpha)}{t-k}$. Then, 
         \begin{align*}
             \lambda_{i(k)} &= \frac{\ell - \Lambda(k-1+\alpha)}{t-k} \\ 
             & \geq \frac{\ell - \Lambda(k-1+\beta)}{t-k} \\ 
             & \geq \lambda_{j(k)} \\ 
             & \geq \lambda_{j(k-1)}, 
         \end{align*}
        where the first inequality follows from the monotonicity of $\Lambda$ within the intervals, which we showed in the first part of the proof. The second inequality follows by the upper bound on $\lambda_{j(k)}$ enforced in the definition of \greq and the last inequality follows from the ex-post monotonicity which we showed in \Cref{thm:GreqMonotone}. 

        Therefore $$\Lambda(t' + \alpha) > \sum_{k=1}^{t'} \lambda_{i(k)} \geq \sum_{k=0}^{t'-1} \lambda_{j(k)} = \Lambda(t'-1 + \beta),$$ which concludes the proof. 
        \item[(i)]{Let $i \in [n]$ and $x \in [0,t)$ be such that \greq sets $\lambda_i(x) = \frac{\ell - \Lambda(x-1)}{t - \lfloor x \rfloor}$. Let $y \in [x,\lceil x \rceil)$ and let $j$ be the city that is active at point $y$ (which might be $x$ itself). Then, $u_i \leq u_j$ and $\Lambda(x-1) \leq \Lambda(y-1)$ by statement (ii). Therefore $$u_j \geq u_i \geq \frac{\ell - \Lambda(x-1)}{t - \lfloor x \rfloor} \geq \frac{\ell - \Lambda(y-1)}{t - \lfloor y \rfloor},$$ and therefore \greq selects the average at point $y$ as well. 
        
        \smallskip 
        For the second part of the statement we refer to the proof of \Cref{thm:GreqMonotone} which directly implies this statement.}
    \end{enumerate}
\end{proof}

\thmGreqTwoApprox*

    \begin{proof}[(second part)]
        With the help of the above claim, we can now apply the scaling operation over all intervals. We overload notation when defining for each interval $[\alpha,\beta)$ and $k \in [t-1]_0$, we define the city $i(k)$ to be defined the unique city with the property that $\lambda_i(k)(k+\alpha)>0$. 
    \begin{align*}
        \sum_{i \in [n]} w_i &= \sum_{i \in [n]}\frac{\pi_i \ell}{u_i} \\
        & > \sum_{i \in [n]} \frac{\int_{0}^t\lambda_i(x)dx}{u_i} \\ 
        & \geq \sum_{i \in [n]} \frac{\sum_{[\alpha,\beta) \in \mathcal{I}}\sum_{k = 0}^{t-2} \lambda_i (k+\alpha) (\beta - \alpha)}{u_i} \\ 
        & = \sum_{[\alpha,\beta) \in \mathcal{I}} (\beta-\alpha)\sum_{k = 0}^{t-2} \frac{\lambda_{i(k)}(k+\alpha) }{u_{i(k)}} \\ 
        & > \sum_{[\alpha,\beta) \in \mathcal{I}} (\beta-\alpha)(t-2) = t-2, 
    \end{align*}
    where the first inequality follows from the fact that \greq failed and therefore some cities received an area of less than their fair share. The second inequality follows by rewriting the integral and dropping the summand $k=t-1$. The equality follows from swapping the sums and the last inequality follows from the claim which we presented in the main part of the paper. As we have argued before, the above inequality implies that an optimal solution to \optimization requires at least a budget of $t-1$ whenever \greq fails. Hence, \greq is an additive $2$-approximation. 
\end{proof}

\thmGreqNotOneApprox* 

\begin{proof}

We prove this statement by providing an example (\Cref{fig:greedy_equa_no_1_appr}) for which a feasible solution with $t=3$ exists while \greq indeed requires $t=5$. Consider the instance with $\vec{\pi} = (\frac{2}{1000}, 9\times\frac{30}{1000}, 10\times\frac{33}{1000}, 5\times\frac{34}{1000}, \frac{228}{1000})$ with $\vec{u} = 1000 \pi$ and $\ell = 100$. There exists a solution with $t = 3$, visualized in \Cref{fig:greedy_equa_no_1_appr_opt}. However, \greq fails for $t = 4$, as illustrated in \Cref{fig:greedy_equa_no_1_appr_greq}.
\end{proof}

    \begin{figure*}[ht]
    \centering
    \begin{subfigure}{0.4\linewidth}
        \includegraphics[width=\linewidth]{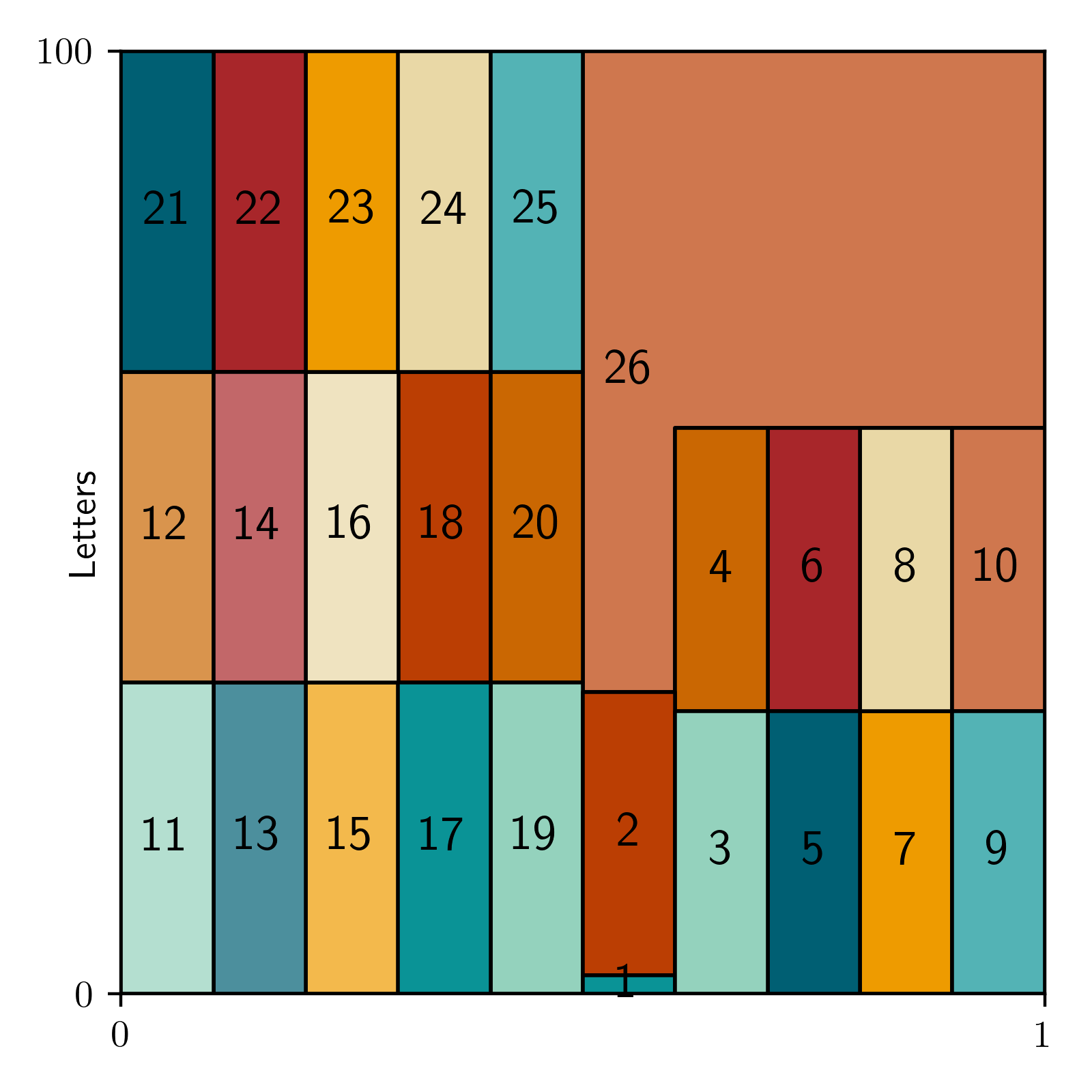}
        \caption{Optimal solution with $t = 3$}
        \label{fig:greedy_equa_no_1_appr_opt}
    \end{subfigure}
    \begin{subfigure}{0.4\linewidth}
        \includegraphics[width=\linewidth]{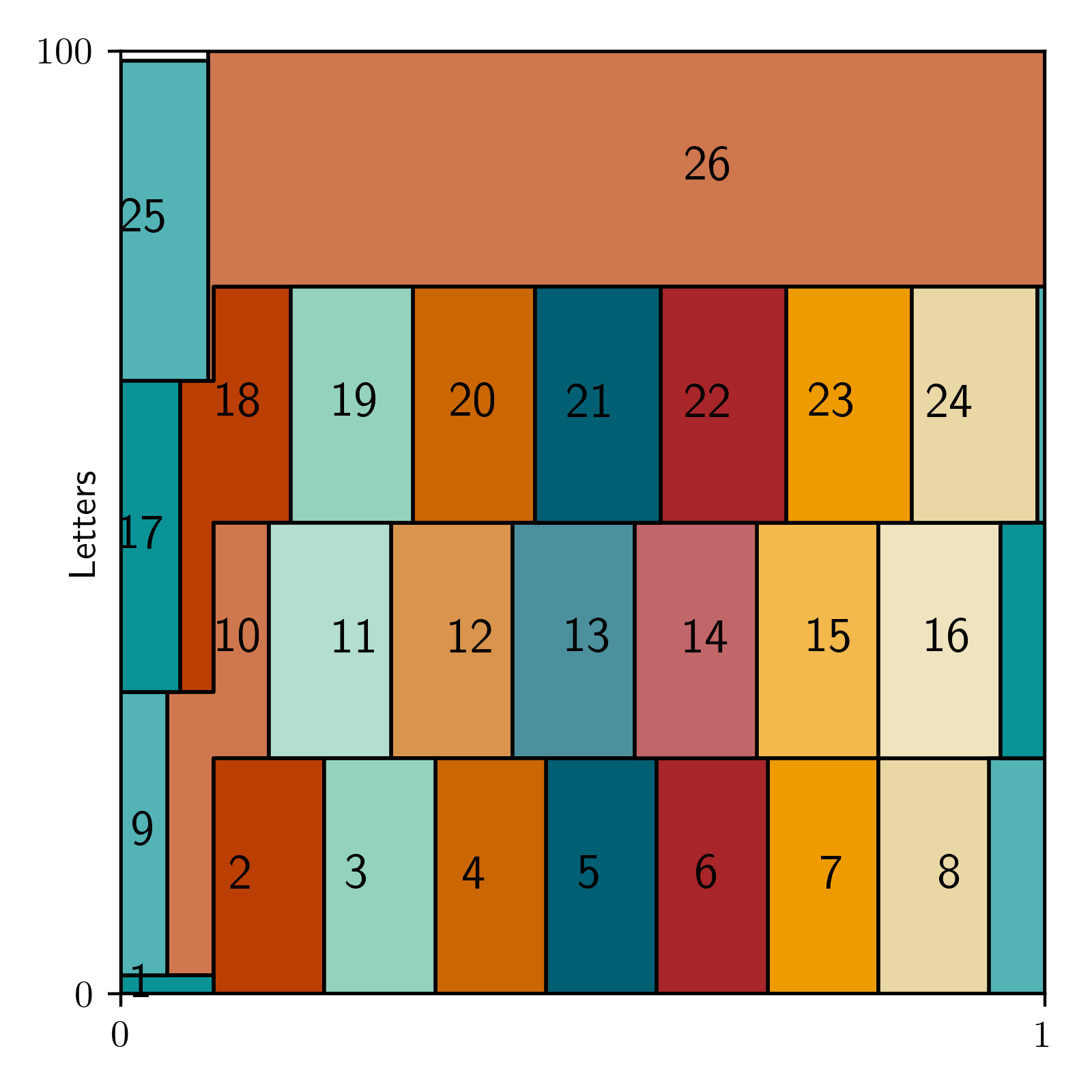}
        \caption{\greq failing for $t = 4$}
        \label{fig:greedy_equa_no_1_appr_greq}
    \end{subfigure}
    \caption{Example instance showing that \greq is not a 1-approximation.}
    \label{fig:greedy_equa_no_1_appr}
\end{figure*}

\section{Missing Proofs and Details of \Cref{sec:proportional}} \label{app:proportional}

\thmBucketsNoApprox*

    \begin{figure}
    \centering
    \begin{subfigure}{0.3\linewidth}
        \includegraphics[width=\linewidth]{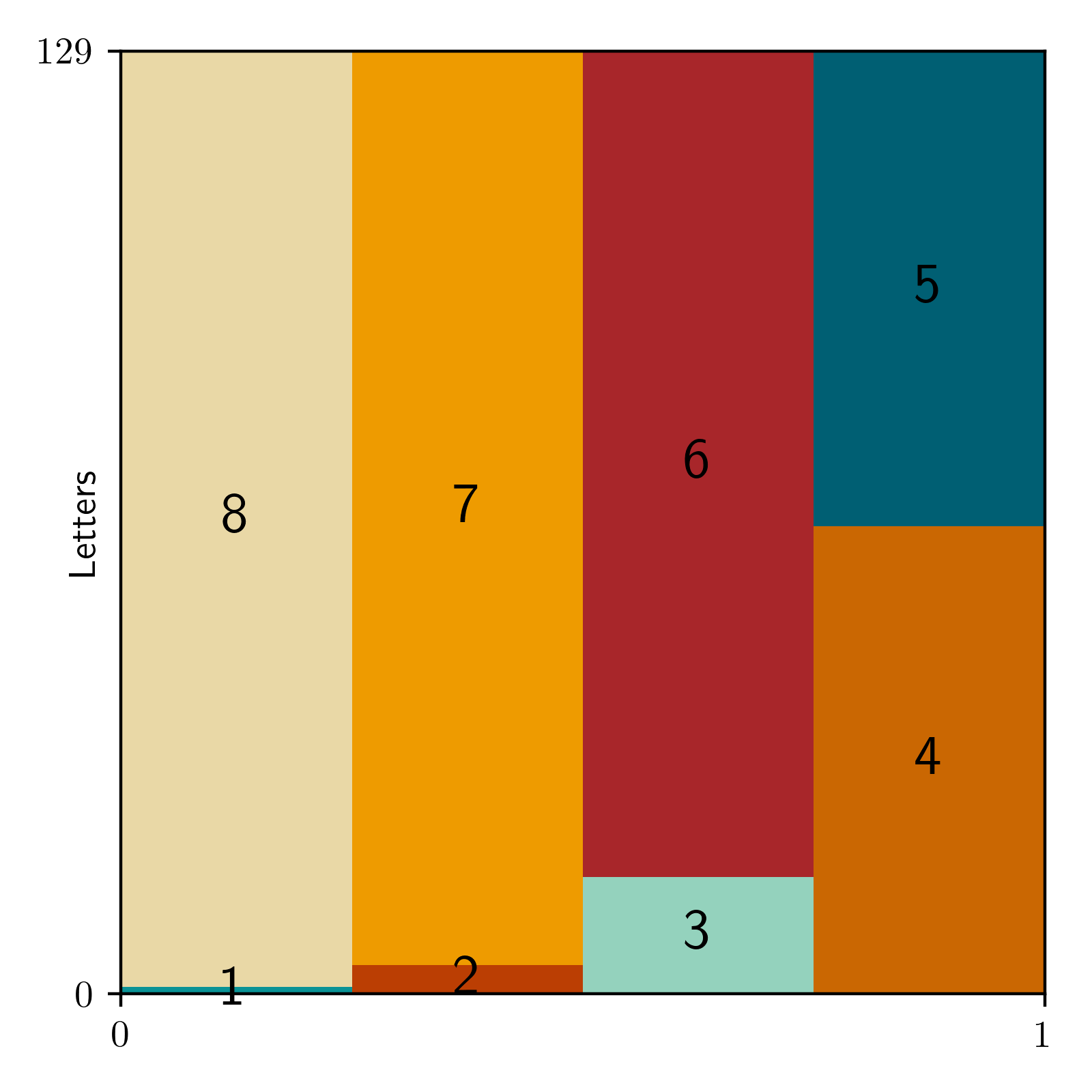}
        \caption{Optimal solution with $t=2$}
        \label{fig:bucket_bad_optimal}
    \end{subfigure}
    \begin{subfigure}{0.3\linewidth}
        \includegraphics[width=\linewidth]{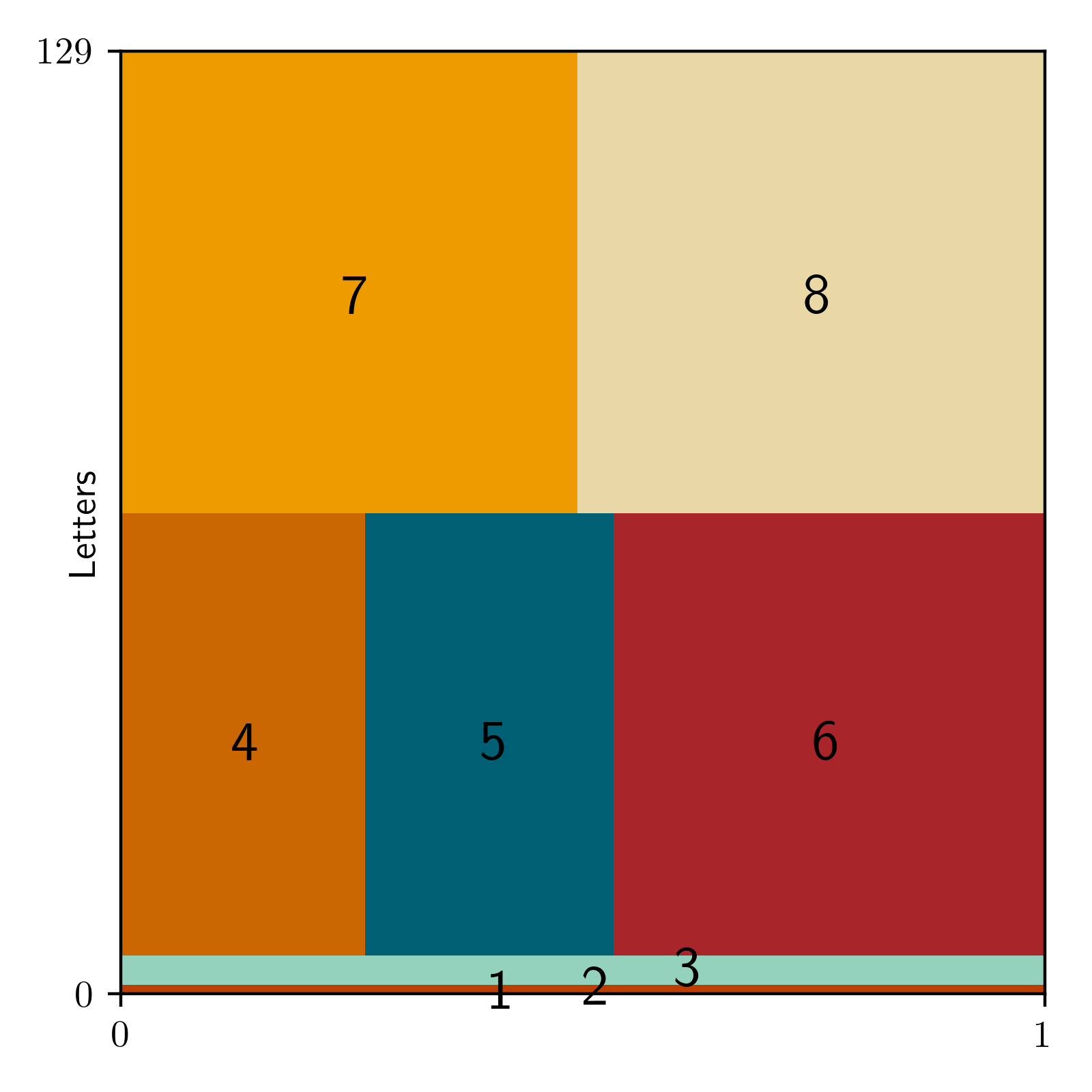}
        \caption{\buckets for $t = 5$}
        \label{fig:bucket_bad_bucket}
    \end{subfigure}
    \caption{Example instance showing that \buckets is not a 2-approximation. $c = 1$, $z = 4$, $\pi = 516$, $\vec{\pi} = (\frac{1}{516}, \frac{4}{516}, \frac{16}{516}, \frac{64}{516}, \frac{65}{516}, \frac{113}{516}, \frac{125}{516}, \frac{128}{516})$, $\vec{u} = 516\vec{\pi} = (1, 4, 16, 64, 65, 113, 125, 128)$, $\ell = 129$}
    \label{fig:bucket_bad}
\end{figure}

\begin{proof}
    Let $c \ge 1$ be a constant. We construct an instance which is feasible for $t = 2$, for which the bucket approach fails for any $t \le 2+c$. Define $z = c+3$ and $\pi = 2z^z+z$. Consider an instance with $n = 2z$ cities and $\ell = \frac{\pi}{z}$ letters, where city $i$ has size $\pi_i = \frac{z^{i-1}}{\pi}$ for $1 \le i \le \frac{n}{2}$ and $\pi_i = \frac{\frac{\pi}{z} - \pi_{n-i-1}}{\pi}$ for $z+1 \le i \le 2z$. Each cities maximum number of letters is given as $u_i = \pi_i \pi$. 
    The optimal solution achieves $t = 2$ by choosing any pair of cities $k$ and $n-k$ for $k \in [w]$ uniformly at random and assigning them their maximum number of letters. More formally, this solution is given as $\lambda_i(x) = u_i$ at positions $x \in [\frac{i-1}{z}, \frac{i}{z})$ for $1 \le i \le z$ and $\lambda_i(x) = u_i$ at positions $x \in [\frac{n+z-i}{z}, \frac{n+z-i+1}{z})$ for $z+1 \le i \le 2z$. \Cref{fig:bucket_bad_optimal} shows the optimal solution for $c = 2$.

    However, \buckets fails for any $t \le c + 2$, independent of the choice of target letters. This is because any city $1 \le i \le z-1$ will end up in a single bucket, i.e. in the algorithm we have $i^\star = i = j$ in the first $\frac{n}{2}-1$ iterations, since $\sum_{k = i}^{i+1} \pi_k \ell = (z^{i-1} + z^i) \frac{\pi}{z} = (z^{i-2} + z^{i-1}) \pi > z^{i-1} \pi = u_i$. \Cref{fig:bucket_bad_optimal} shows the probablity distribution constructed by \buckets solution for $c = 2$ and $t = 5$. Thus, for any $t \le z-1 = c+2$, \buckets will have $i < n$ after the loop and fails.
\end{proof}

\subsection{Column Generation}

We formulate finding the most proportional probability distribution as a linear program with a variable $x_a$ for each $t$-bounded allocation $a \in A_t$, representing its probability.
\begin{alignat*}{5}
 & \text{minimize} & \sum_{a \in A_t} x_a \varphi(a)& \\
 & \text{subject to} \quad& \sum_{a \in A_t} x_a & = 1, &&& (y)\\
                 && \sum_{a \in A_t} x_a a_i & \ge \pi_i \ell \quad & \text{for } i & \in [n], \qquad & (y_i)\\
                 && x_a & \in \mathbb{R} & \text{for } a & \in A_t,\\
                 && x_a & \ge 0, & \text{for } a & \in A_t.\\
\end{alignat*}

To approximate or find an optimal solution for the LP, we first formulate the dual LP with variables $y$ and $y_i$ for each city $i \in [n]$.
\begin{alignat*}{5}
 & \text{maximize} & y + \sum_{i \in [n]} \pi_i \ell y_i& \\
 & \text{subject to} \quad& y + \sum_{i \in [n]} a_i y_i & \le \varphi(a) \quad & \text{for } a & \in A_t, \qquad & (x_a)\\
                 && y & \in \mathbb{R}, &  & \\
                 && y_i & \ge 0, & \text{for } i & \in [n].\\
\end{alignat*}

We start with any ex-ante fair probability distribution over $A_t$ and solve the above LP with only the constraints corresponding to allocations $a$ with positive probability. We then iteratively compute an allocation that we can add to decrease the disproportionality measure.

Given a solution $y$, $y_i$ for a (partial version of) the dual above, we define the separation oracle as a mixed integer program with variables $a_i$ and auxiliary binary variables $z_i$ for each city $i \in [n]$:
\begin{alignat*}{5}
 & \text{maximize} & y + \sum_{i \in [n]} a_i y_i &- \sum_{i \in [n]} \Big\lvert \frac{a_i}{\tau_i} - 1 \Big\rvert& \\
 & \text{subject to} \quad& \sum_{i \in [n]} a_i & = \ell,\\
                 && \sum_{i \in [n]} z_i & \le t,\\
                 && a_i & \le u_i, & \text{for } i & \in [n],\\
                 && a_i & \le z_i \ell \quad & \text{for } i & \in [n],\\
                 && z_i & \in \{0, 1\}, & \text{for } i & \in [n],\\
                 && a_i & \ge 0, & \text{for } i & \in [n].\\
\end{alignat*}
Note that the objective contains an absolute value, which can be linearized using standard ILP techniques. We implemented this separation oracle using Gurobi 12 (under academic license).

\subsection{Bucket Approach}

Below we provide a formal description of \buckets which creates in each while loop a bucket containing cities $i$ to $i^*$. Note that the second part of the pseudocode (starting from the for loop) is only there to draw a picture similar to those of \greq. Each bucket is filled, until its height would be higher than the letter upper bound of its smallest city or the bucket target width would be larger than 1. Note, that the stop condition in the algorithm implicitly rescales the remaining municipalities target width to a total of $t-j+1$

\begin{algorithm}
    \begin{algorithmic}
        \Procedure{Buckets}{$\Vec{\pi}, \Vec{u}, t, \vec{\tau}$}
        \State $i \gets 1$, $j \gets 1$
        \While{$j \le t$ and $i \le n$}
            \State \parbox[t]{\dimexpr\linewidth-\algorithmicindent}{%
               $i^\star \gets \max\Big\{i' \in [n] \,\, \Big| \, \sum\limits_{k = i}^{i'} \pi_k \ell \le u_i \text{ and } (t-j+1) \sum\limits_{k = i}^{i'} \tau_k \le \sum\limits_{k = i}^{n} \tau_k \Big\}$ 
            }
            \State $x \gets j-1$
            \For{$k = i, \dots, i^\star$}
                \State $h \gets \sum_{k = i}^{i^\star} \pi_k \ell$
                \State $\lambda_k(z) \leftarrow h$ for $z \in [x, x+ \frac{\pi_k \ell}{h})$
                \State $x \gets x + \frac{\pi_k \ell}{h}$
            \EndFor
            \State $j \gets j+1$, $i \gets i^\star + 1$
        \EndWhile
        \If{$i < n$}
            \Return ``fail''
        \EndIf
        \EndProcedure
    \end{algorithmic} 
\end{algorithm}

\begin{figure}
    \centering
    \includegraphics[width=0.4\linewidth]{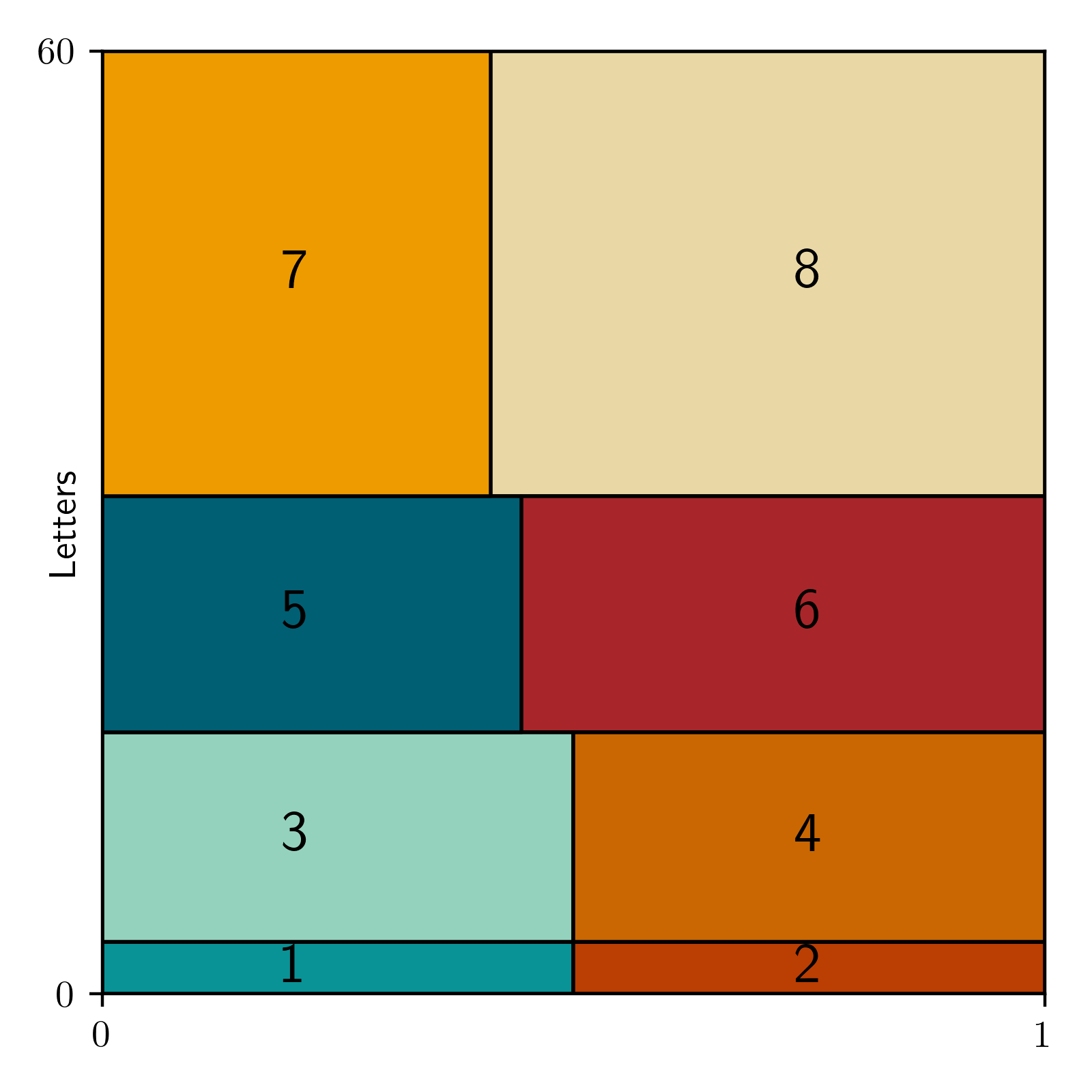}
    \caption{Probability distribution constructed by \buckets for \Cref{ex:running_example}, with $t = 4$ and square root target letters.}
    \label{fig:running_example_bucket}
\end{figure}

\section{Missing Details of \Cref{sec:federal}} \label{app:federal}

In this section we explain in more detail how we distribute the total number of letters $\ell$ and the number of cities to be selected $t=80$ over the $42$ groups of cities. For each group, we want a number of letters $\ell_G$ and a number of cities to contact $t_G$ to apply one of the methods presented in \Cref{sec:monotone,sec:proportional}. Given that the total number of letters is large ($\ell = 20\,000$), we can apply randomized rounding to these fair shares without introducing significant deviations.

The distribution of $t$ over the groups on the other hand is not determined by ex-ante fairness. However, there are upper and lower bounds on the number of cities we select from each group. For each group, $t_G$ must be at least the minimum number $t_G^{\min}$ that the corresponding algorithm requires for the group (note, that this will always be at least 1). On the other hand, $t_G$ can not be larger than the number of cities $n_g$ in a group, which is mostly relevant for the groups with only one or two cities (compare \Cref{tab:apportionment}). Under these constraints, we then try to apportion $t$ in a way, such that similarly sized cities receive similar number of letters if selected, across groups. First, we compute the \textit{global} target letters $\tau_i^{\text{glob}}$ for each city with respect to the total number of letters $\ell = 20\,000$ and $t = 80$. Since we will have to recompute these targets within each group $G$, depending on the value we assign to $t_G$ and we want the recomputed targets to be close to the global ones, we try to keep $t_G$ close to the sum the sum of (global) target widths within a group, limiting the amount of rescaling required. To achieve this, define the parametrized target width $t_G^\gamma$ of a group as the rounded sum of (global) target widths of its cities, bounded by the aforementioned bounds.  
\begin{equation}\label{eq:apportionment}
    t_G^\gamma = \max \left(\min \left(\left\lceil \gamma \sum_{i \in N_G} \frac{\pi_i \ell}{t_G^{\text{glob}}}\right\rceil, n_g  \right), t_G^{\min} \right).
\end{equation}
Then, find a value $\gamma$, such that $\sum_{G \in \mathcal{G}} t_G^\gamma  = 80$ and assign $t_G = t_G^\gamma$ to each group $G$. This process can be described as running Adam's apportionment method \citep{BaYo01a} on the global target widths of the groups, while enforcing their upper and lower bounds.
In principle, we could use any rounding function in \Cref{eq:apportionment}, however rounding up seems like a natural choice, since it ensures that each group receives at least one, slightly reducing the bias that is introduced by the values of $t_G^{\min}$.

For the experimental results and an evaluation of this method, see \Cref{app:all_results}.

\section{Experimental Results}\label{app:all_results}

In this section we present the results of our experiments in more detail. 
All experiments were run on a machine with specifications as indicated in \Cref{tab:specs}. The total running time of \greq and \buckets was below one minute, while \colgen took several hours to compute.

\subsection{Results Within Groups}
We compute the results of \greq, \colgen and \buckets for each of the $42$ groups of cities and present the results in Figures \ref{fig:results_Baden-Württemberg_Large} to \ref{fig:results_Thüringen_Small_greedy_equal}.

Note, that since the choice of $t_G$ depends on the method and target letter function, $t_G$ can differ between methods. 
In addition to the output probability distributions, we also visualize how well each method aligns with the target letters. We refer to the explanations in \Cref{sec:proportional} on how to read these figures. We use the square-root target function $f(x) = \sqrt{x}$ for \colgen and \buckets, and also evaluate \greq with respect to the target function $f(x) = \ell$, even though it does not take any targets as input.

For the groups of medium and large cities \greq mostly achieves its goal of assigning all cities the same number of letters.
For the groups of small cities this is not possible, thus some of the larger cities can receive variable numbers of letters.
Even though the assumption from \Cref{thm:GreqMonotone} on city sizes is not met in all groups, there are no ex-post monotonicity violations in any of the probability distributions constructed by \greq.

\colgen iteratively computes the probability distribution, solving a mixed integer program in each iteration. In practice, the number of iterations needed until an optimal solution is found is in the hundreds for (non-trivial) groups of medium and large cities and in the thousands for the small ones. While perfectly matching the targets is not always possible (see for example \Cref{fig:results_Nordrhein-Westfalen_Large,fig:results_Sachsen-Anhalt_Small}), \colgen comes very close for most of the groups (e.g., \Cref{fig:results_Bayern_Small,fig:results_Nordrhein-Westfalen_Medium}). On the downside, \colgen often violates ex-post monotonicity and the number of letters sent to similarly sized cities can fluctuate, resulting in the spikes e.g. in \Cref{fig:results_Sachsen-Anhalt_Small}. The former could potentially be fixed, by restricting the separation oracle to monotone allocations, while the latter might be reduced by penalizing larger deviations more in the objective function.

\buckets groups cities of equal size and assigns them the same number of letters upon selection, trivially satisfying the binary outcome property. This often works well to approximate the target letters, as can be seen for example in \Cref{fig:results_Bayern_Small,fig:results_Hessen_Small}. The lower the value of $t_G$ the coarser the approximation will get. \buckets works less well for groups which have municipalities with a very low letter upper bound $u_i$, as this forces all municipalities in the same bucket to receive the same number of letters upon selection (see for example \Cref{fig:results_Thüringen_Small,fig:results_Baden-Württemberg_Small}). Again, this issue worsens for low values of $t_G$. Similarly to \greq, even though ex-post monotonicity is not theoretically guaranteed, there are no monotonicity violations on our data. An additional advantage of the probability distributions returned by \buckets is that each bucket can be sampled independently. This reduces correlations and reshuffles the combinations of municipalities selected across multiple samplings.

\begin{table}[]
    \centering
    \begin{tabular}{>{\raggedright\arraybackslash}p{4.5cm}>{\raggedright\arraybackslash}p{9cm}}
        \toprule
        Hardware Model & HP ZBook Power 15.6 inch G9 Mobile Workstation PC \\
        Memory & 16.0 GiB \\
        Processor & 12th Gen Intel\textsuperscript{\textregistered} Core\texttrademark{} i7-12700H × 20 \\
        Graphics & Mesa Intel\textsuperscript{\textregistered} Graphics (ADL GT2) \\
        Disk Capacity & 2.5 TB \\
        OS Name & Ubuntu 22.04.4 LTS \\
        \bottomrule
    \end{tabular}
    \caption{Computer specifications}
    \label{tab:specs}
\end{table}

\subsection{Results Across Groups}
Apart from the sampling methods for each group, we also proposed a way to initially assign the number of cities to be selected from each group, with the aim of keeping the target letters for equally sized cities close across groups. \Cref{fig:global_local_targets_colgen_small,fig:global_local_targets_colgen_all,fig:global_local_targets_bucket_small,fig:global_local_targets_bucket_all} compare the global and local target functions and visualize how well \colgen and \buckets meet these targets.
\Cref{fig:global_local_targets_colgen_small,fig:global_local_targets_bucket_small} restrict their view to the targets for small cities. The global targets are drawn as a thick black line and the local targets of the groups are shown as thinner, colored lines. Additionally, each scatter point represents one city with their size on the x-axis and the expected number of letters they receive if selected on the y-axis. We can see that most local targets are fairly close to each other, with some notable exceptions. \colgen is generally quite close to the targets (as we have seen in the previous section), while \buckets approximates the targets function through a series of horizontal lines. 

\Cref{fig:global_local_targets_colgen_all,fig:global_local_targets_bucket_all} compare the local target functions of all groups (note the log scales). Here, we observe a fairly large number of groups with constant local targets much lower than the global ones. These are exactly the groups that were assigned $t_G = 1$ in the apportionment step. When only allowed to choose one city, there is only one way to sample: select a city with probability proportional to size and assign it the full number of letters $\ell_G$. The local target function reflects this and since these local targets are essentially independent of the target function $f$, the global and local targets can be arbitrarily far apart. This affects mostly the groups of medium and large, which is due to the target function and the resulting apportionment, which implicitly favors groups with small cities. An extreme example of this are the large cities of Saarland (compare \Cref{tab:apportionment}). The cities have a joint global target width of $0.180371$, but they must receive at least $t_G = 1$, resulting in local targets much lower than the global targets.

Note, that this is less an effect of the apportionment method and more of the choice of target function. Choosing a target function with slower growth or lowering the thresholds that define medium and large cities, could potentially increase the target width of these groups and consequently bring the local targets closer to the global ones.

\begin{table}
    \centering
    \caption{Results of the apportionment. $|G|$ is the number of cities in the group, $t_G^{\textsc{GE}}$, $t_G^{\textsc{CG}}$ and $t_G^{\textsc{B}}$ are the values of $t_G$ for \greq, \colgen and \buckets, respectively.}
    \begin{tabular}{lrrrrrrr}
\toprule
Group & Population & $\pi_i$ & $|G|$ & $\ell_G$ & $t_G^{\textsc{GE}}$ & $t_G^{\textsc{CG}}$ & $t_G^{\textsc{B}}$ \\
\midrule
Baden-Württemberg (Large) & 2158197 & 0.0256 & 9 & 511 & 2 & 1 & 1 \\
Baden-Württemberg (Medium) & 3619302 & 0.0429 & 98 & 858 & 3 & 2 & 2 \\
Baden-Württemberg (Small) & 5502758 & 0.0652 & 994 & 1305 & 4 & 6 & 6 \\
Bayern (Large) & 3010827 & 0.0357 & 8 & 714 & 2 & 1 & 1 \\
Bayern (Medium) & 2326541 & 0.0276 & 67 & 551 & 2 & 1 & 1 \\
Bayern (Small) & 8032025 & 0.0952 & 1981 & 1905 & 6 & 10 & 10 \\
Berlin (Large) & 3755251 & 0.0445 & 1 & 891 & 1 & 1 & 1 \\
Brandenburg (Large) & 185750 & 0.0022 & 1 & 44 & 1 & 1 & 1 \\
Brandenburg (Medium) & 940363 & 0.0111 & 27 & 223 & 1 & 1 & 1 \\
Brandenburg (Small) & 1447022 & 0.0172 & 385 & 343 & 2 & 2 & 2 \\
Bremen (Large) & 684864 & 0.0081 & 2 & 162 & 1 & 1 & 1 \\
Hamburg (Large) & 1892122 & 0.0224 & 1 & 449 & 1 & 1 & 1 \\
Hessen (Large) & 1658130 & 0.0197 & 6 & 393 & 2 & 1 & 1 \\
Hessen (Medium) & 1805347 & 0.0214 & 53 & 428 & 2 & 1 & 1 \\
Hessen (Small) & 2927883 & 0.0347 & 362 & 694 & 2 & 3 & 3 \\
Mecklenburg-Vorpommern (Large) & 209920 & 0.0025 & 1 & 50 & 1 & 1 & 1 \\
Mecklenburg-Vorpommern (Medium) & 396680 & 0.0047 & 8 & 94 & 1 & 1 & 1 \\
Mecklenburg-Vorpommern (Small) & 1021778 & 0.0121 & 716 & 242 & 2 & 2 & 2 \\
Niedersachsen (Large) & 1588358 & 0.0188 & 8 & 377 & 2 & 1 & 1 \\
Niedersachsen (Medium) & 2974786 & 0.0353 & 86 & 705 & 2 & 2 & 2 \\
Niedersachsen (Small) & 3577098 & 0.0424 & 847 & 848 & 3 & 4 & 4 \\
Nordrhein-Westfalen (Large) & 8438299 & 0.1000 & 30 & 2001 & 6 & 2 & 2 \\
Nordrhein-Westfalen (Medium) & 7369437 & 0.0874 & 182 & 1747 & 5 & 3 & 3 \\
Nordrhein-Westfalen (Small) & 2331380 & 0.0276 & 184 & 553 & 2 & 2 & 2 \\
Rheinland-Pfalz (Large) & 723508 & 0.0086 & 5 & 172 & 1 & 1 & 1 \\
Rheinland-Pfalz (Medium) & 690561 & 0.0082 & 17 & 163 & 1 & 1 & 1 \\
Rheinland-Pfalz (Small) & 2745081 & 0.0325 & 2279 & 651 & 3 & 6 & 6 \\
Saarland (Large) & 181959 & 0.0022 & 1 & 43 & 1 & 1 & 1 \\
Saarland (Medium) & 275178 & 0.0033 & 8 & 65 & 1 & 1 & 1 \\
Saarland (Small) & 535529 & 0.0063 & 43 & 127 & 1 & 1 & 1 \\
Sachsen (Large) & 1427967 & 0.0169 & 3 & 338 & 1 & 1 & 1 \\
Sachsen (Medium) & 723183 & 0.0086 & 21 & 172 & 1 & 1 & 1 \\
Sachsen (Small) & 1935002 & 0.0229 & 394 & 458 & 2 & 3 & 3 \\
Sachsen-Anhalt (Large) & 481447 & 0.0057 & 2 & 114 & 1 & 1 & 1 \\
Sachsen-Anhalt (Medium) & 708172 & 0.0084 & 22 & 168 & 1 & 1 & 1 \\
Sachsen-Anhalt (Small) & 997024 & 0.0118 & 194 & 237 & 1 & 2 & 1 \\
Schleswig-Holstein (Large) & 465812 & 0.0055 & 2 & 110 & 1 & 1 & 1 \\
Schleswig-Holstein (Medium) & 753307 & 0.0089 & 20 & 179 & 1 & 1 & 1 \\
Schleswig-Holstein (Small) & 1734151 & 0.0206 & 1082 & 411 & 3 & 3 & 4 \\
Thüringen (Large) & 326160 & 0.0039 & 2 & 77 & 1 & 1 & 1 \\
Thüringen (Medium) & 692661 & 0.0082 & 20 & 164 & 1 & 1 & 1 \\
Thüringen (Small) & 1108025 & 0.0131 & 583 & 263 & 2 & 2 & 2 \\
\bottomrule
\end{tabular}

    \label{tab:apportionment}
\end{table}

\subsection{Other Data}

While we evaluated our methods for the large nation-wide citizen assemblies in Germany, there are other similar projects to which they could be applied. As an example, we applied our methods to the whole state of Baden-Württemberg with $t = 11$ and $\ell = 2674$. The results can be seen in \Cref{fig:results_Baden-Württemberg_All}.

\clearpage

\begin{figure}
    \centering
    \begin{subfigure}{0.32\textwidth}
        \includegraphics[draft=\draft, width=\linewidth]{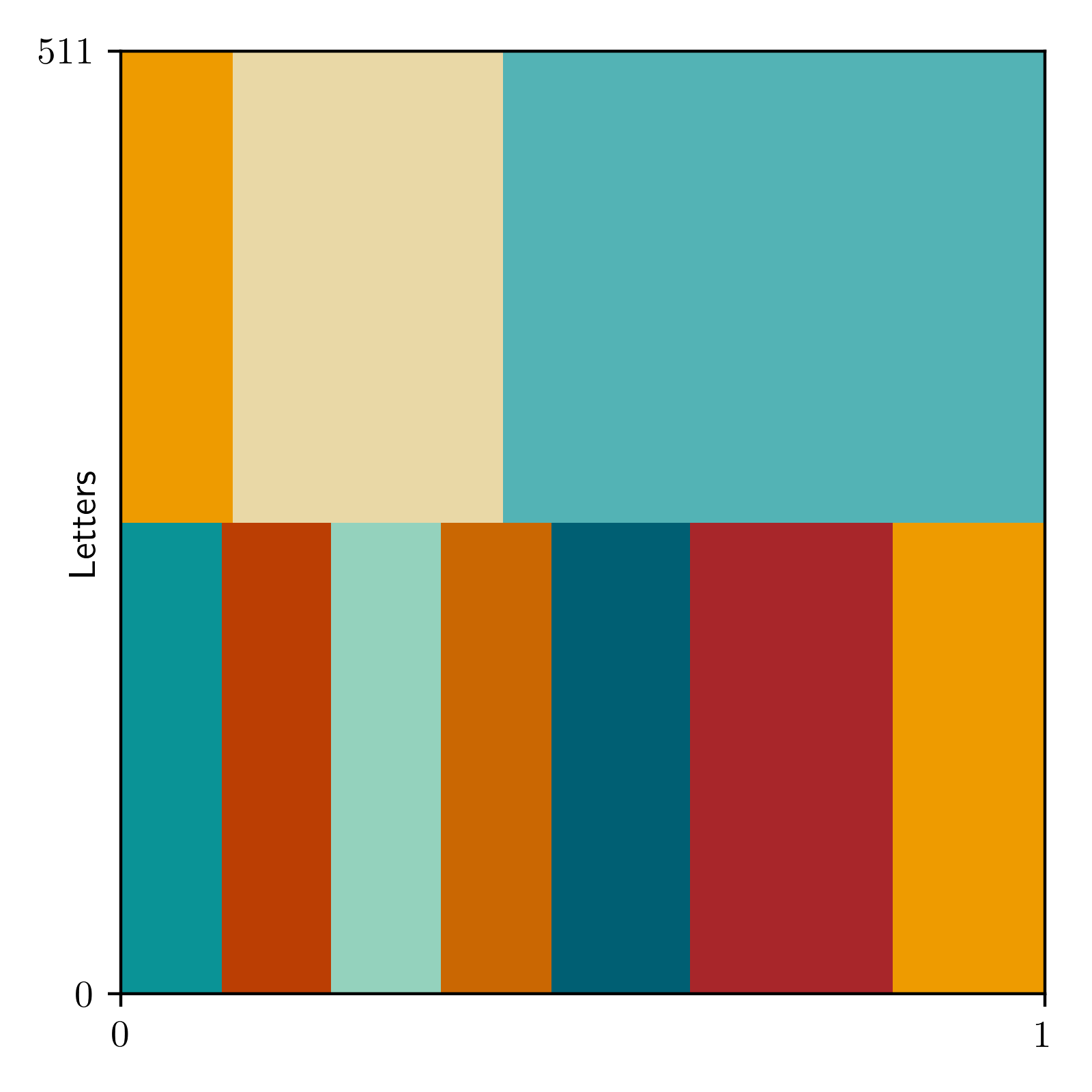}
        \includegraphics[draft=\draft, width=\linewidth]{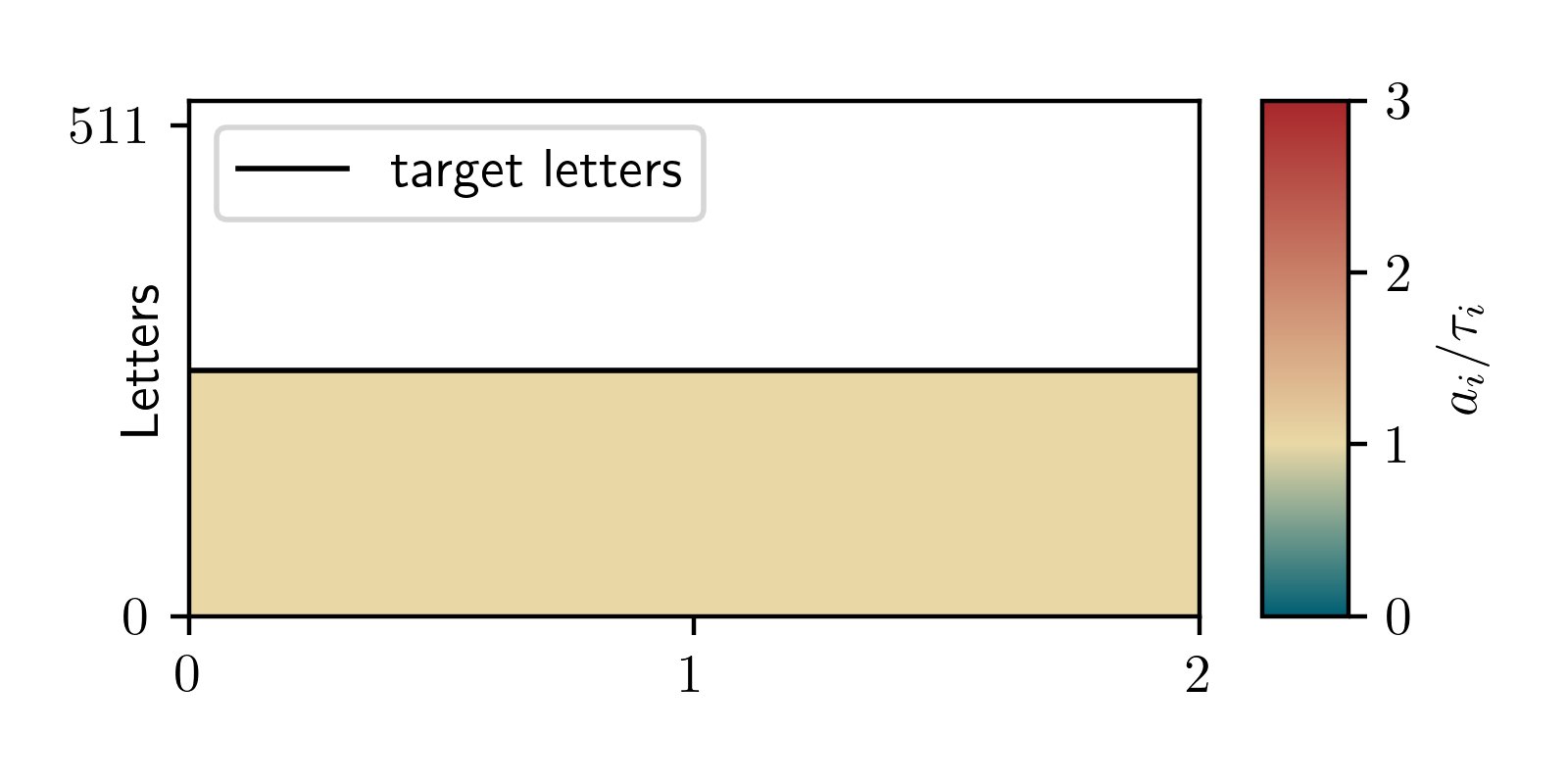}
        \caption{\greq ($t_G = 2$)}
        \label{fig:results_Baden-Württemberg_Large_greedy_equal}
    \end{subfigure}
    \begin{subfigure}{0.32\textwidth}
        \includegraphics[draft=\draft, width=\linewidth]{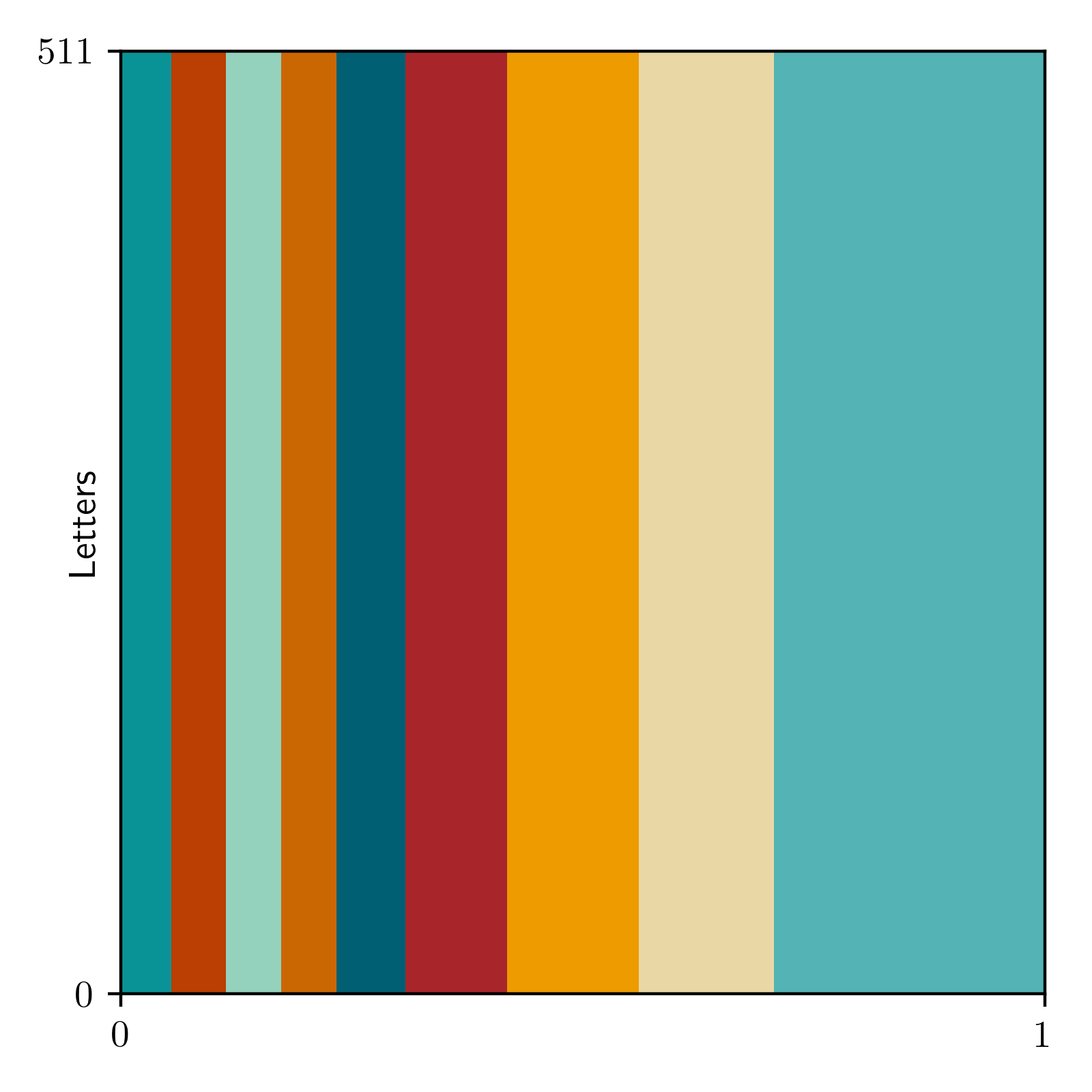}
        \includegraphics[draft=\draft, width=\linewidth]{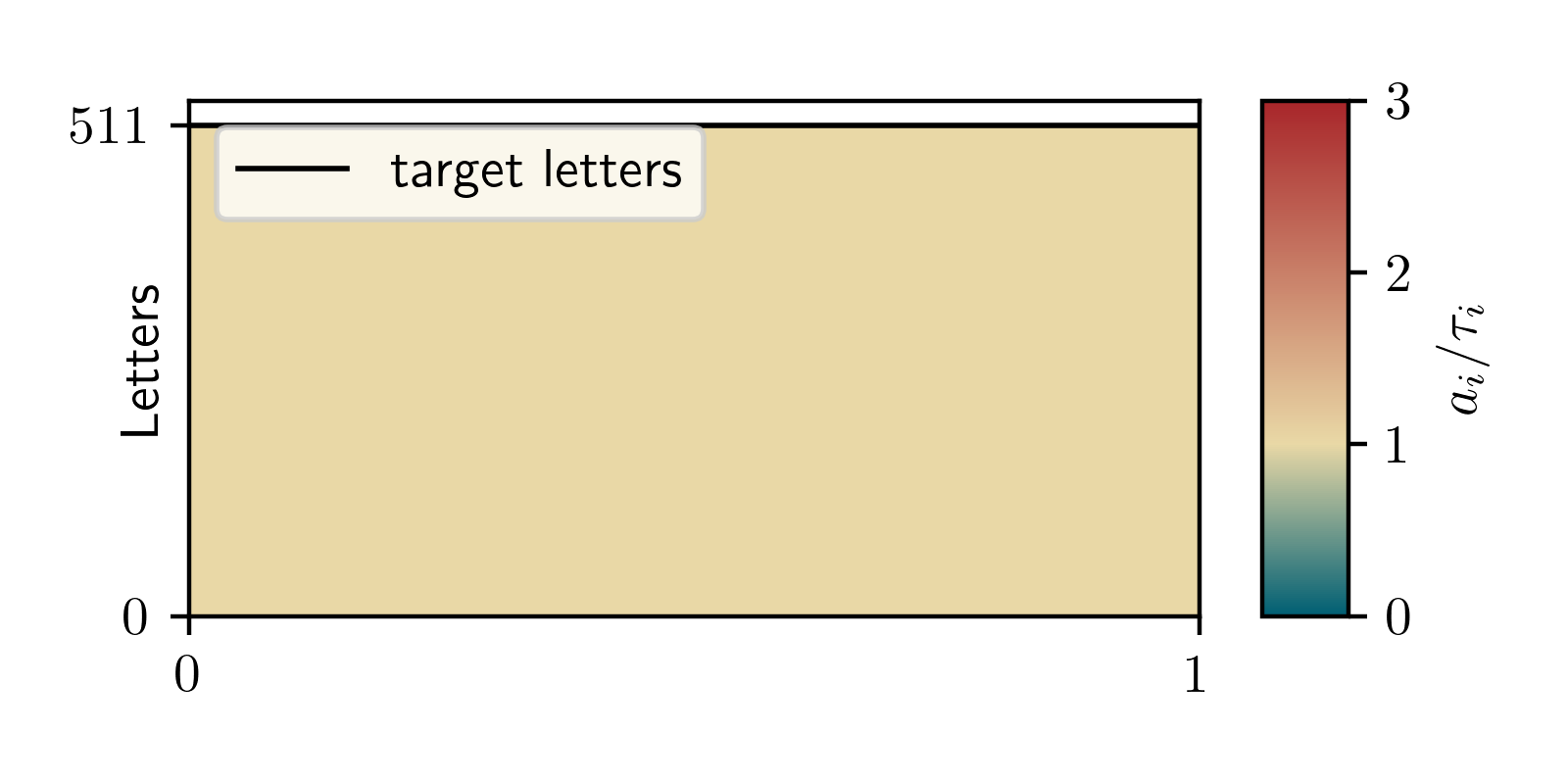}
        \caption{\colgen ($t_G\!=\!1$)}
        \label{fig:results_Baden-Württemberg_Large_column_generation}
    \end{subfigure}
    \begin{subfigure}{0.32\textwidth}
        \includegraphics[draft=\draft, width=\linewidth]{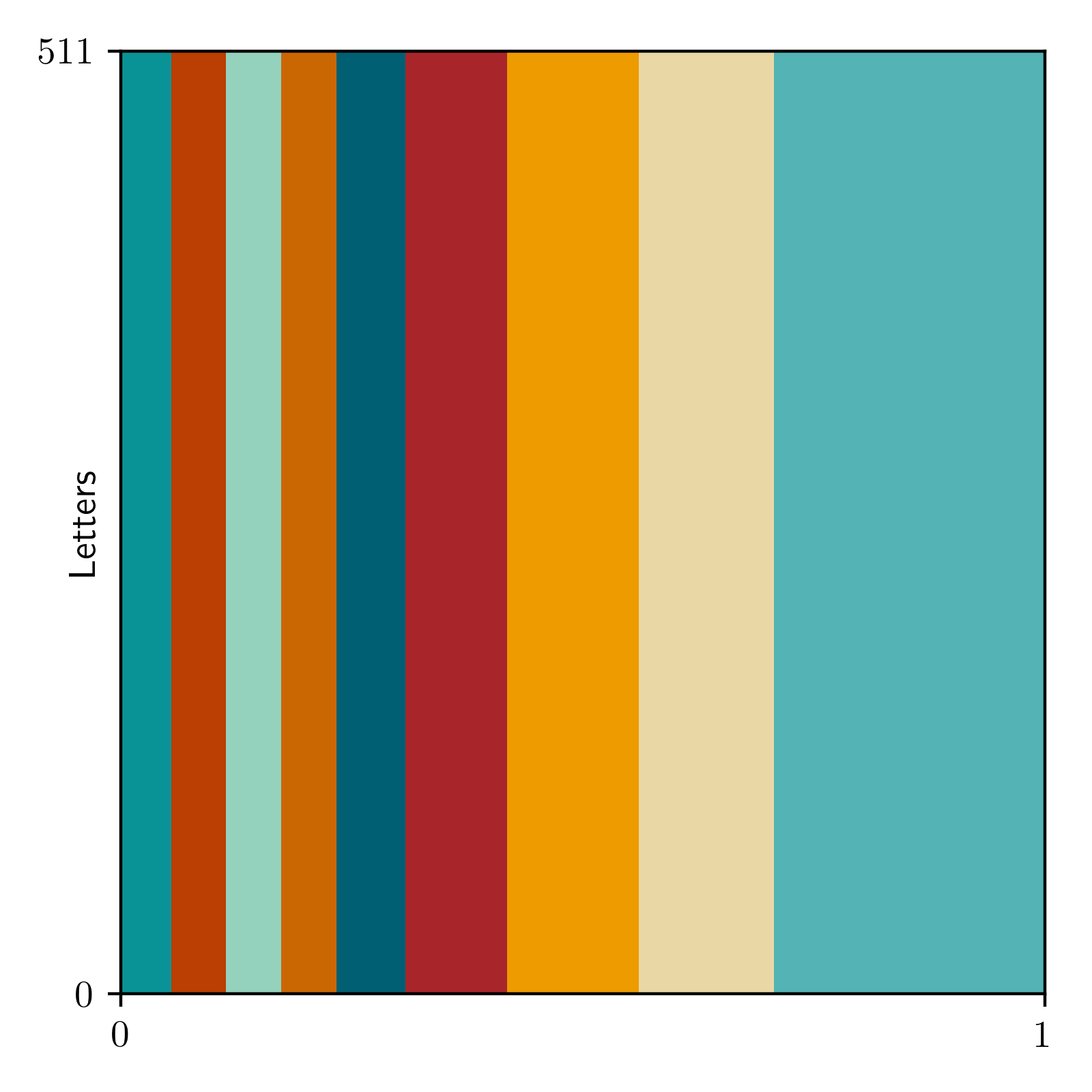}
        \includegraphics[draft=\draft, width=\linewidth]{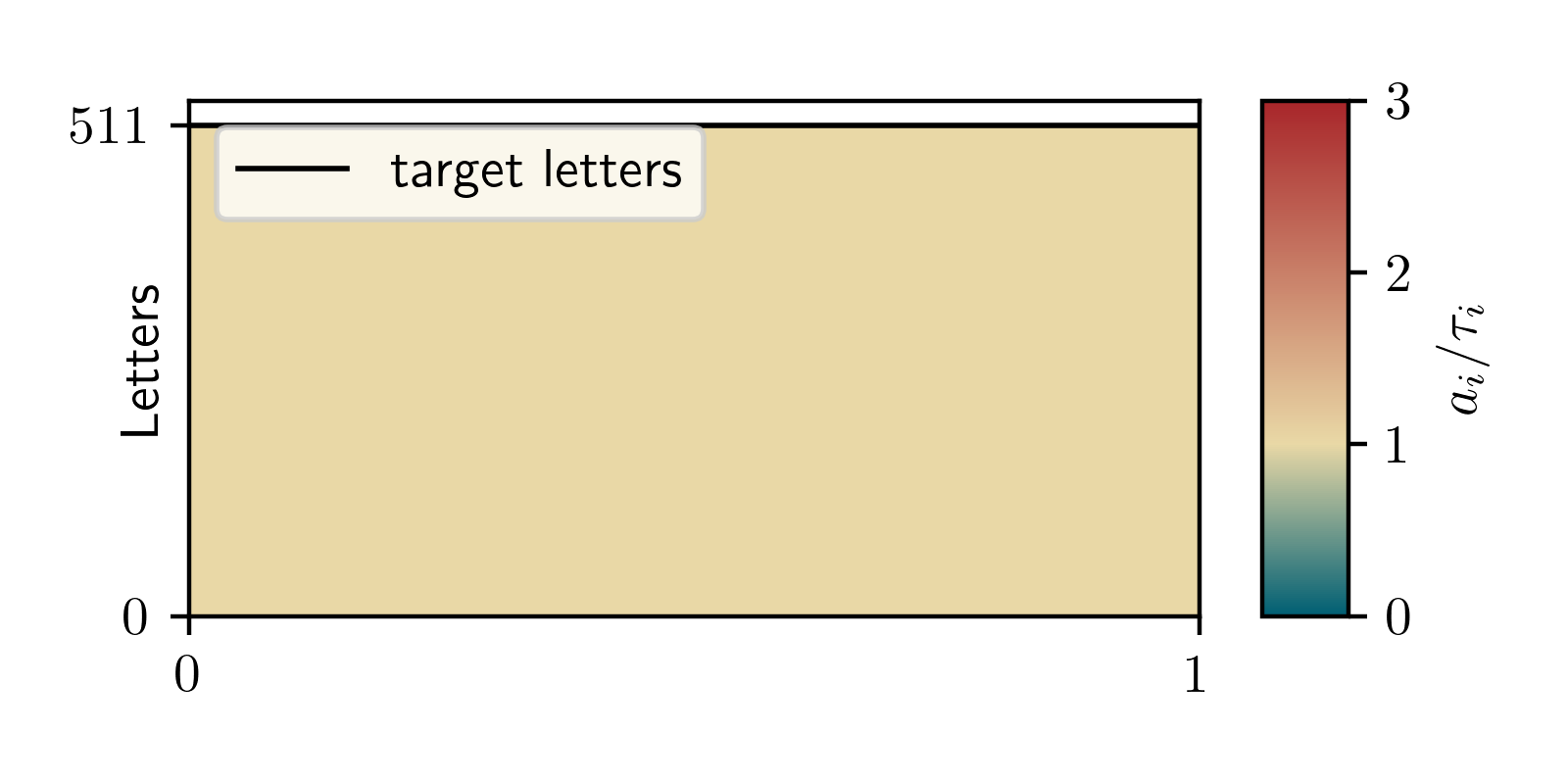}
        \caption{\buckets ($t_G = 1$)}
        \label{fig:results_Baden-Württemberg_Large_greedy_bucket_fill}
    \end{subfigure}
    \caption{Large municipalities of Baden-Württemberg ($\ell_G = 511$)}
    \label{fig:results_Baden-Württemberg_Large}
\end{figure} 

\begin{figure}
    \centering
    \begin{subfigure}{0.32\textwidth}
        \includegraphics[draft=\draft, width=\linewidth]{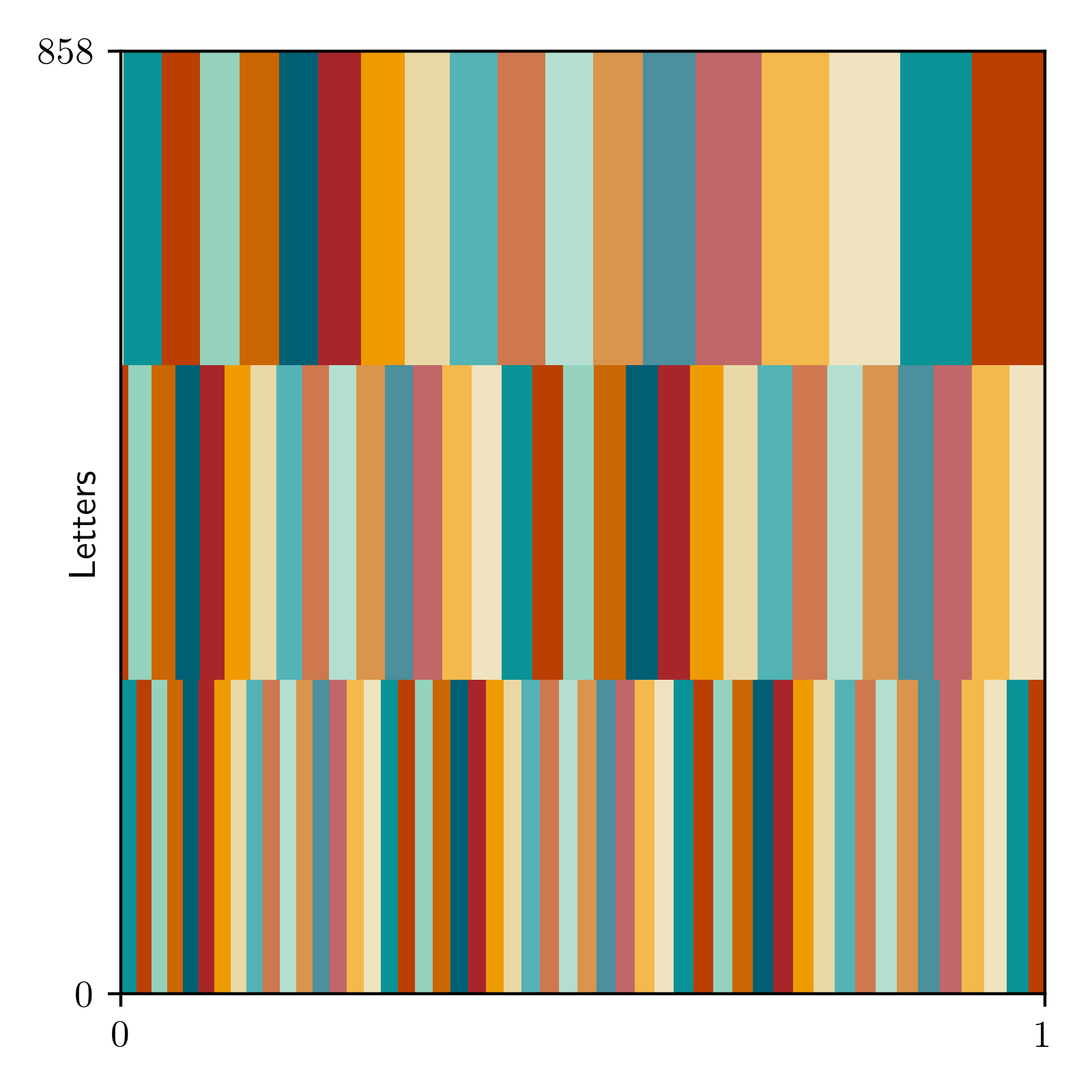}
        \includegraphics[draft=\draft, width=\linewidth]{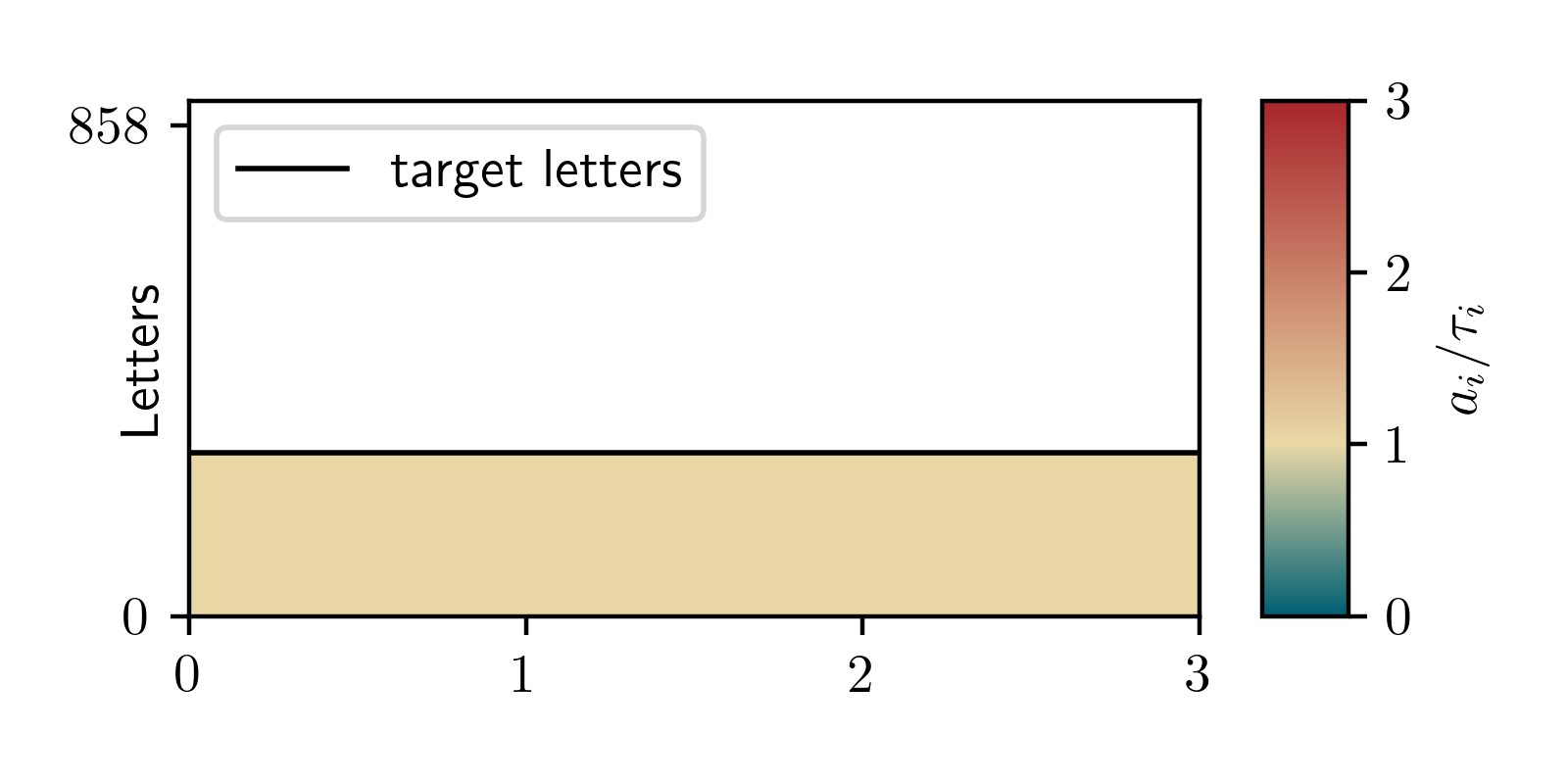}
        \caption{\greq ($t_G = 3$)}
        \label{fig:results_Baden-Württemberg_Medium_greedy_equal}
    \end{subfigure}
    \begin{subfigure}{0.32\textwidth}
        \includegraphics[draft=\draft, width=\linewidth]{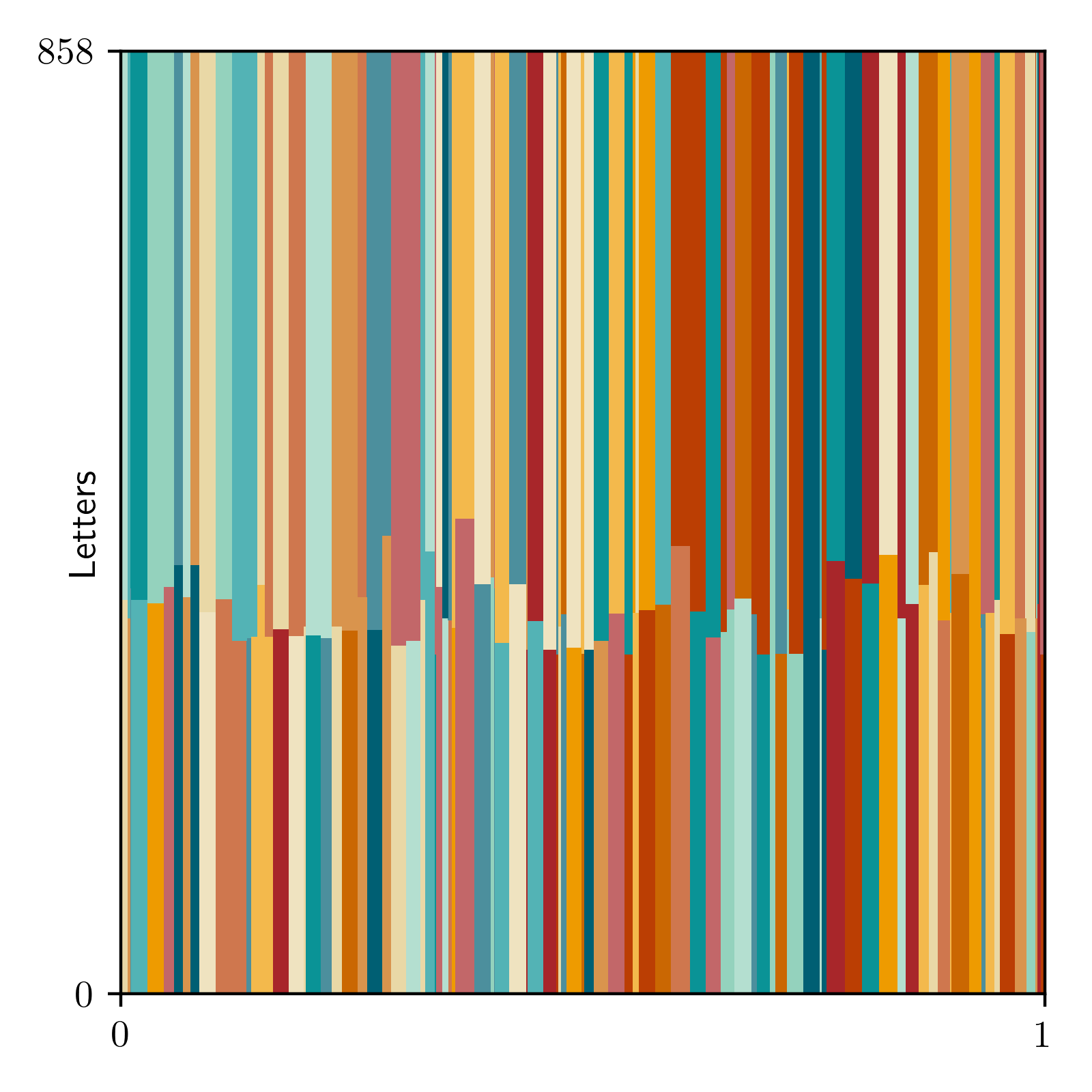}
        \includegraphics[draft=\draft, width=\linewidth]{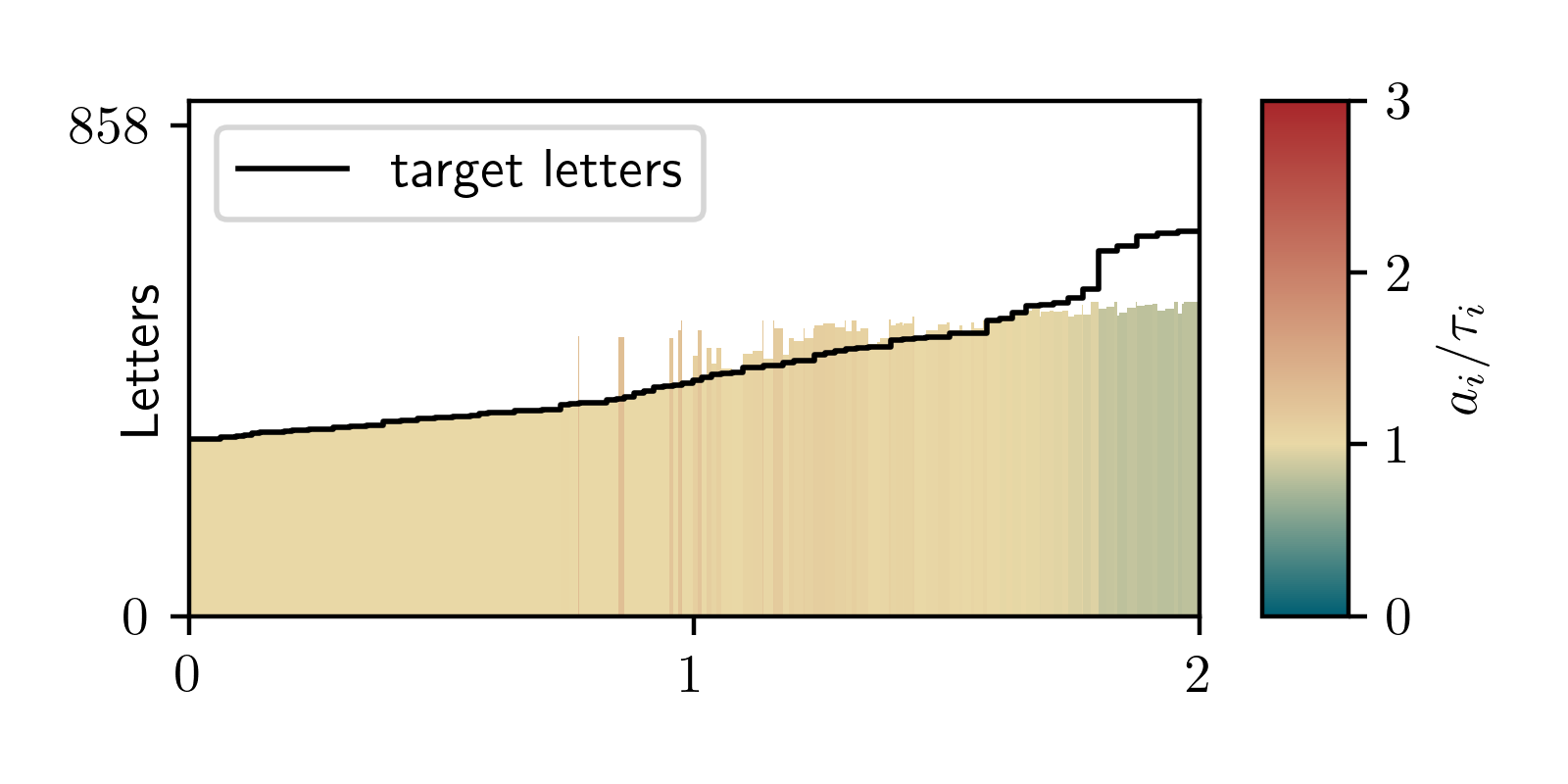}
        \caption{\colgen ($t_G\!=\!2$)}
        \label{fig:results_Baden-Württemberg_Medium_column_generation}
    \end{subfigure}
    \begin{subfigure}{0.32\textwidth}
        \includegraphics[draft=\draft, width=\linewidth]{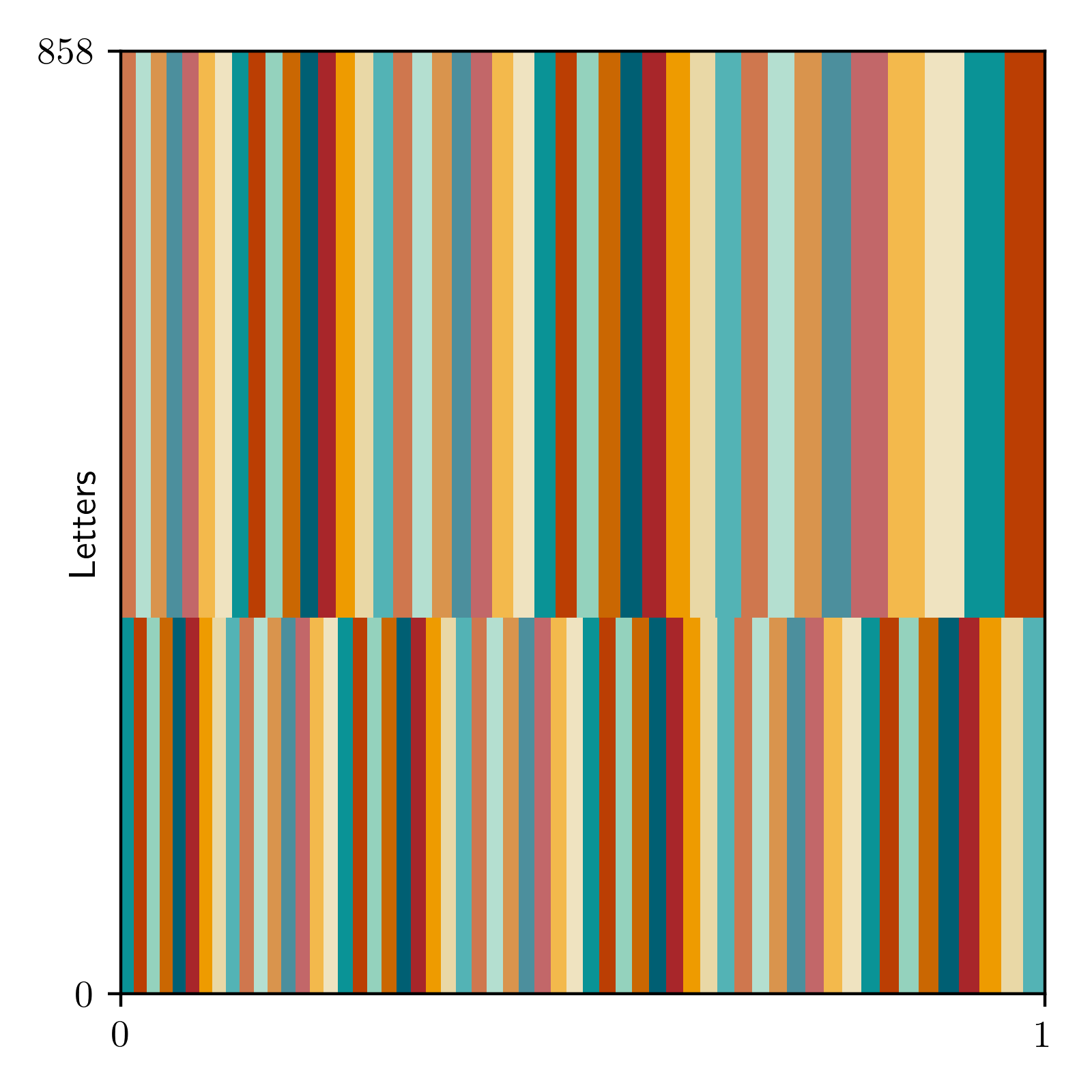}
        \includegraphics[draft=\draft, width=\linewidth]{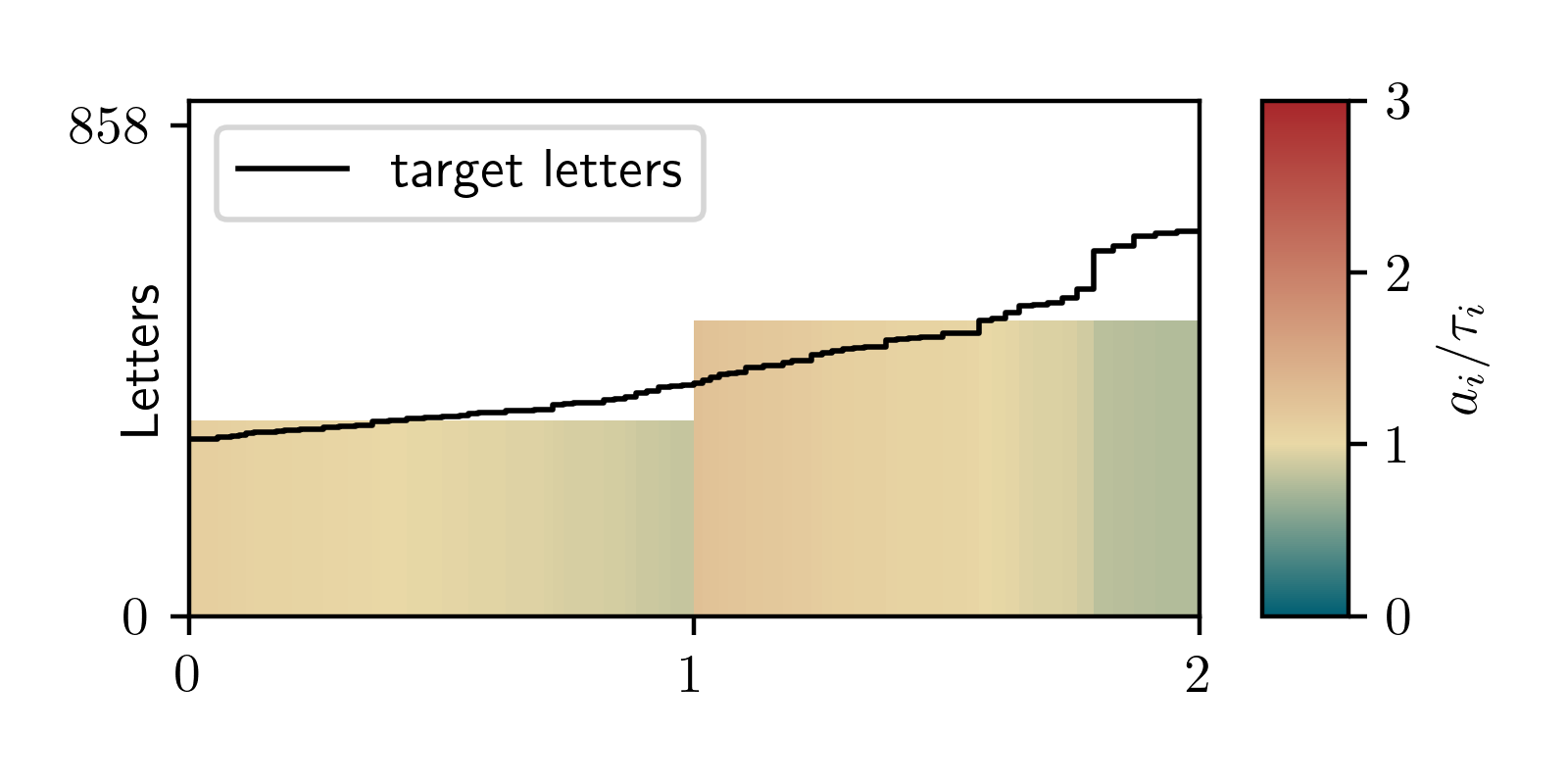}
        \caption{\buckets ($t_G = 2$)}
        \label{fig:results_Baden-Württemberg_Medium_greedy_bucket_fill}
    \end{subfigure}
    \caption{Medium municipalities of Baden-Württemberg ($\ell_G = 858$)}
    \label{fig:results_Baden-Württemberg_Medium}
\end{figure} 

\begin{figure}
    \centering
    \begin{subfigure}{0.32\textwidth}
        \includegraphics[draft=\draft, width=\linewidth]{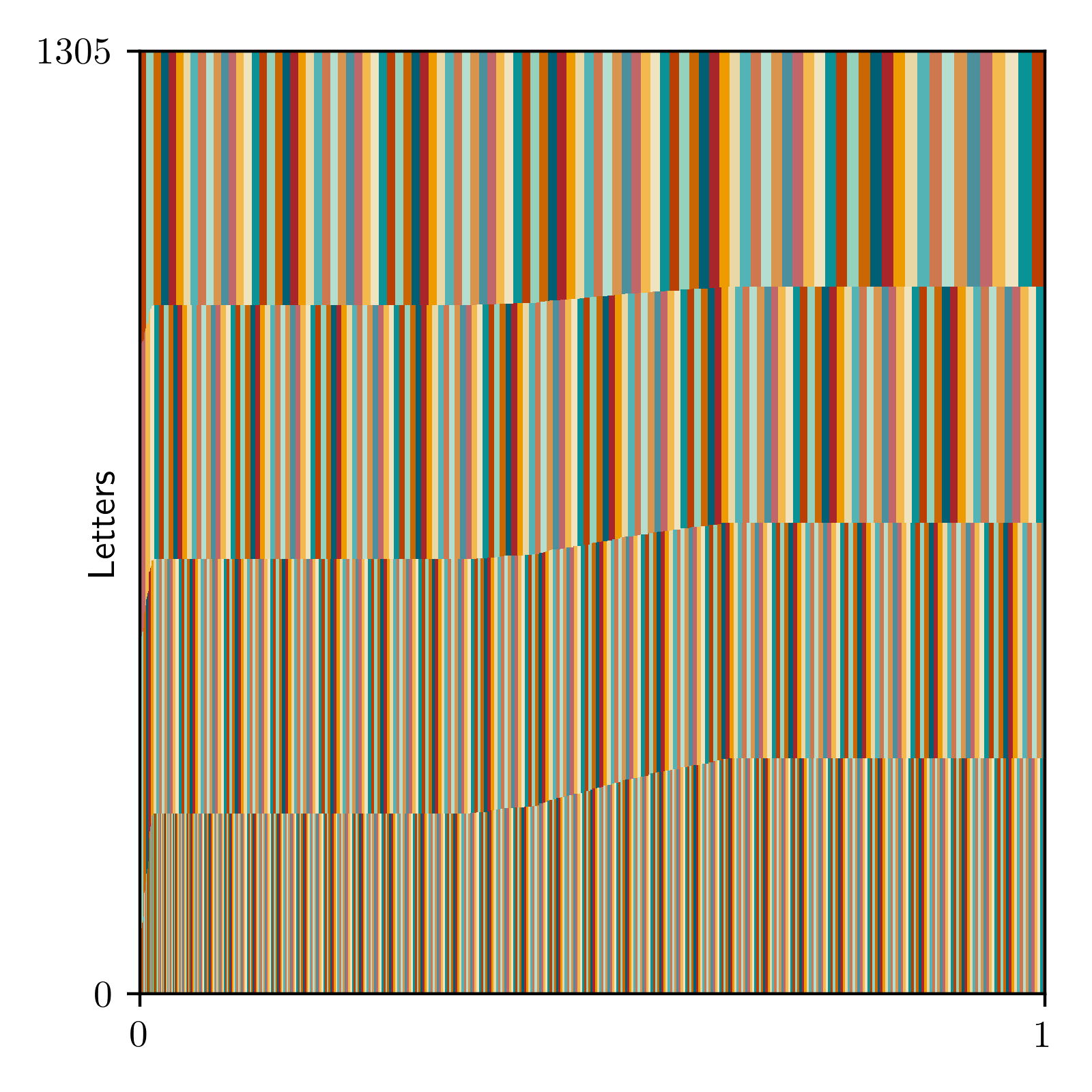}
        \includegraphics[draft=\draft, width=\linewidth]{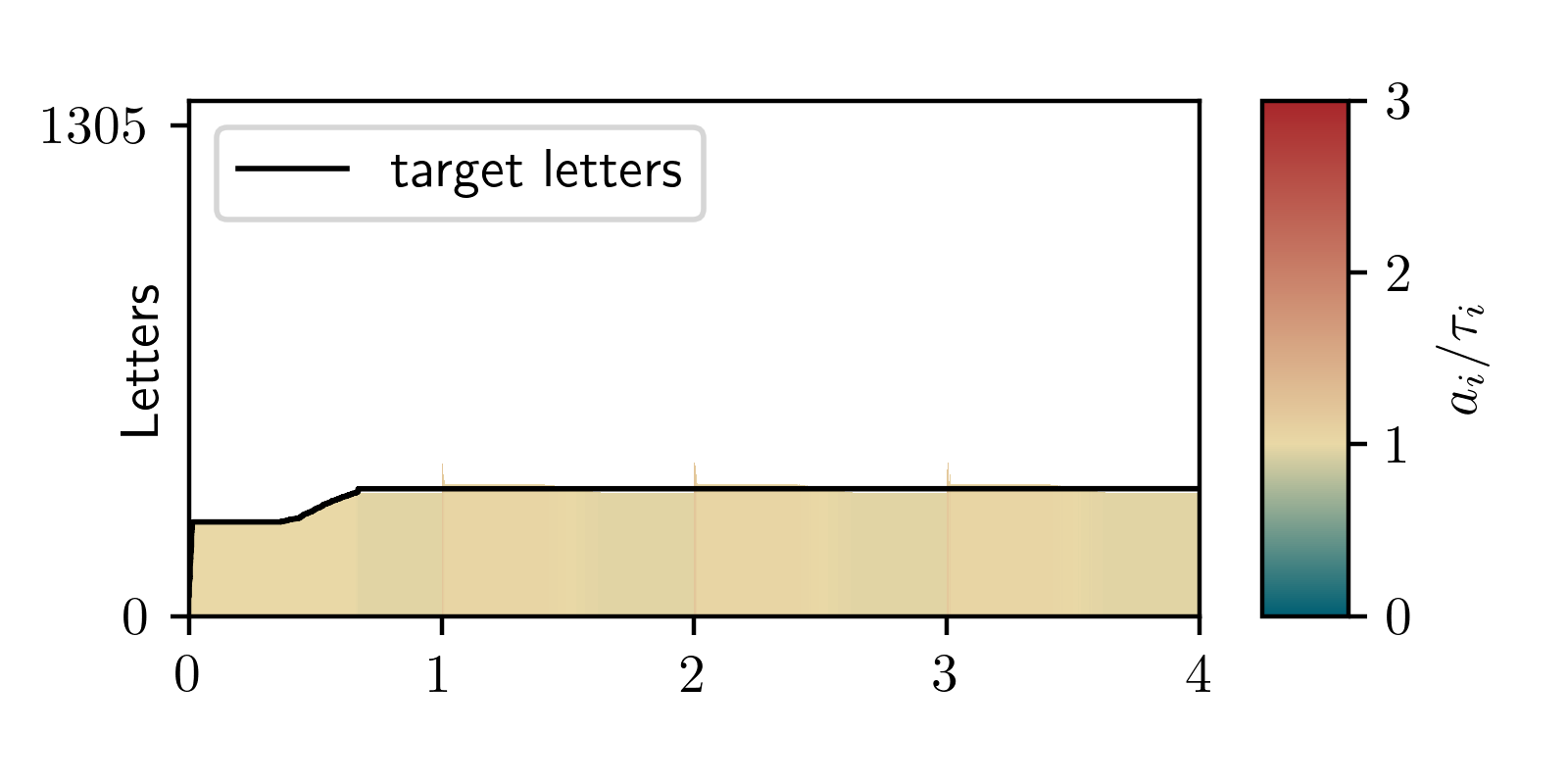}
        \caption{\greq ($t_G = 4$)}
        \label{fig:results_Baden-Württemberg_Small_greedy_equal}
    \end{subfigure}
    \begin{subfigure}{0.32\textwidth}
        \includegraphics[draft=\draft, width=\linewidth]{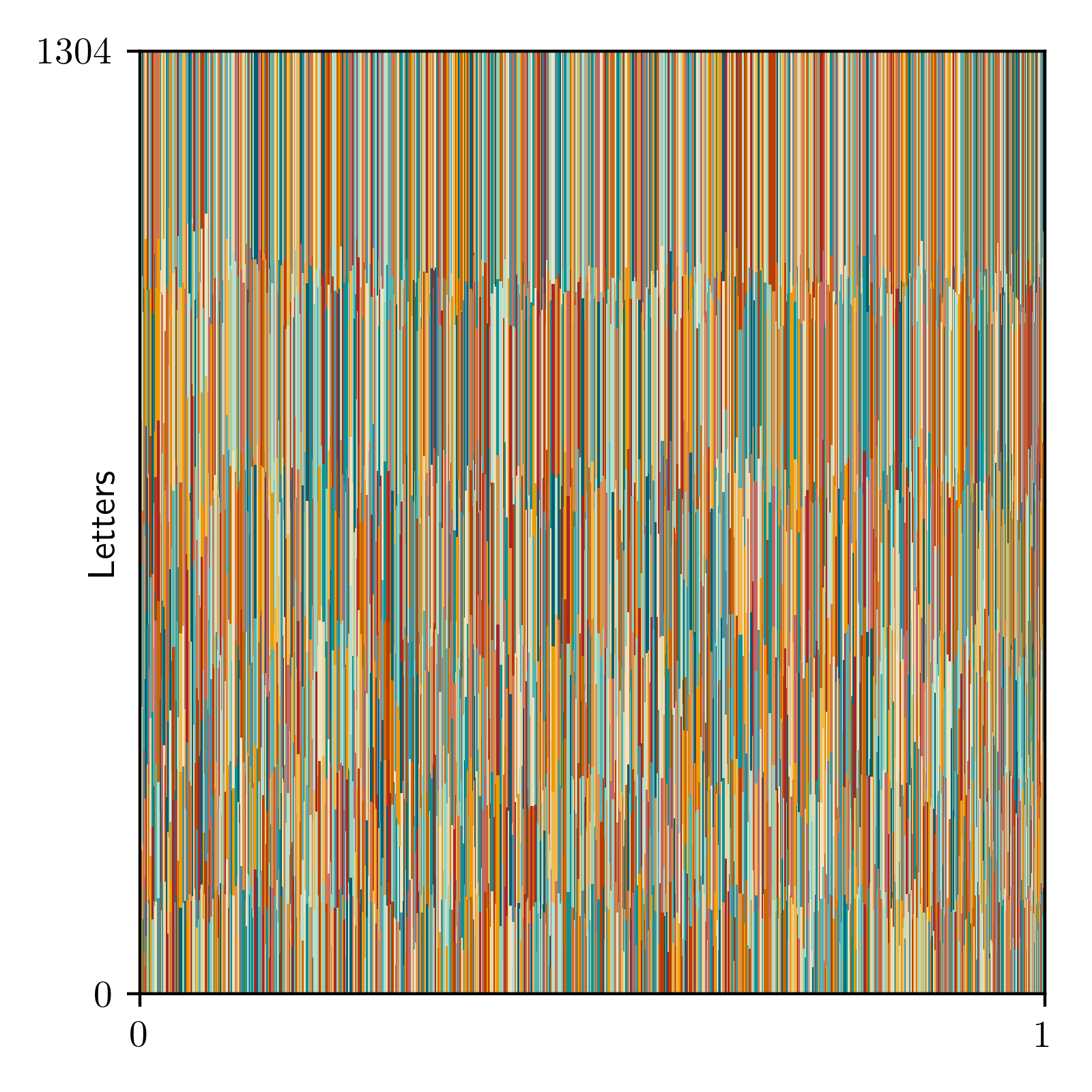}
        \includegraphics[draft=\draft, width=\linewidth]{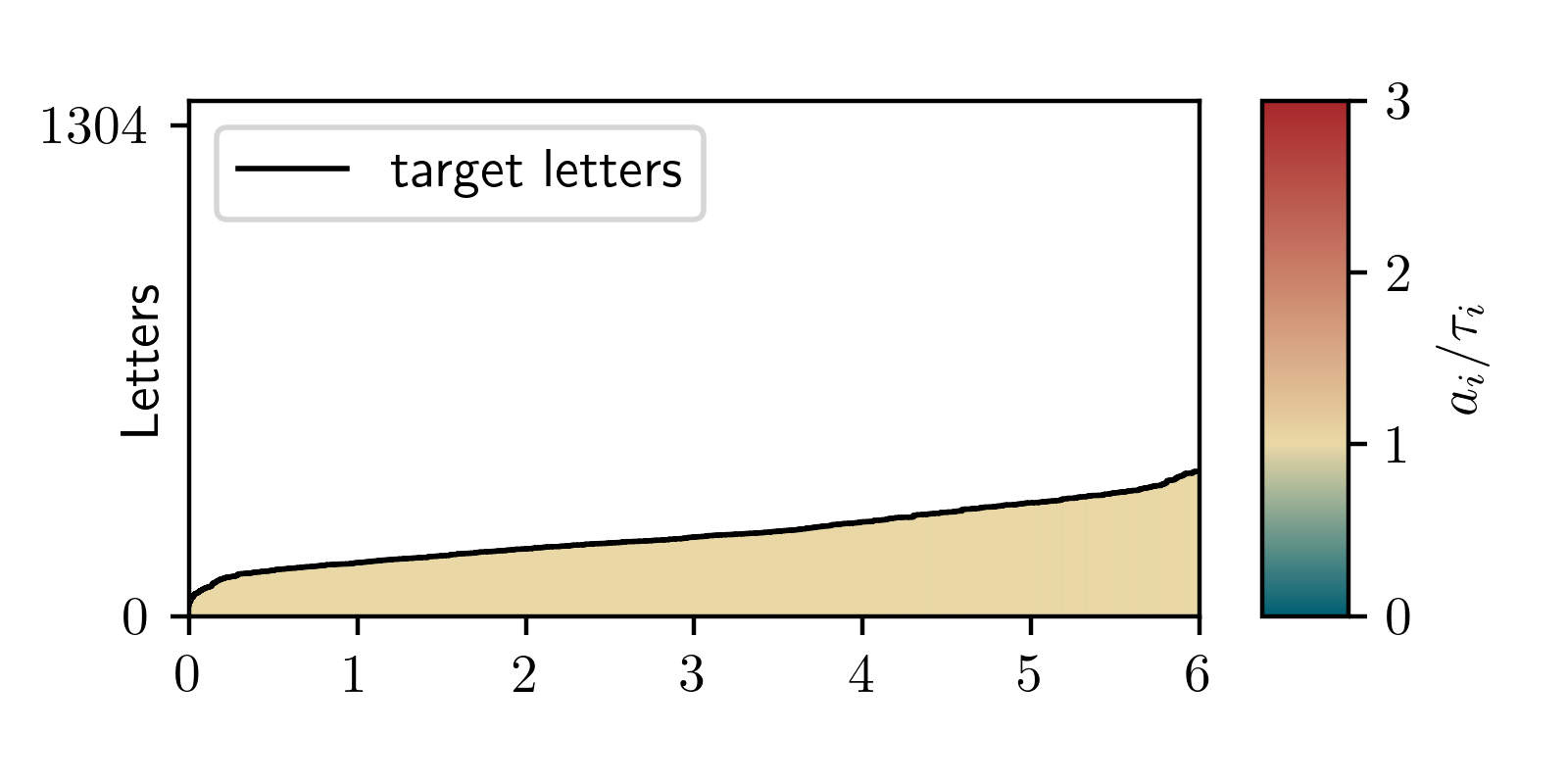}
        \caption{\colgen ($t_G\!=\!6$)}
        \label{fig:results_Baden-Württemberg_Small_column_generation}
    \end{subfigure}
    \begin{subfigure}{0.32\textwidth}
        \includegraphics[draft=\draft, width=\linewidth]{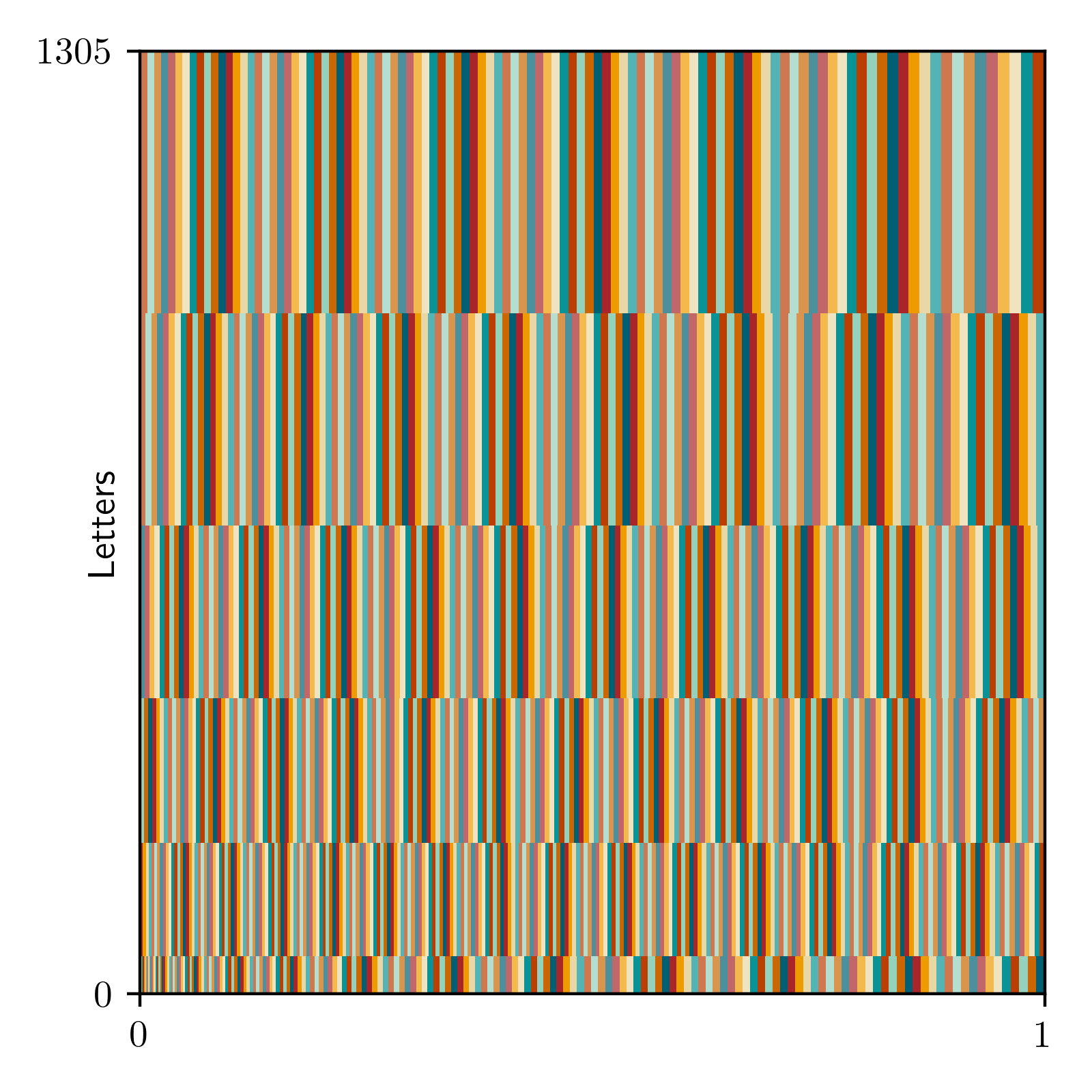}
        \includegraphics[draft=\draft, width=\linewidth]{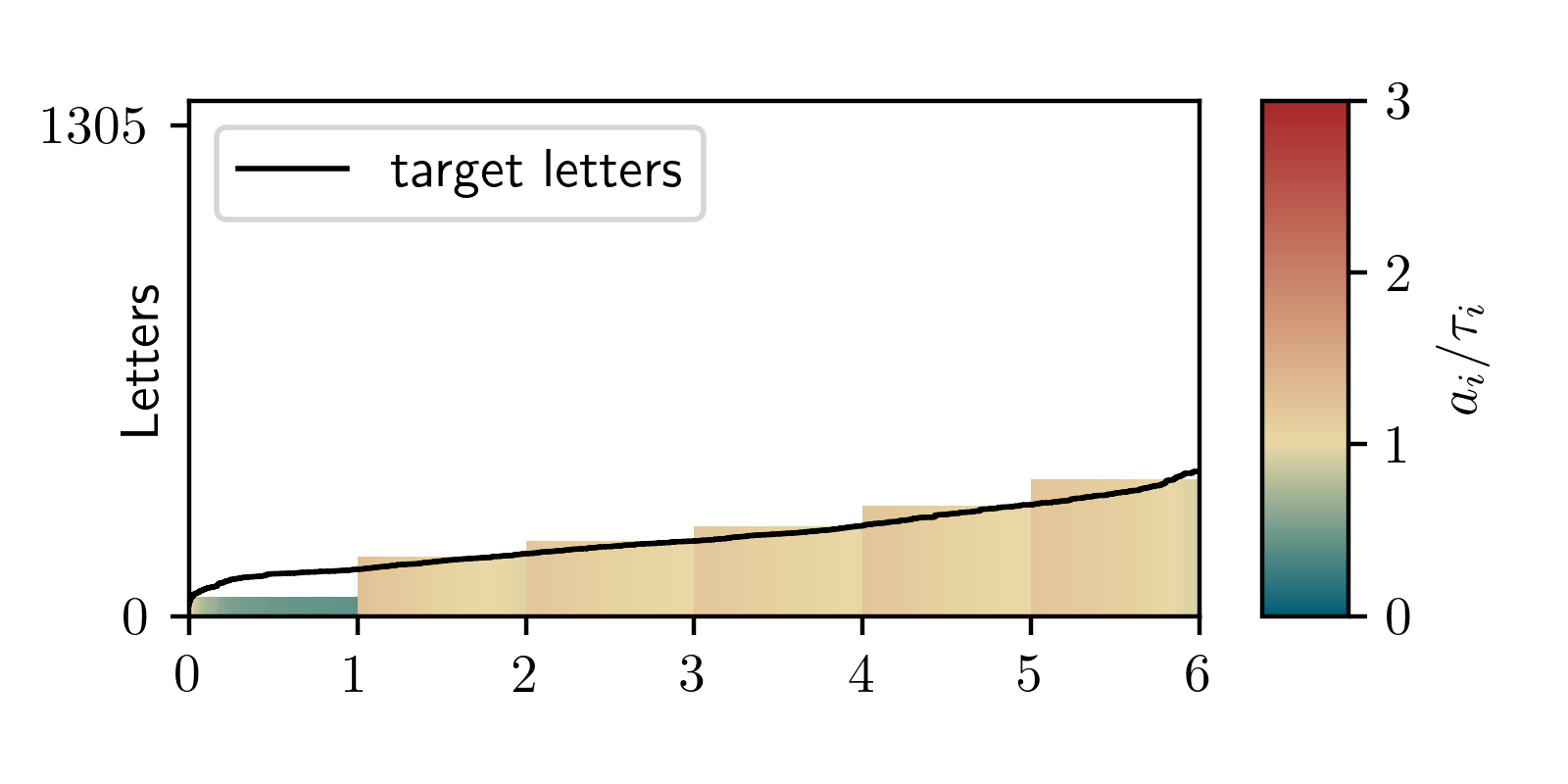}
        \caption{\buckets ($t_G = 6$)}
        \label{fig:results_Baden-Württemberg_Small_greedy_bucket_fill}
    \end{subfigure}
    \caption{Small municipalities of Baden-Württemberg ($\ell_G = 1305$)}
    \label{fig:results_Baden-Württemberg_Small}
\end{figure} 

\begin{figure}
    \centering
    \begin{subfigure}{0.32\textwidth}
        \includegraphics[draft=\draft, width=\linewidth]{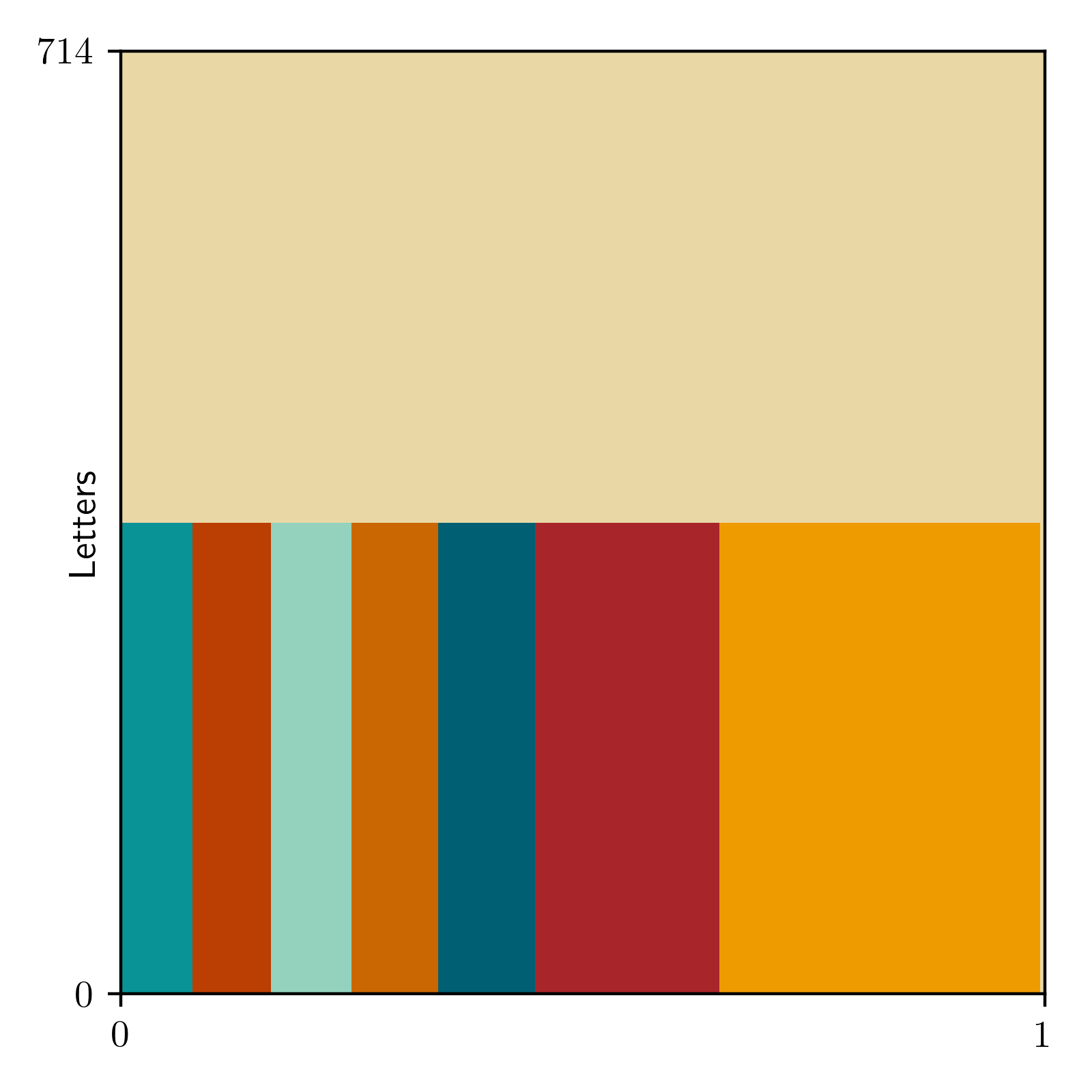}
        \includegraphics[draft=\draft, width=\linewidth]{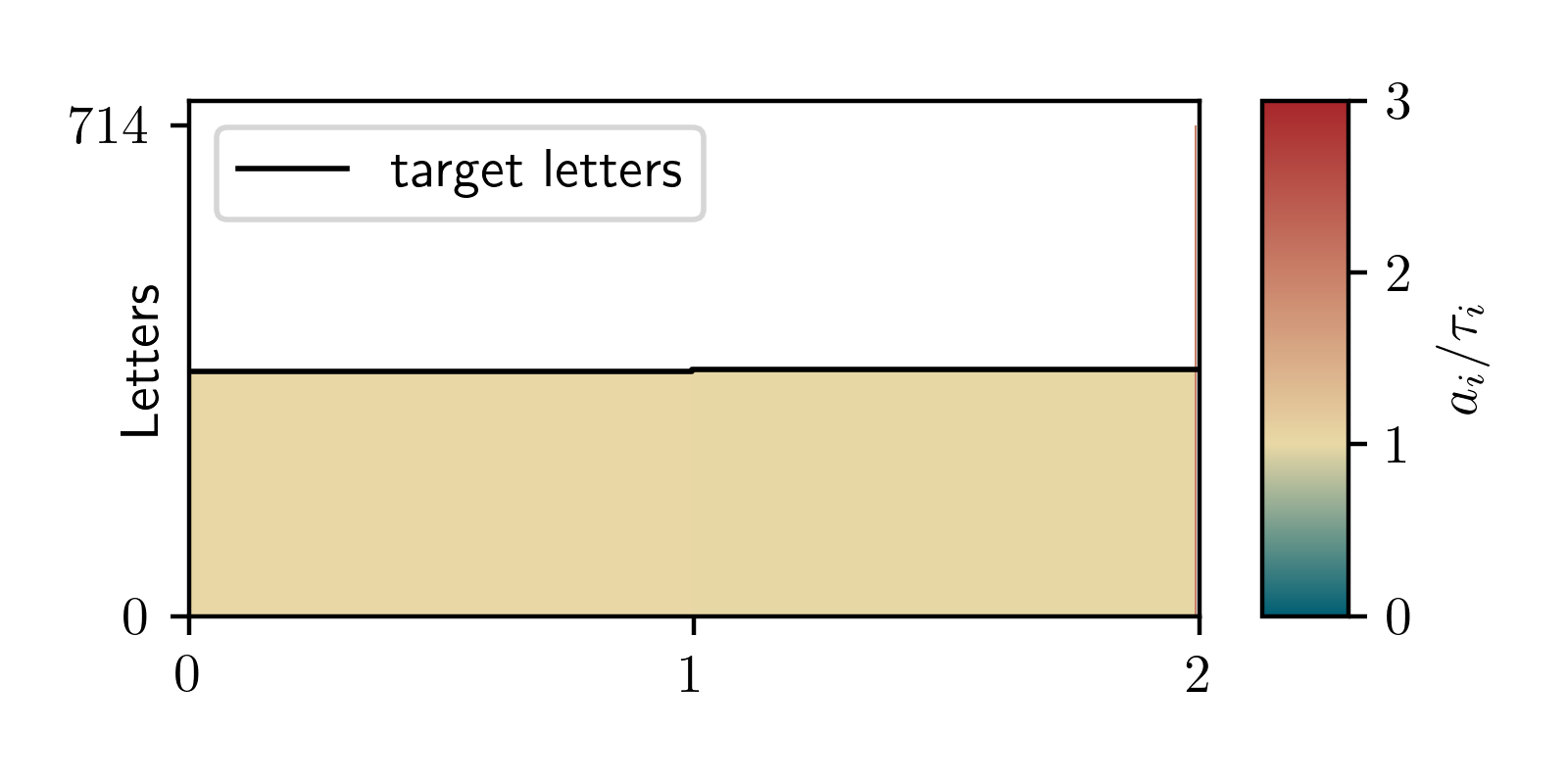}
        \caption{\greq ($t_G = 2$)}
        \label{fig:results_Bayern_Large_greedy_equal}
    \end{subfigure}
    \begin{subfigure}{0.32\textwidth}
        \includegraphics[draft=\draft, width=\linewidth]{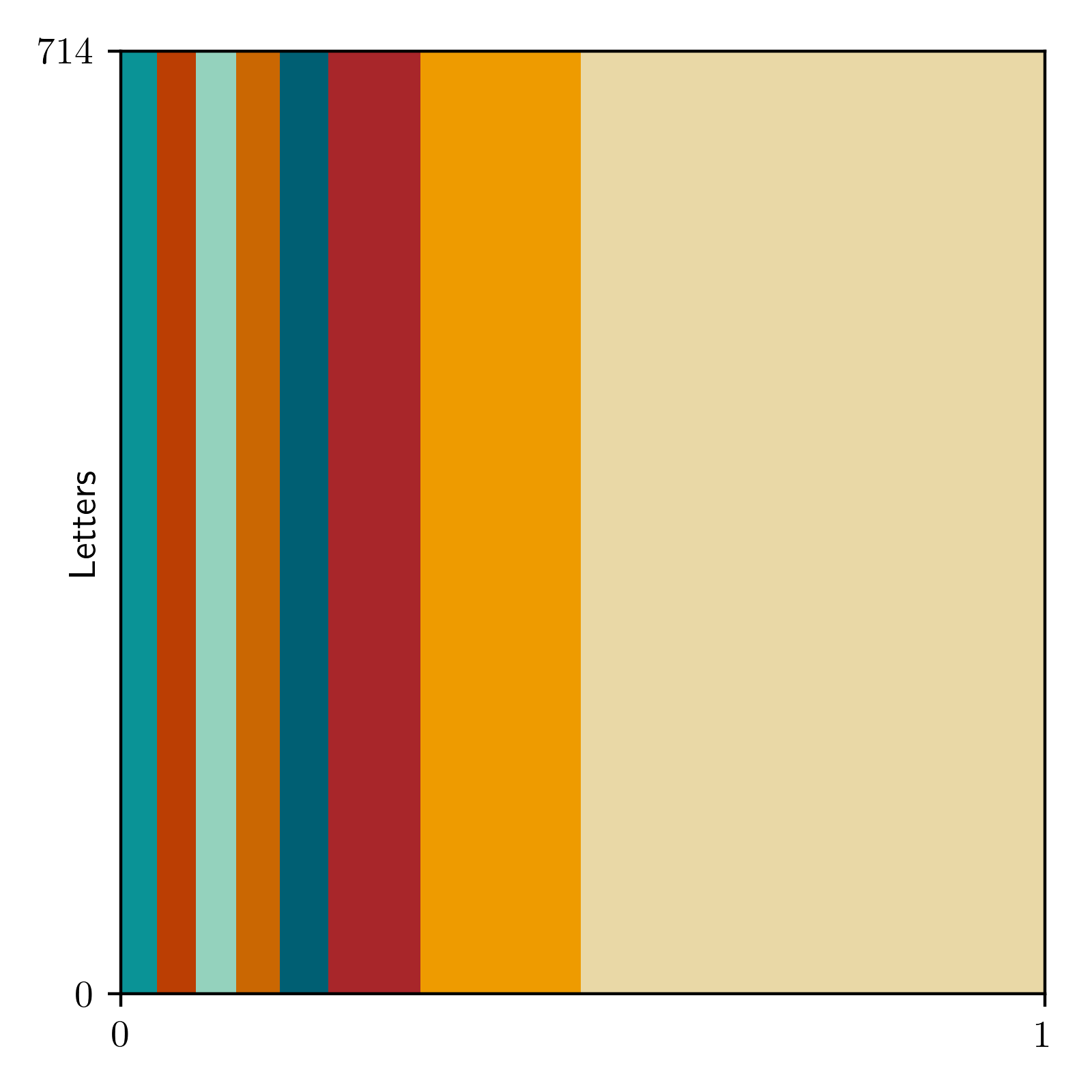}
        \includegraphics[draft=\draft, width=\linewidth]{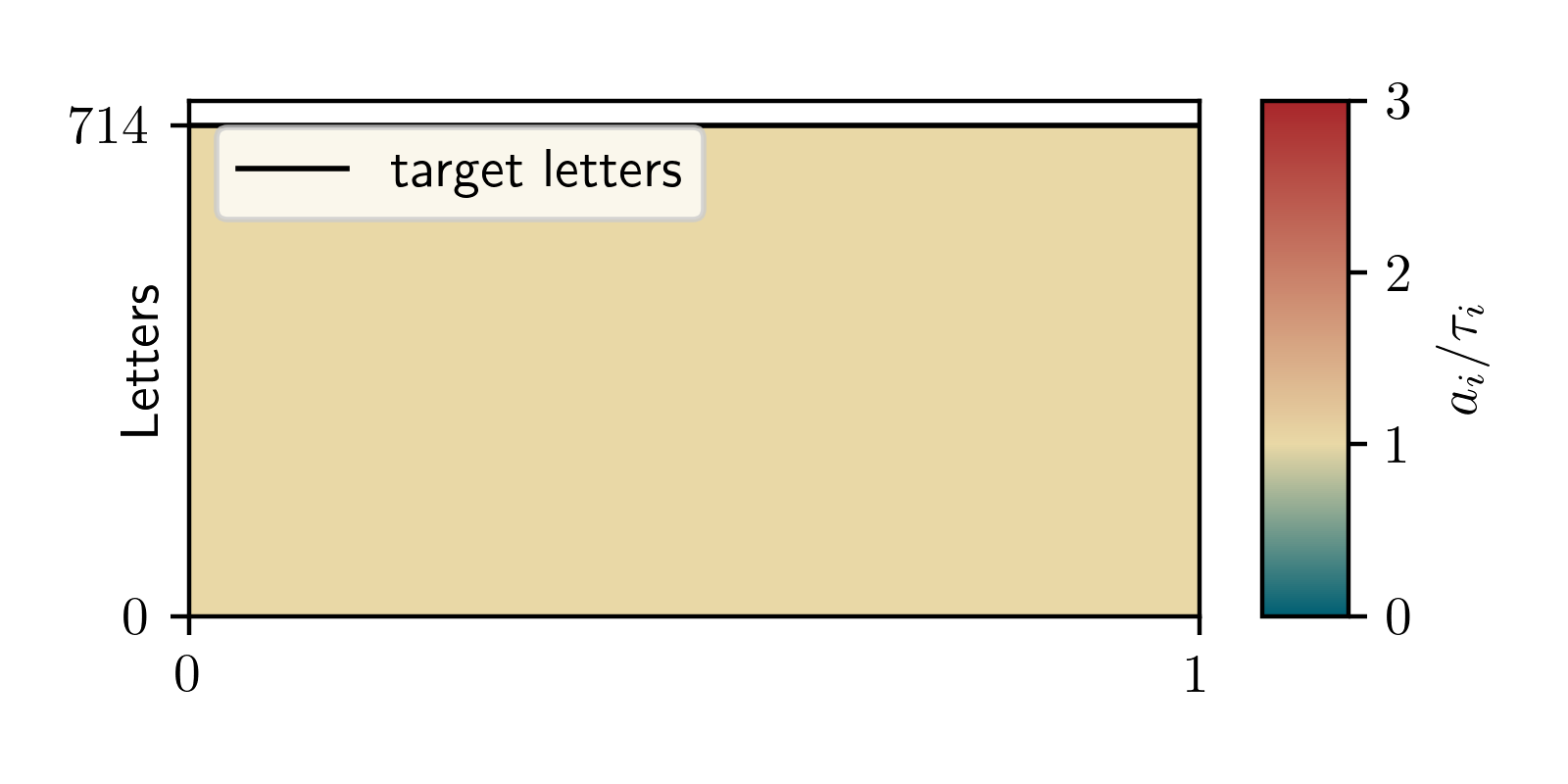}
        \caption{\colgen ($t_G\!=\!1$)}
        \label{fig:results_Bayern_Large_column_generation}
    \end{subfigure}
    \begin{subfigure}{0.32\textwidth}
        \includegraphics[draft=\draft, width=\linewidth]{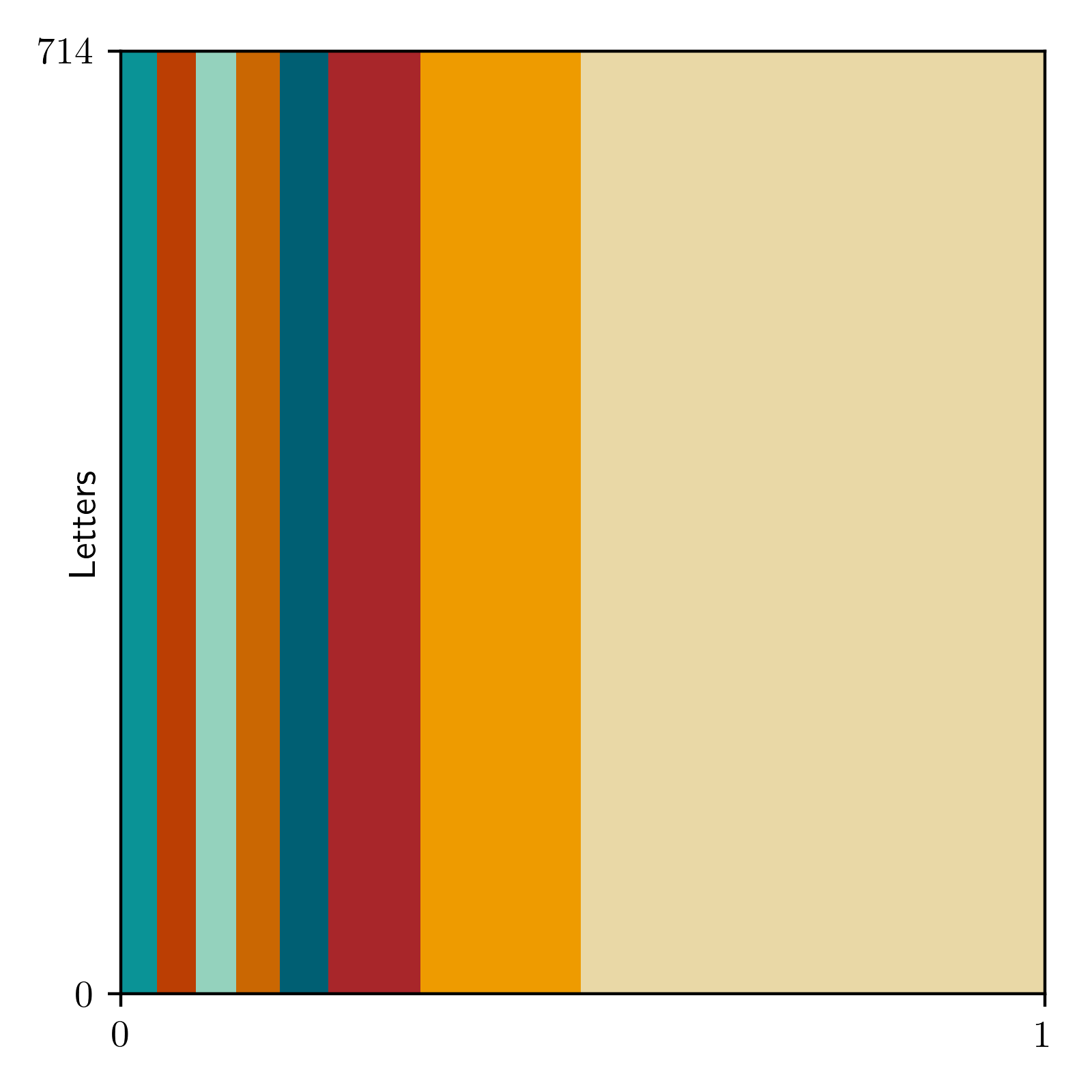}
        \includegraphics[draft=\draft, width=\linewidth]{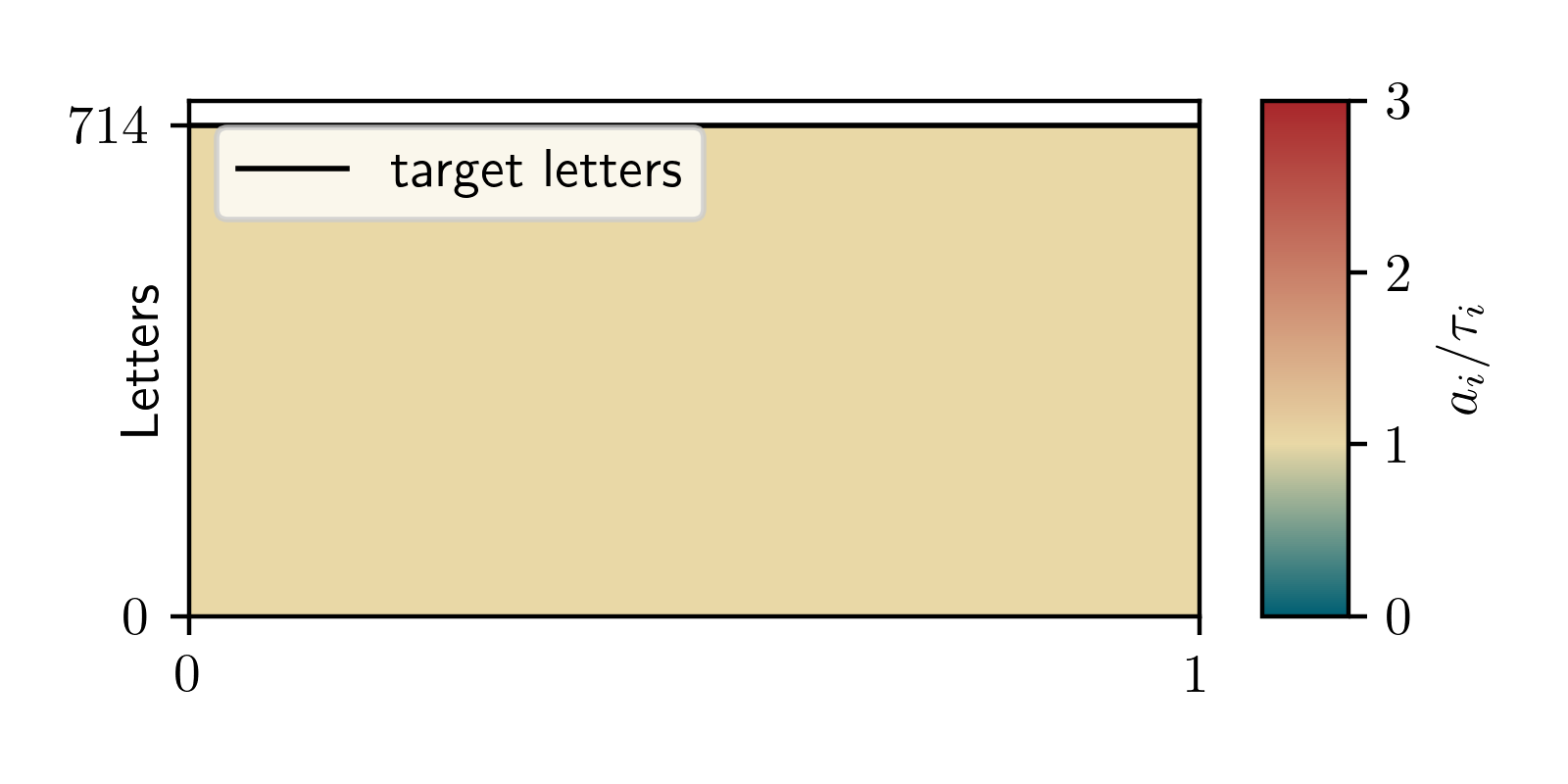}
        \caption{\buckets ($t_G = 1$)}
        \label{fig:results_Bayern_Large_greedy_bucket_fill}
    \end{subfigure}
    \caption{Large municipalities of Bayern ($\ell_G = 714$)}
    \label{fig:results_Bayern_Large}
\end{figure} 

\begin{figure}
    \centering
    \begin{subfigure}{0.32\textwidth}
        \includegraphics[draft=\draft, width=\linewidth]{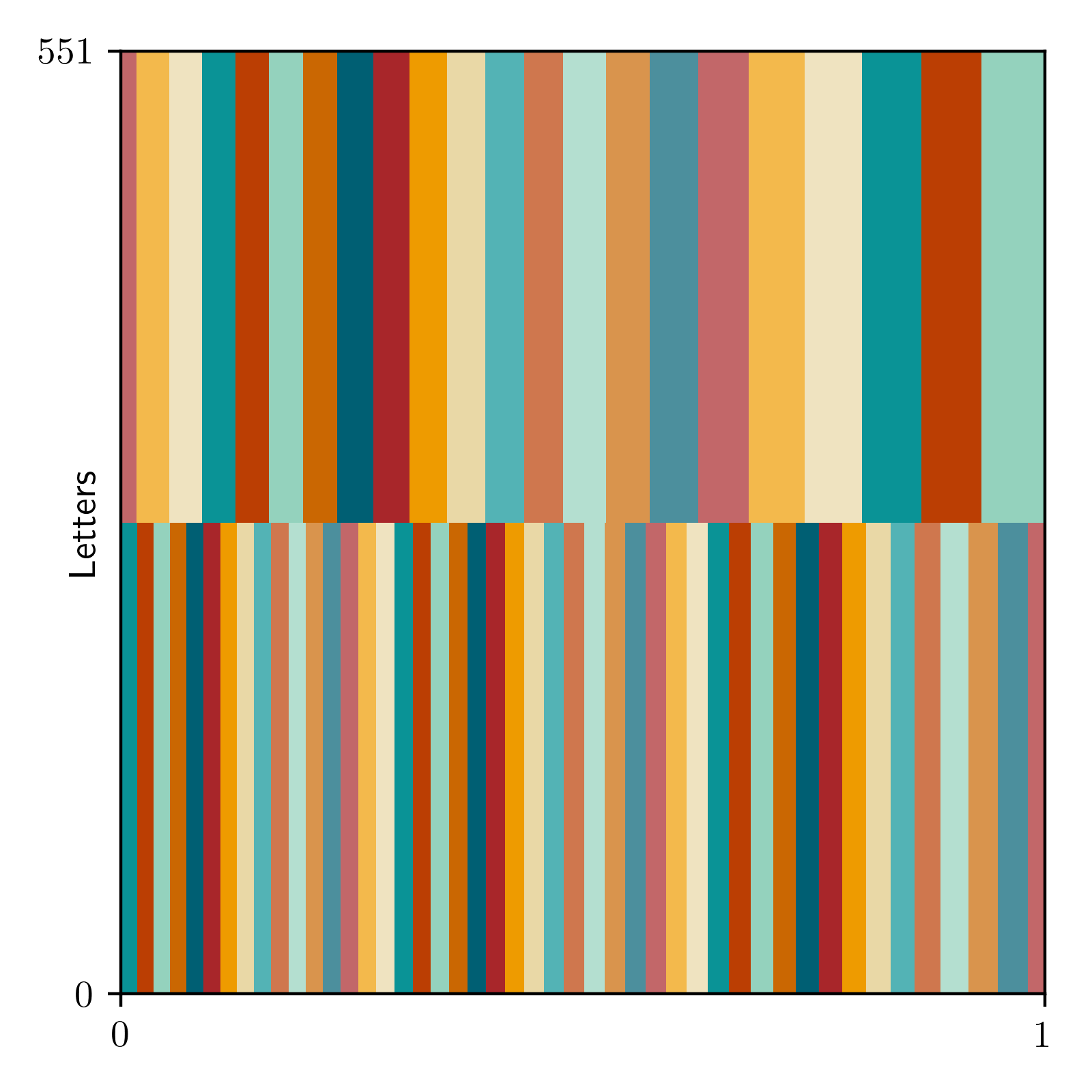}
        \includegraphics[draft=\draft, width=\linewidth]{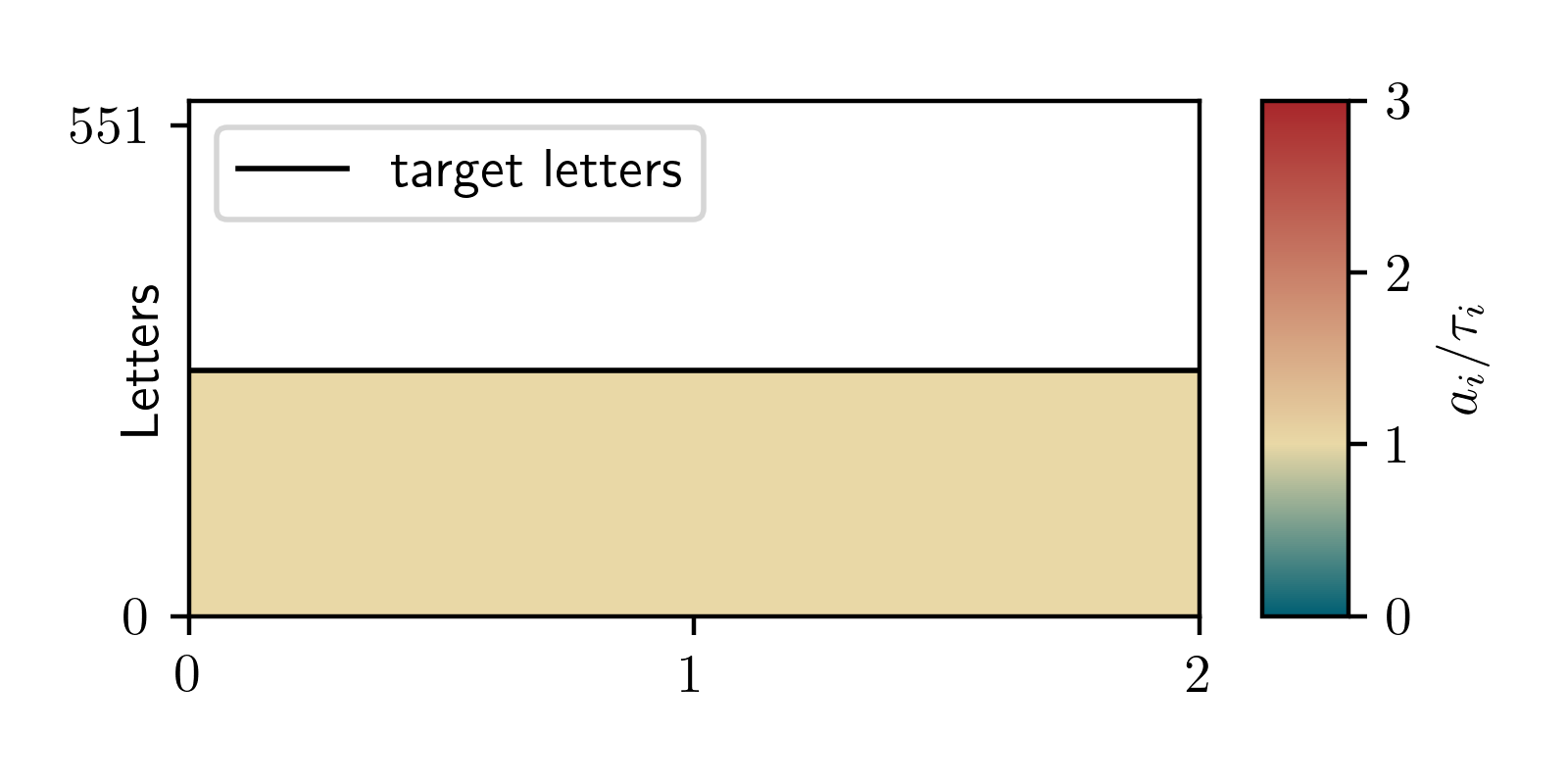}
        \caption{\greq ($t_G = 2$)}
        \label{fig:results_Bayern_Medium_greedy_equal}
    \end{subfigure}
    \begin{subfigure}{0.32\textwidth}
        \includegraphics[draft=\draft, width=\linewidth]{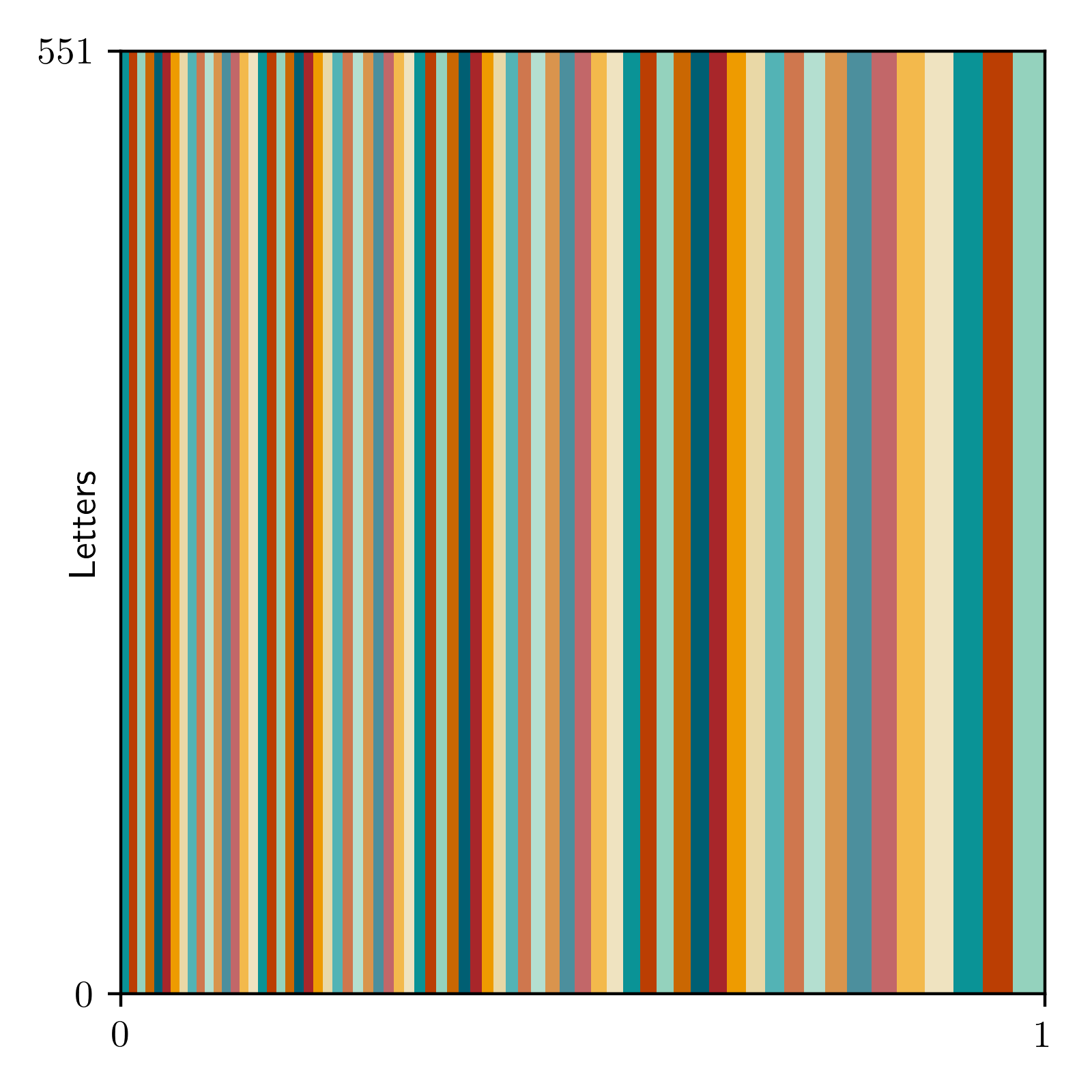}
        \includegraphics[draft=\draft, width=\linewidth]{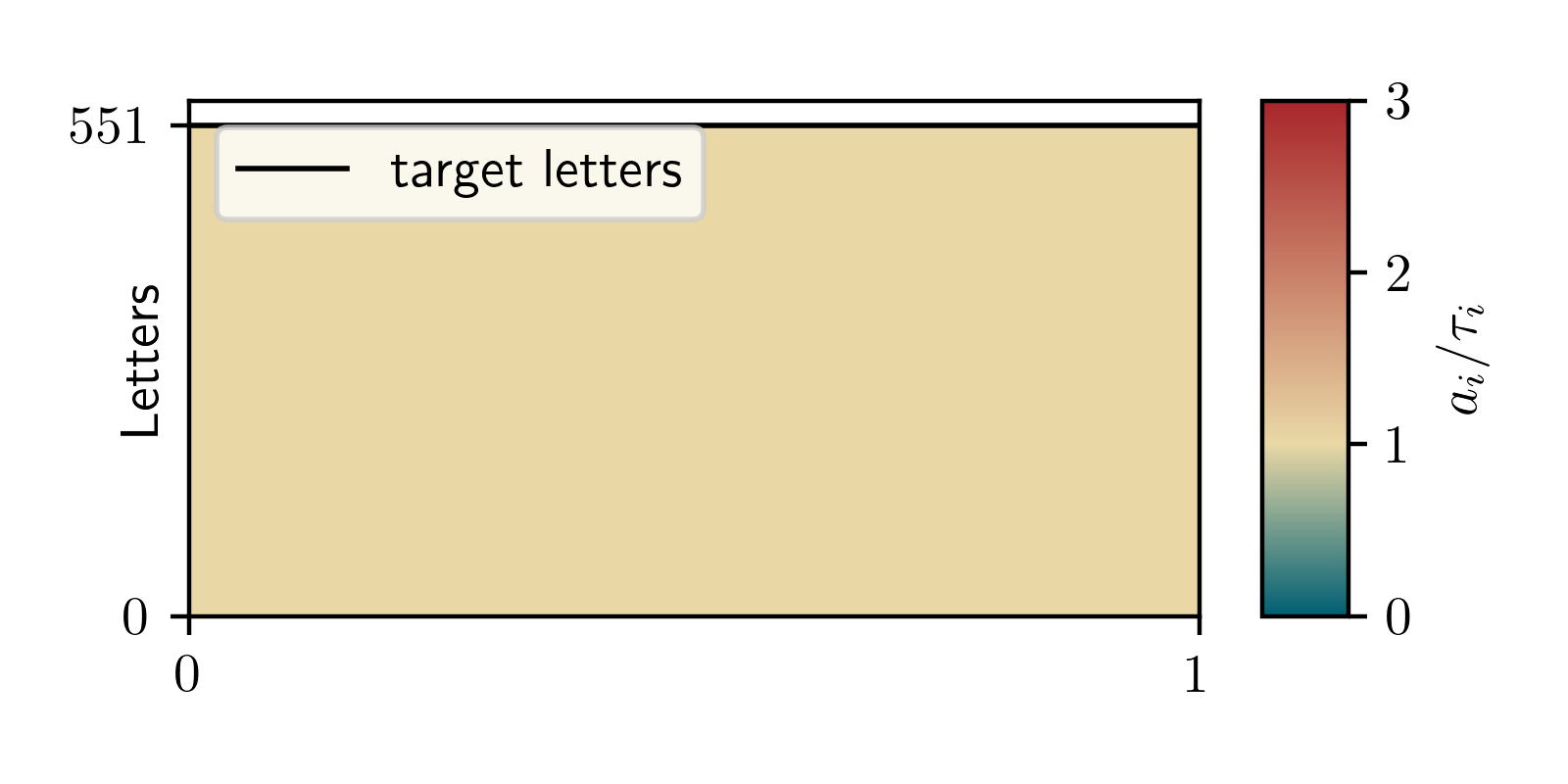}
        \caption{\colgen ($t_G\!=\!1$)}
        \label{fig:results_Bayern_Medium_column_generation}
    \end{subfigure}
    \begin{subfigure}{0.32\textwidth}
        \includegraphics[draft=\draft, width=\linewidth]{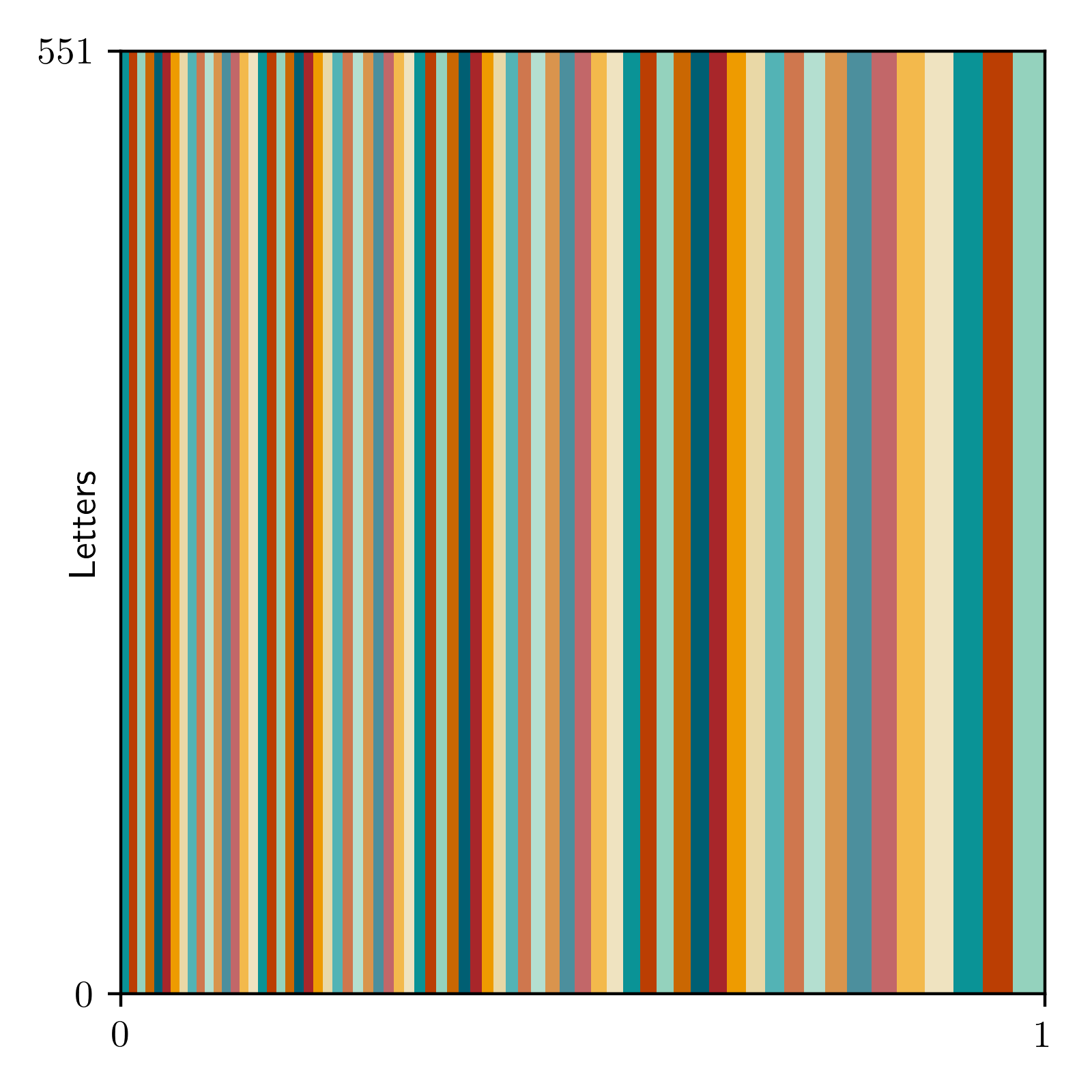}
        \includegraphics[draft=\draft, width=\linewidth]{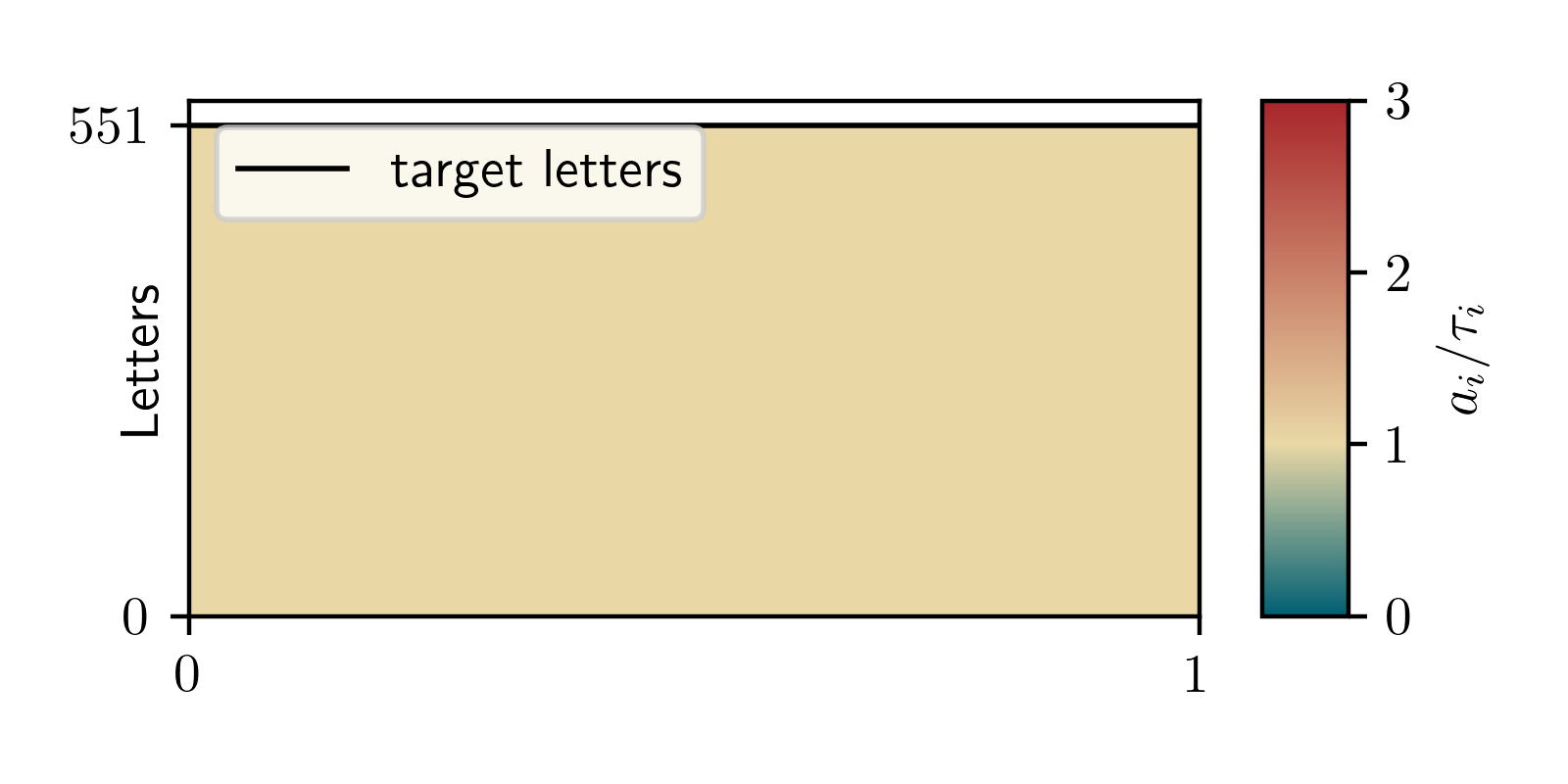}
        \caption{\buckets ($t_G = 1$)}
        \label{fig:results_Bayern_Medium_greedy_bucket_fill}
    \end{subfigure}
    \caption{Medium municipalities of Bayern ($\ell_G = 551$)}
    \label{fig:results_Bayern_Medium}
\end{figure} 

\begin{figure}
    \centering
    \begin{subfigure}{0.32\textwidth}
        \includegraphics[draft=\draft, width=\linewidth]{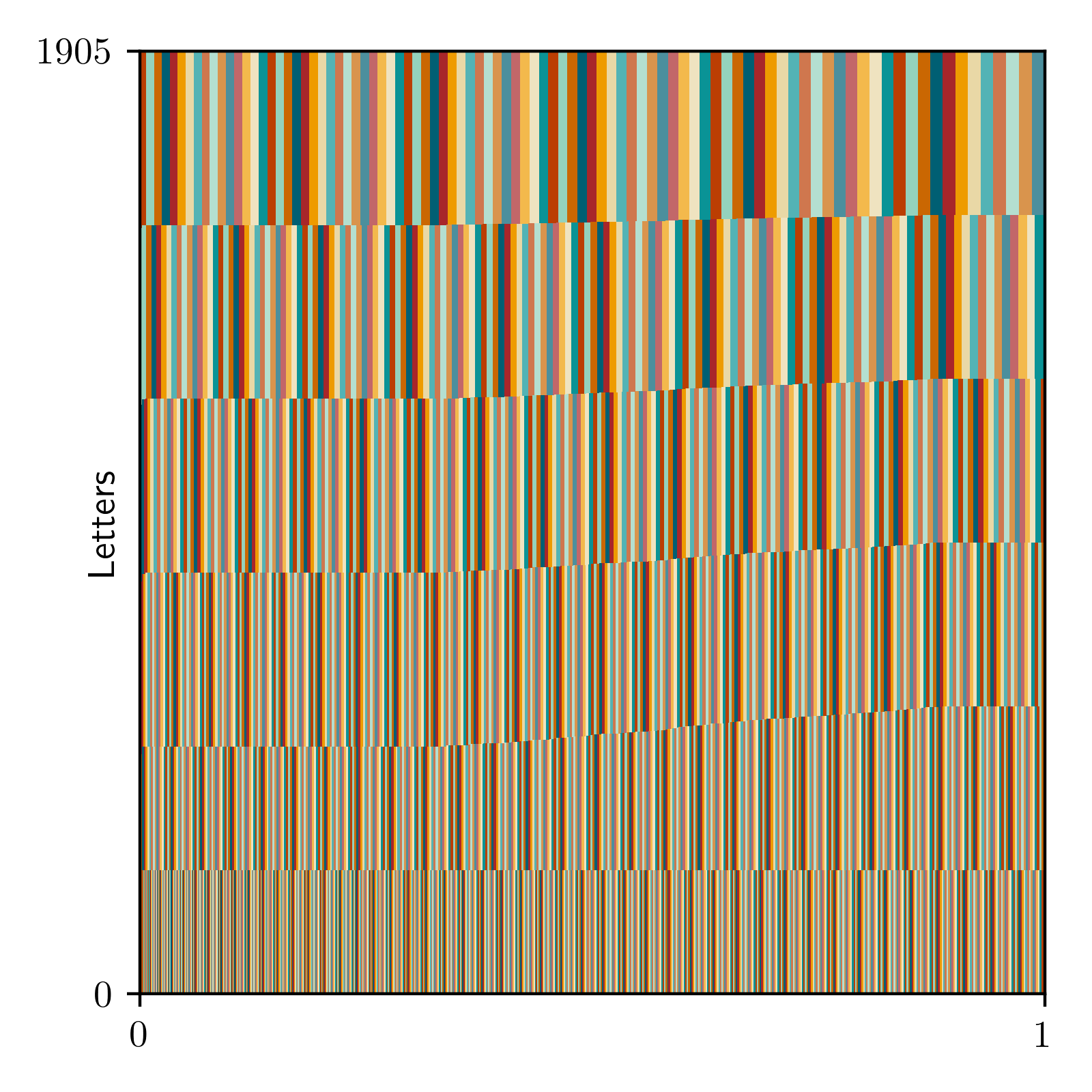}
        \includegraphics[draft=\draft, width=\linewidth]{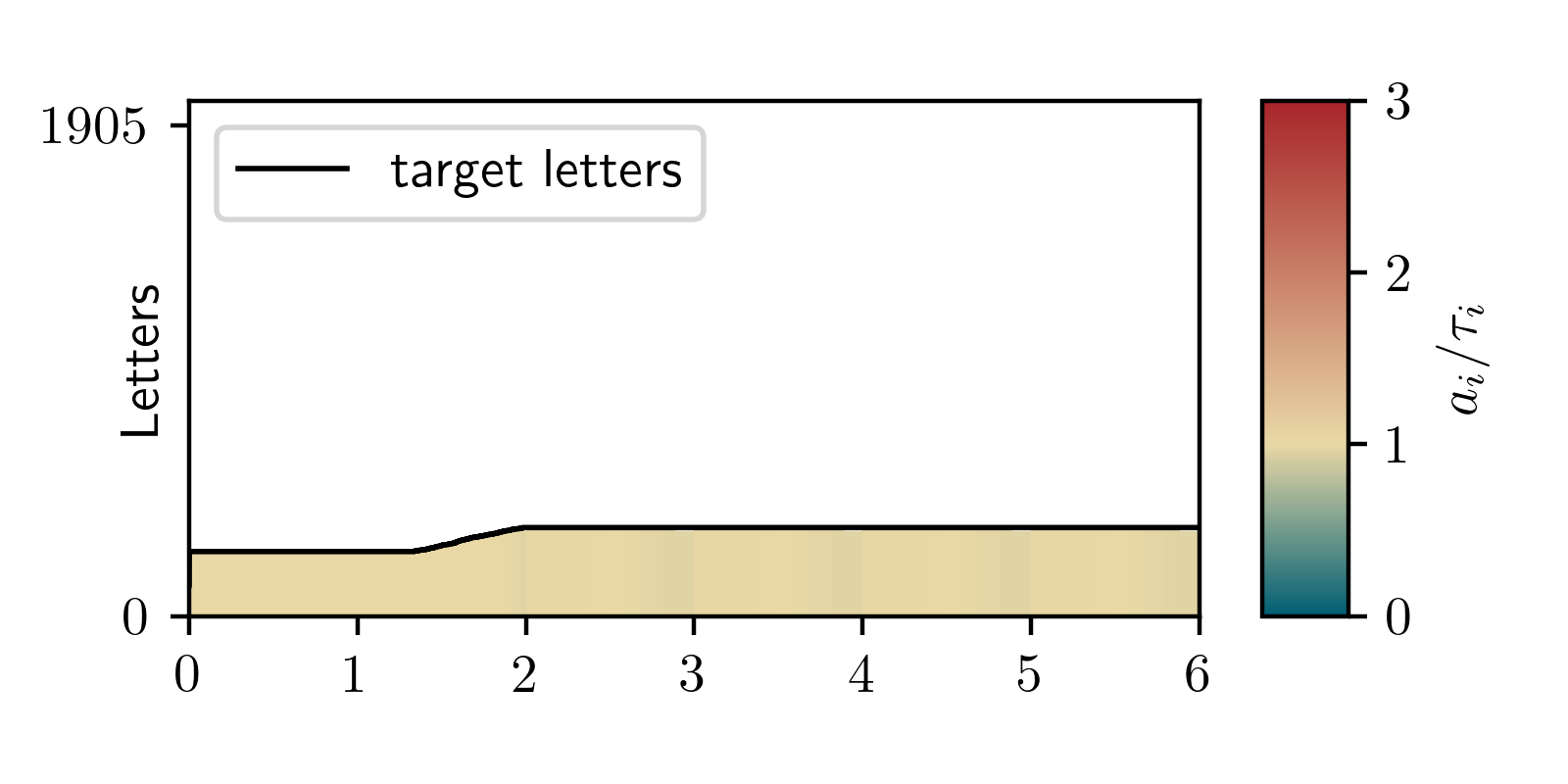}
        \caption{\greq ($t_G = 6$)}
        \label{fig:results_Bayern_Small_greedy_equal}
    \end{subfigure}
    \begin{subfigure}{0.32\textwidth}
        \includegraphics[draft=\draft, width=\linewidth]{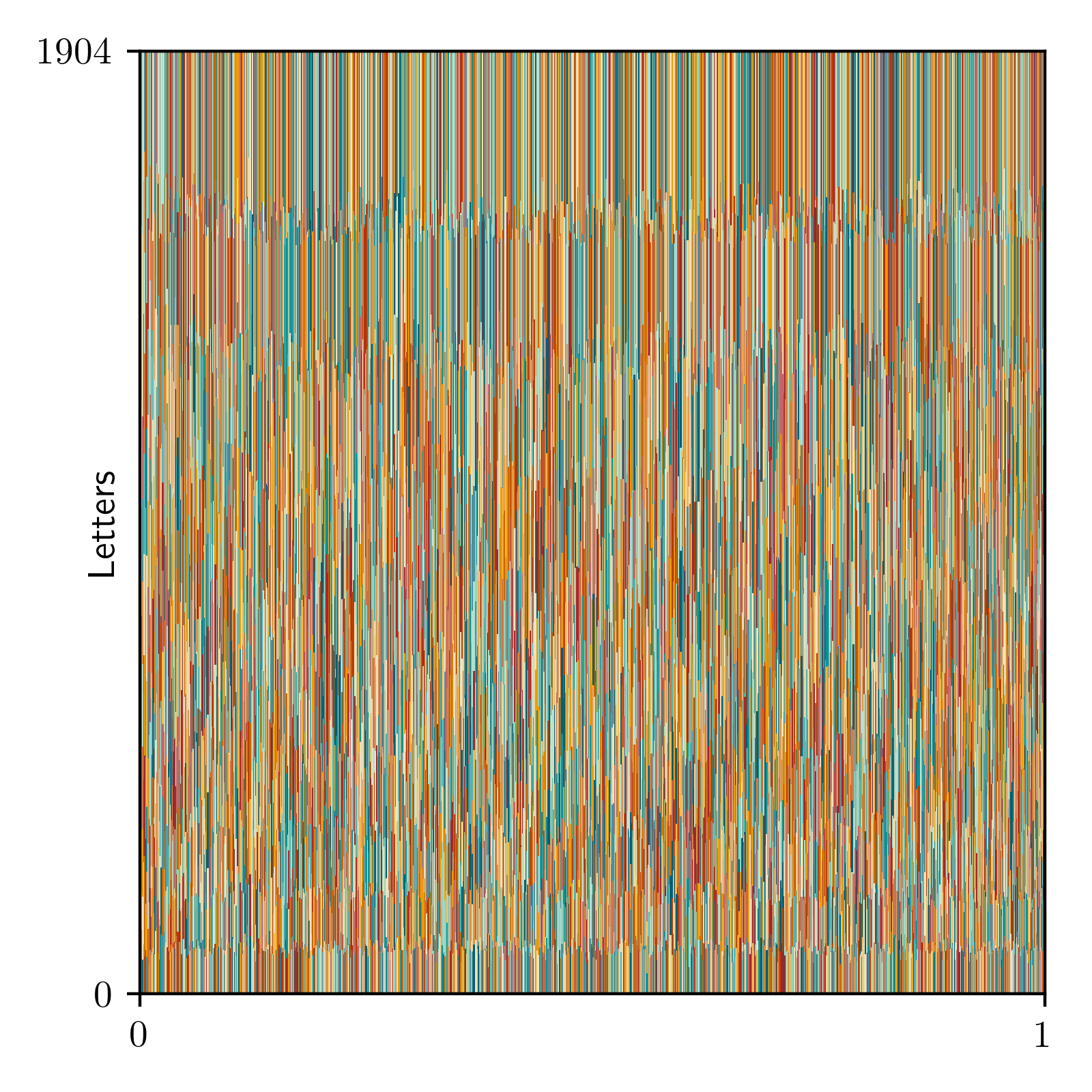}
        \includegraphics[draft=\draft, width=\linewidth]{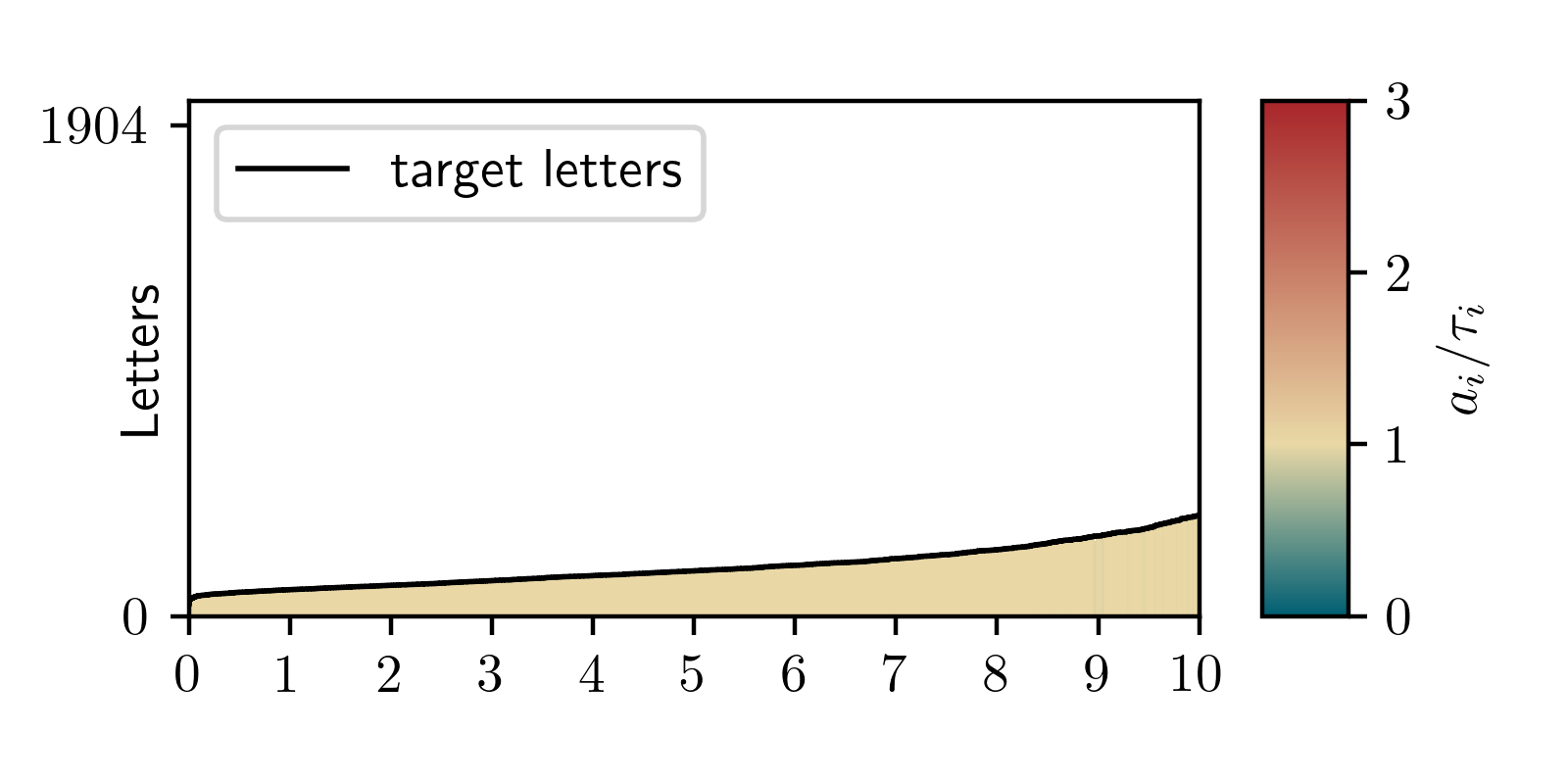}
        \caption{\colgen\!\hspace{-0.5mm}($t_G\!\!=\!\!10$)}
        \label{fig:results_Bayern_Small_column_generation}
    \end{subfigure}
    \begin{subfigure}{0.32\textwidth}
        \includegraphics[draft=\draft, width=\linewidth]{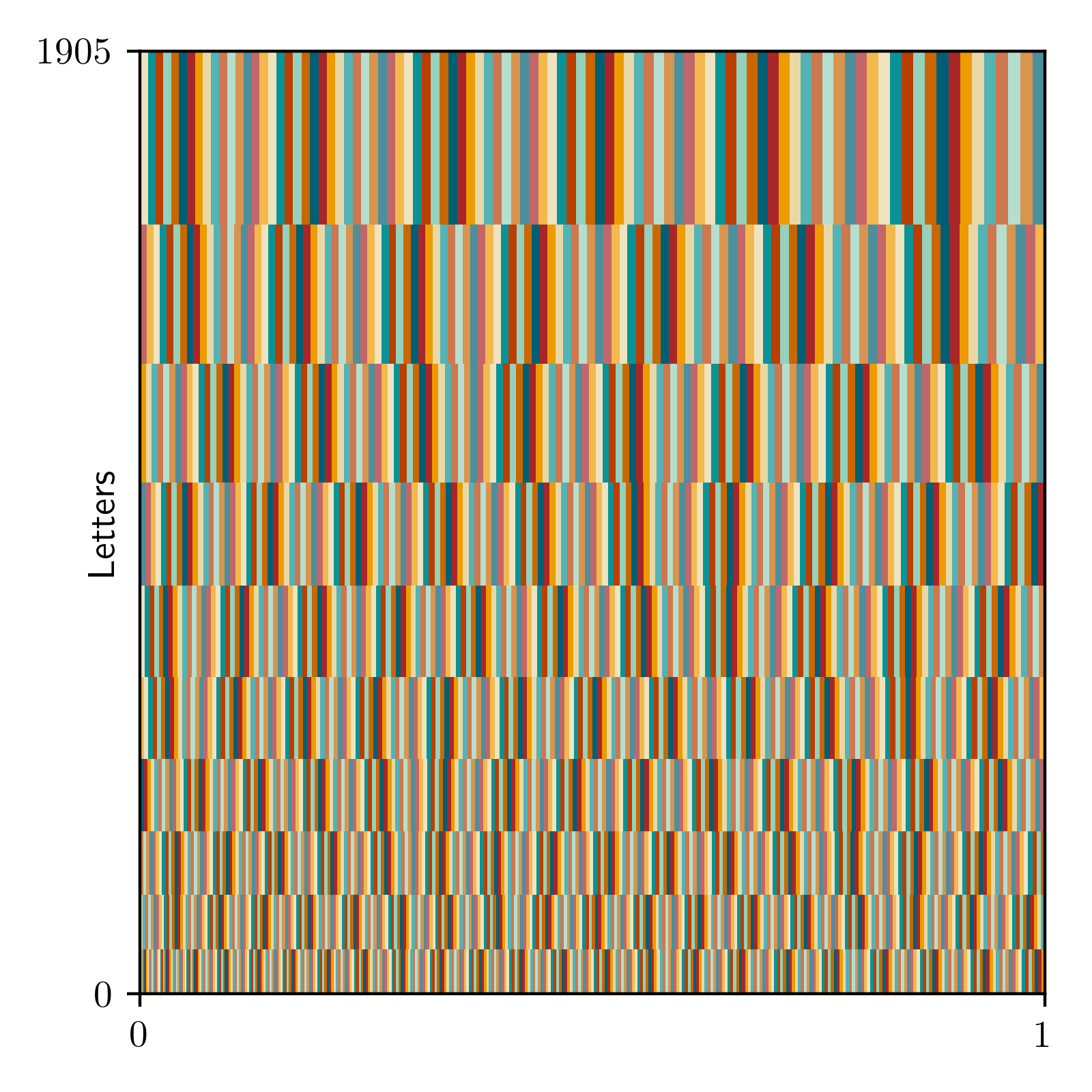}
        \includegraphics[draft=\draft, width=\linewidth]{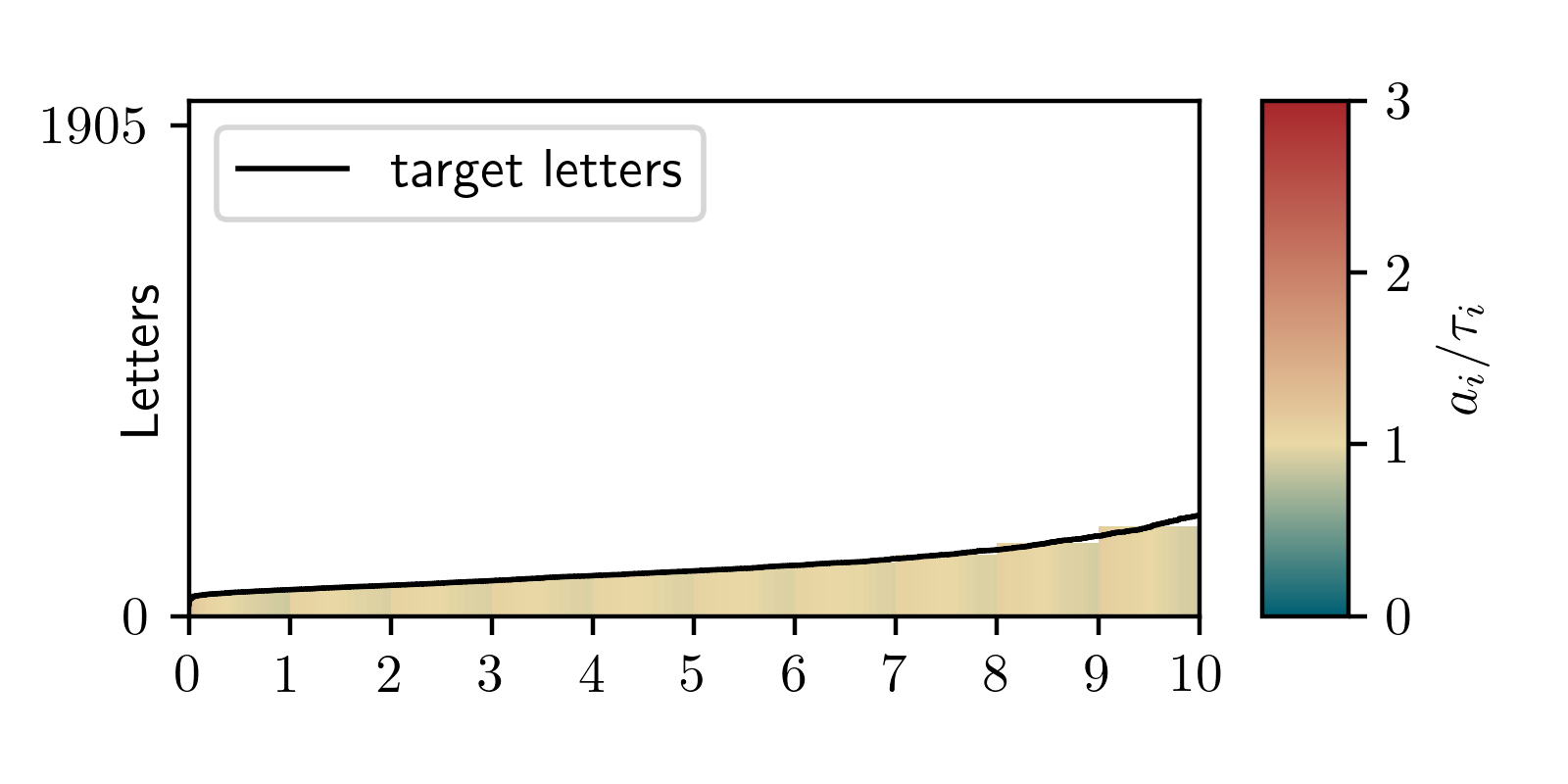}
        \caption{\buckets ($t_G = 10$)}
        \label{fig:results_Bayern_Small_greedy_bucket_fill}
    \end{subfigure}
    \caption{Small municipalities of Bayern ($\ell_G = 1905$)}
    \label{fig:results_Bayern_Small}
\end{figure} 

\begin{figure}
    \centering
    \begin{subfigure}{0.32\textwidth}
        \includegraphics[draft=\draft, width=\linewidth]{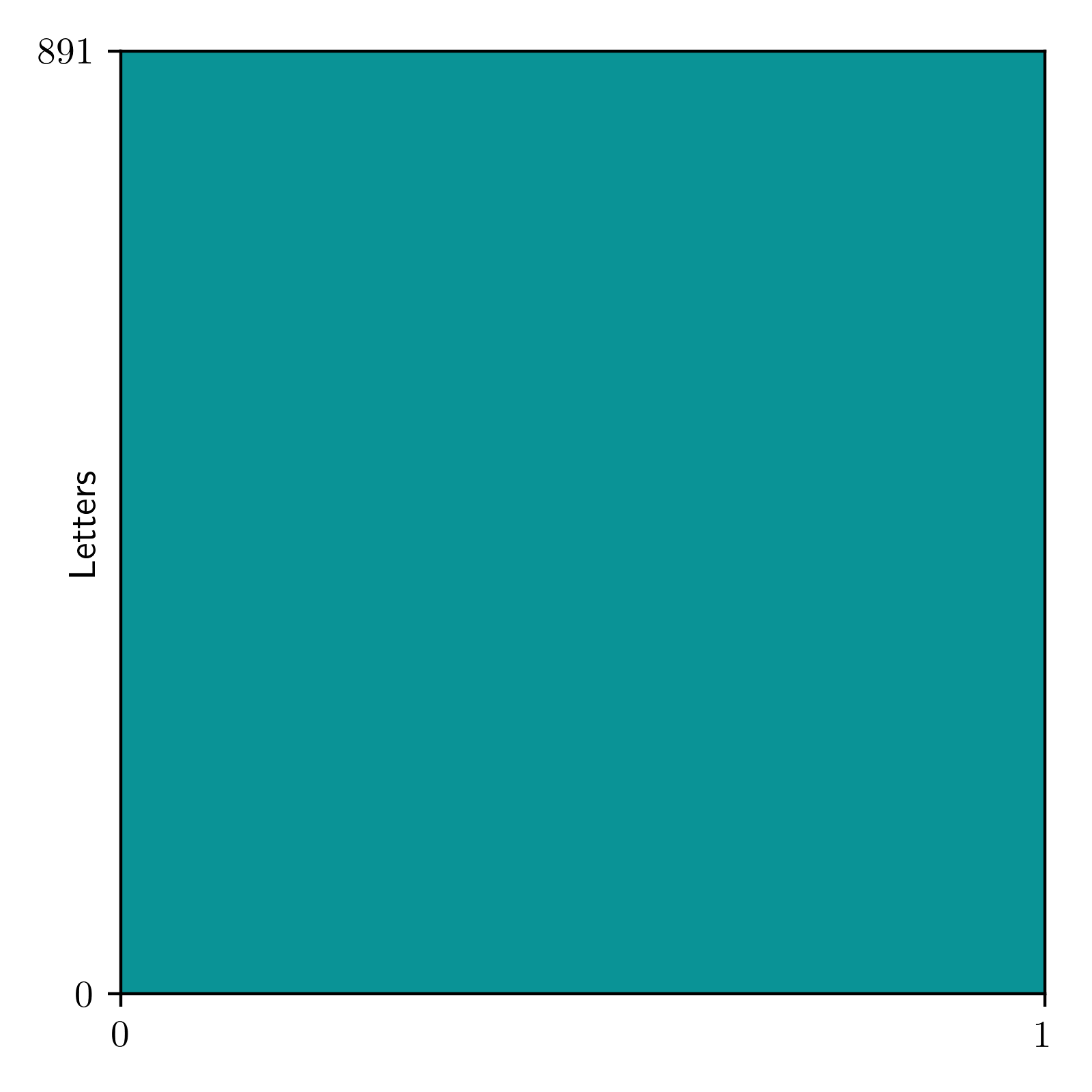}
        \includegraphics[draft=\draft, width=\linewidth]{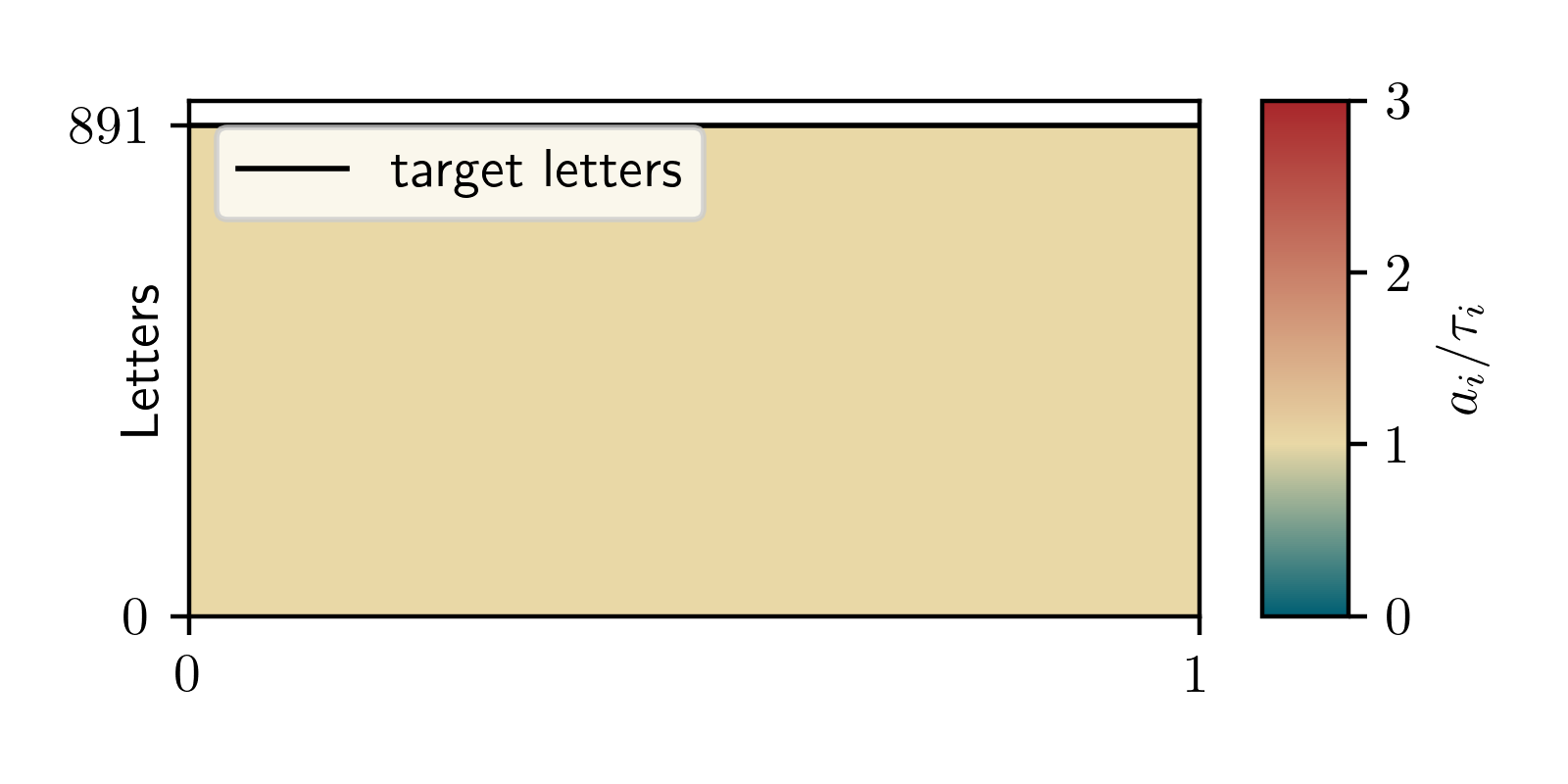}
        \caption{\greq ($t_G = 1$)}
        \label{fig:results_Berlin_Large_greedy_equal}
    \end{subfigure}
    \begin{subfigure}{0.32\textwidth}
        \includegraphics[draft=\draft, width=\linewidth]{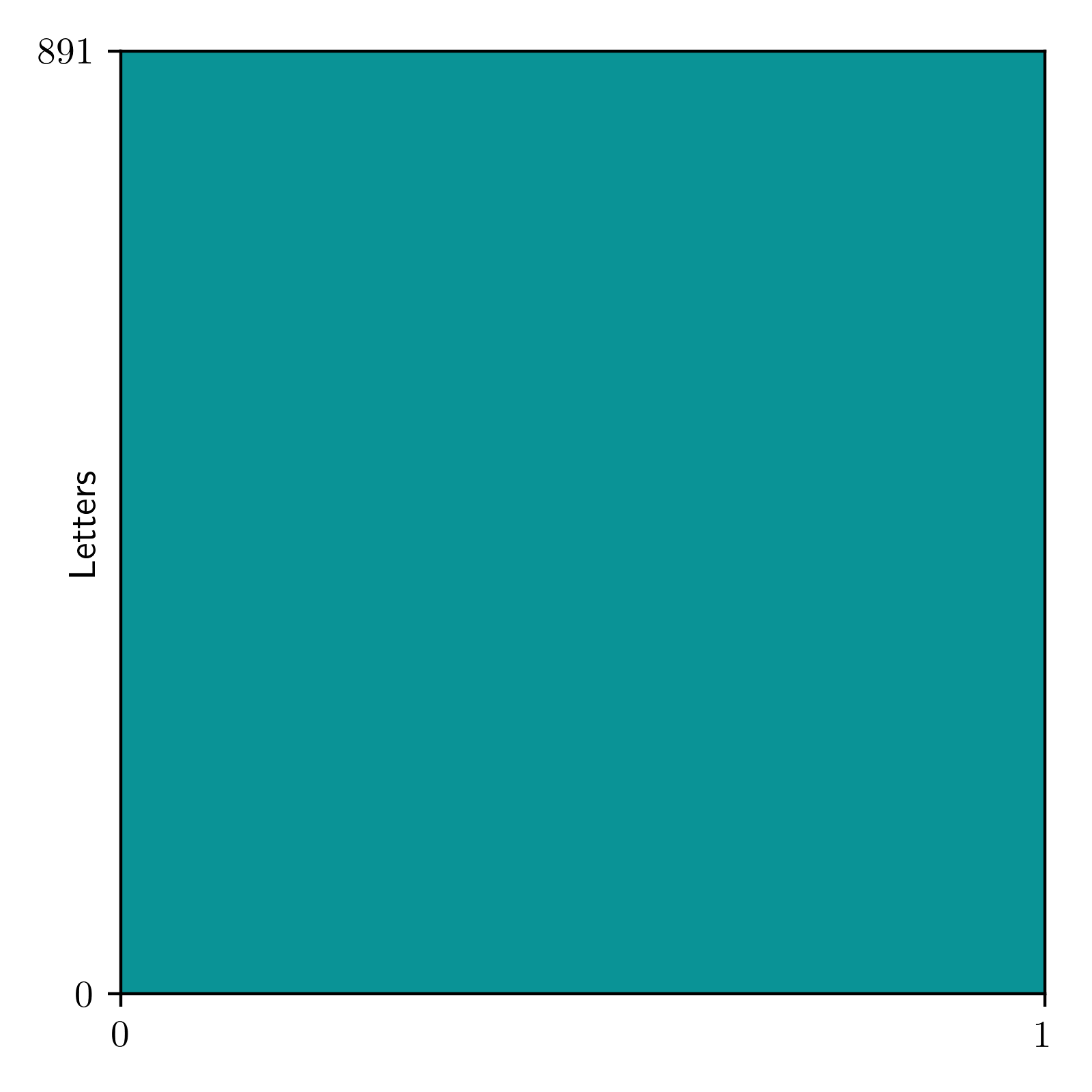}
        \includegraphics[draft=\draft, width=\linewidth]{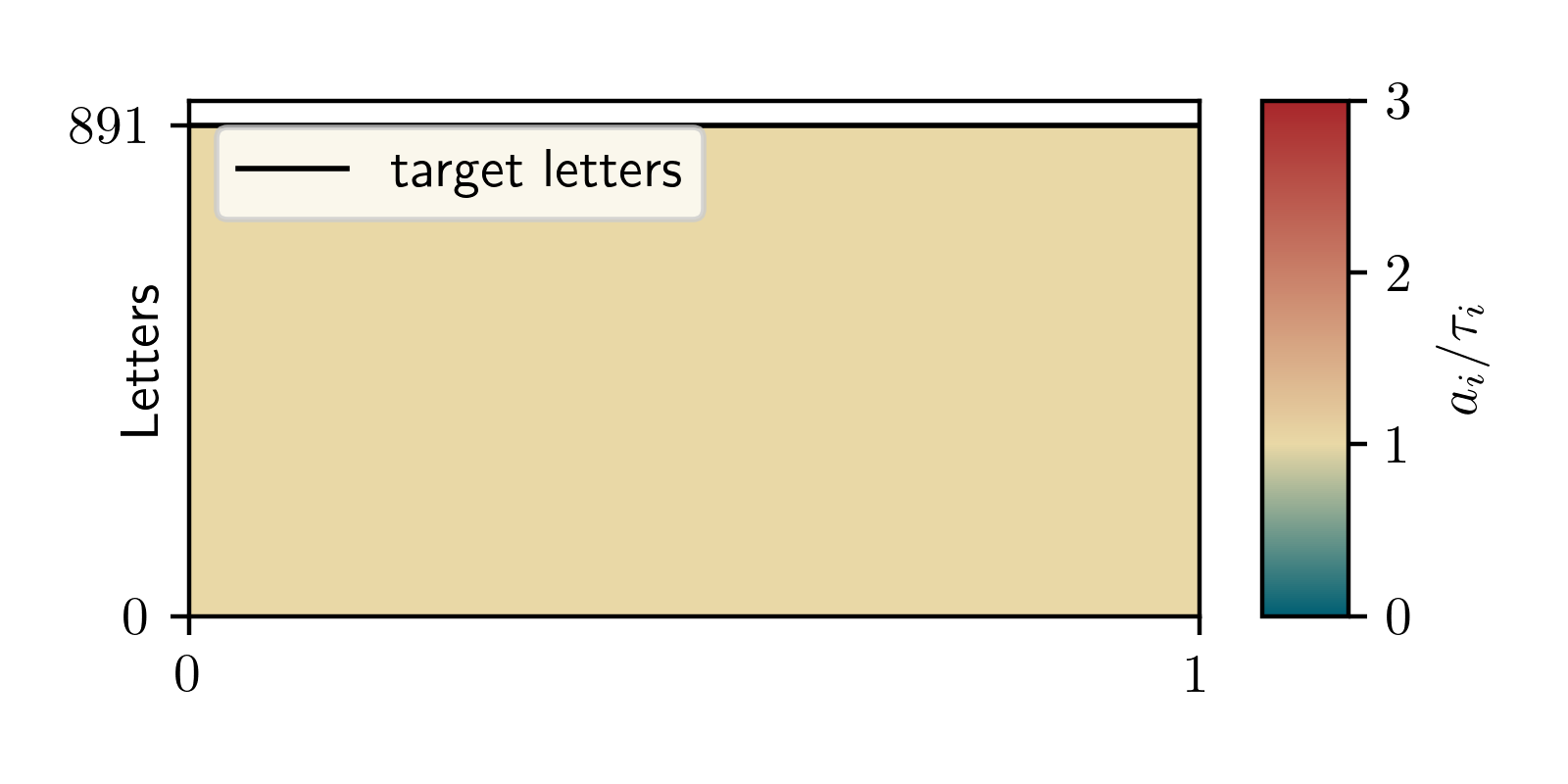}
        \caption{\colgen ($t_G\!=\!1$)}
        \label{fig:results_Berlin_Large_column_generation}
    \end{subfigure}
    \begin{subfigure}{0.32\textwidth}
        \includegraphics[draft=\draft, width=\linewidth]{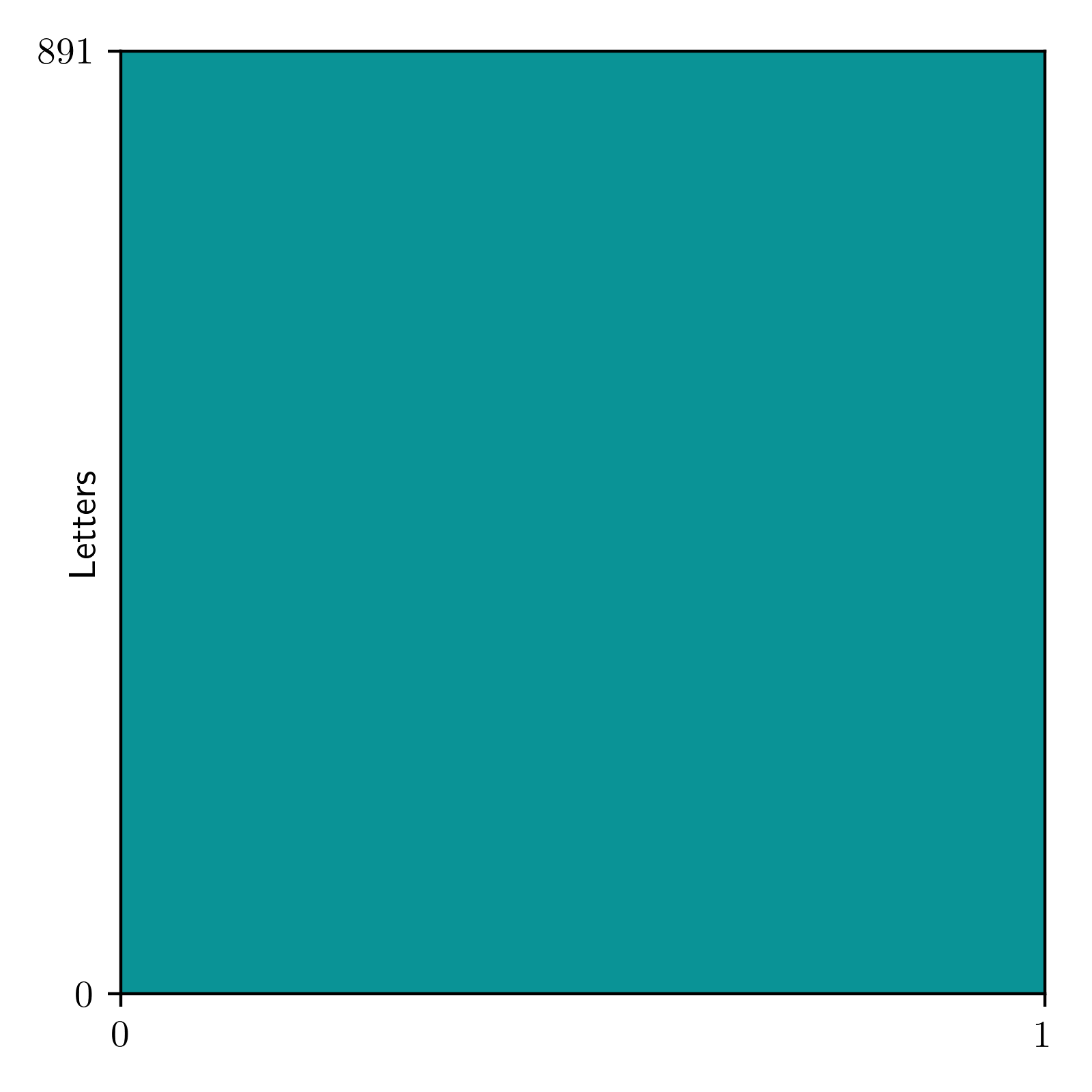}
        \includegraphics[draft=\draft, width=\linewidth]{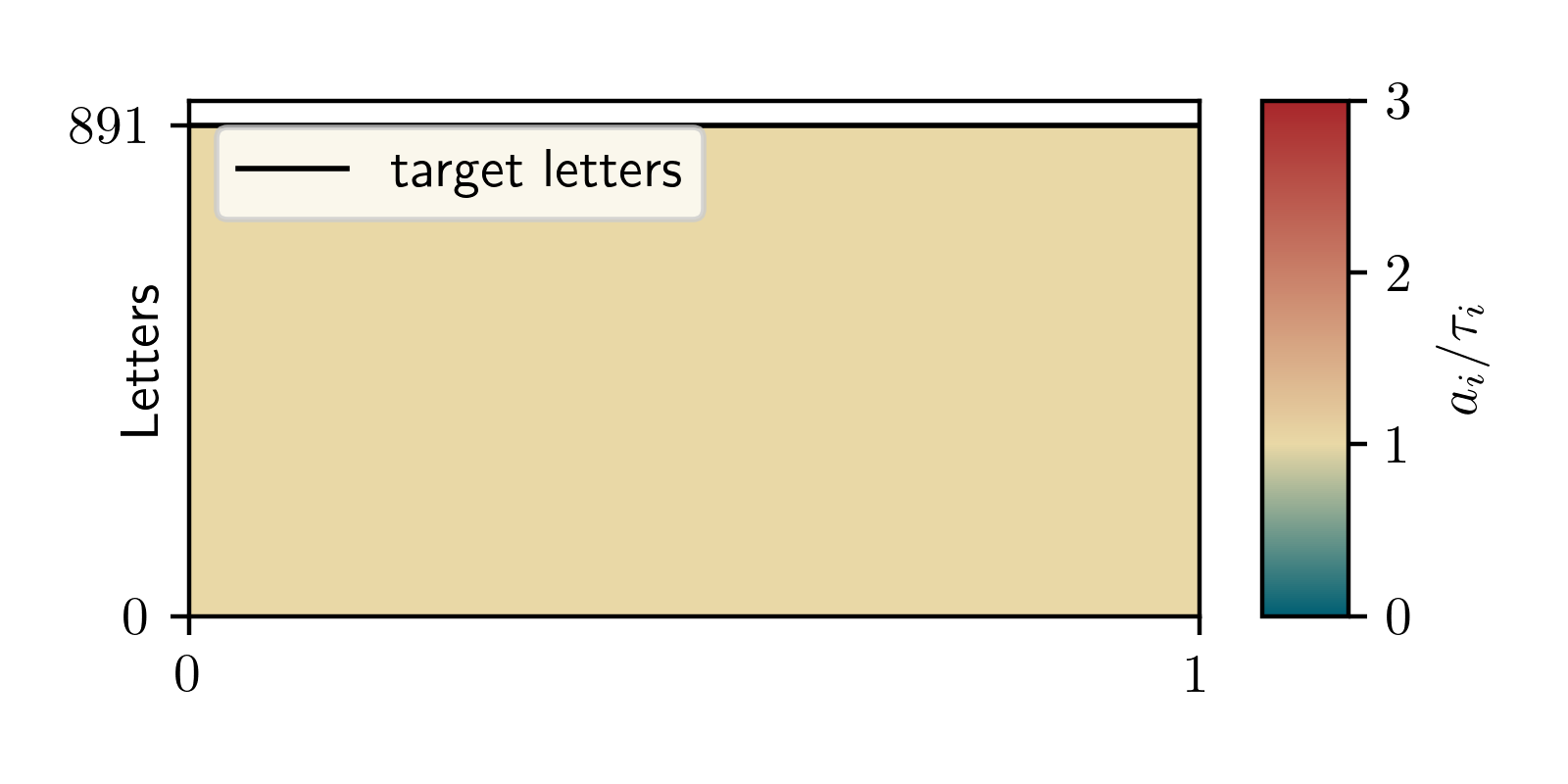}
        \caption{\buckets ($t_G = 1$)}
        \label{fig:results_Berlin_Large_greedy_bucket_fill}
    \end{subfigure}
    \caption{Large municipalities of Berlin ($\ell_G = 891$)}
    \label{fig:results_Berlin_Large}
\end{figure}

\begin{figure}
    \centering
    \begin{subfigure}{0.32\textwidth}
        \includegraphics[draft=\draft, width=\linewidth]{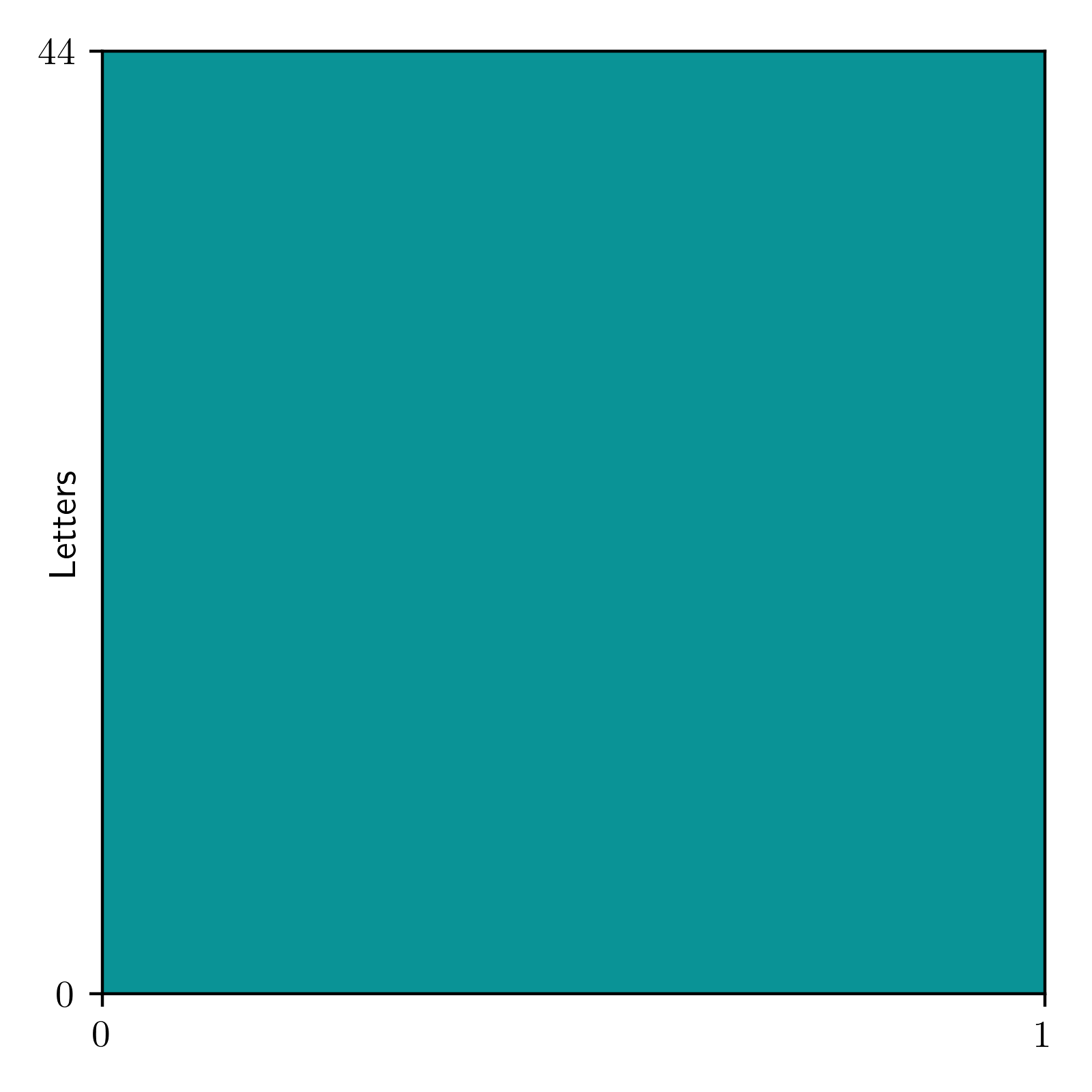}
        \includegraphics[draft=\draft, width=\linewidth]{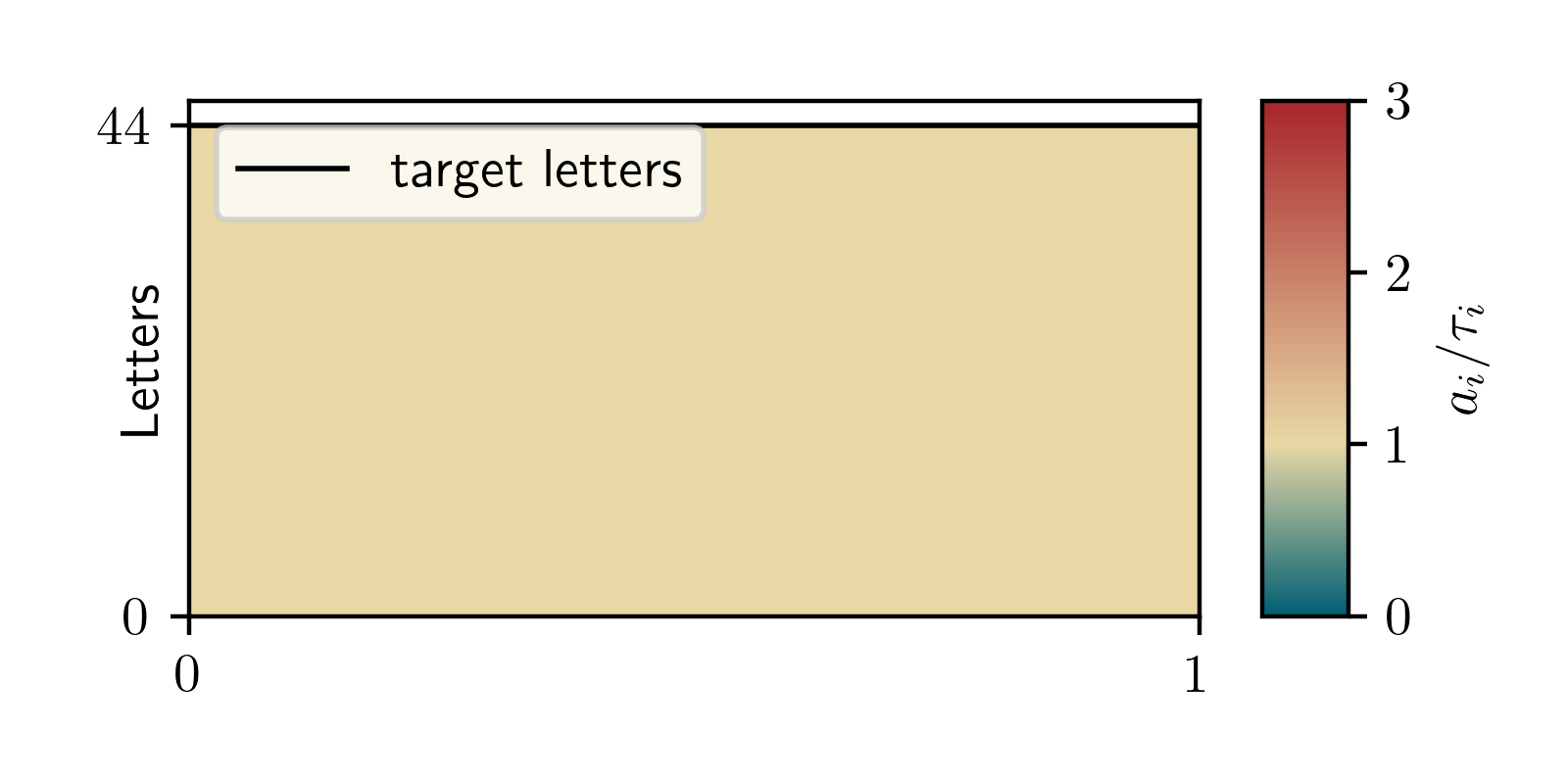}
        \caption{\greq ($t_G = 1$)}
        \label{fig:results_Brandenburg_Large_greedy_equal}
    \end{subfigure}
    \begin{subfigure}{0.32\textwidth}
        \includegraphics[draft=\draft, width=\linewidth]{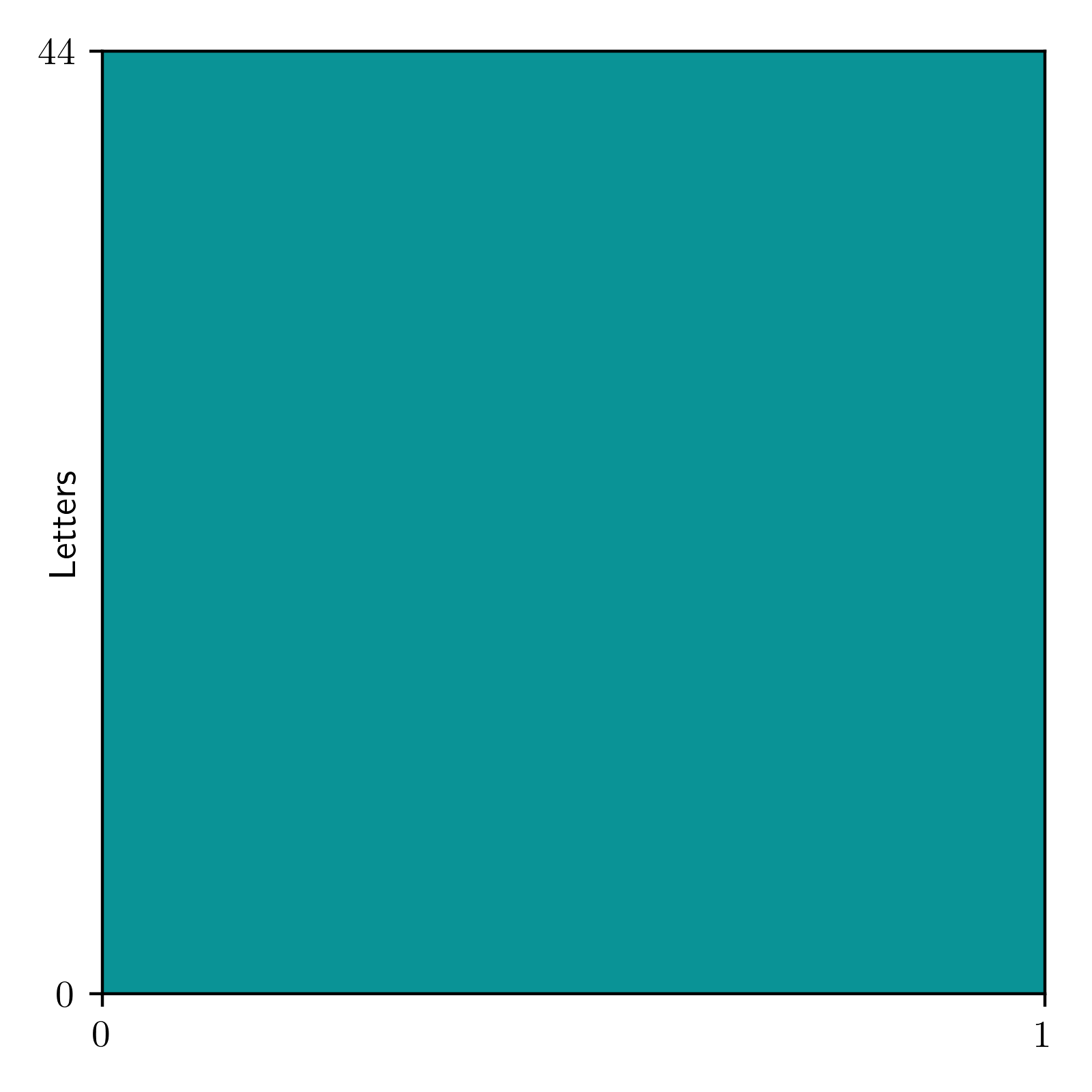}
        \includegraphics[draft=\draft, width=\linewidth]{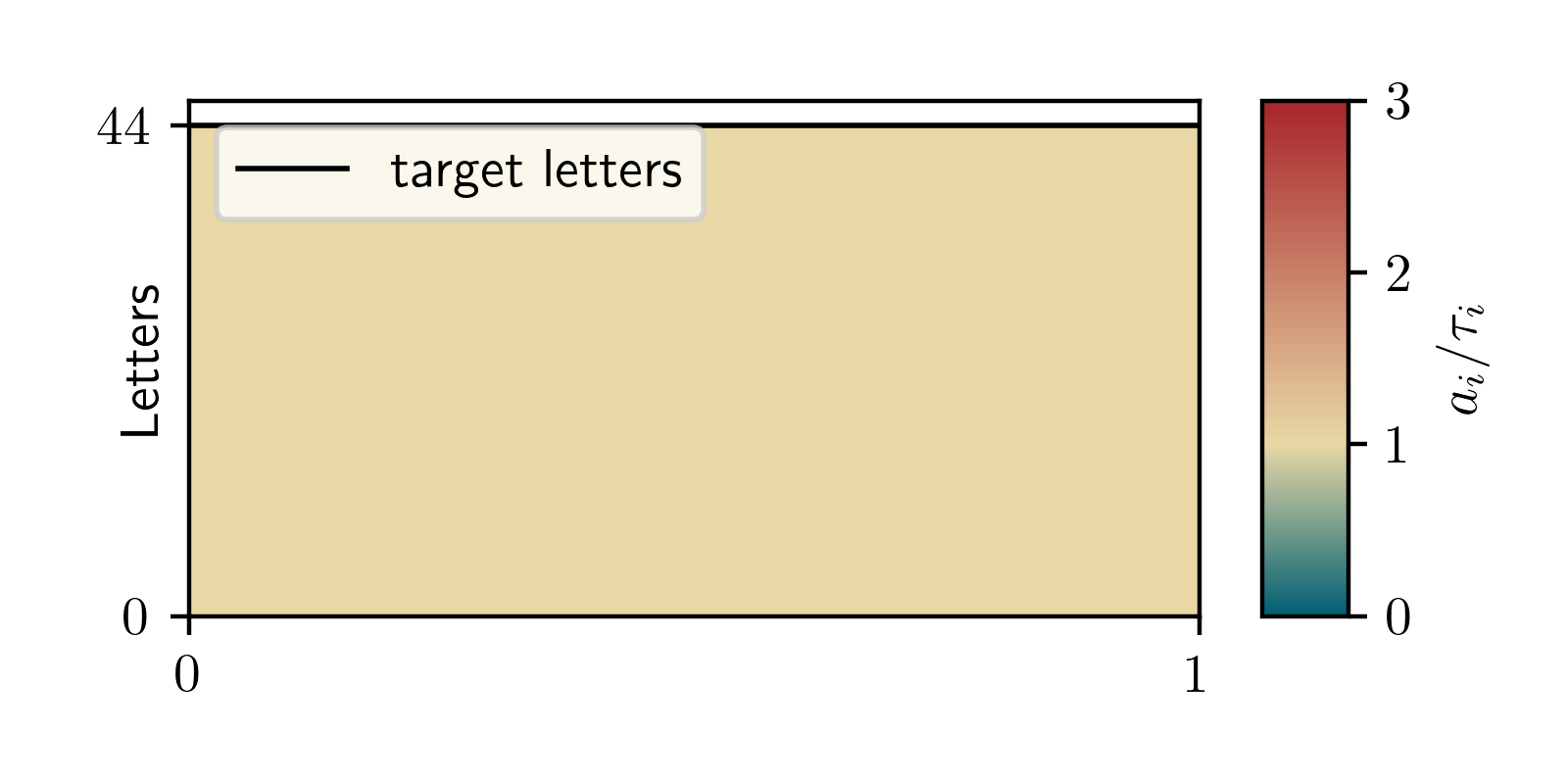}
        \caption{\colgen ($t_G\!=\!1$)}
        \label{fig:results_Brandenburg_Large_column_generation}
    \end{subfigure}
    \begin{subfigure}{0.32\textwidth}
        \includegraphics[draft=\draft, width=\linewidth]{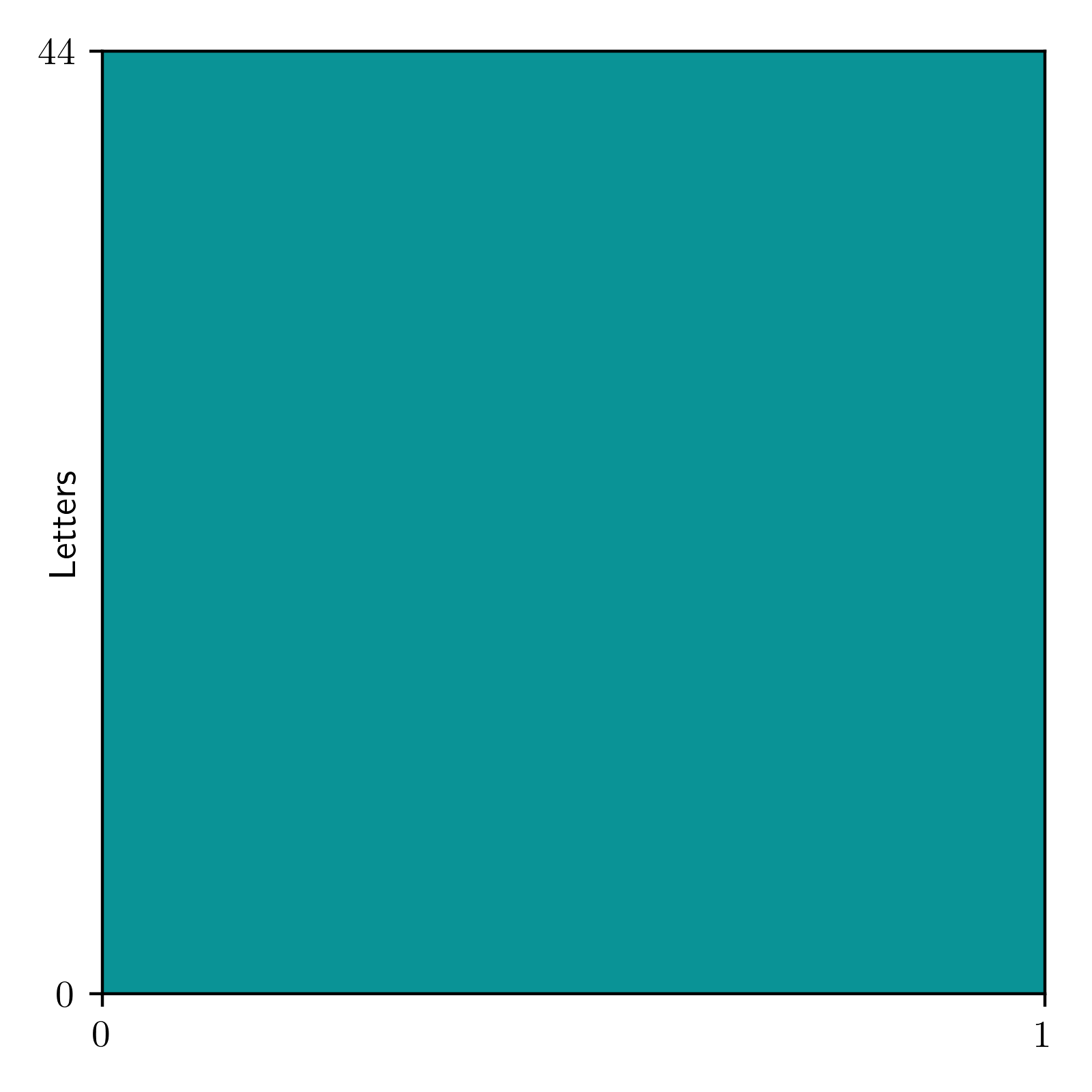}
        \includegraphics[draft=\draft, width=\linewidth]{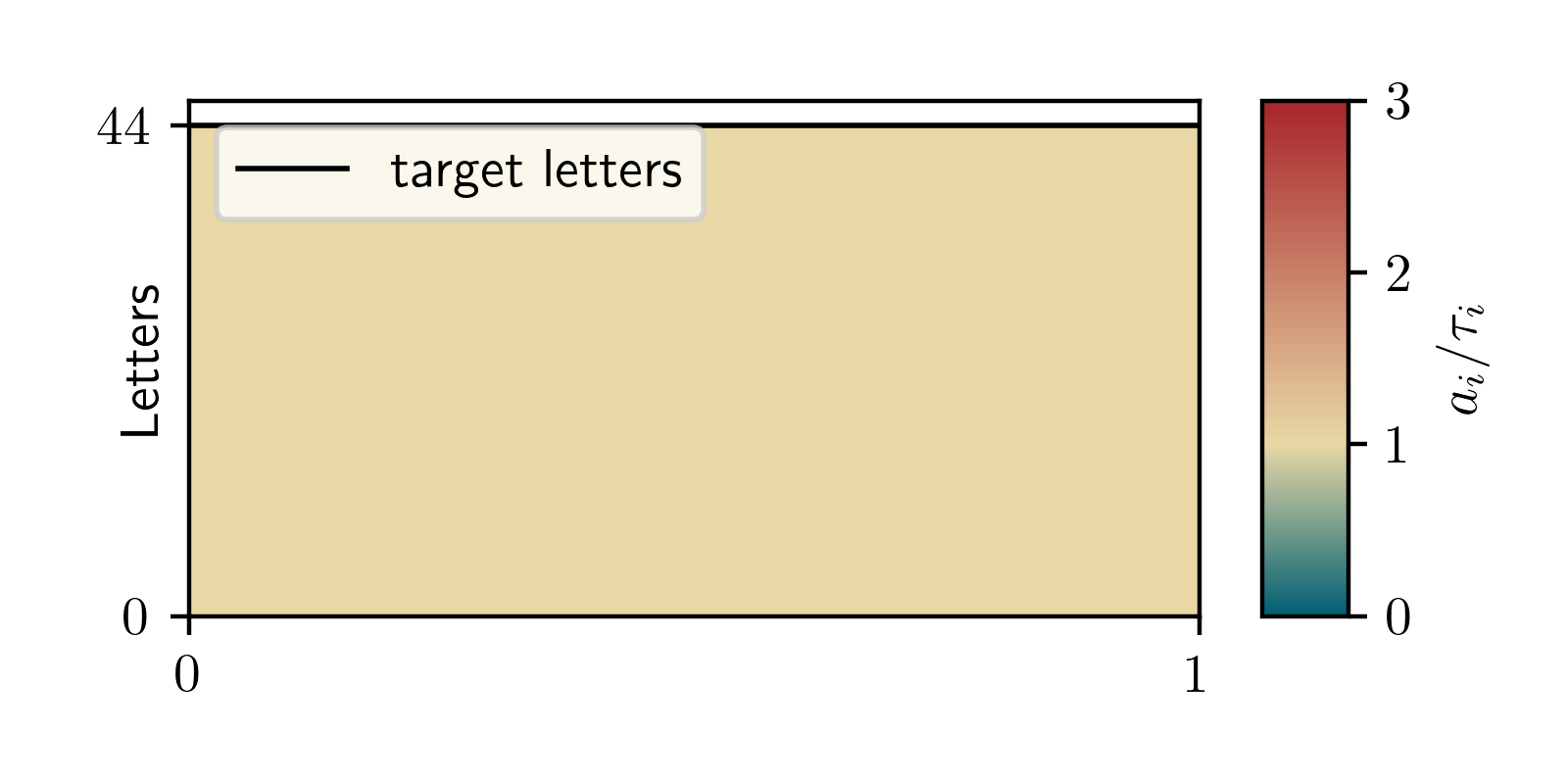}
        \caption{\buckets ($t_G = 1$)}
        \label{fig:results_Brandenburg_Large_greedy_bucket_fill}
    \end{subfigure}
    \caption{Large municipalities of Brandenburg ($\ell_G = 44$)}
    \label{fig:results_Brandenburg_Large}
\end{figure} 

\begin{figure}
    \centering
    \begin{subfigure}{0.32\textwidth}
        \includegraphics[draft=\draft, width=\linewidth]{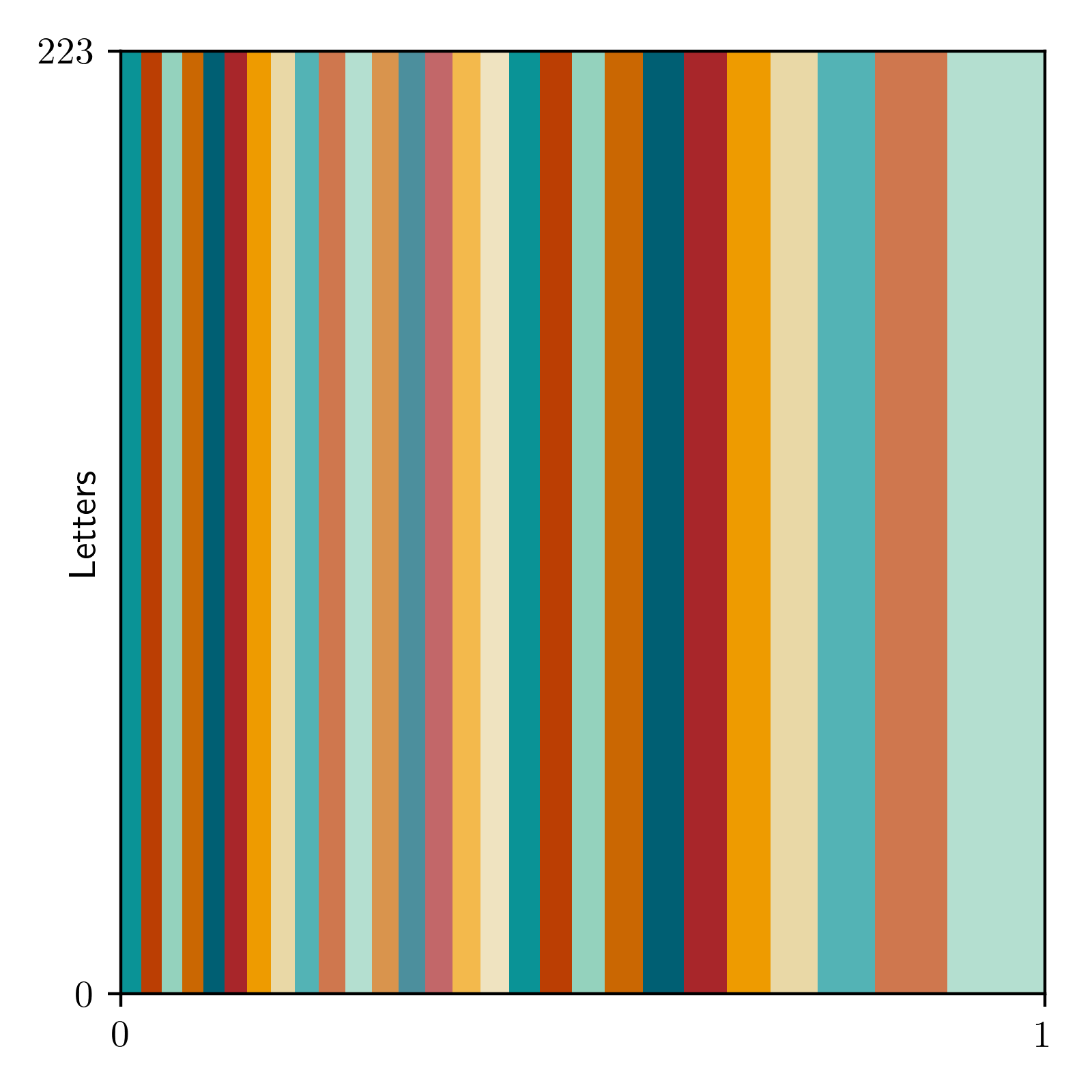}
        \includegraphics[draft=\draft, width=\linewidth]{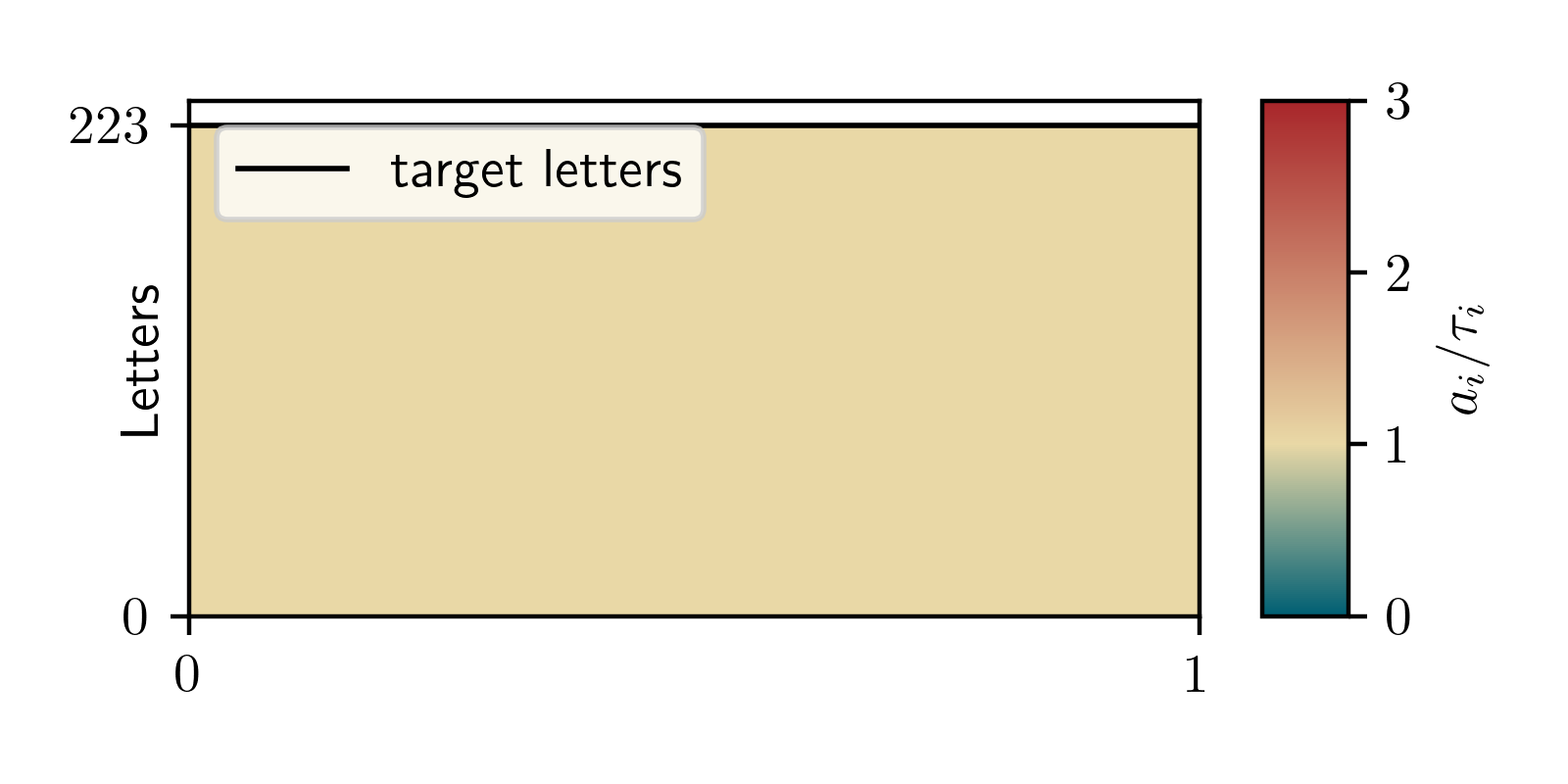}
        \caption{\greq ($t_G = 1$)}
        \label{fig:results_Brandenburg_Medium_greedy_equal}
    \end{subfigure}
    \begin{subfigure}{0.32\textwidth}
        \includegraphics[draft=\draft, width=\linewidth]{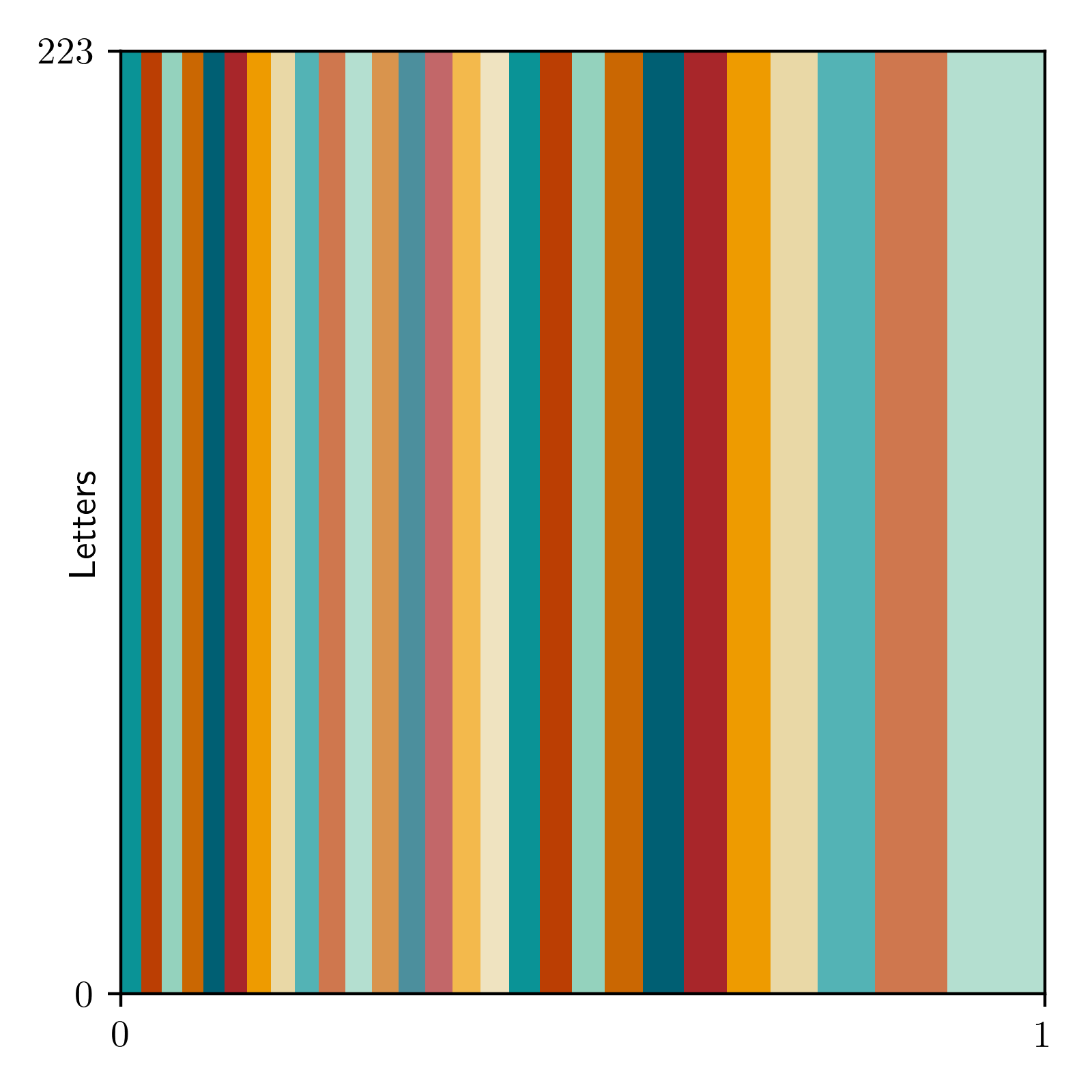}
        \includegraphics[draft=\draft, width=\linewidth]{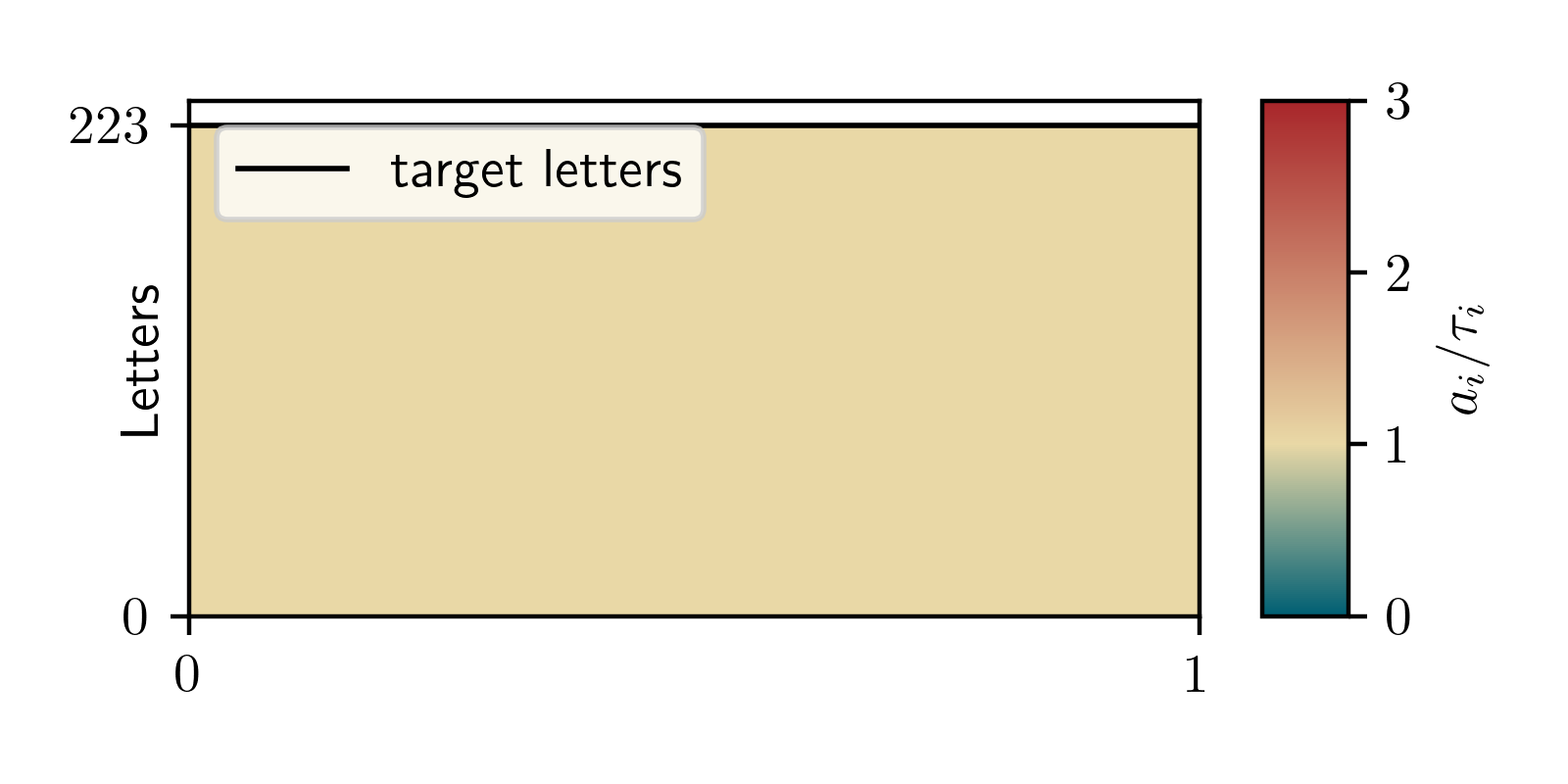}
        \caption{\colgen ($t_G\!=\!1$)}
        \label{fig:results_Brandenburg_Medium_column_generation}
    \end{subfigure}
    \begin{subfigure}{0.32\textwidth}
        \includegraphics[draft=\draft, width=\linewidth]{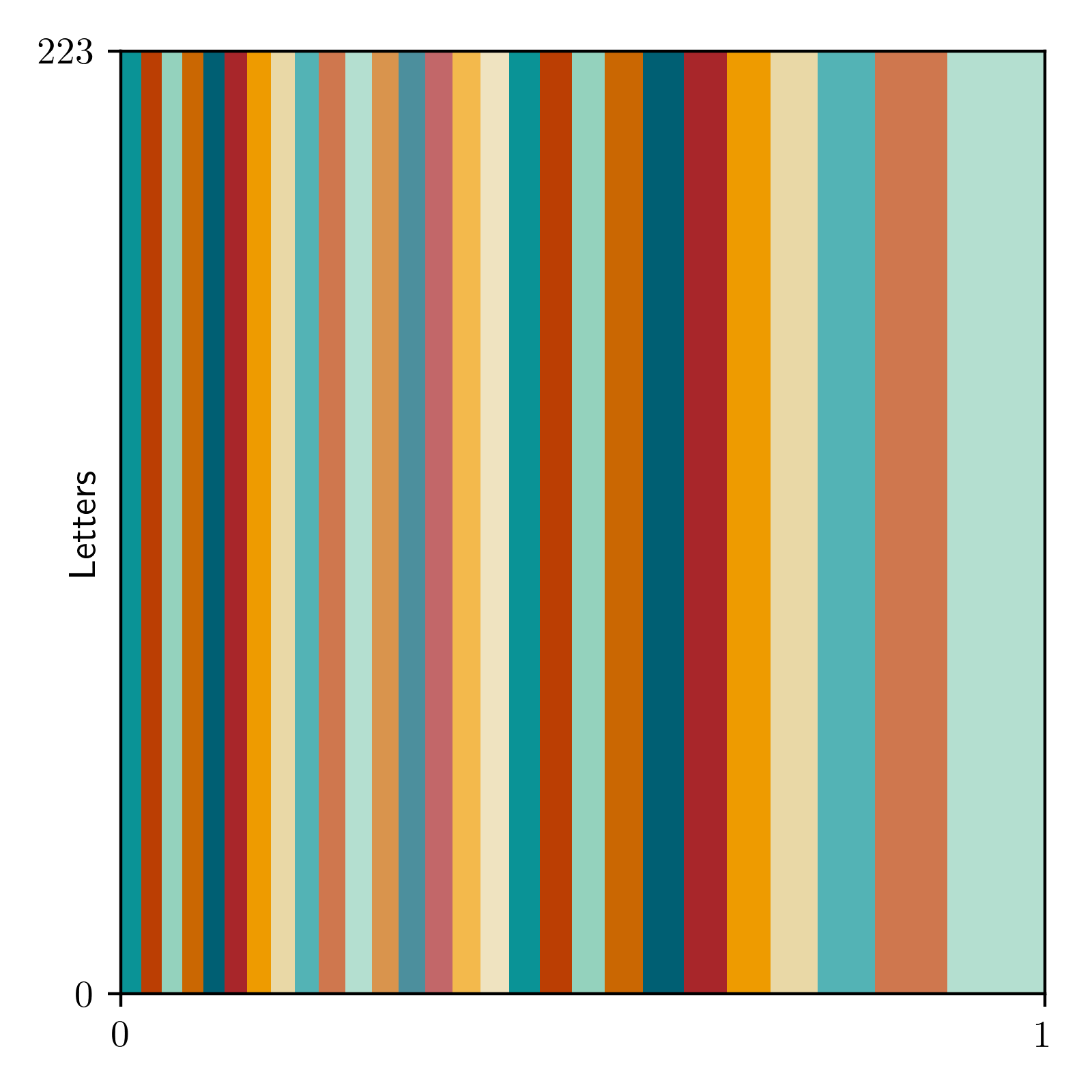}
        \includegraphics[draft=\draft, width=\linewidth]{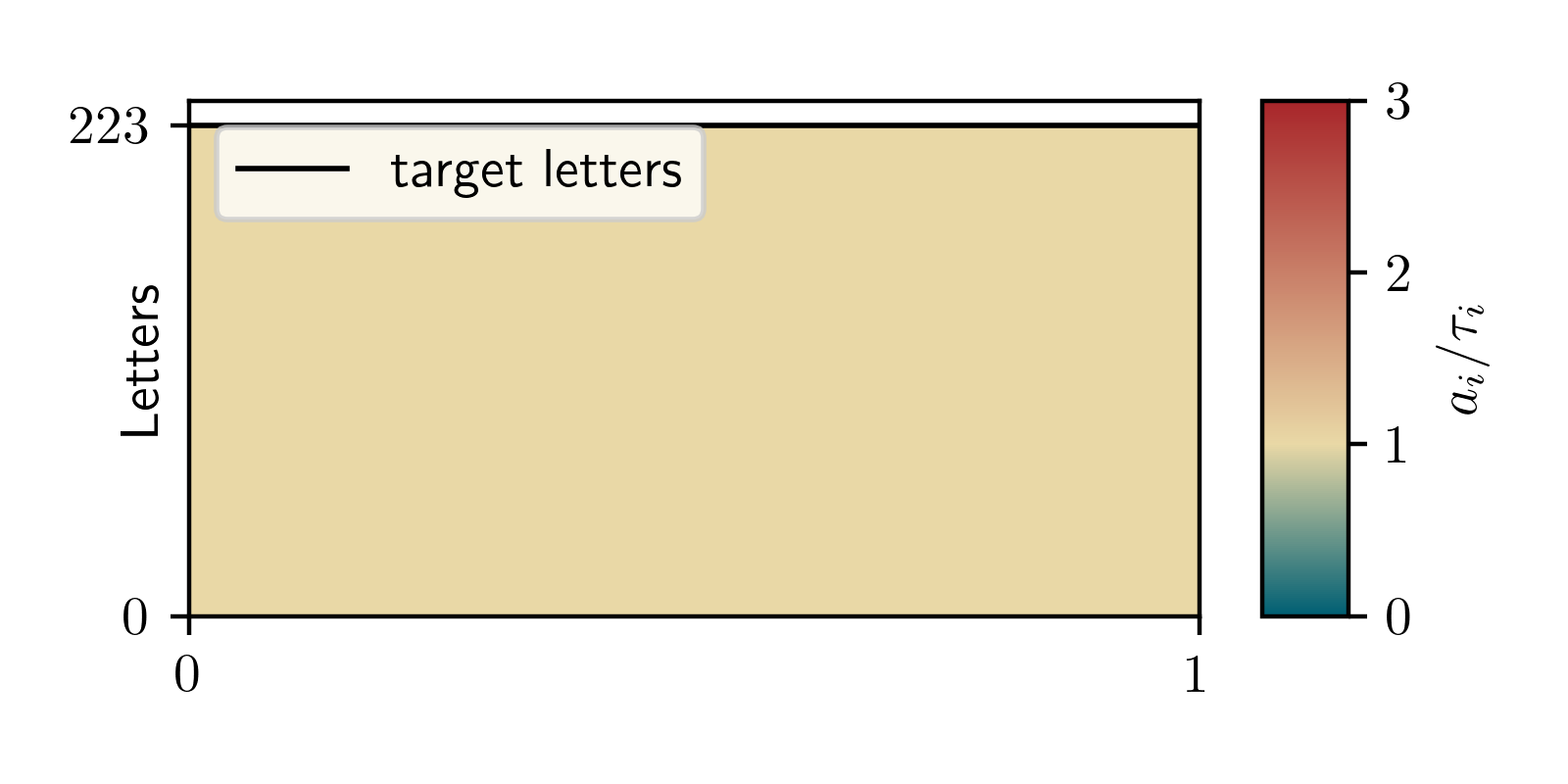}
        \caption{\buckets ($t_G = 1$)}
        \label{fig:results_Brandenburg_Medium_greedy_bucket_fill}
    \end{subfigure}
    \caption{Medium municipalities of Brandenburg ($\ell_G = 223$)}
    \label{fig:results_Brandenburg_Medium}
\end{figure} 

\begin{figure}
    \centering
    \begin{subfigure}{0.32\textwidth}
        \includegraphics[draft=\draft, width=\linewidth]{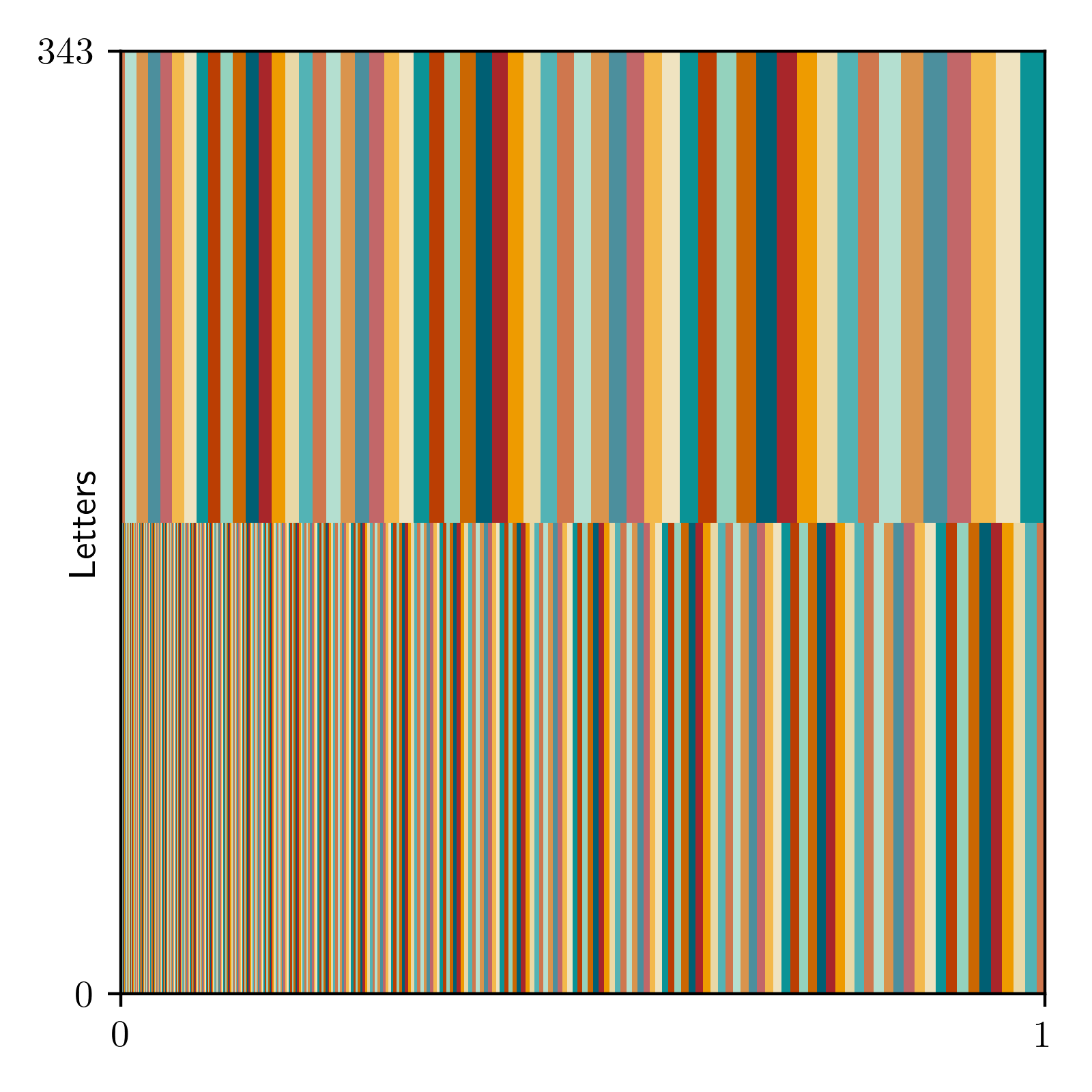}
        \includegraphics[draft=\draft, width=\linewidth]{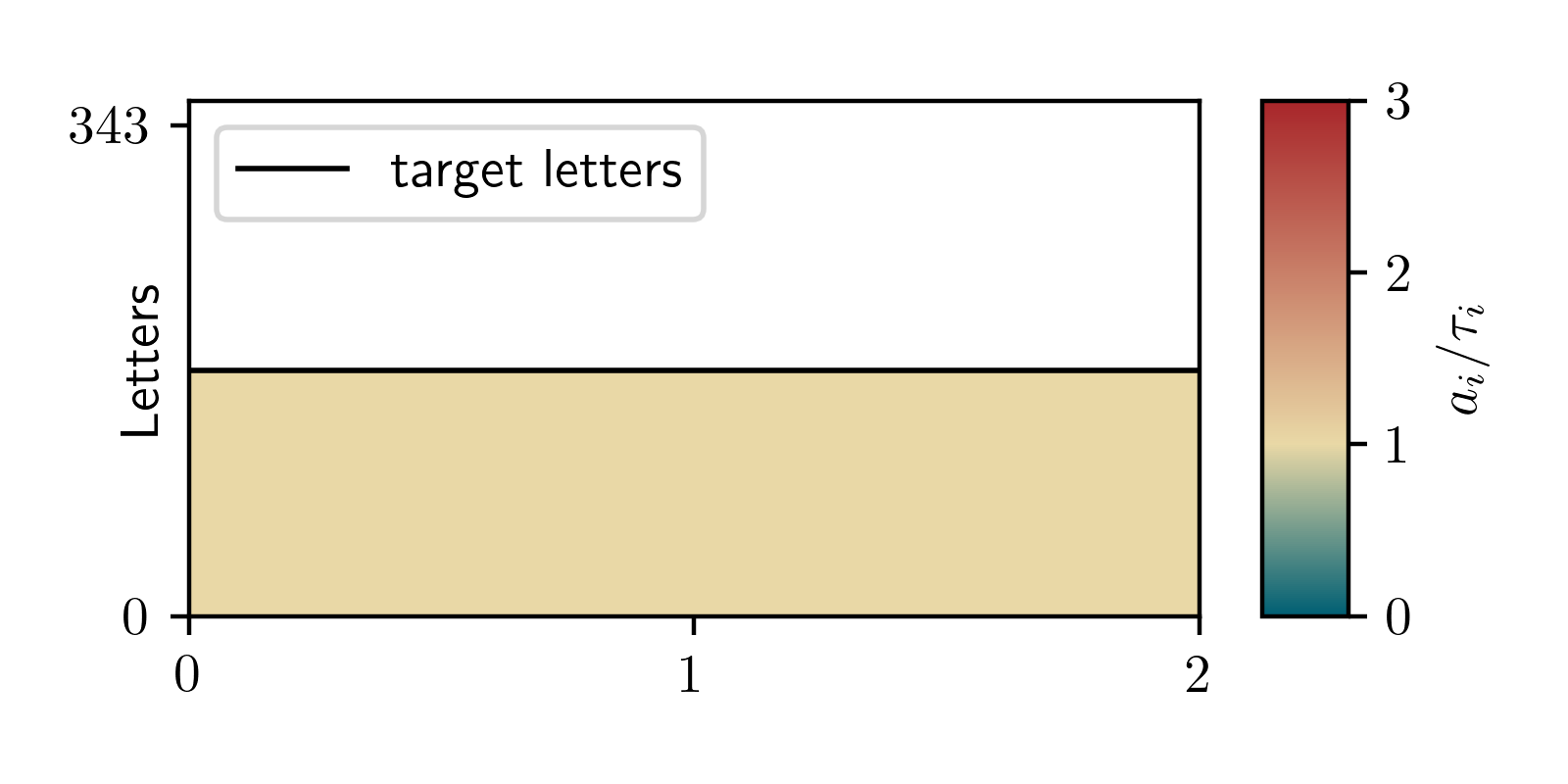}
        \caption{\greq ($t_G = 2$)}
        \label{fig:results_Brandenburg_Small_greedy_equal}
    \end{subfigure}
    \begin{subfigure}{0.32\textwidth}
        \includegraphics[draft=\draft, width=\linewidth]{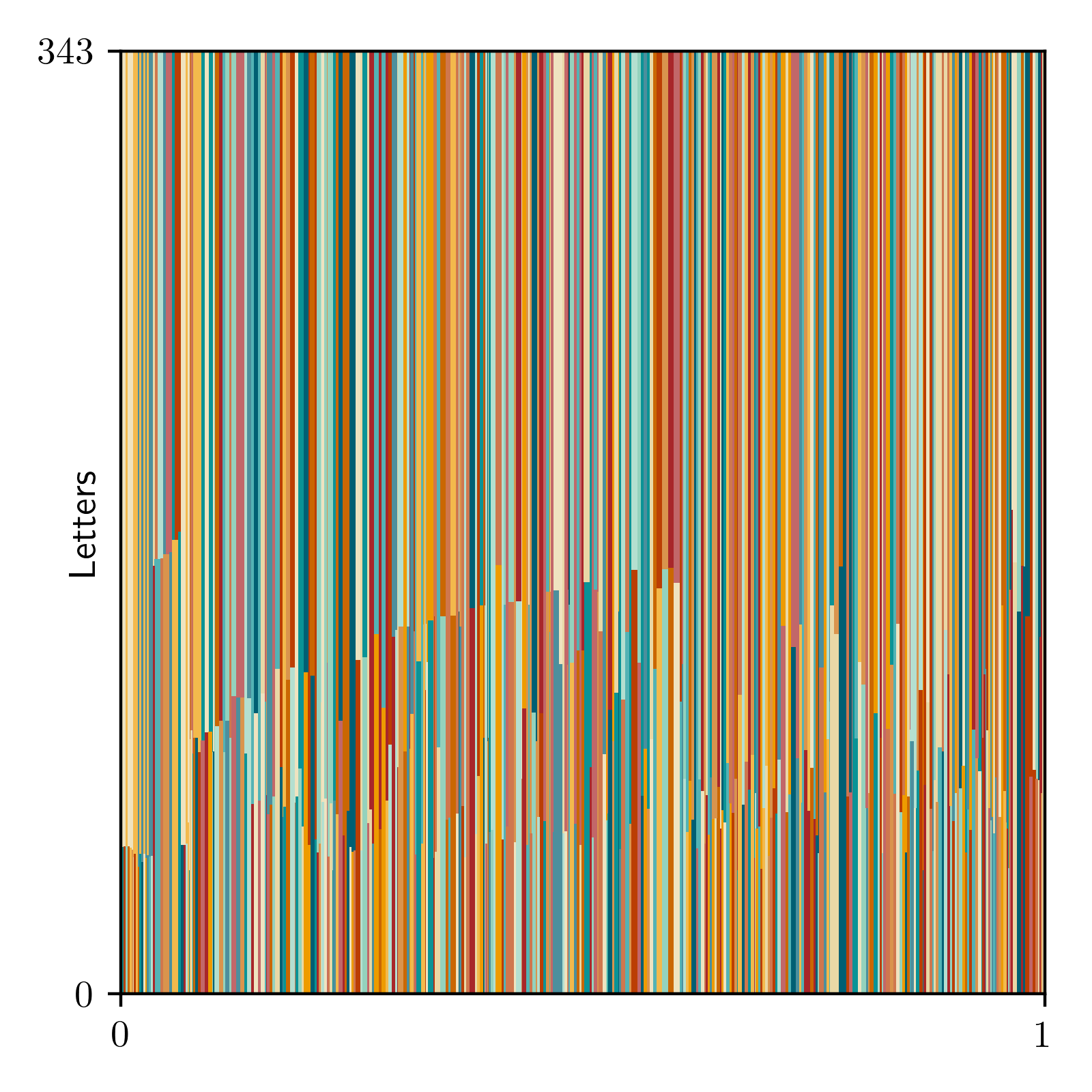}
        \includegraphics[draft=\draft, width=\linewidth]{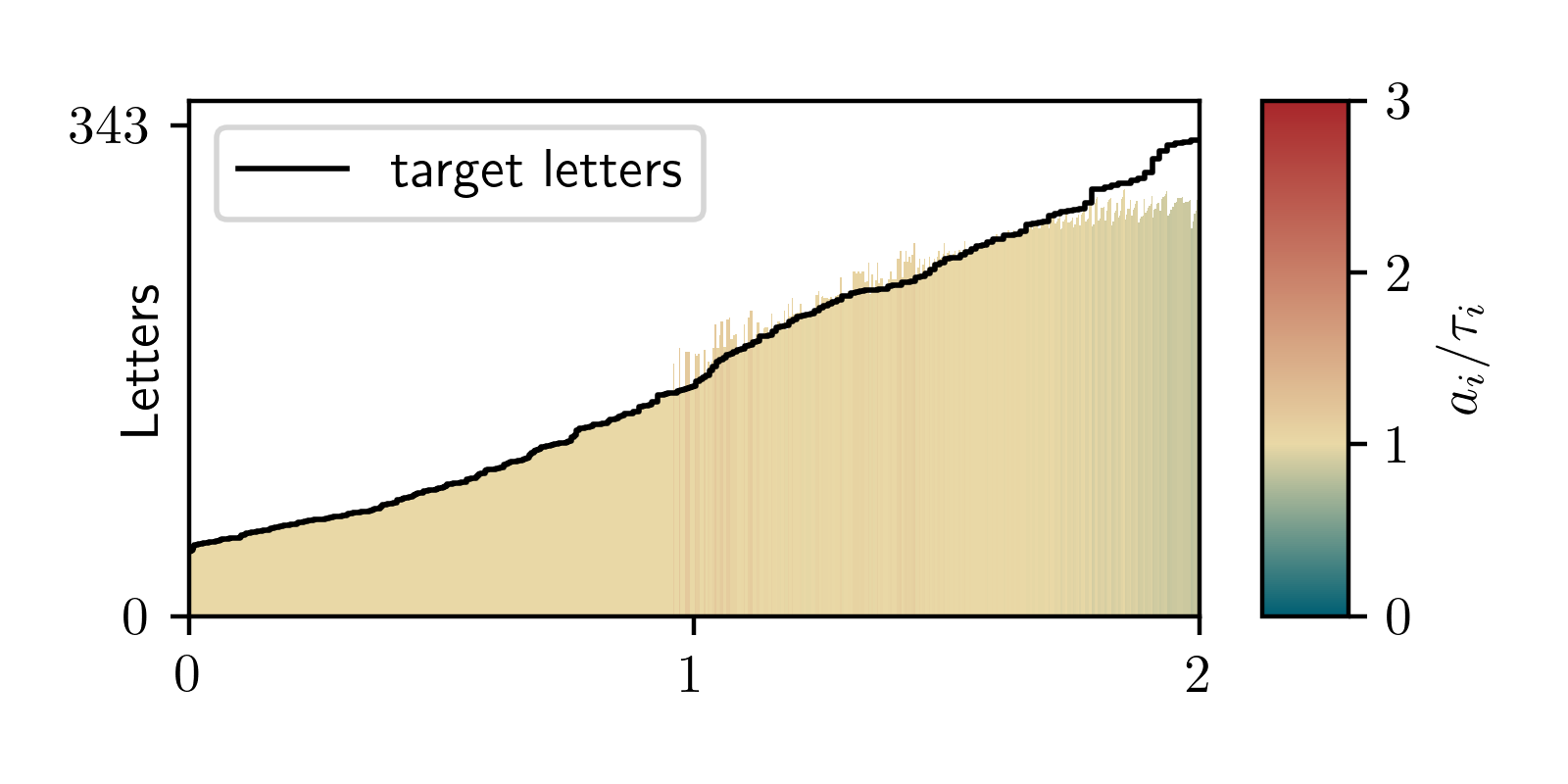}
        \caption{\colgen ($t_G\!=\!2$)}
        \label{fig:results_Brandenburg_Small_column_generation}
    \end{subfigure}
    \begin{subfigure}{0.32\textwidth}
        \includegraphics[draft=\draft, width=\linewidth]{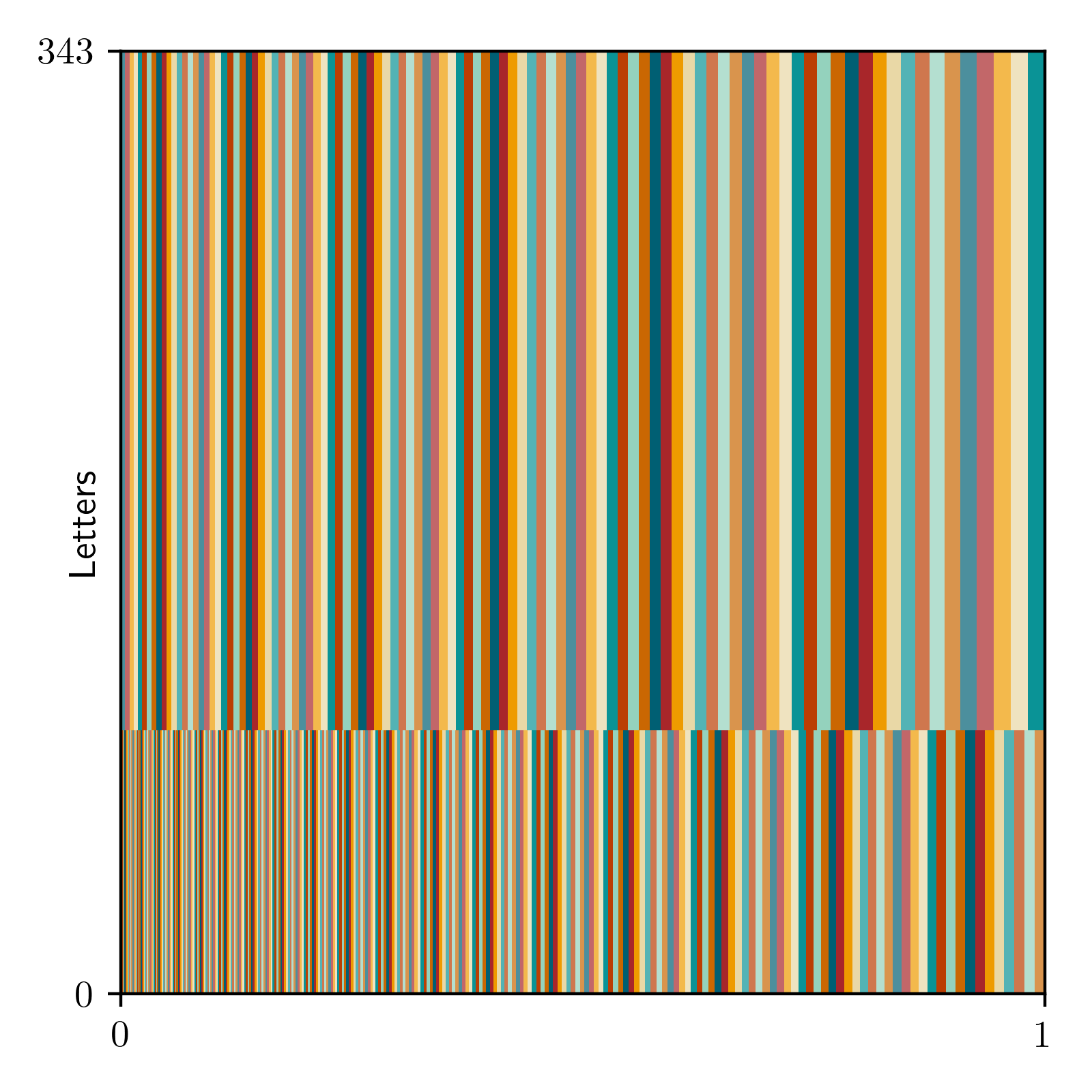}
        \includegraphics[draft=\draft, width=\linewidth]{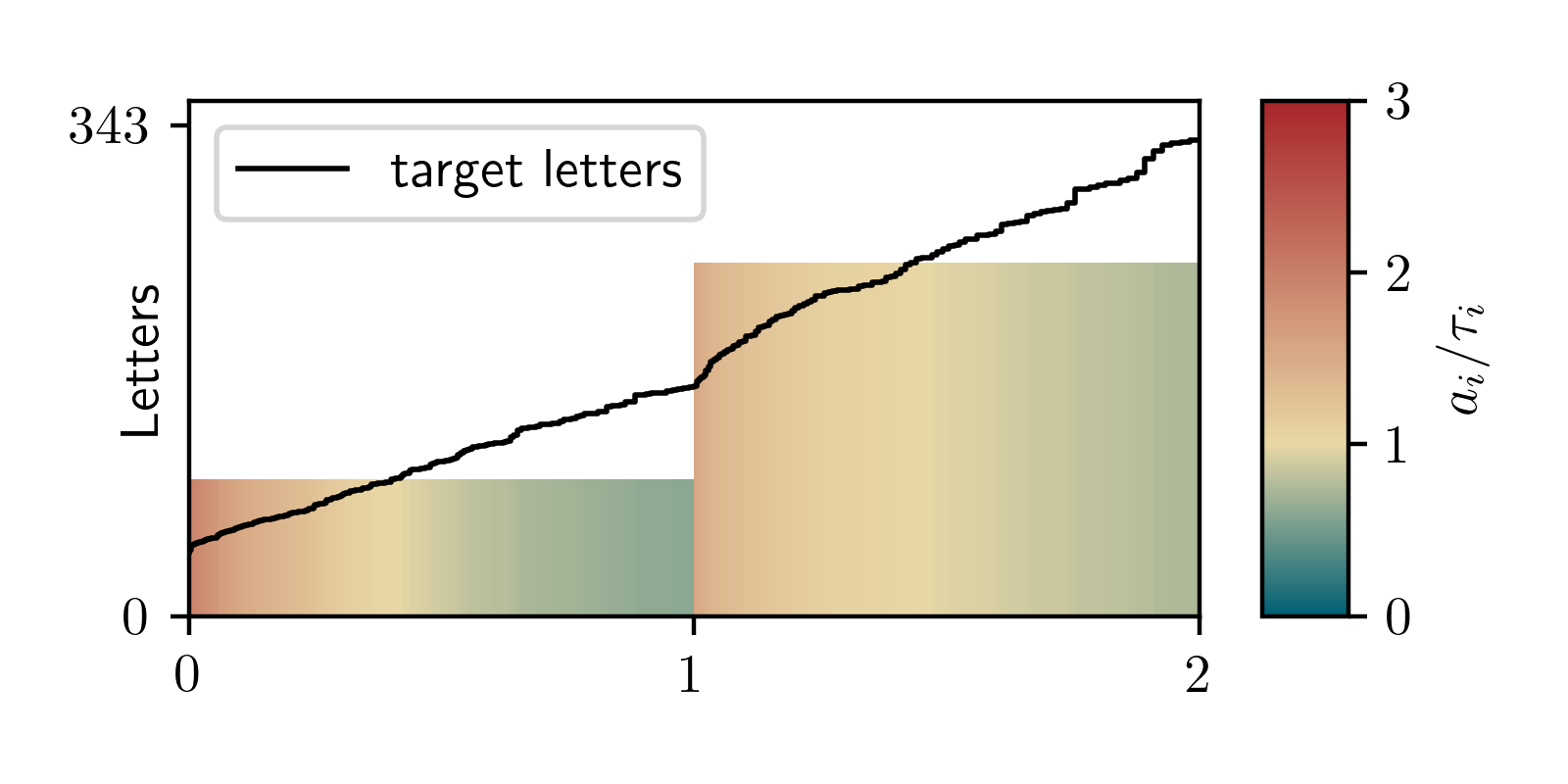}
        \caption{\buckets ($t_G = 2$)}
        \label{fig:results_Brandenburg_Small_greedy_bucket_fill}
    \end{subfigure}
    \caption{Small municipalities of Brandenburg ($\ell_G = 343$)}
    \label{fig:results_Brandenburg_Small}
\end{figure} 

\begin{figure}
    \centering
    \begin{subfigure}{0.32\textwidth}
        \includegraphics[draft=\draft, width=\linewidth]{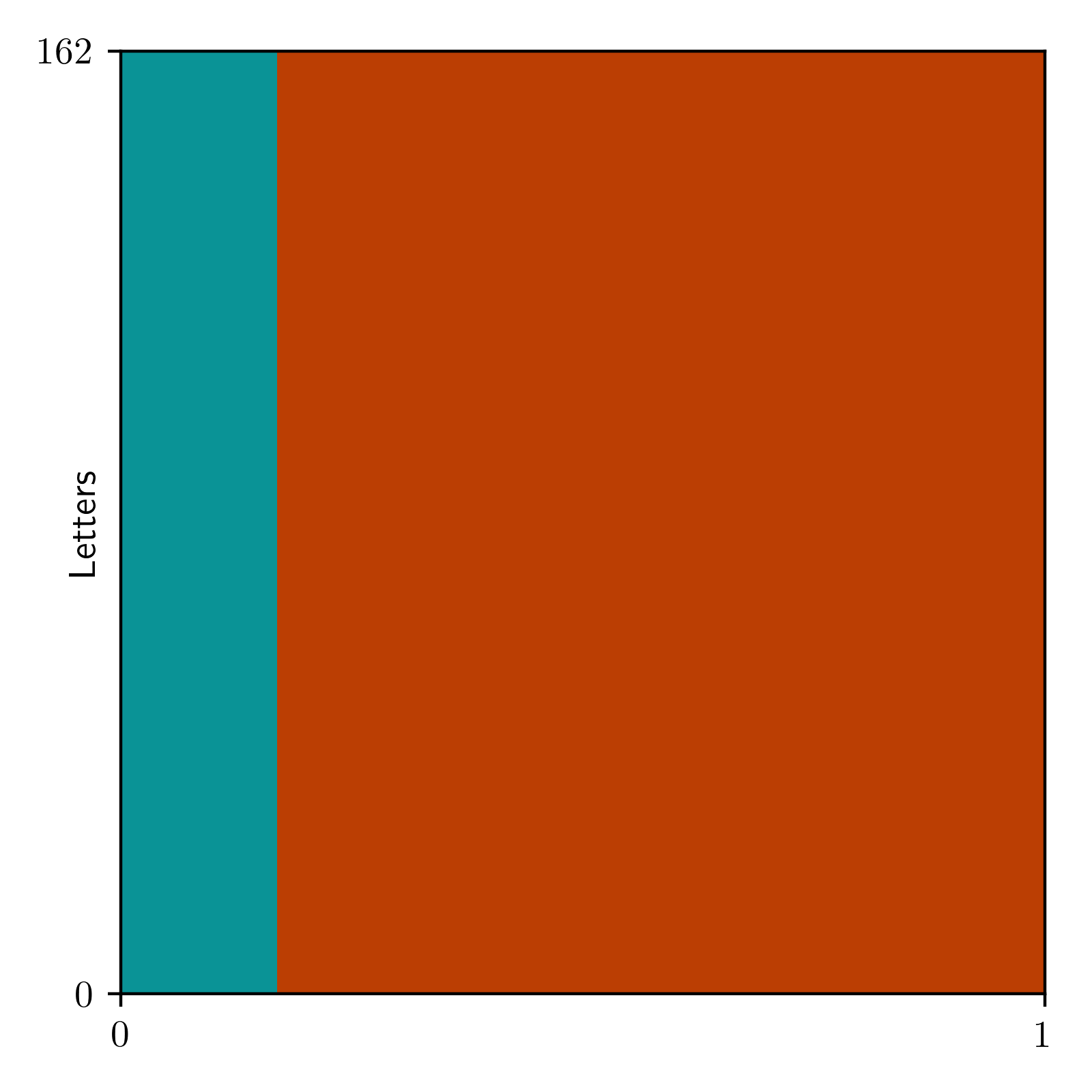}
        \includegraphics[draft=\draft, width=\linewidth]{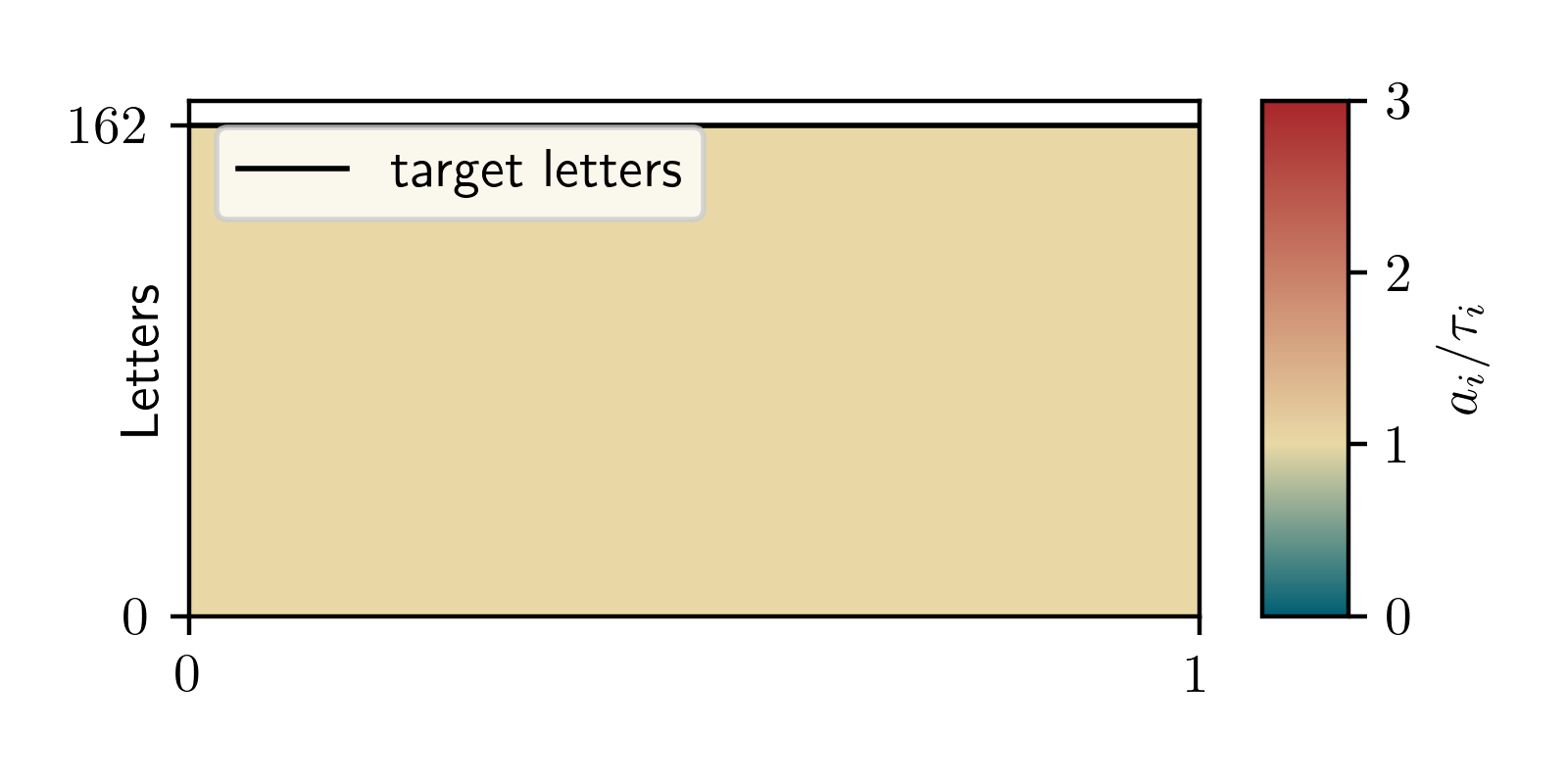}
        \caption{\greq ($t_G = 1$)}
        \label{fig:results_Bremen_Large_greedy_equal}
    \end{subfigure}
    \begin{subfigure}{0.32\textwidth}
        \includegraphics[draft=\draft, width=\linewidth]{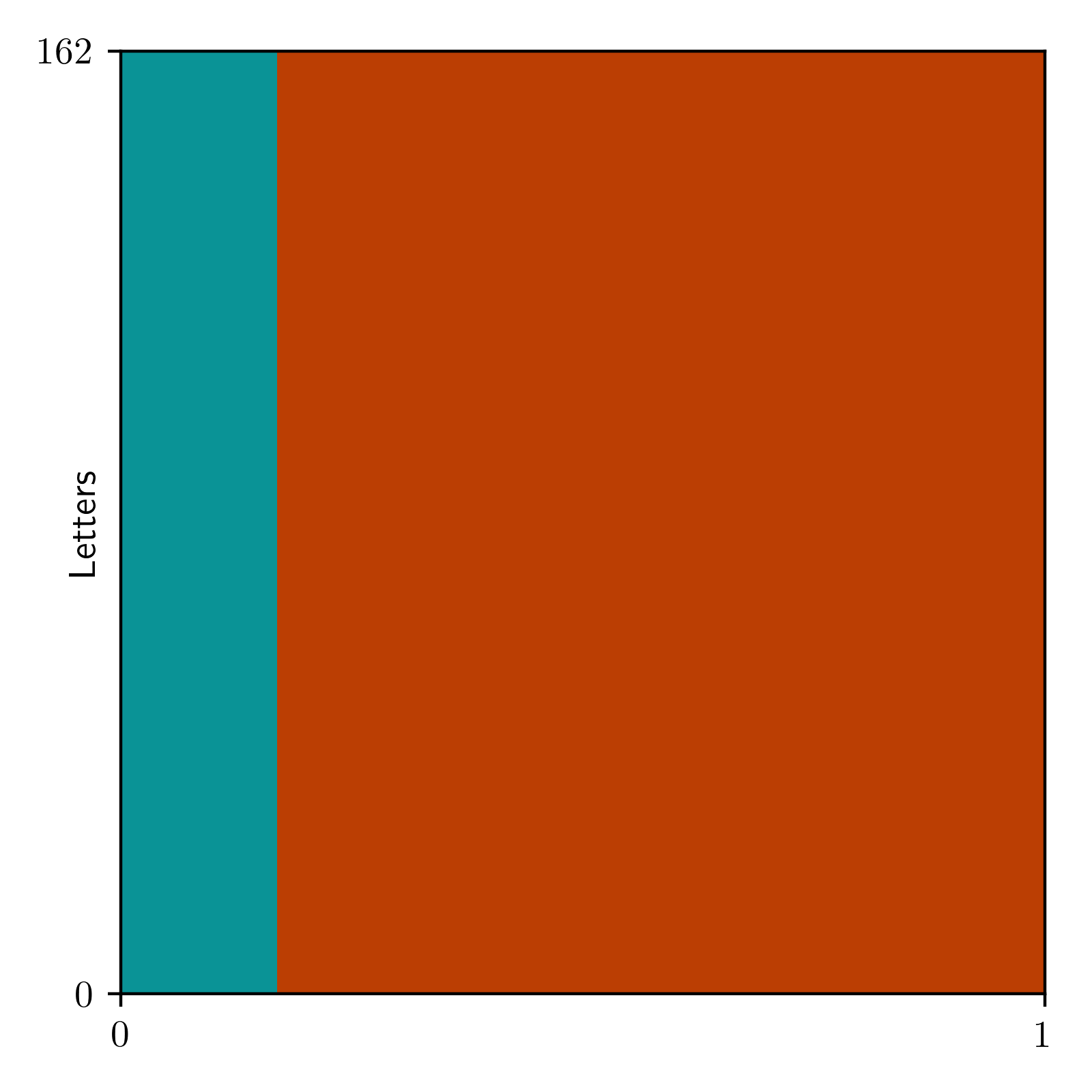}
        \includegraphics[draft=\draft, width=\linewidth]{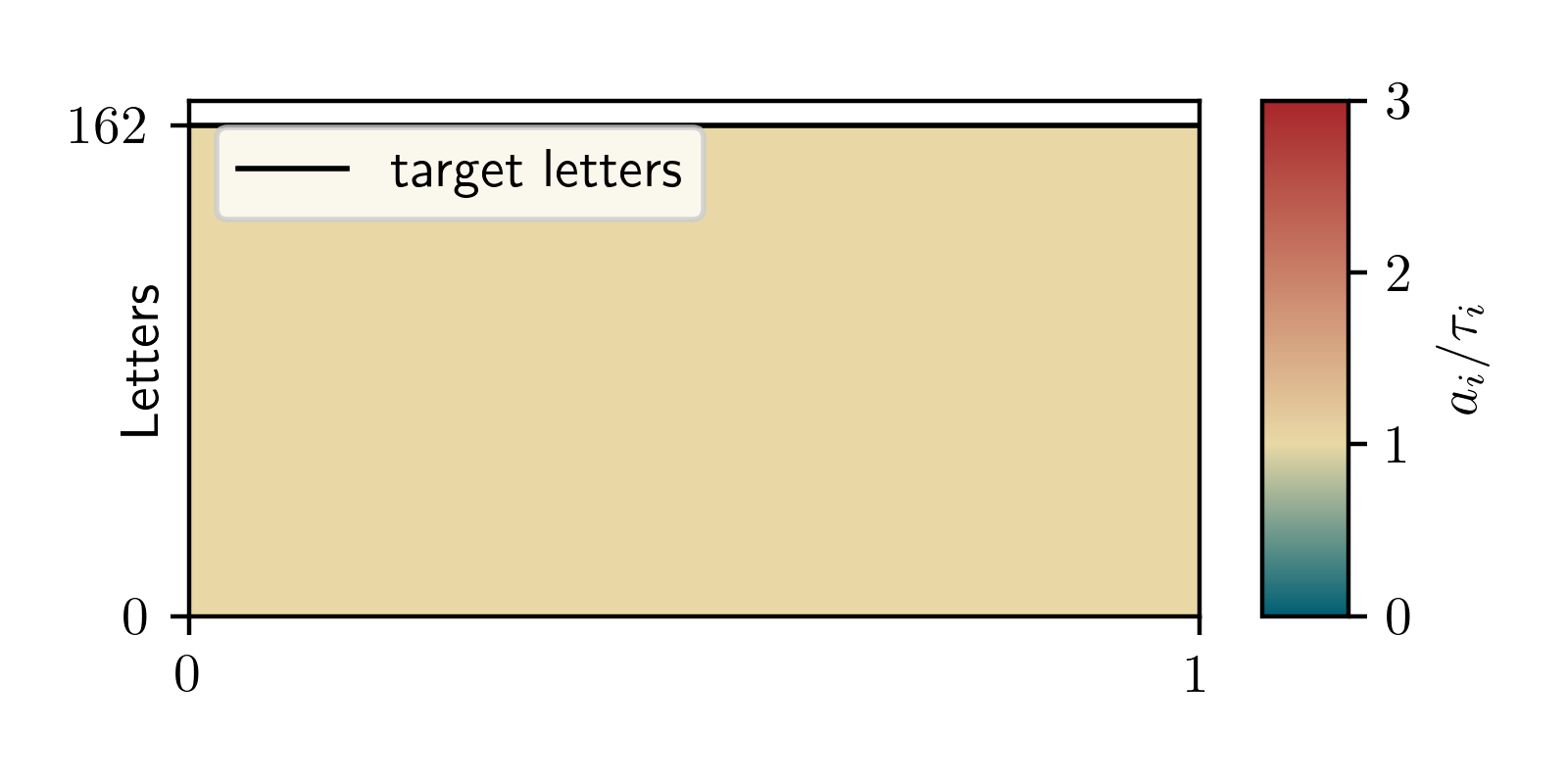}
        \caption{\colgen ($t_G\!=\!1$)}
        \label{fig:results_Bremen_Large_column_generation}
    \end{subfigure}
    \begin{subfigure}{0.32\textwidth}
        \includegraphics[draft=\draft, width=\linewidth]{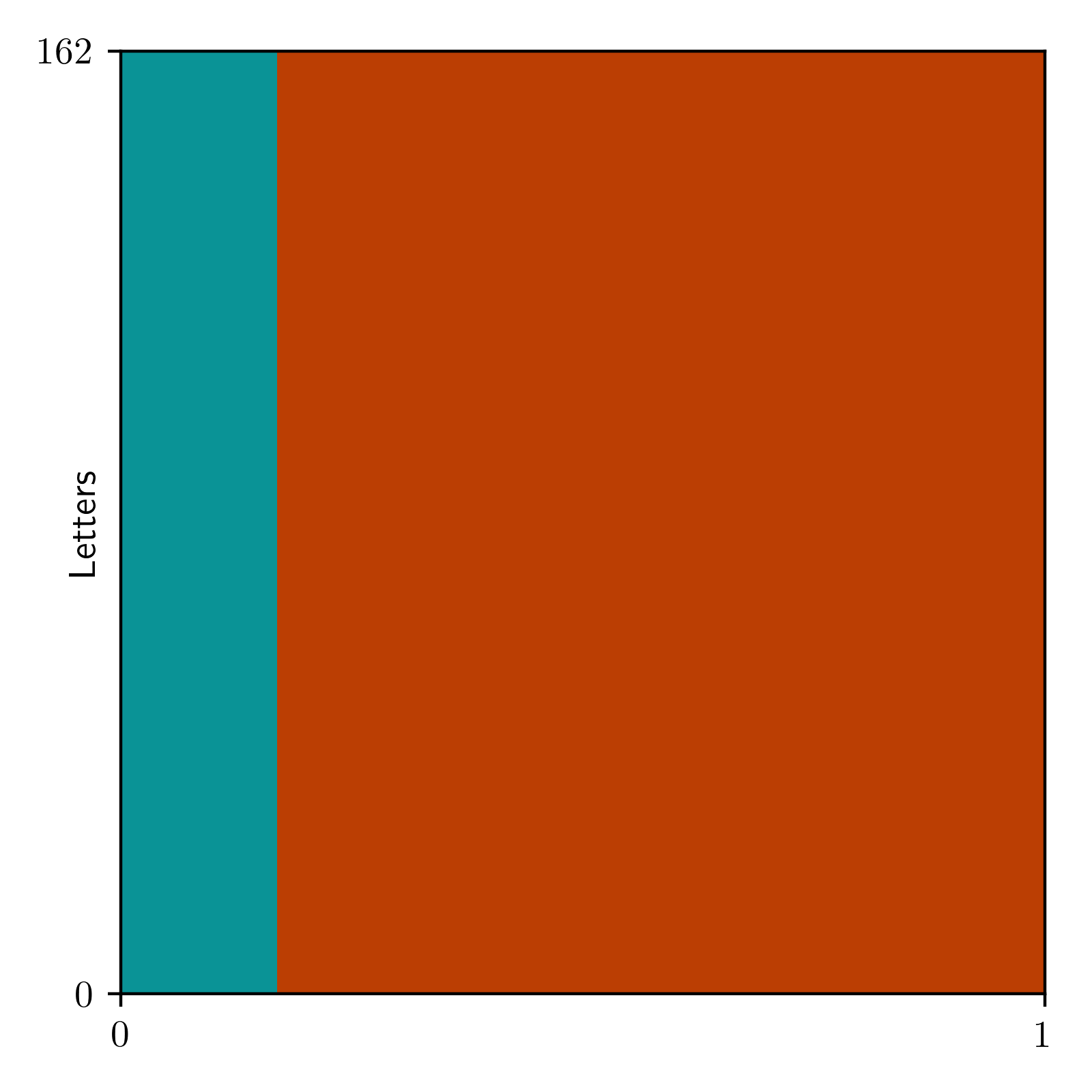}
        \includegraphics[draft=\draft, width=\linewidth]{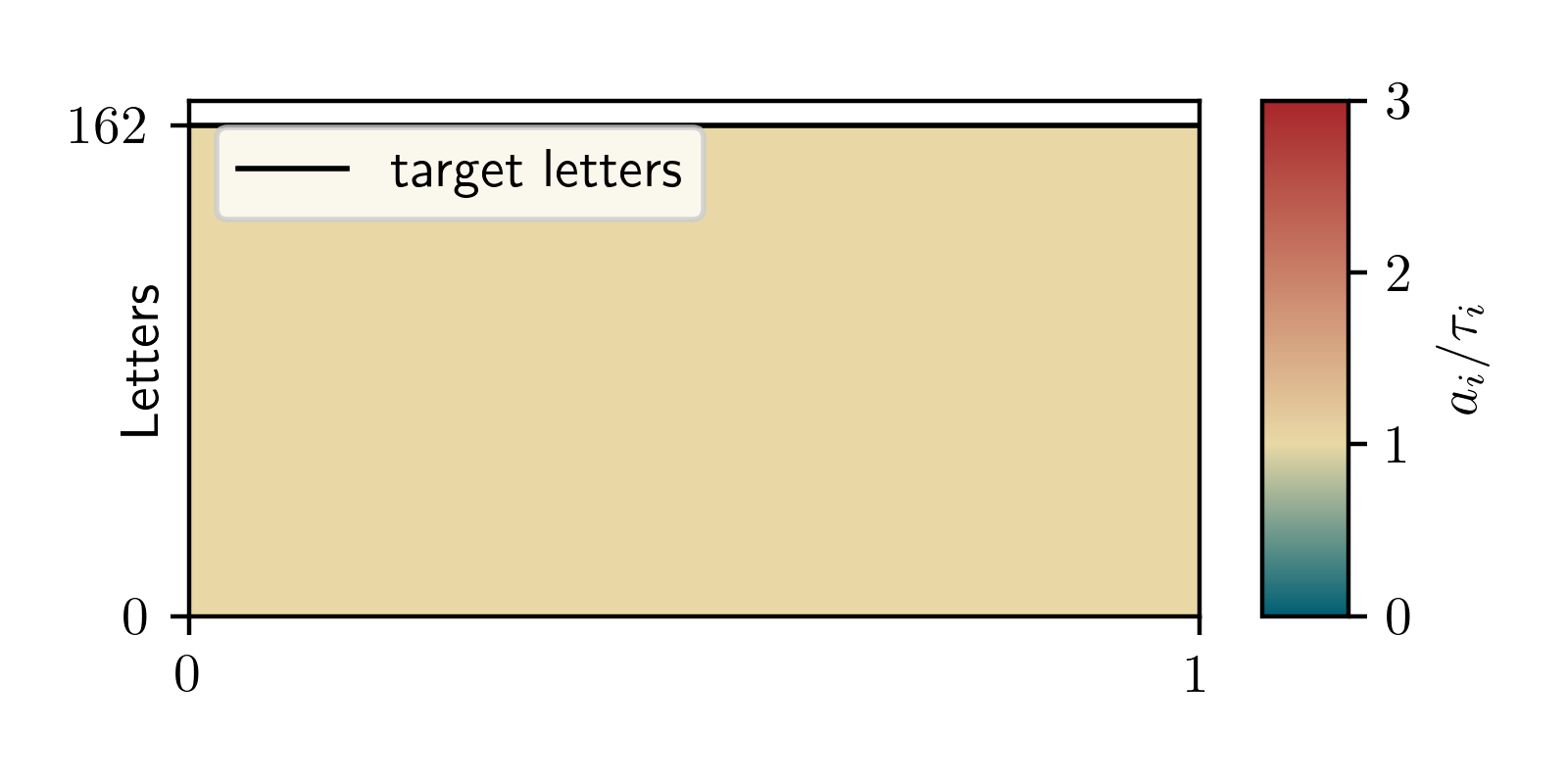}
        \caption{\buckets ($t_G = 1$)}
        \label{fig:results_Bremen_Large_greedy_bucket_fill}
    \end{subfigure}
    \caption{Large municipalities of Bremen ($\ell_G = 162$)}
    \label{fig:results_Bremen_Large}
\end{figure}

\begin{figure}
    \centering
    \begin{subfigure}{0.32\textwidth}
        \includegraphics[draft=\draft, width=\linewidth]{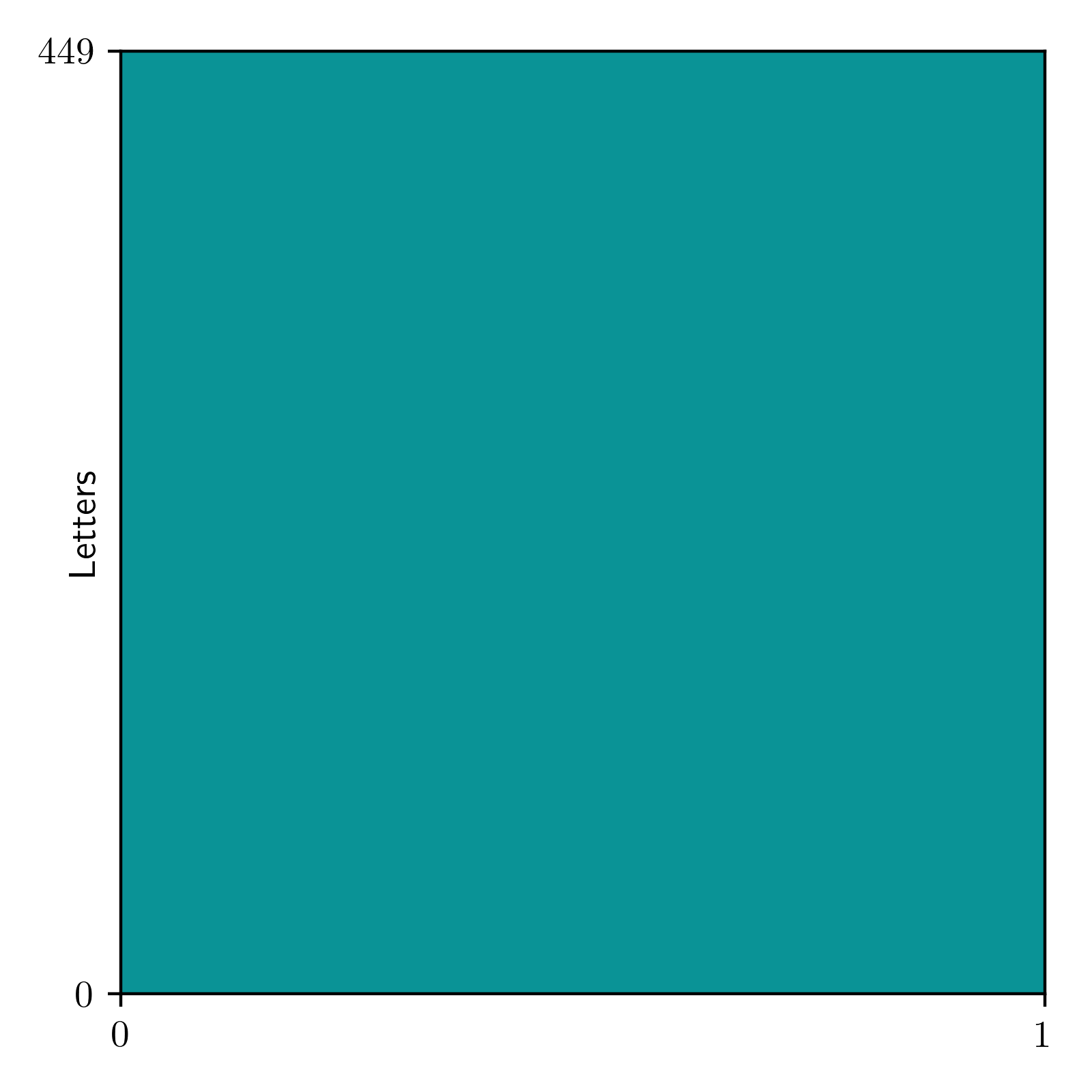}
        \includegraphics[draft=\draft, width=\linewidth]{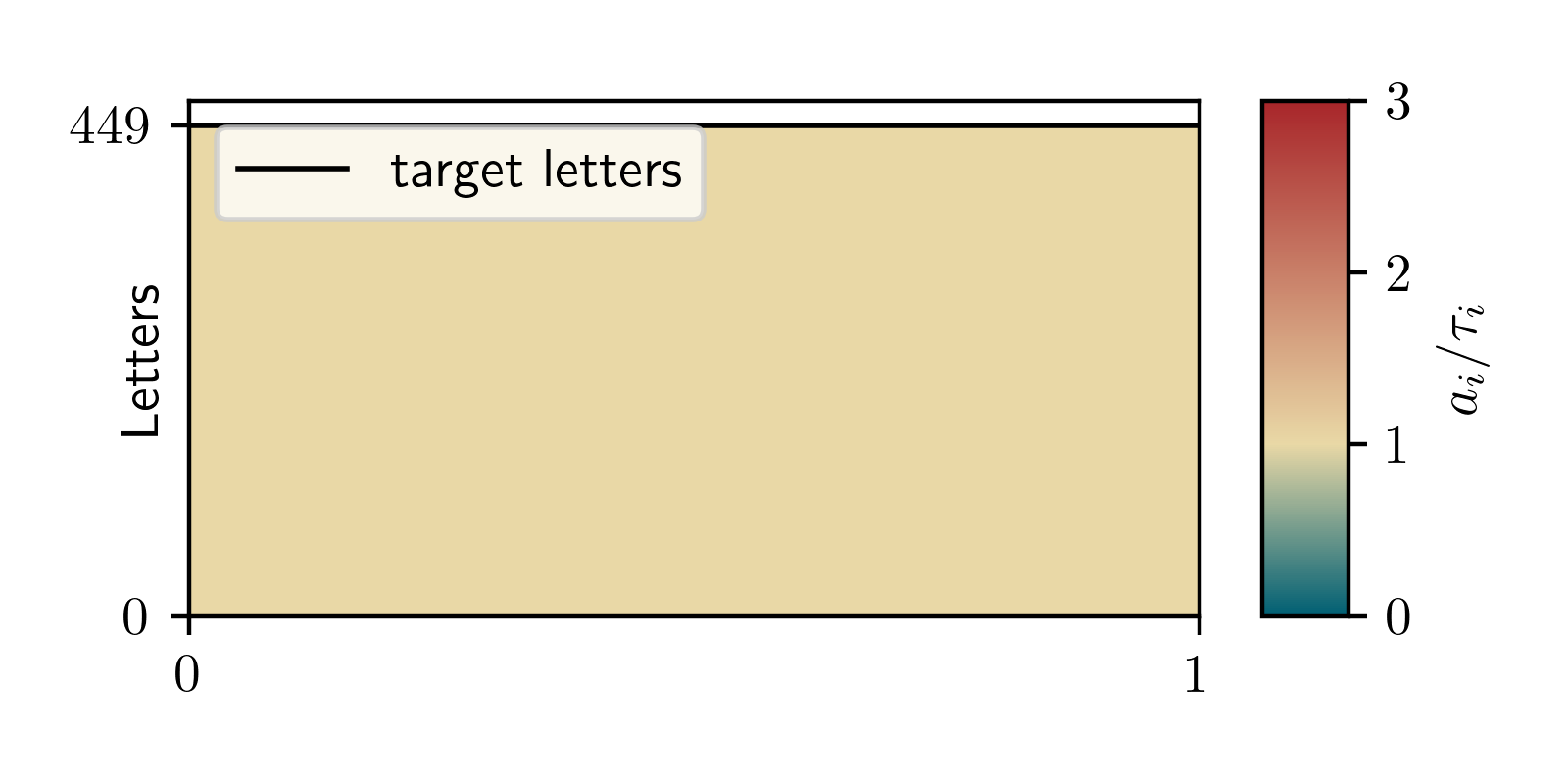}
        \caption{\greq ($t_G = 1$)}
        \label{fig:results_Hamburg_Large_greedy_equal}
    \end{subfigure}
    \begin{subfigure}{0.32\textwidth}
        \includegraphics[draft=\draft, width=\linewidth]{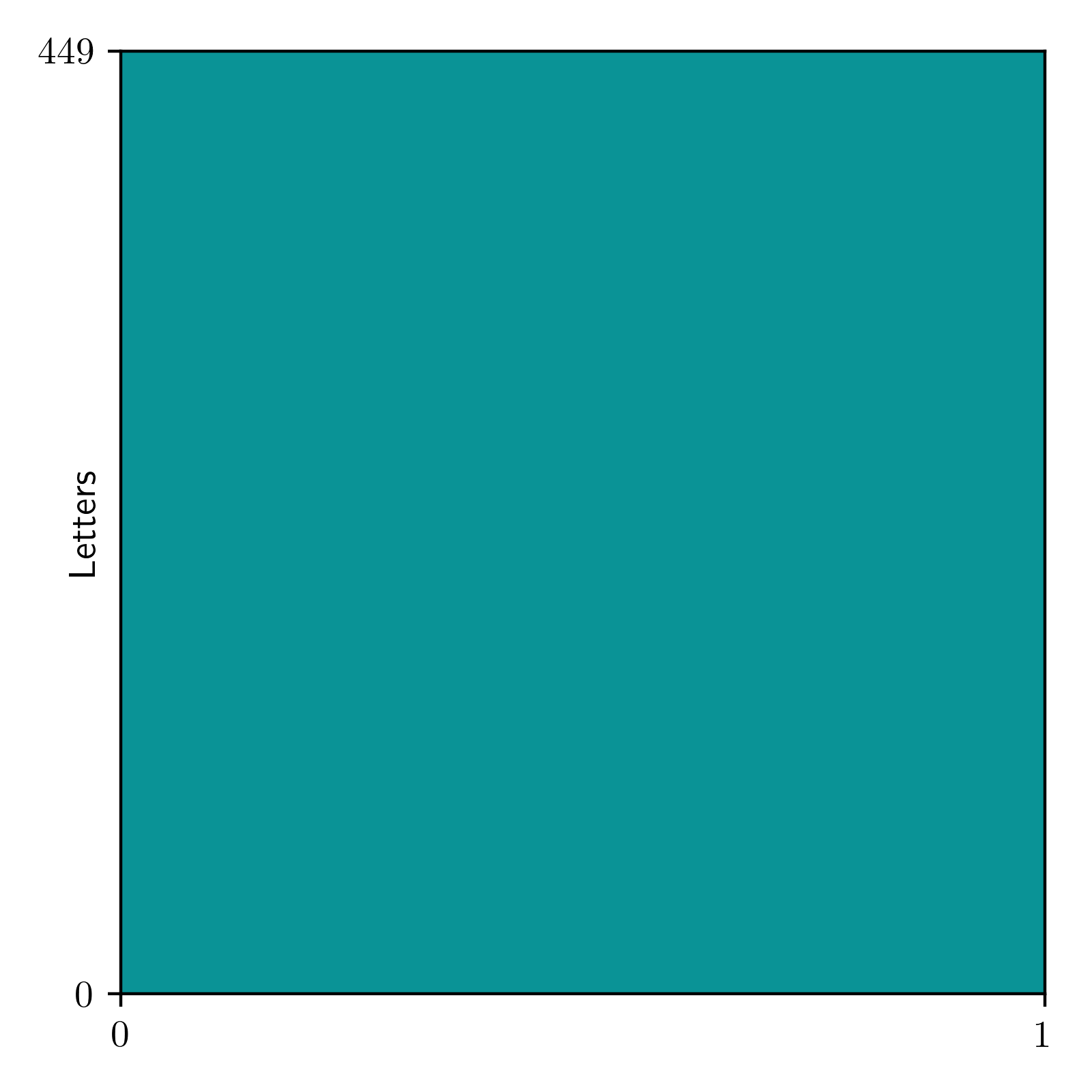}
        \includegraphics[draft=\draft, width=\linewidth]{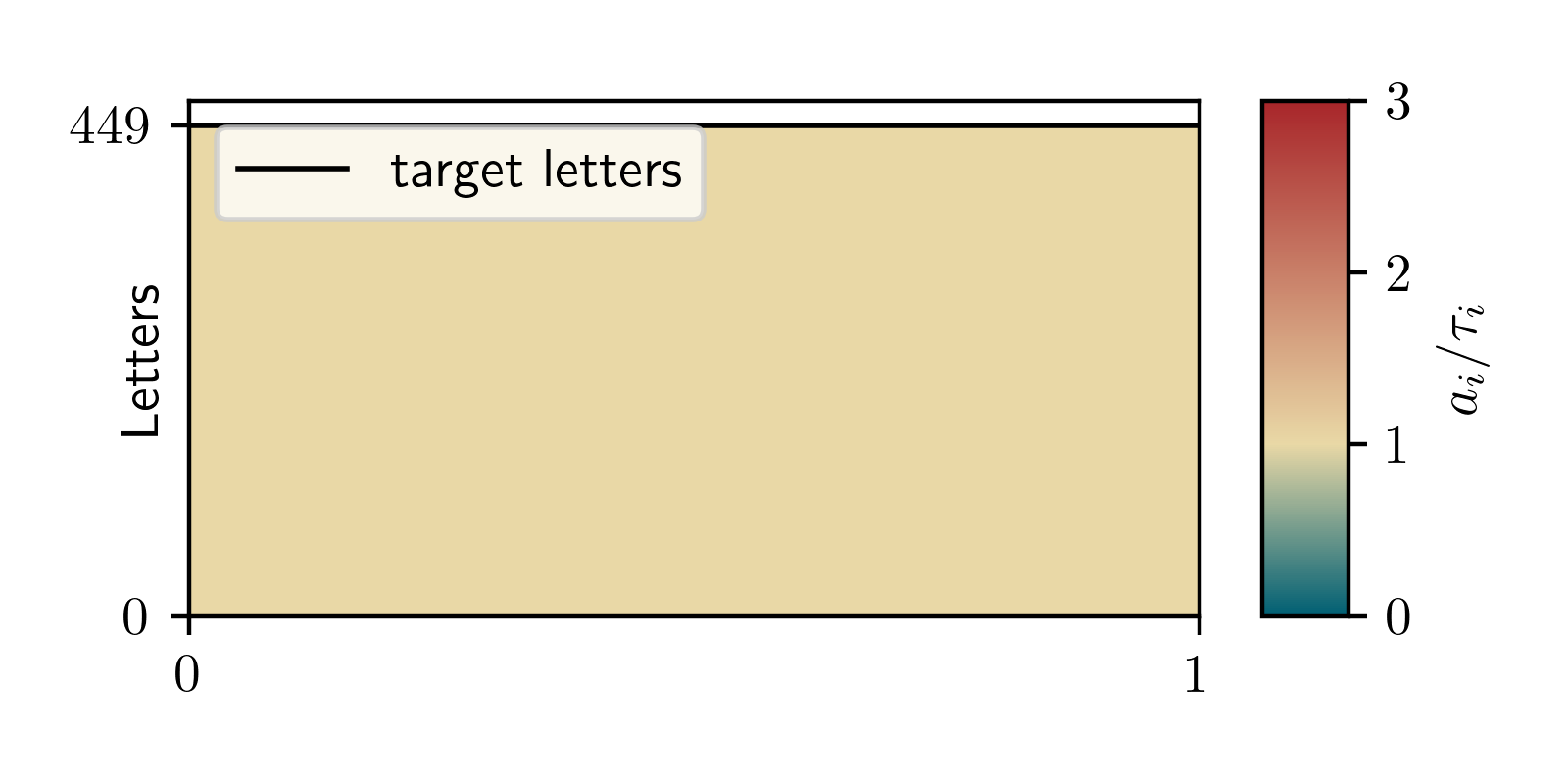}
        \caption{\colgen ($t_G\!=\!1$)}
        \label{fig:results_Hamburg_Large_column_generation}
    \end{subfigure}
    \begin{subfigure}{0.32\textwidth}
        \includegraphics[draft=\draft, width=\linewidth]{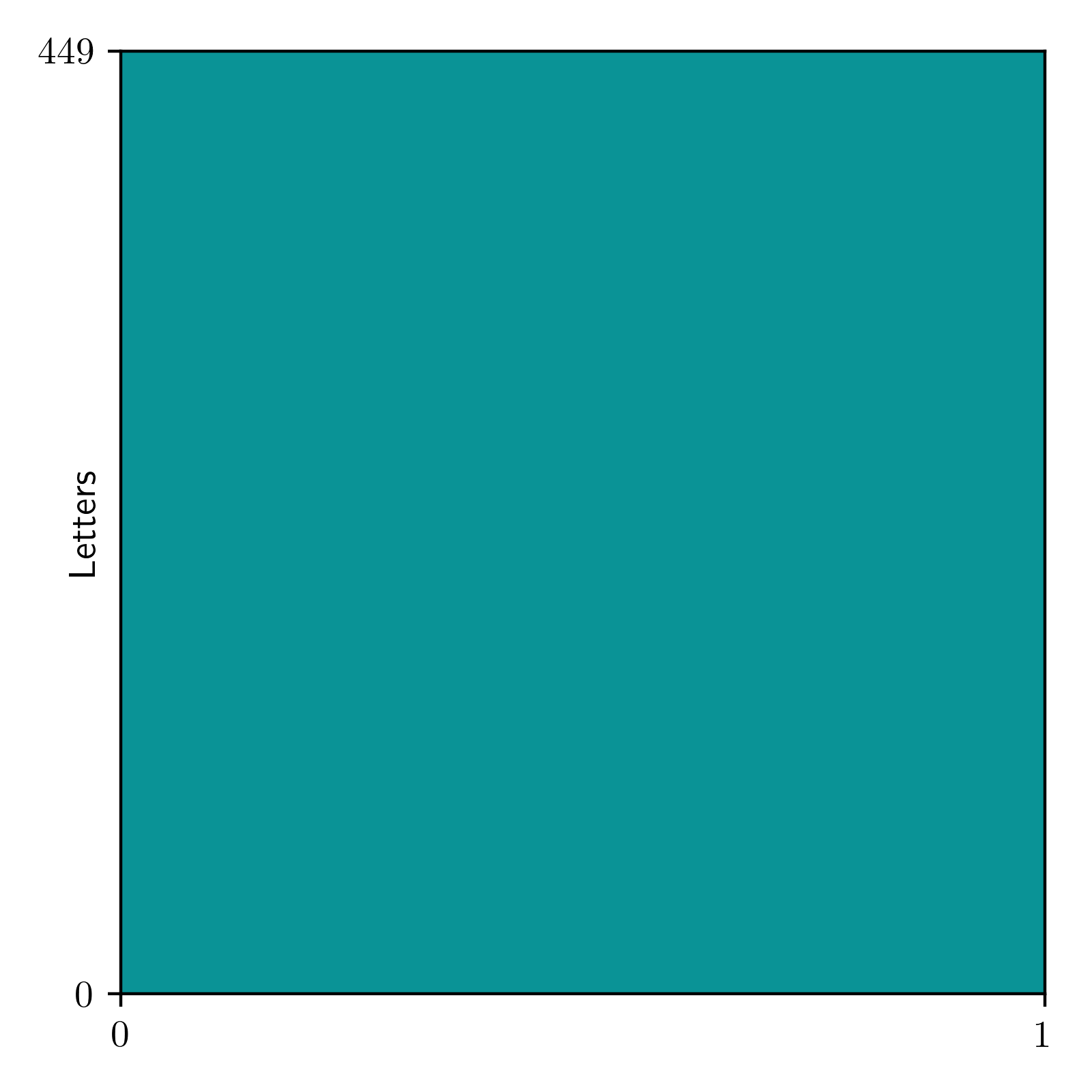}
        \includegraphics[draft=\draft, width=\linewidth]{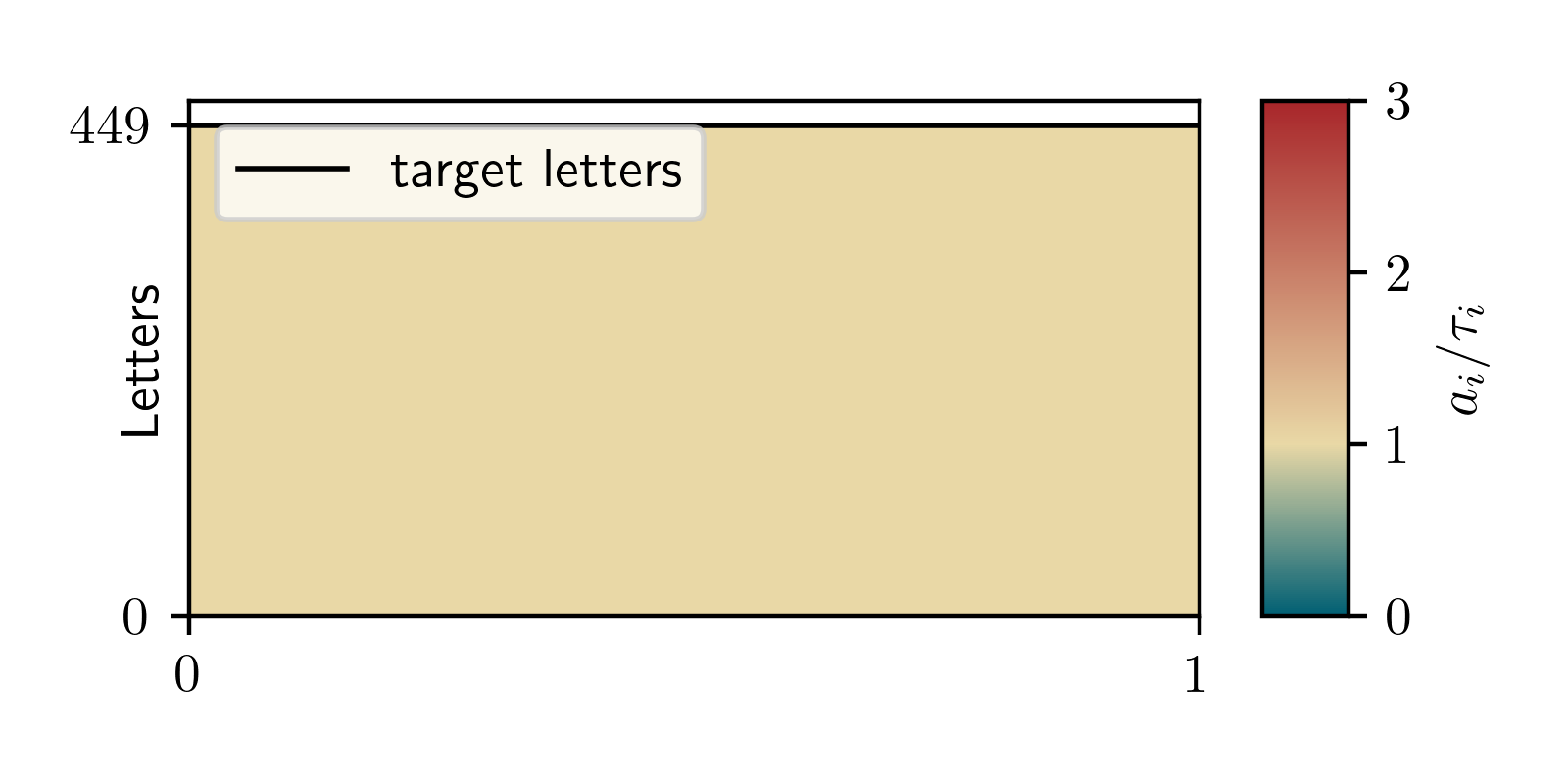}
        \caption{\buckets ($t_G = 1$)}
        \label{fig:results_Hamburg_Large_greedy_bucket_fill}
    \end{subfigure}
    \caption{Large municipalities of Hamburg ($\ell_G = 449$)}
    \label{fig:results_Hamburg_Large}
\end{figure}

\begin{figure}
    \centering
    \begin{subfigure}{0.32\textwidth}
        \includegraphics[draft=\draft, width=\linewidth]{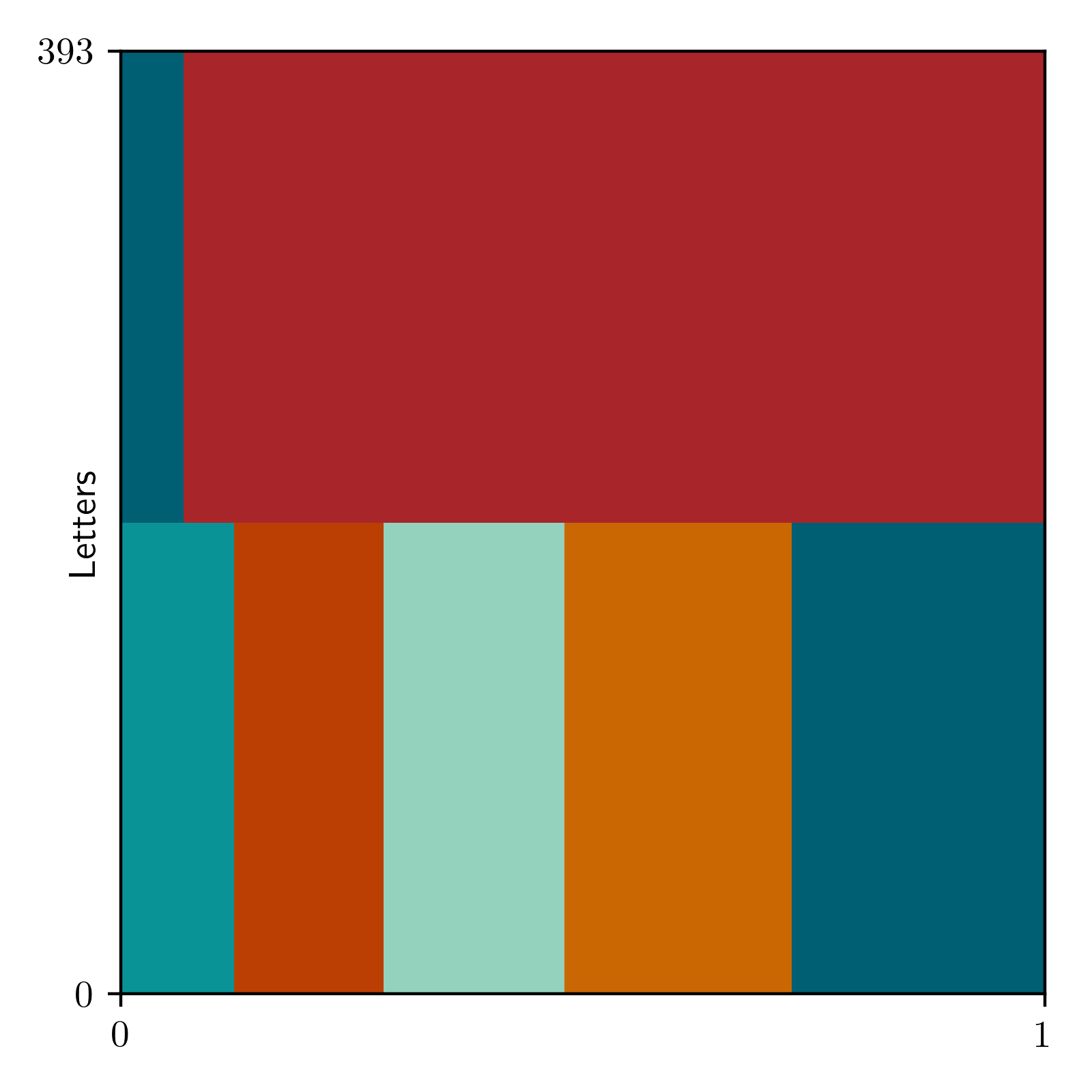}
        \includegraphics[draft=\draft, width=\linewidth]{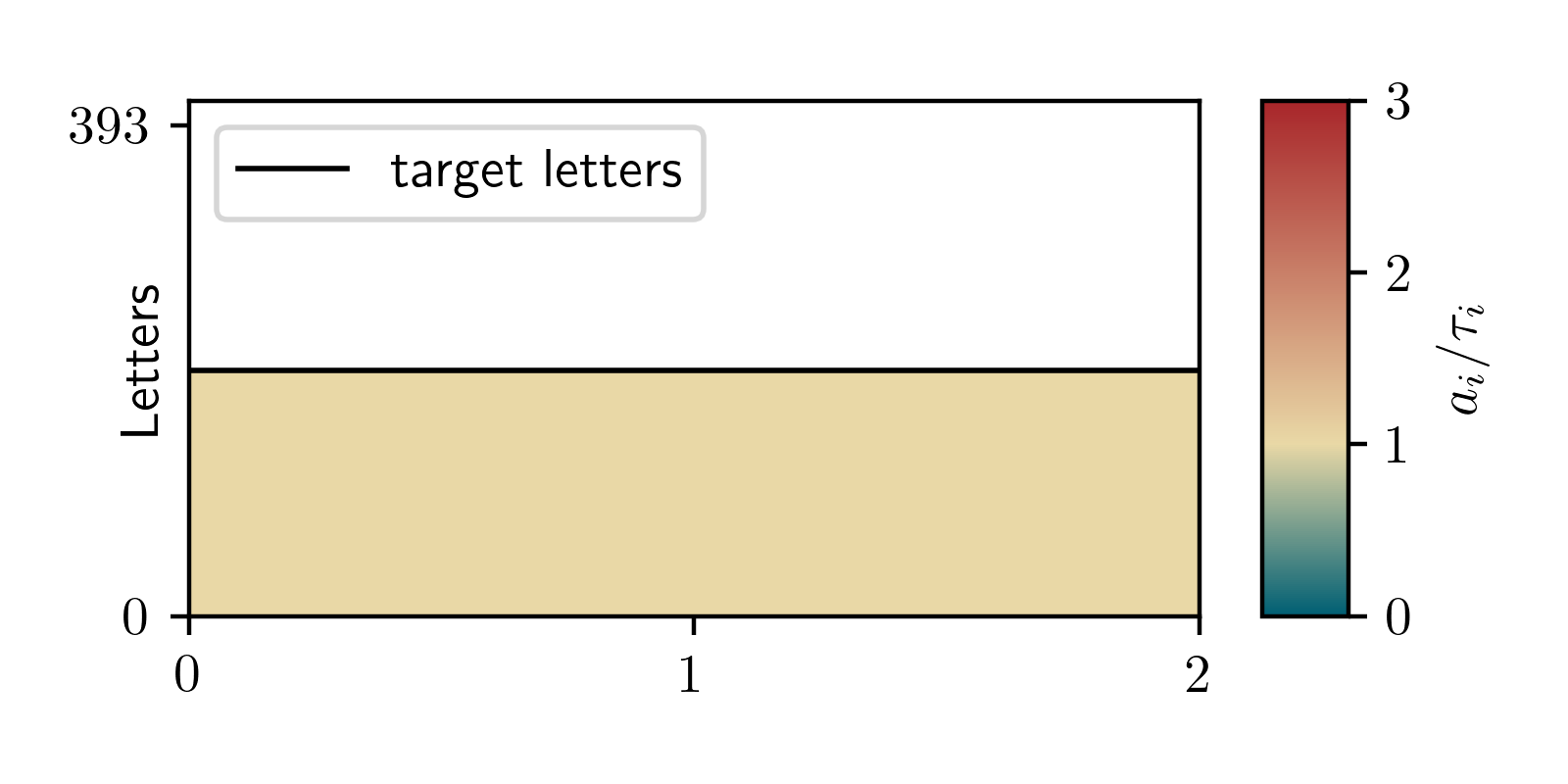}
        \caption{\greq ($t_G = 2$)}
        \label{fig:results_Hessen_Large_greedy_equal}
    \end{subfigure}
    \begin{subfigure}{0.32\textwidth}
        \includegraphics[draft=\draft, width=\linewidth]{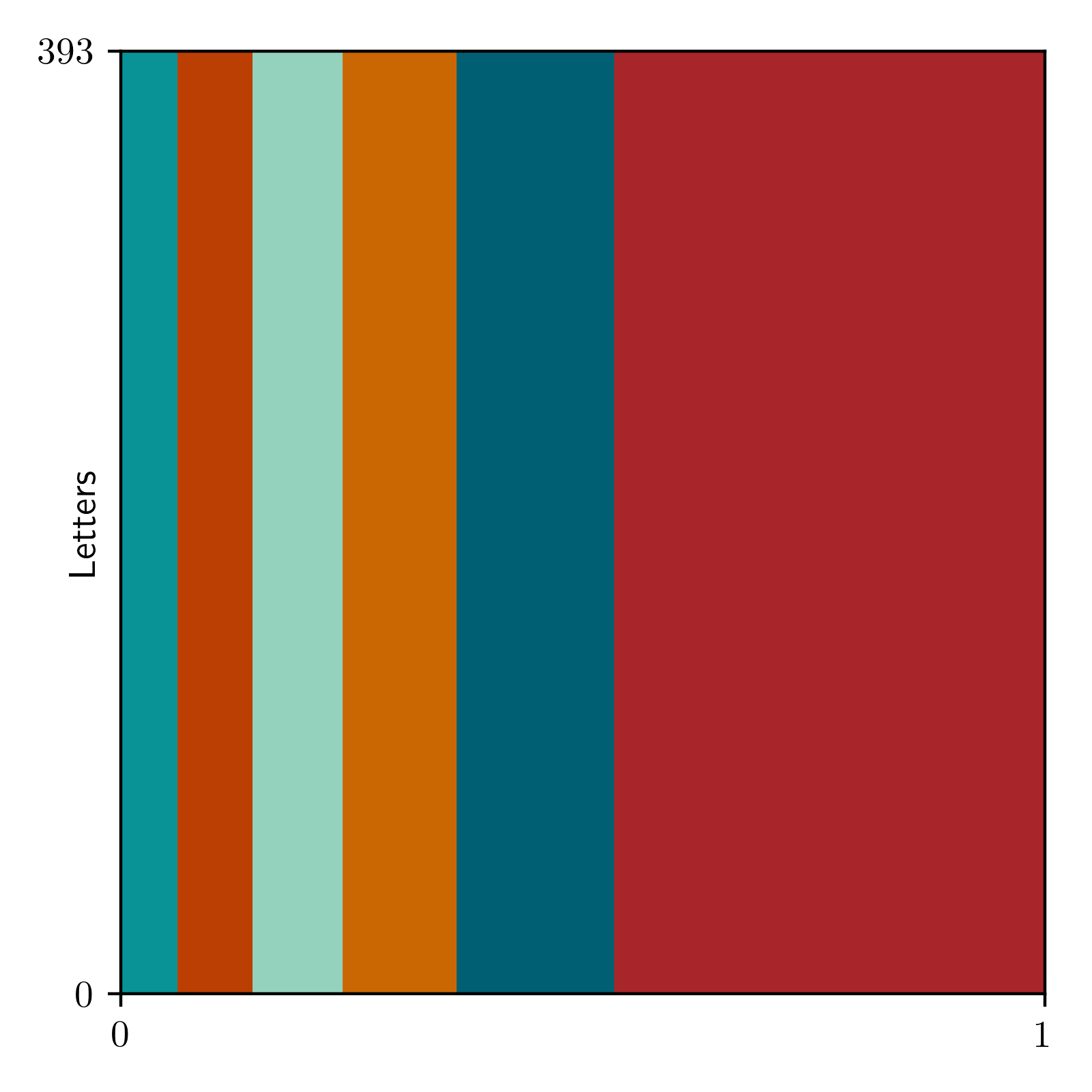}
        \includegraphics[draft=\draft, width=\linewidth]{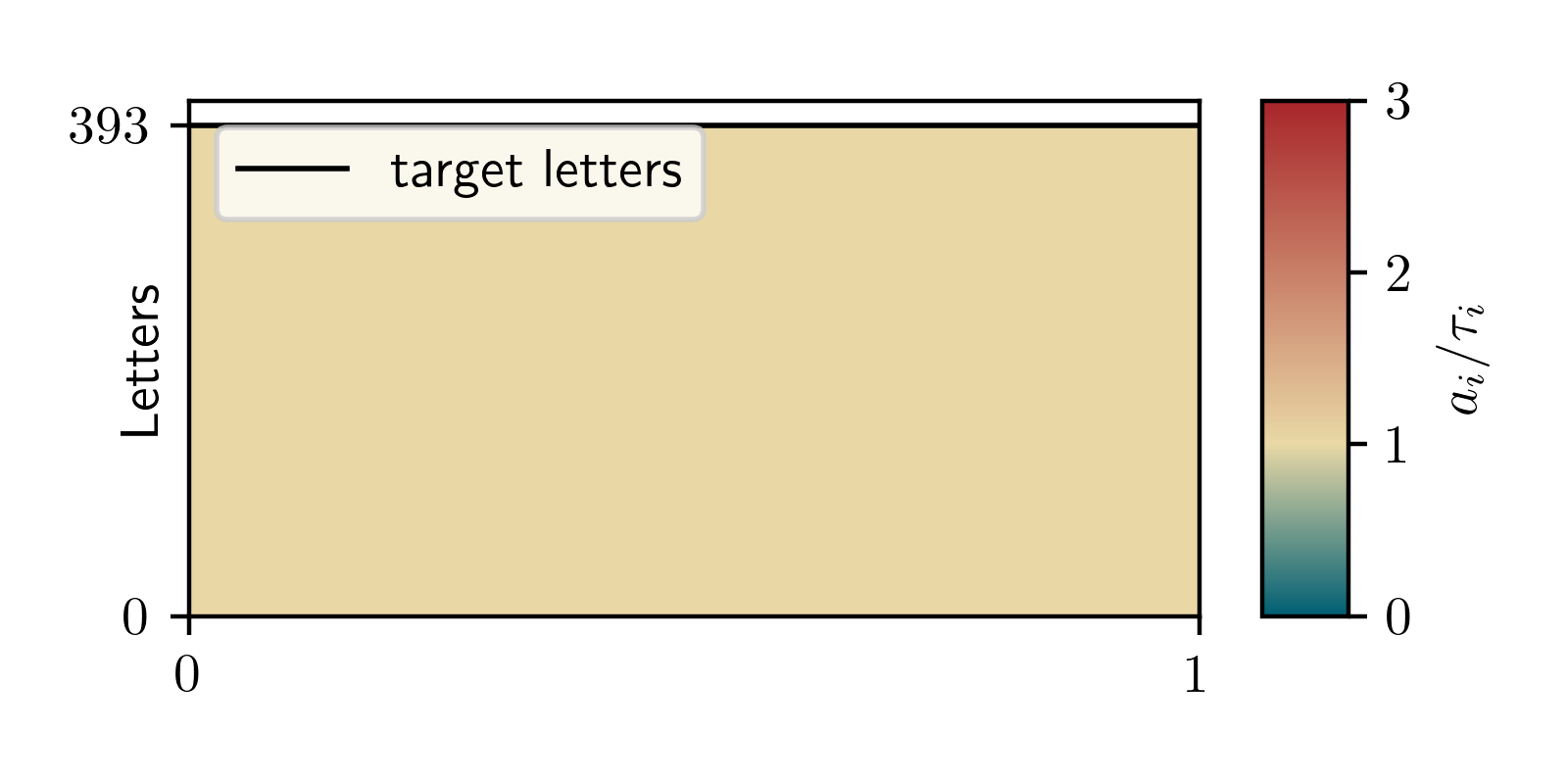}
        \caption{\colgen ($t_G\!=\!1$)}
        \label{fig:results_Hessen_Large_column_generation}
    \end{subfigure}
    \begin{subfigure}{0.32\textwidth}
        \includegraphics[draft=\draft, width=\linewidth]{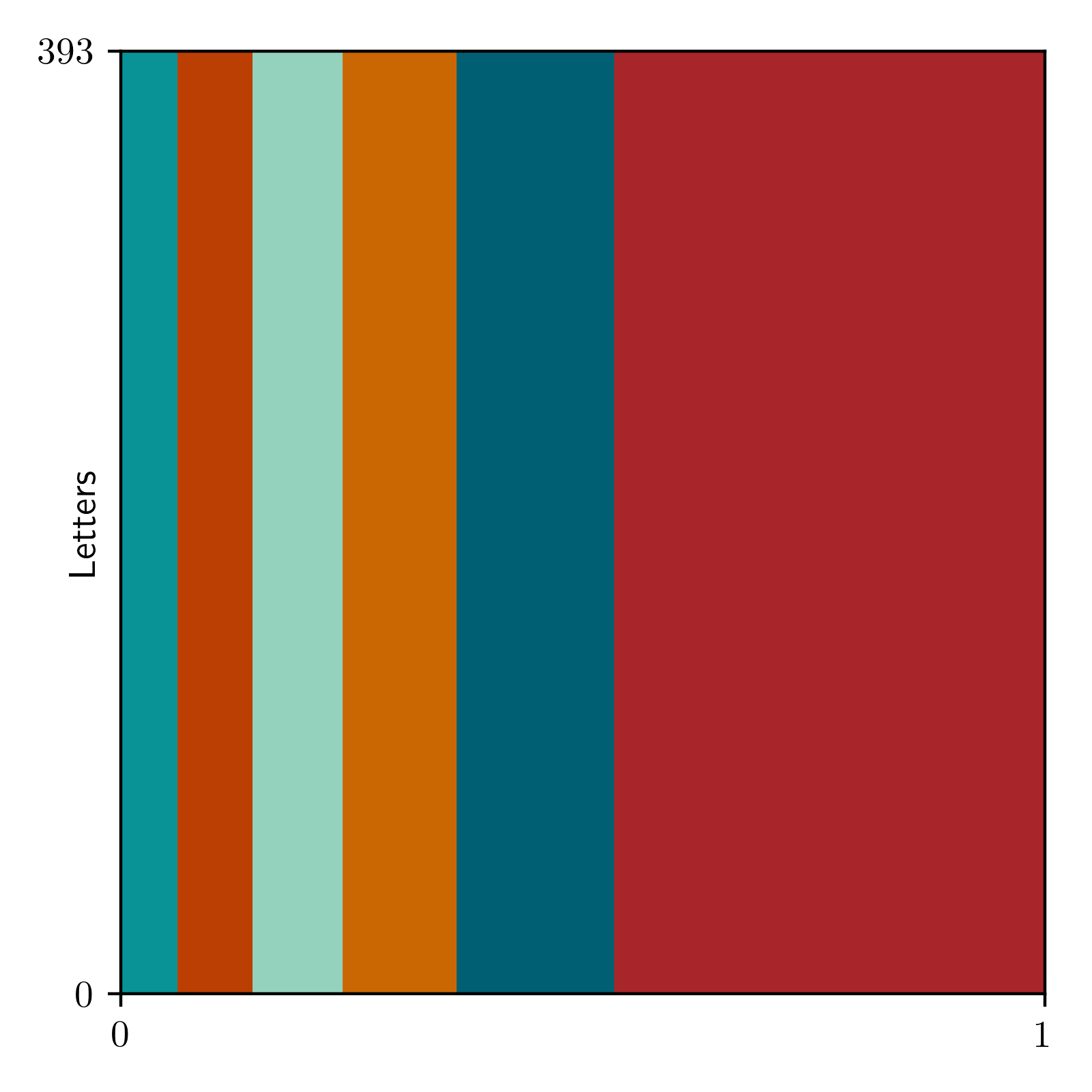}
        \includegraphics[draft=\draft, width=\linewidth]{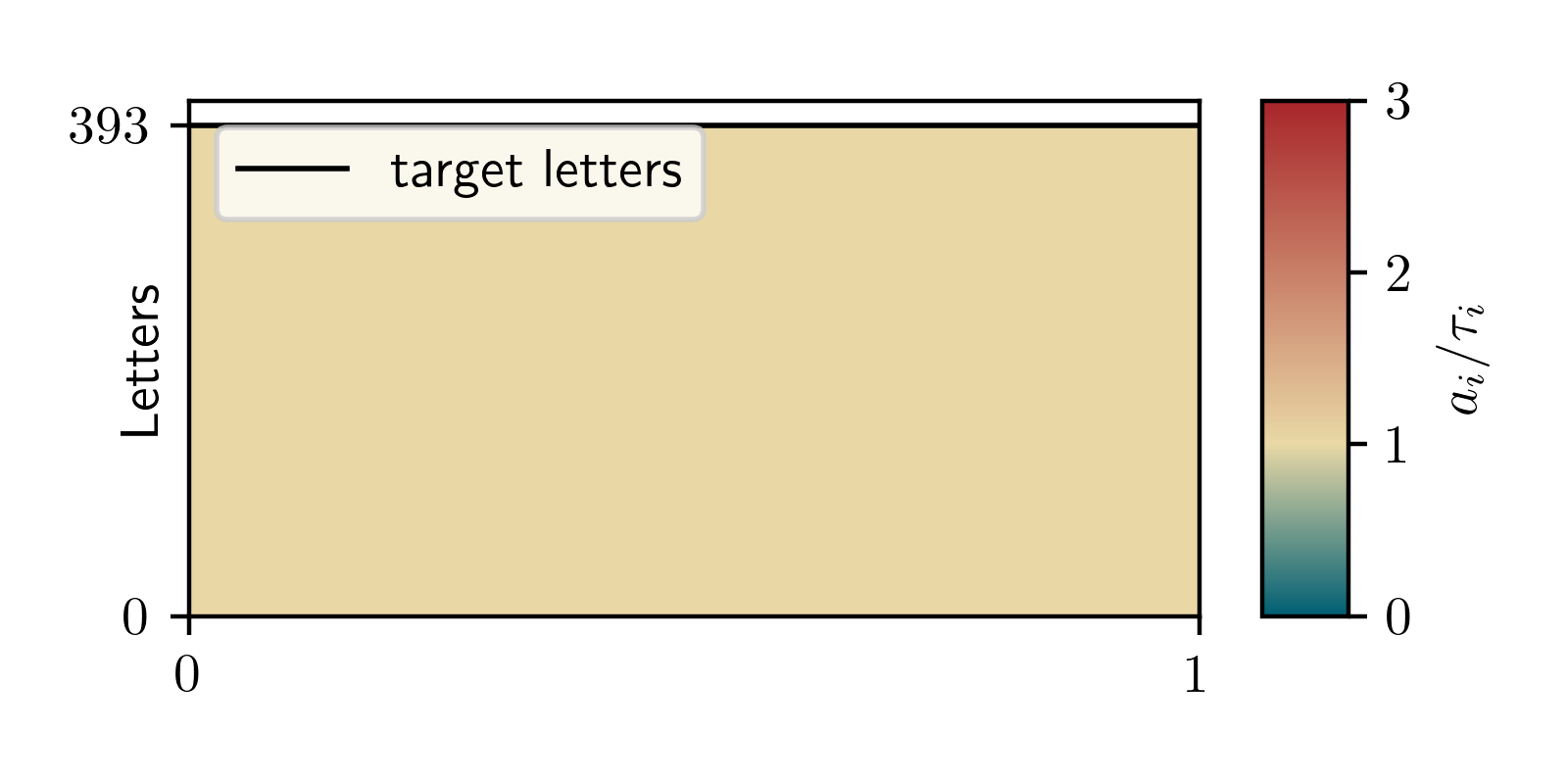}
        \caption{\buckets ($t_G = 1$)}
        \label{fig:results_Hessen_Large_greedy_bucket_fill}
    \end{subfigure}
    \caption{Large municipalities of Hessen ($\ell_G = 393$)}
    \label{fig:results_Hessen_Large}
\end{figure} 

\begin{figure}
    \centering
    \begin{subfigure}{0.32\textwidth}
        \includegraphics[draft=\draft, width=\linewidth]{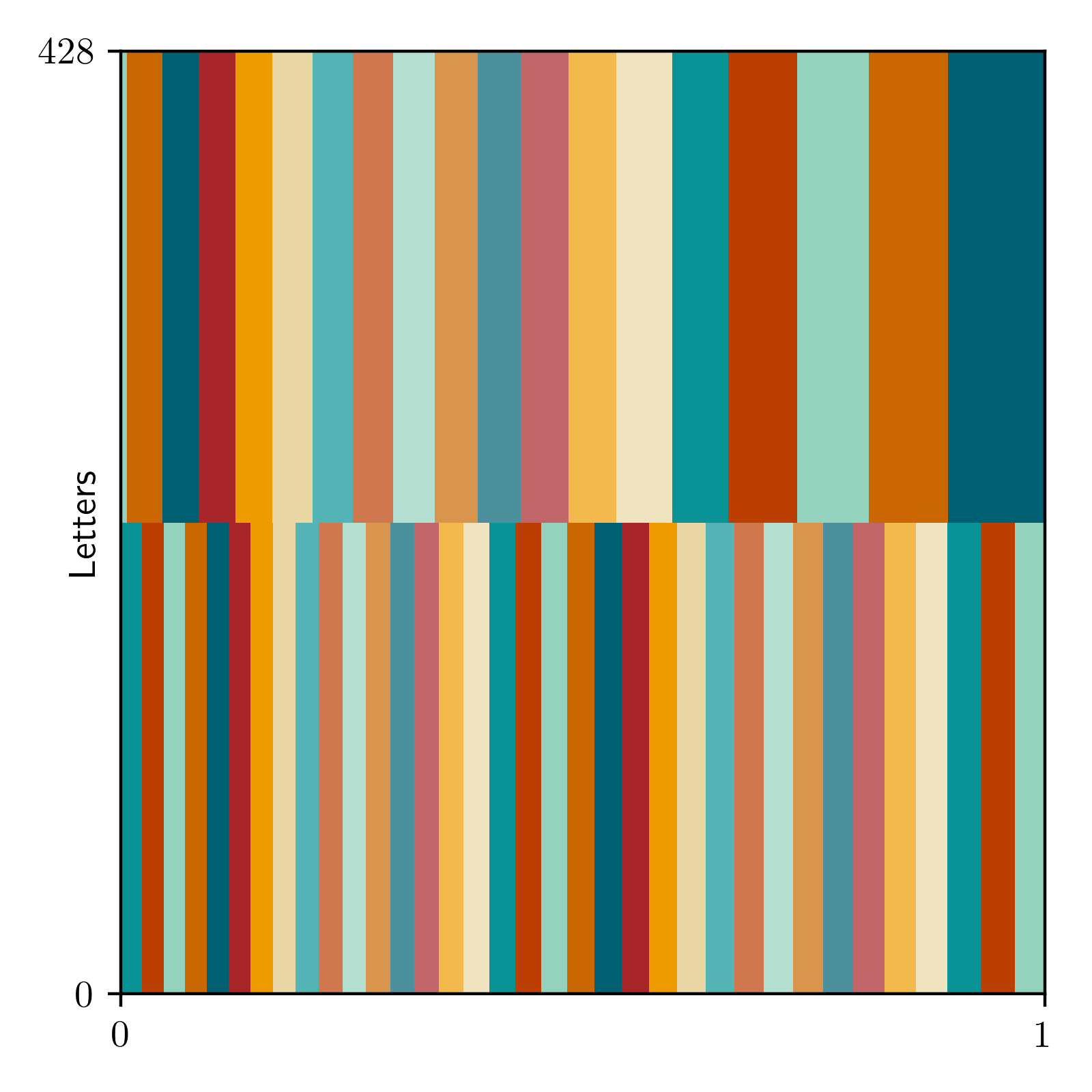}
        \includegraphics[draft=\draft, width=\linewidth]{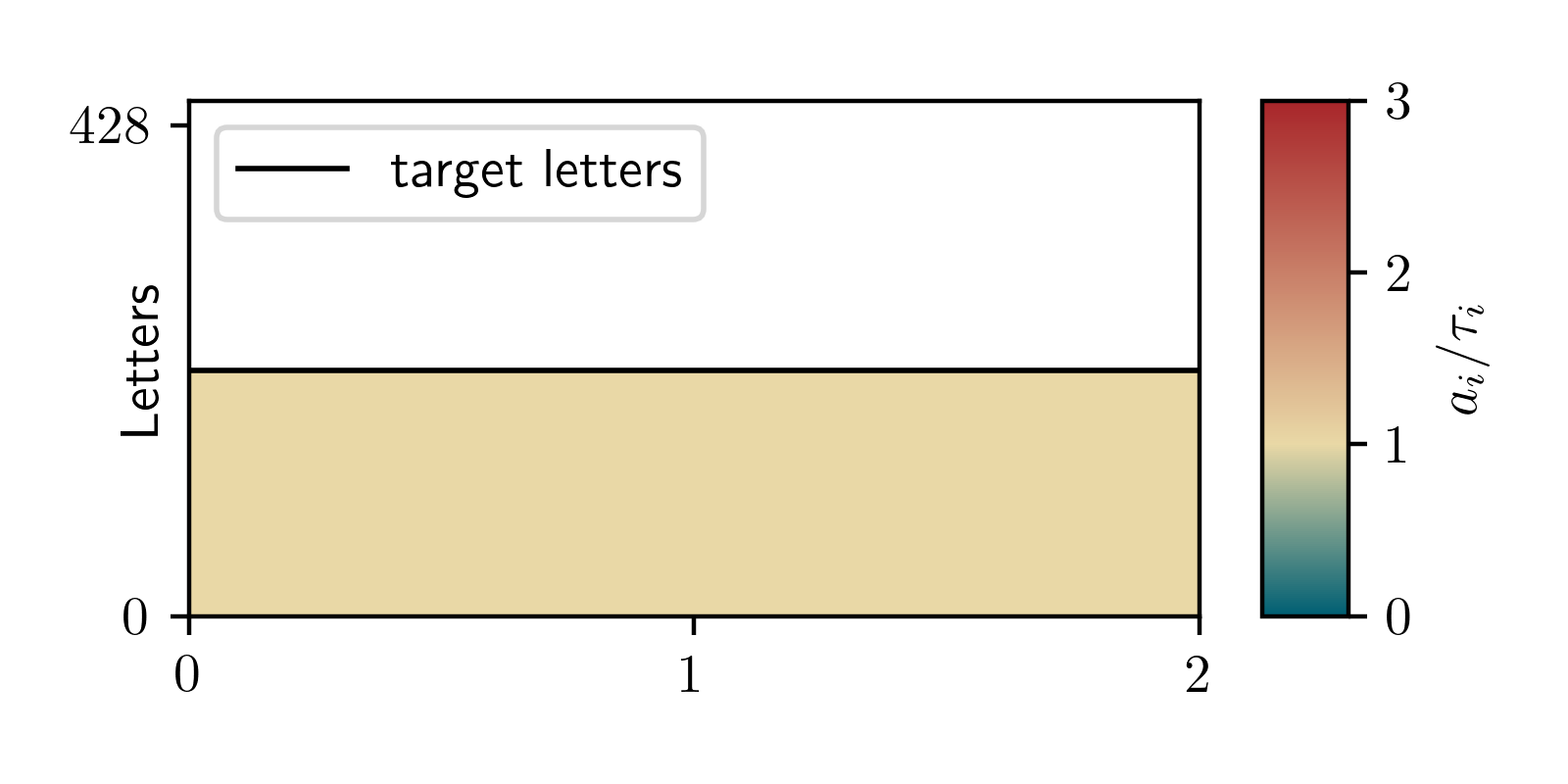}
        \caption{\greq ($t_G = 2$)}
        \label{fig:results_Hessen_Medium_greedy_equal}
    \end{subfigure}
    \begin{subfigure}{0.32\textwidth}
        \includegraphics[draft=\draft, width=\linewidth]{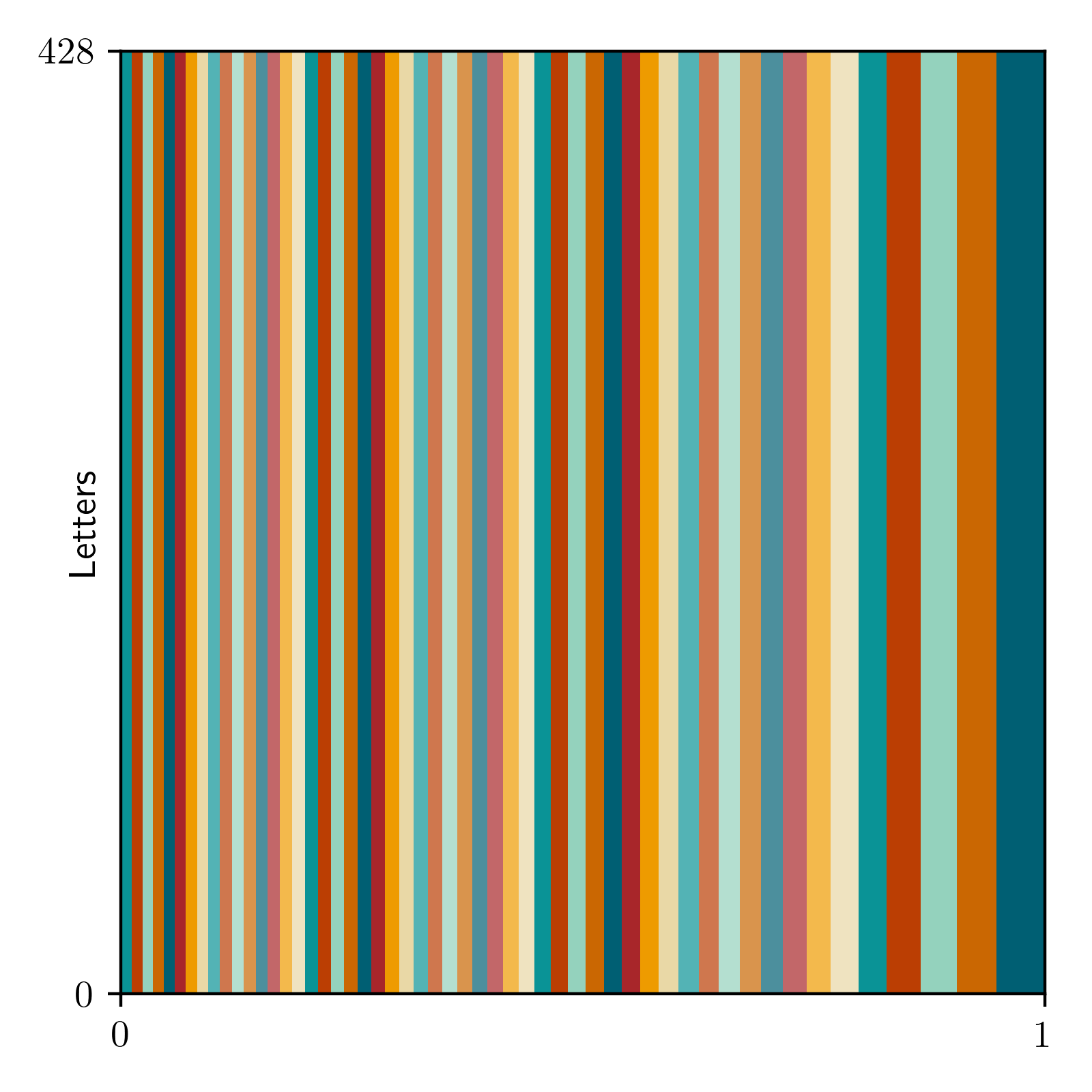}
        \includegraphics[draft=\draft, width=\linewidth]{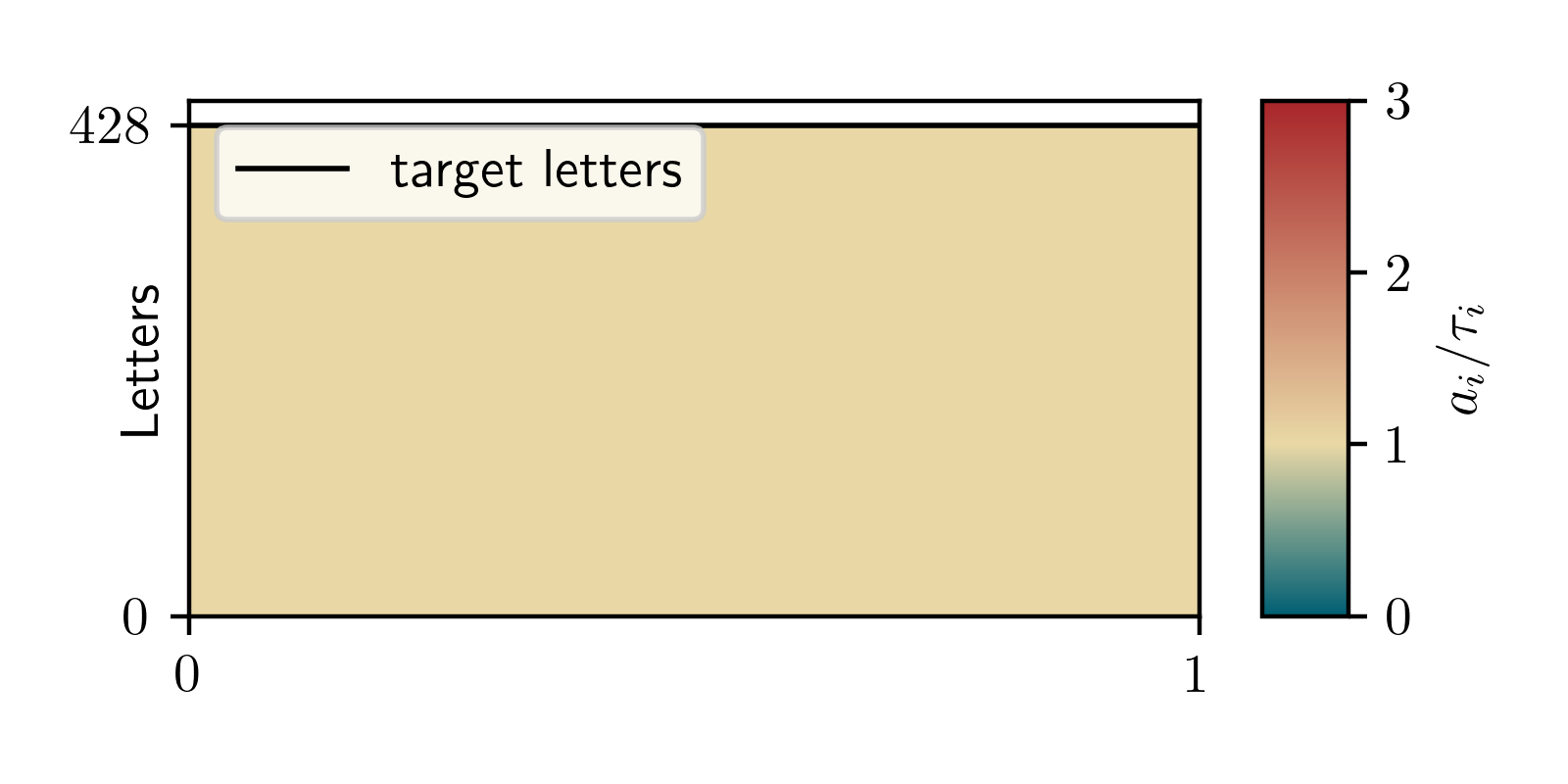}
        \caption{\colgen ($t_G\!=\!1$)}
        \label{fig:results_Hessen_Medium_column_generation}
    \end{subfigure}
    \begin{subfigure}{0.32\textwidth}
        \includegraphics[draft=\draft, width=\linewidth]{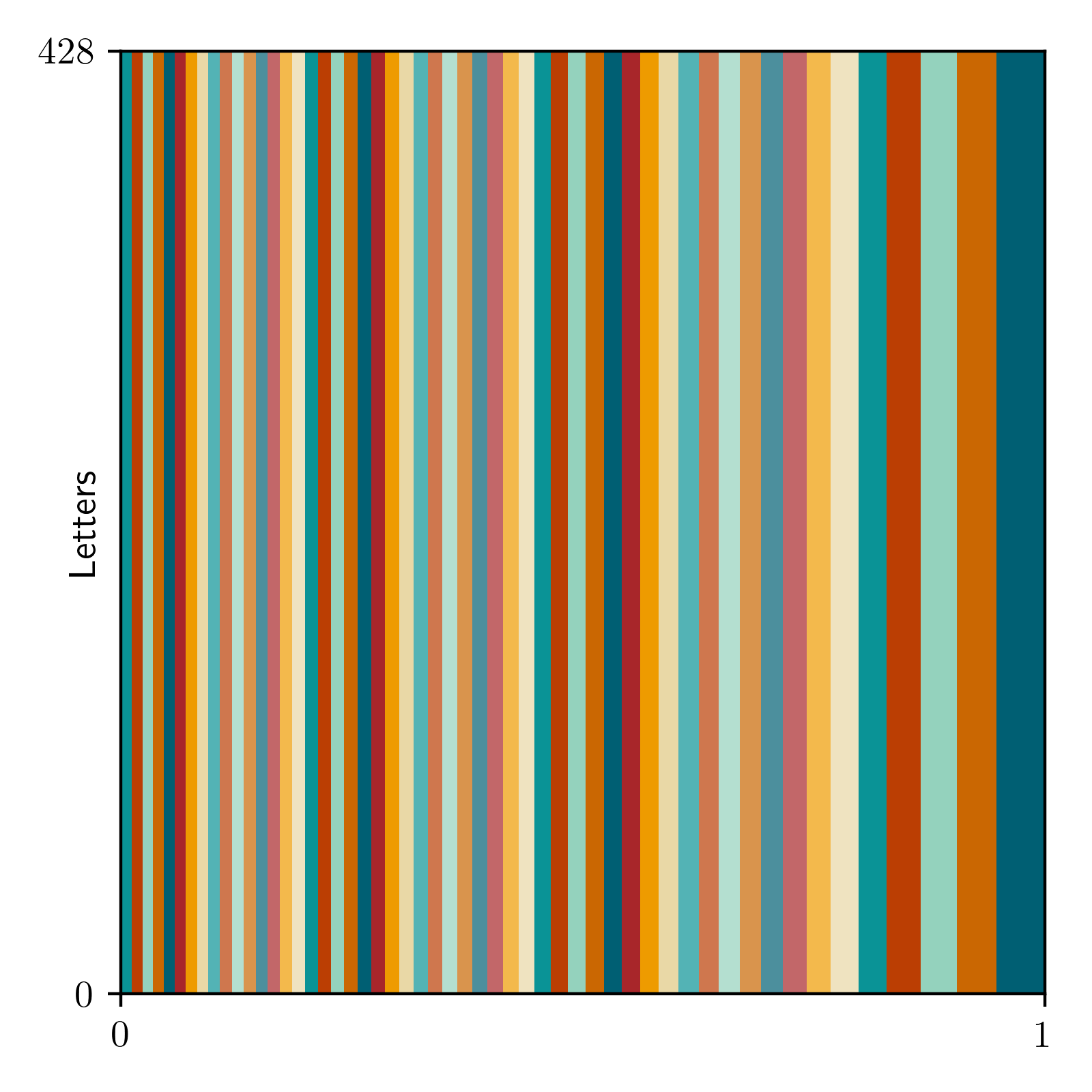}
        \includegraphics[draft=\draft, width=\linewidth]{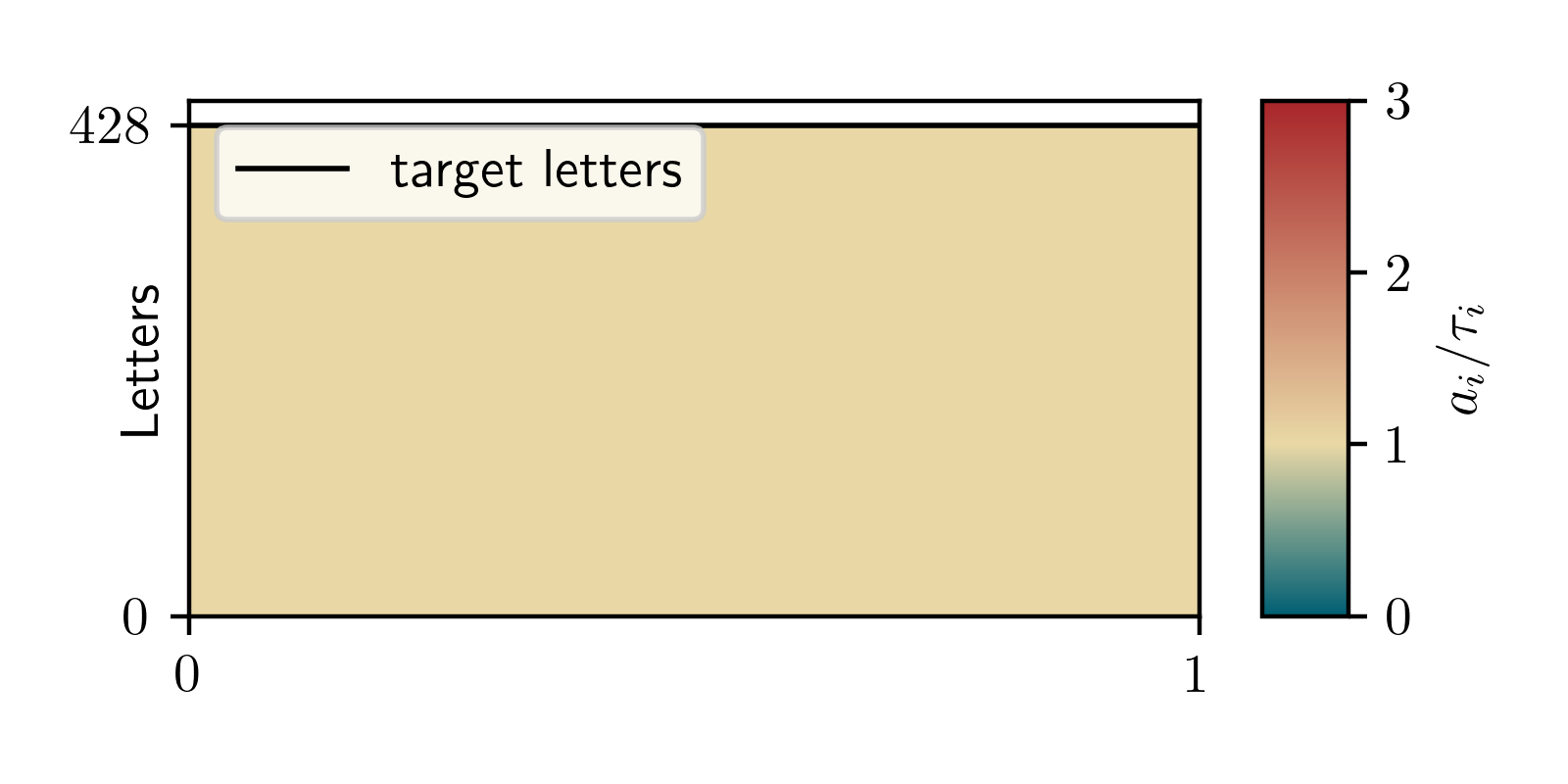}
        \caption{\buckets ($t_G = 1$)}
        \label{fig:results_Hessen_Medium_greedy_bucket_fill}
    \end{subfigure}
    \caption{Medium municipalities of Hessen ($\ell_G = 428$)}
    \label{fig:results_Hessen_Medium}
\end{figure} 

\begin{figure}
    \centering
    \begin{subfigure}{0.32\textwidth}
        \includegraphics[draft=\draft, width=\linewidth]{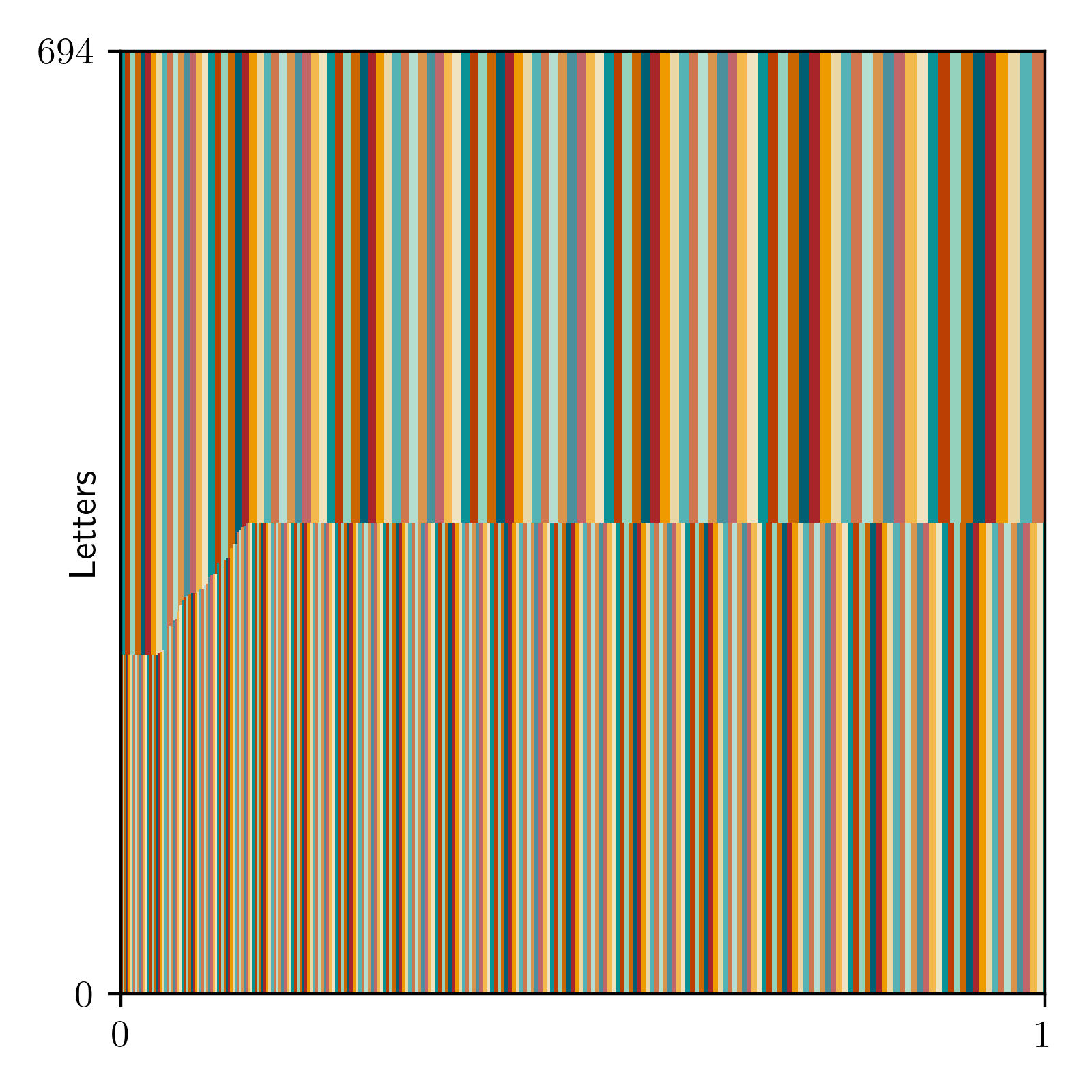}
        \includegraphics[draft=\draft, width=\linewidth]{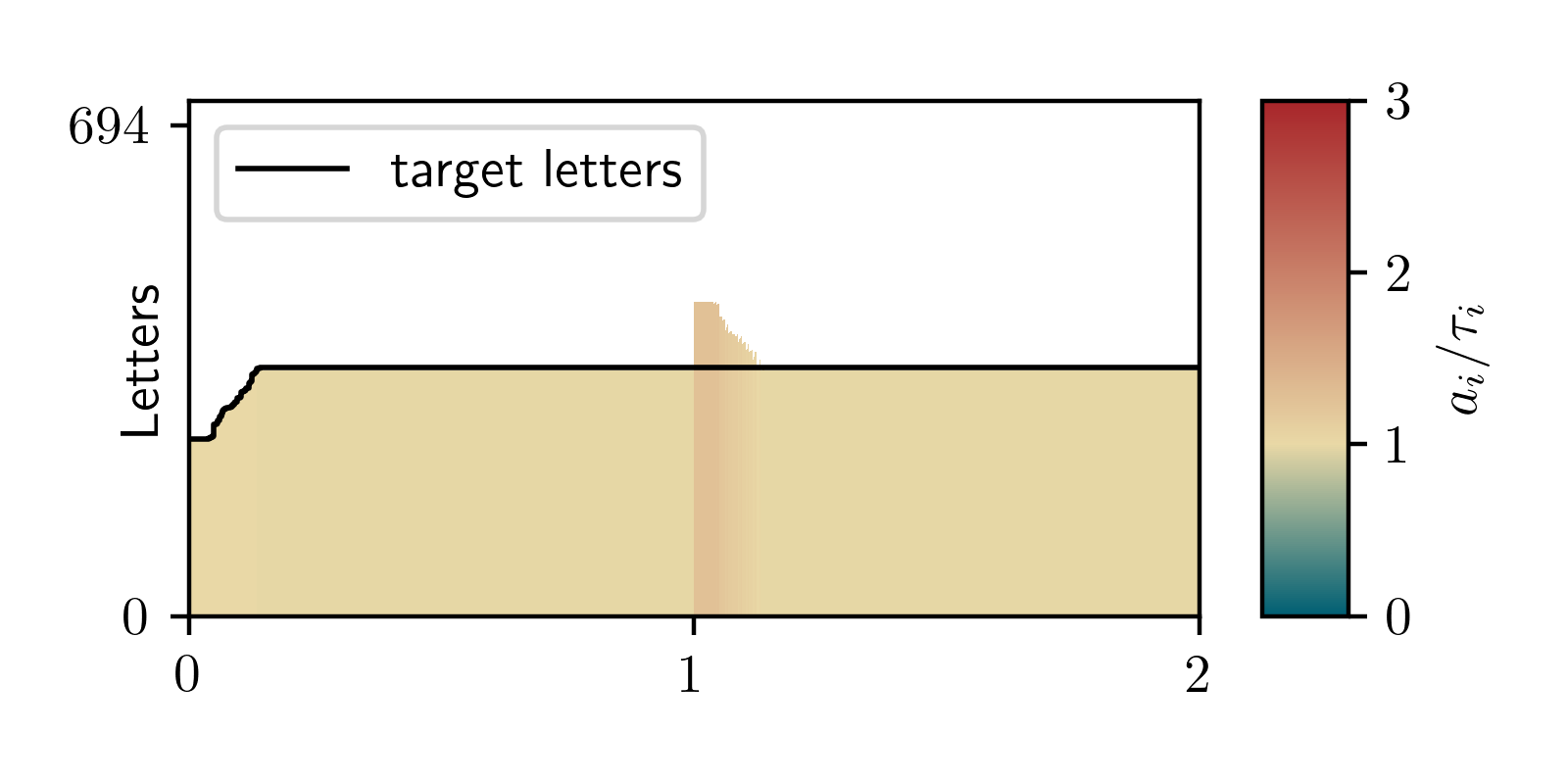}
        \caption{\greq ($t_G = 2$)}
        \label{fig:results_Hessen_Small_greedy_equal}
    \end{subfigure}
    \begin{subfigure}{0.32\textwidth}
        \includegraphics[draft=\draft, width=\linewidth]{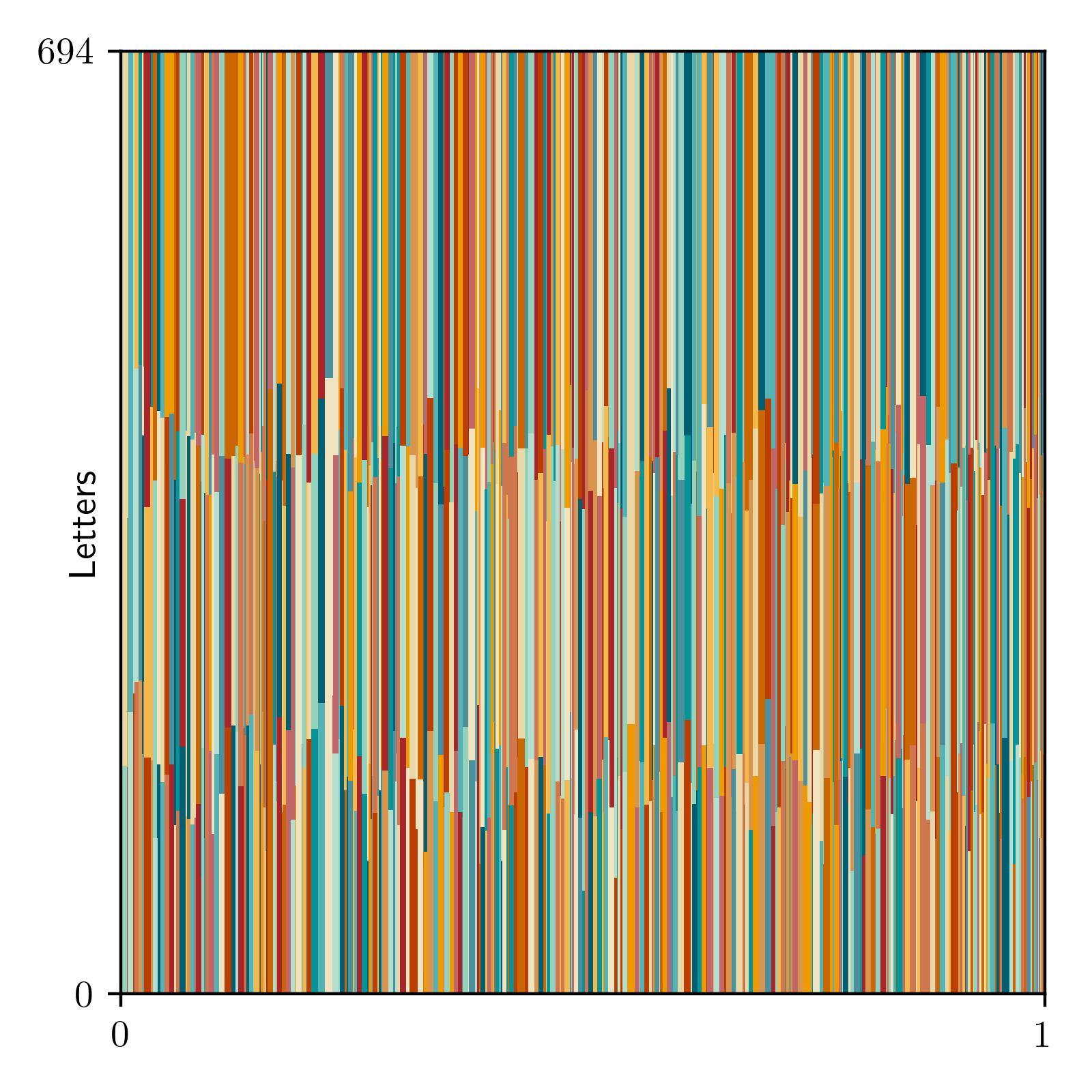}
        \includegraphics[draft=\draft, width=\linewidth]{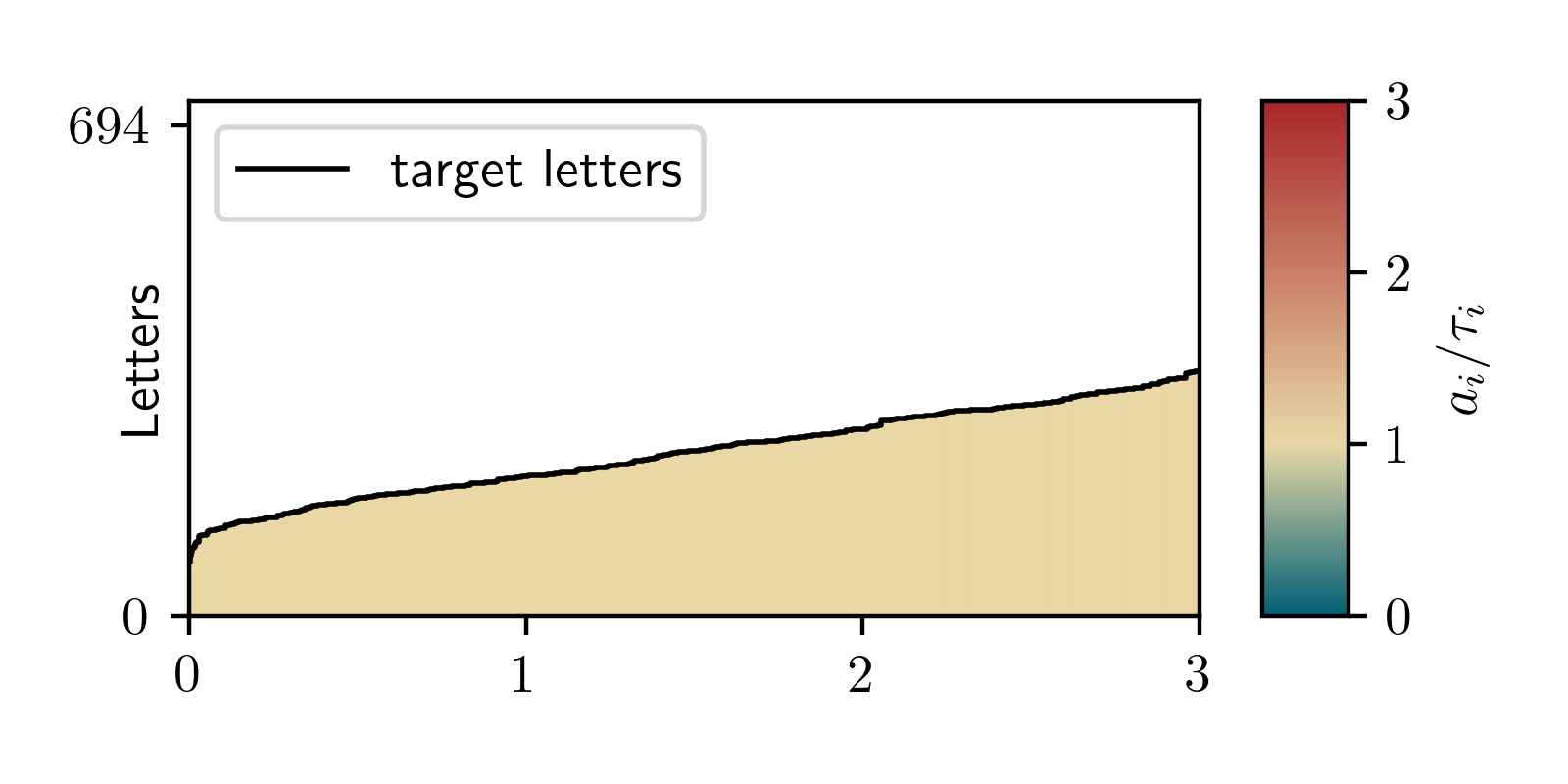}
        \caption{\colgen ($t_G\!=\!3$)}
        \label{fig:results_Hessen_Small_column_generation}
    \end{subfigure}
    \begin{subfigure}{0.32\textwidth}
        \includegraphics[draft=\draft, width=\linewidth]{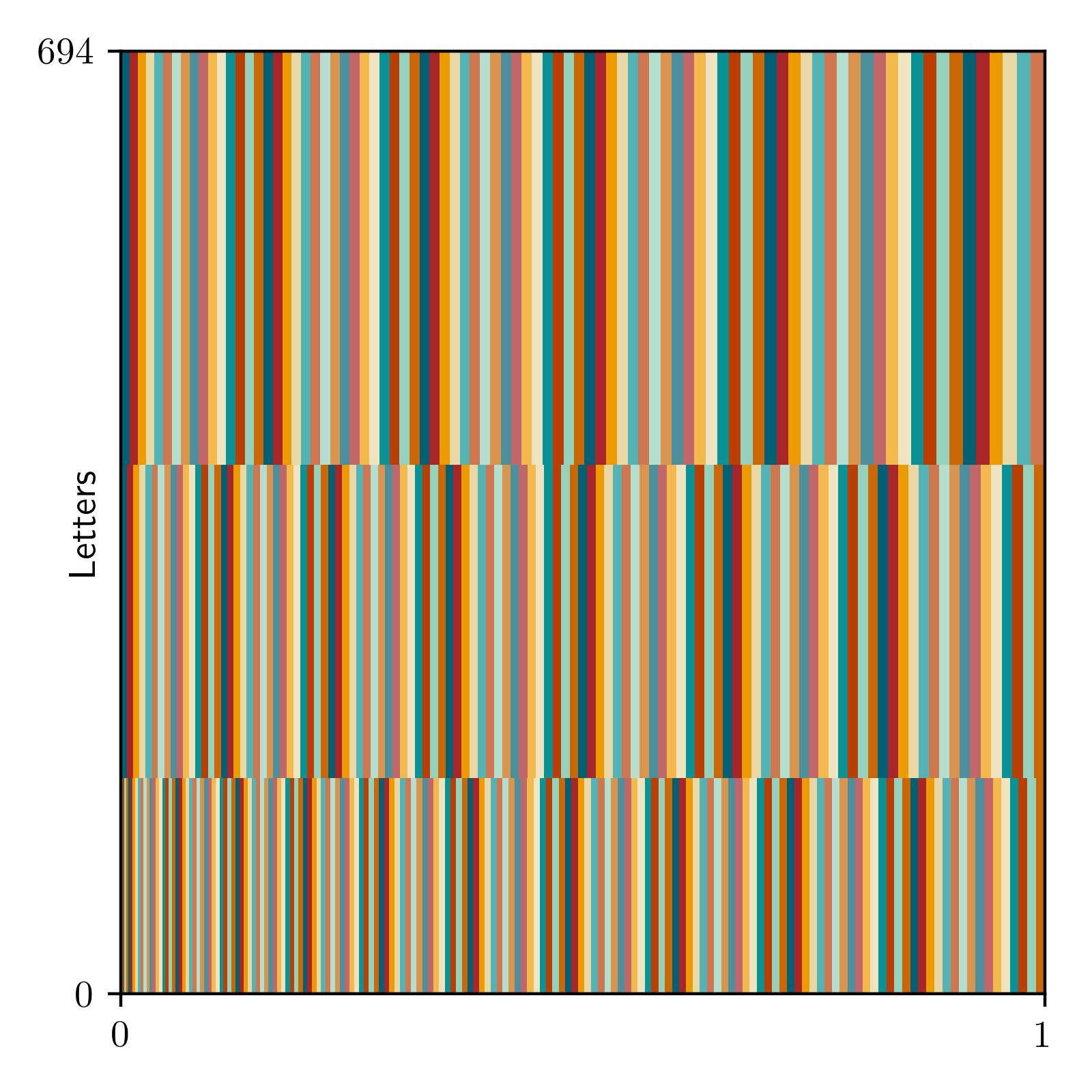}
        \includegraphics[draft=\draft, width=\linewidth]{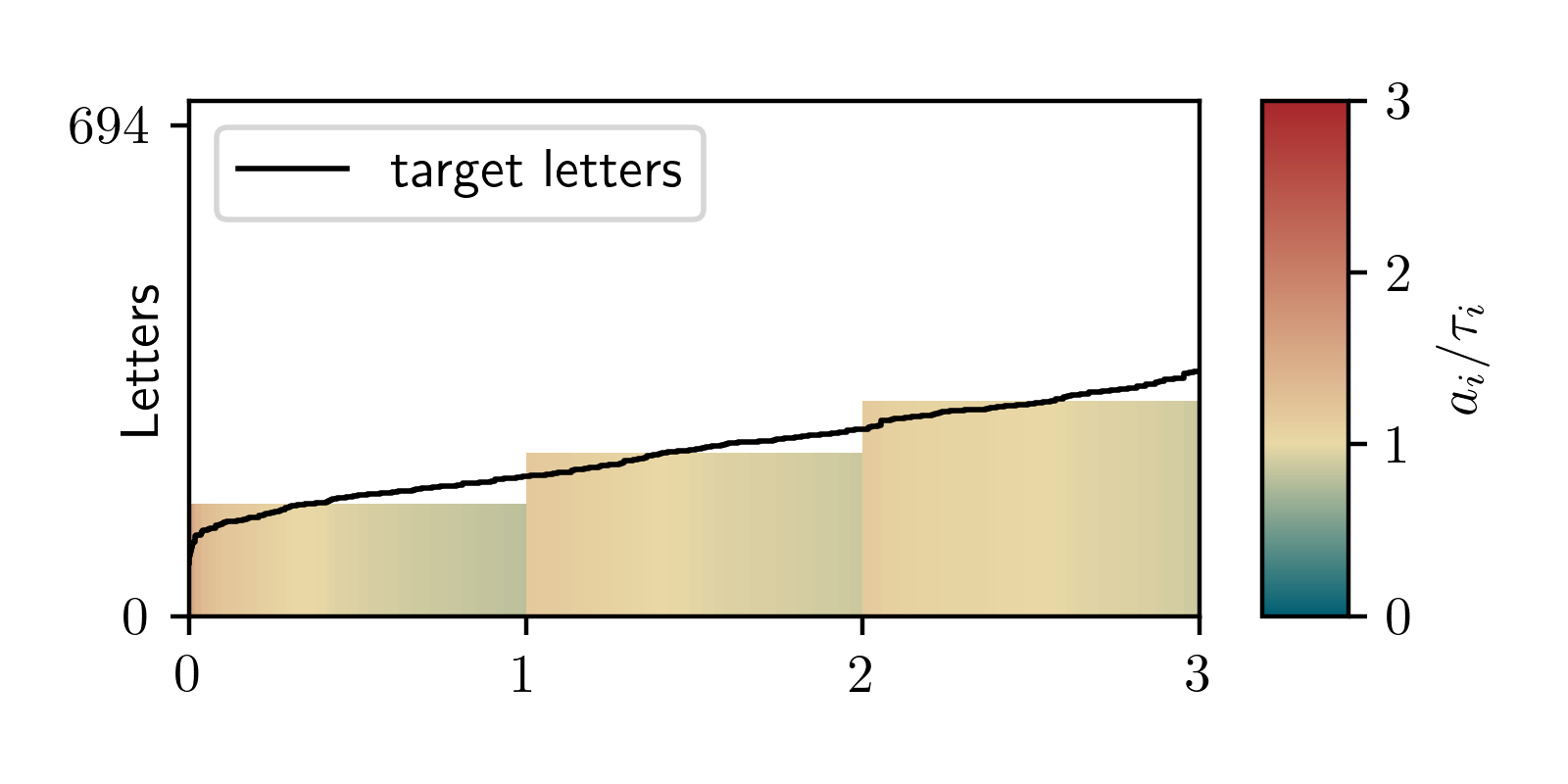}
        \caption{\buckets ($t_G = 3$)}
        \label{fig:results_Hessen_Small_greedy_bucket_fill}
    \end{subfigure}
    \caption{Small municipalities of Hessen ($\ell_G = 694$)}
    \label{fig:results_Hessen_Small}
\end{figure} 

\begin{figure}
    \centering
    \begin{subfigure}{0.32\textwidth}
        \includegraphics[draft=\draft, width=\linewidth]{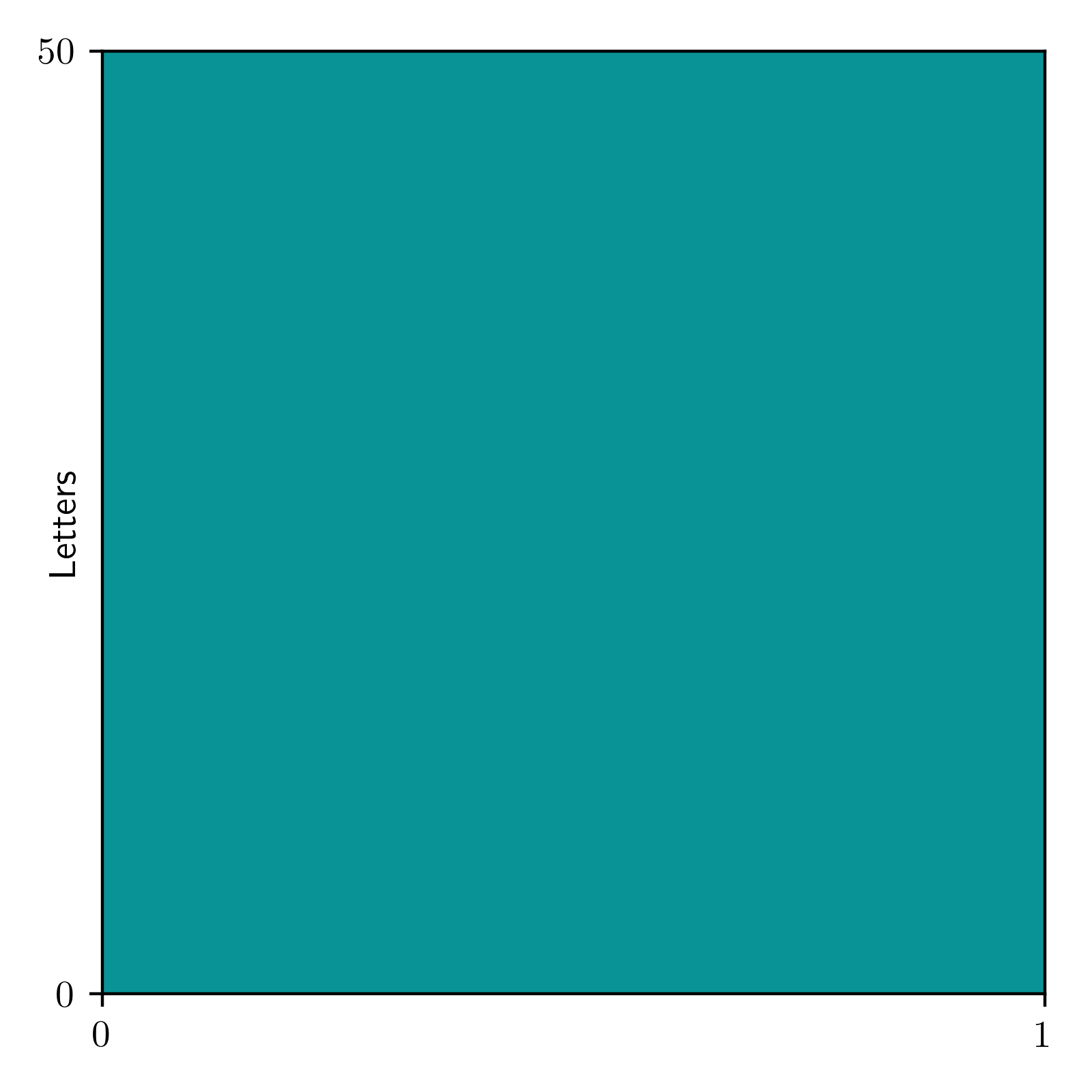}
        \includegraphics[draft=\draft, width=\linewidth]{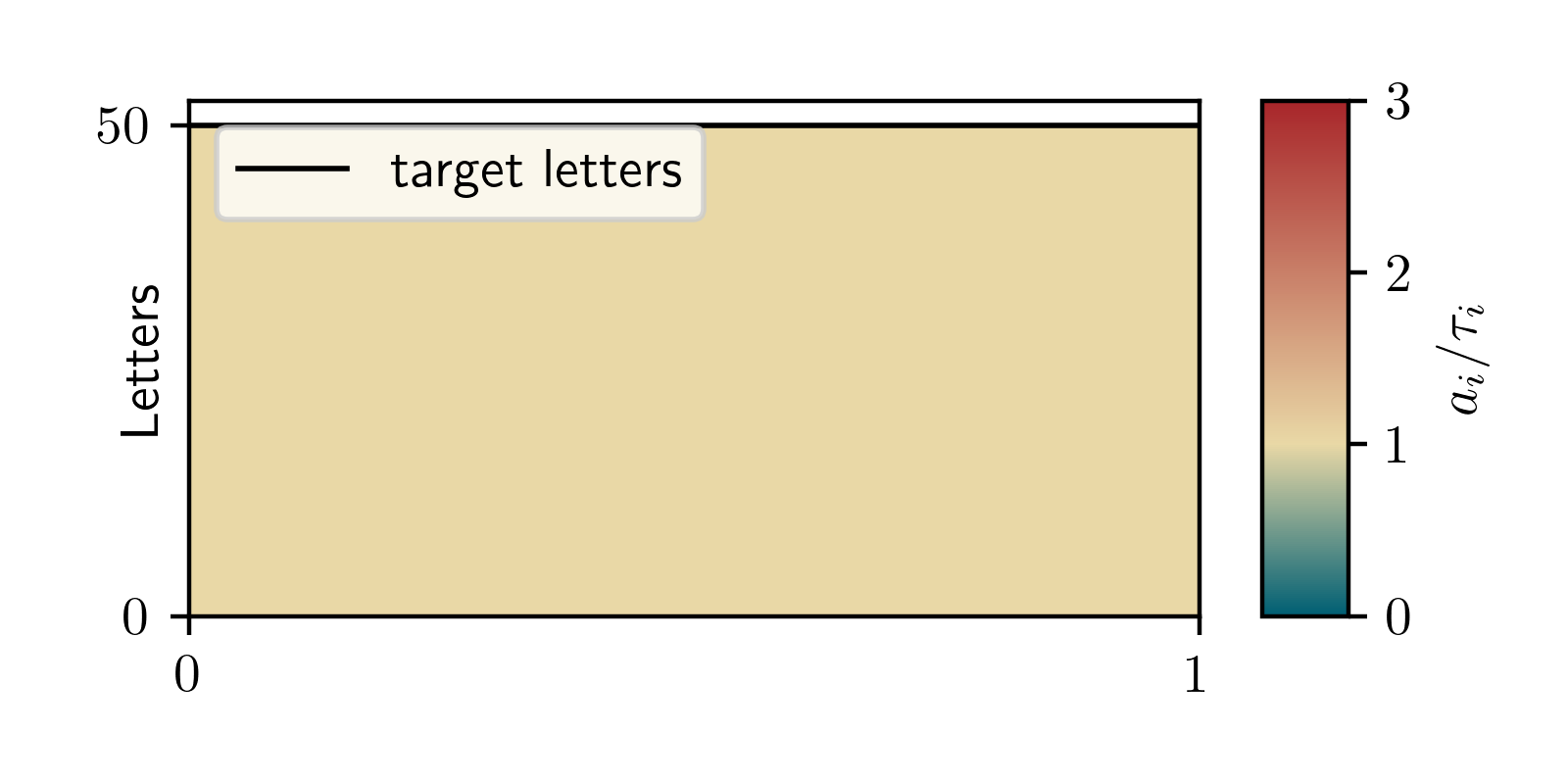}
        \caption{\greq ($t_G = 1$)}
        \label{fig:results_Mecklenburg-Vorpommern_Large_greedy_equal}
    \end{subfigure}
    \begin{subfigure}{0.32\textwidth}
        \includegraphics[draft=\draft, width=\linewidth]{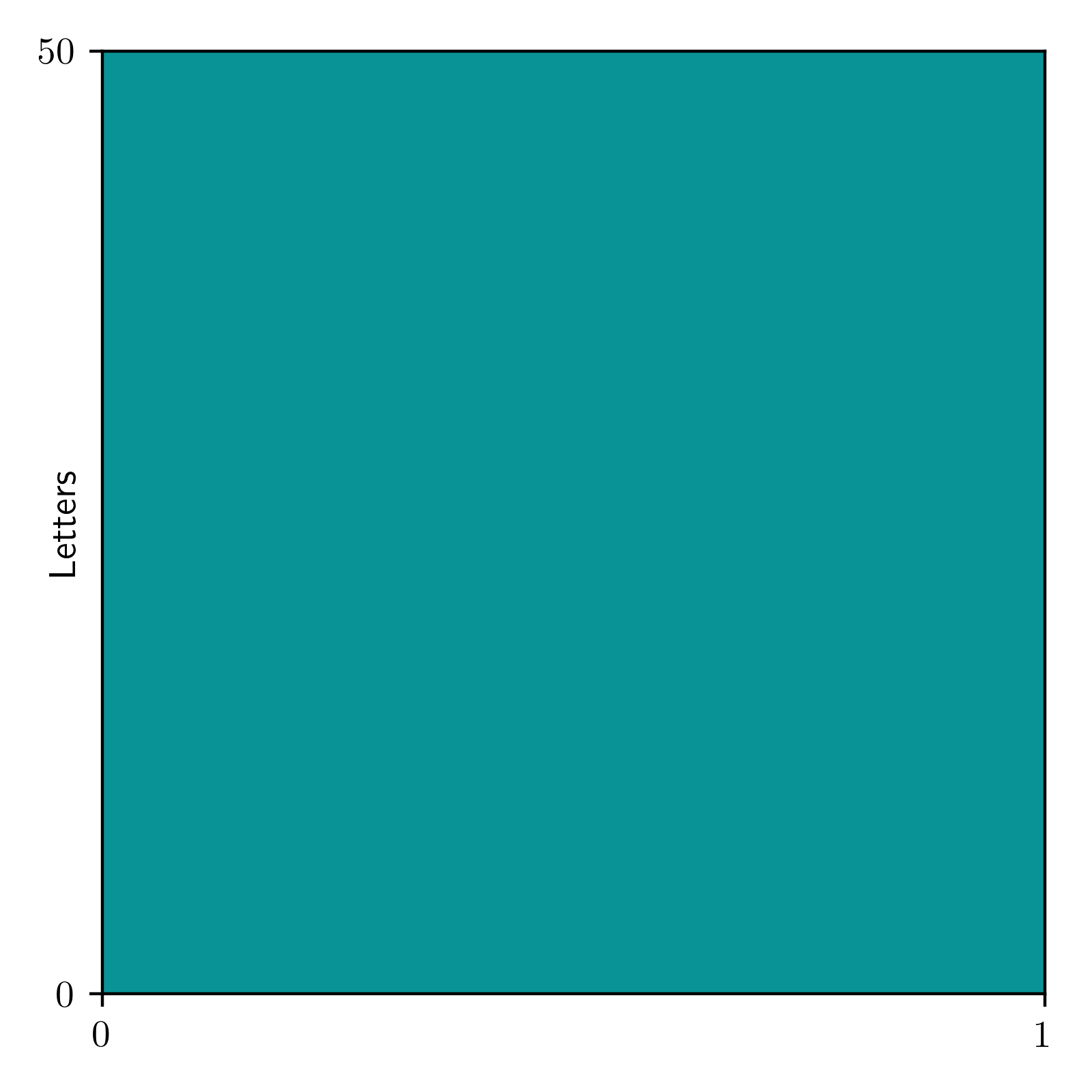}
        \includegraphics[draft=\draft, width=\linewidth]{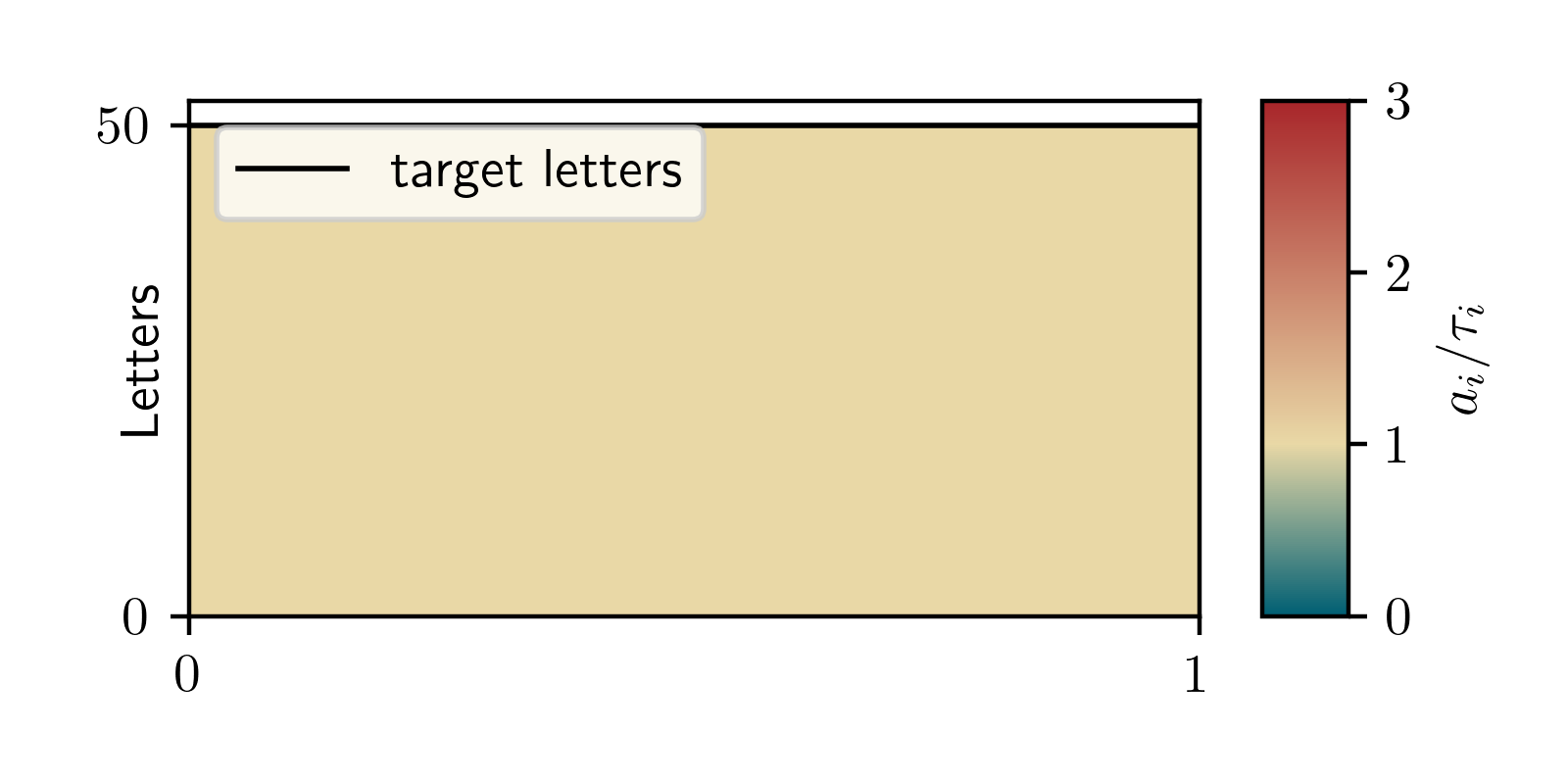}
        \caption{\colgen ($t_G\!=\!1$)}
        \label{fig:results_Mecklenburg-Vorpommern_Large_column_generation}
    \end{subfigure}
    \begin{subfigure}{0.32\textwidth}
        \includegraphics[draft=\draft, width=\linewidth]{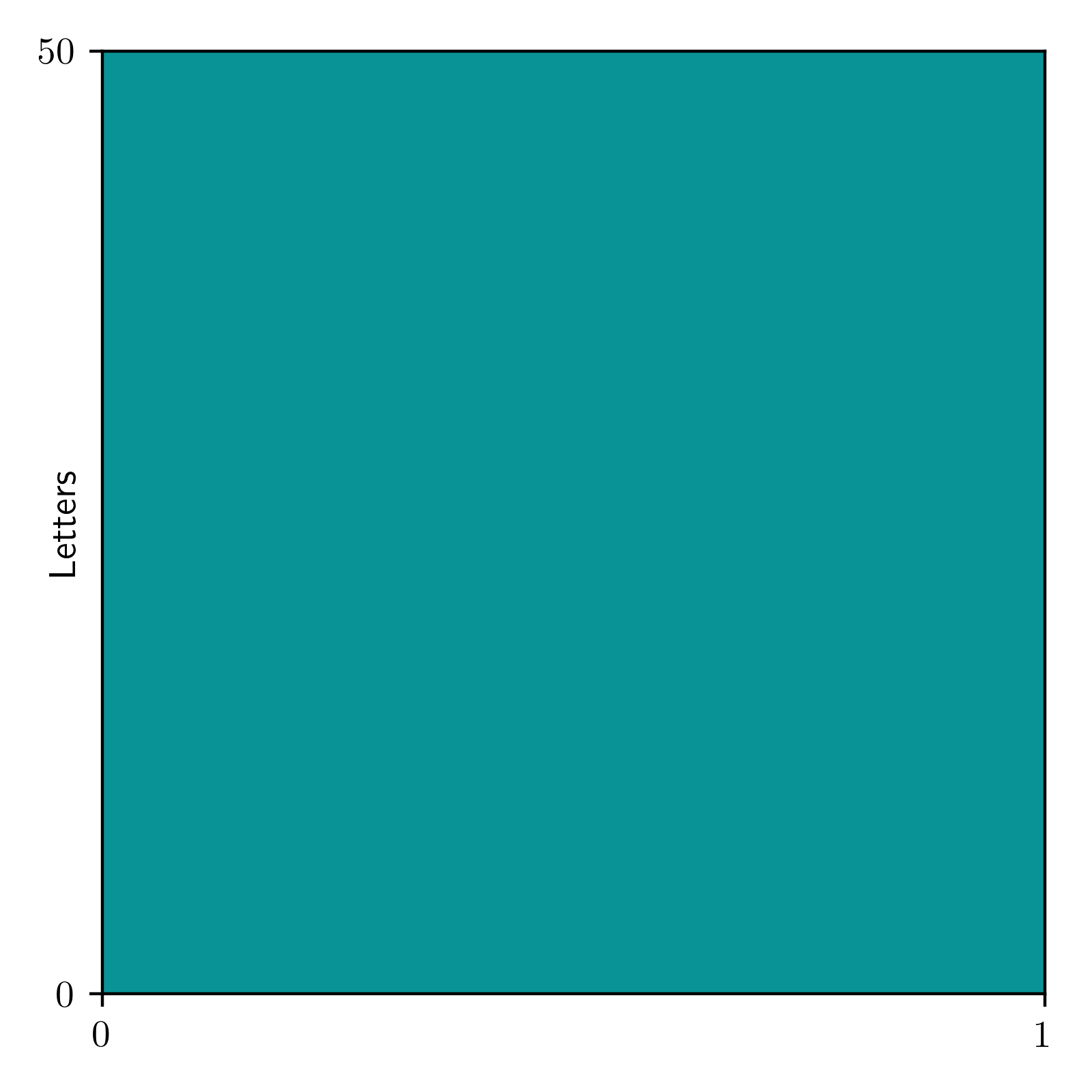}
        \includegraphics[draft=\draft, width=\linewidth]{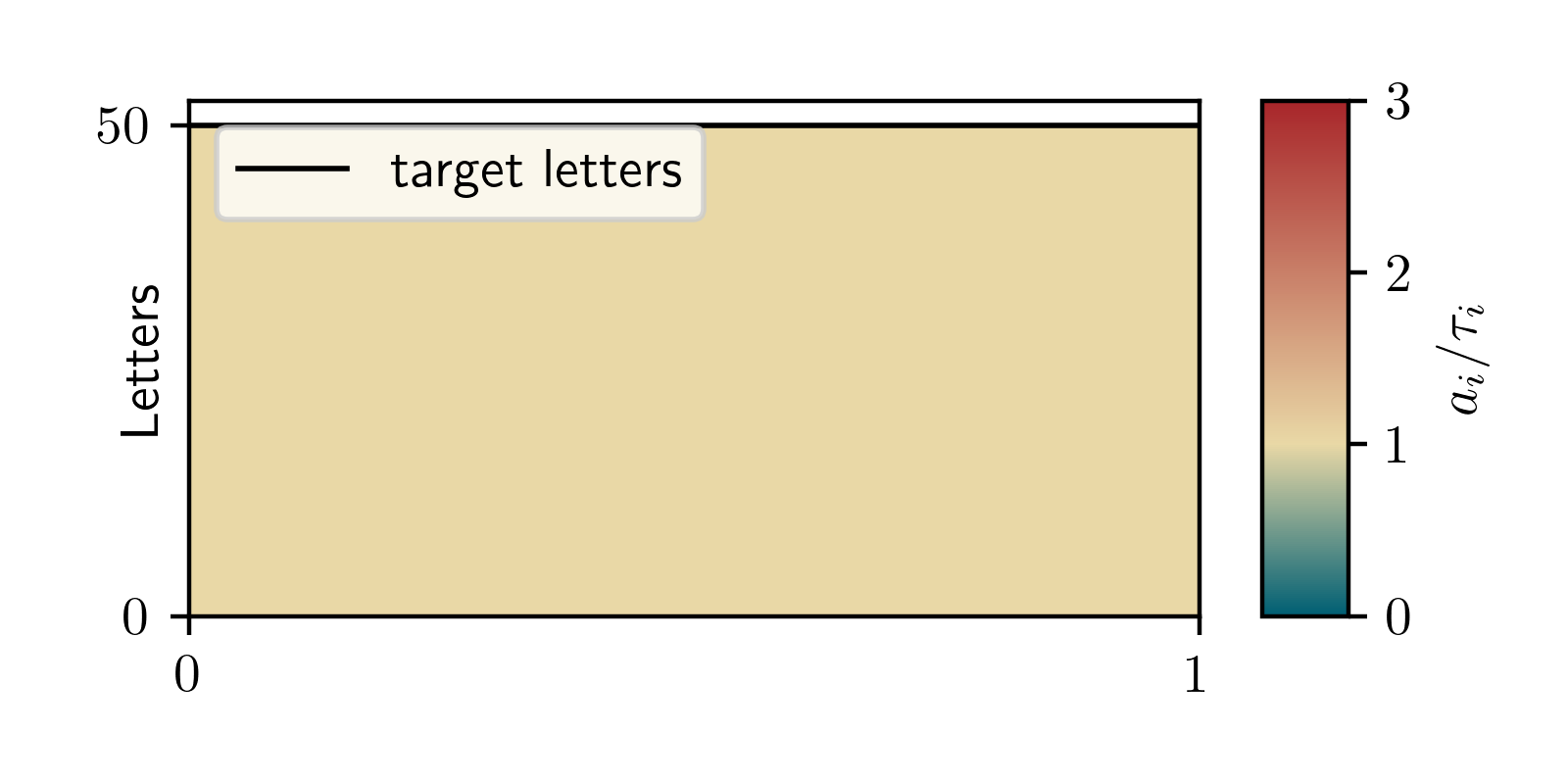}
        \caption{\buckets ($t_G = 1$)}
        \label{fig:results_Mecklenburg-Vorpommern_Large_greedy_bucket_fill}
    \end{subfigure}
    \caption{Large municipalities of Mecklenburg-Vorpommern ($\ell_G = 50$)}
    \label{fig:results_Mecklenburg-Vorpommern_Large}
\end{figure} 

\begin{figure}
    \centering
    \begin{subfigure}{0.32\textwidth}
        \includegraphics[draft=\draft, width=\linewidth]{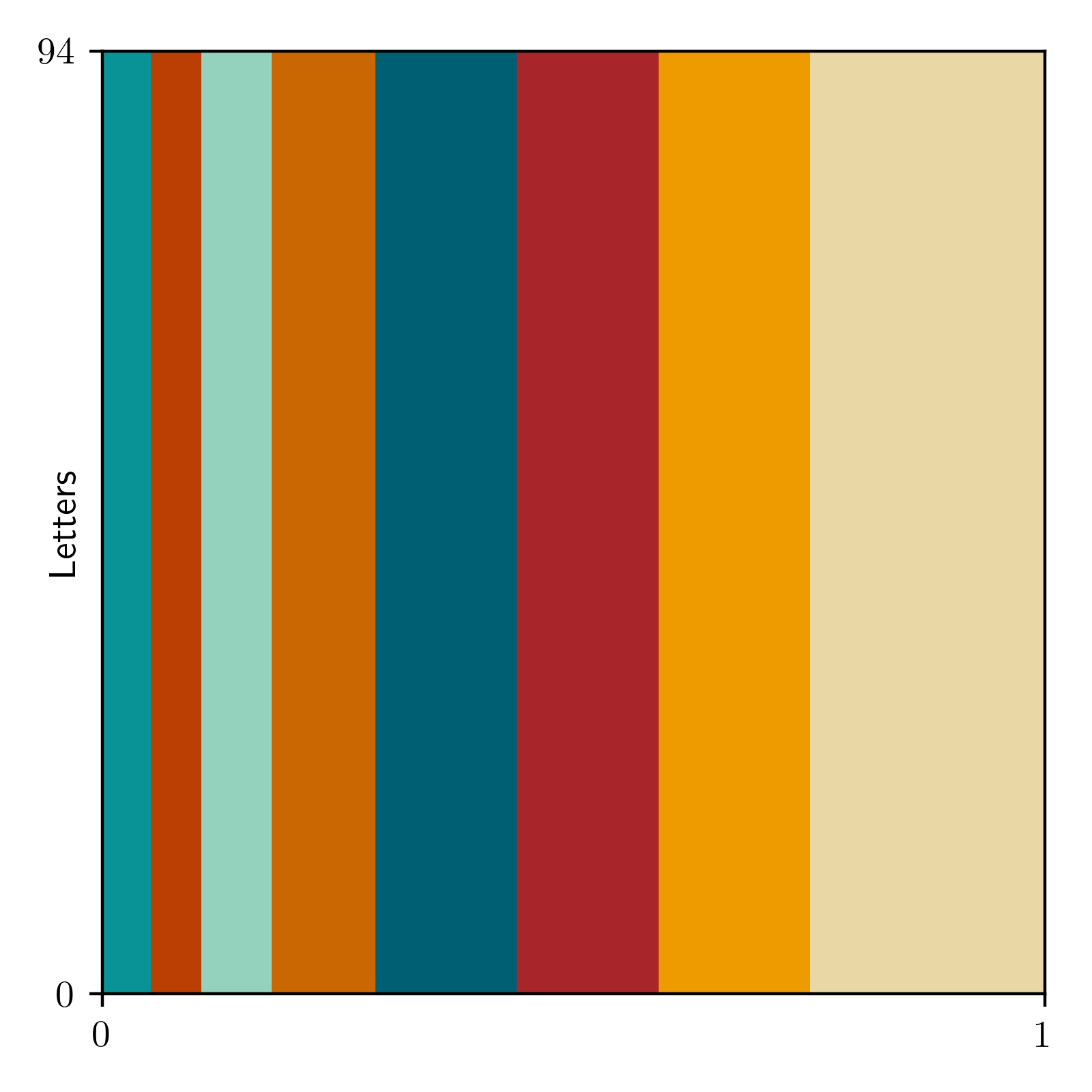}
        \includegraphics[draft=\draft, width=\linewidth]{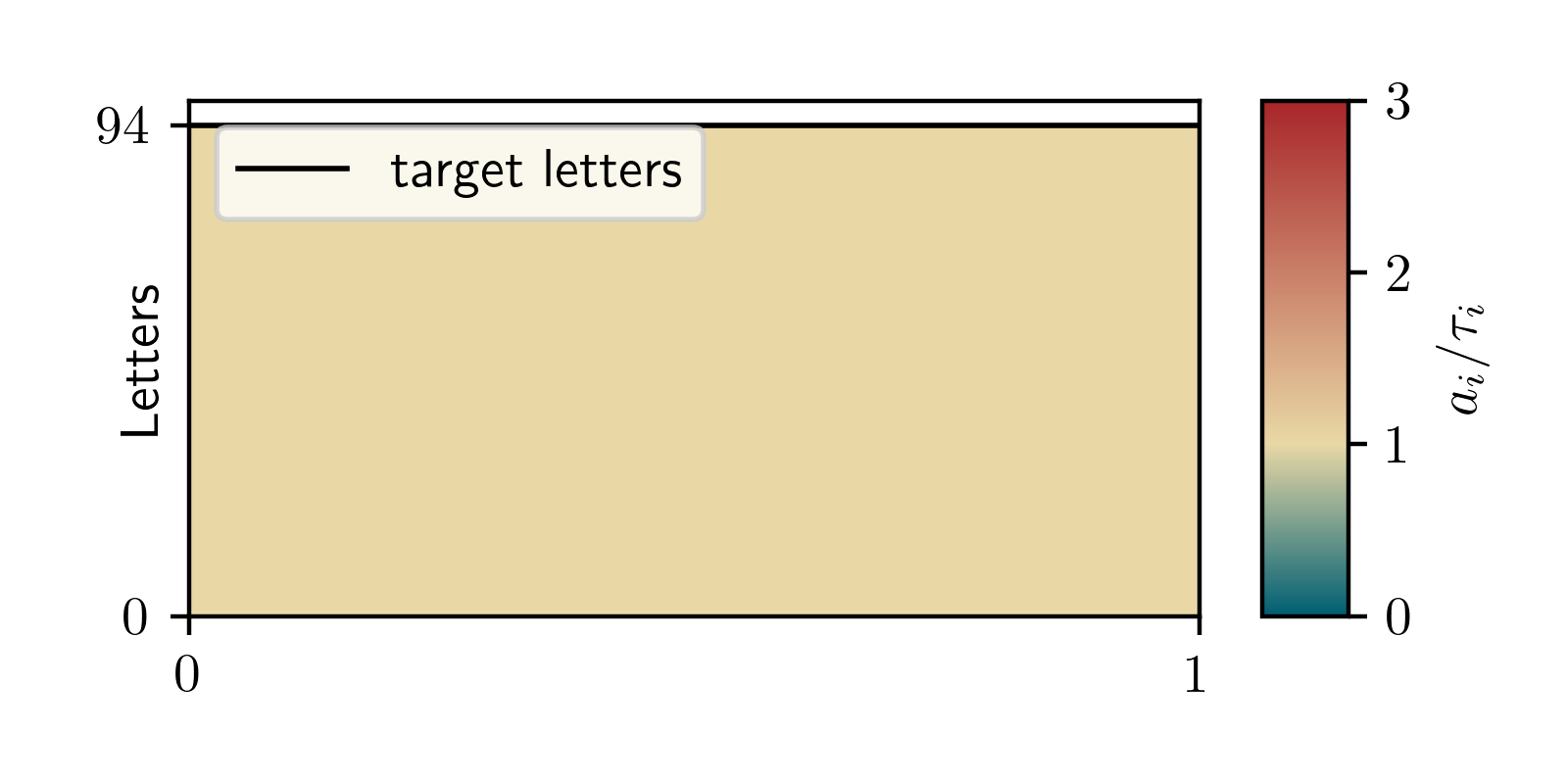}
        \caption{\greq ($t_G = 1$)}
        \label{fig:results_Mecklenburg-Vorpommern_Medium_greedy_equal}
    \end{subfigure}
    \begin{subfigure}{0.32\textwidth}
        \includegraphics[draft=\draft, width=\linewidth]{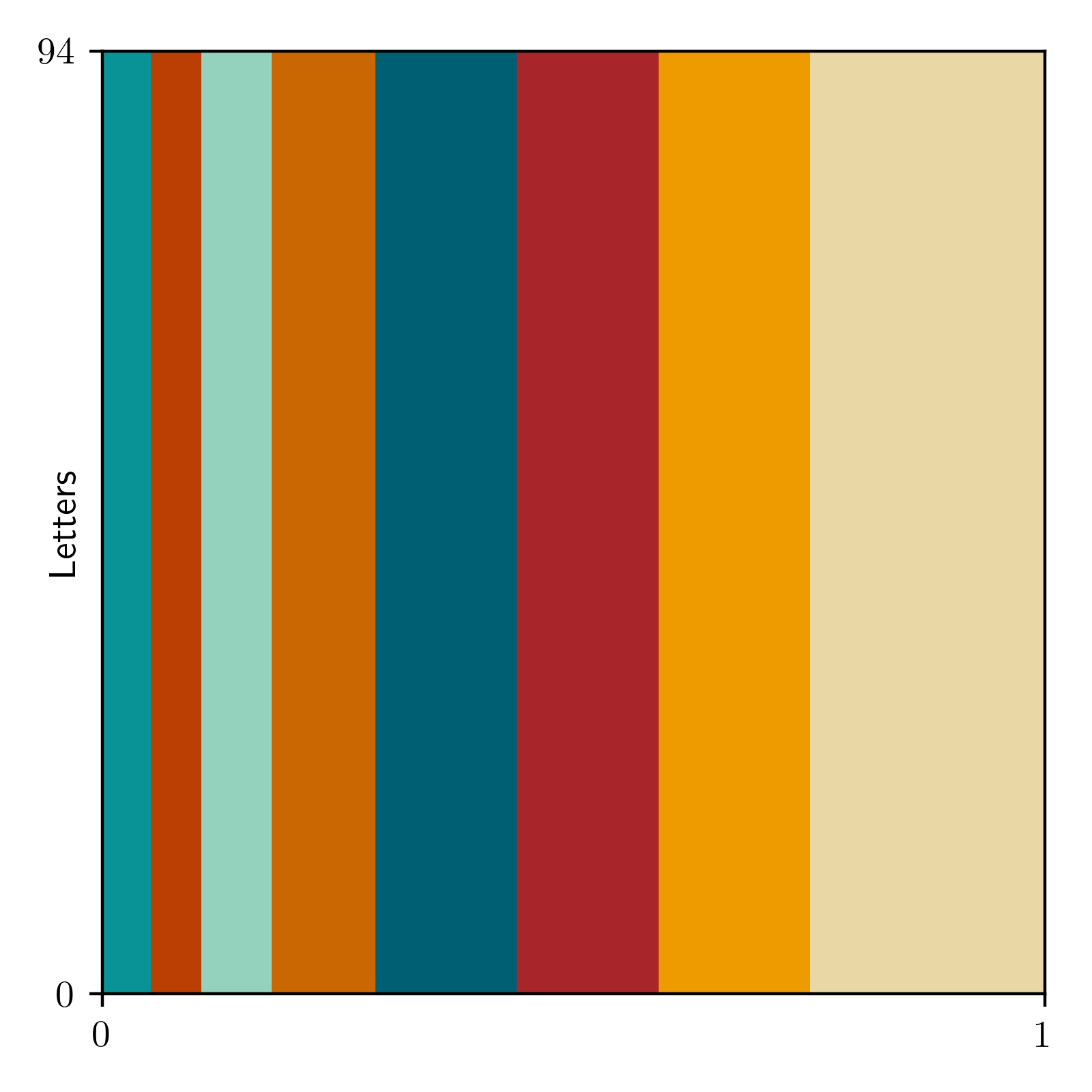}
        \includegraphics[draft=\draft, width=\linewidth]{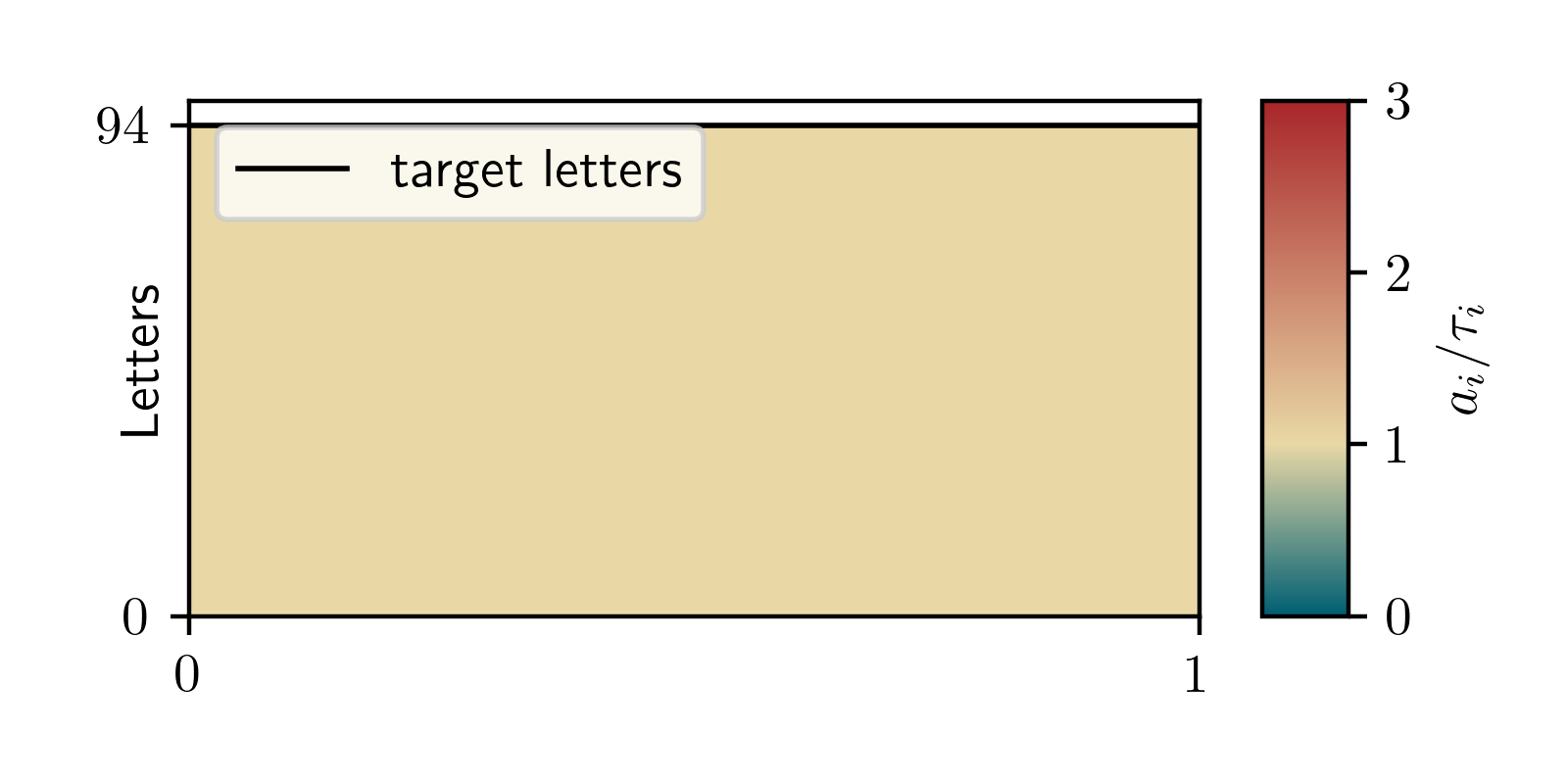}
        \caption{\colgen ($t_G\!=\!1$)}
        \label{fig:results_Mecklenburg-Vorpommern_Medium_column_generation}
    \end{subfigure}
    \begin{subfigure}{0.32\textwidth}
        \includegraphics[draft=\draft, width=\linewidth]{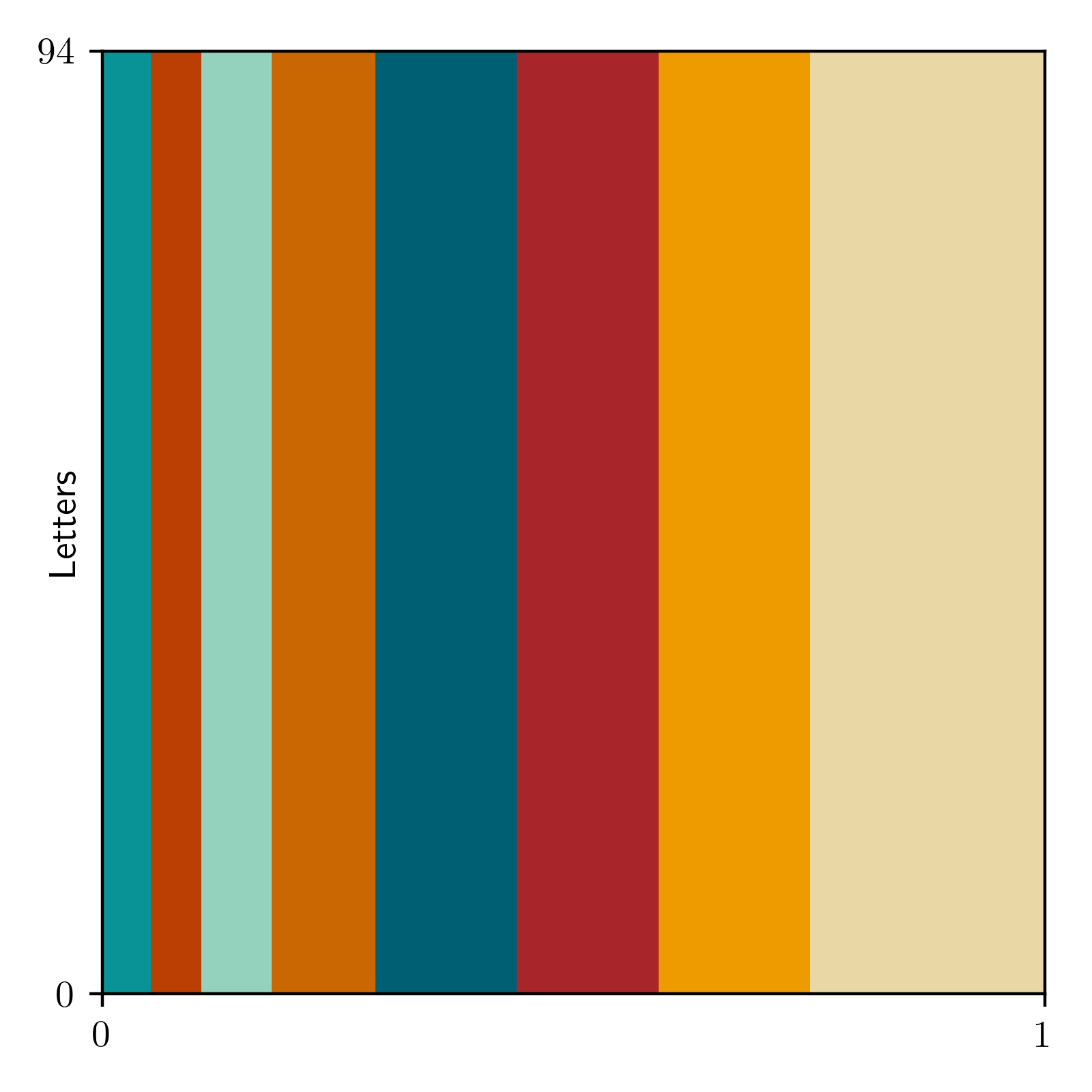}
        \includegraphics[draft=\draft, width=\linewidth]{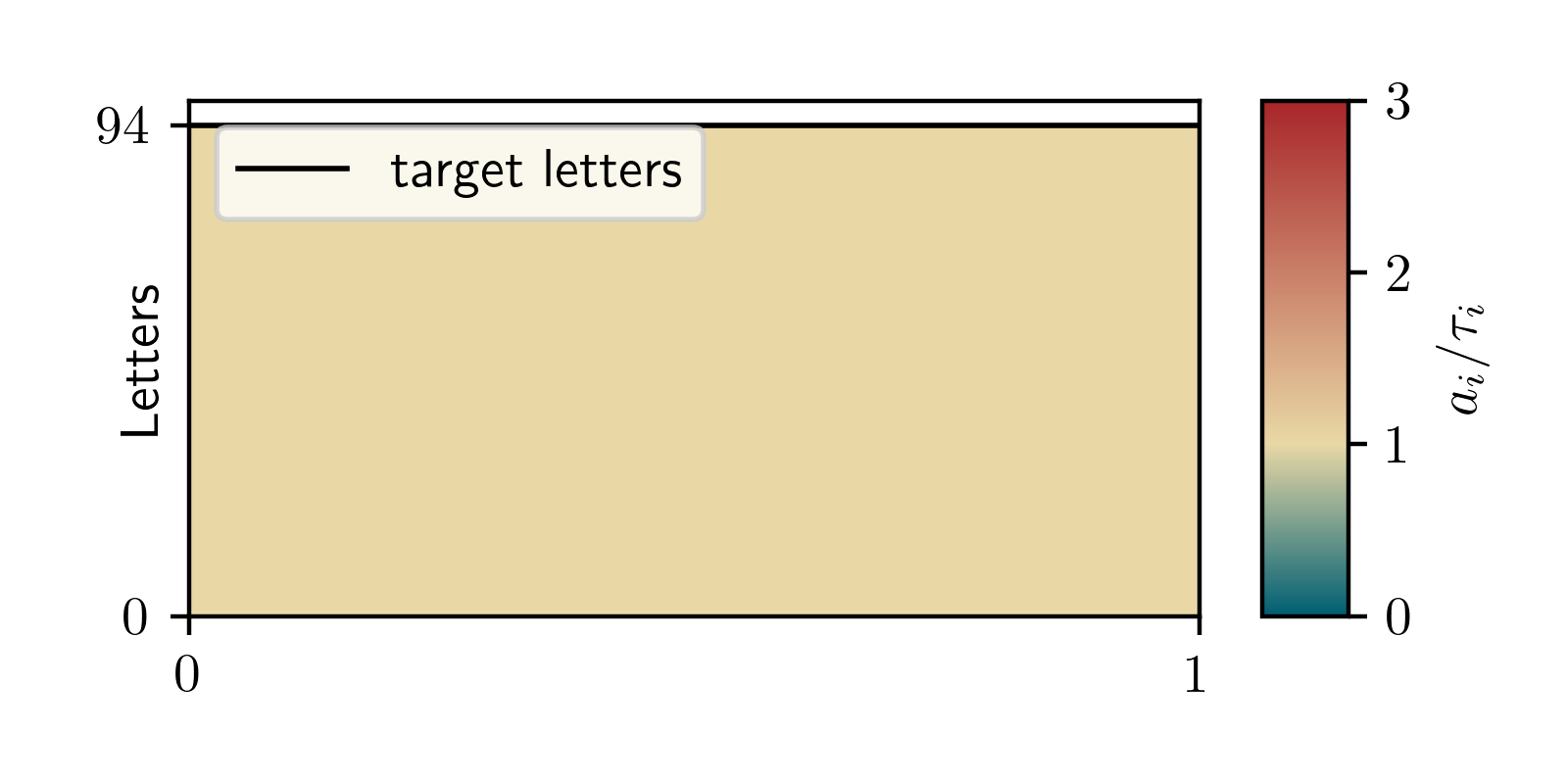}
        \caption{\buckets ($t_G = 1$)}
        \label{fig:results_Mecklenburg-Vorpommern_Medium_greedy_bucket_fill}
    \end{subfigure}
    \caption{Medium municipalities of Mecklenburg-Vorpommern ($\ell_G = 94$)}
    \label{fig:results_Mecklenburg-Vorpommern_Medium}
\end{figure} 

\begin{figure}
    \centering
    \begin{subfigure}{0.32\textwidth}
        \includegraphics[draft=\draft, width=\linewidth]{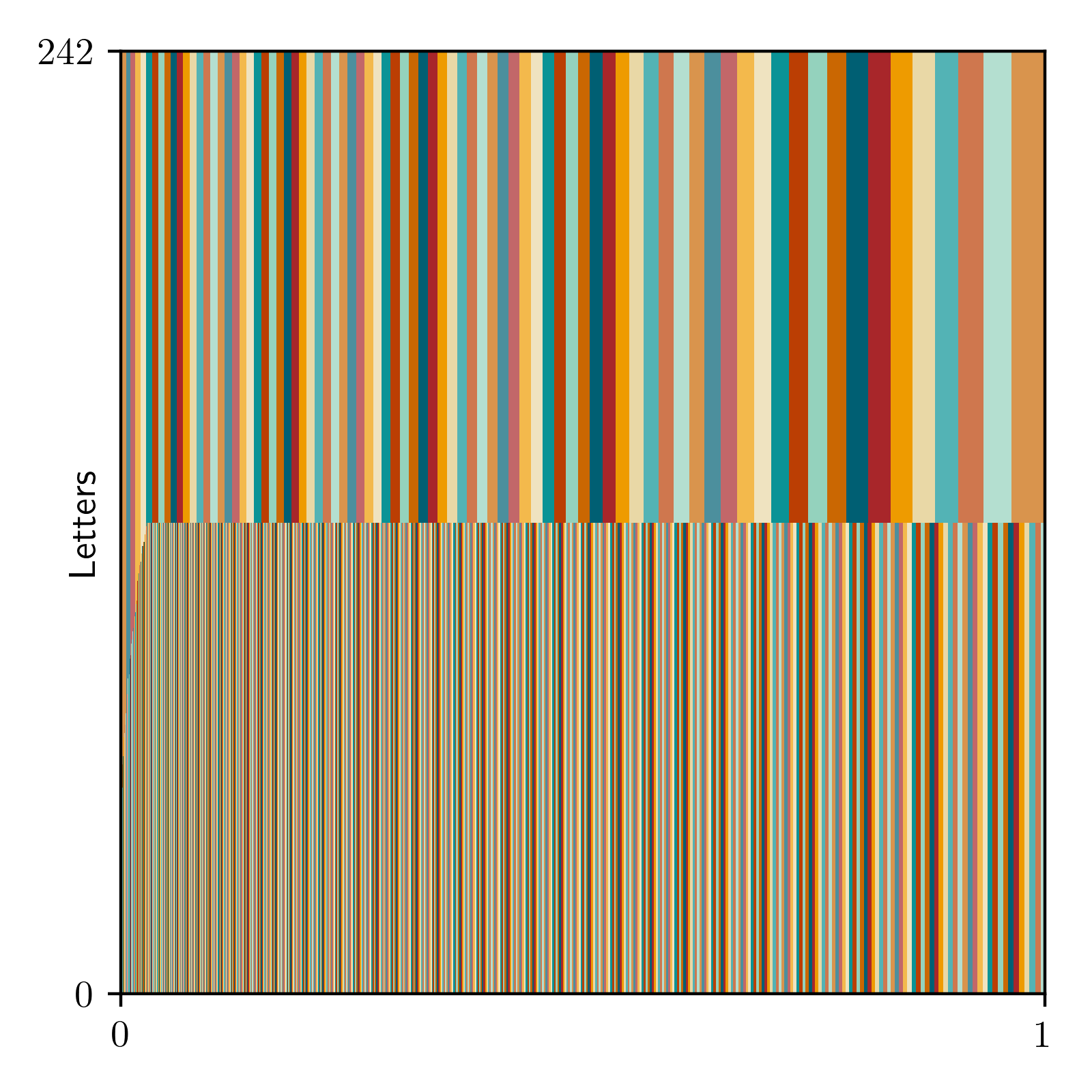}
        \includegraphics[draft=\draft, width=\linewidth]{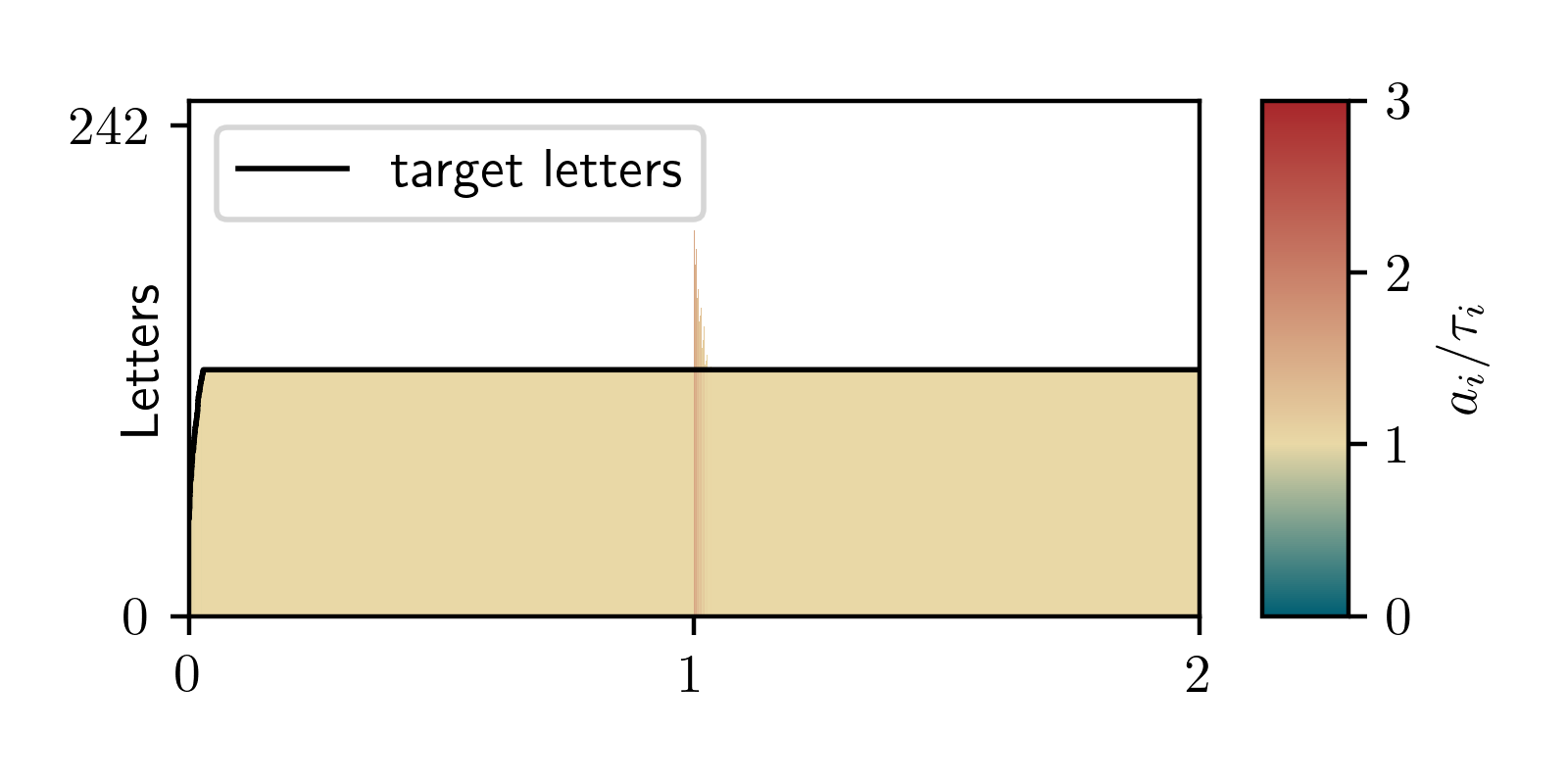}
        \caption{\greq ($t_G = 2$)}
        \label{fig:results_Mecklenburg-Vorpommern_Small_greedy_equal}
    \end{subfigure}
    \begin{subfigure}{0.32\textwidth}
        \includegraphics[draft=\draft, width=\linewidth]{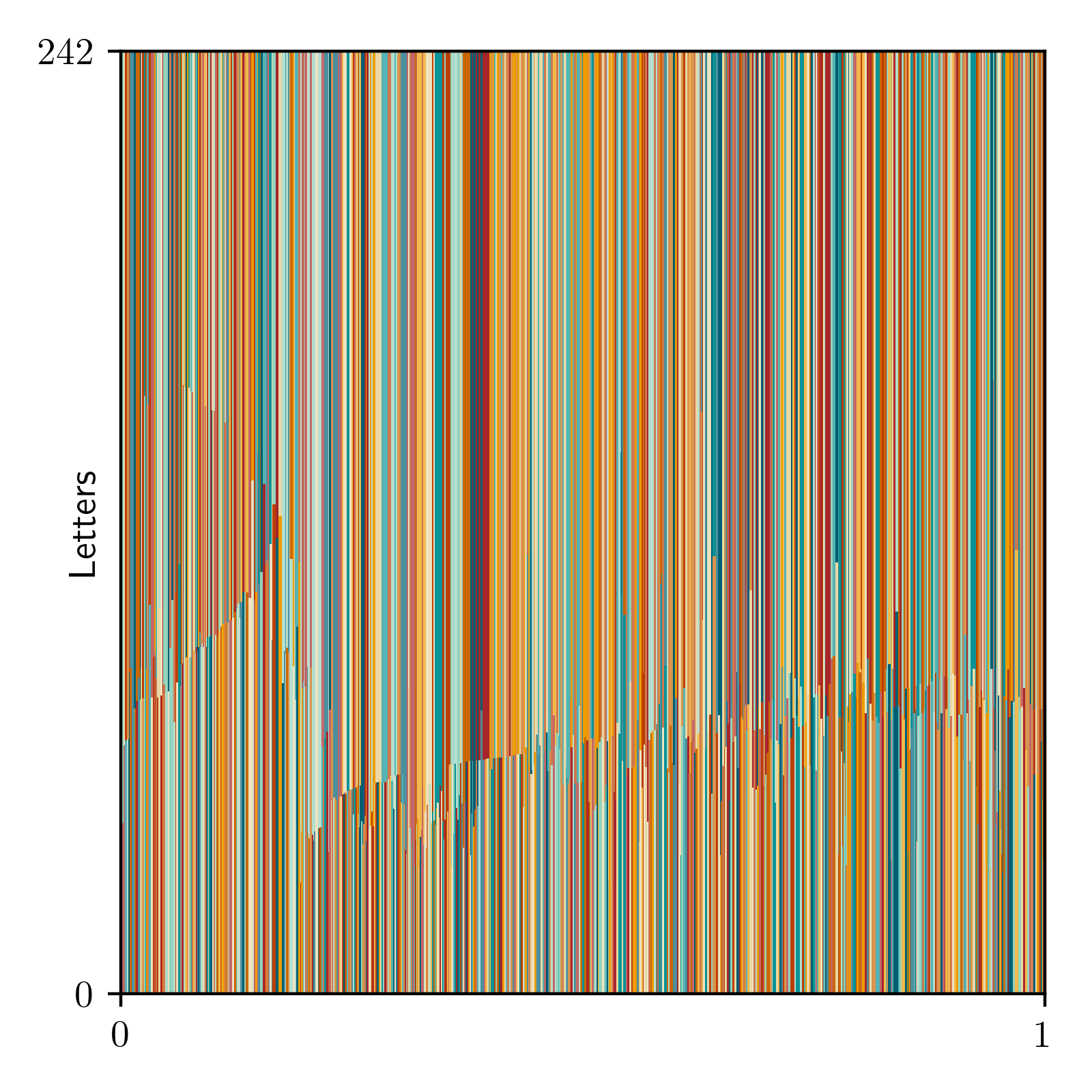}
        \includegraphics[draft=\draft, width=\linewidth]{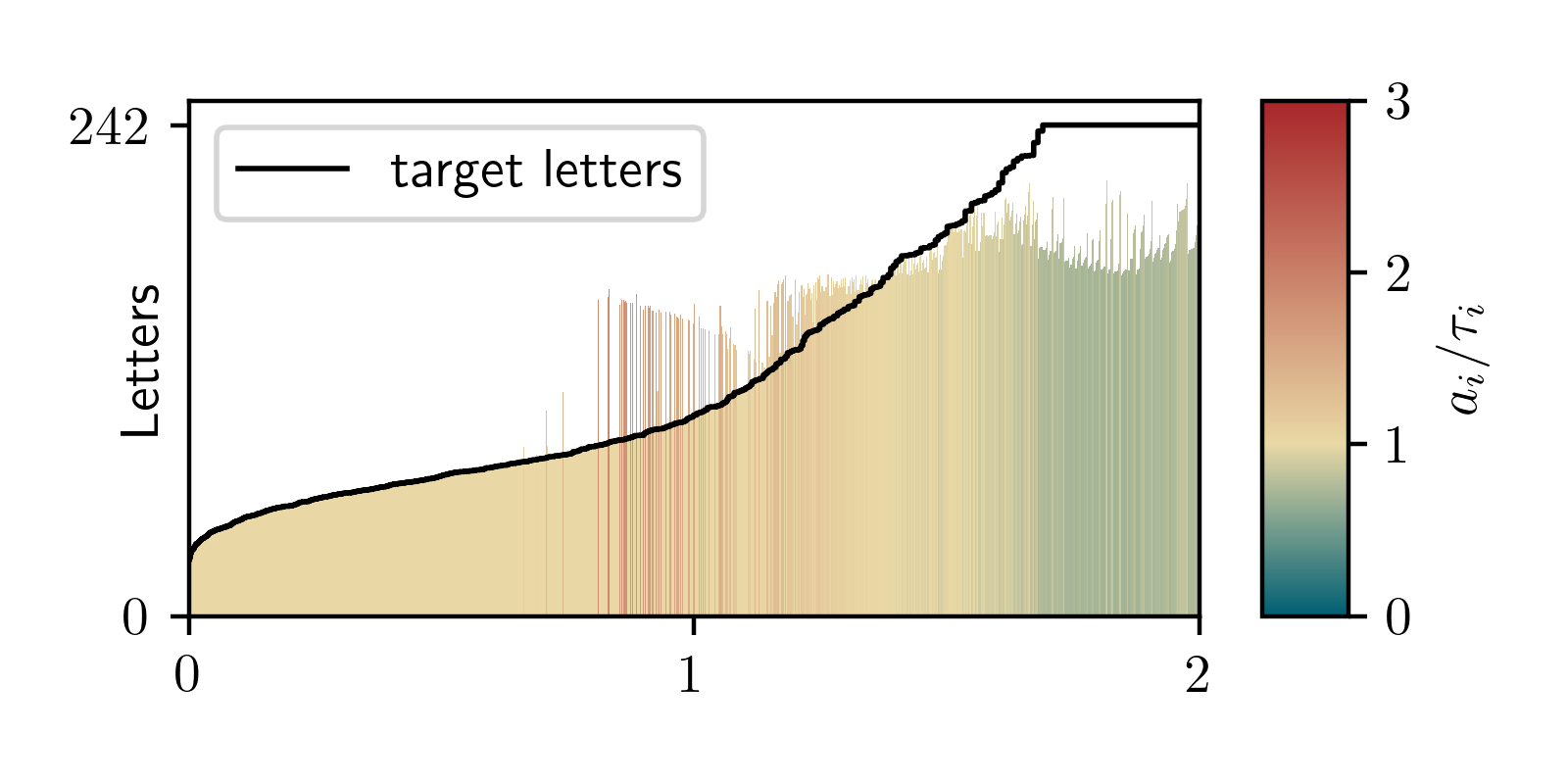}
        \caption{\colgen ($t_G\!=\!2$)}
        \label{fig:results_Mecklenburg-Vorpommern_Small_column_generation}
    \end{subfigure}
    \begin{subfigure}{0.32\textwidth}
        \includegraphics[draft=\draft, width=\linewidth]{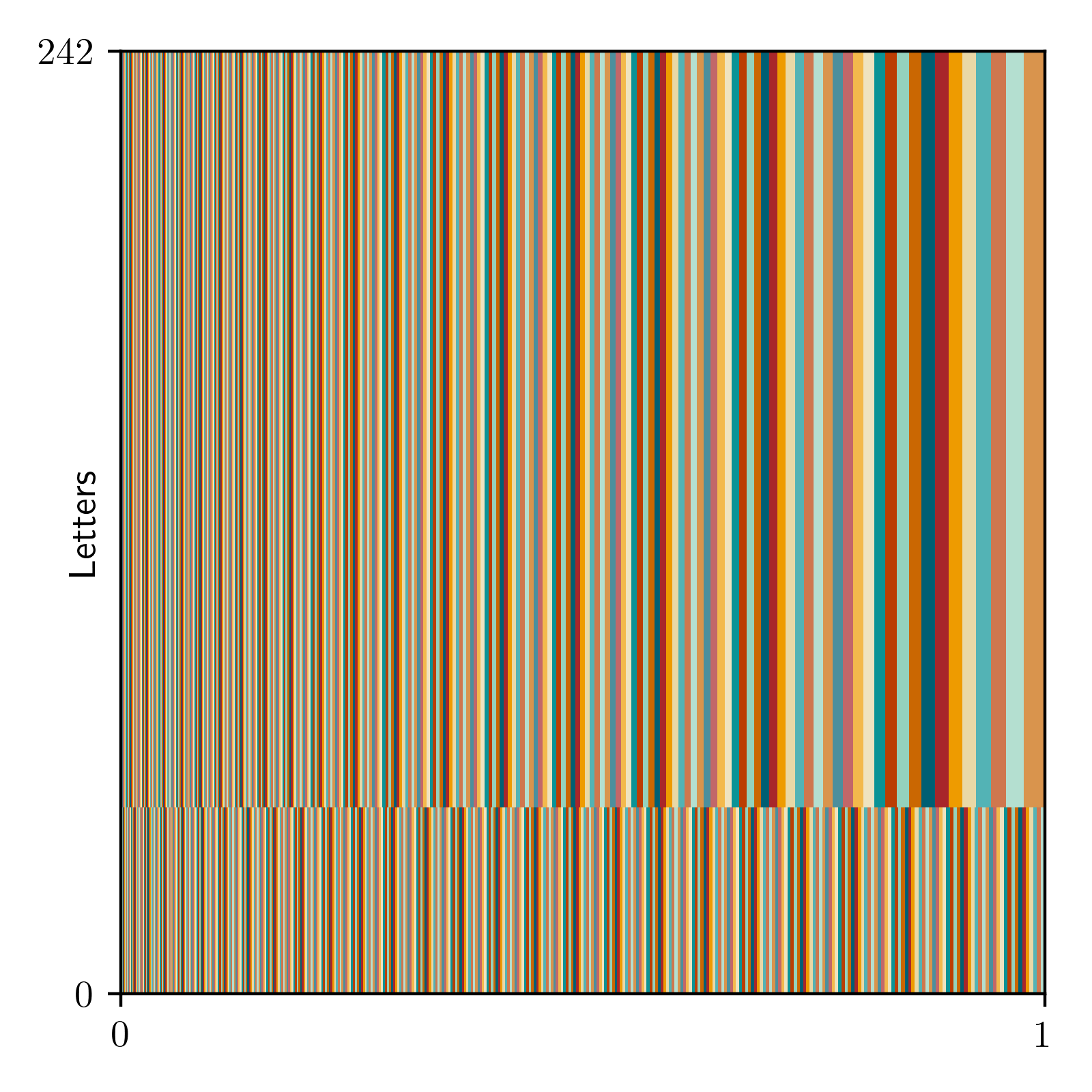}
        \includegraphics[draft=\draft, width=\linewidth]{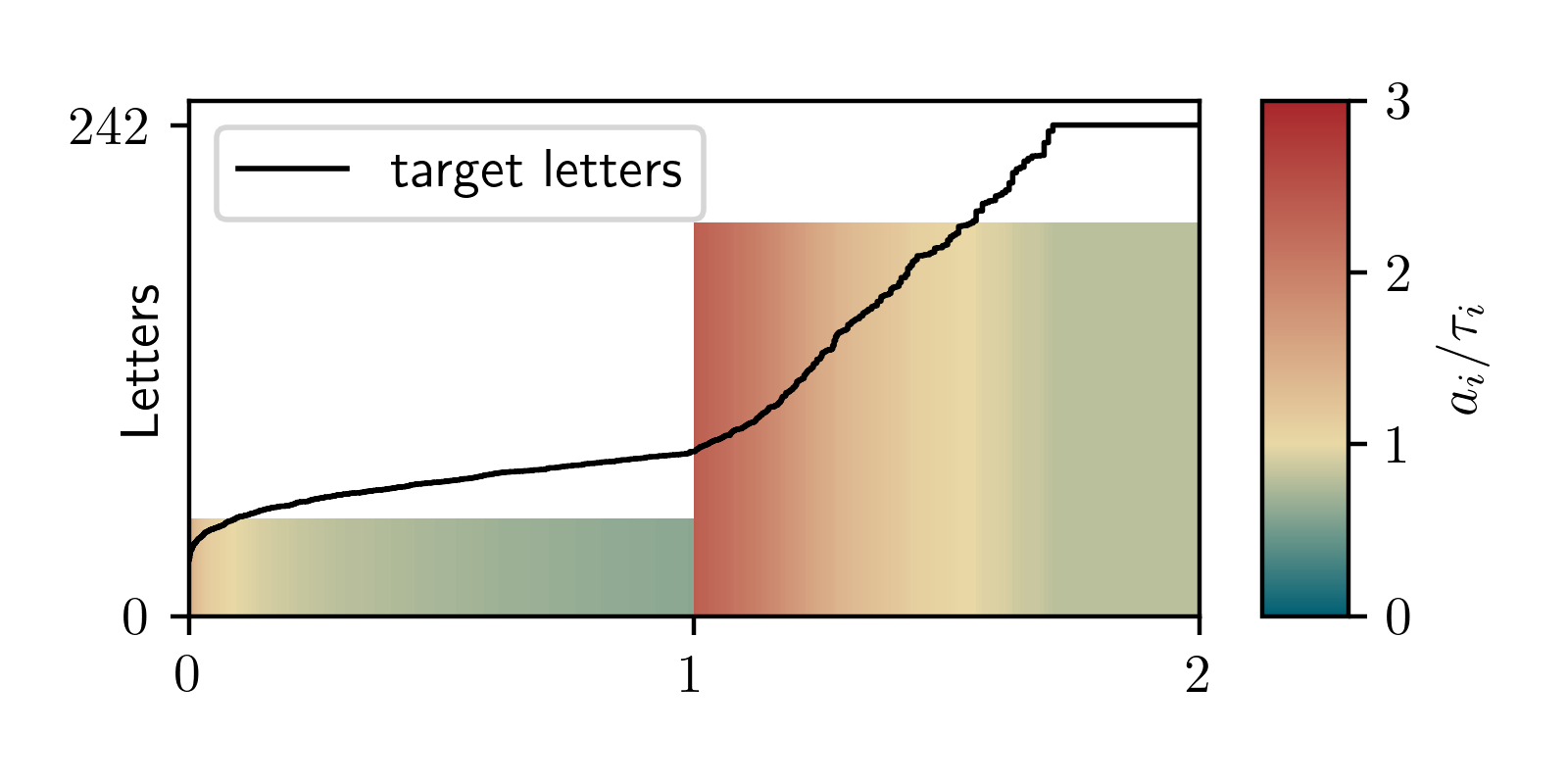}
        \caption{\buckets ($t_G = 2$)}
        \label{fig:results_Mecklenburg-Vorpommern_Small_greedy_bucket_fill}
    \end{subfigure}
    \caption{Small municipalities of Mecklenburg-Vorpommern ($\ell_G = 242$)}
    \label{fig:results_Mecklenburg-Vorpommern_Small}
\end{figure} 

\begin{figure}
    \centering
    \begin{subfigure}{0.32\textwidth}
        \includegraphics[draft=\draft, width=\linewidth]{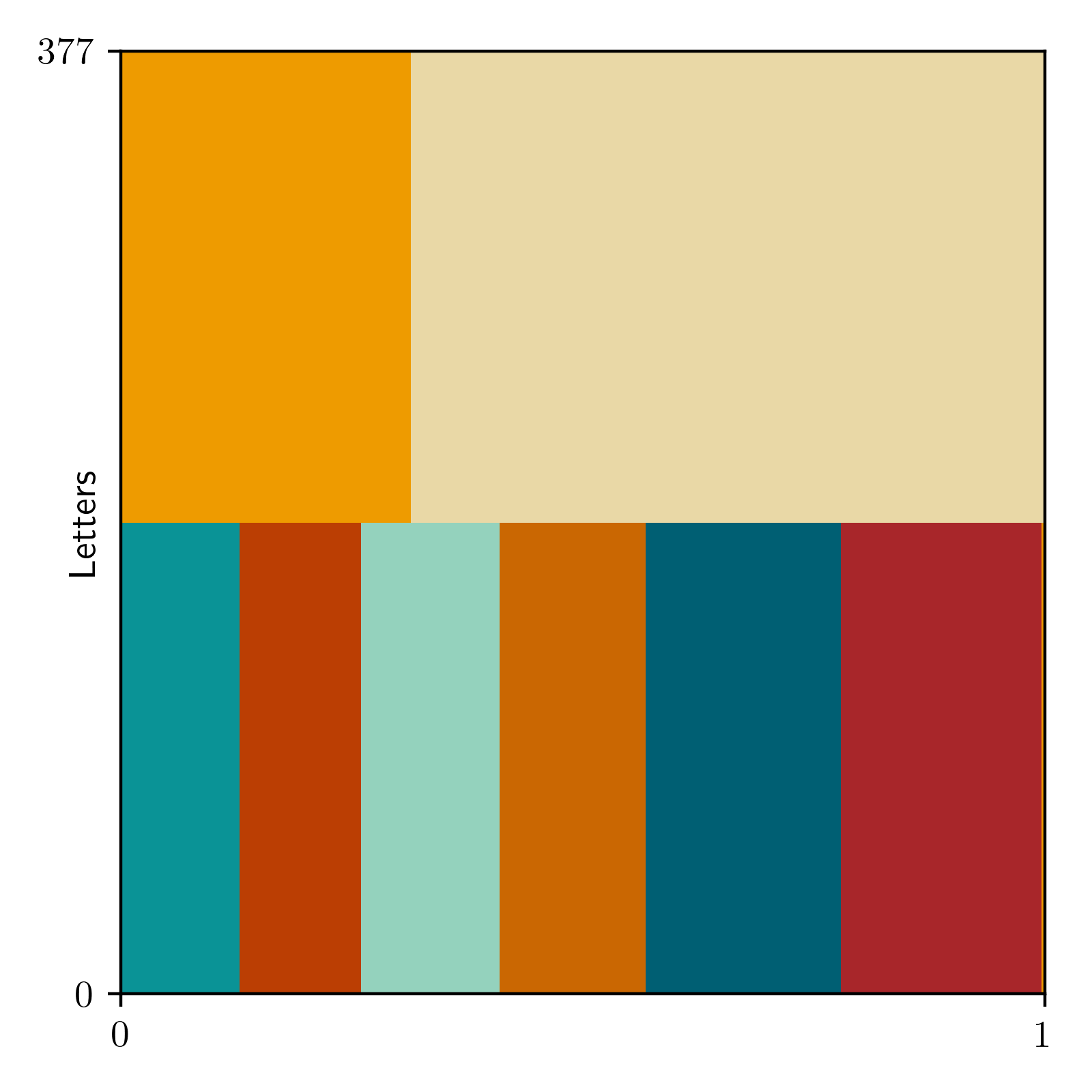}
        \includegraphics[draft=\draft, width=\linewidth]{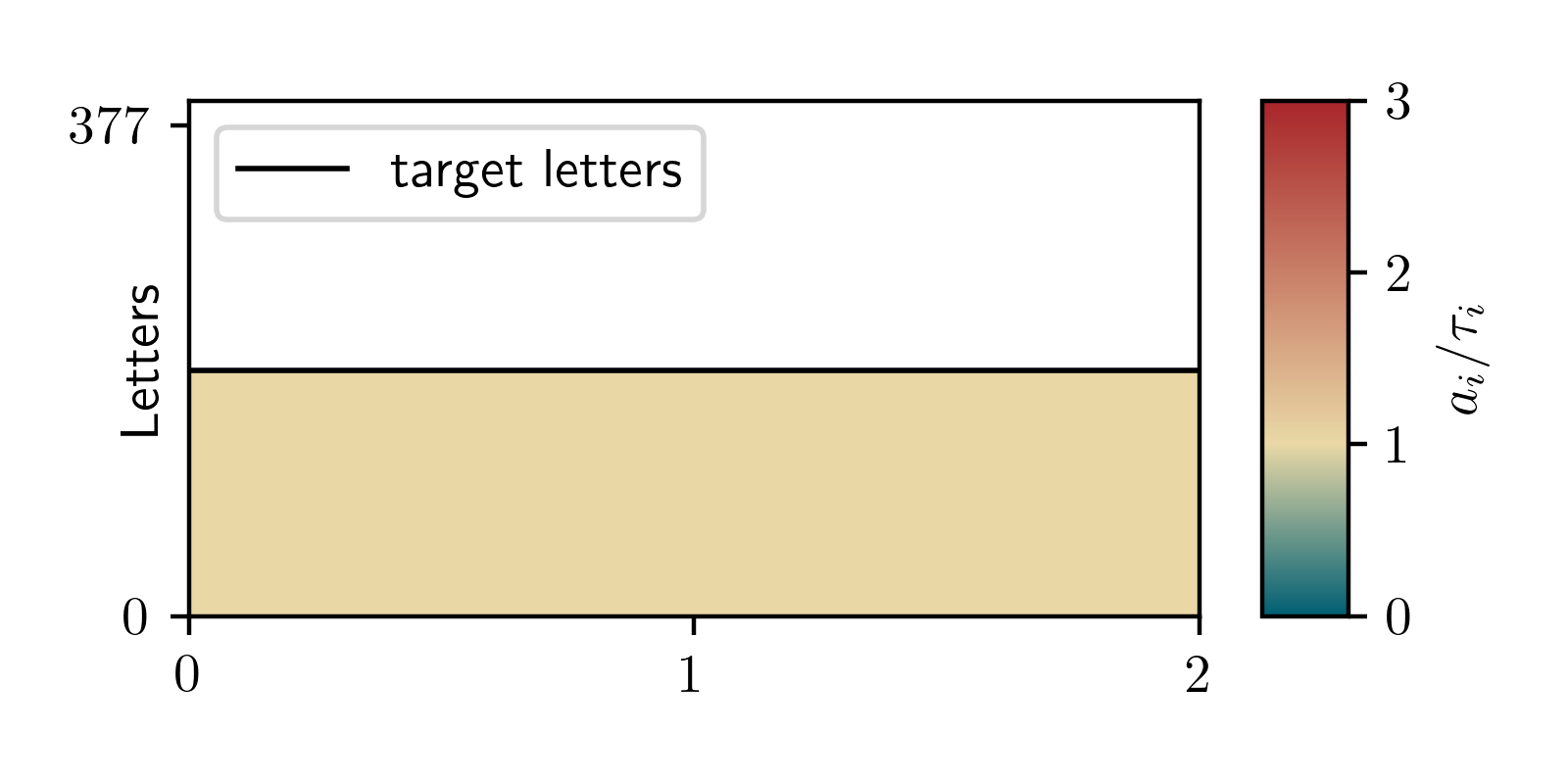}
        \caption{\greq ($t_G = 2$)}
        \label{fig:results_Niedersachsen_Large_greedy_equal}
    \end{subfigure}
    \begin{subfigure}{0.32\textwidth}
        \includegraphics[draft=\draft, width=\linewidth]{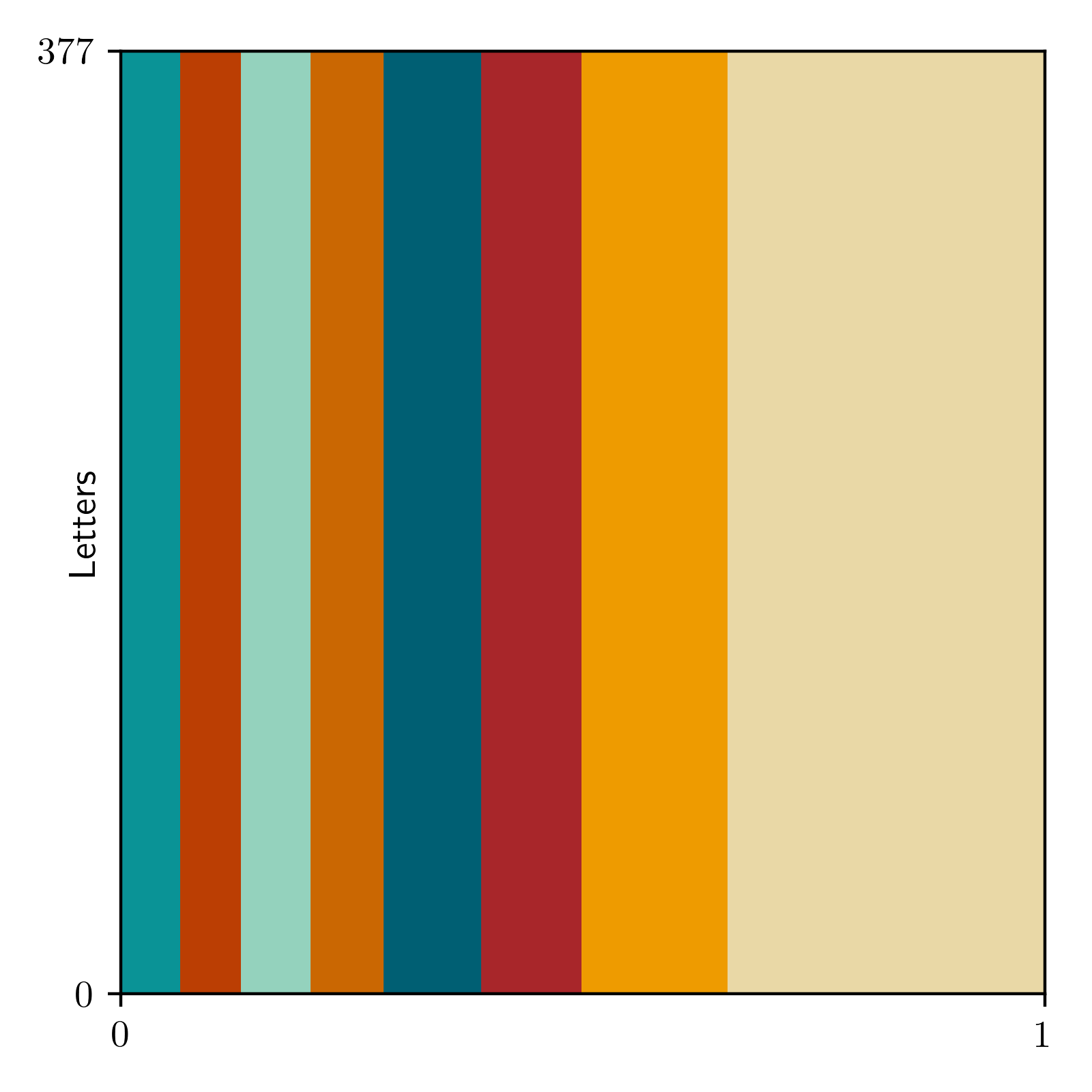}
        \includegraphics[draft=\draft, width=\linewidth]{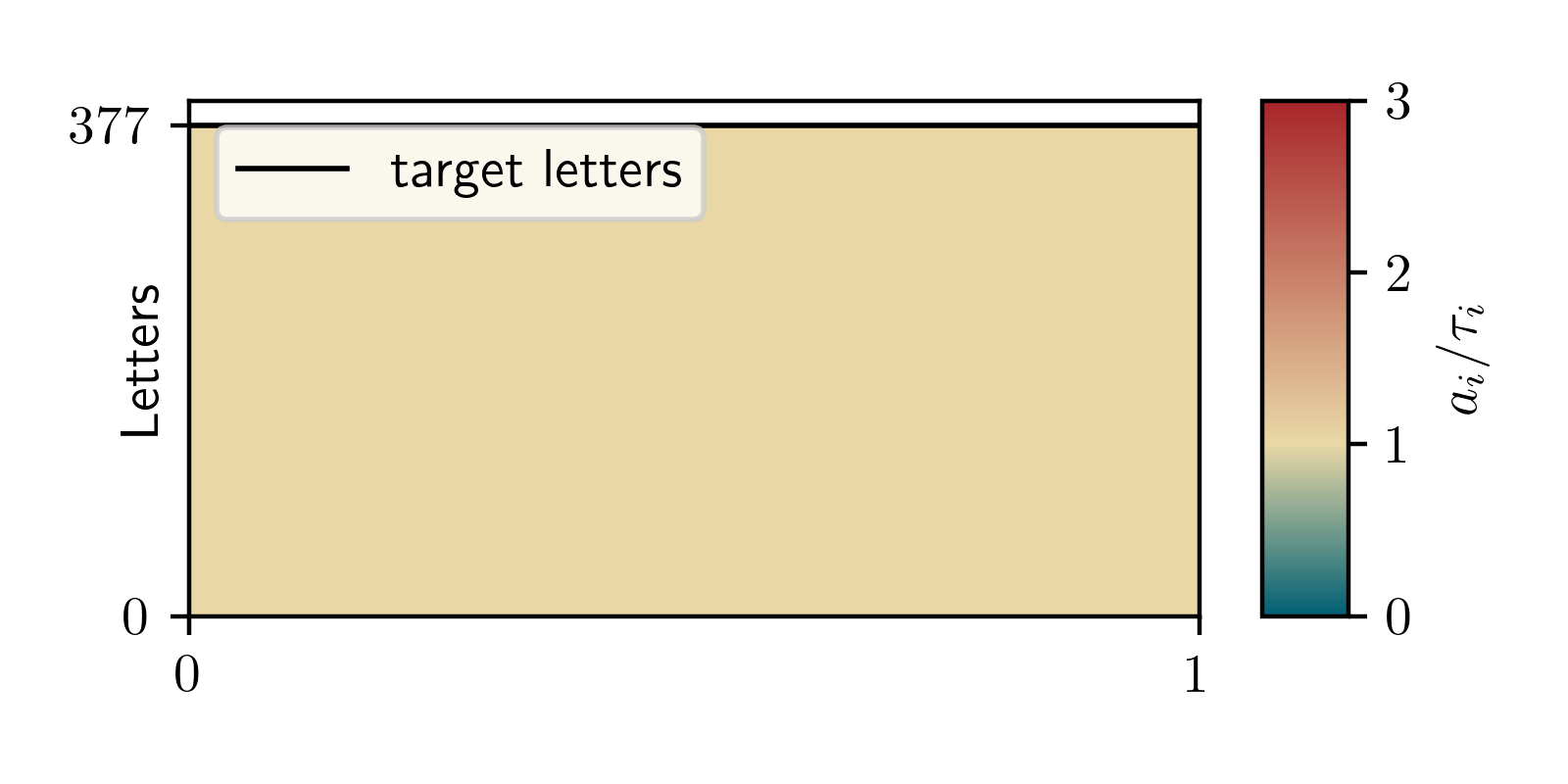}
        \caption{\colgen ($t_G\!=\!1$)}
        \label{fig:results_Niedersachsen_Large_column_generation}
    \end{subfigure}
    \begin{subfigure}{0.32\textwidth}
        \includegraphics[draft=\draft, width=\linewidth]{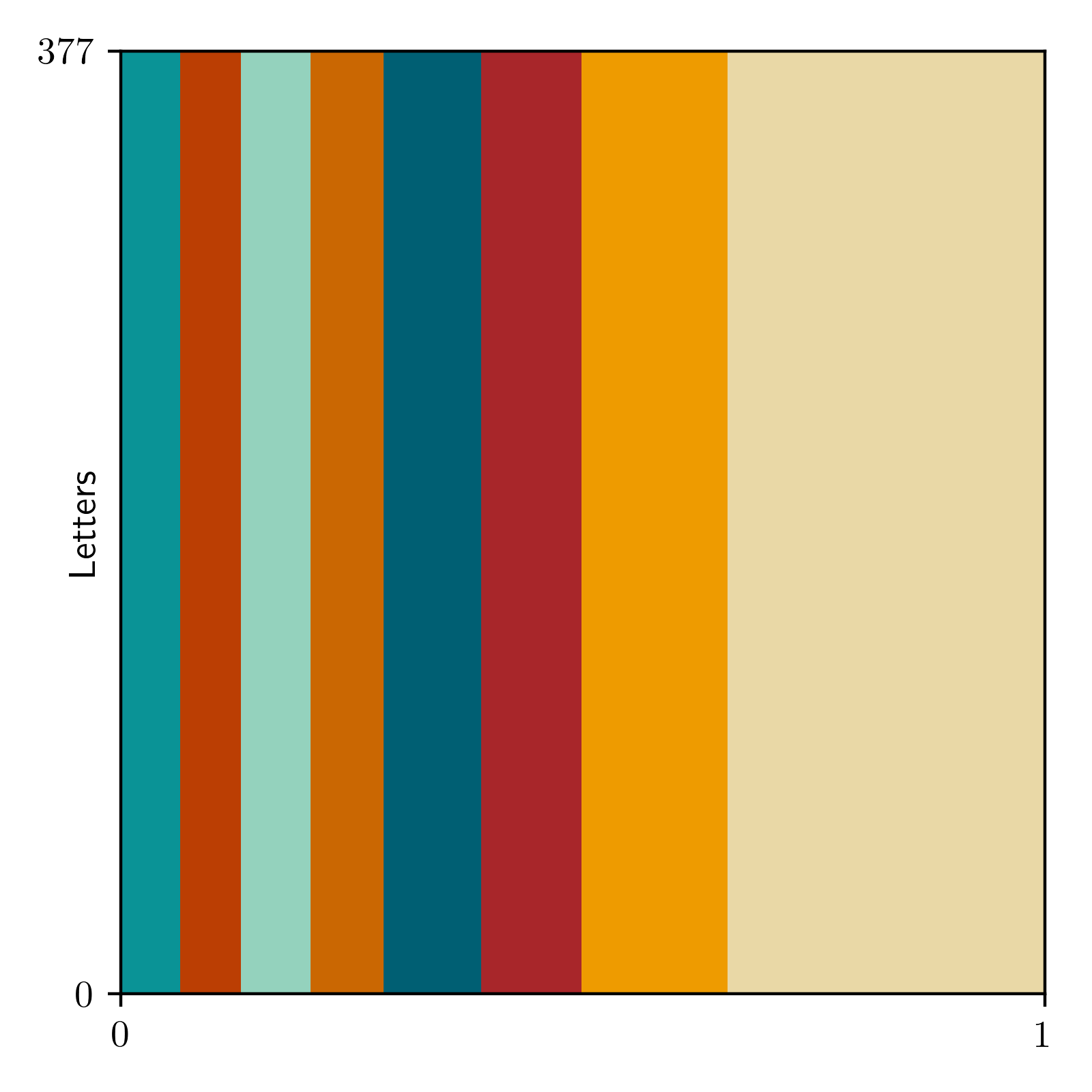}
        \includegraphics[draft=\draft, width=\linewidth]{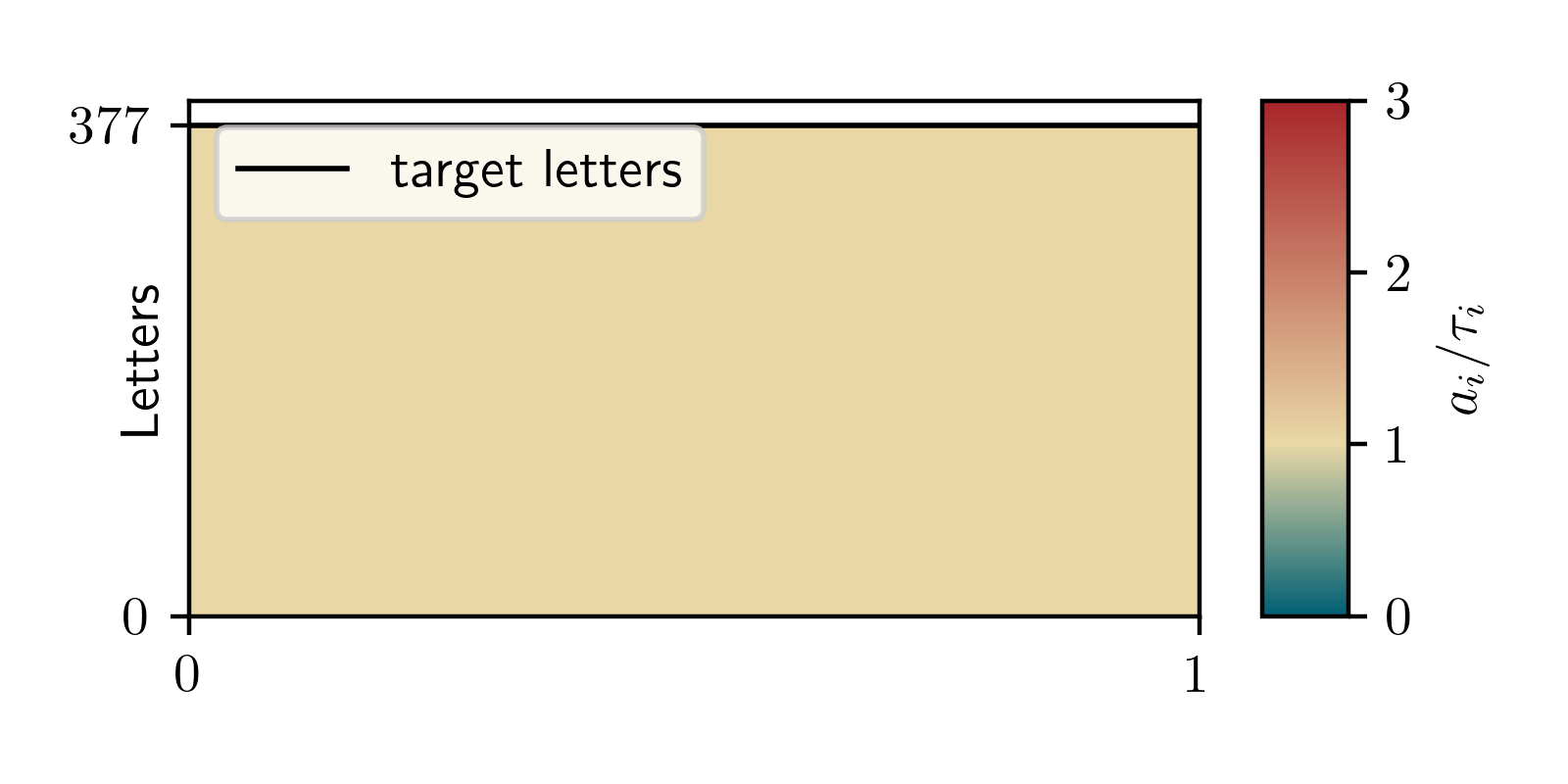}
        \caption{\buckets ($t_G = 1$)}
        \label{fig:results_Niedersachsen_Large_greedy_bucket_fill}
    \end{subfigure}
    \caption{Large municipalities of Niedersachsen ($\ell_G = 377$)}
    \label{fig:results_Niedersachsen_Large}
\end{figure} 

\begin{figure}
    \centering
    \begin{subfigure}{0.32\textwidth}
        \includegraphics[draft=\draft, width=\linewidth]{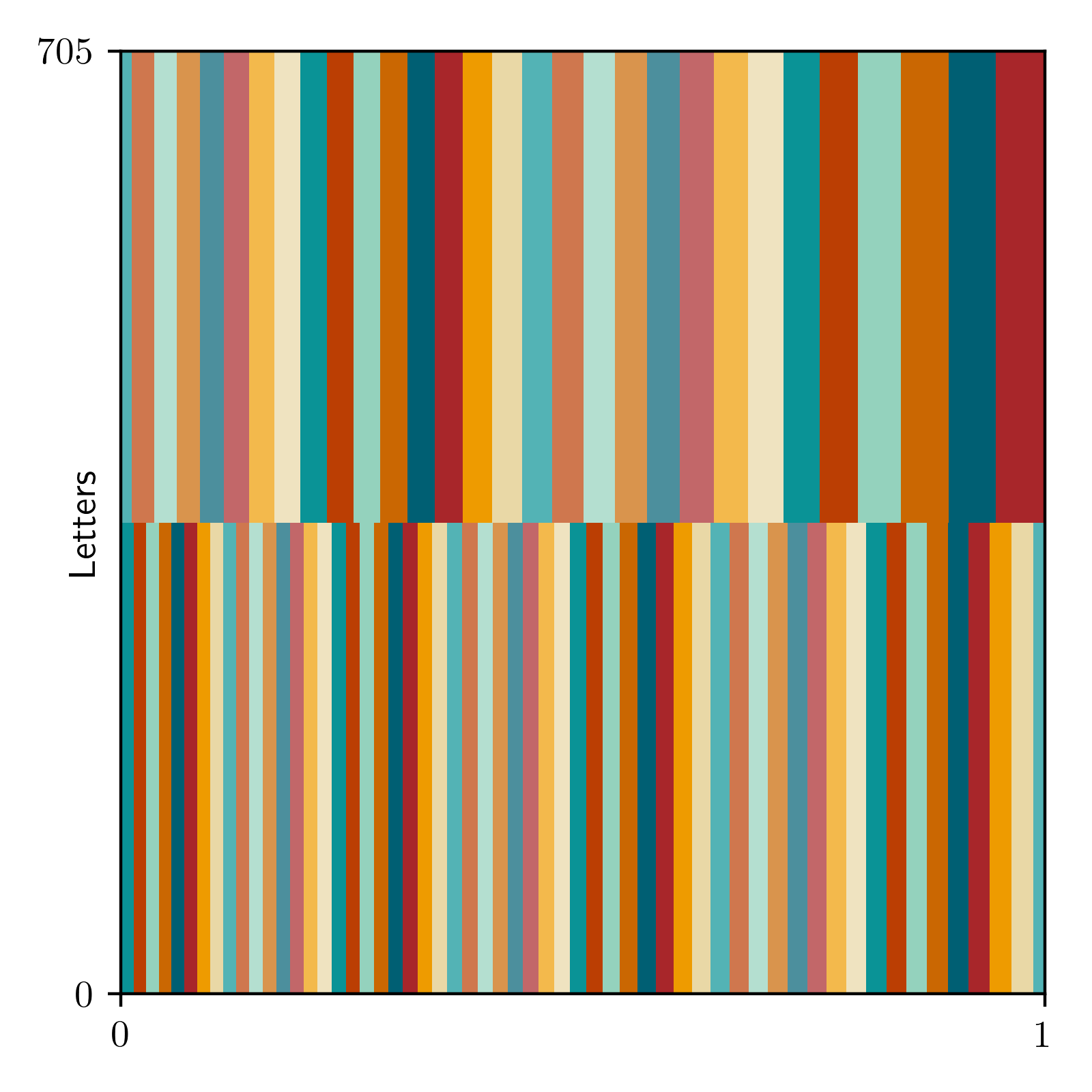}
        \includegraphics[draft=\draft, width=\linewidth]{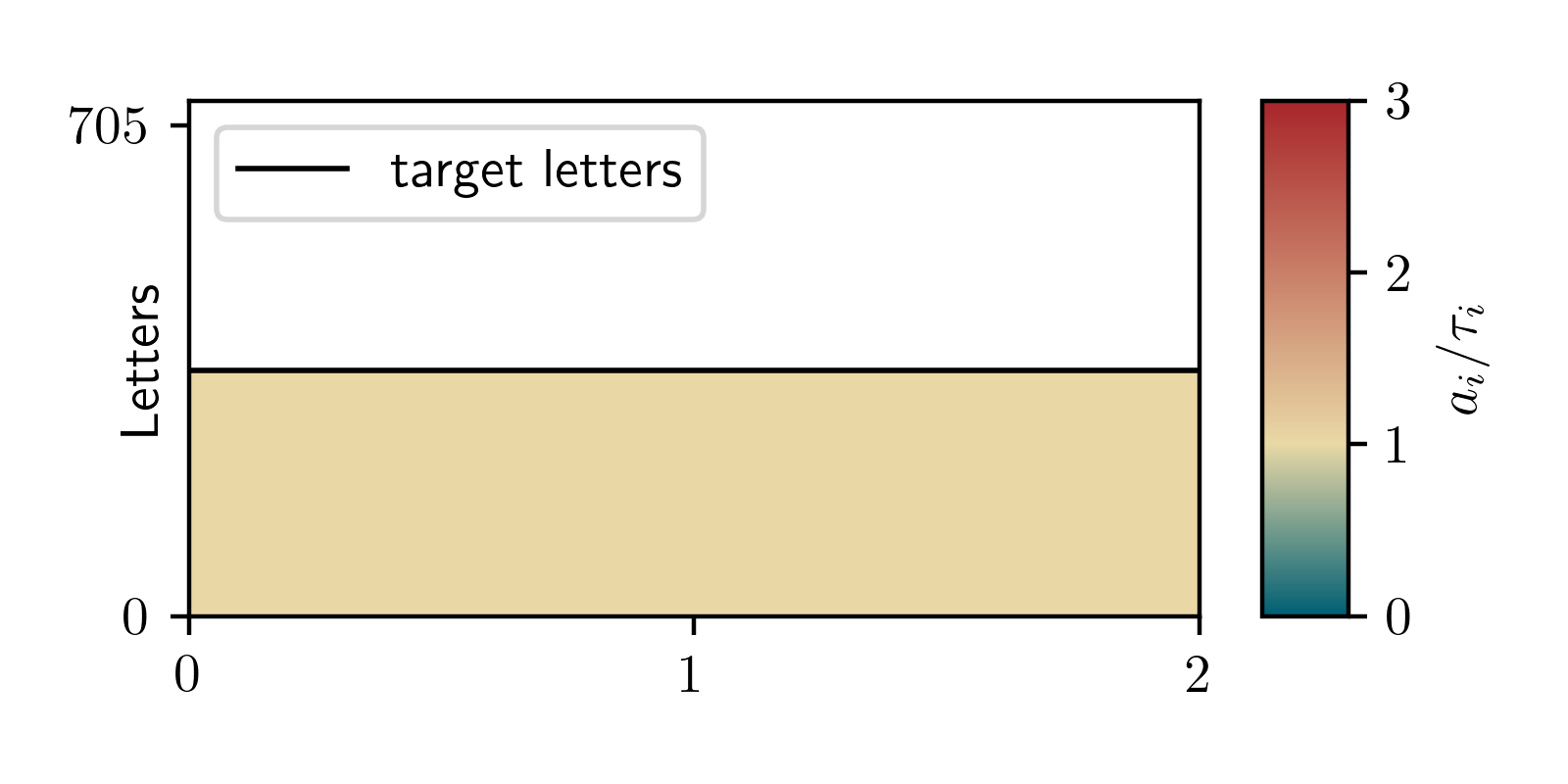}
        \caption{\greq ($t_G = 2$)}
        \label{fig:results_Niedersachsen_Medium_greedy_equal}
    \end{subfigure}
    \begin{subfigure}{0.32\textwidth}
        \includegraphics[draft=\draft, width=\linewidth]{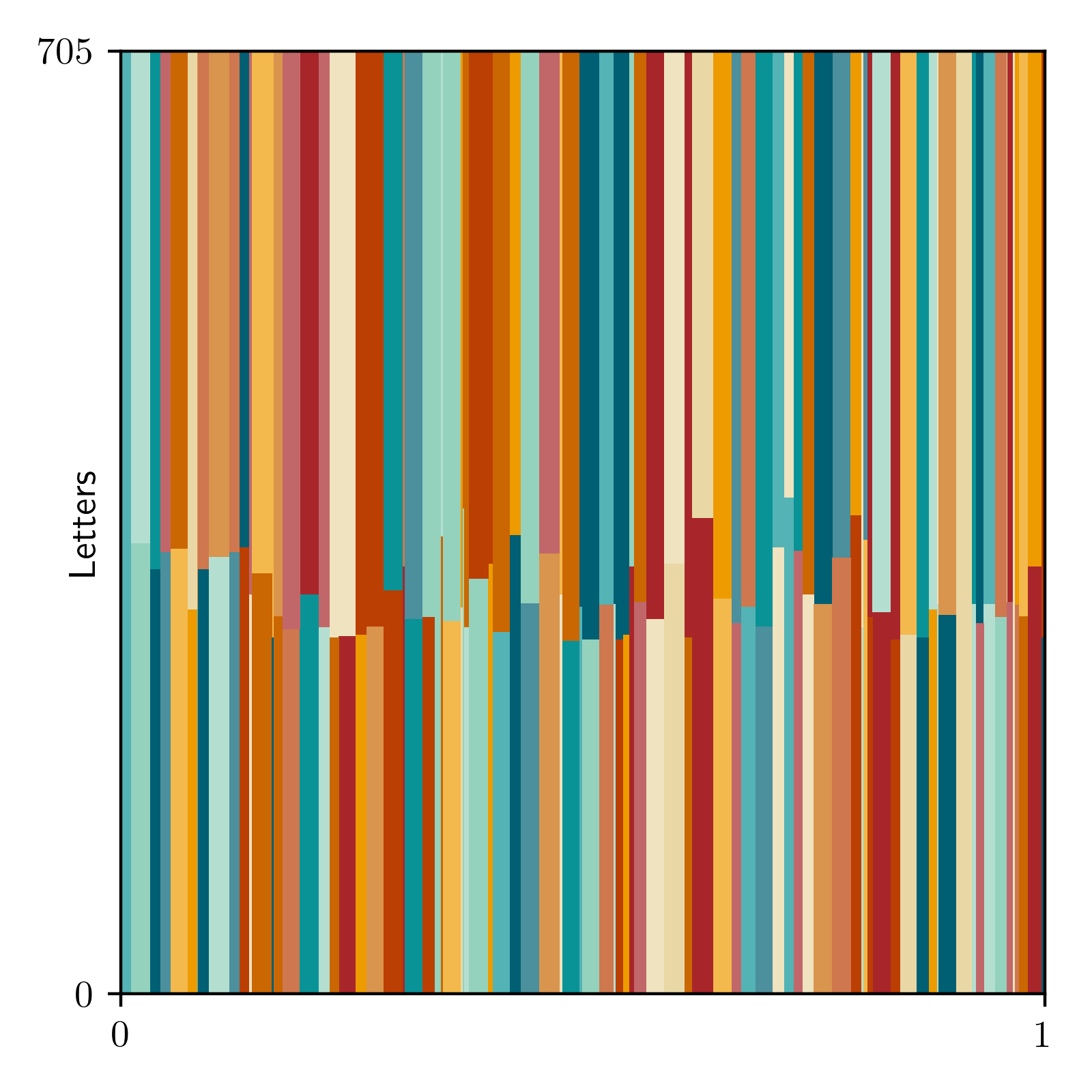}
        \includegraphics[draft=\draft, width=\linewidth]{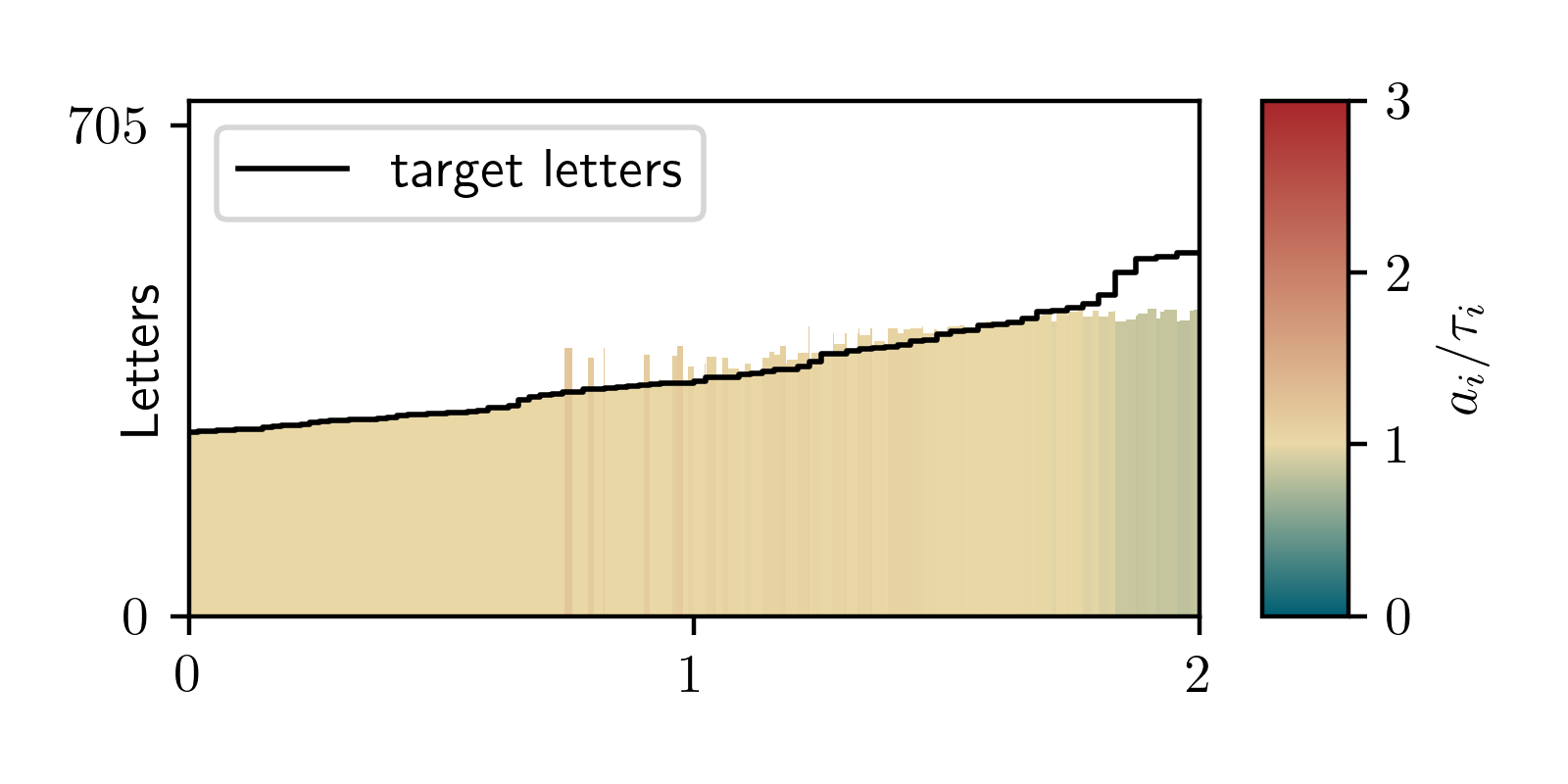}
        \caption{\colgen ($t_G\!=\!2$)}
        \label{fig:results_Niedersachsen_Medium_column_generation}
    \end{subfigure}
    \begin{subfigure}{0.32\textwidth}
        \includegraphics[draft=\draft, width=\linewidth]{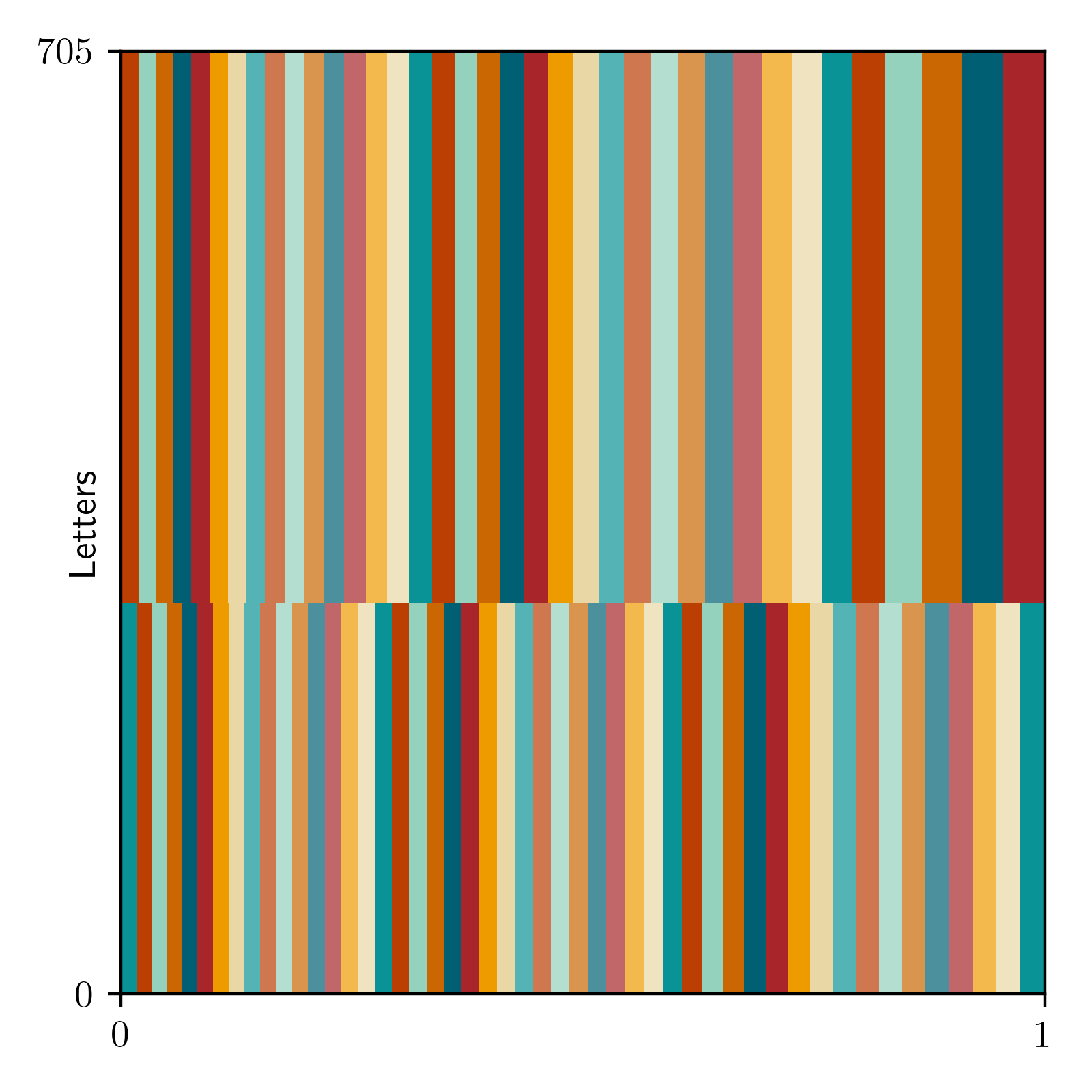}
        \includegraphics[draft=\draft, width=\linewidth]{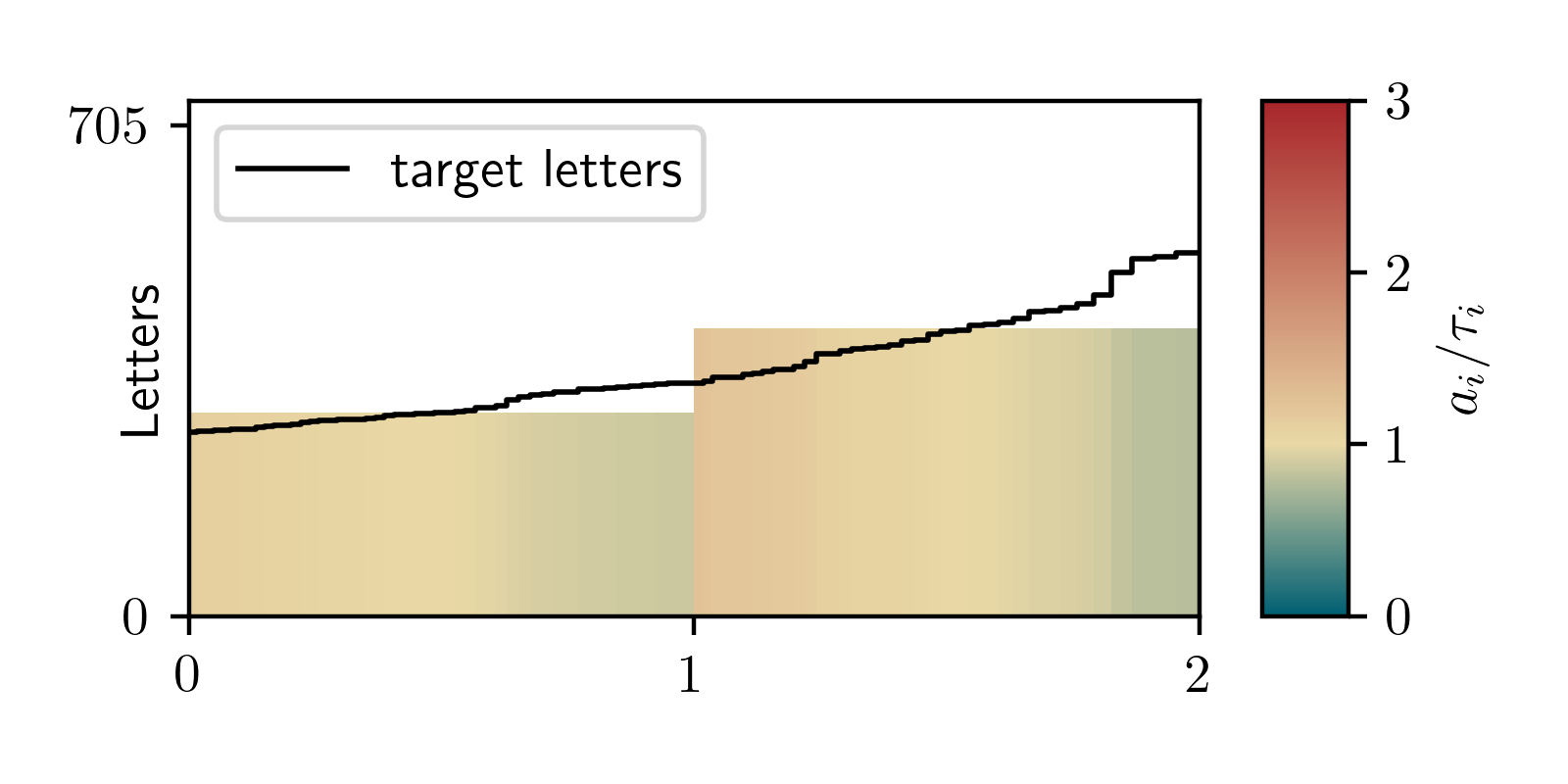}
        \caption{\buckets ($t_G = 2$)}
        \label{fig:results_Niedersachsen_Medium_greedy_bucket_fill}
    \end{subfigure}
    \caption{Medium municipalities of Niedersachsen ($\ell_G = 705$)}
    \label{fig:results_Niedersachsen_Medium}
\end{figure} 

\begin{figure}
    \centering
    \begin{subfigure}{0.32\textwidth}
        \includegraphics[draft=\draft, width=\linewidth]{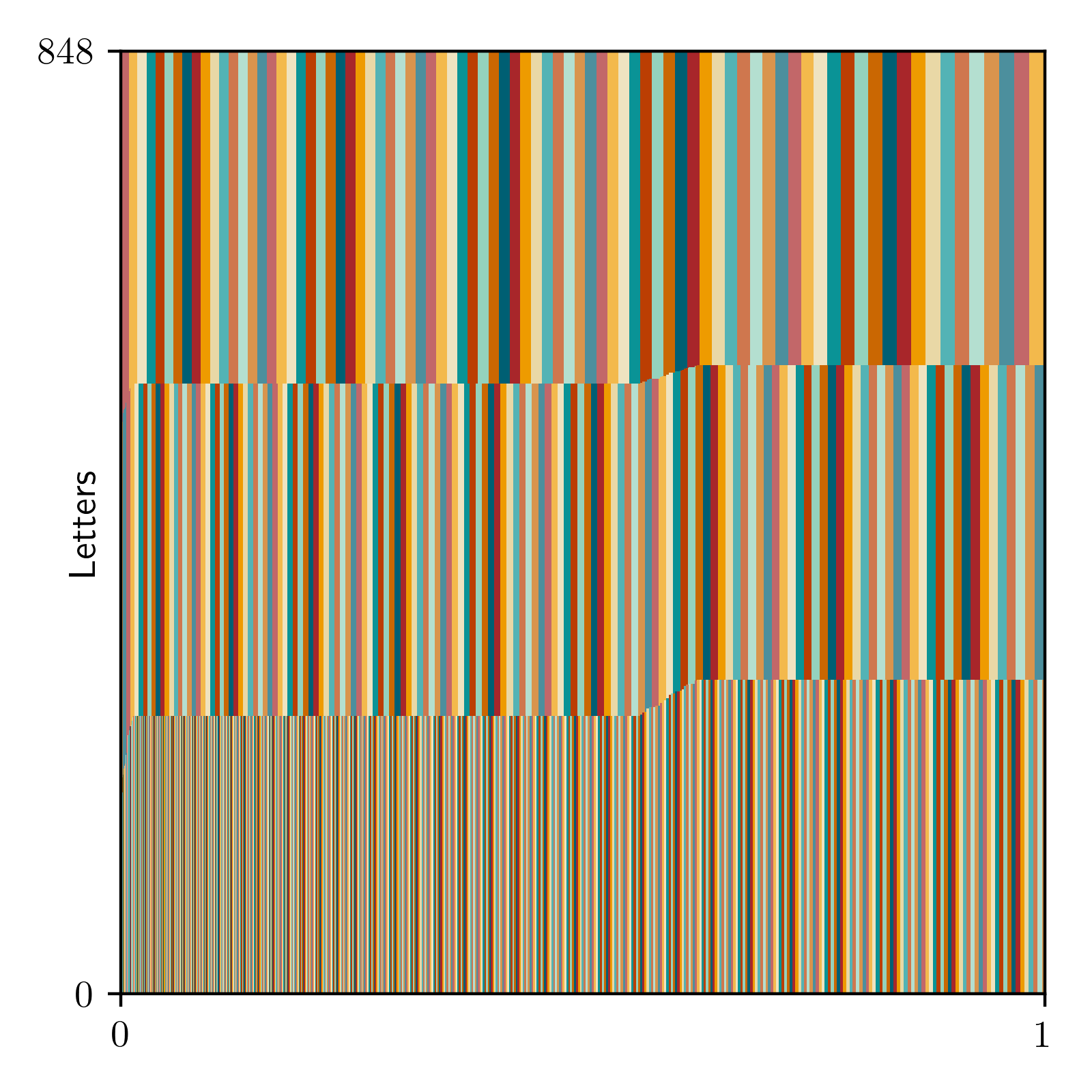}
        \includegraphics[draft=\draft, width=\linewidth]{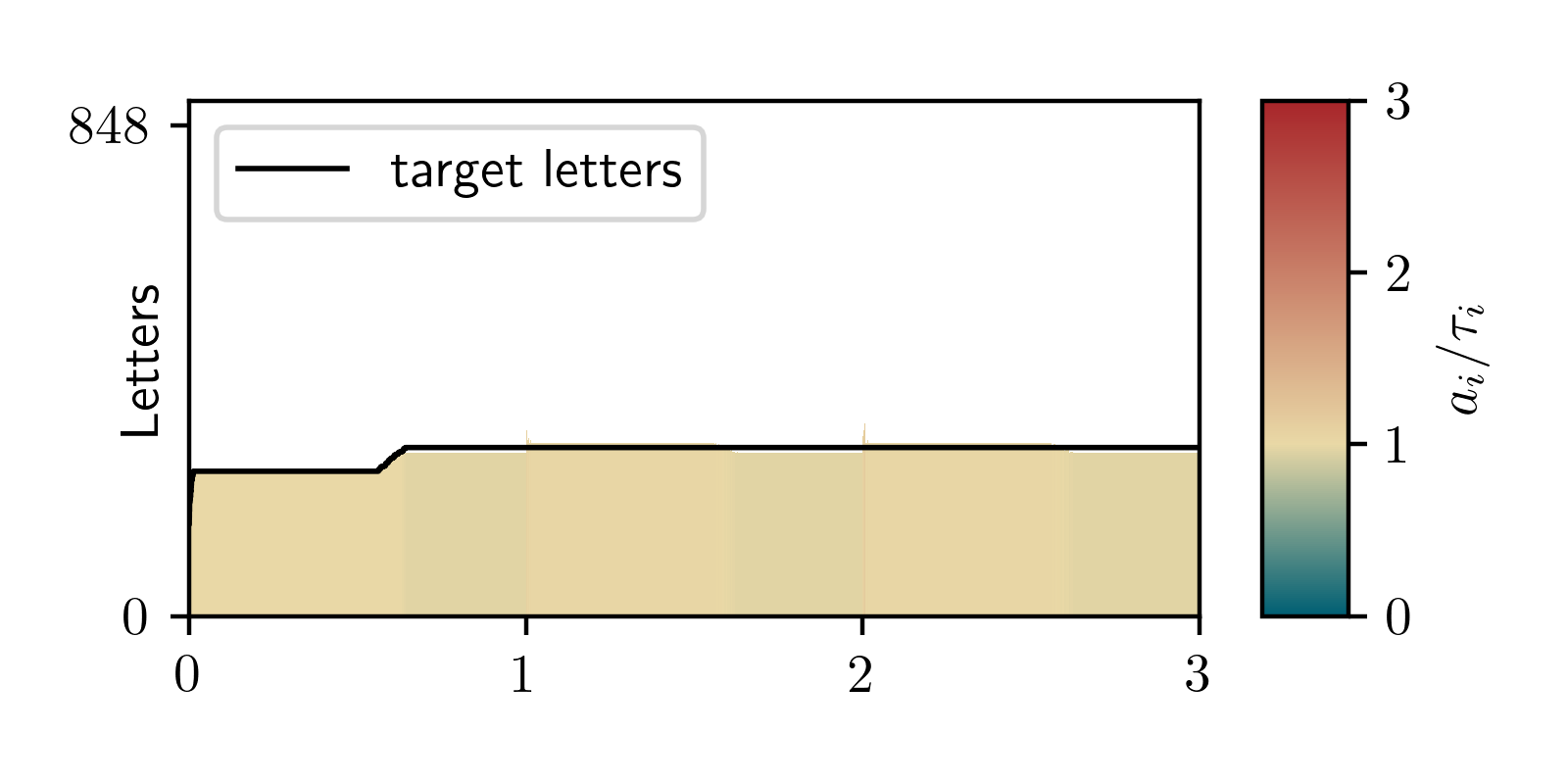}
        \caption{\greq ($t_G = 3$)}
        \label{fig:results_Niedersachsen_Small_greedy_equal}
    \end{subfigure}
    \begin{subfigure}{0.32\textwidth}
        \includegraphics[draft=\draft, width=\linewidth]{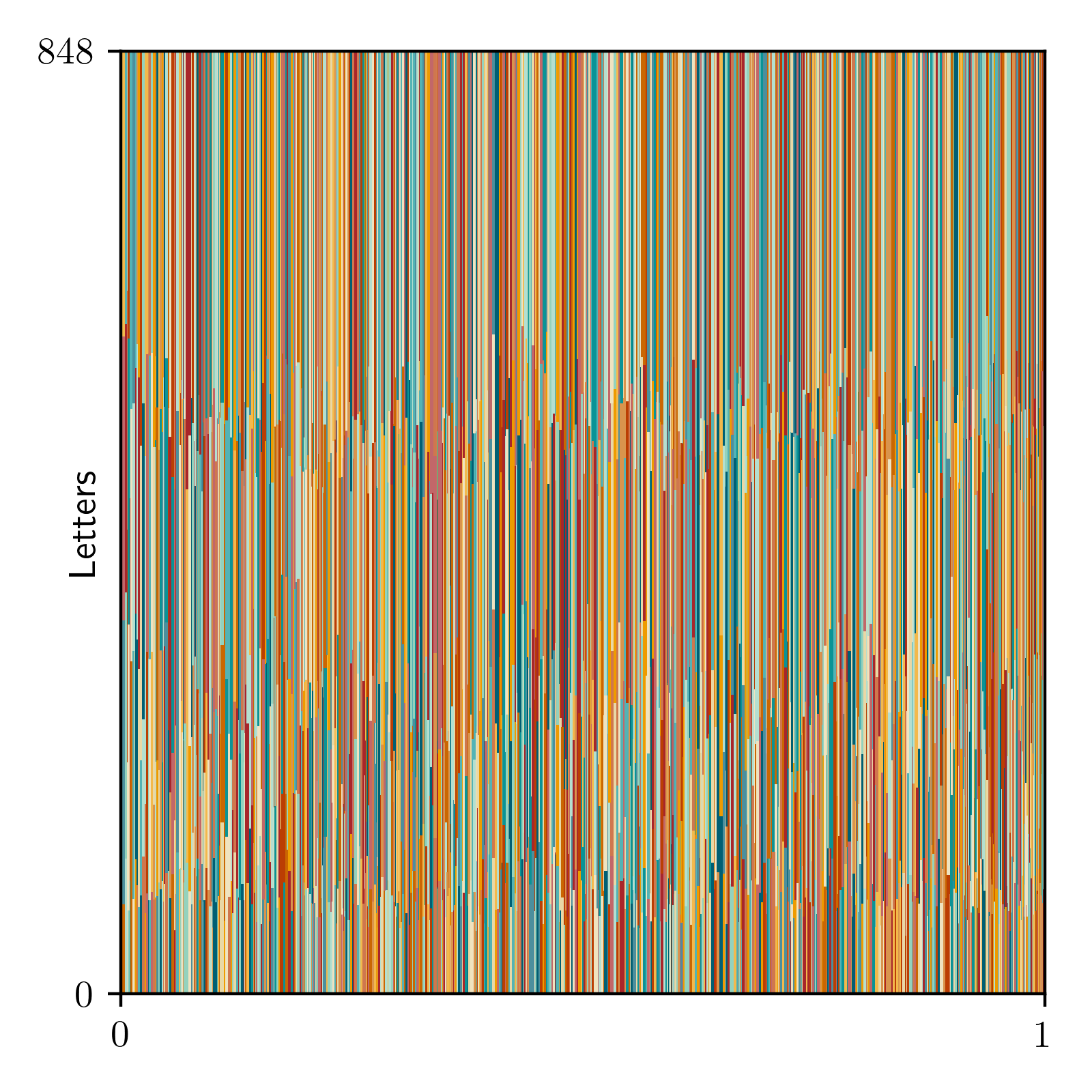}
        \includegraphics[draft=\draft, width=\linewidth]{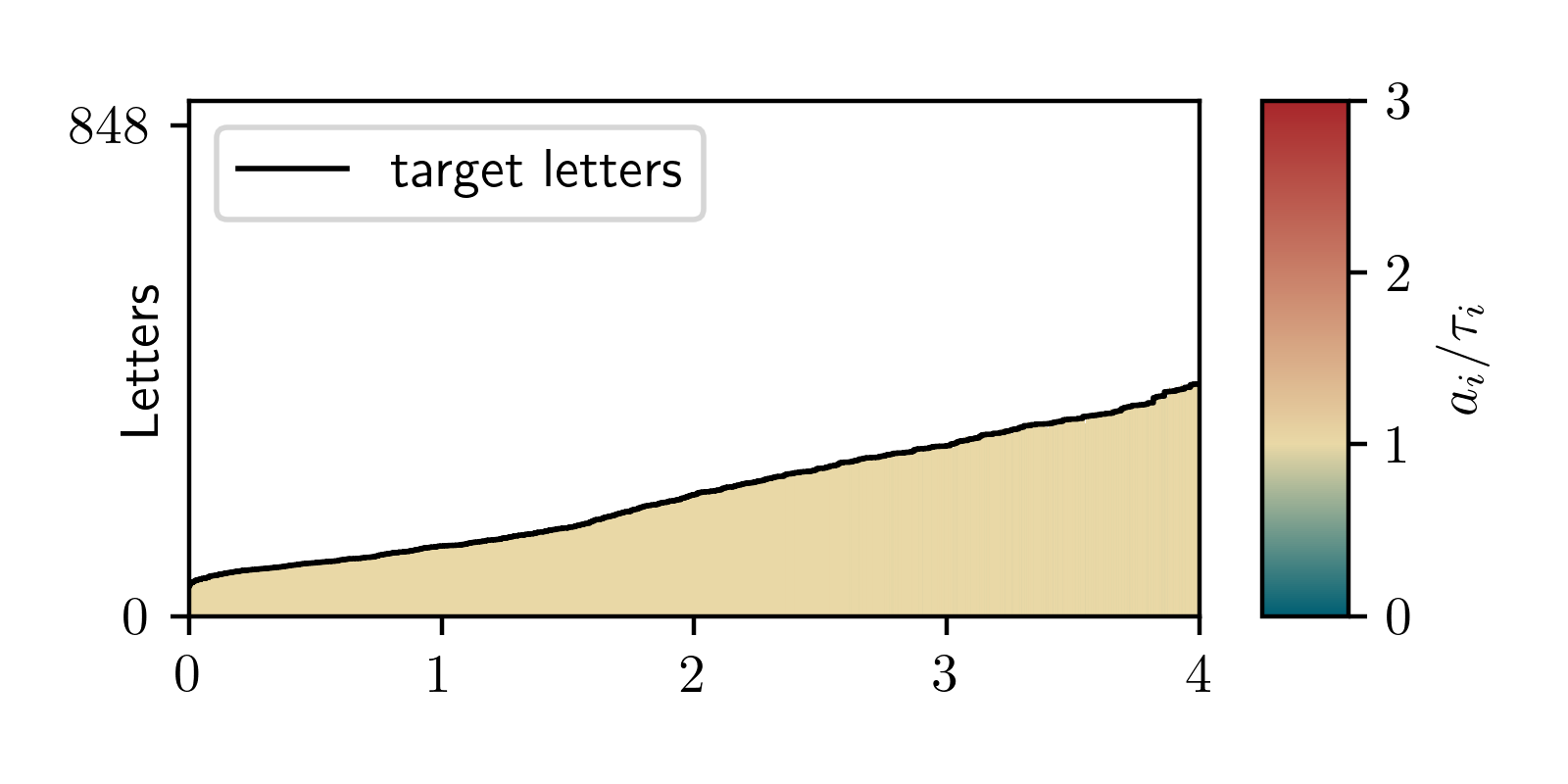}
        \caption{\colgen ($t_G\!=\!4$)}
        \label{fig:results_Niedersachsen_Small_column_generation}
    \end{subfigure}
    \begin{subfigure}{0.32\textwidth}
        \includegraphics[draft=\draft, width=\linewidth]{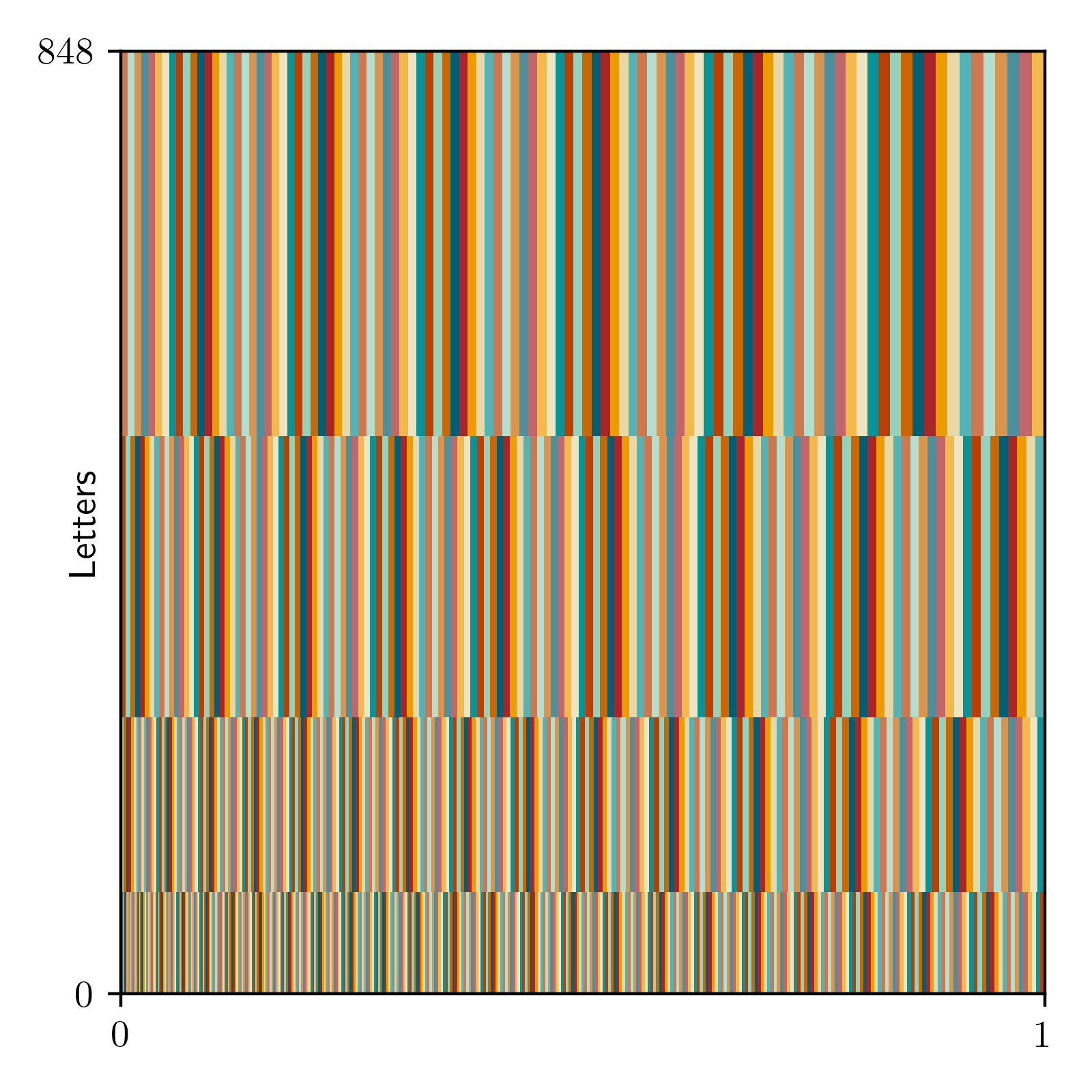}
        \includegraphics[draft=\draft, width=\linewidth]{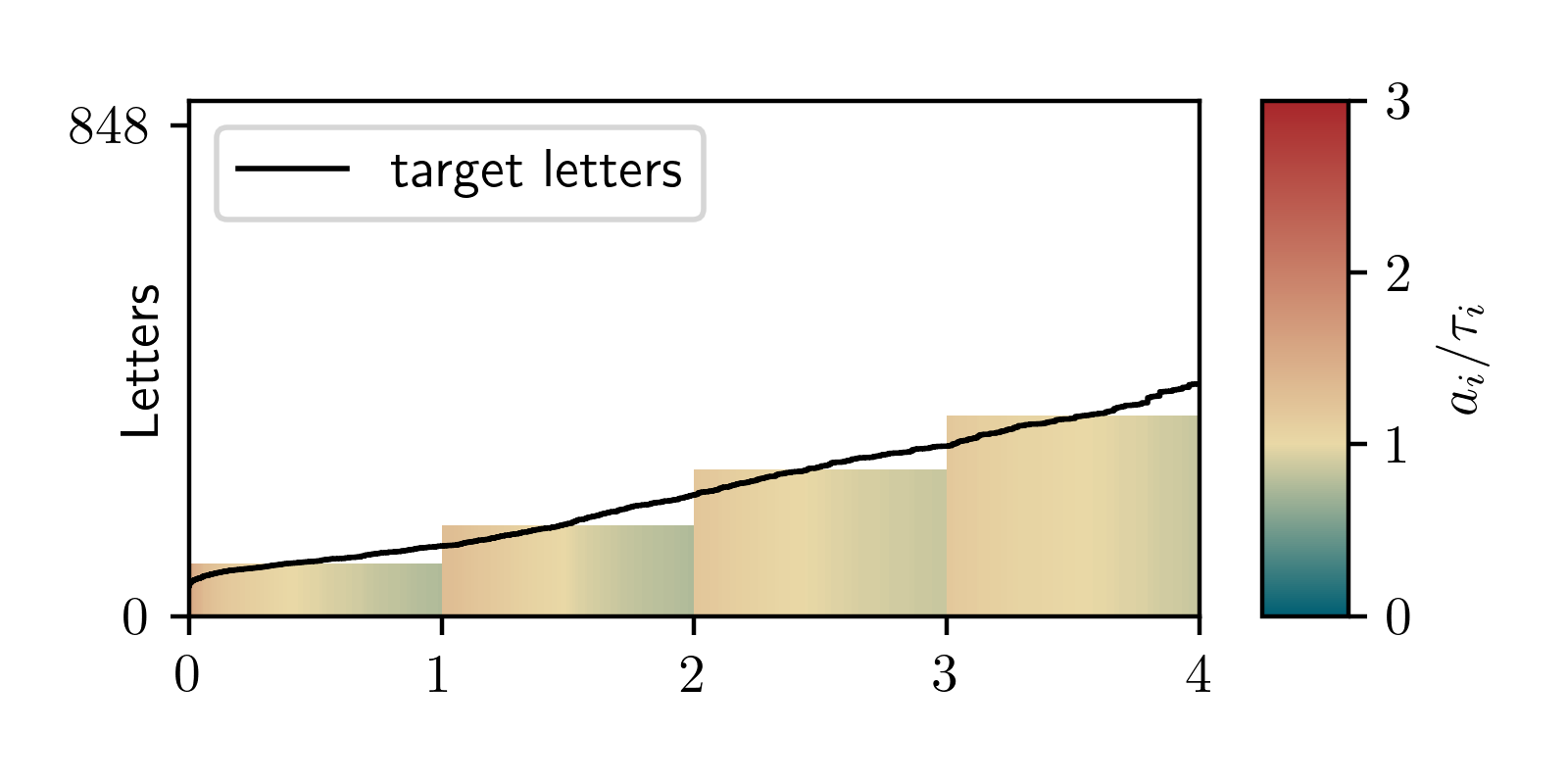}
        \caption{\buckets ($t_G = 4$)}
        \label{fig:results_Niedersachsen_Small_greedy_bucket_fill}
    \end{subfigure}
    \caption{Small municipalities of Niedersachsen ($\ell_G = 848$)}
    \label{fig:results_Niedersachsen_Small}
\end{figure} 

\begin{figure}
    \centering
    \begin{subfigure}{0.32\textwidth}
        \includegraphics[draft=\draft, width=\linewidth]{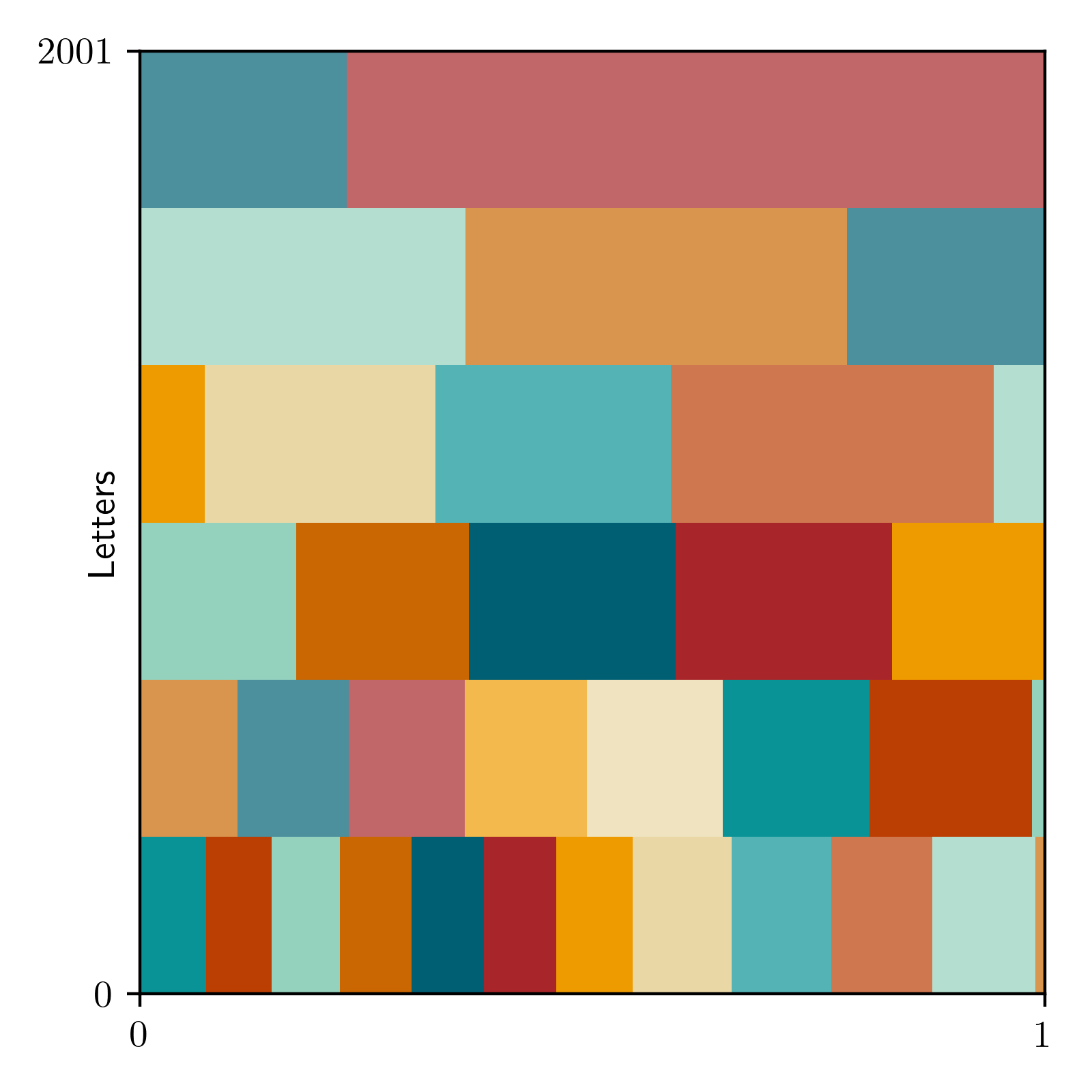}
        \includegraphics[draft=\draft, width=\linewidth]{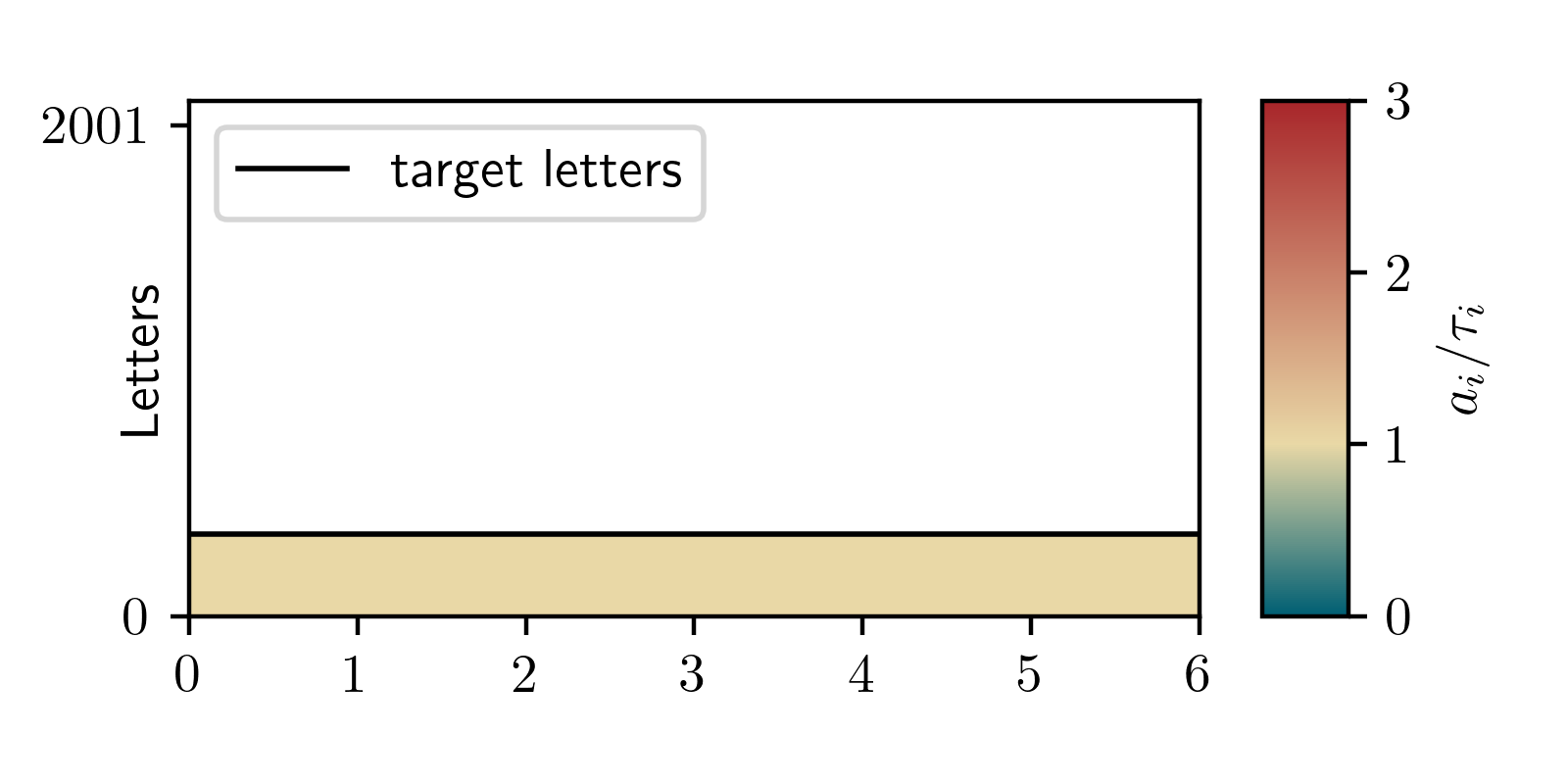}
        \caption{\greq ($t_G = 6$)}
        \label{fig:results_Nordrhein-Westfalen_Large_greedy_equal}
    \end{subfigure}
    \begin{subfigure}{0.32\textwidth}
        \includegraphics[draft=\draft, width=\linewidth]{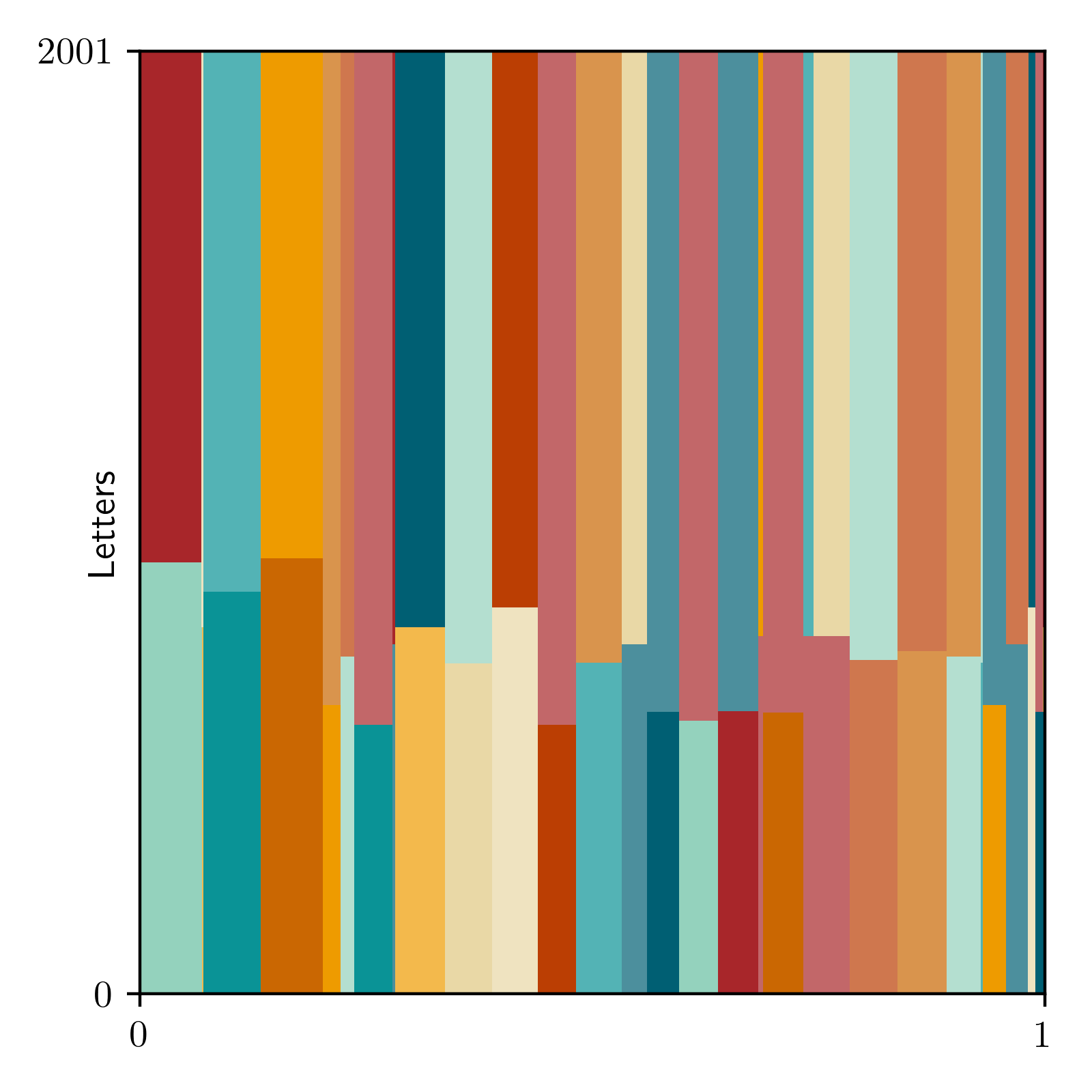}
        \includegraphics[draft=\draft, width=\linewidth]{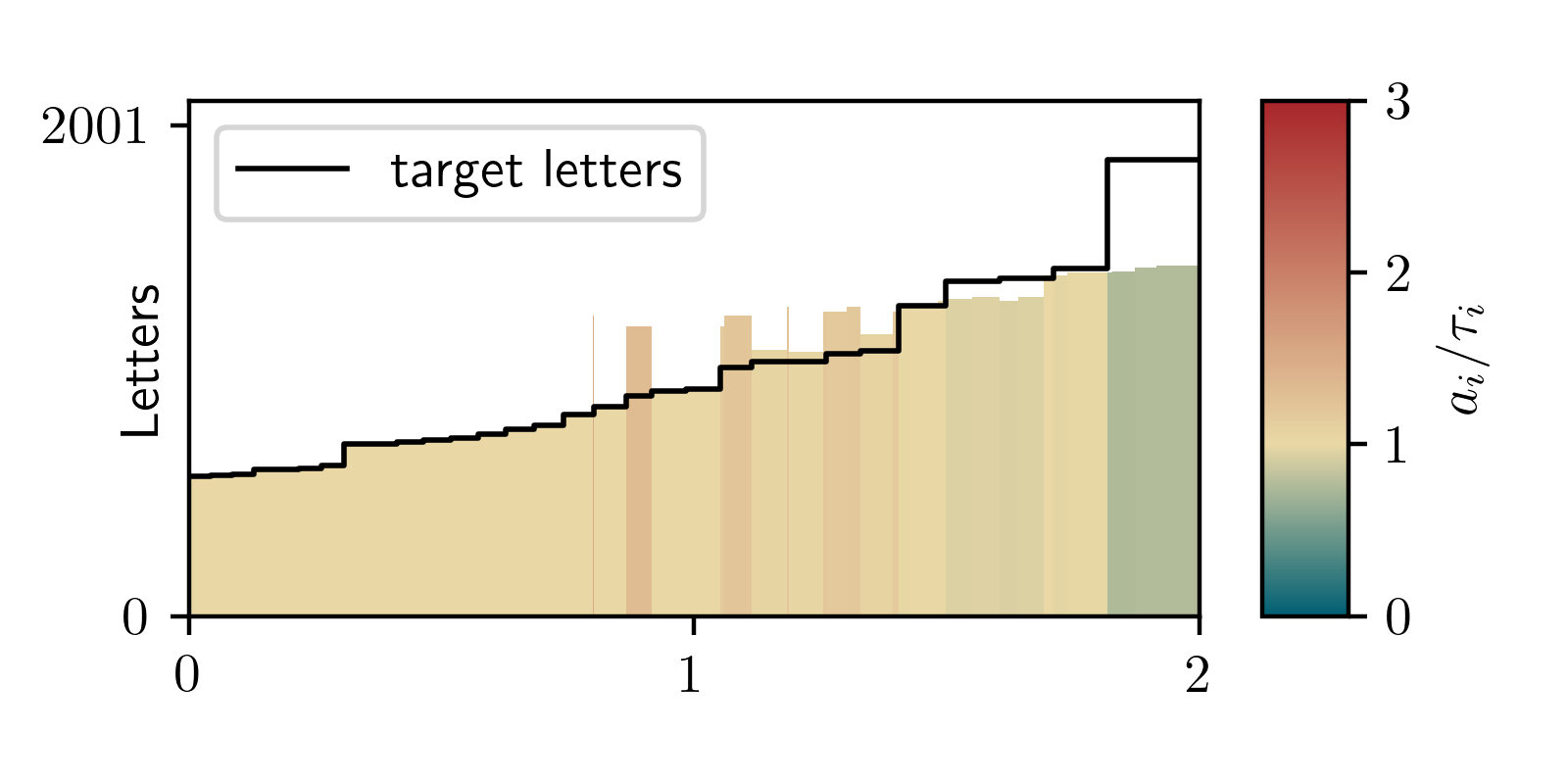}
        \caption{\colgen ($t_G\!=\!2$)}
        \label{fig:results_Nordrhein-Westfalen_Large_column_generation}
    \end{subfigure}
    \begin{subfigure}{0.32\textwidth}
        \includegraphics[draft=\draft, width=\linewidth]{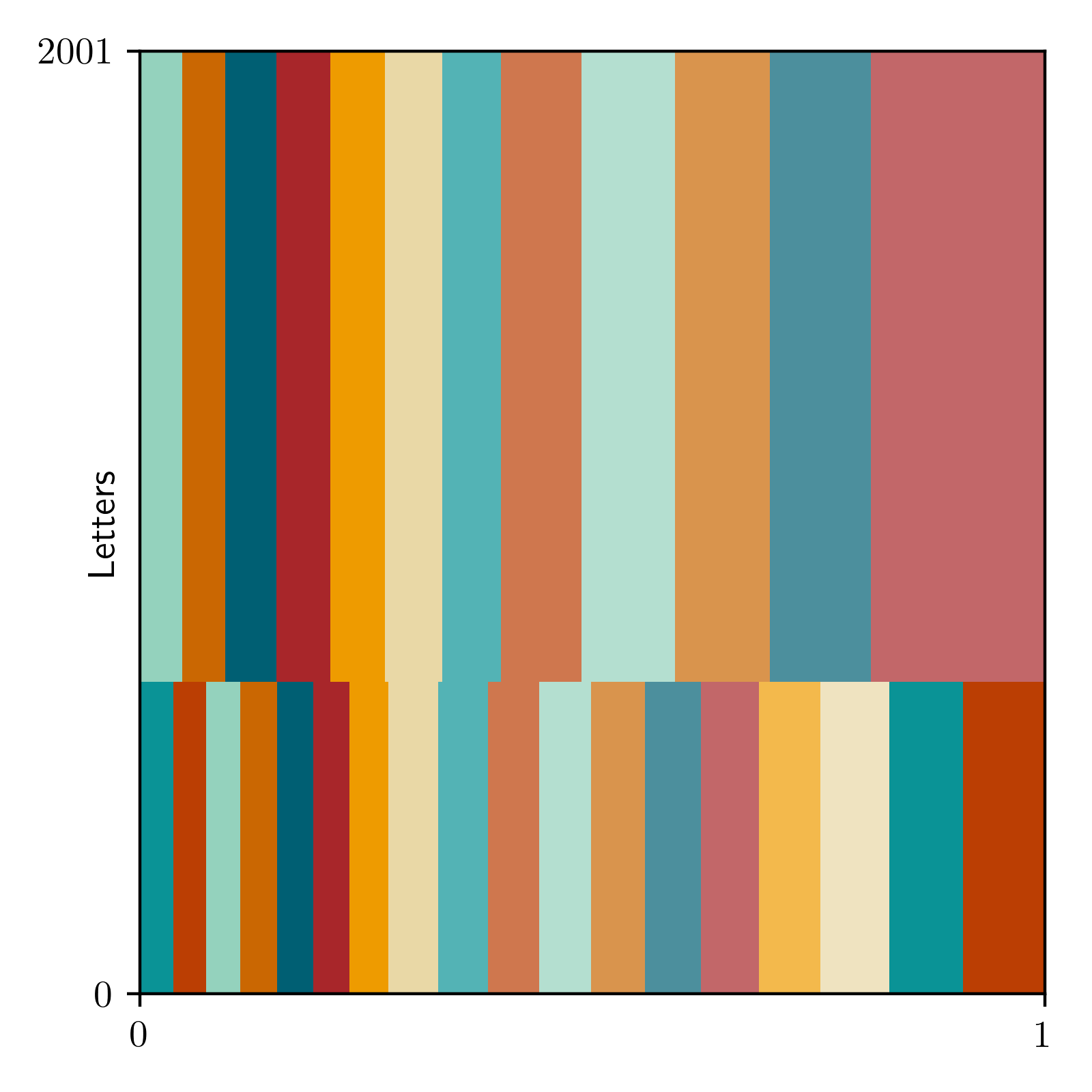}
        \includegraphics[draft=\draft, width=\linewidth]{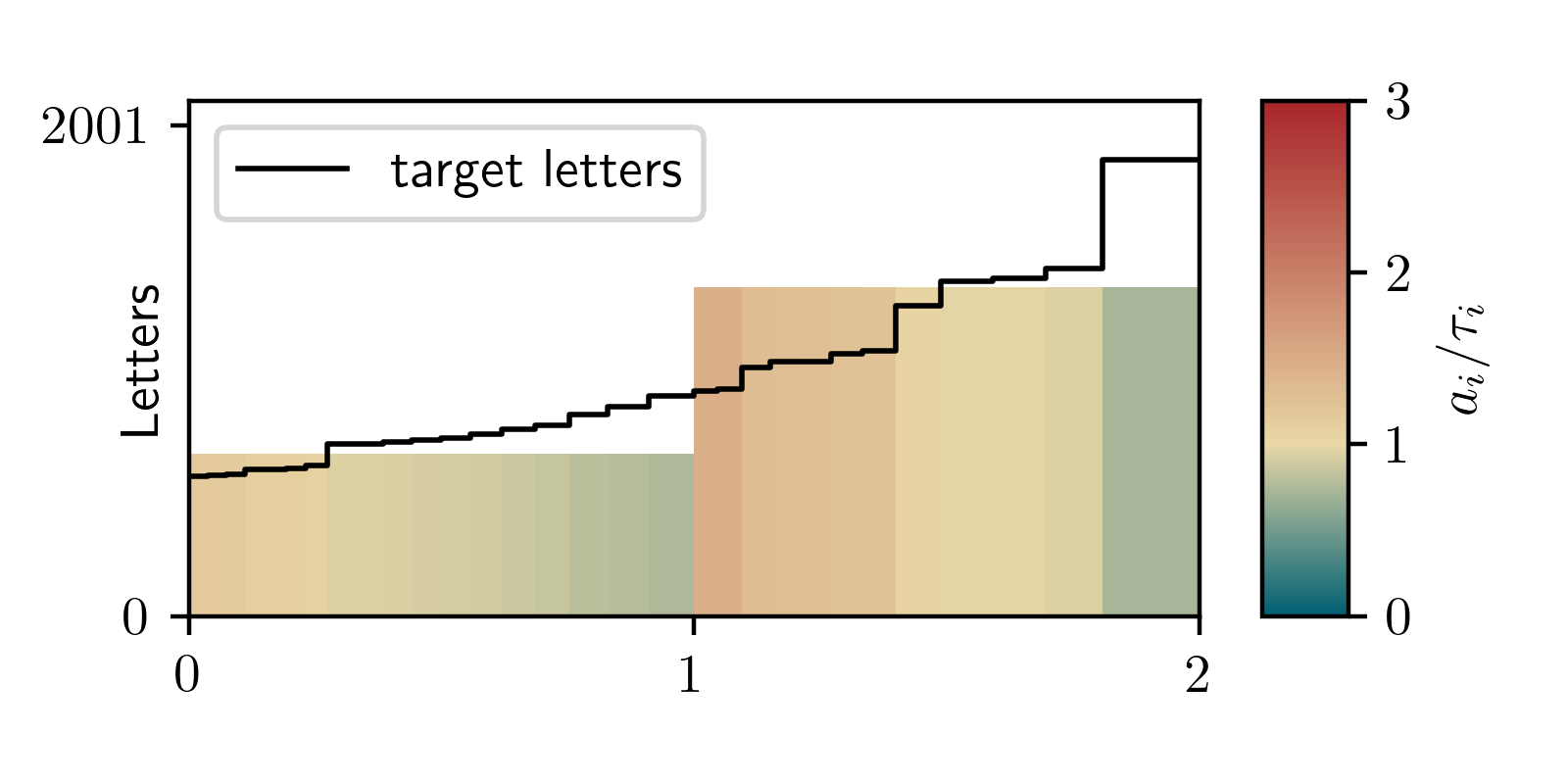}
        \caption{\buckets ($t_G = 2$)}
        \label{fig:results_Nordrhein-Westfalen_Large_greedy_bucket_fill}
    \end{subfigure}
    \caption{Large municipalities of Nordrhein-Westfalen ($\ell_G = 2001$)}
    \label{fig:results_Nordrhein-Westfalen_Large}
\end{figure} 

\begin{figure}
    \centering
    \begin{subfigure}{0.32\textwidth}
        \includegraphics[draft=\draft, width=\linewidth]{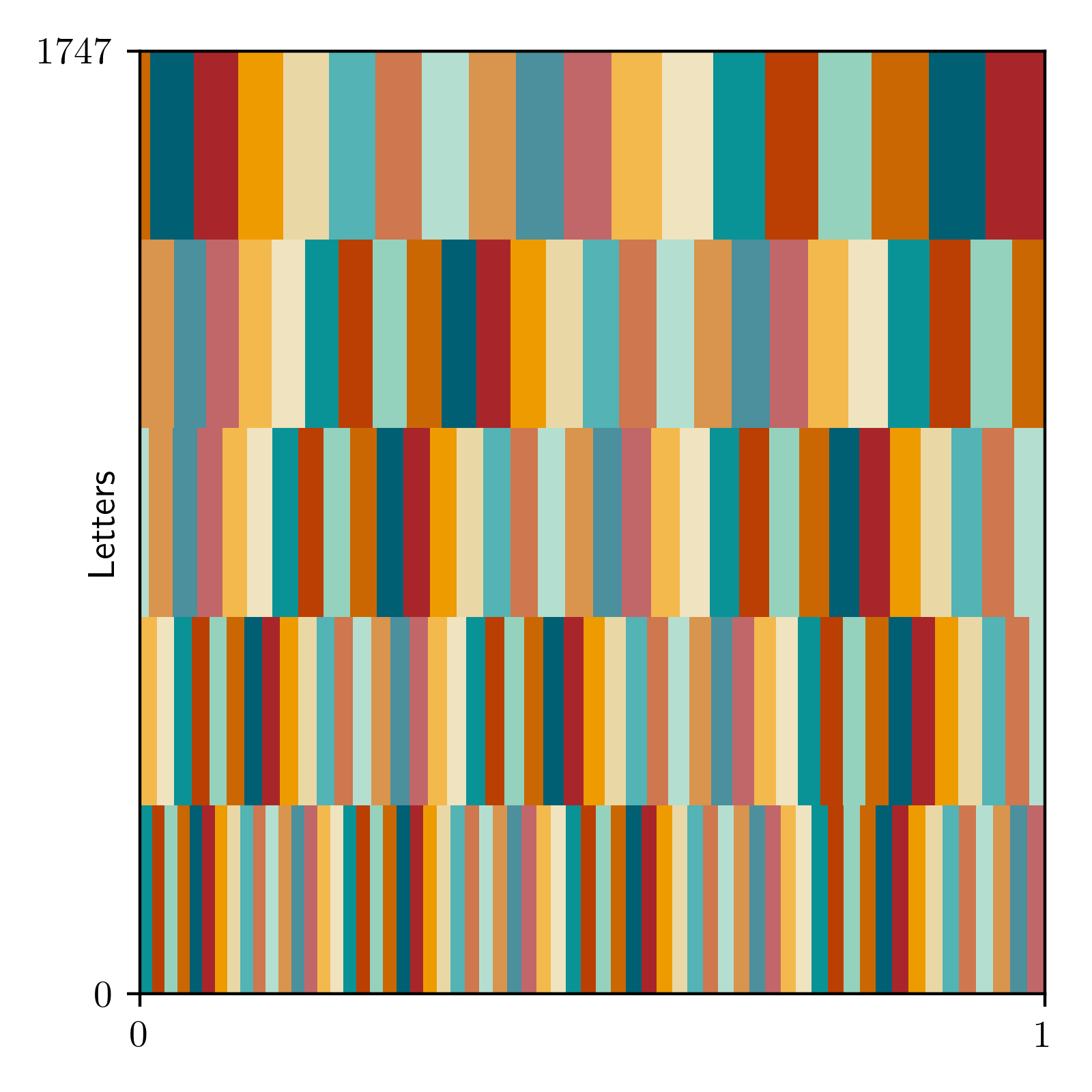}
        \includegraphics[draft=\draft, width=\linewidth]{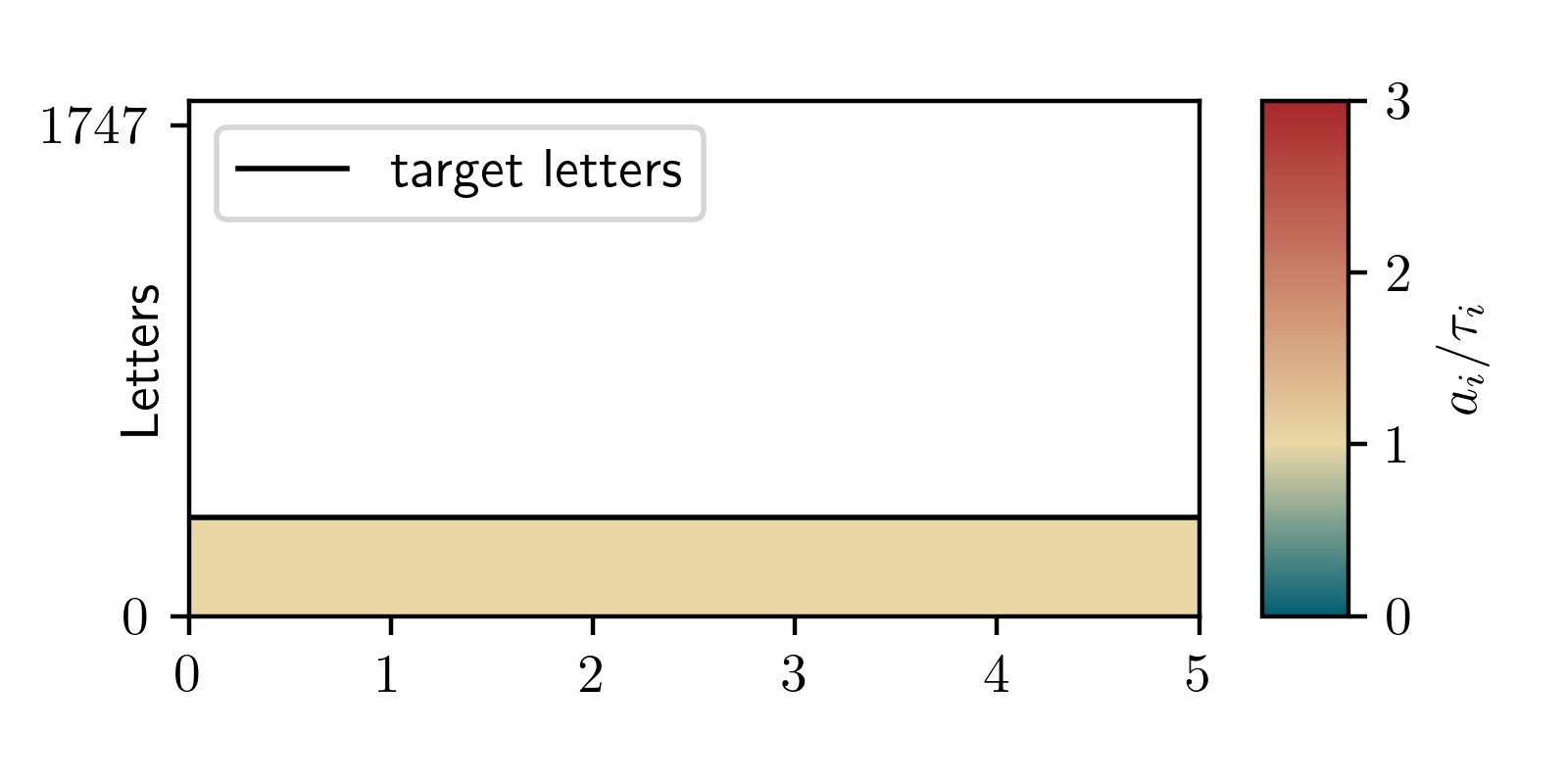}
        \caption{\greq ($t_G = 5$)}
        \label{fig:results_Nordrhein-Westfalen_Medium_greedy_equal}
    \end{subfigure}
    \begin{subfigure}{0.32\textwidth}
        \includegraphics[draft=\draft, width=\linewidth]{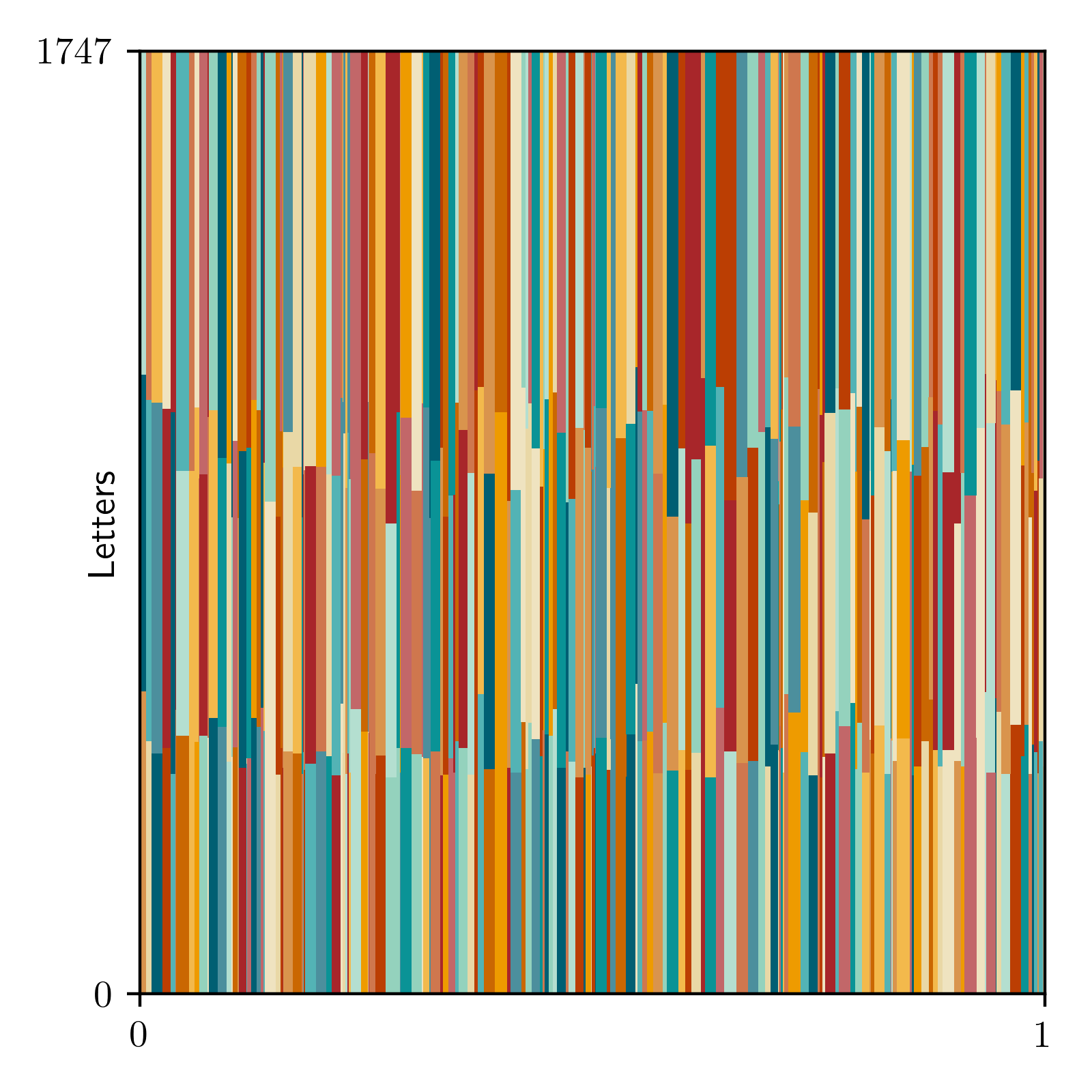}
        \includegraphics[draft=\draft, width=\linewidth]{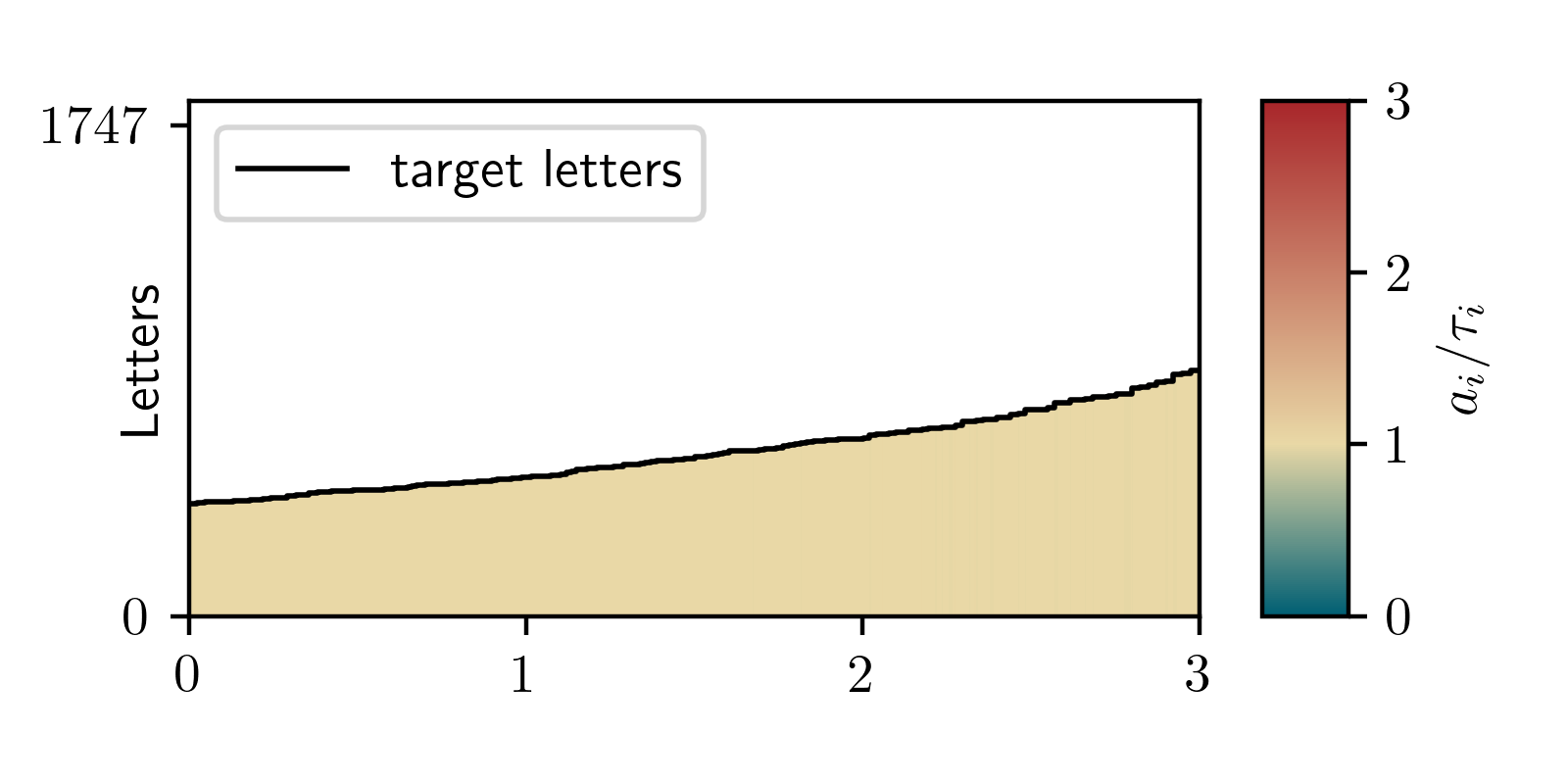}
        \caption{\colgen ($t_G\!=\!3$)}
        \label{fig:results_Nordrhein-Westfalen_Medium_column_generation}
    \end{subfigure}
    \begin{subfigure}{0.32\textwidth}
        \includegraphics[draft=\draft, width=\linewidth]{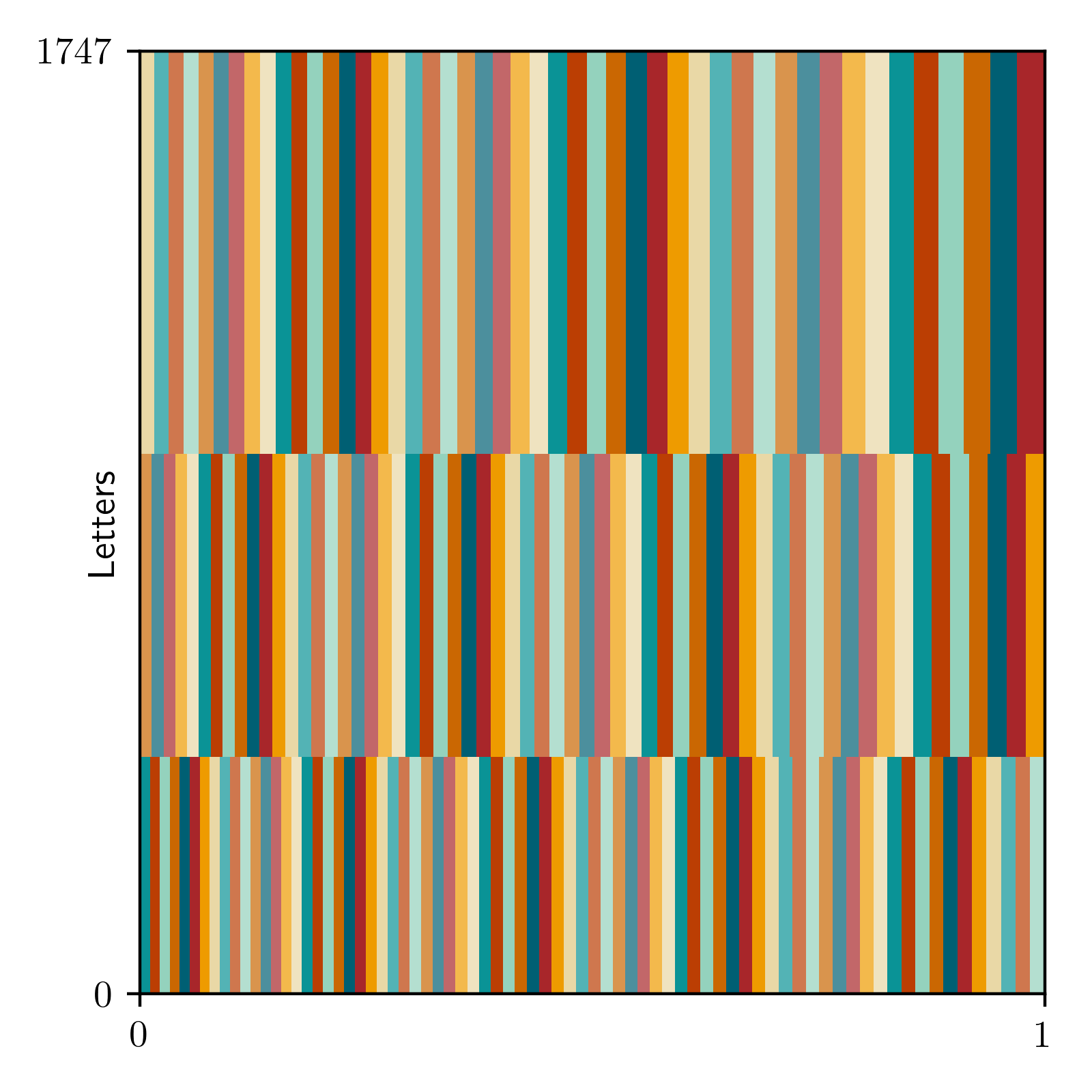}
        \includegraphics[draft=\draft, width=\linewidth]{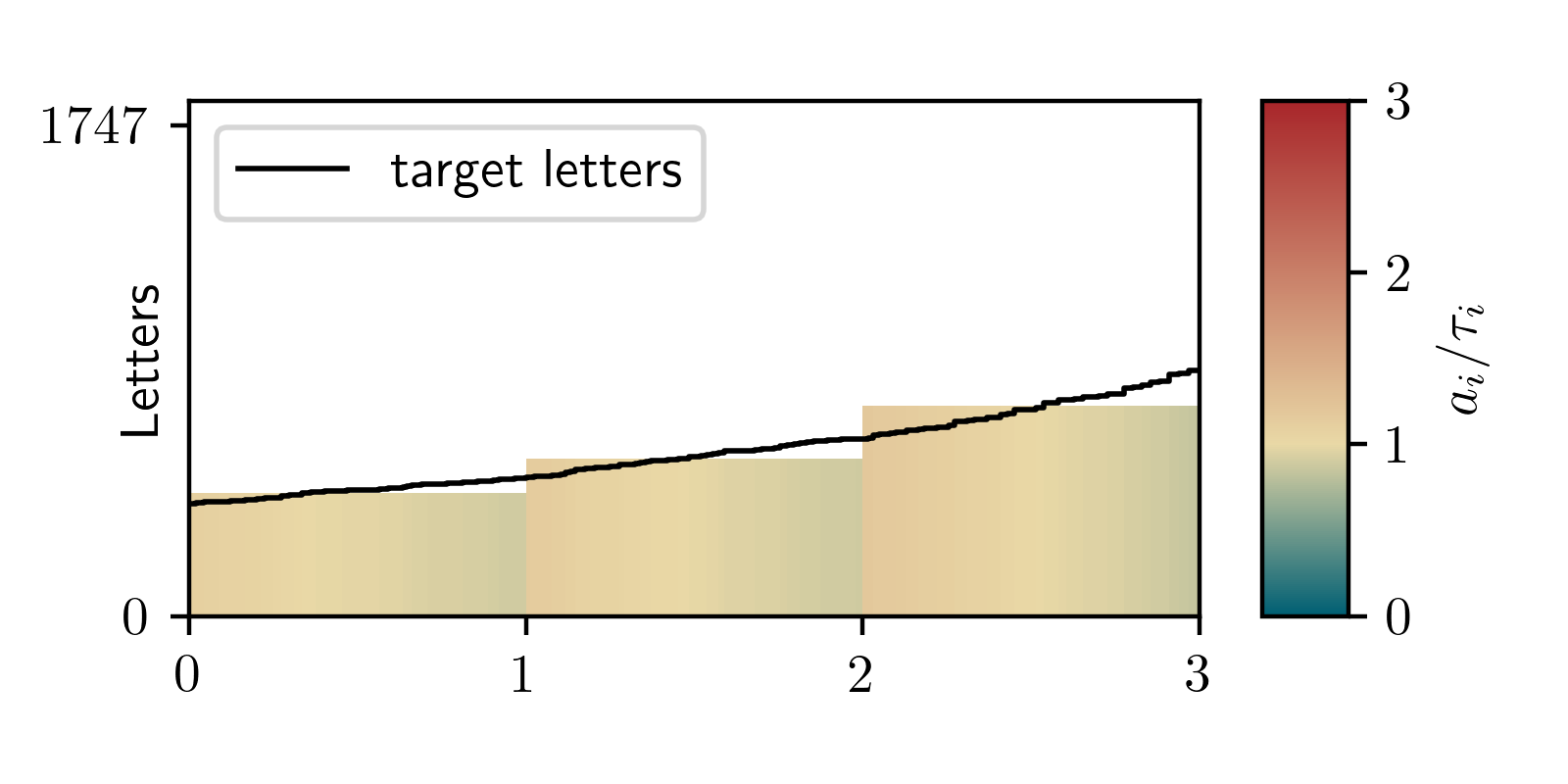}
        \caption{\buckets ($t_G = 3$)}
        \label{fig:results_Nordrhein-Westfalen_Medium_greedy_bucket_fill}
    \end{subfigure}
    \caption{Medium municipalities of Nordrhein-Westfalen ($\ell_G = 1747$)}
    \label{fig:results_Nordrhein-Westfalen_Medium}
\end{figure} 

\begin{figure}
    \centering
    \begin{subfigure}{0.32\textwidth}
        \includegraphics[draft=\draft, width=\linewidth]{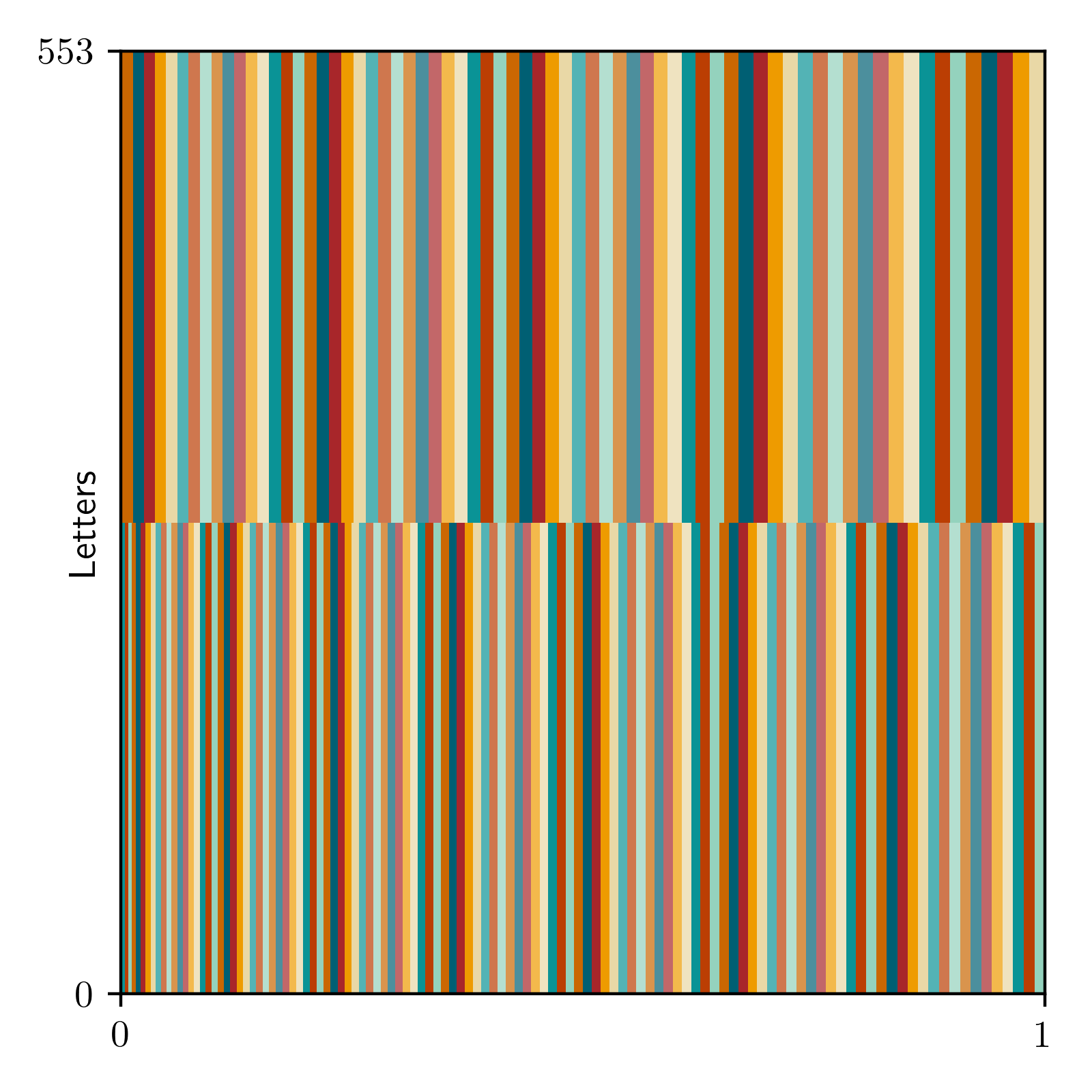}
        \includegraphics[draft=\draft, width=\linewidth]{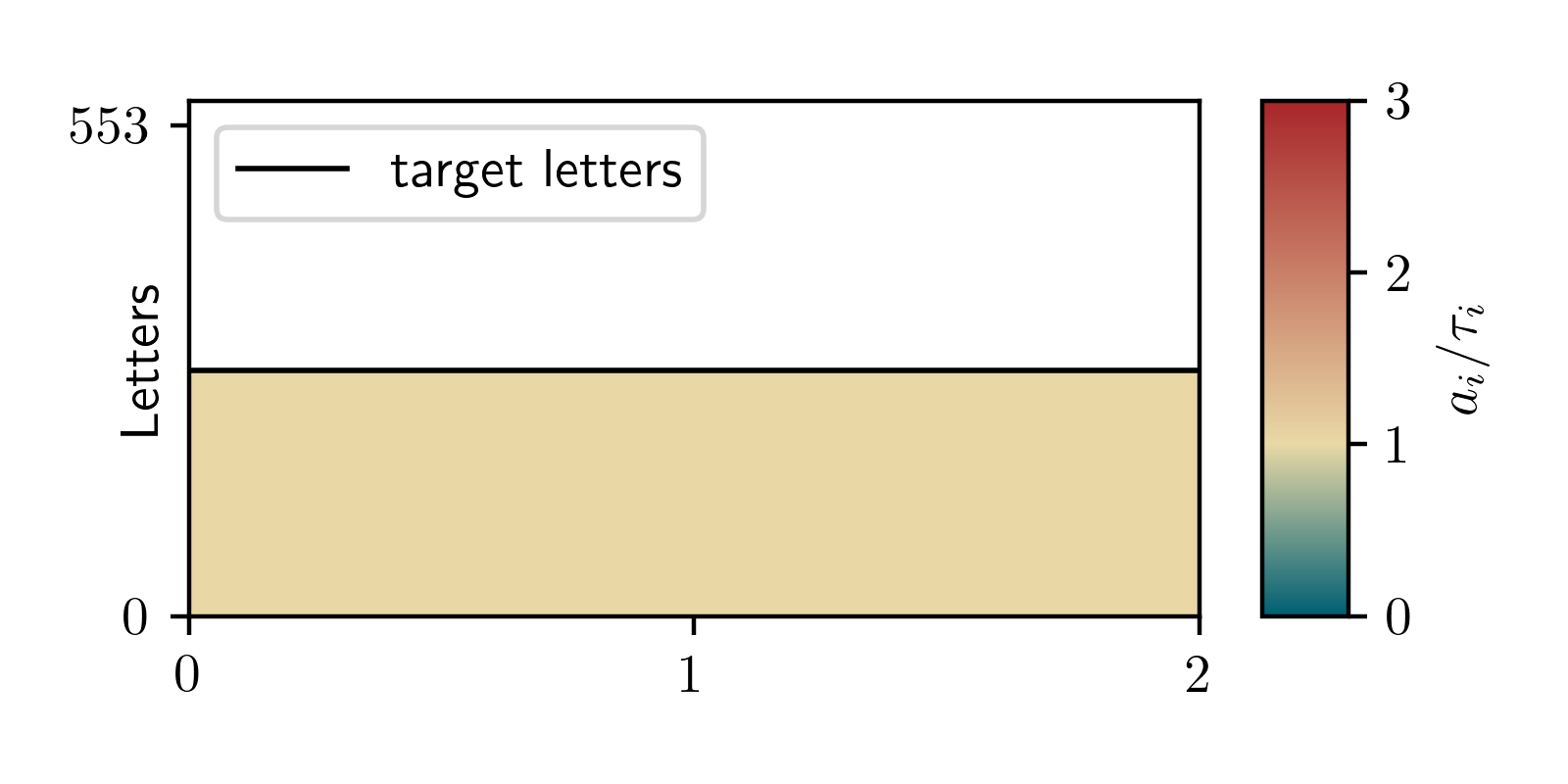}
        \caption{\greq ($t_G = 2$)}
        \label{fig:results_Nordrhein-Westfalen_Small_greedy_equal}
    \end{subfigure}
    \begin{subfigure}{0.32\textwidth}
        \includegraphics[draft=\draft, width=\linewidth]{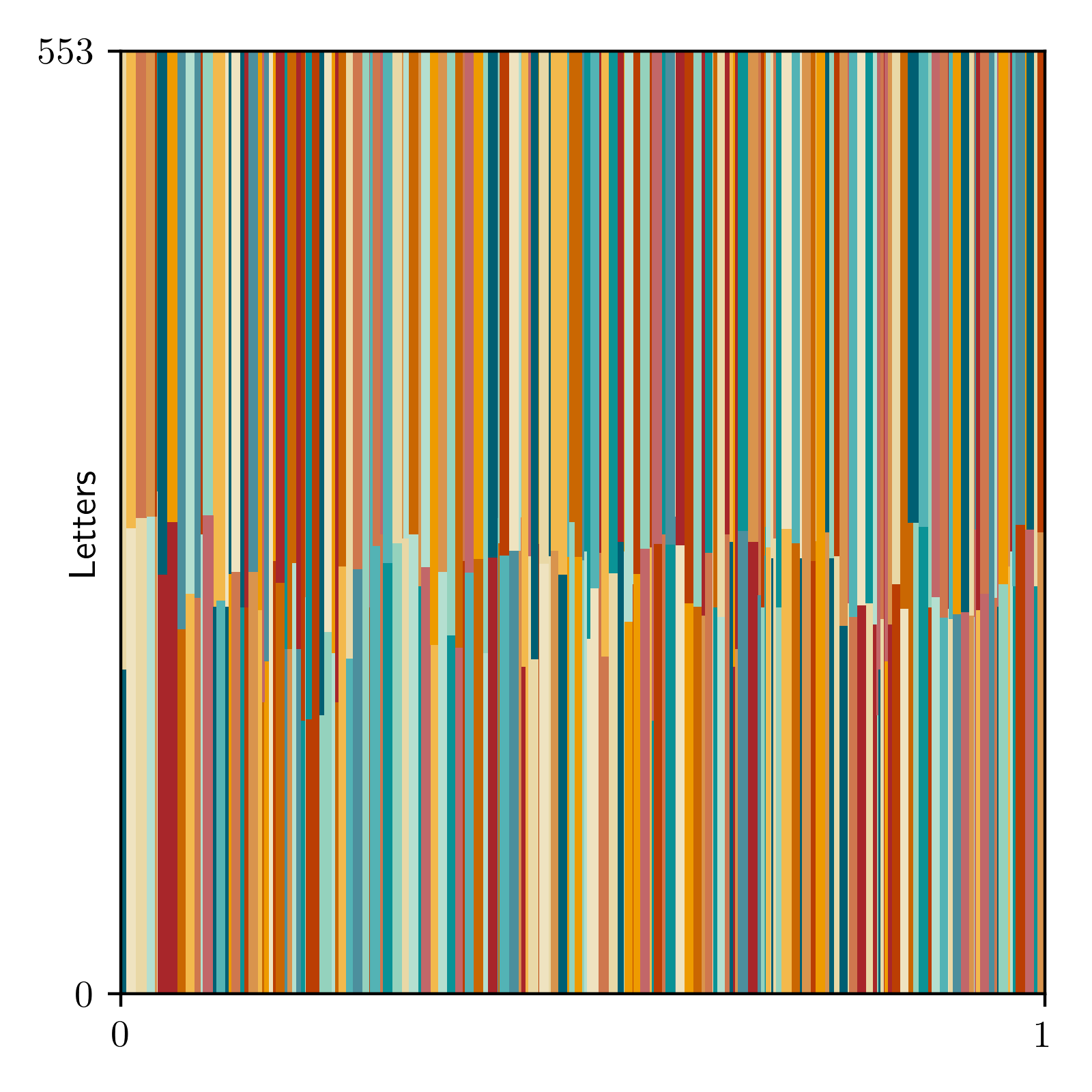}
        \includegraphics[draft=\draft, width=\linewidth]{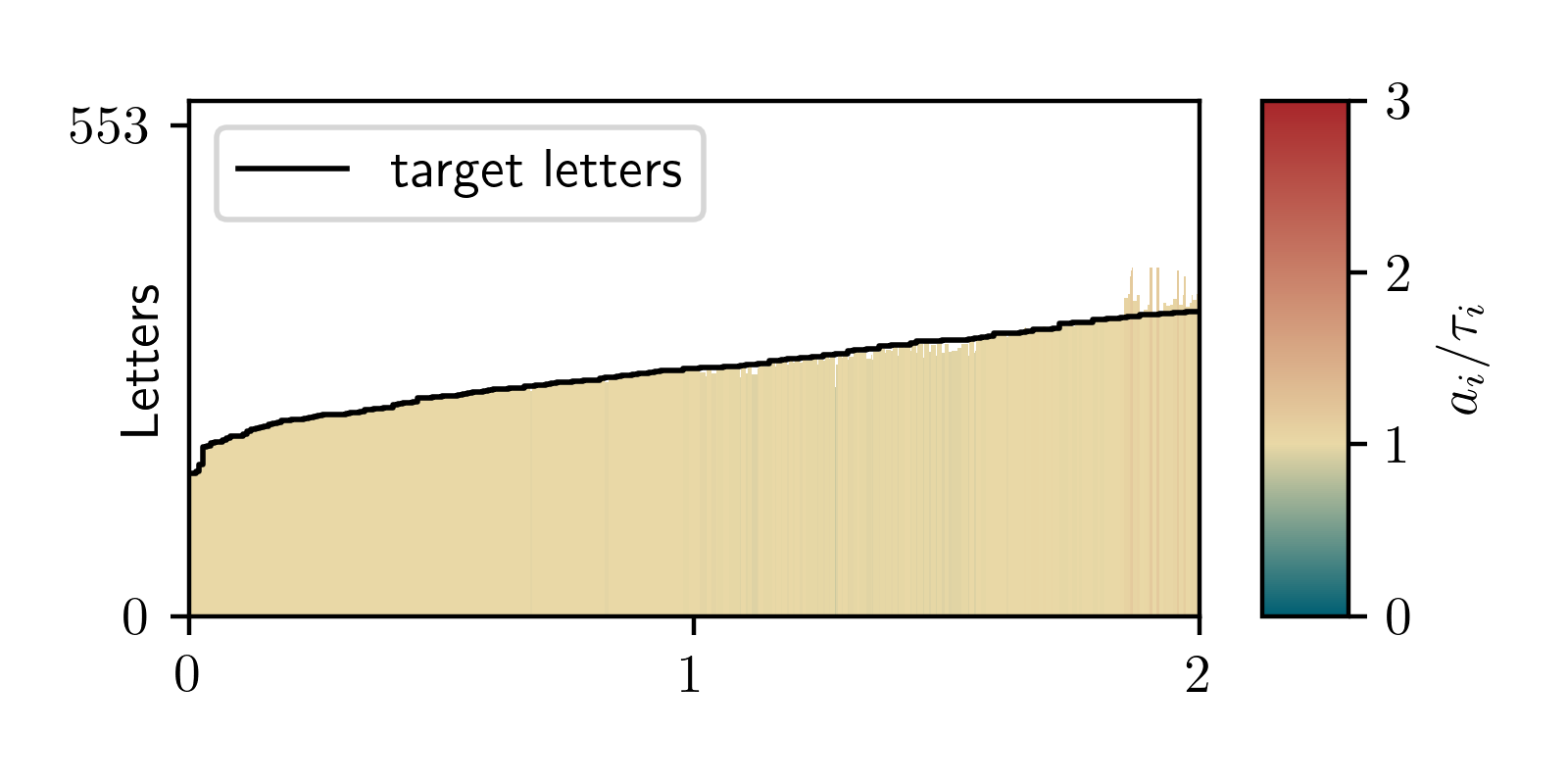}
        \caption{\colgen ($t_G\!=\!2$)}
        \label{fig:results_Nordrhein-Westfalen_Small_column_generation}
    \end{subfigure}
    \begin{subfigure}{0.32\textwidth}
        \includegraphics[draft=\draft, width=\linewidth]{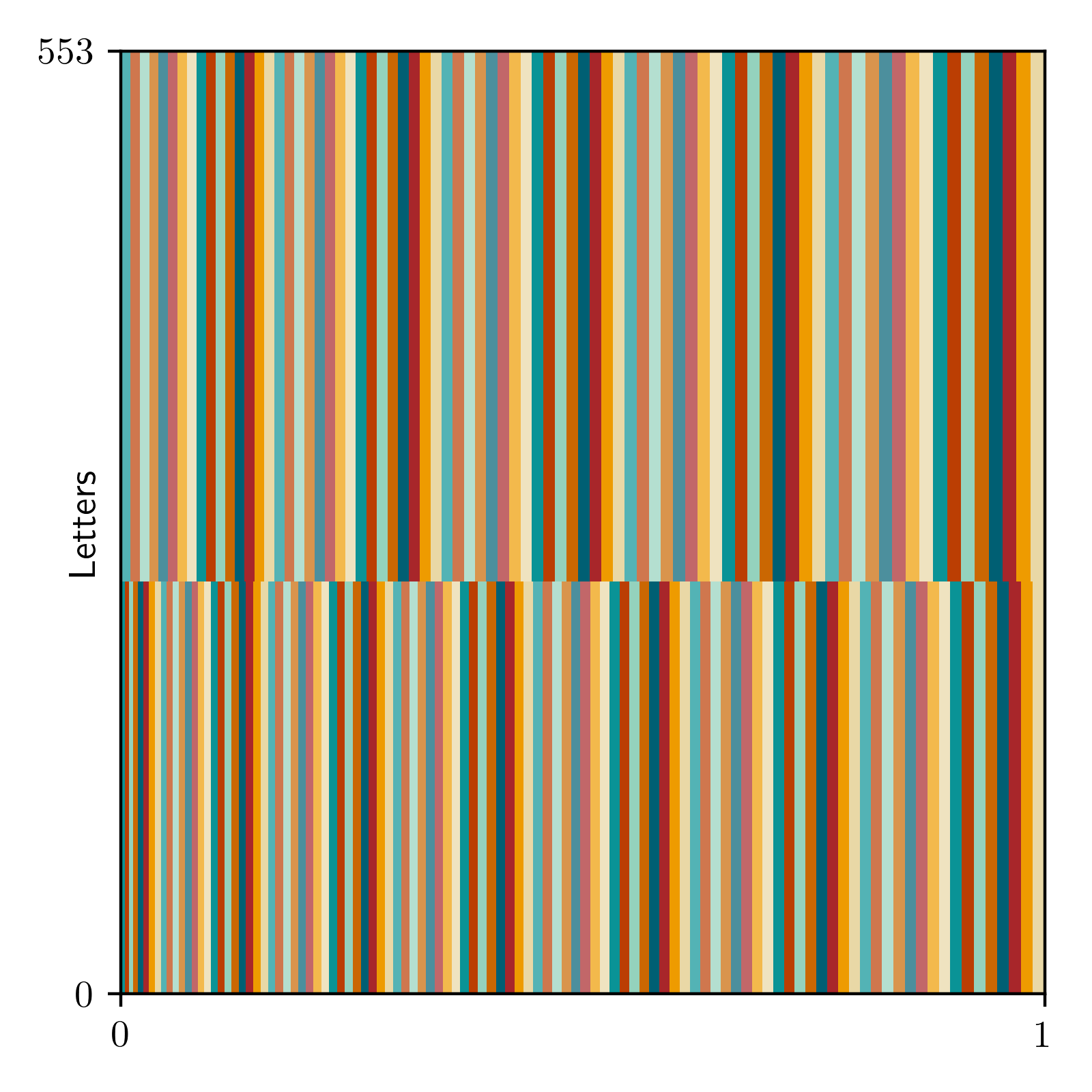}
        \includegraphics[draft=\draft, width=\linewidth]{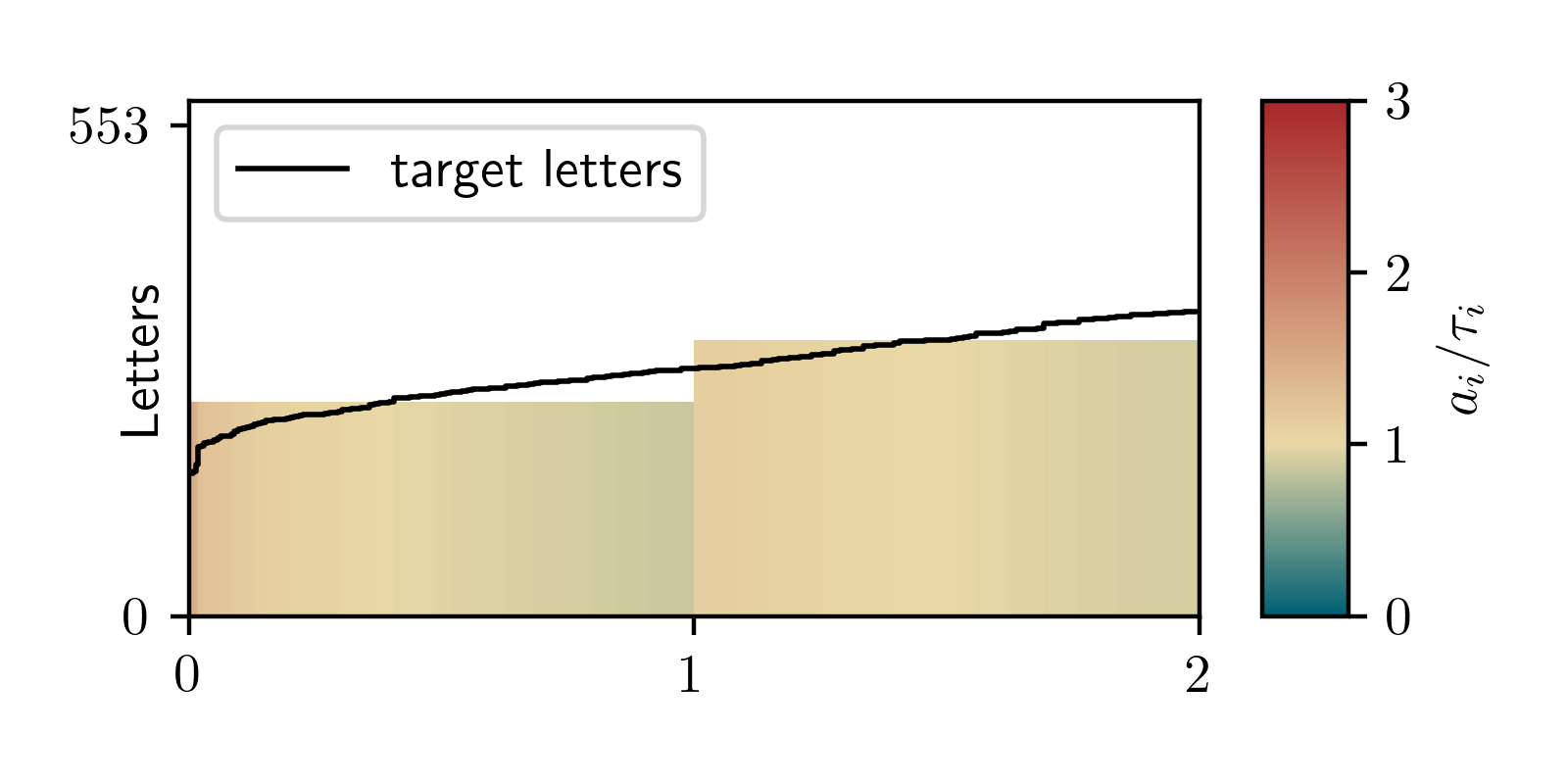}
        \caption{\buckets ($t_G = 2$)}
        \label{fig:results_Nordrhein-Westfalen_Small_greedy_bucket_fill}
    \end{subfigure}
    \caption{Small municipalities of Nordrhein-Westfalen ($\ell_G = 553$)}
    \label{fig:results_Nordrhein-Westfalen_Small}
\end{figure} 

\begin{figure}
    \centering
    \begin{subfigure}{0.32\textwidth}
        \includegraphics[draft=\draft, width=\linewidth]{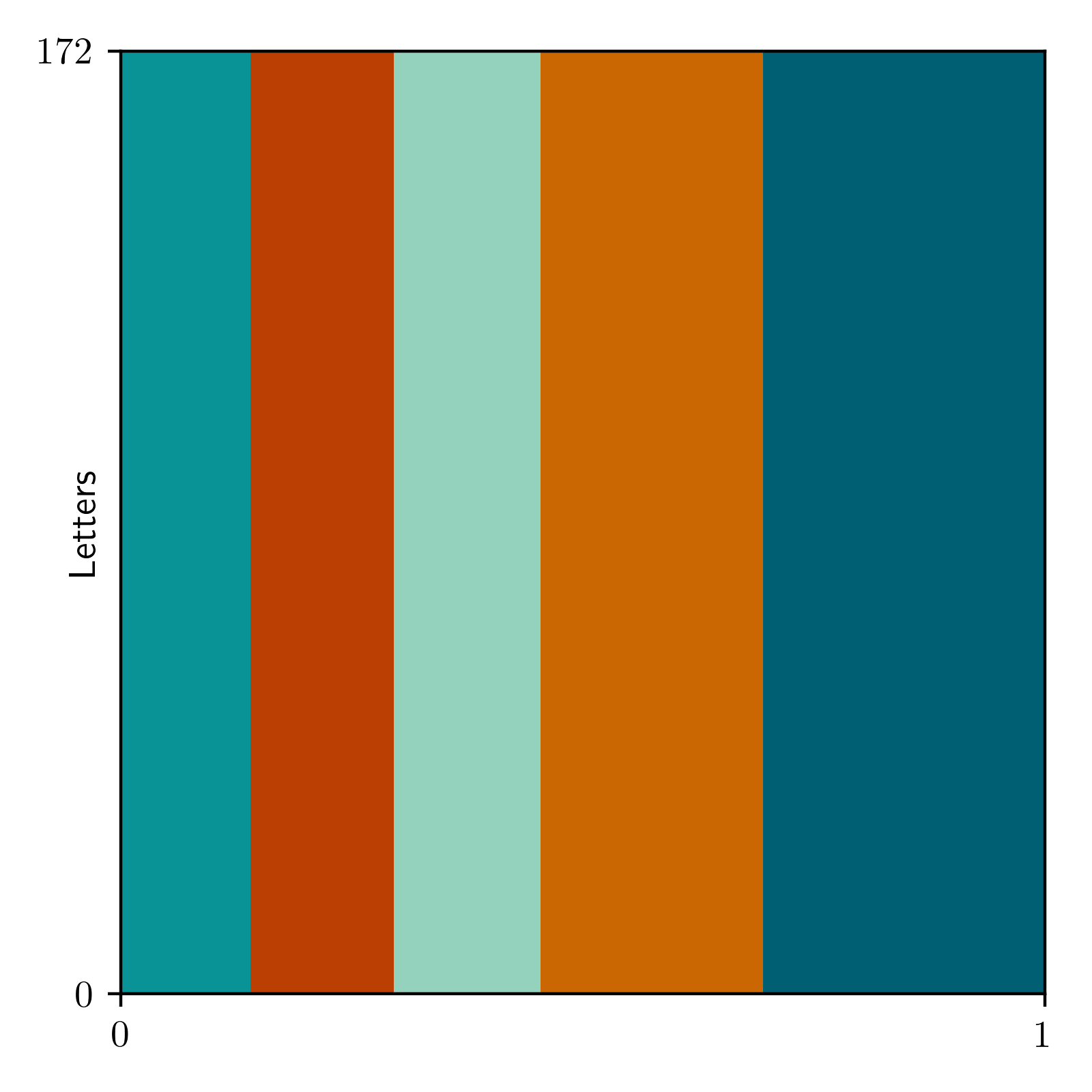}
        \includegraphics[draft=\draft, width=\linewidth]{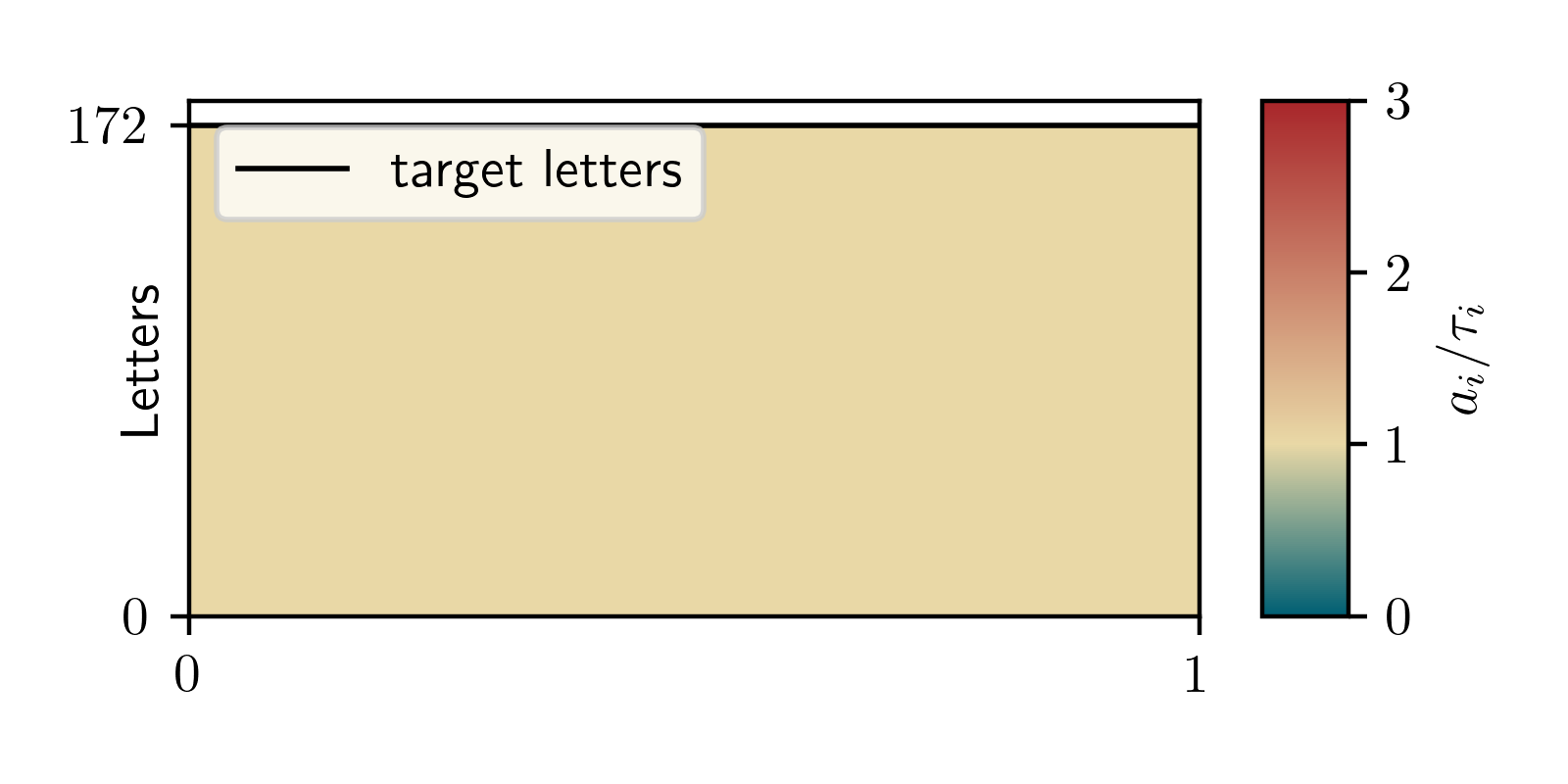}
        \caption{\greq ($t_G = 1$)}
        \label{fig:results_Rheinland-Pfalz_Large_greedy_equal}
    \end{subfigure}
    \begin{subfigure}{0.32\textwidth}
        \includegraphics[draft=\draft, width=\linewidth]{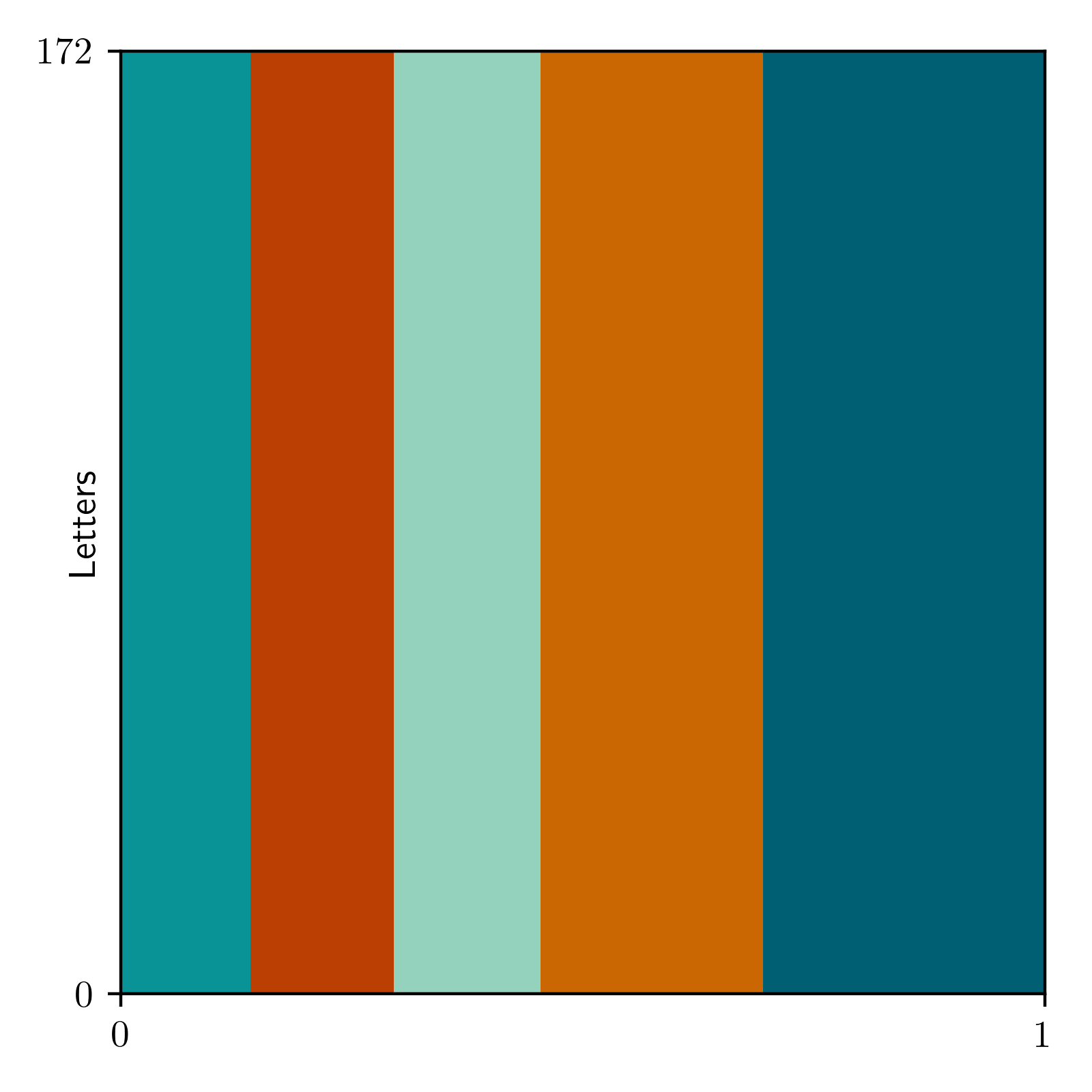}
        \includegraphics[draft=\draft, width=\linewidth]{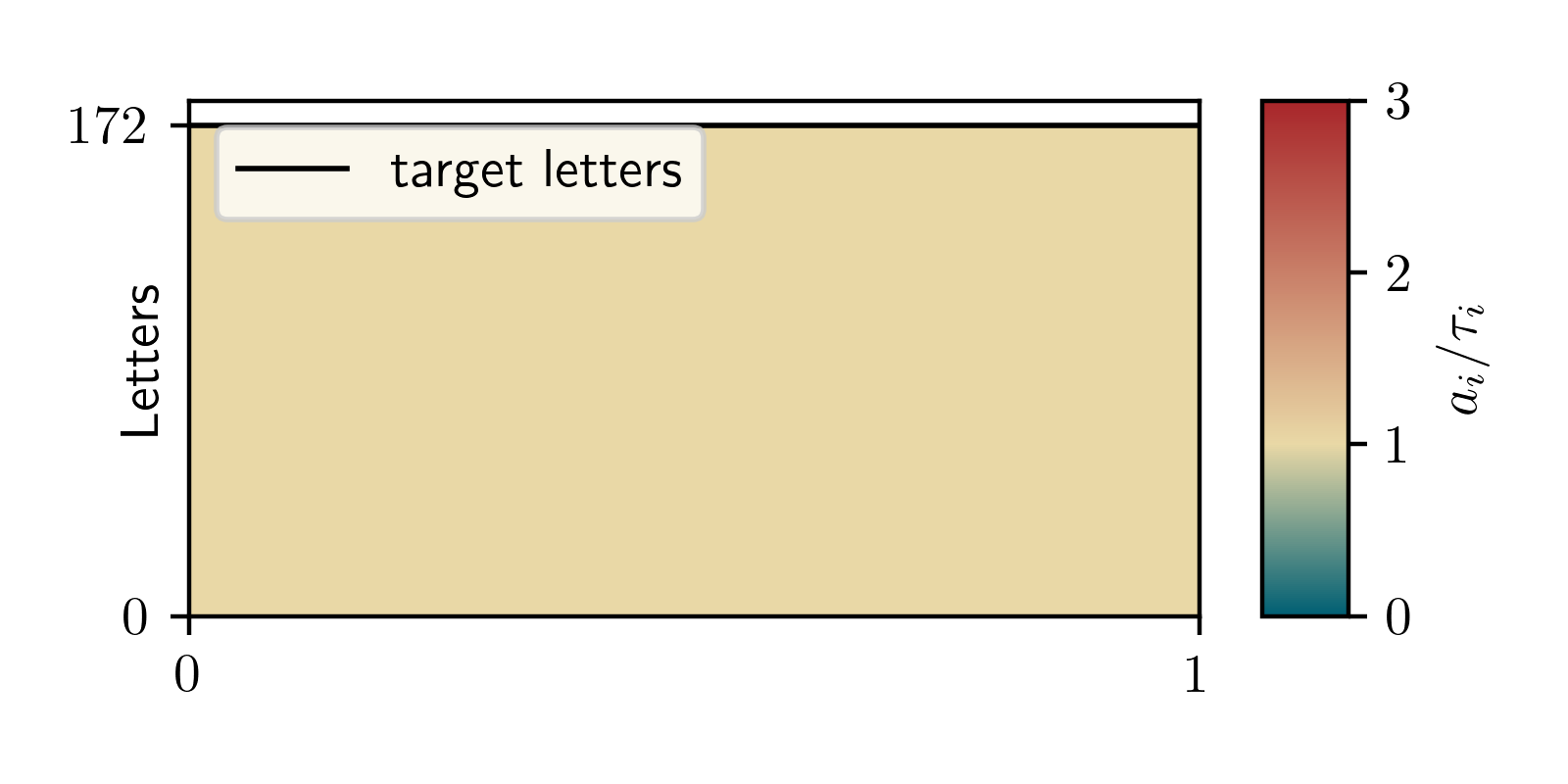}
        \caption{\colgen ($t_G\!=\!1$)}
        \label{fig:results_Rheinland-Pfalz_Large_column_generation}
    \end{subfigure}
    \begin{subfigure}{0.32\textwidth}
        \includegraphics[draft=\draft, width=\linewidth]{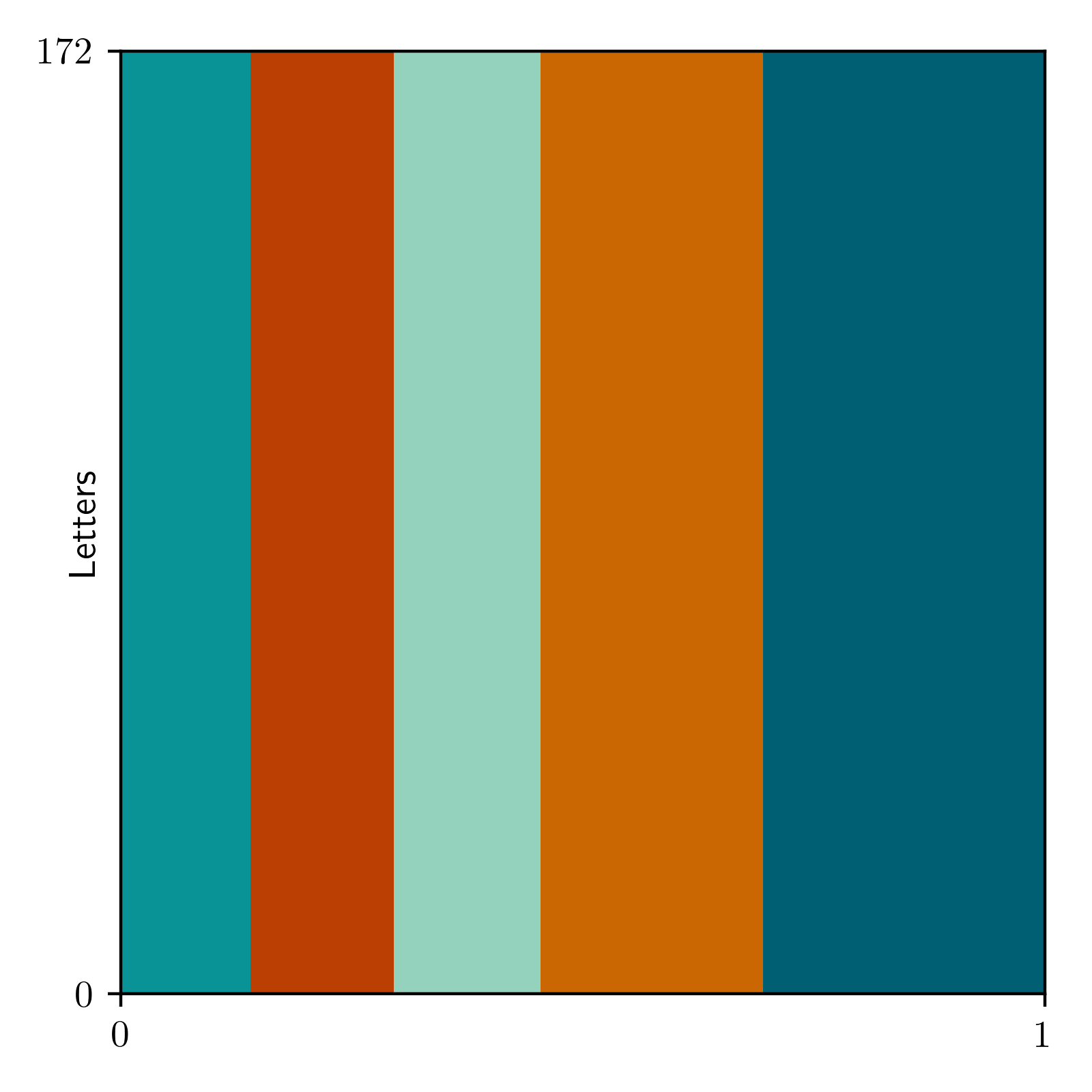}
        \includegraphics[draft=\draft, width=\linewidth]{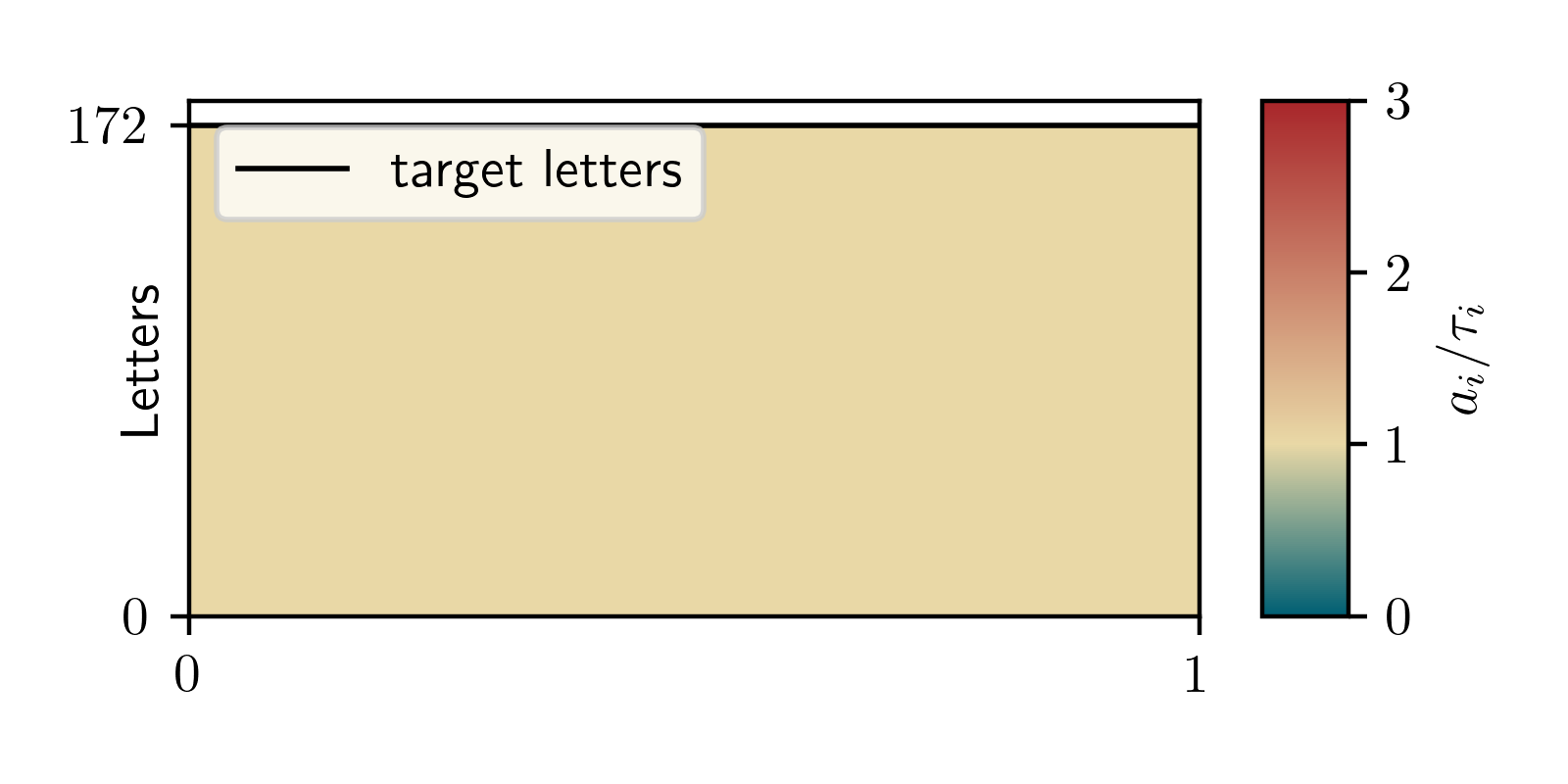}
        \caption{\buckets ($t_G = 1$)}
        \label{fig:results_Rheinland-Pfalz_Large_greedy_bucket_fill}
    \end{subfigure}
    \caption{Large municipalities of Rheinland-Pfalz ($\ell_G = 172$)}
    \label{fig:results_Rheinland-Pfalz_Large}
\end{figure} 

\begin{figure}
    \centering
    \begin{subfigure}{0.32\textwidth}
        \includegraphics[draft=\draft, width=\linewidth]{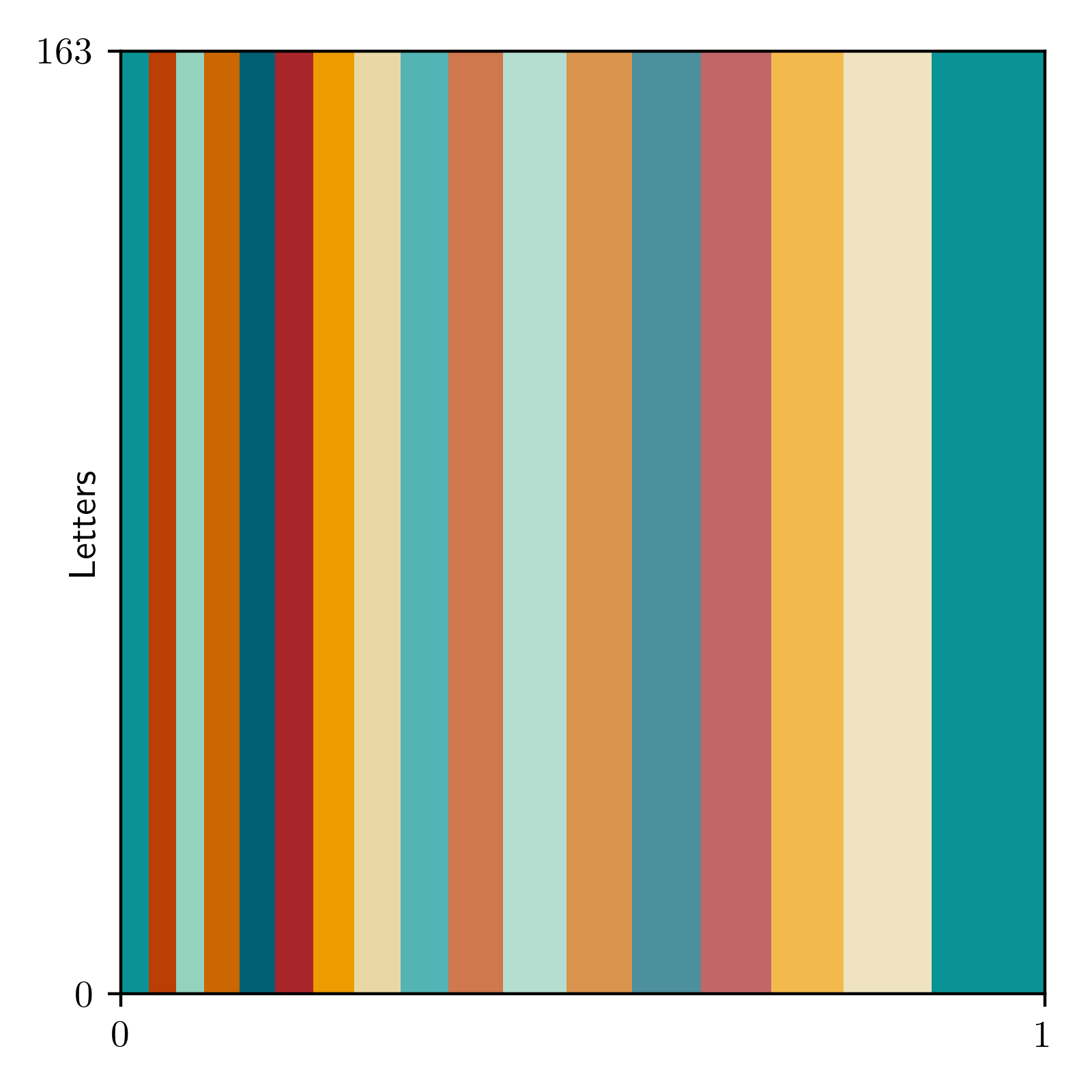}
        \includegraphics[draft=\draft, width=\linewidth]{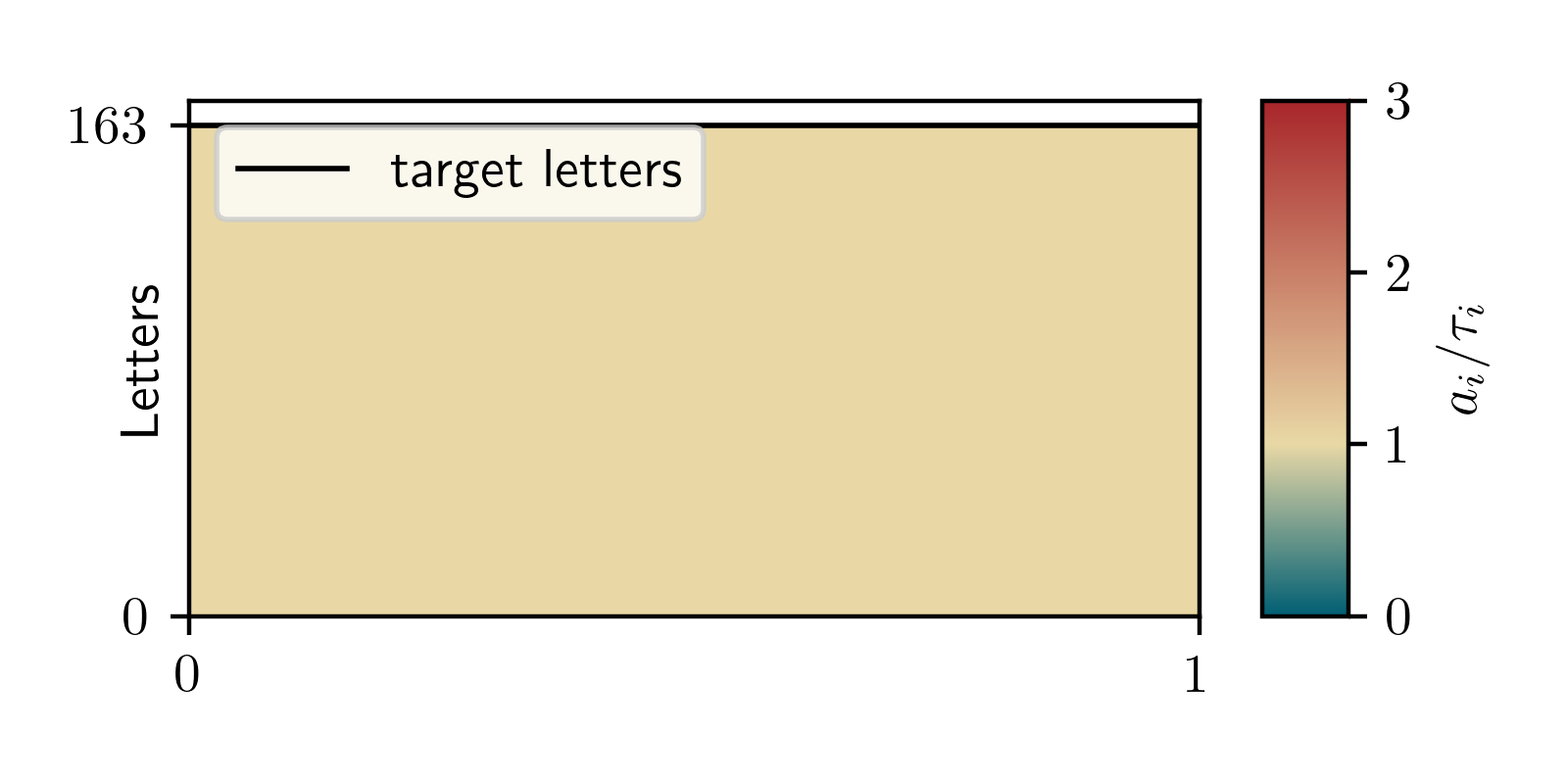}
        \caption{\greq ($t_G = 1$)}
        \label{fig:results_Rheinland-Pfalz_Medium_greedy_equal}
    \end{subfigure}
    \begin{subfigure}{0.32\textwidth}
        \includegraphics[draft=\draft, width=\linewidth]{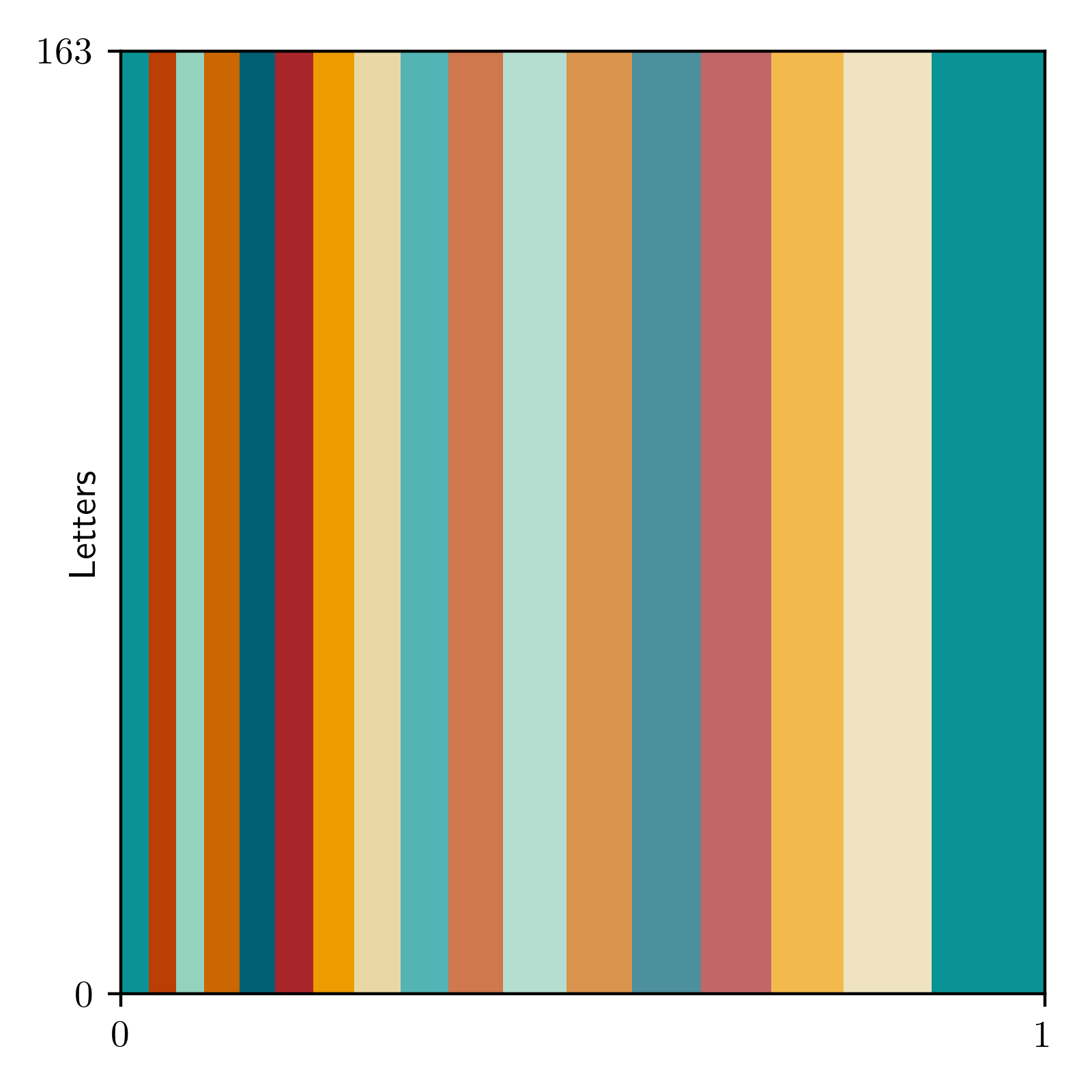}
        \includegraphics[draft=\draft, width=\linewidth]{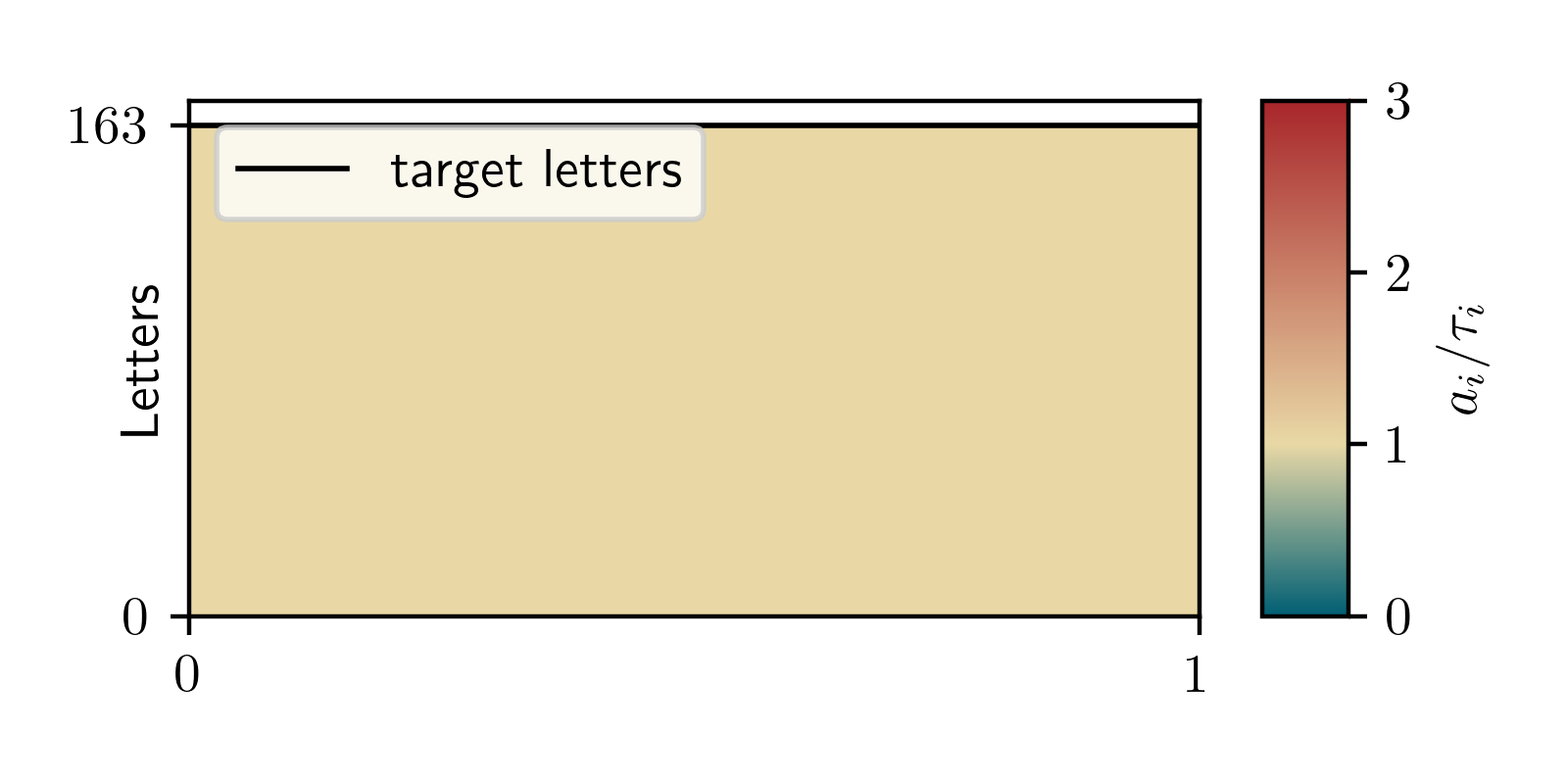}
        \caption{\colgen ($t_G\!=\!1$)}
        \label{fig:results_Rheinland-Pfalz_Medium_column_generation}
    \end{subfigure}
    \begin{subfigure}{0.32\textwidth}
        \includegraphics[draft=\draft, width=\linewidth]{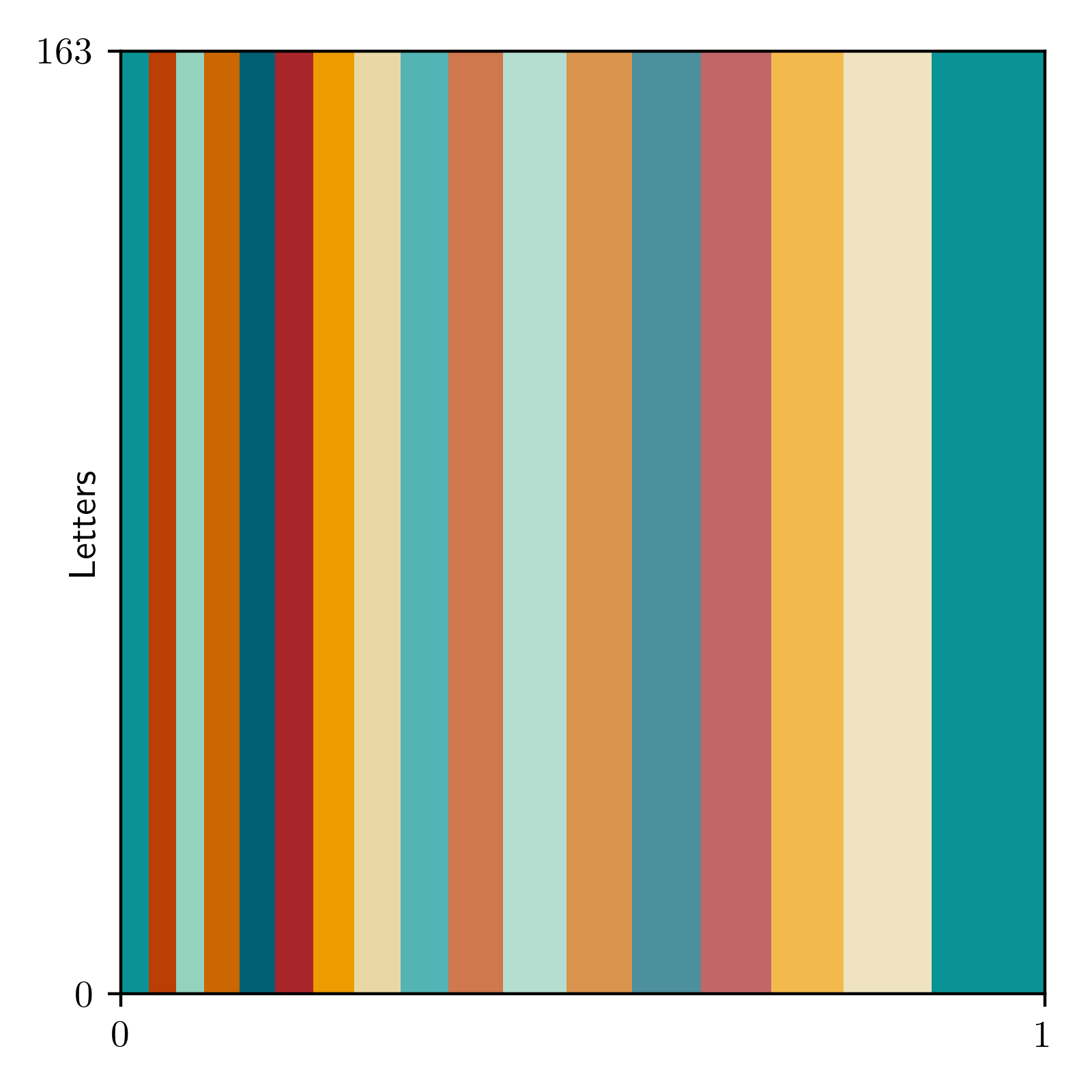}
        \includegraphics[draft=\draft, width=\linewidth]{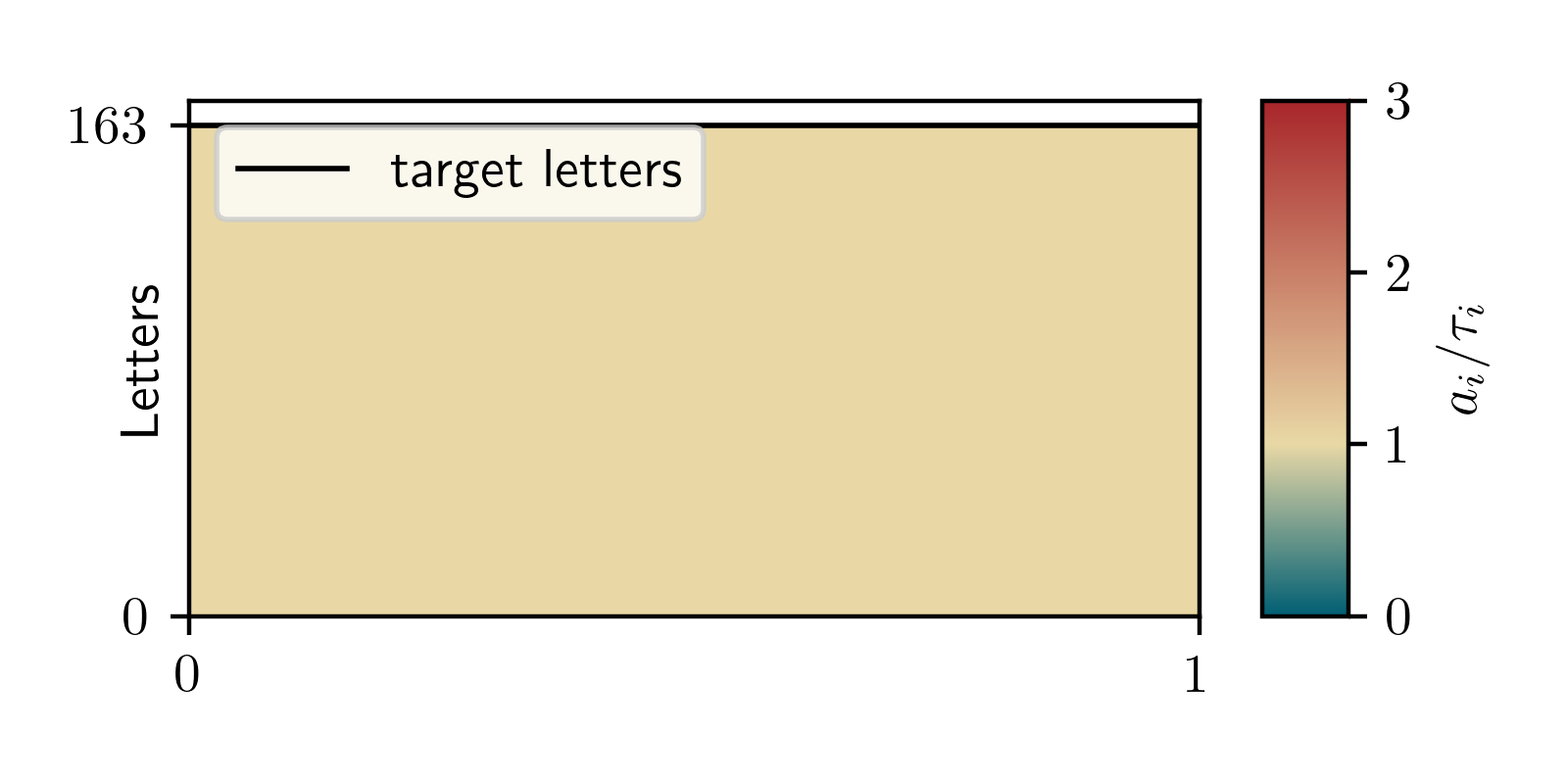}
        \caption{\buckets ($t_G = 1$)}
        \label{fig:results_Rheinland-Pfalz_Medium_greedy_bucket_fill}
    \end{subfigure}
    \caption{Medium municipalities of Rheinland-Pfalz ($\ell_G = 163$)}
    \label{fig:results_Rheinland-Pfalz_Medium}
\end{figure} 

\begin{figure}
    \centering
    \begin{subfigure}{0.32\textwidth}
        \includegraphics[draft=\draft, width=\linewidth]{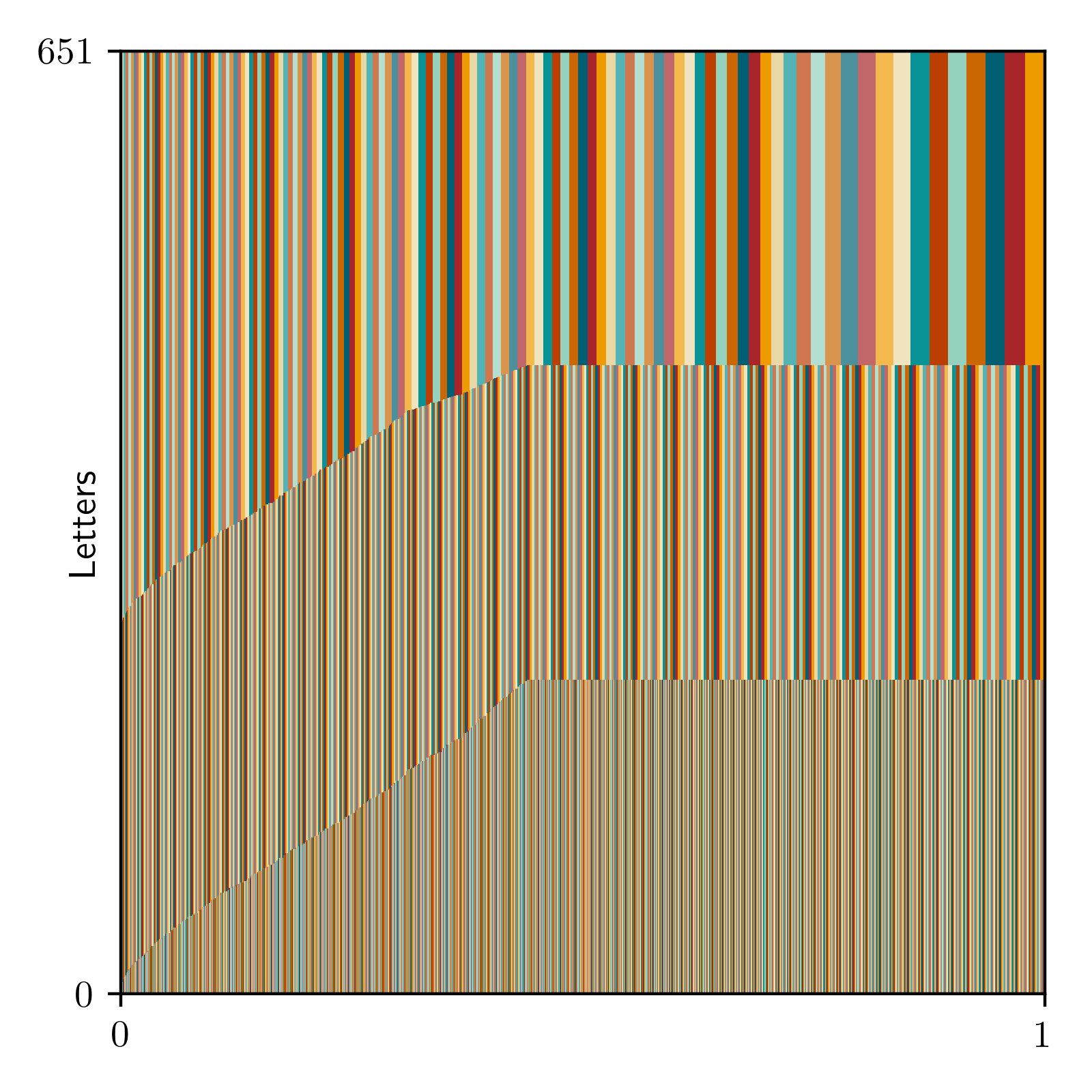}
        \includegraphics[draft=\draft, width=\linewidth]{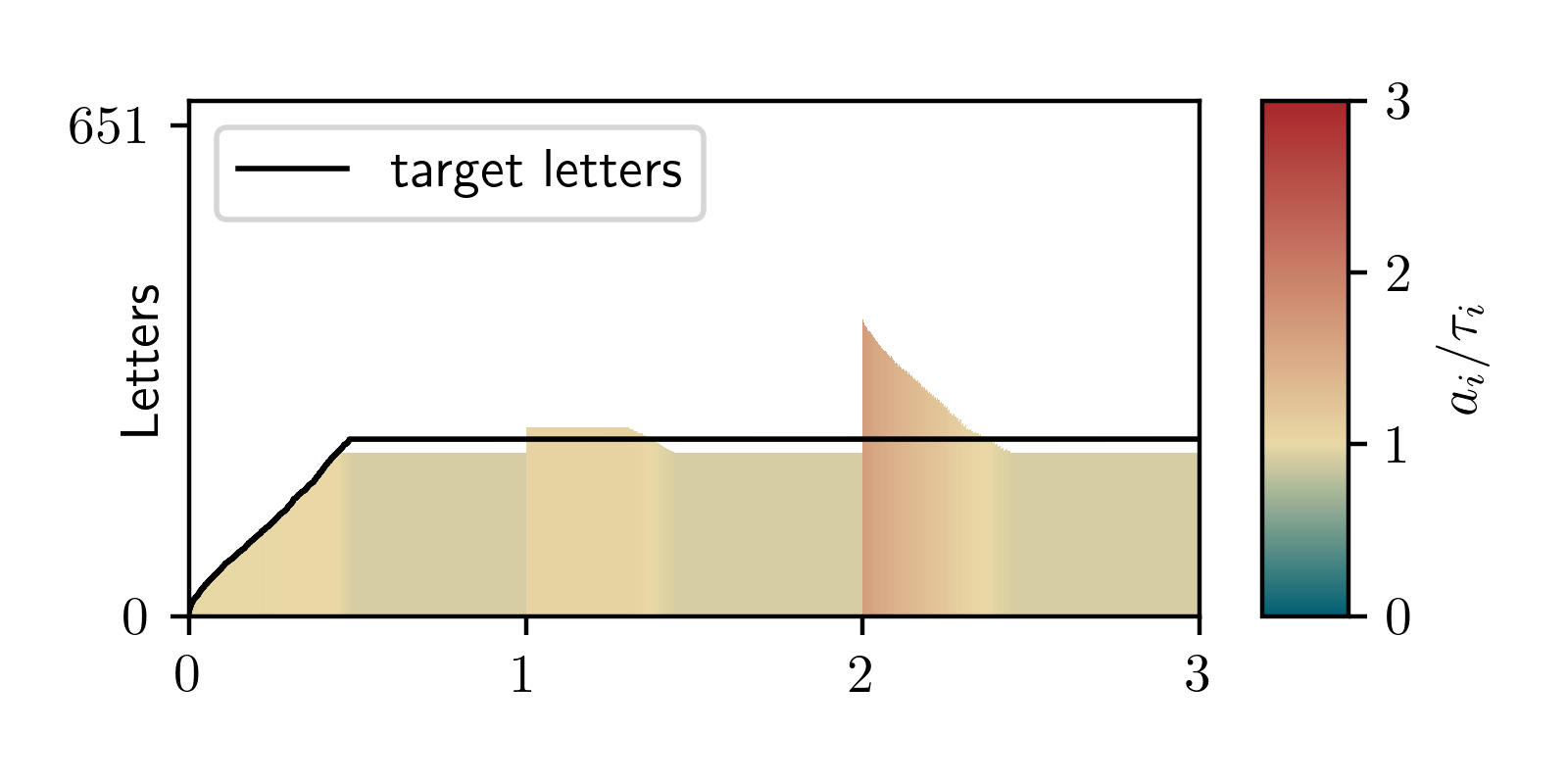}
        \caption{\greq ($t_G = 3$)}
        \label{fig:results_Rheinland-Pfalz_Small_greedy_equal}
    \end{subfigure}
    \begin{subfigure}{0.32\textwidth}
        \includegraphics[draft=\draft, width=\linewidth]{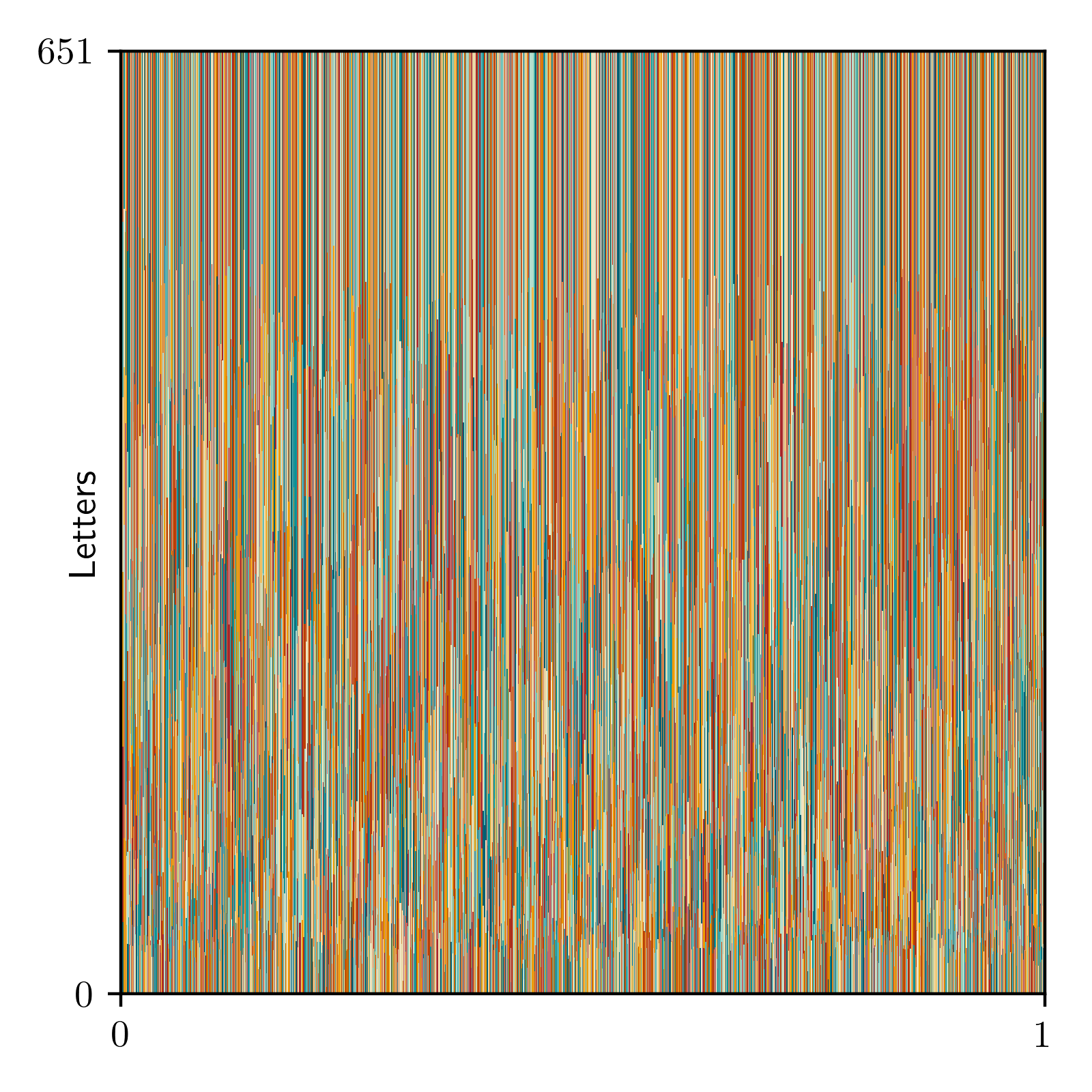}
        \includegraphics[draft=\draft, width=\linewidth]{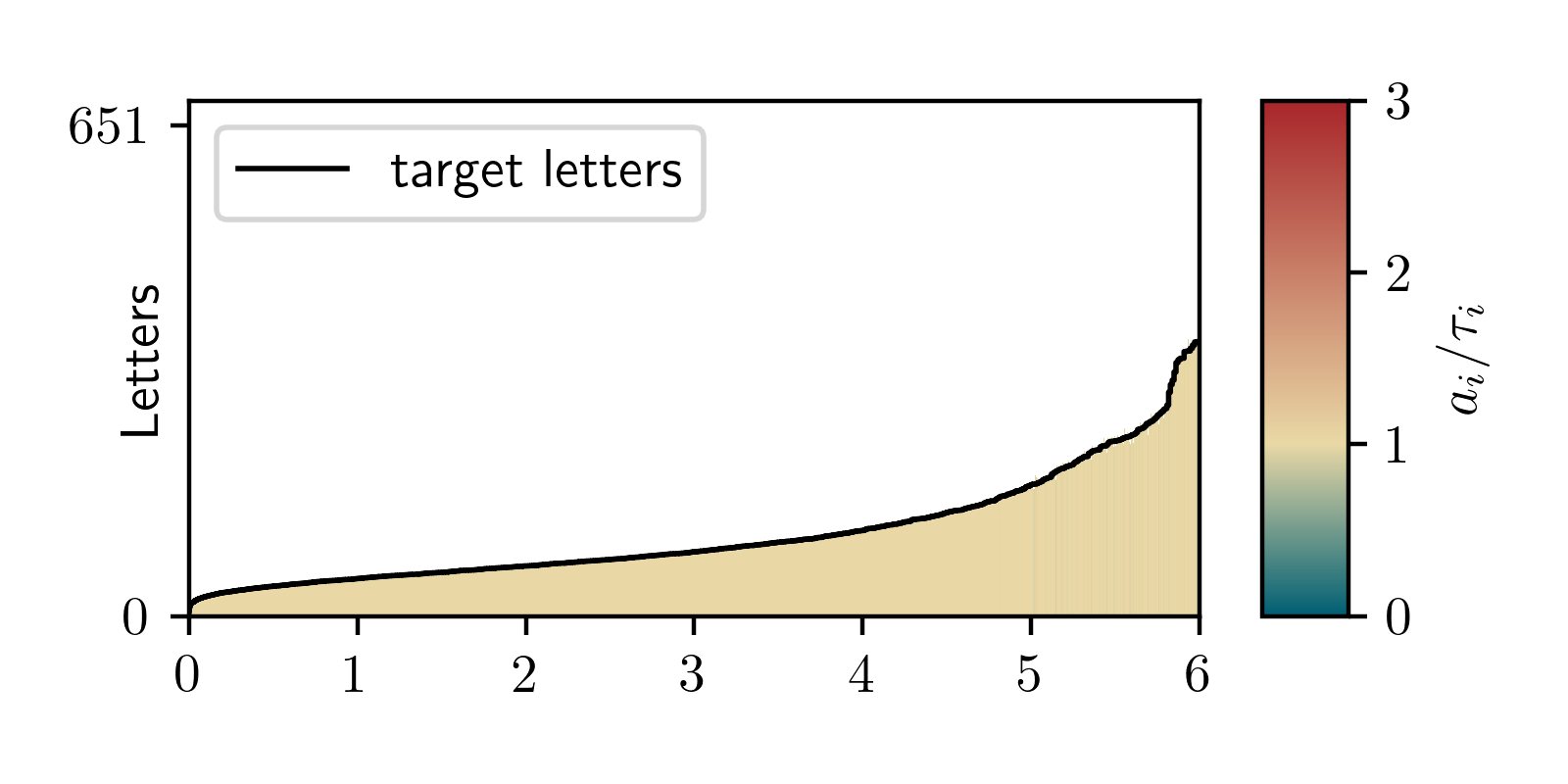}
        \caption{\colgen ($t_G\!=\!6$)}
        \label{fig:results_Rheinland-Pfalz_Small_column_generation}
    \end{subfigure}
    \begin{subfigure}{0.32\textwidth}
        \includegraphics[draft=\draft, width=\linewidth]{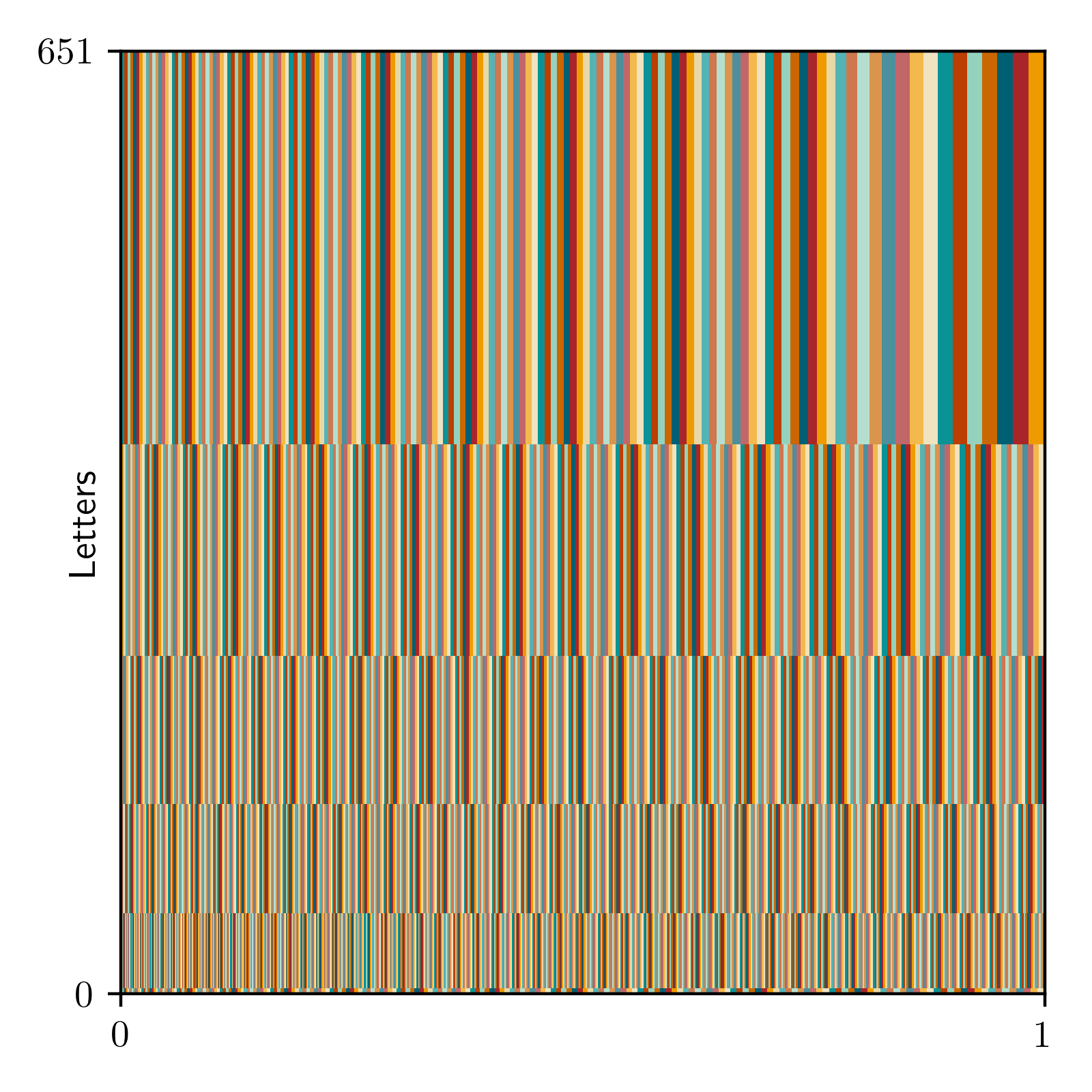}
        \includegraphics[draft=\draft, width=\linewidth]{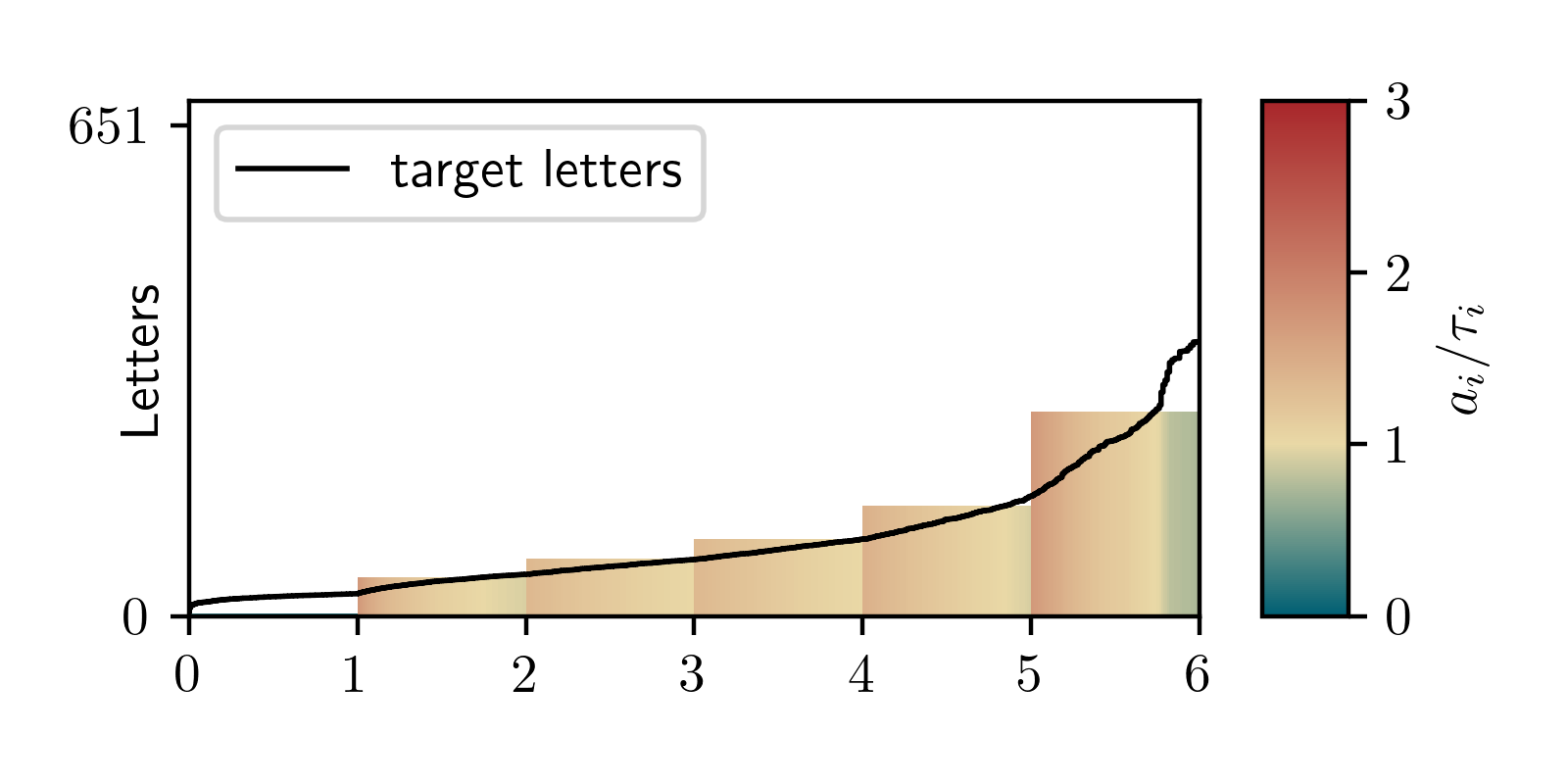}
        \caption{\buckets ($t_G = 6$)}
        \label{fig:results_Rheinland-Pfalz_Small_greedy_bucket_fill}
    \end{subfigure}
    \caption{Small municipalities of Rheinland-Pfalz ($\ell_G = 651$)}
    \label{fig:results_Rheinland-Pfalz_Small}
\end{figure} 

\begin{figure}
    \centering
    \begin{subfigure}{0.32\textwidth}
        \includegraphics[draft=\draft, width=\linewidth]{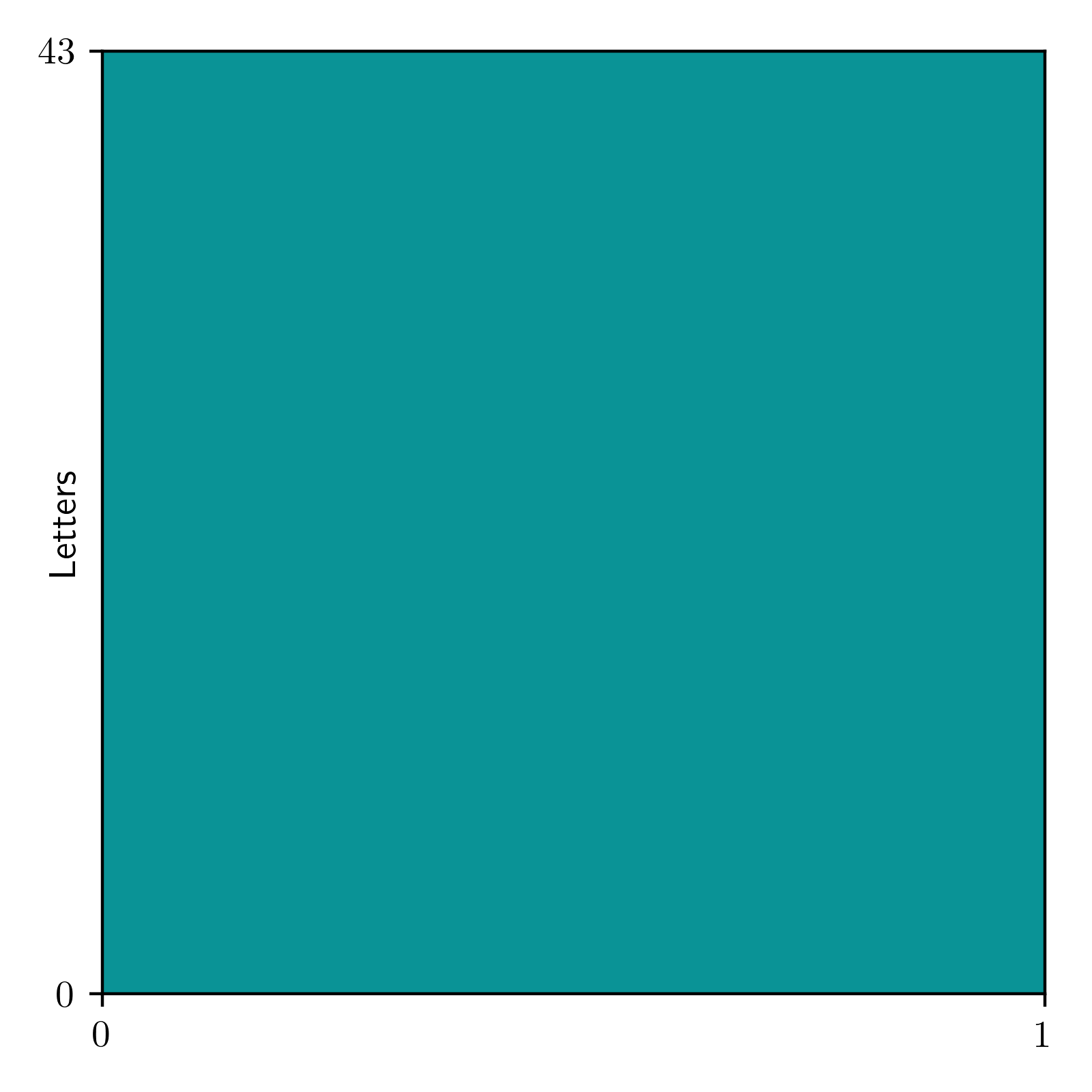}
        \includegraphics[draft=\draft, width=\linewidth]{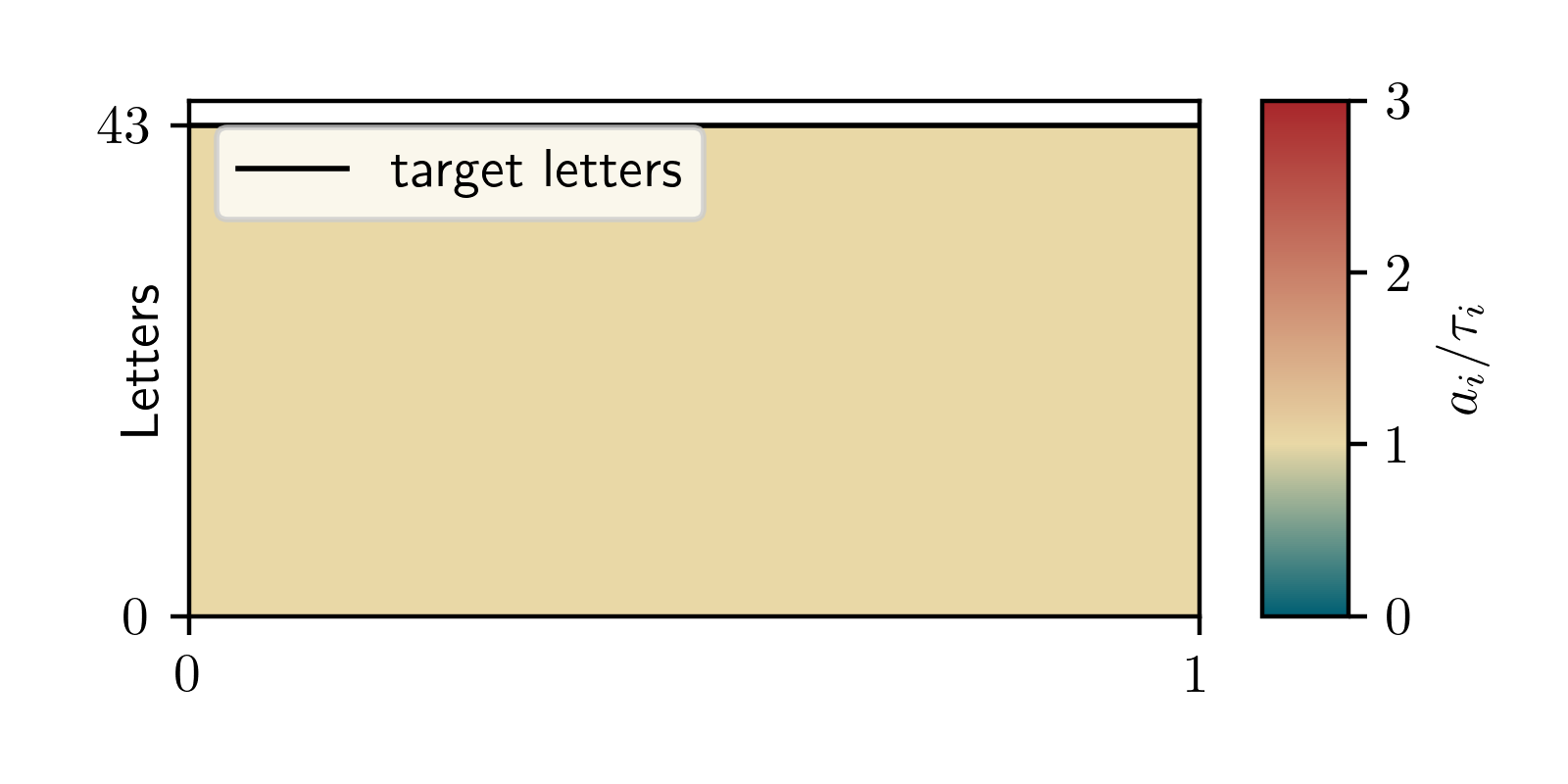}
        \caption{\greq ($t_G = 1$)}
        \label{fig:results_Saarland_Large_greedy_equal}
    \end{subfigure}
    \begin{subfigure}{0.32\textwidth}
        \includegraphics[draft=\draft, width=\linewidth]{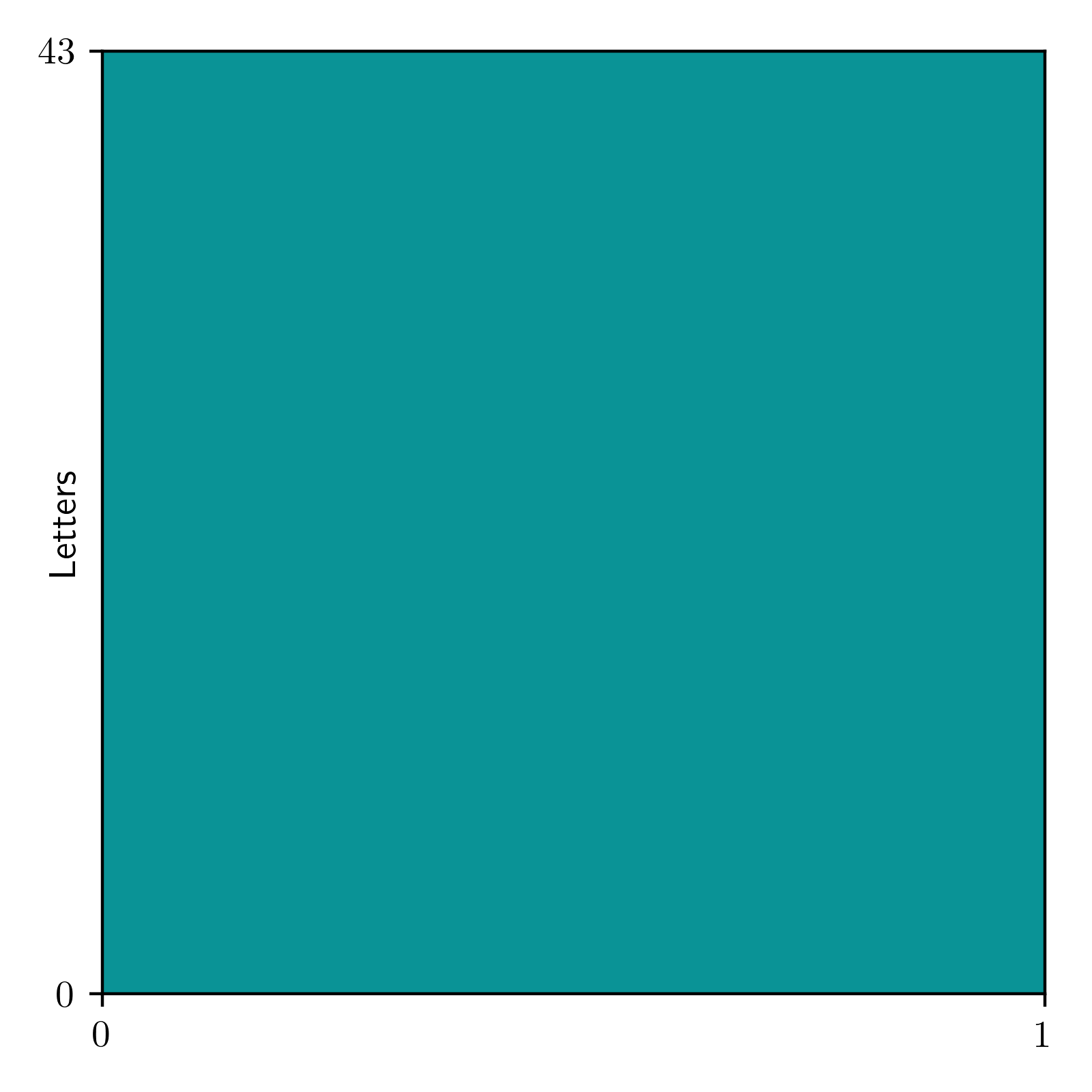}
        \includegraphics[draft=\draft, width=\linewidth]{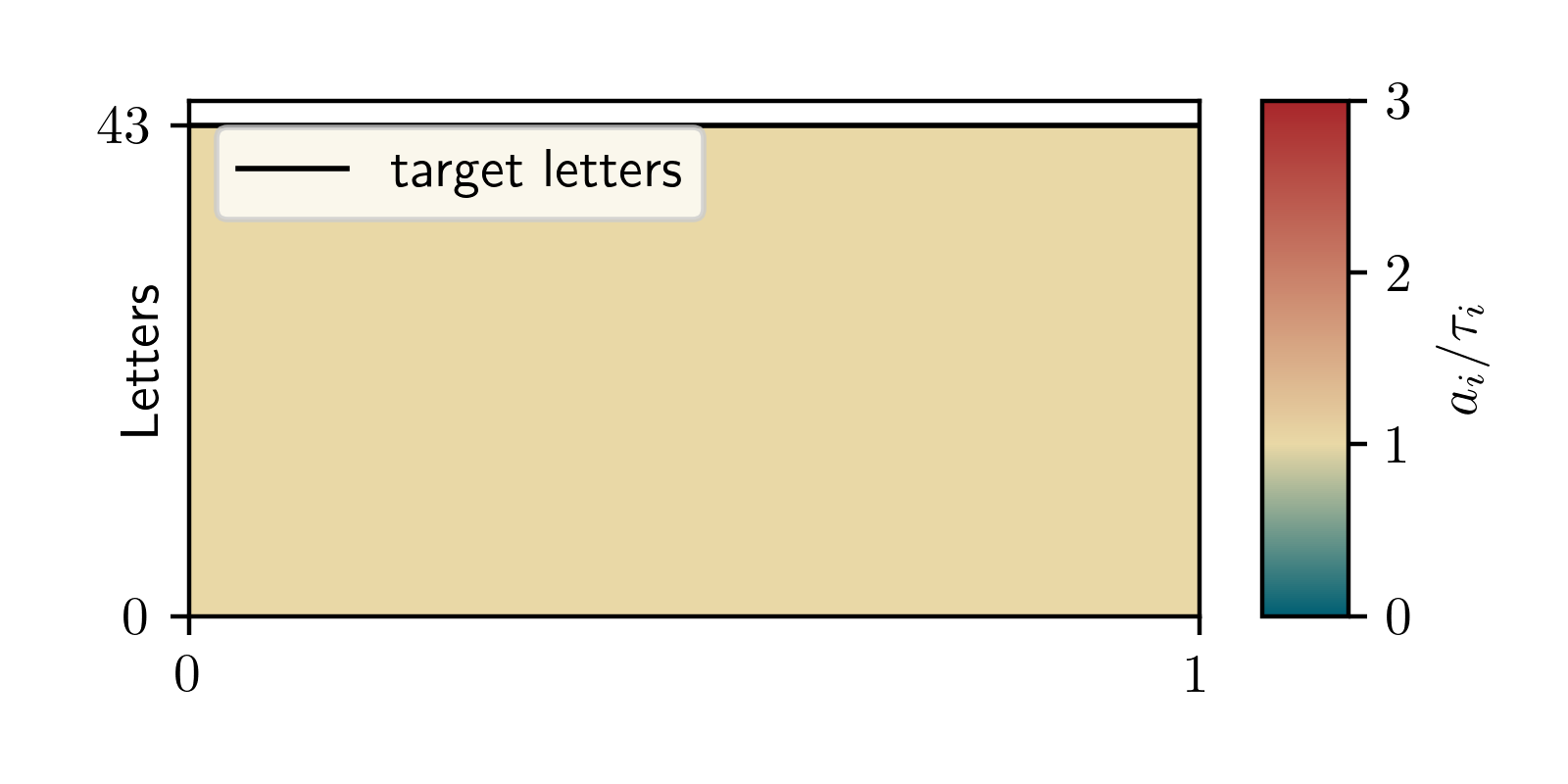}
        \caption{\colgen ($t_G\!=\!1$)}
        \label{fig:results_Saarland_Large_column_generation}
    \end{subfigure}
    \begin{subfigure}{0.32\textwidth}
        \includegraphics[draft=\draft, width=\linewidth]{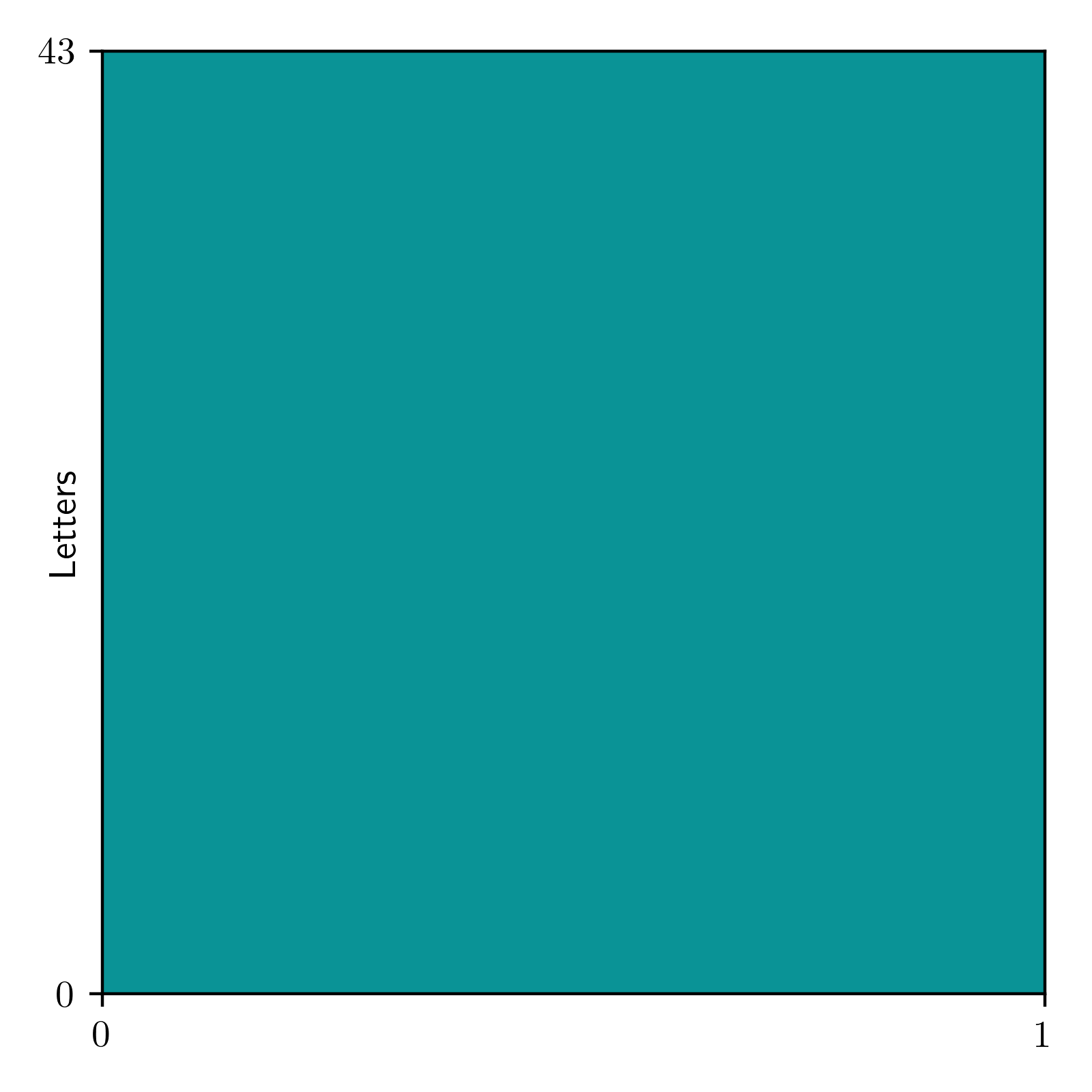}
        \includegraphics[draft=\draft, width=\linewidth]{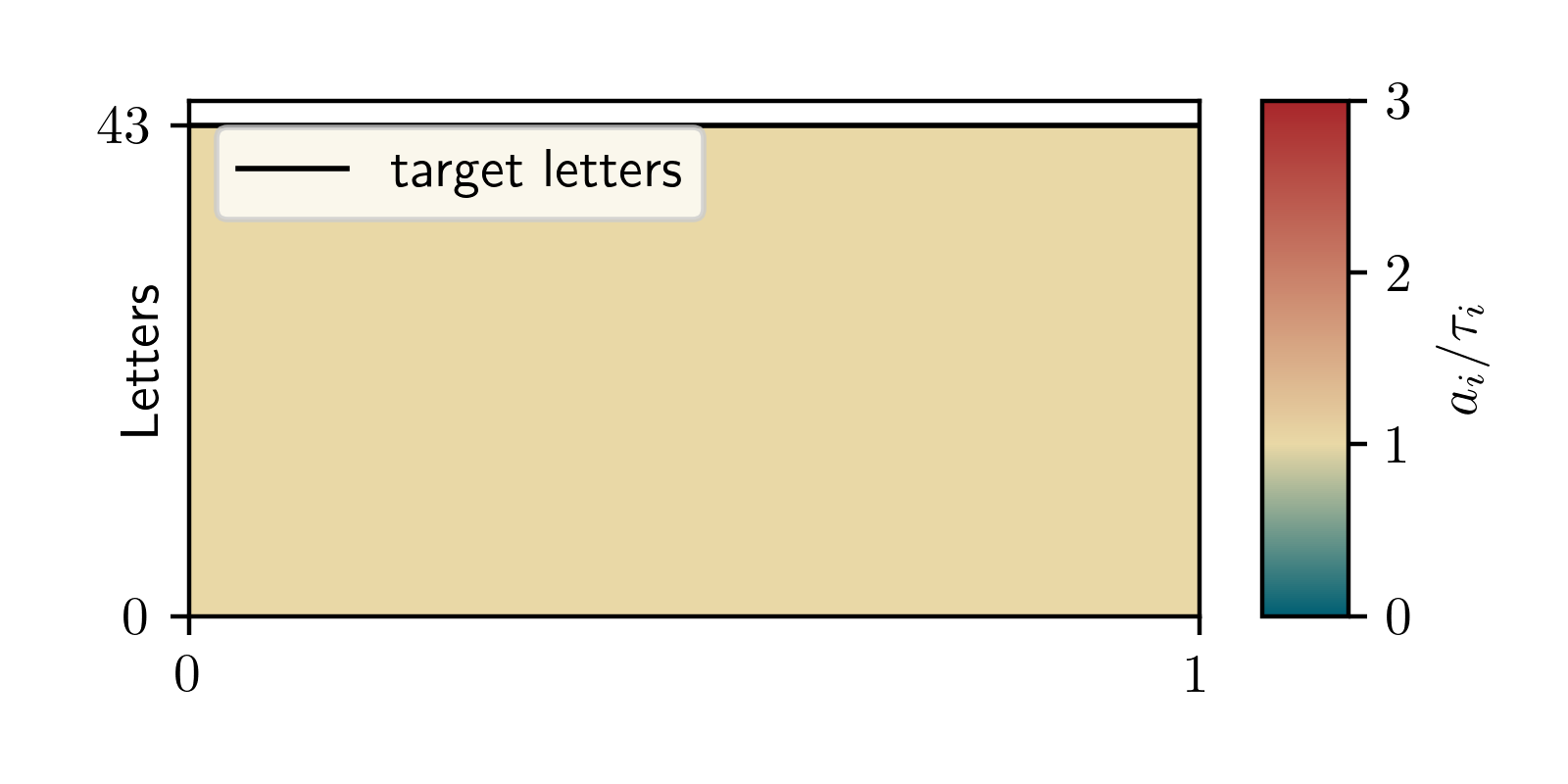}
        \caption{\buckets ($t_G = 1$)}
        \label{fig:results_Saarland_Large_greedy_bucket_fill}
    \end{subfigure}
    \caption{Large municipalities of Saarland ($\ell_G = 43$)}
    \label{fig:results_Saarland_Large}
\end{figure} 

\begin{figure}
    \centering
    \begin{subfigure}{0.32\textwidth}
        \includegraphics[draft=\draft, width=\linewidth]{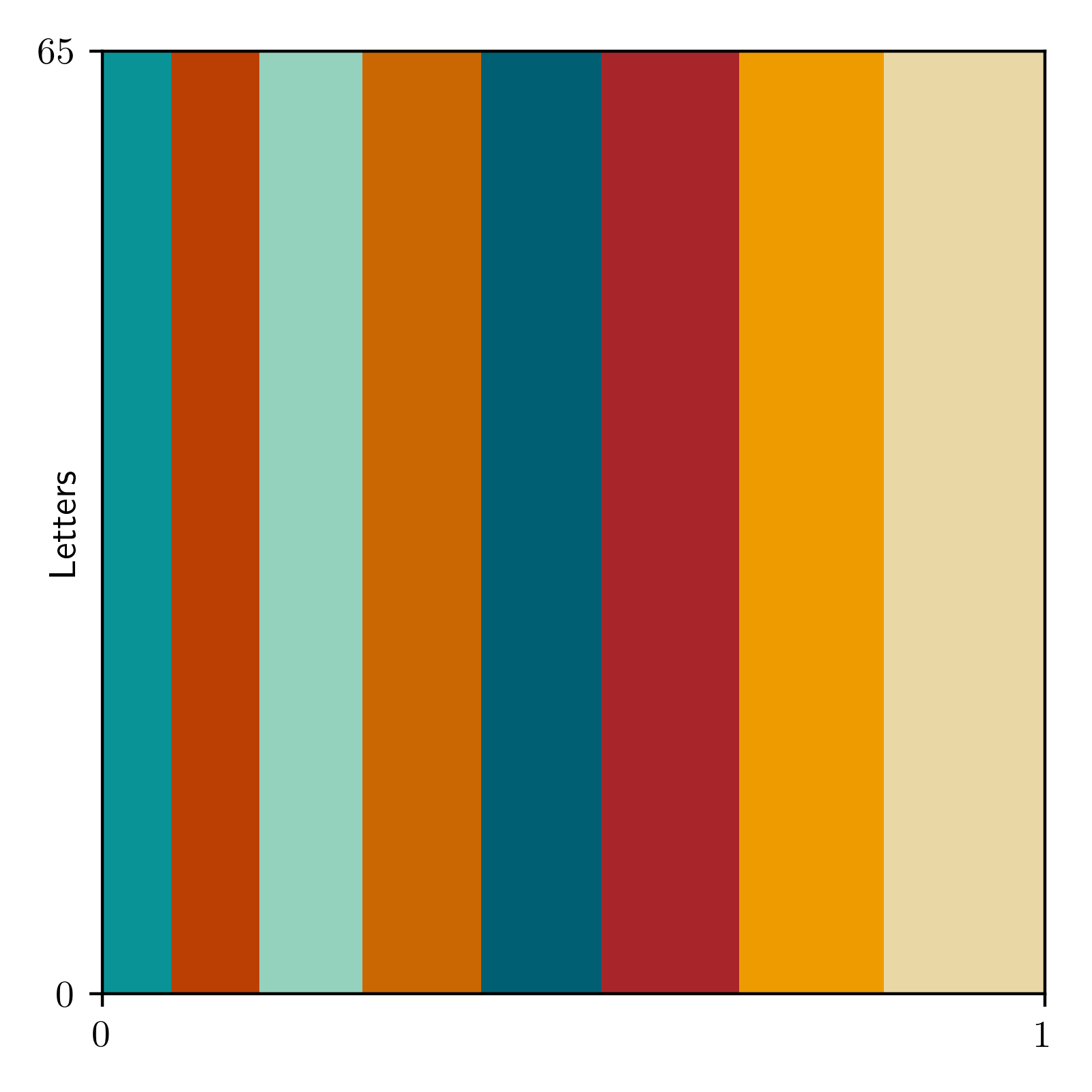}
        \includegraphics[draft=\draft, width=\linewidth]{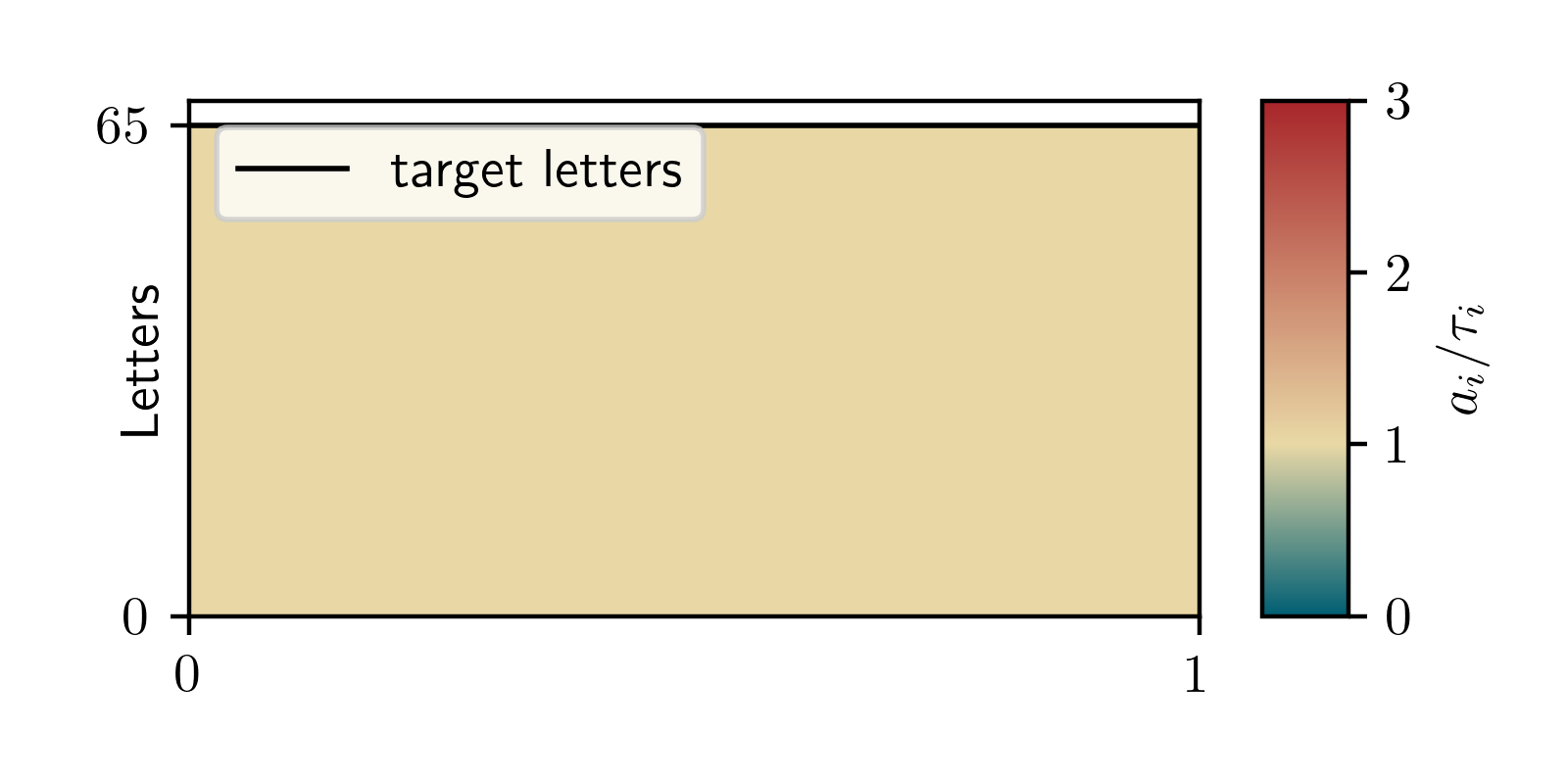}
        \caption{\greq ($t_G = 1$)}
        \label{fig:results_Saarland_Medium_greedy_equal}
    \end{subfigure}
    \begin{subfigure}{0.32\textwidth}
        \includegraphics[draft=\draft, width=\linewidth]{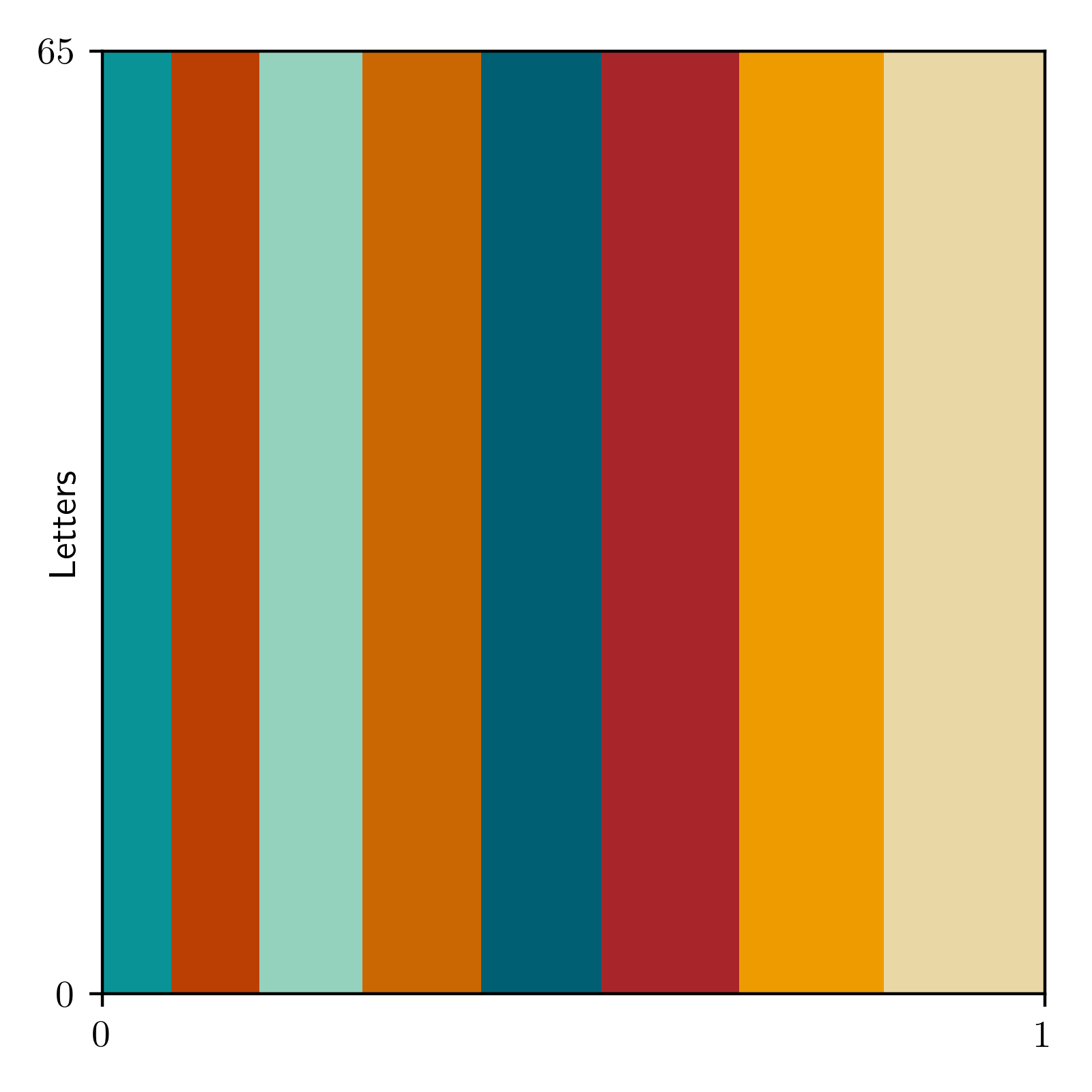}
        \includegraphics[draft=\draft, width=\linewidth]{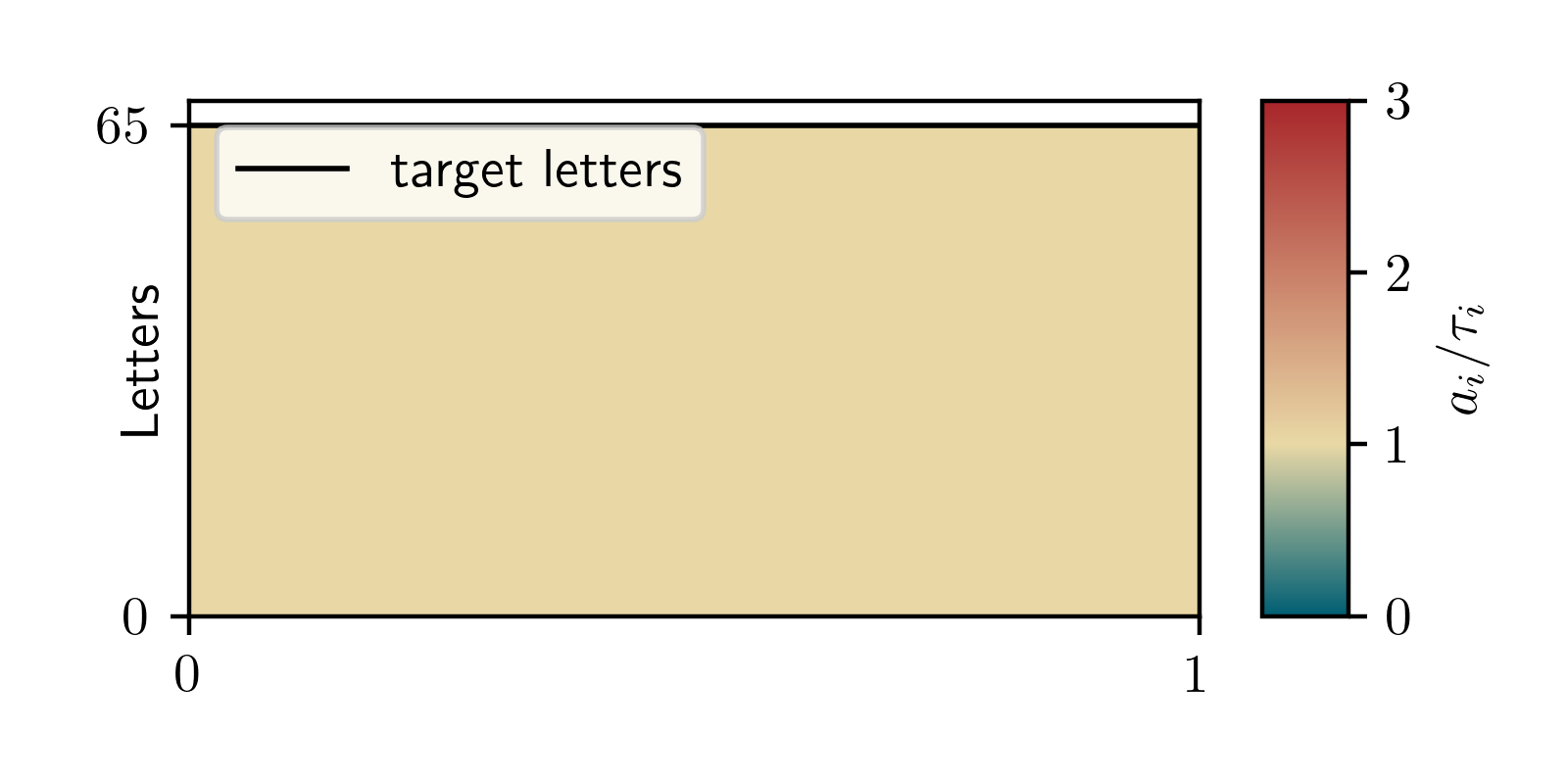}
        \caption{\colgen ($t_G\!=\!1$)}
        \label{fig:results_Saarland_Medium_column_generation}
    \end{subfigure}
    \begin{subfigure}{0.32\textwidth}
        \includegraphics[draft=\draft, width=\linewidth]{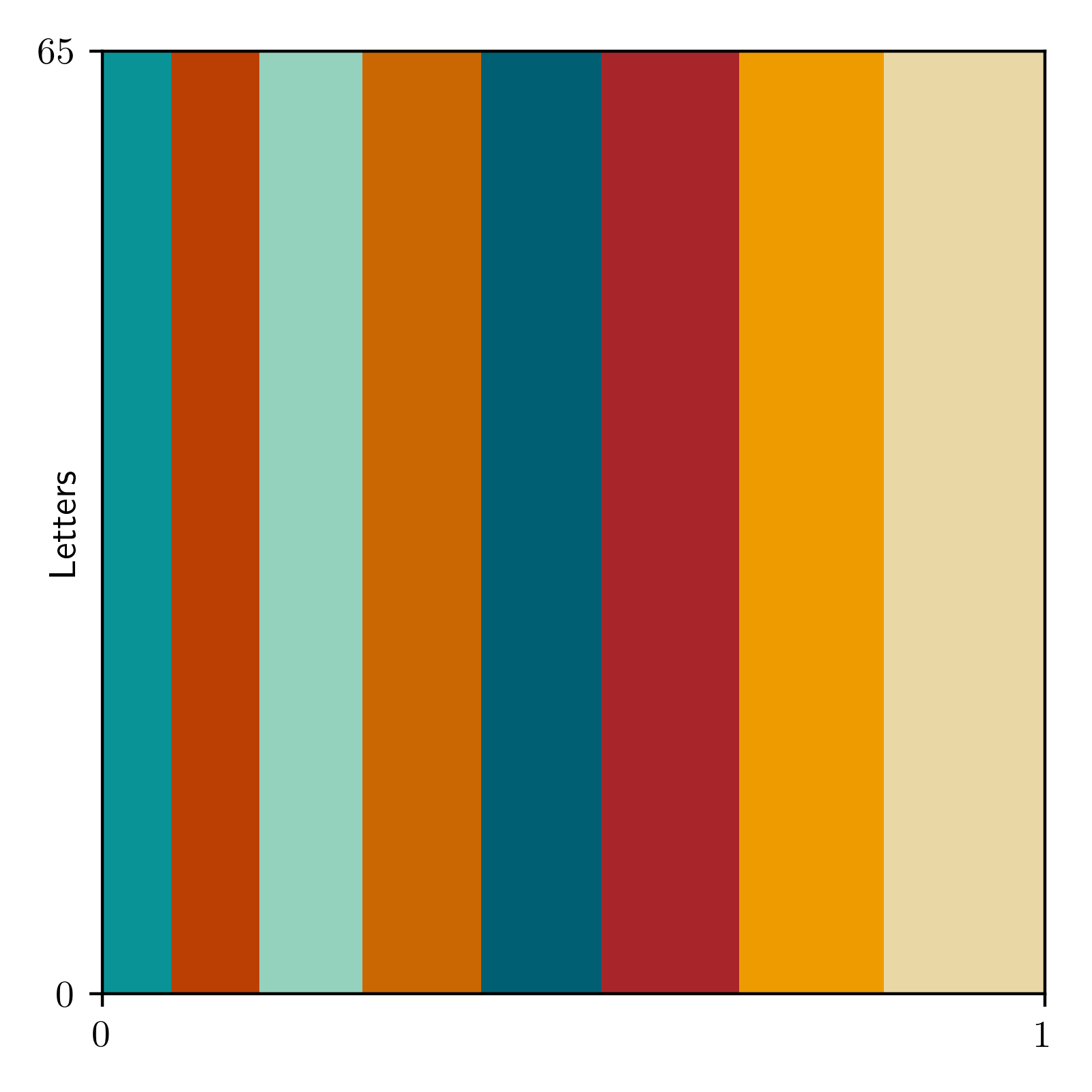}
        \includegraphics[draft=\draft, width=\linewidth]{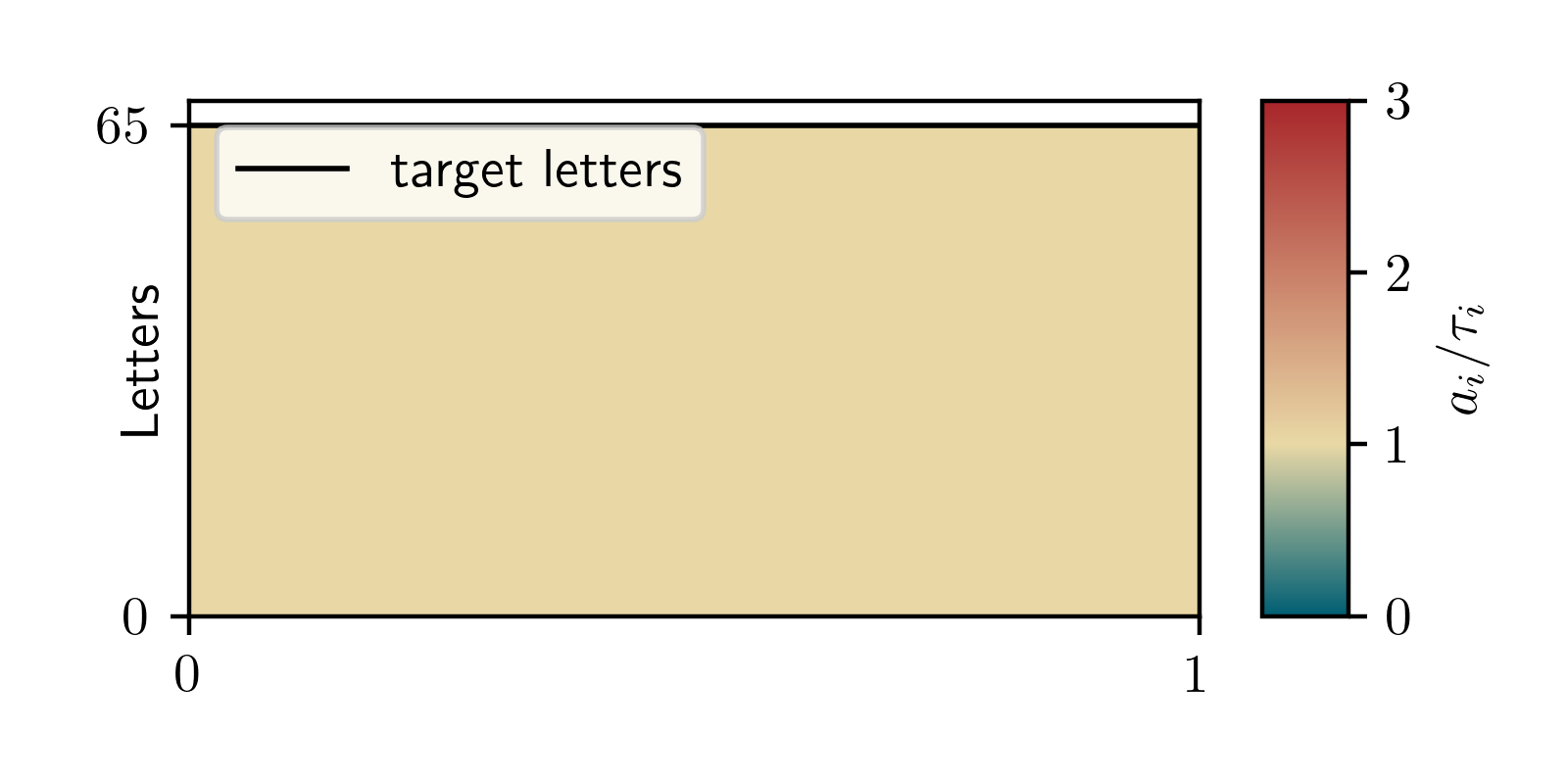}
        \caption{\buckets ($t_G = 1$)}
        \label{fig:results_Saarland_Medium_greedy_bucket_fill}
    \end{subfigure}
    \caption{Medium municipalities of Saarland ($\ell_G = 65$)}
    \label{fig:results_Saarland_Medium}
\end{figure} 

\begin{figure}
    \centering
    \begin{subfigure}{0.32\textwidth}
        \includegraphics[draft=\draft, width=\linewidth]{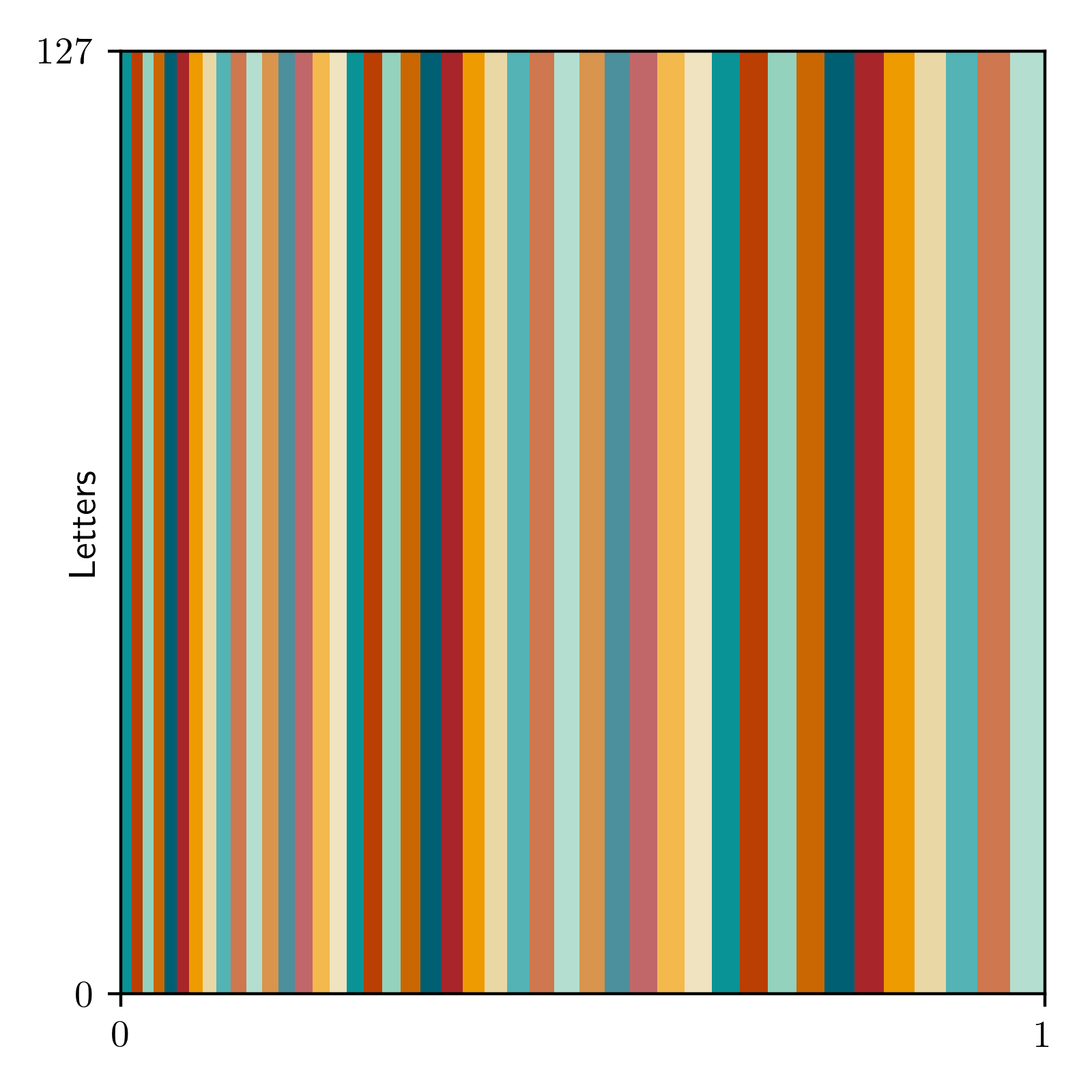}
        \includegraphics[draft=\draft, width=\linewidth]{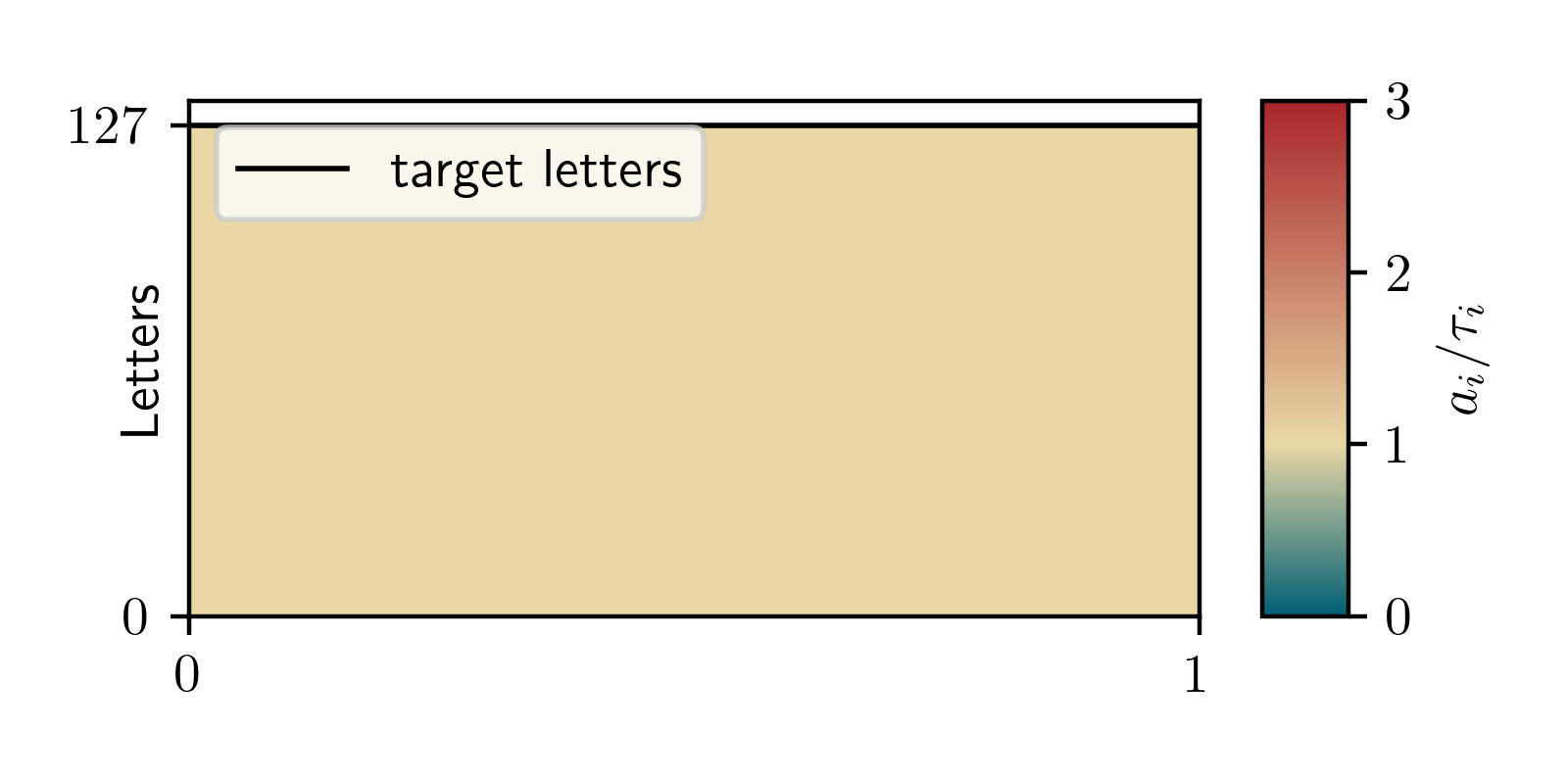}
        \caption{\greq ($t_G = 1$)}
        \label{fig:results_Saarland_Small_greedy_equal}
    \end{subfigure}
    \begin{subfigure}{0.32\textwidth}
        \includegraphics[draft=\draft, width=\linewidth]{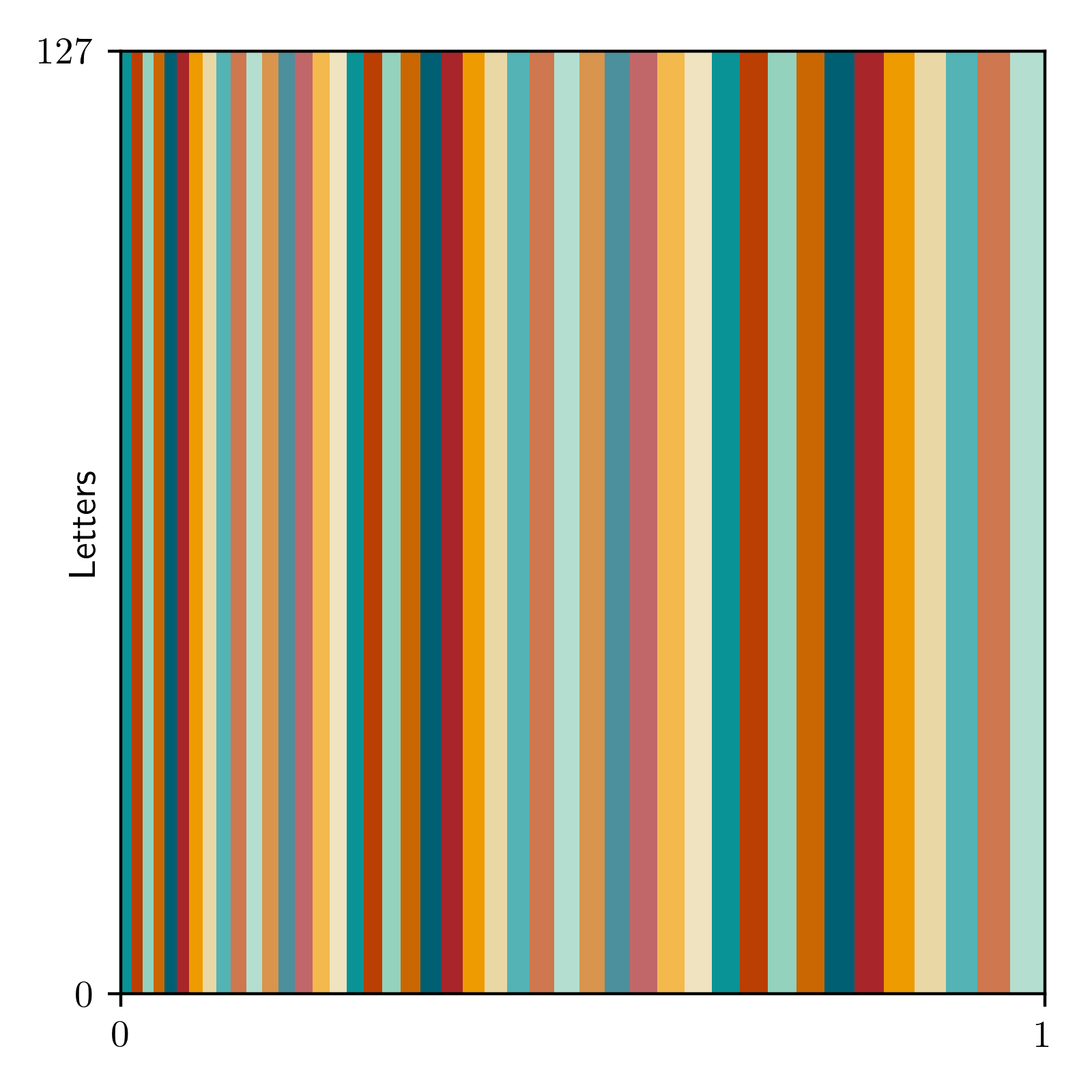}
        \includegraphics[draft=\draft, width=\linewidth]{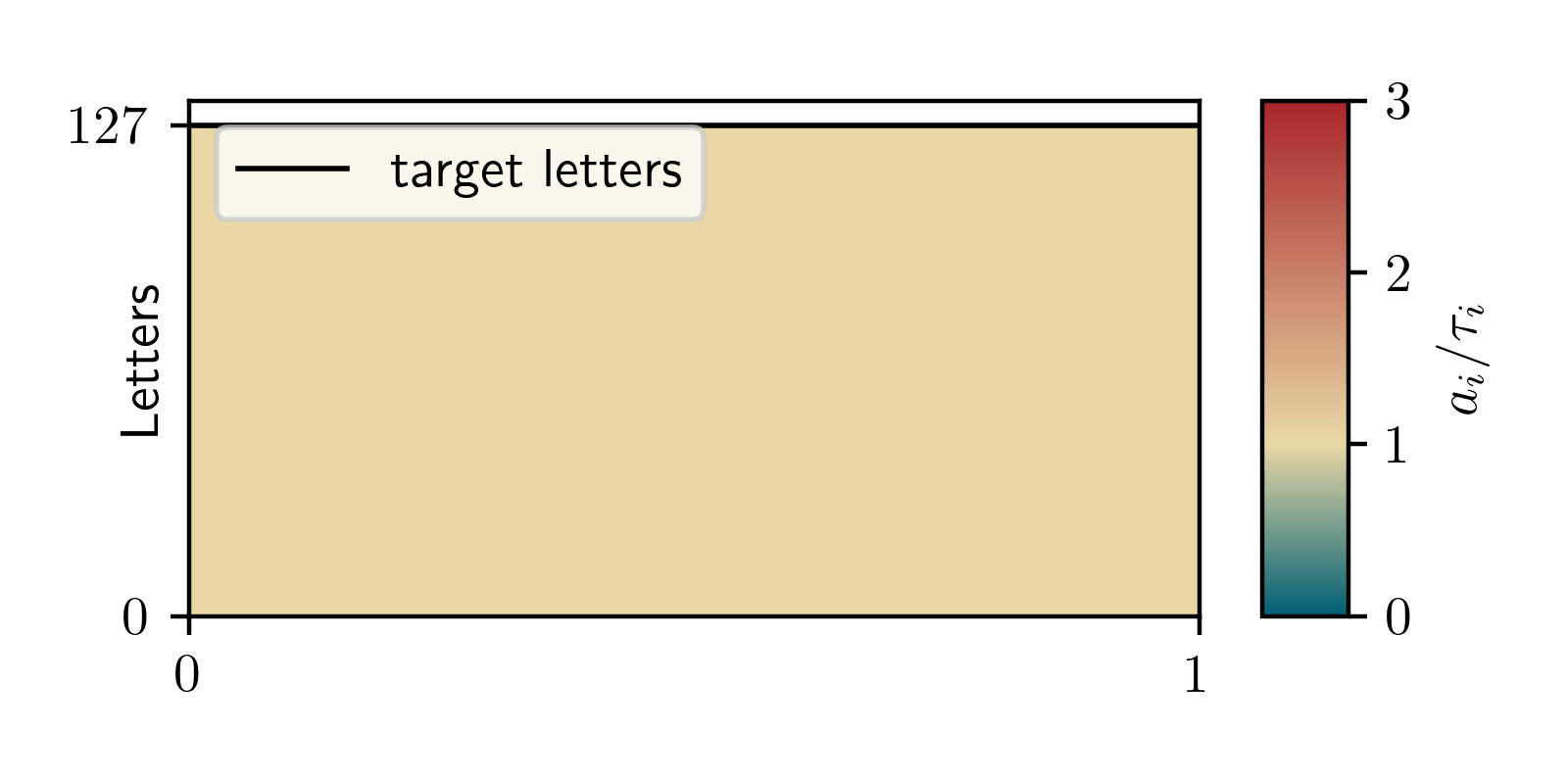}
        \caption{\colgen ($t_G\!=\!1$)}
        \label{fig:results_Saarland_Small_column_generation}
    \end{subfigure}
    \begin{subfigure}{0.32\textwidth}
        \includegraphics[draft=\draft, width=\linewidth]{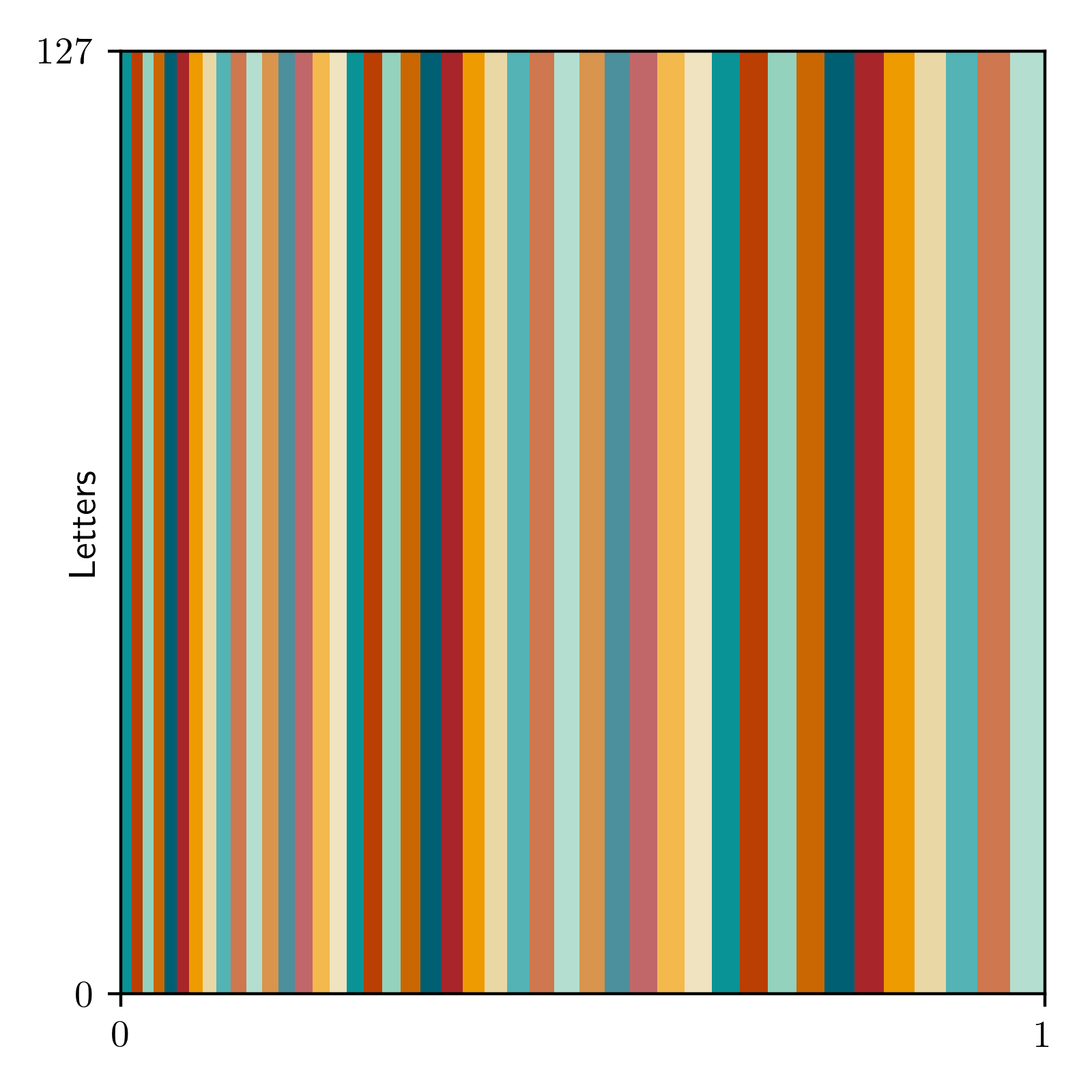}
        \includegraphics[draft=\draft, width=\linewidth]{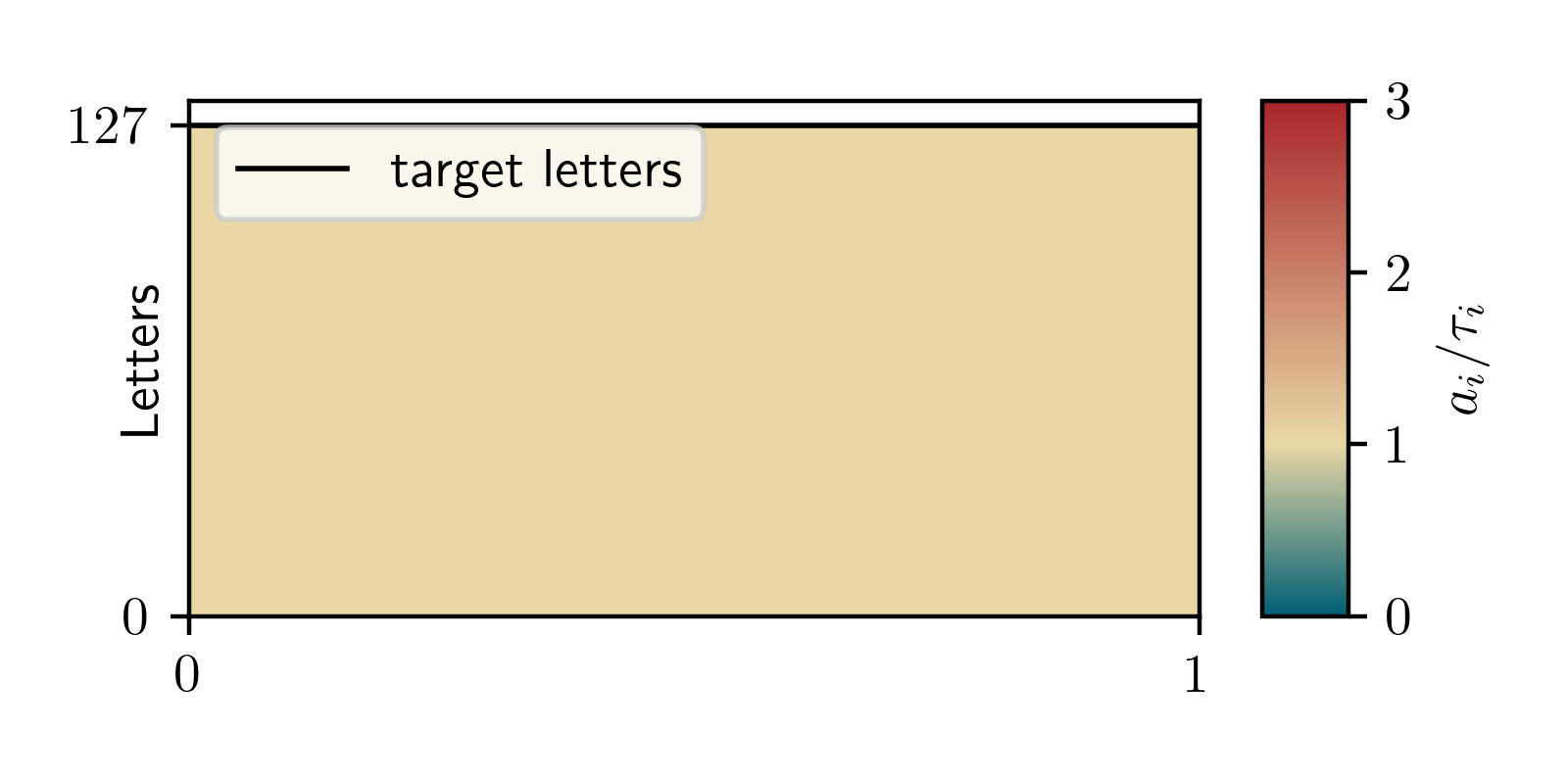}
        \caption{\buckets ($t_G = 1$)}
        \label{fig:results_Saarland_Small_greedy_bucket_fill}
    \end{subfigure}
    \caption{Small municipalities of Saarland ($\ell_G = 127$)}
    \label{fig:results_Saarland_Small}
\end{figure} 

\begin{figure}
    \centering
    \begin{subfigure}{0.32\textwidth}
        \includegraphics[draft=\draft, width=\linewidth]{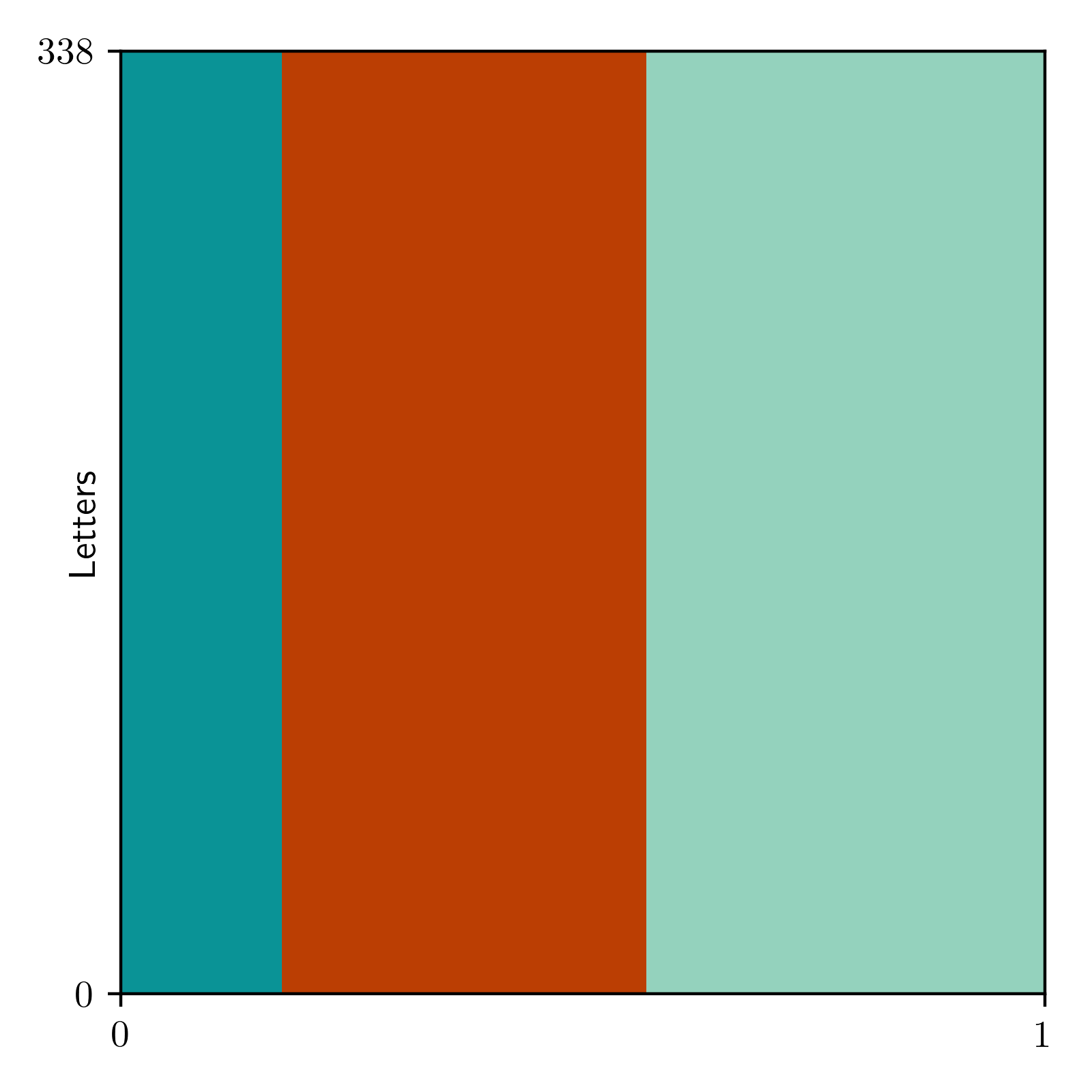}
        \includegraphics[draft=\draft, width=\linewidth]{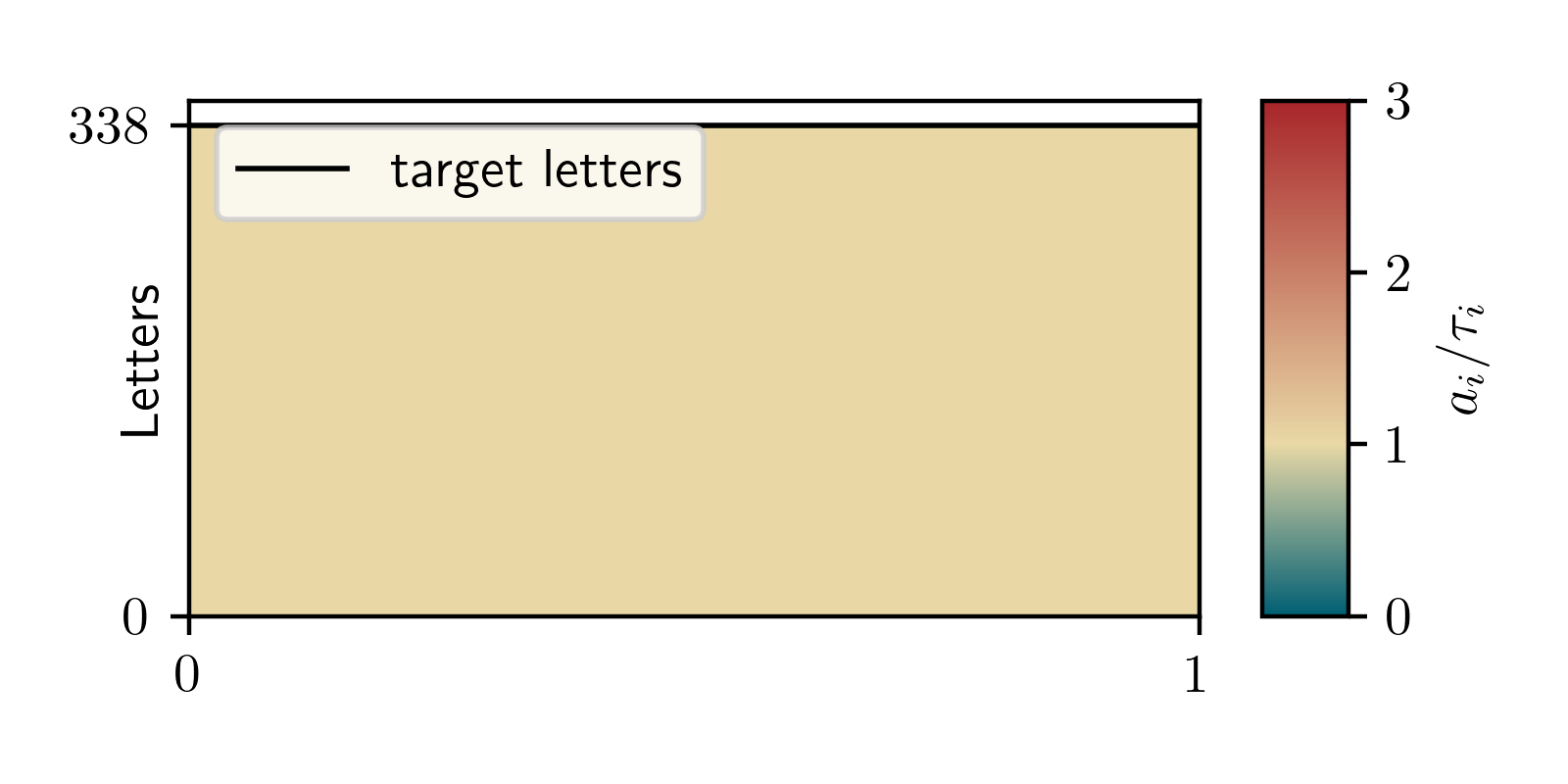}
        \caption{\greq ($t_G = 1$)}
        \label{fig:results_Sachsen_Large_greedy_equal}
    \end{subfigure}
    \begin{subfigure}{0.32\textwidth}
        \includegraphics[draft=\draft, width=\linewidth]{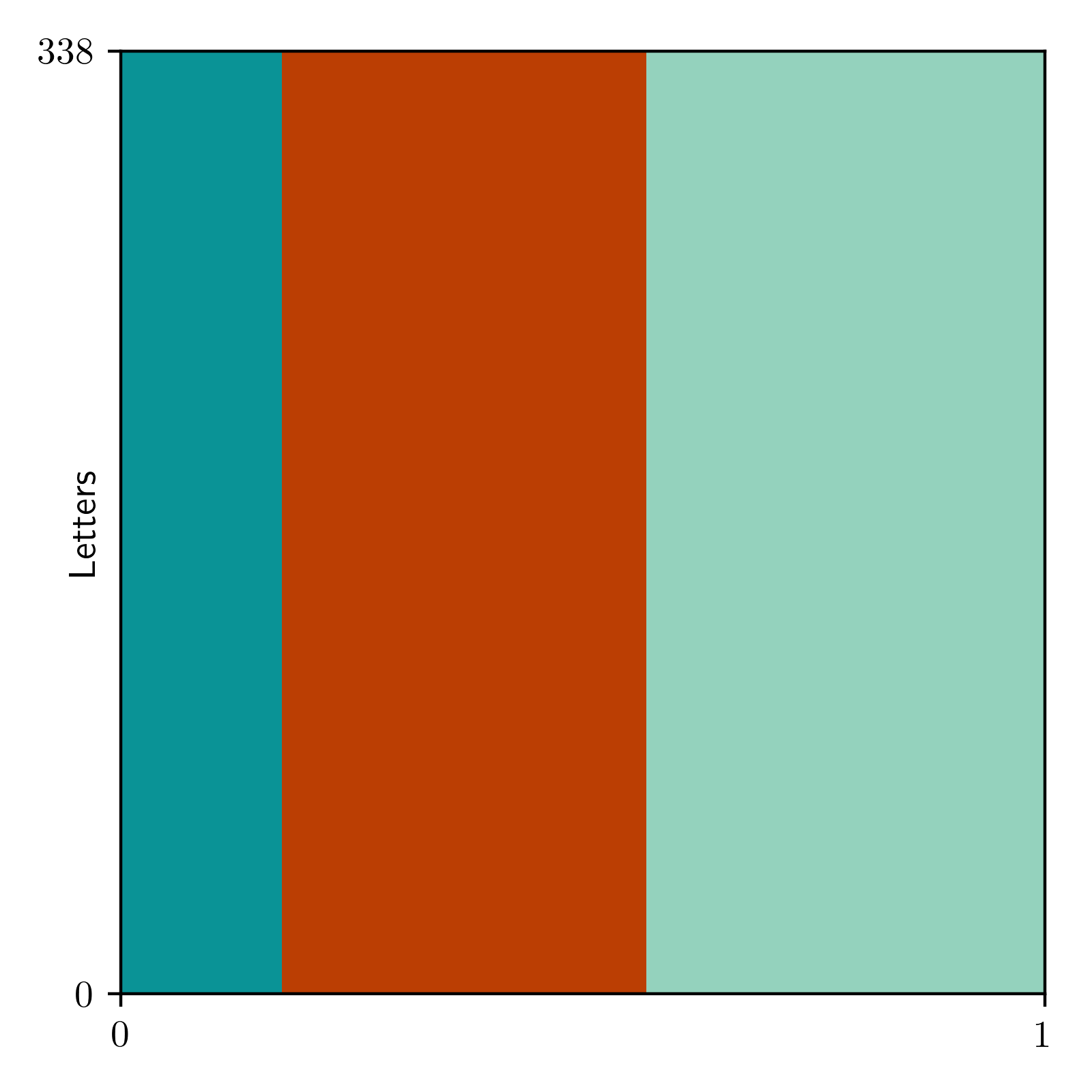}
        \includegraphics[draft=\draft, width=\linewidth]{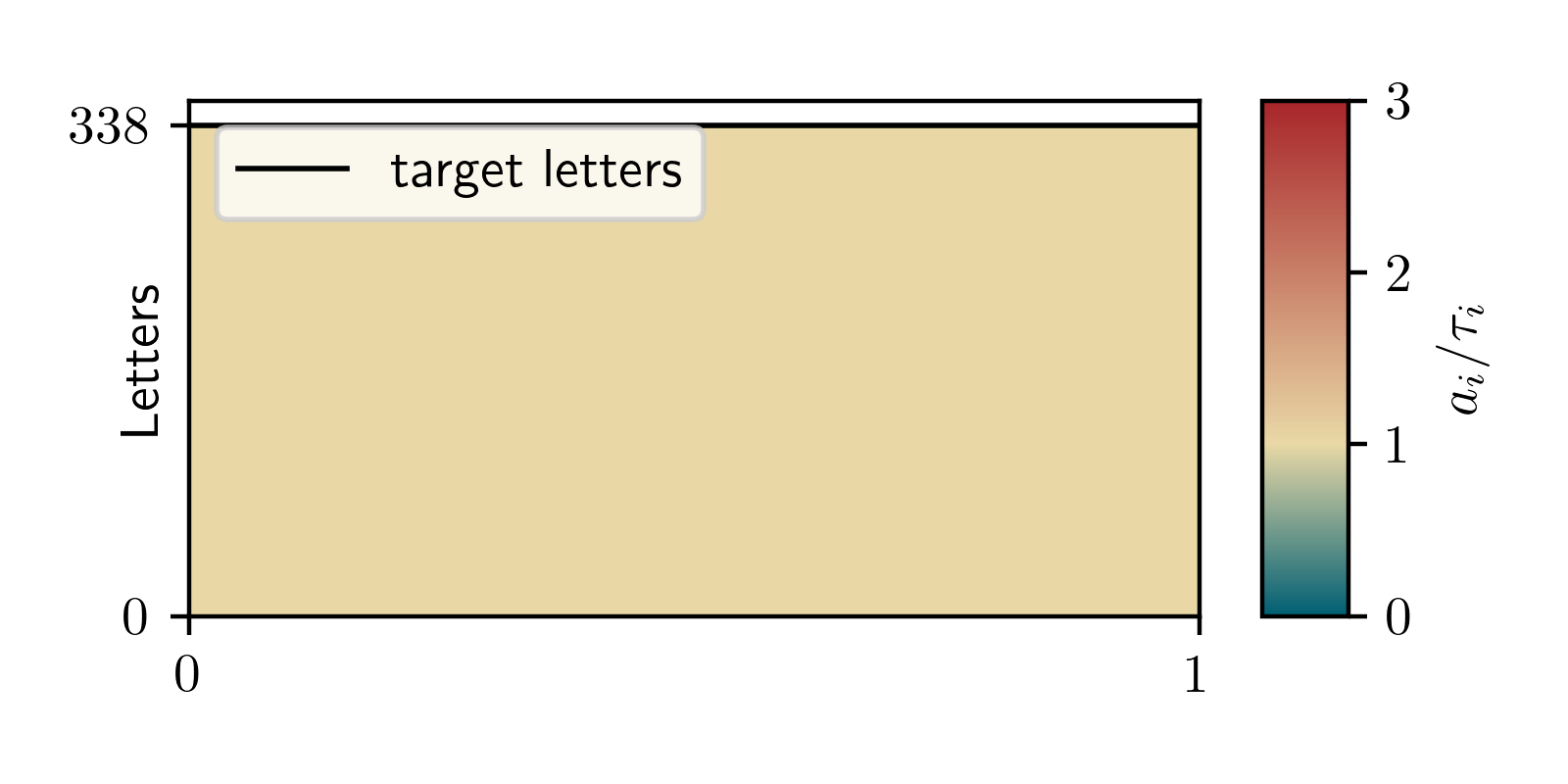}
        \caption{\colgen ($t_G\!=\!1$)}
        \label{fig:results_Sachsen_Large_column_generation}
    \end{subfigure}
    \begin{subfigure}{0.32\textwidth}
        \includegraphics[draft=\draft, width=\linewidth]{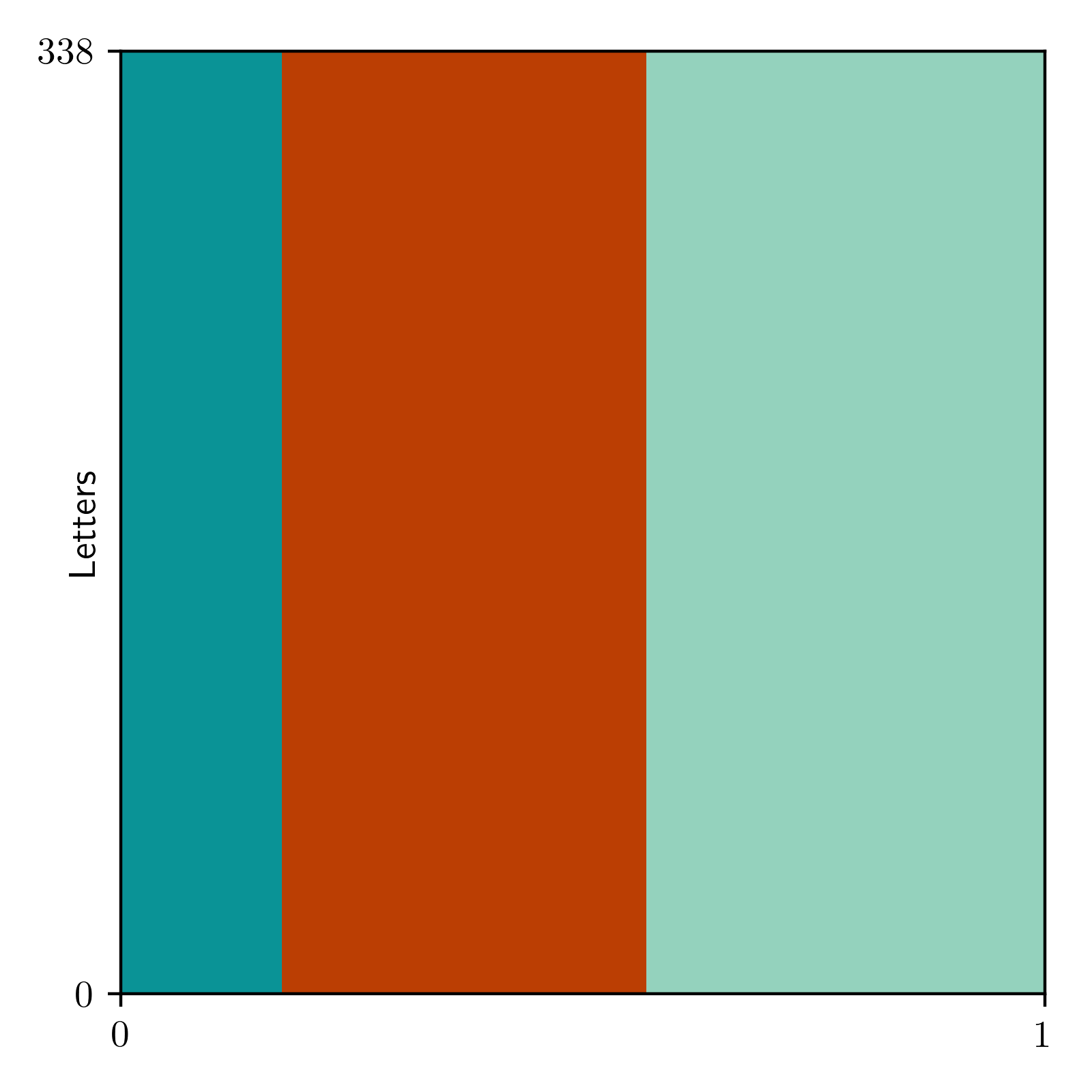}
        \includegraphics[draft=\draft, width=\linewidth]{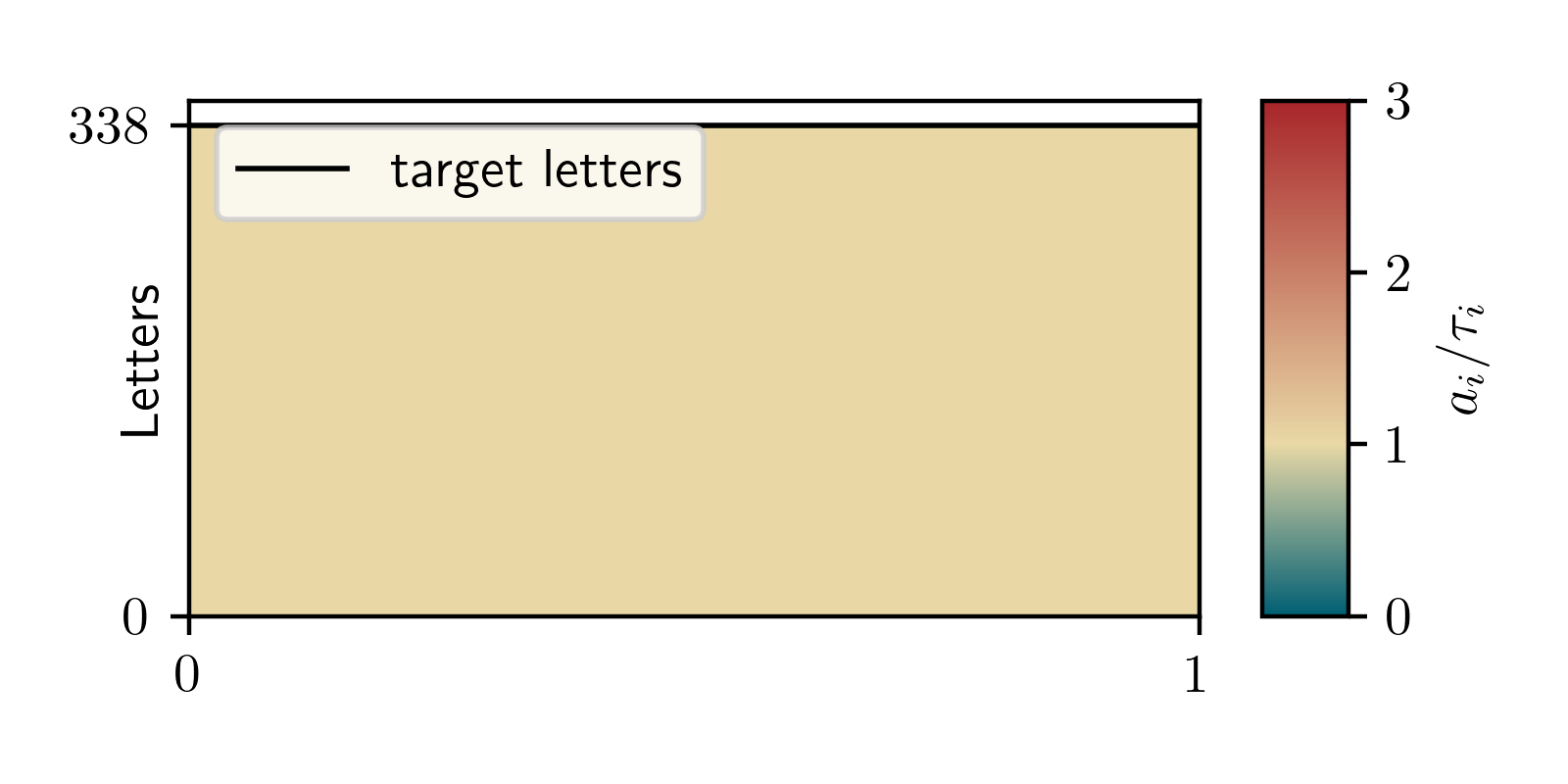}
        \caption{\buckets ($t_G = 1$)}
        \label{fig:results_Sachsen_Large_greedy_bucket_fill}
    \end{subfigure}
    \caption{Large municipalities of Sachsen ($\ell_G = 338$)}
    \label{fig:results_Sachsen_Large}
\end{figure} 

\begin{figure}
    \centering
    \begin{subfigure}{0.32\textwidth}
        \includegraphics[draft=\draft, width=\linewidth]{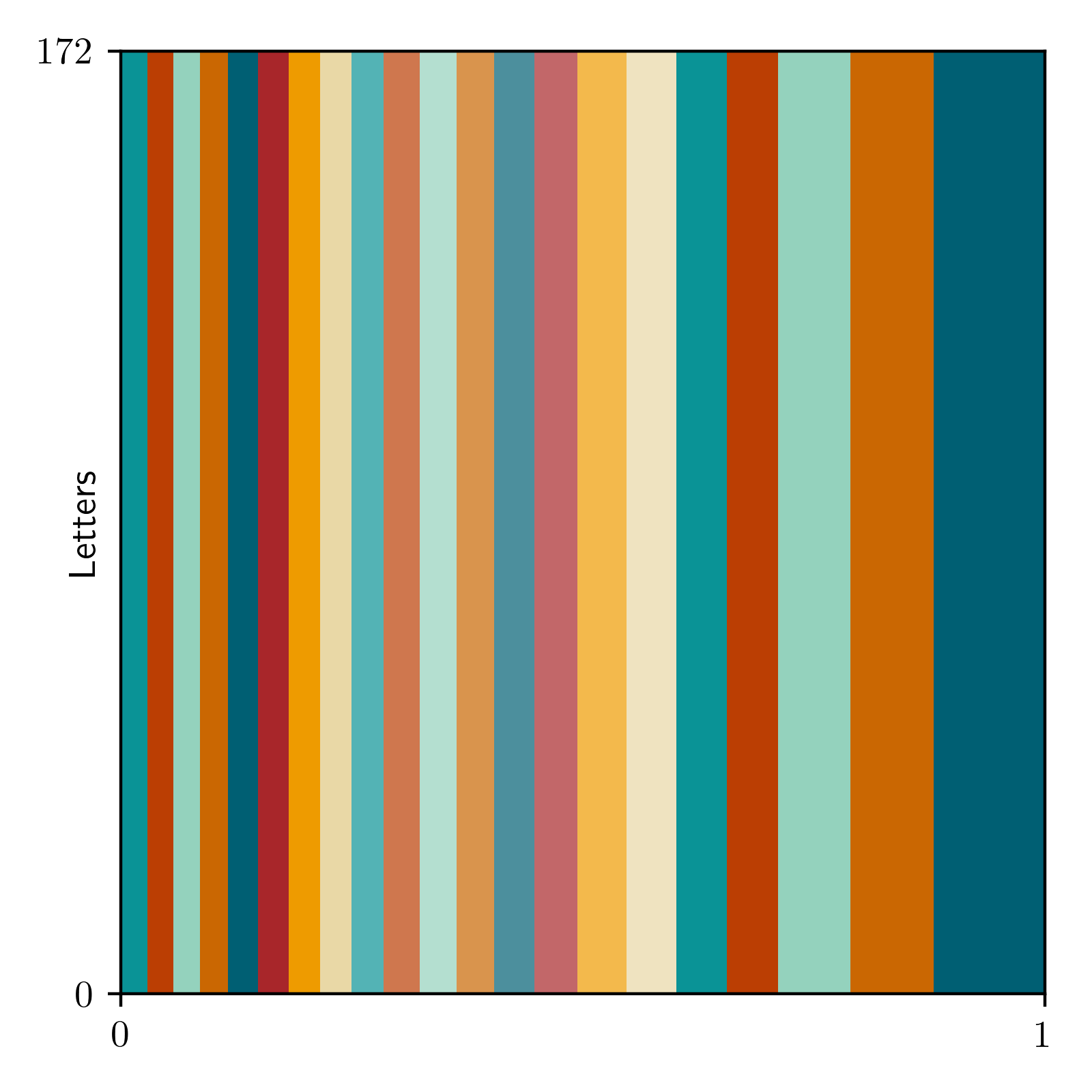}
        \includegraphics[draft=\draft, width=\linewidth]{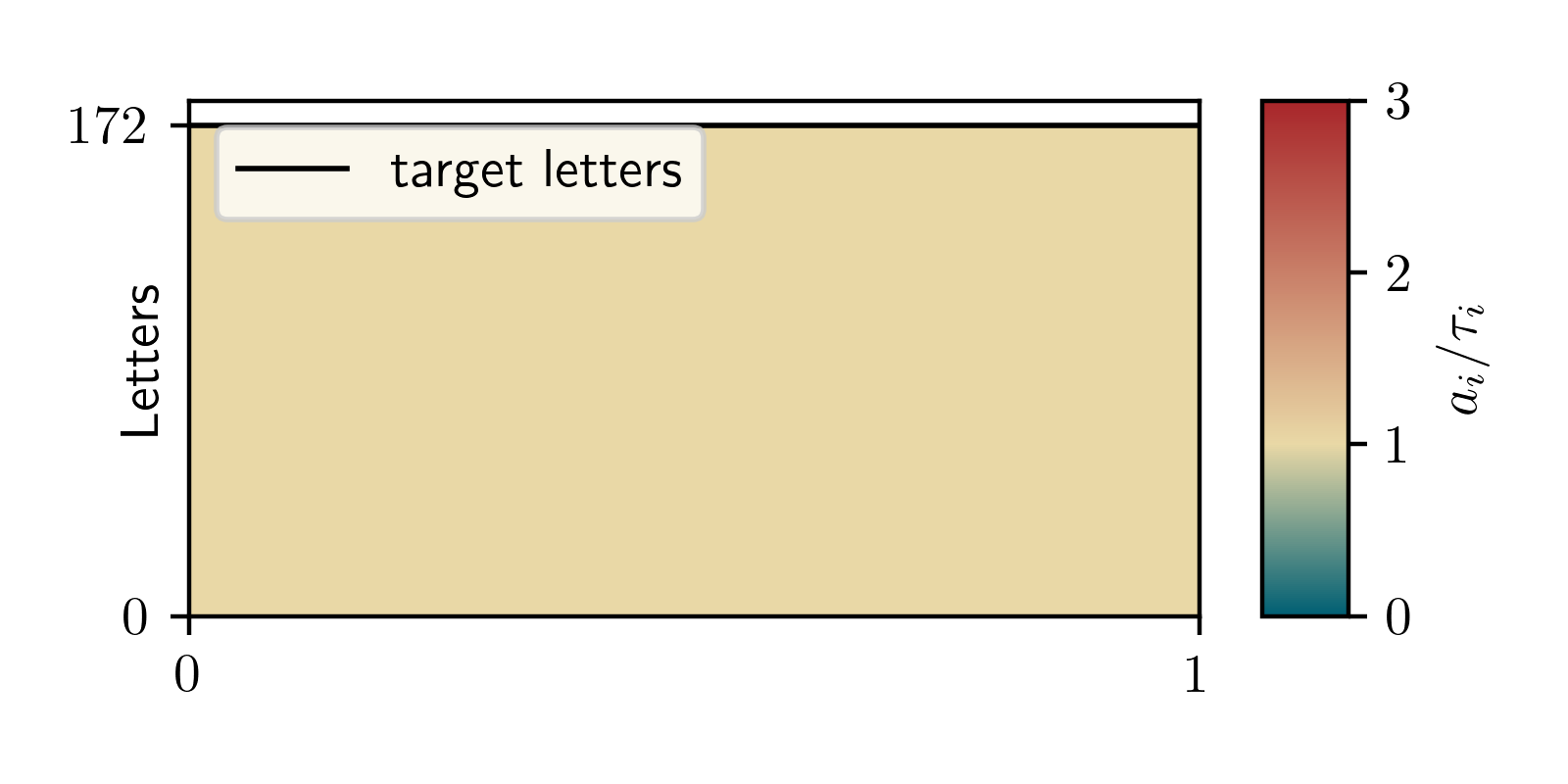}
        \caption{\greq ($t_G = 1$)}
        \label{fig:results_Sachsen_Medium_greedy_equal}
    \end{subfigure}
    \begin{subfigure}{0.32\textwidth}
        \includegraphics[draft=\draft, width=\linewidth]{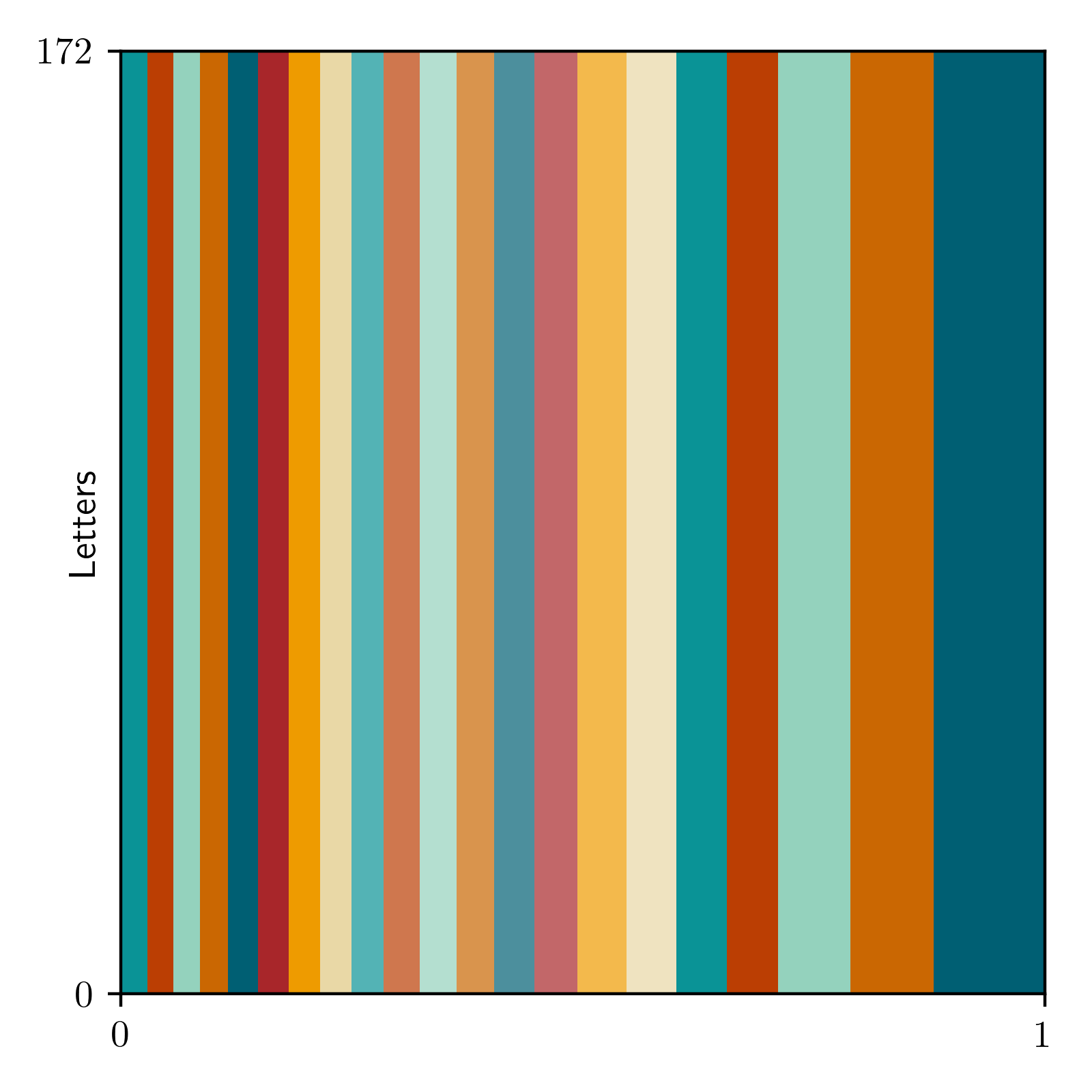}
        \includegraphics[draft=\draft, width=\linewidth]{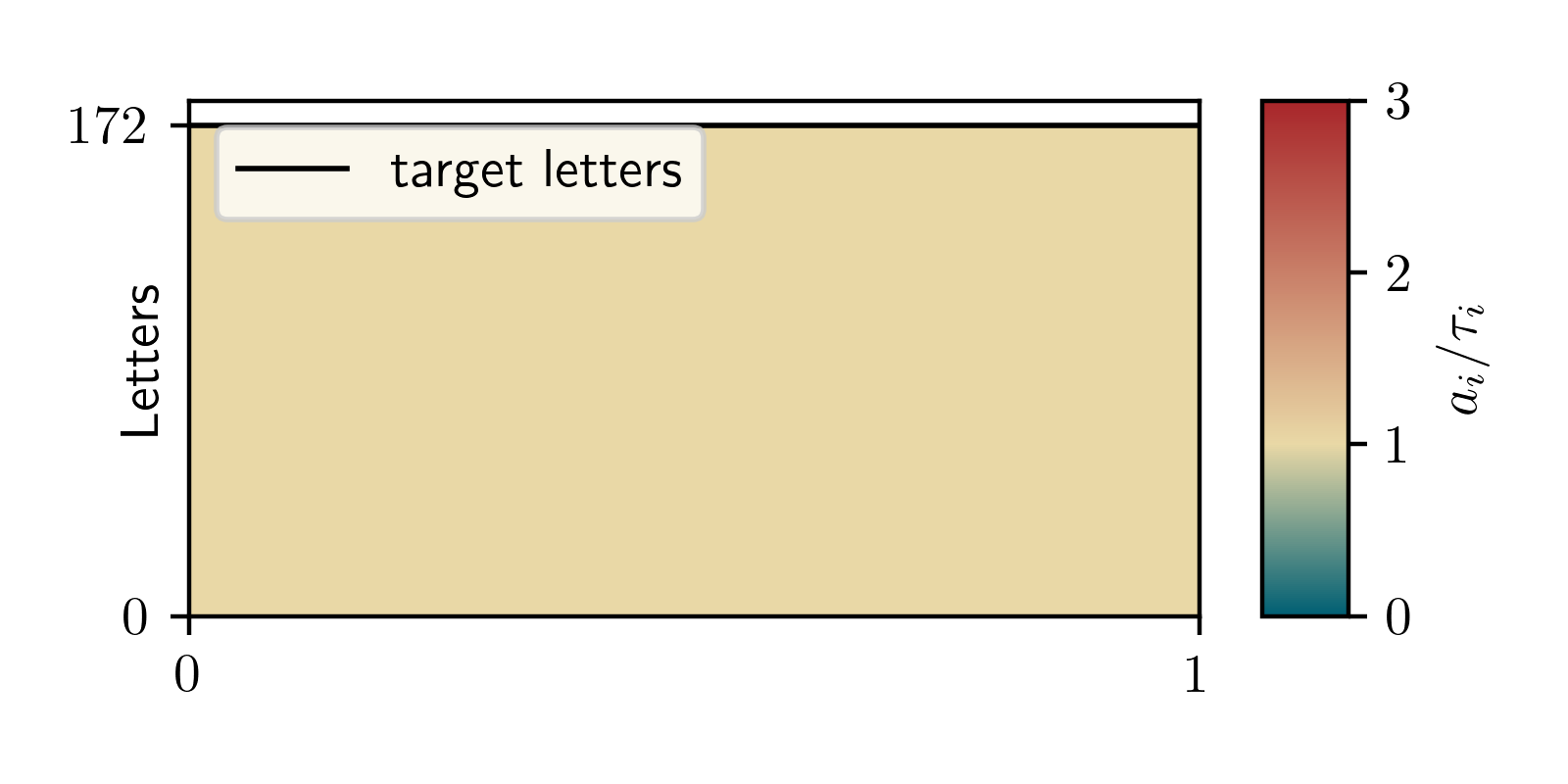}
        \caption{\colgen ($t_G\!=\!1$)}
        \label{fig:results_Sachsen_Medium_column_generation}
    \end{subfigure}
    \begin{subfigure}{0.32\textwidth}
        \includegraphics[draft=\draft, width=\linewidth]{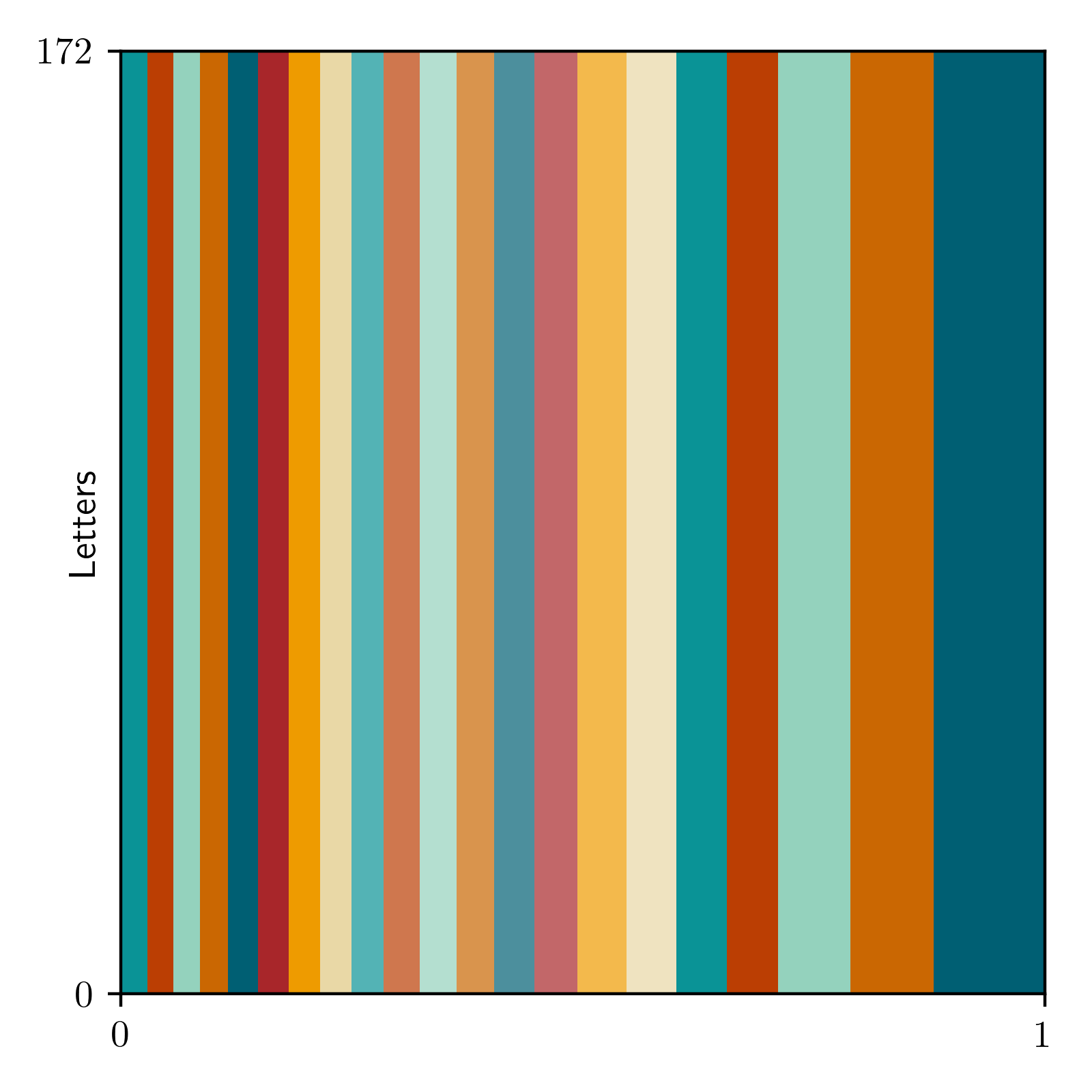}
        \includegraphics[draft=\draft, width=\linewidth]{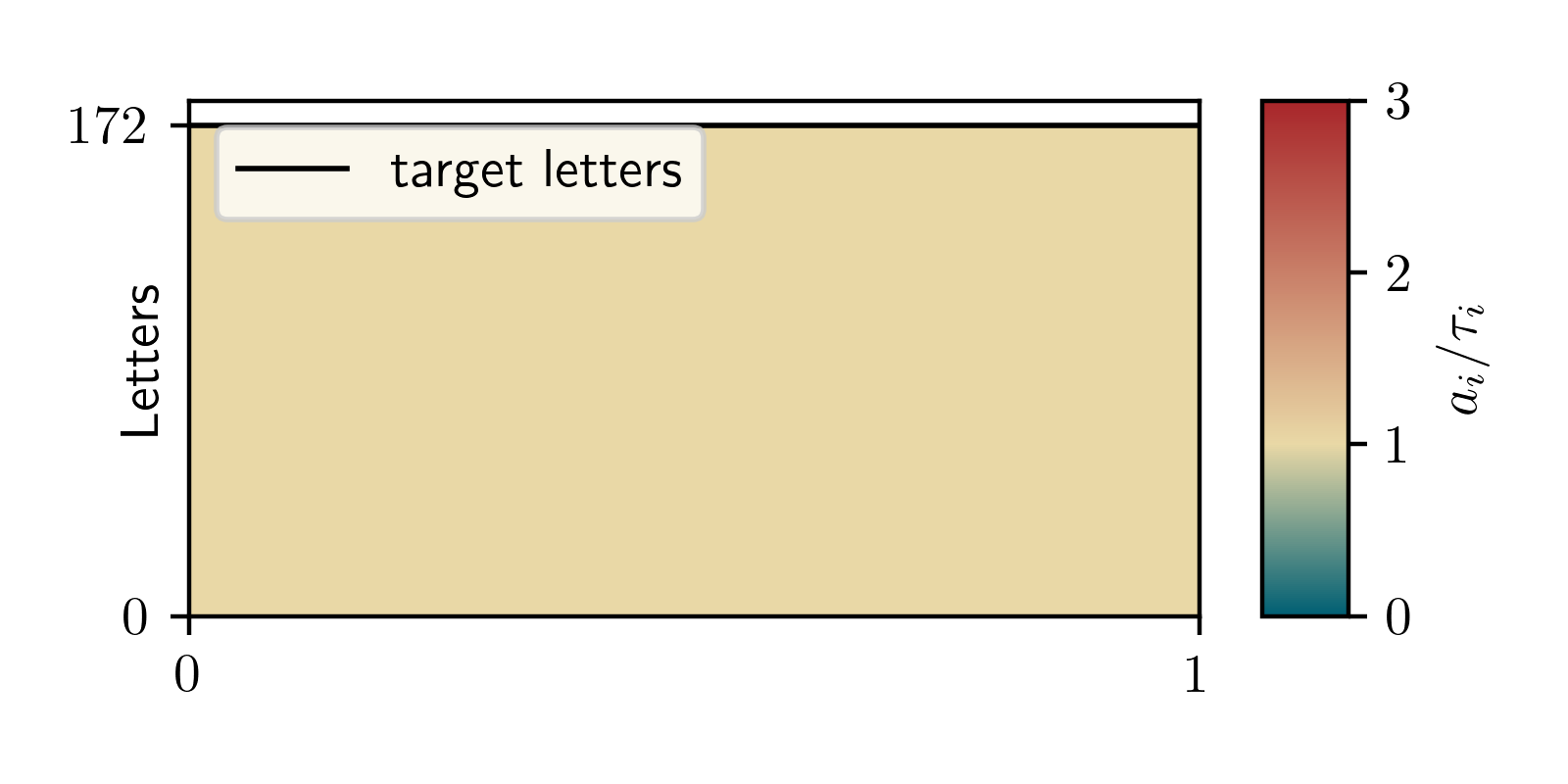}
        \caption{\buckets ($t_G = 1$)}
        \label{fig:results_Sachsen_Medium_greedy_bucket_fill}
    \end{subfigure}
    \caption{Medium municipalities of Sachsen ($\ell_G = 172$)}
    \label{fig:results_Sachsen_Medium}
\end{figure} 

\begin{figure}
    \centering
    \begin{subfigure}{0.32\textwidth}
        \includegraphics[draft=\draft, width=\linewidth]{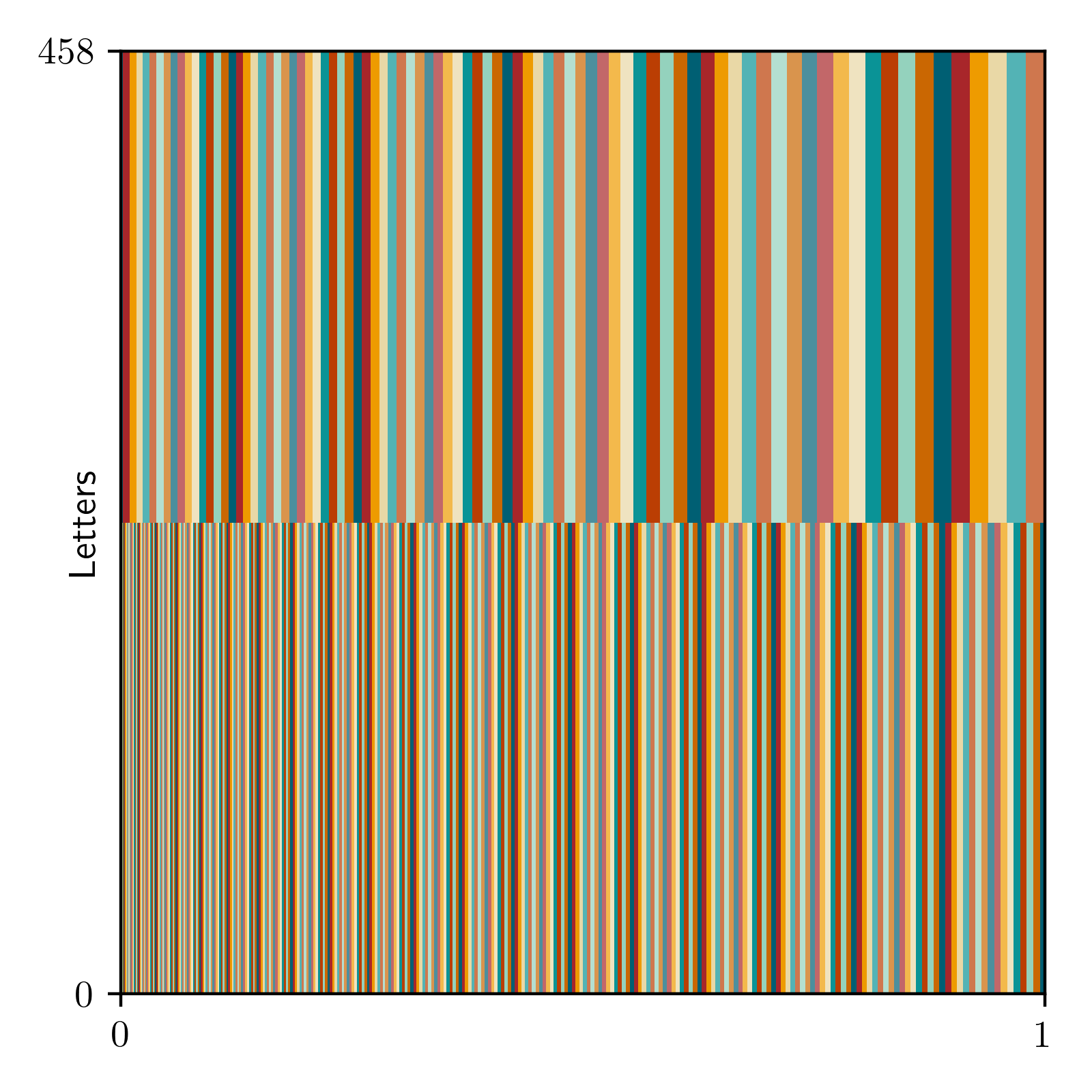}
        \includegraphics[draft=\draft, width=\linewidth]{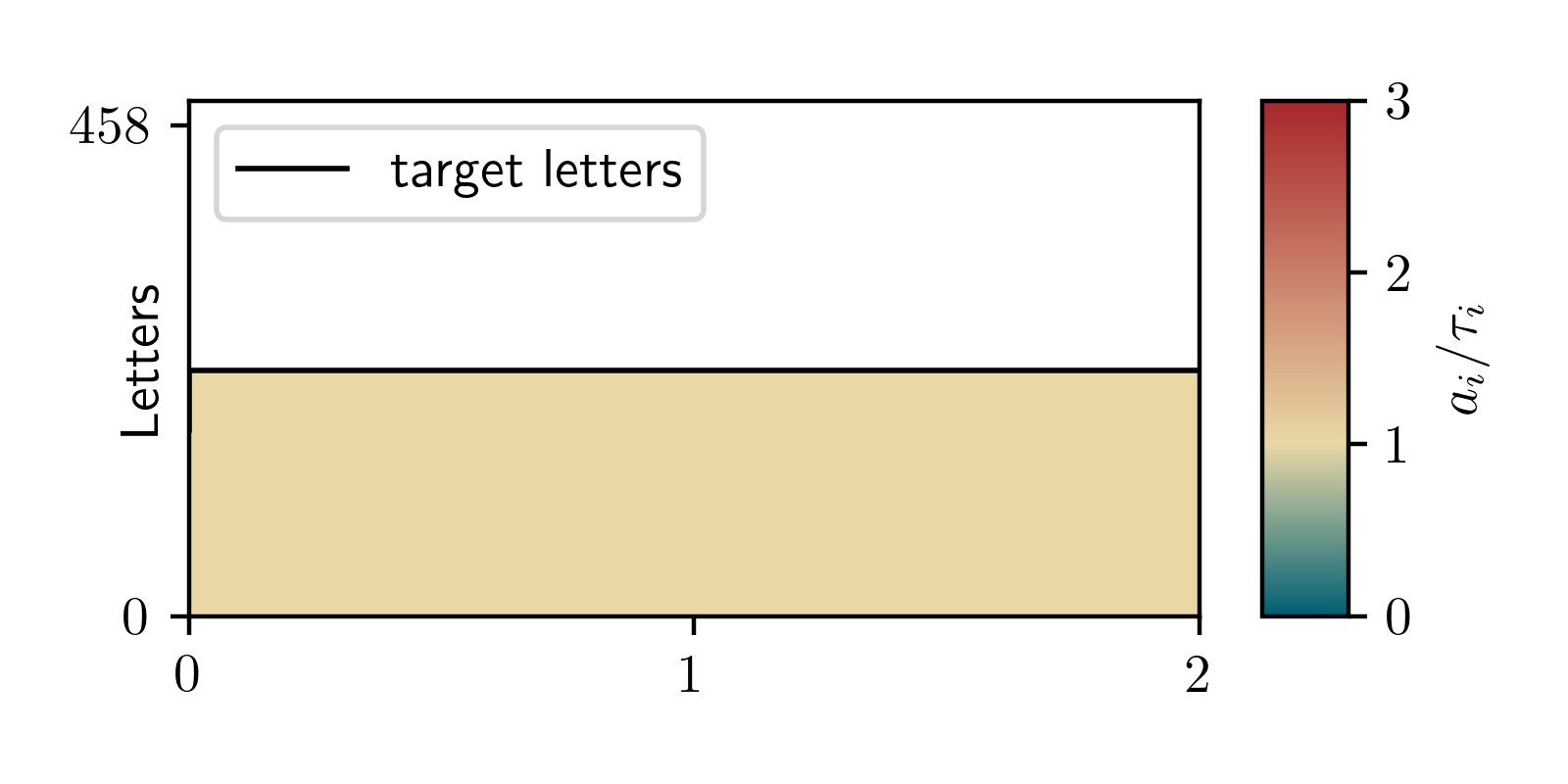}
        \caption{\greq ($t_G = 2$)}
        \label{fig:results_Sachsen_Small_greedy_equal}
    \end{subfigure}
    \begin{subfigure}{0.32\textwidth}
        \includegraphics[draft=\draft, width=\linewidth]{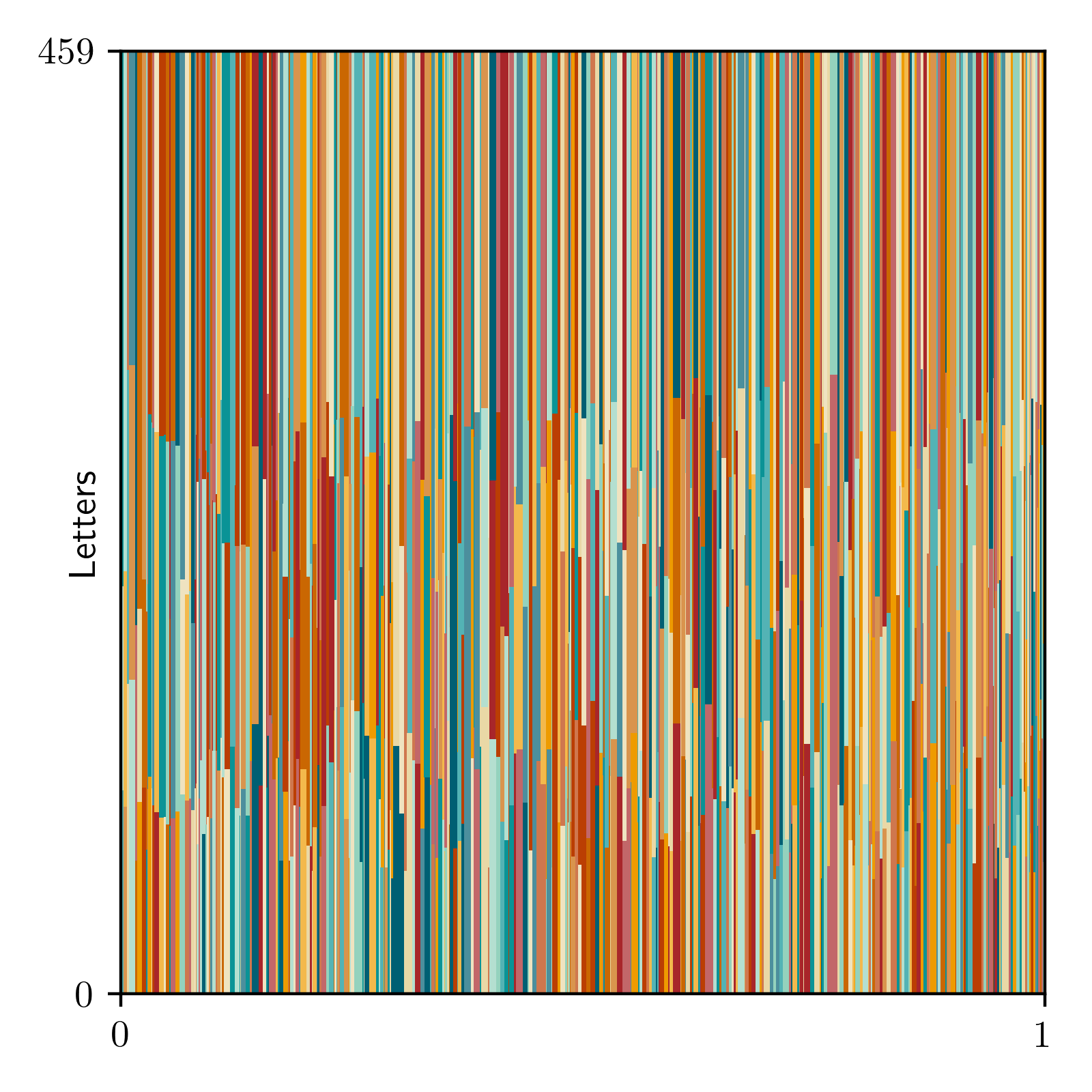}
        \includegraphics[draft=\draft, width=\linewidth]{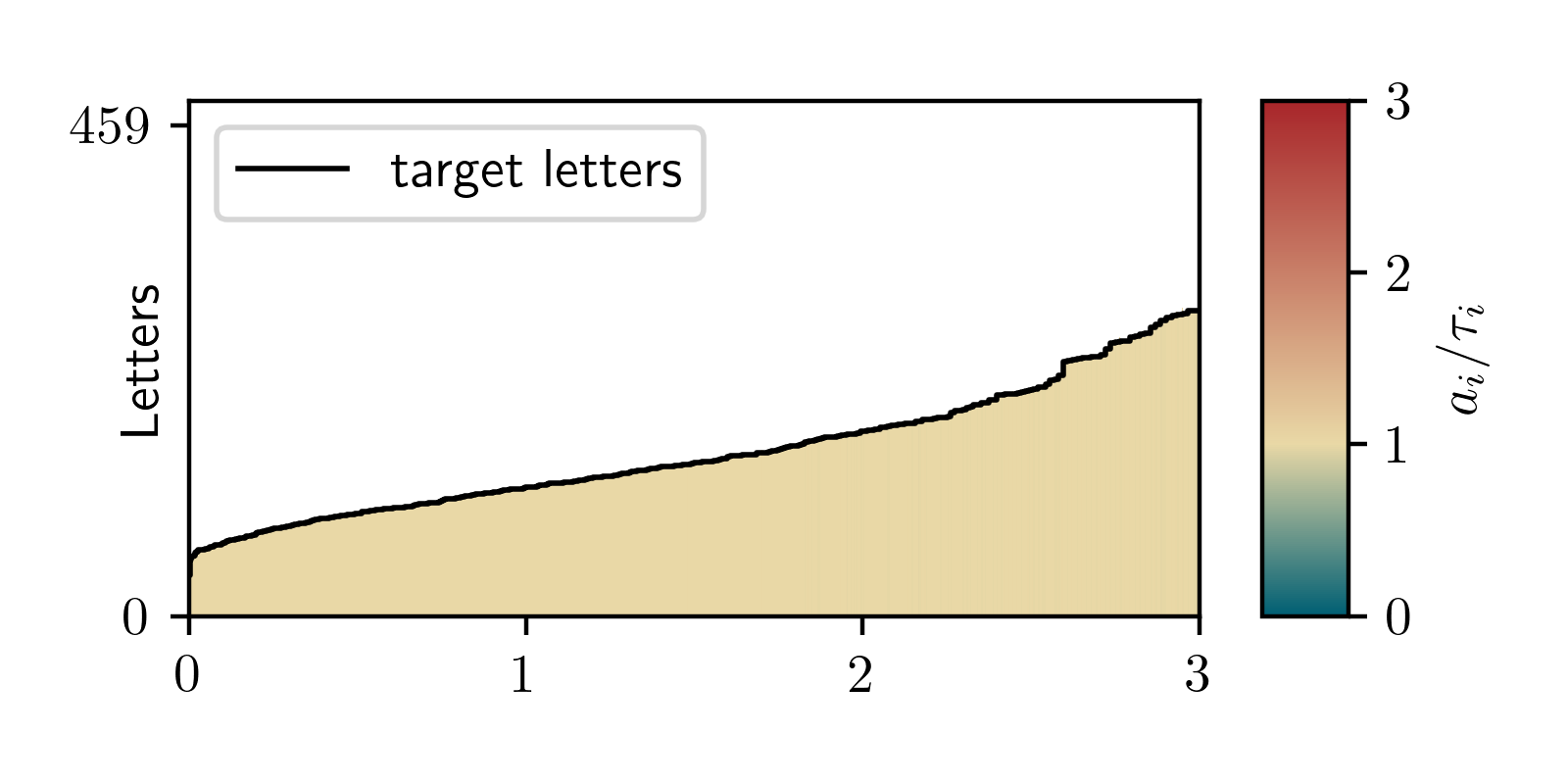}
        \caption{\colgen ($t_G\!=\!3$)}
        \label{fig:results_Sachsen_Small_column_generation}
    \end{subfigure}
    \begin{subfigure}{0.32\textwidth}
        \includegraphics[draft=\draft, width=\linewidth]{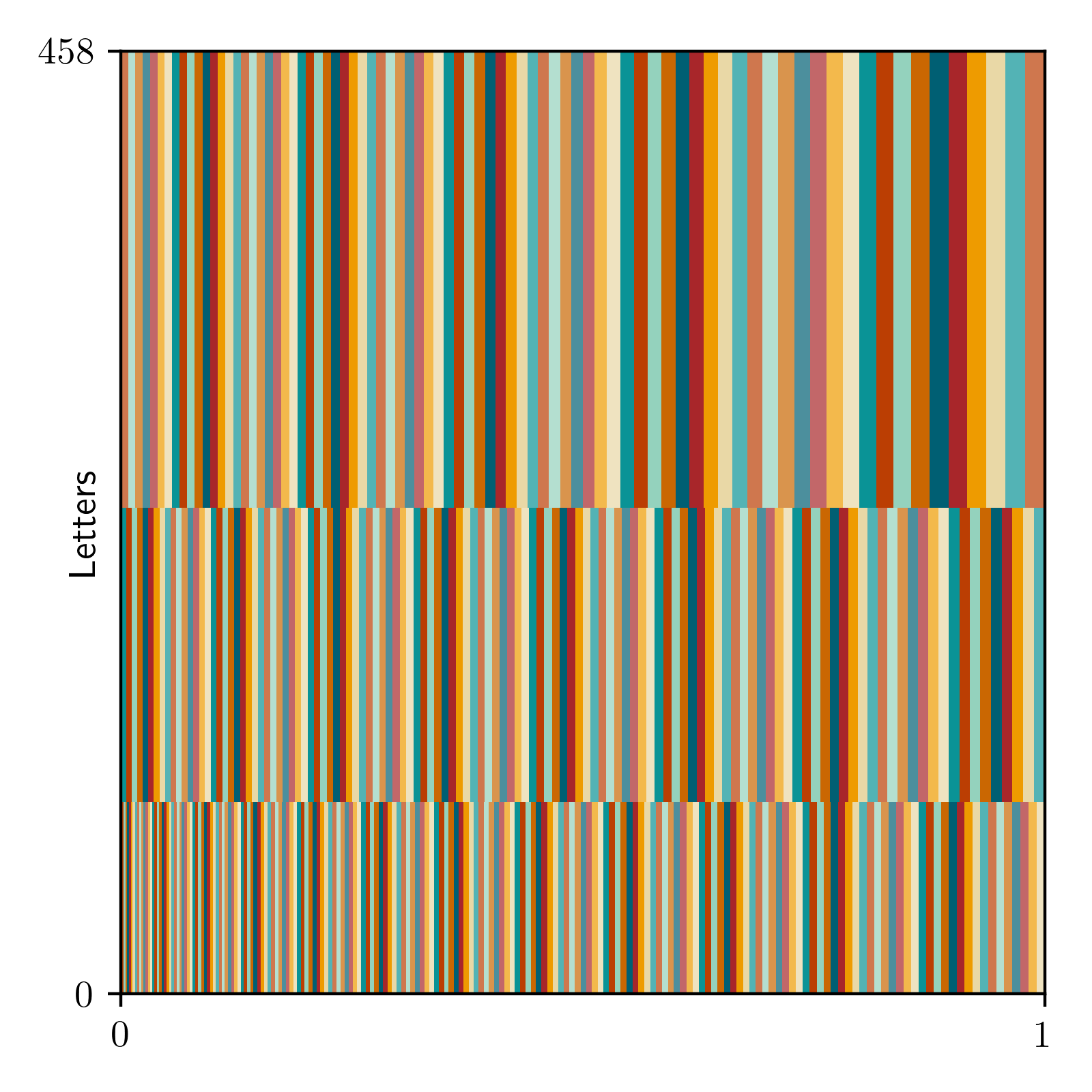}
        \includegraphics[draft=\draft, width=\linewidth]{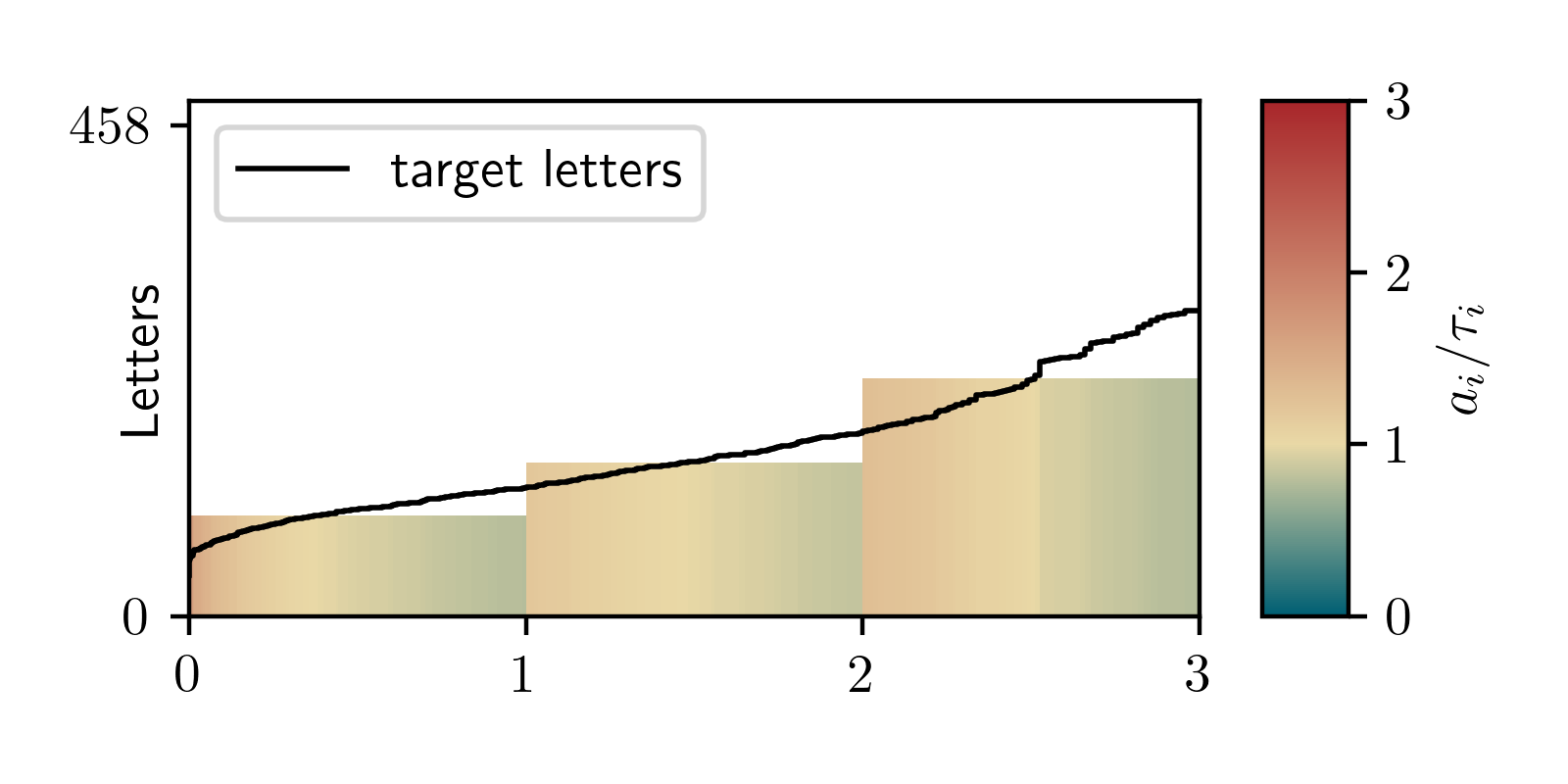}
        \caption{\buckets ($t_G = 3$)}
        \label{fig:results_Sachsen_Small_greedy_bucket_fill}
    \end{subfigure}
    \caption{Small municipalities of Sachsen ($\ell_G = 458$)}
    \label{fig:results_Sachsen_Small}
\end{figure} 

\begin{figure}
    \centering
    \begin{subfigure}{0.32\textwidth}
        \includegraphics[draft=\draft, width=\linewidth]{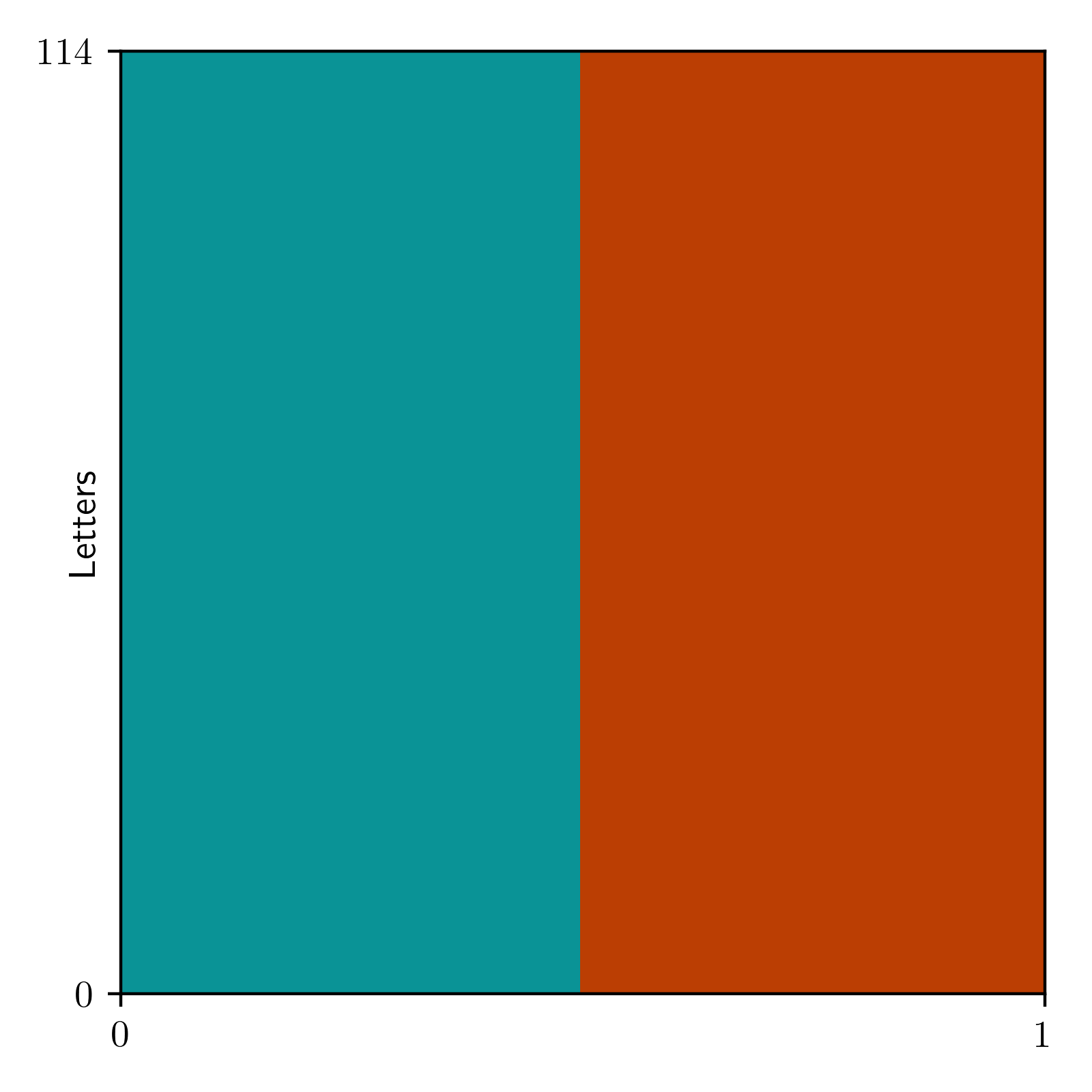}
        \includegraphics[draft=\draft, width=\linewidth]{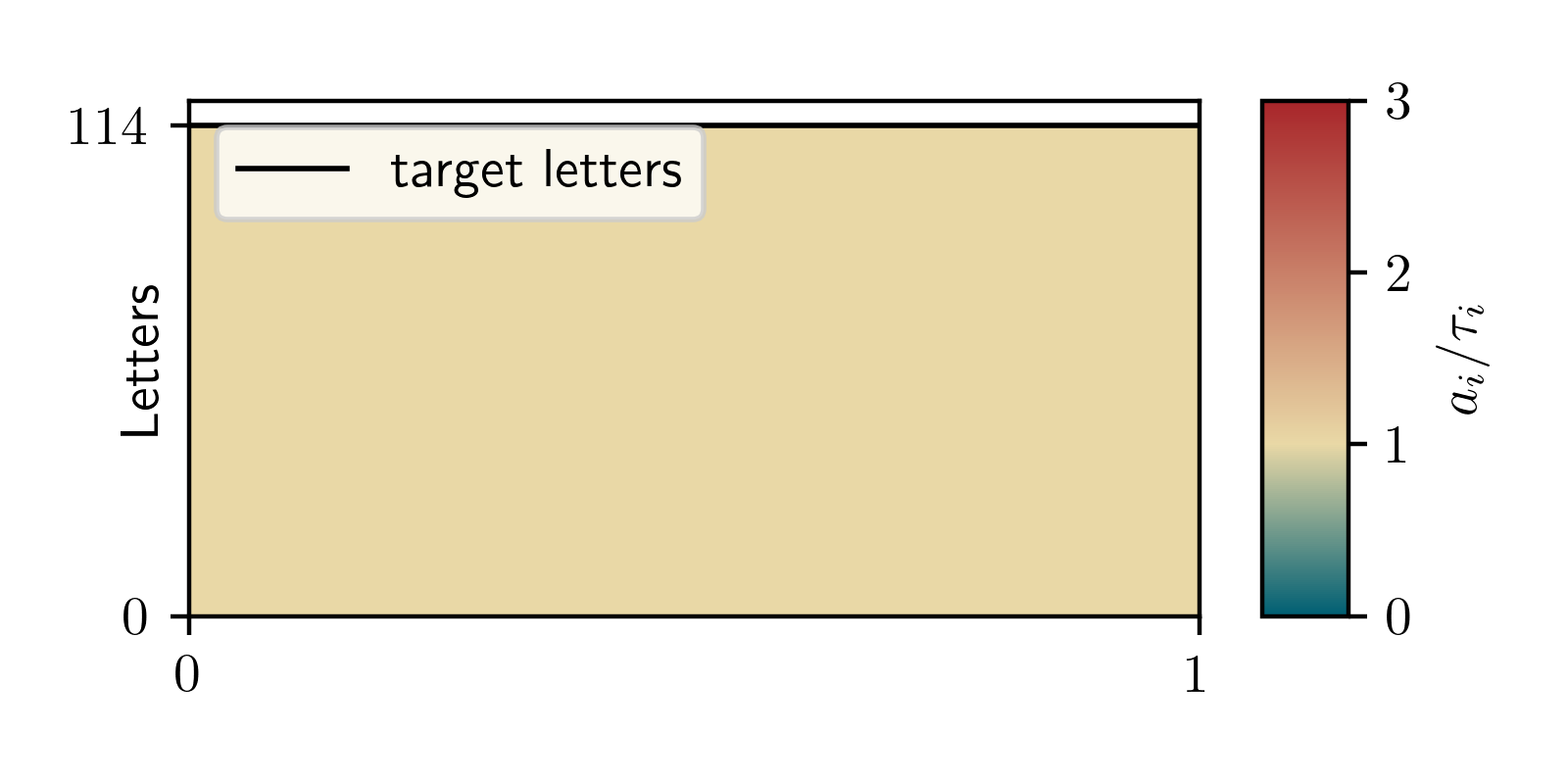}
        \caption{\greq ($t_G = 1$)}
        \label{fig:results_Sachsen-Anhalt_Large_greedy_equal}
    \end{subfigure}
    \begin{subfigure}{0.32\textwidth}
        \includegraphics[draft=\draft, width=\linewidth]{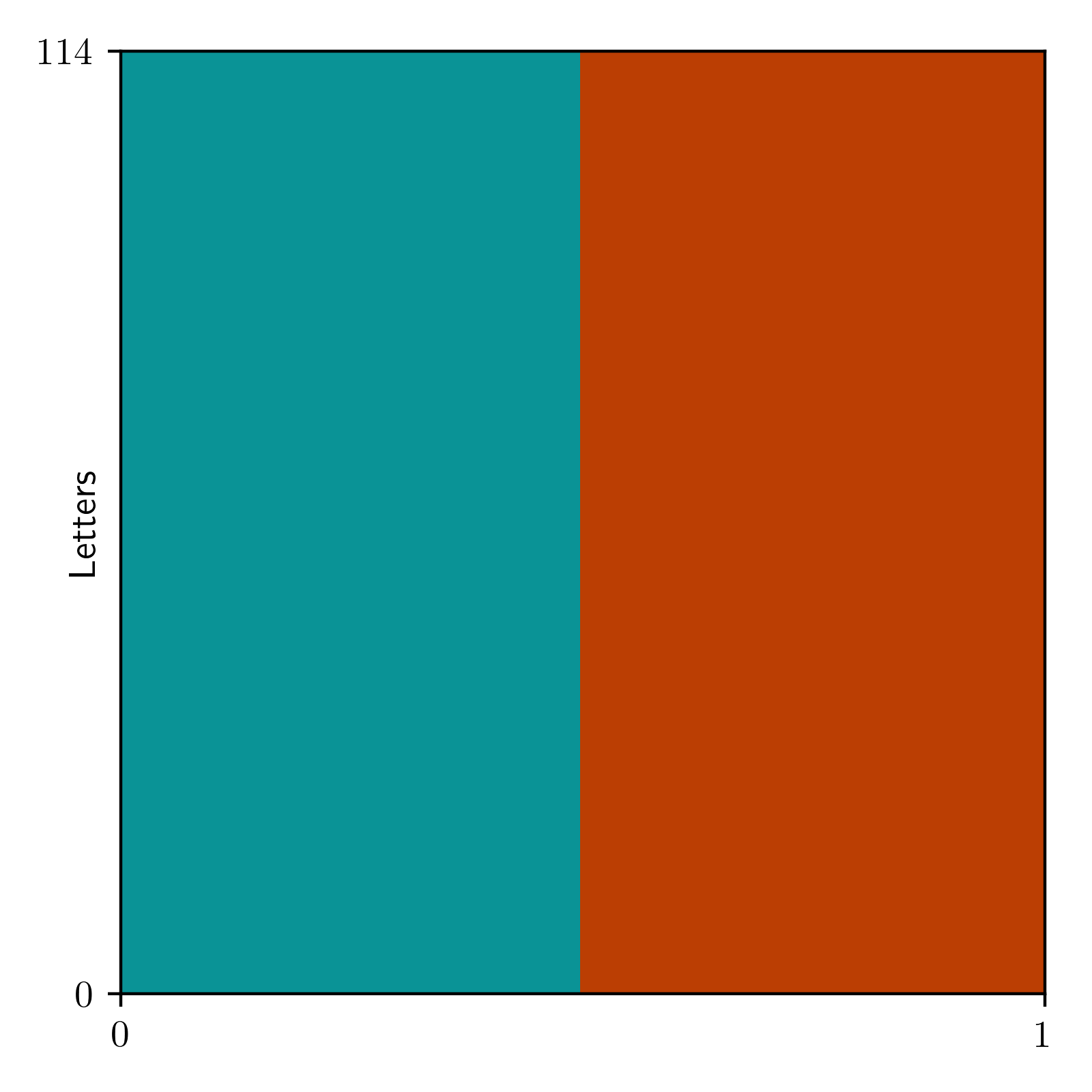}
        \includegraphics[draft=\draft, width=\linewidth]{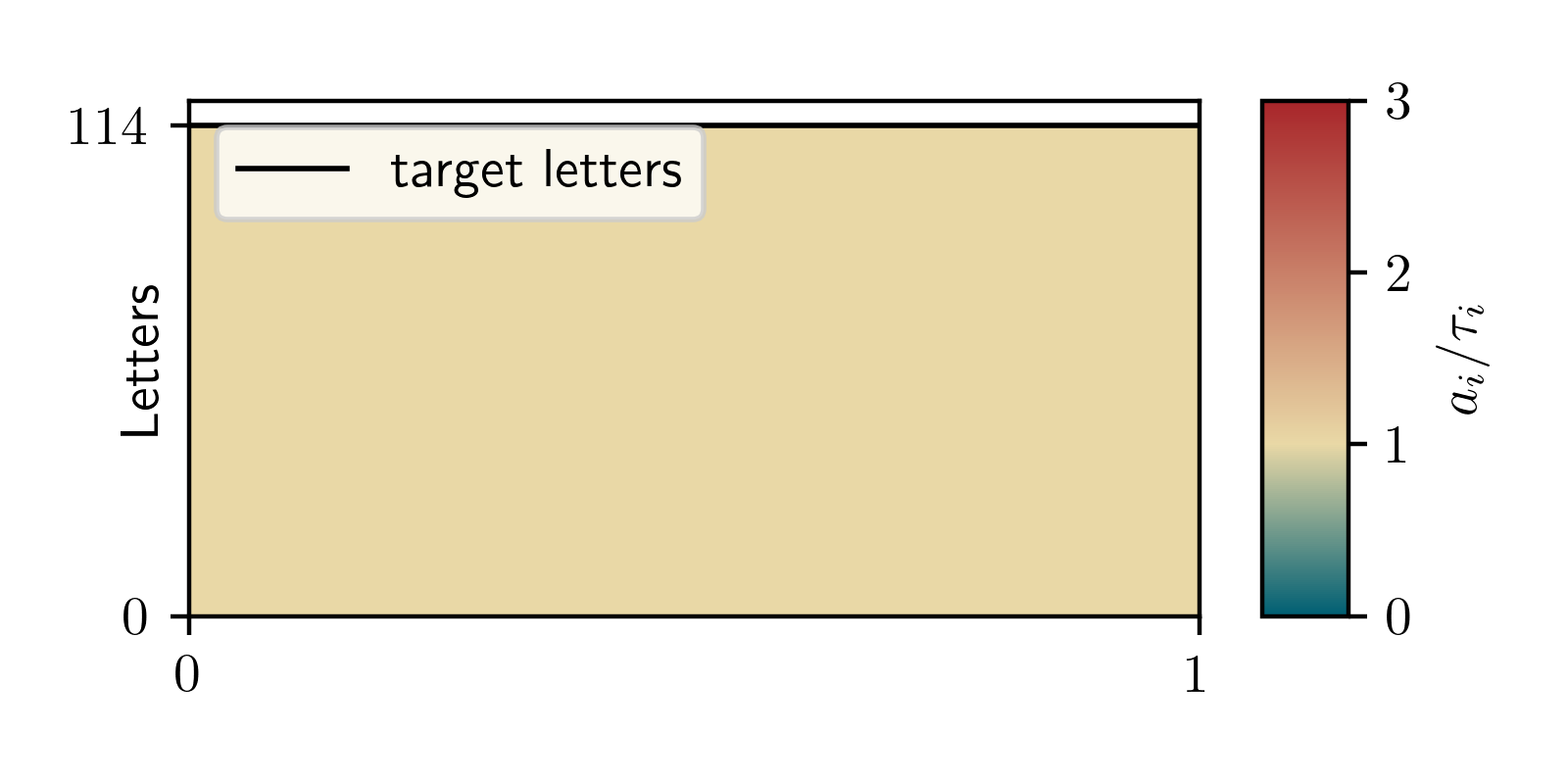}
        \caption{\colgen ($t_G\!=\!1$)}
        \label{fig:results_Sachsen-Anhalt_Large_column_generation}
    \end{subfigure}
    \begin{subfigure}{0.32\textwidth}
        \includegraphics[draft=\draft, width=\linewidth]{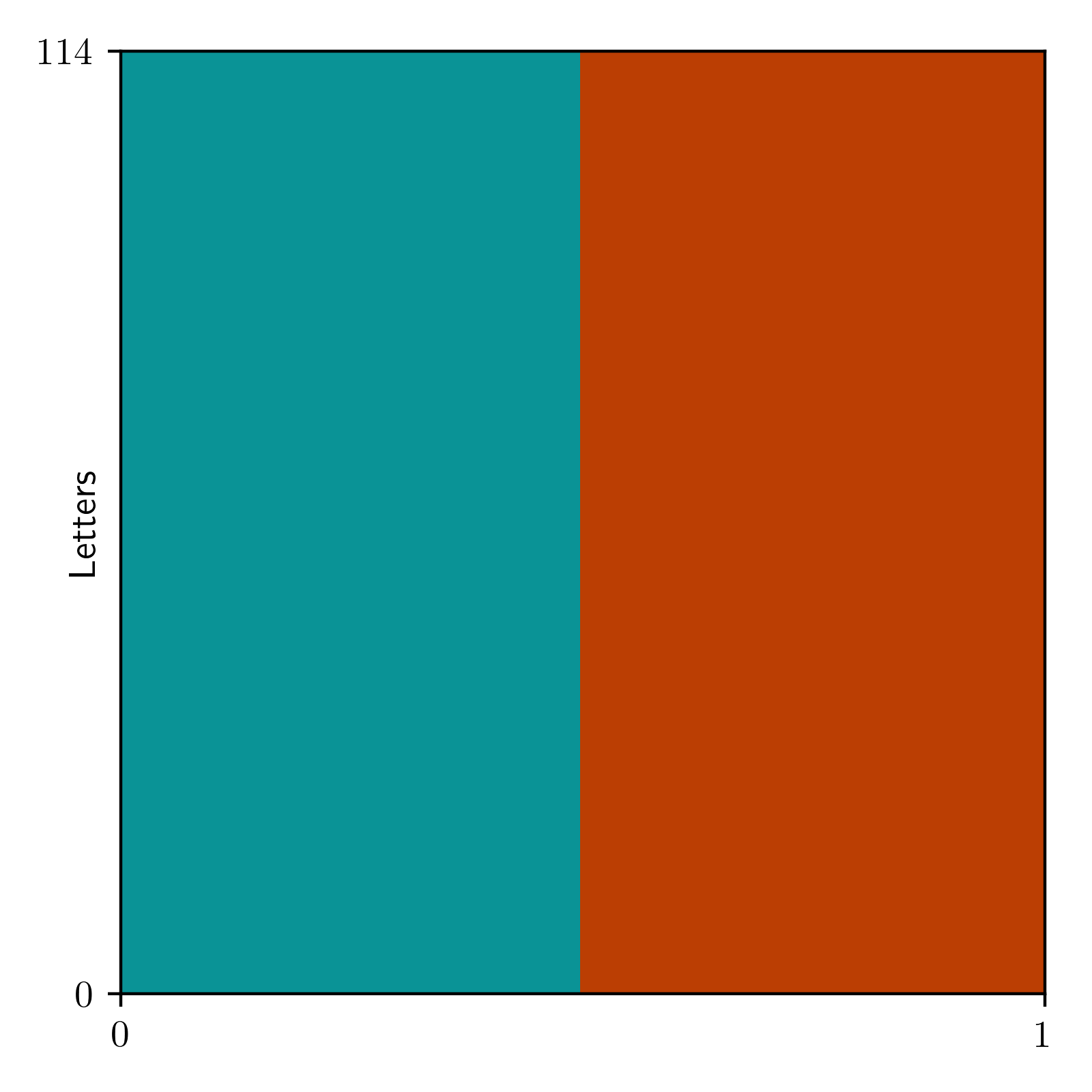}
        \includegraphics[draft=\draft, width=\linewidth]{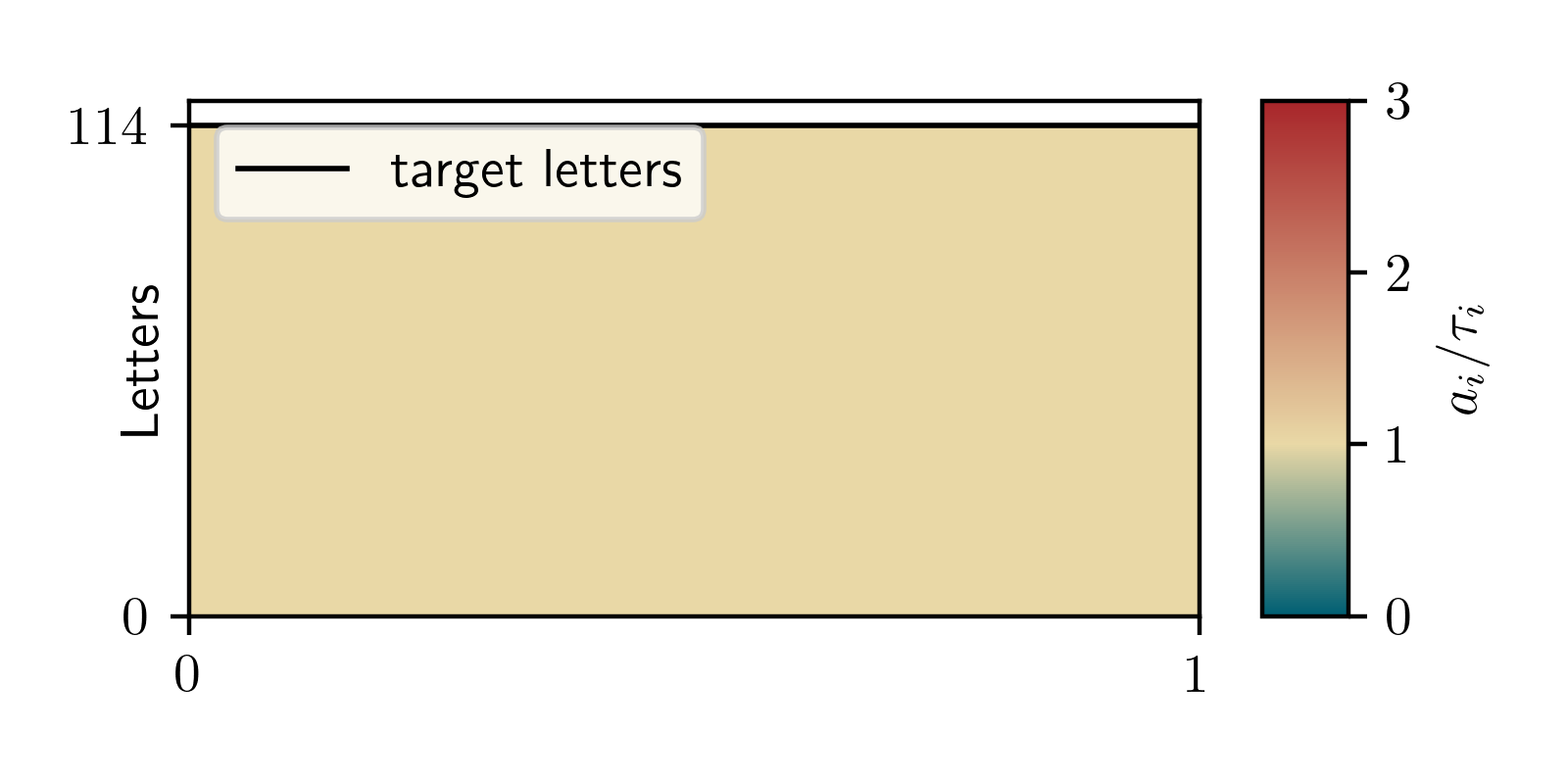}
        \caption{\buckets ($t_G = 1$)}
        \label{fig:results_Sachsen-Anhalt_Large_greedy_bucket_fill}
    \end{subfigure}
    \caption{Large municipalities of Sachsen-Anhalt ($\ell_G = 114$)}
    \label{fig:results_Sachsen-Anhalt_Large}
\end{figure} 

\begin{figure}
    \centering
    \begin{subfigure}{0.32\textwidth}
        \includegraphics[draft=\draft, width=\linewidth]{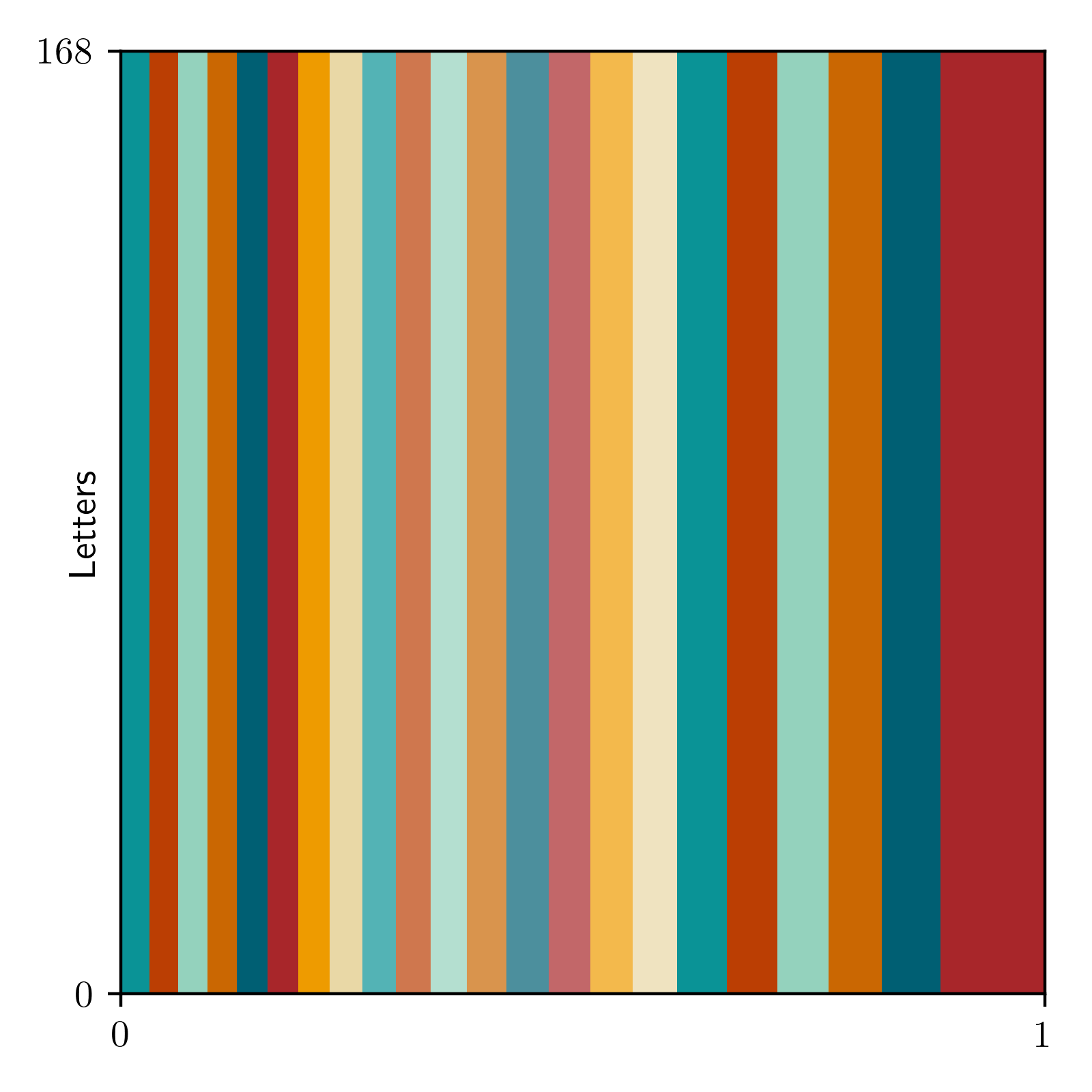}
        \includegraphics[draft=\draft, width=\linewidth]{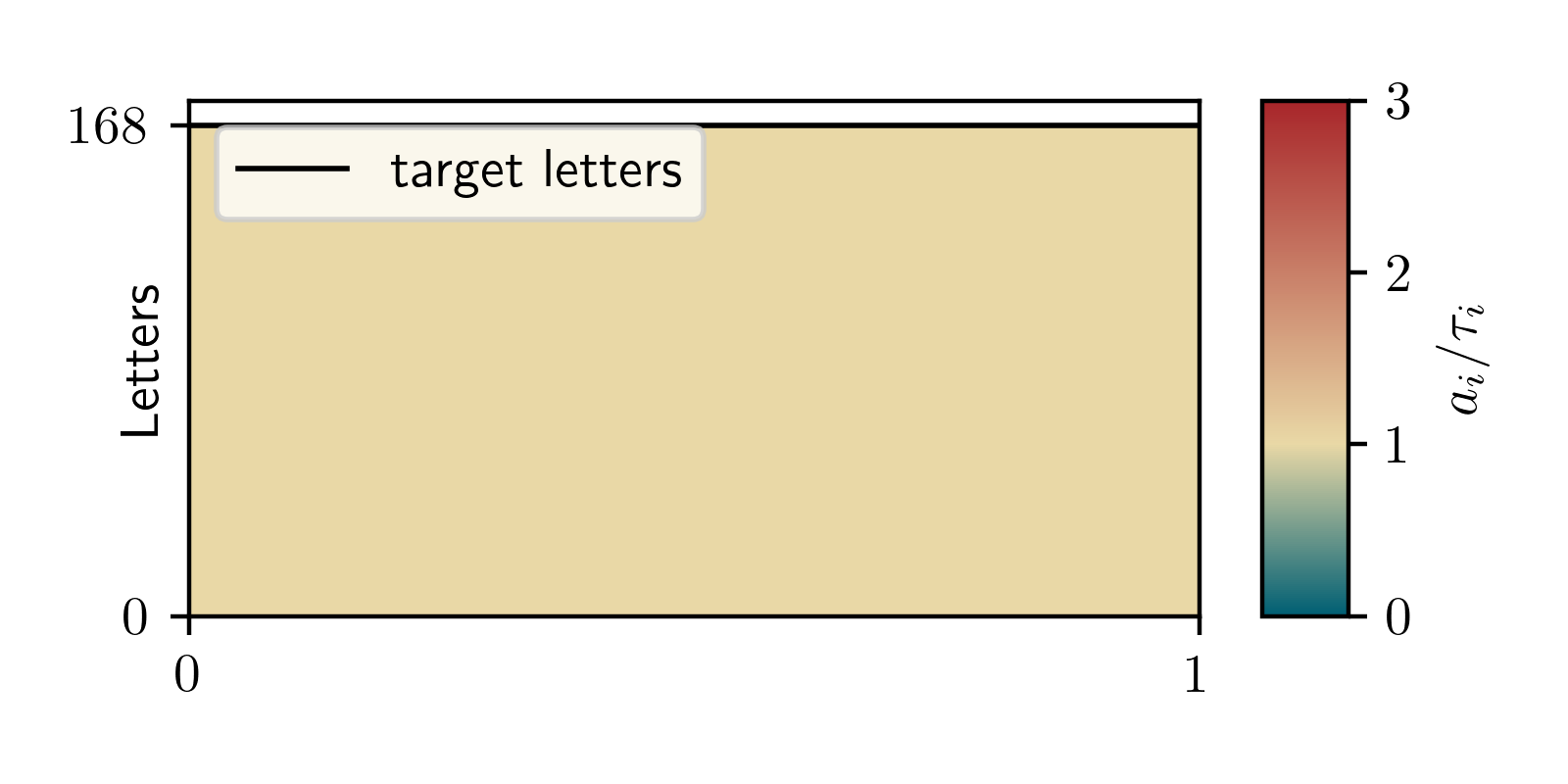}
        \caption{\greq ($t_G = 1$)}
        \label{fig:results_Sachsen-Anhalt_Medium_greedy_equal}
    \end{subfigure}
    \begin{subfigure}{0.32\textwidth}
        \includegraphics[draft=\draft, width=\linewidth]{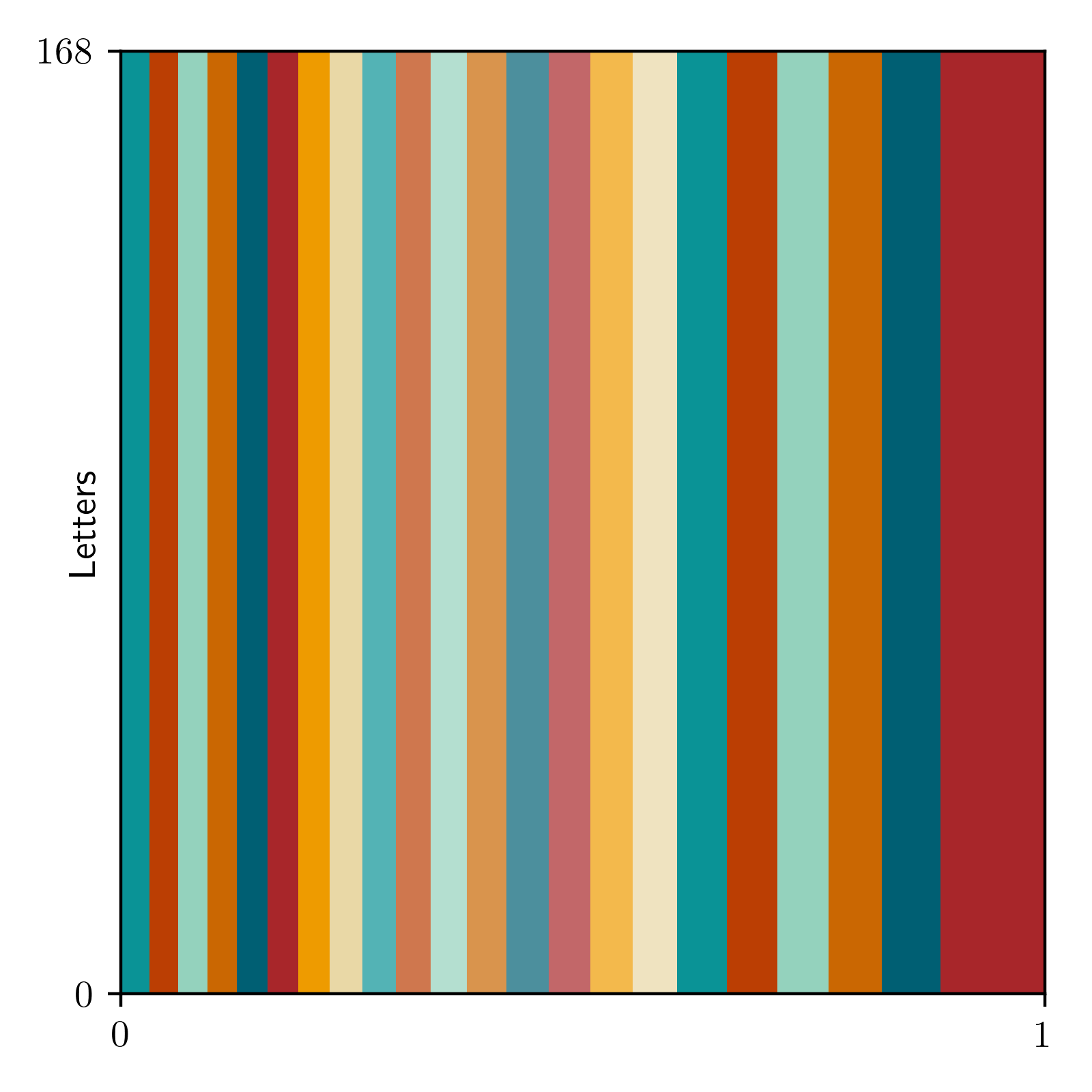}
        \includegraphics[draft=\draft, width=\linewidth]{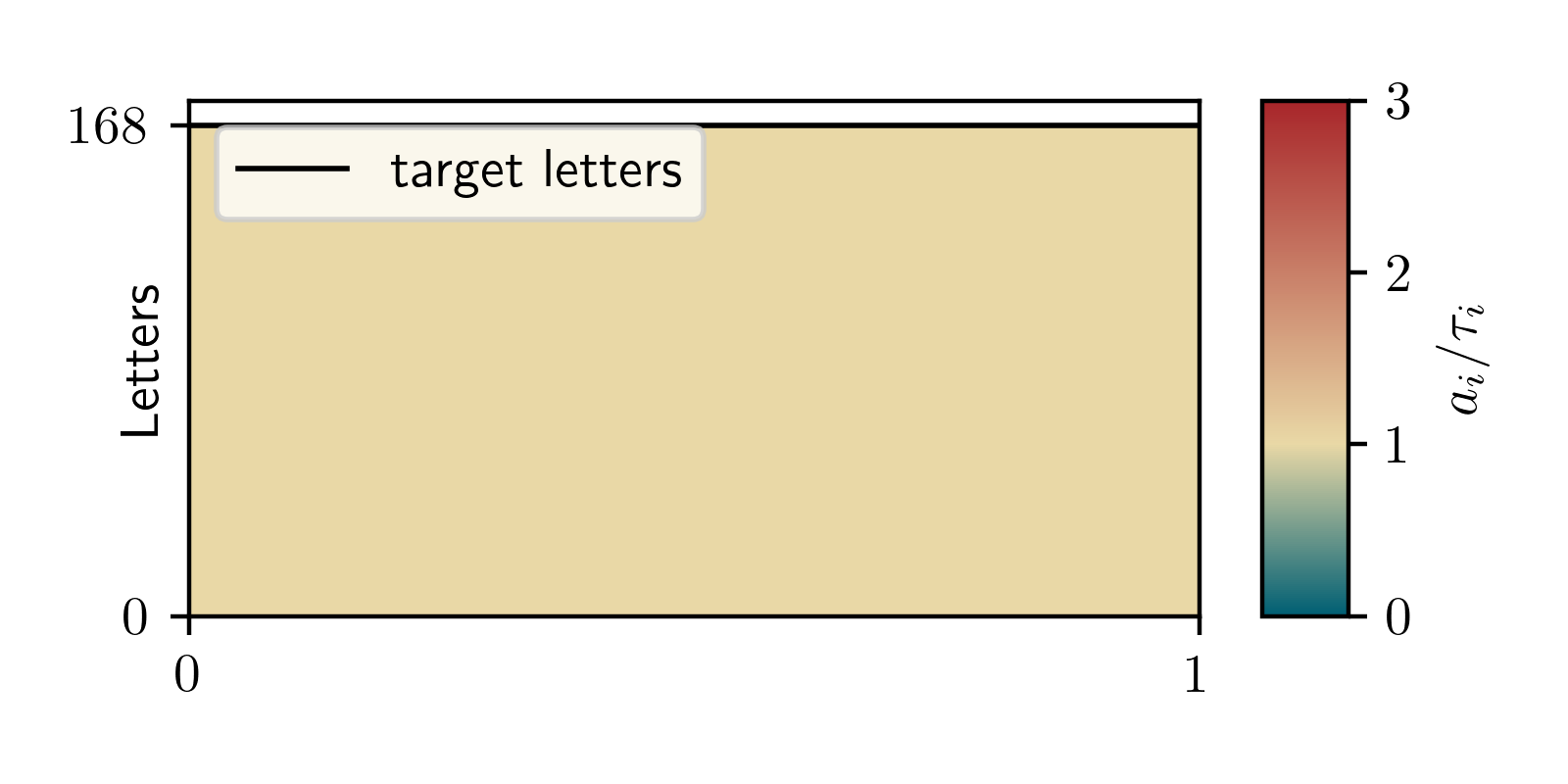}
        \caption{\colgen ($t_G\!=\!1$)}
        \label{fig:results_Sachsen-Anhalt_Medium_column_generation}
    \end{subfigure}
    \begin{subfigure}{0.32\textwidth}
        \includegraphics[draft=\draft, width=\linewidth]{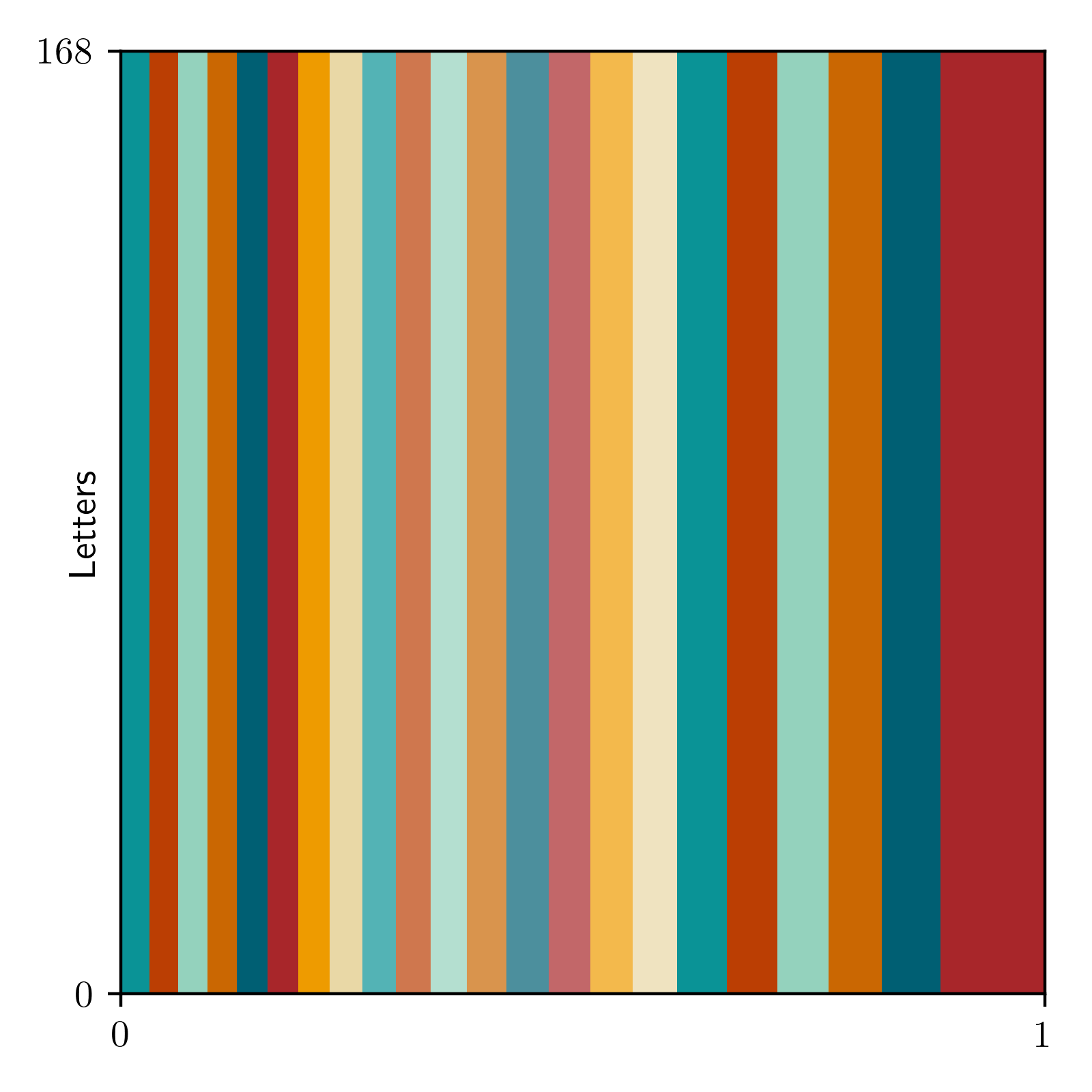}
        \includegraphics[draft=\draft, width=\linewidth]{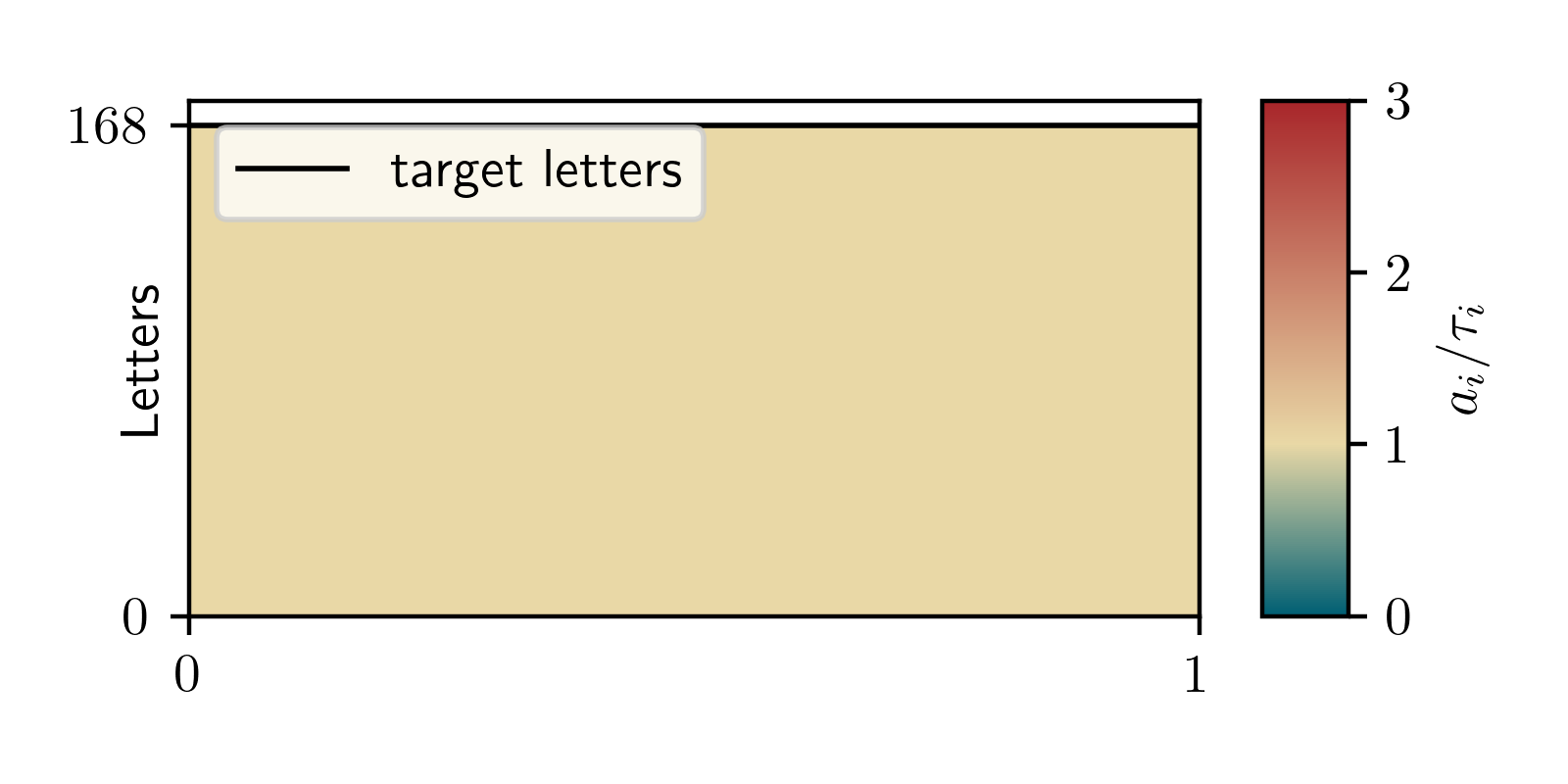}
        \caption{\buckets ($t_G = 1$)}
        \label{fig:results_Sachsen-Anhalt_Medium_greedy_bucket_fill}
    \end{subfigure}
    \caption{Medium municipalities of Sachsen-Anhalt ($\ell_G = 168$)}
    \label{fig:results_Sachsen-Anhalt_Medium}
\end{figure} 

\begin{figure}
    \centering
    \begin{subfigure}{0.32\textwidth}
        \includegraphics[draft=\draft, width=\linewidth]{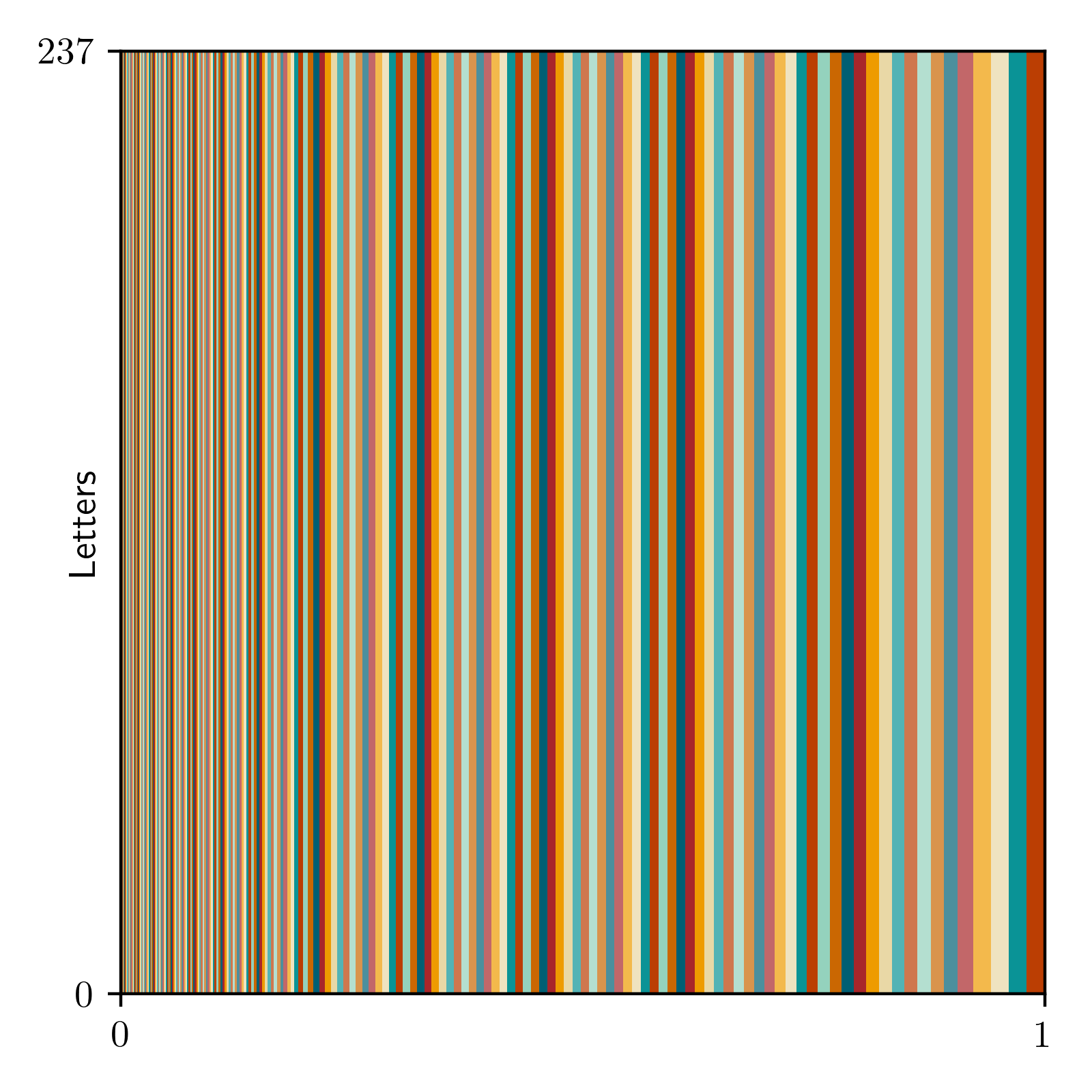}
        \includegraphics[draft=\draft, width=\linewidth]{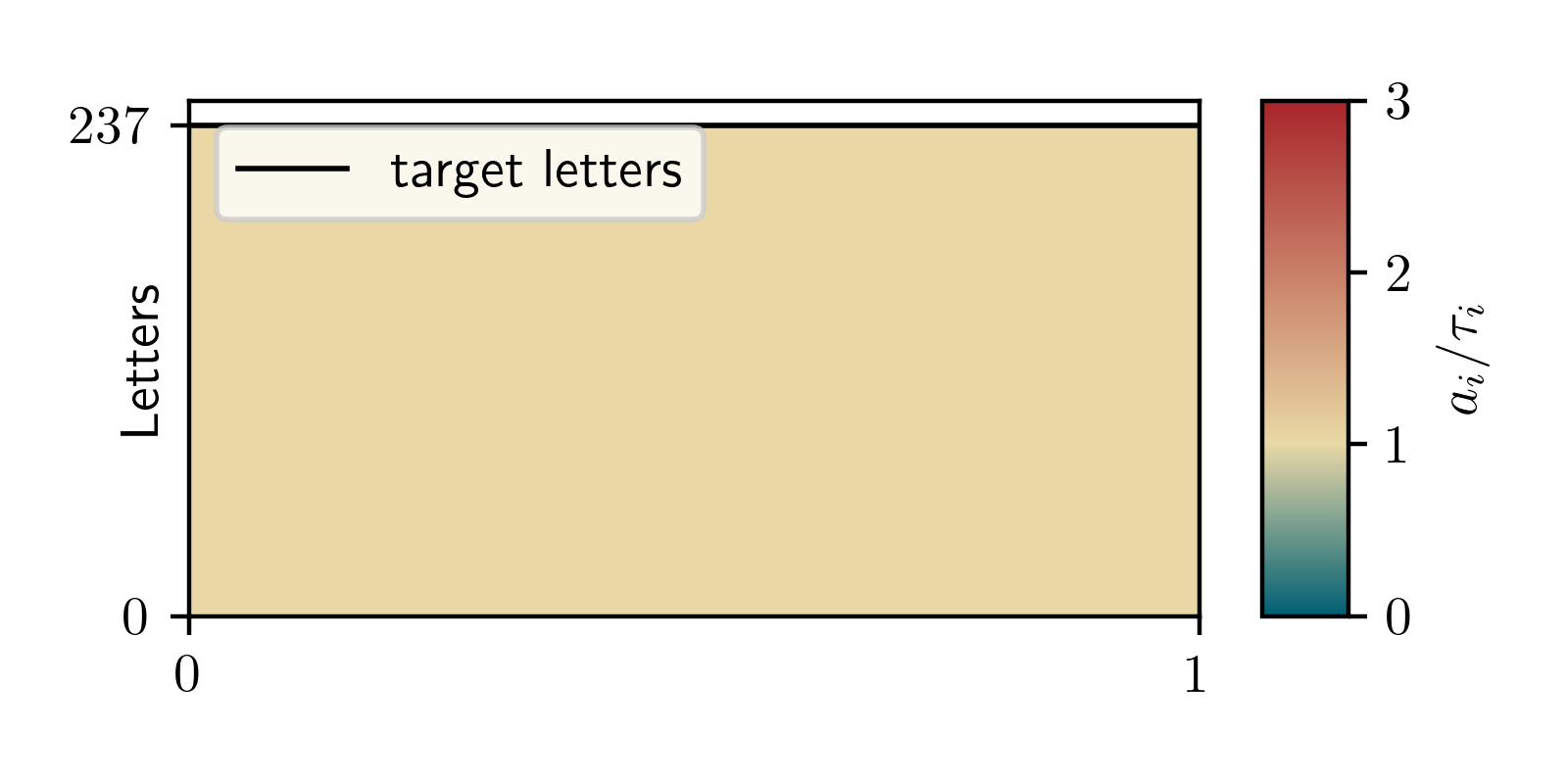}
        \caption{\greq ($t_G = 1$)}
        \label{fig:results_Sachsen-Anhalt_Small_greedy_equal}
    \end{subfigure}
    \begin{subfigure}{0.32\textwidth}
        \includegraphics[draft=\draft, width=\linewidth]{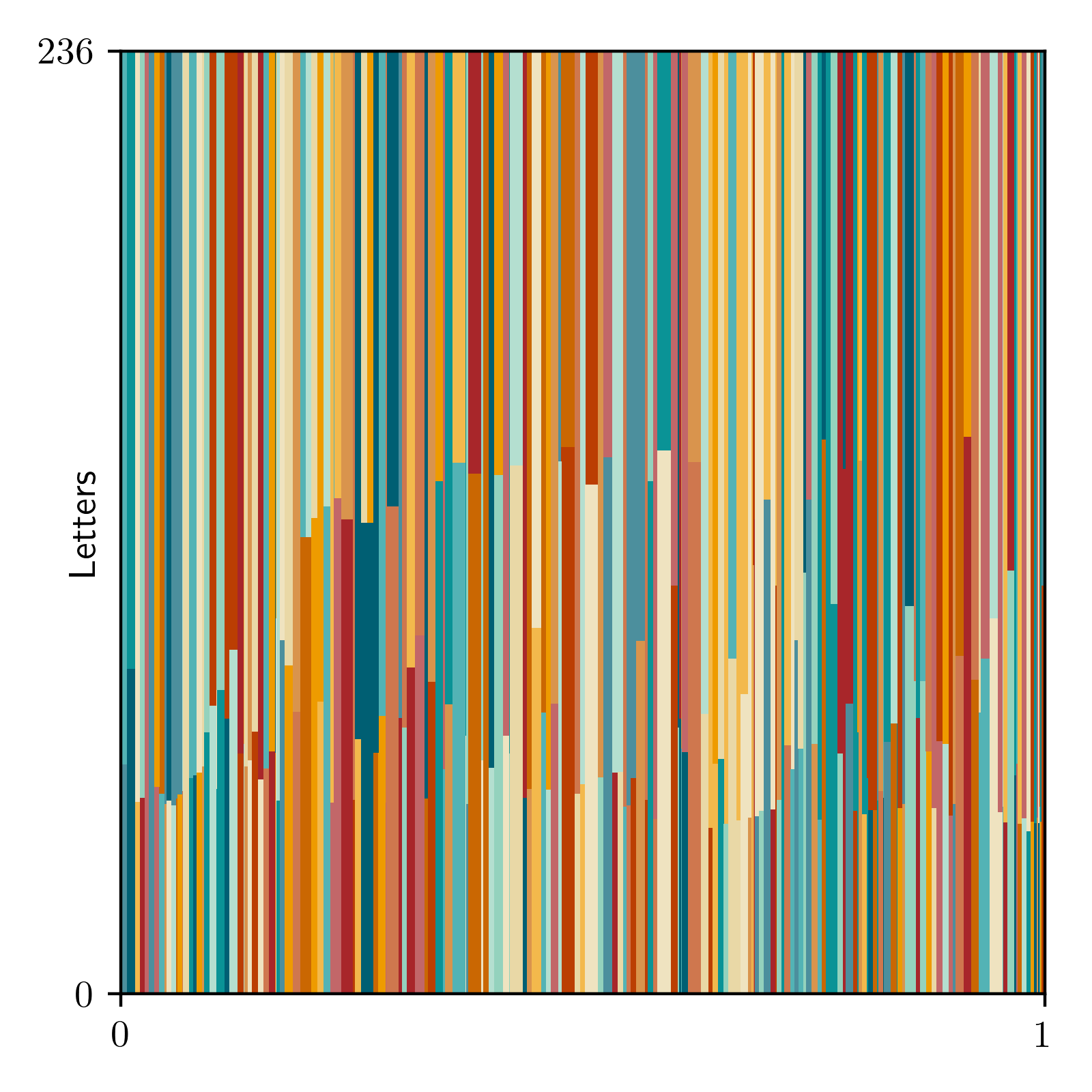}
        \includegraphics[draft=\draft, width=\linewidth]{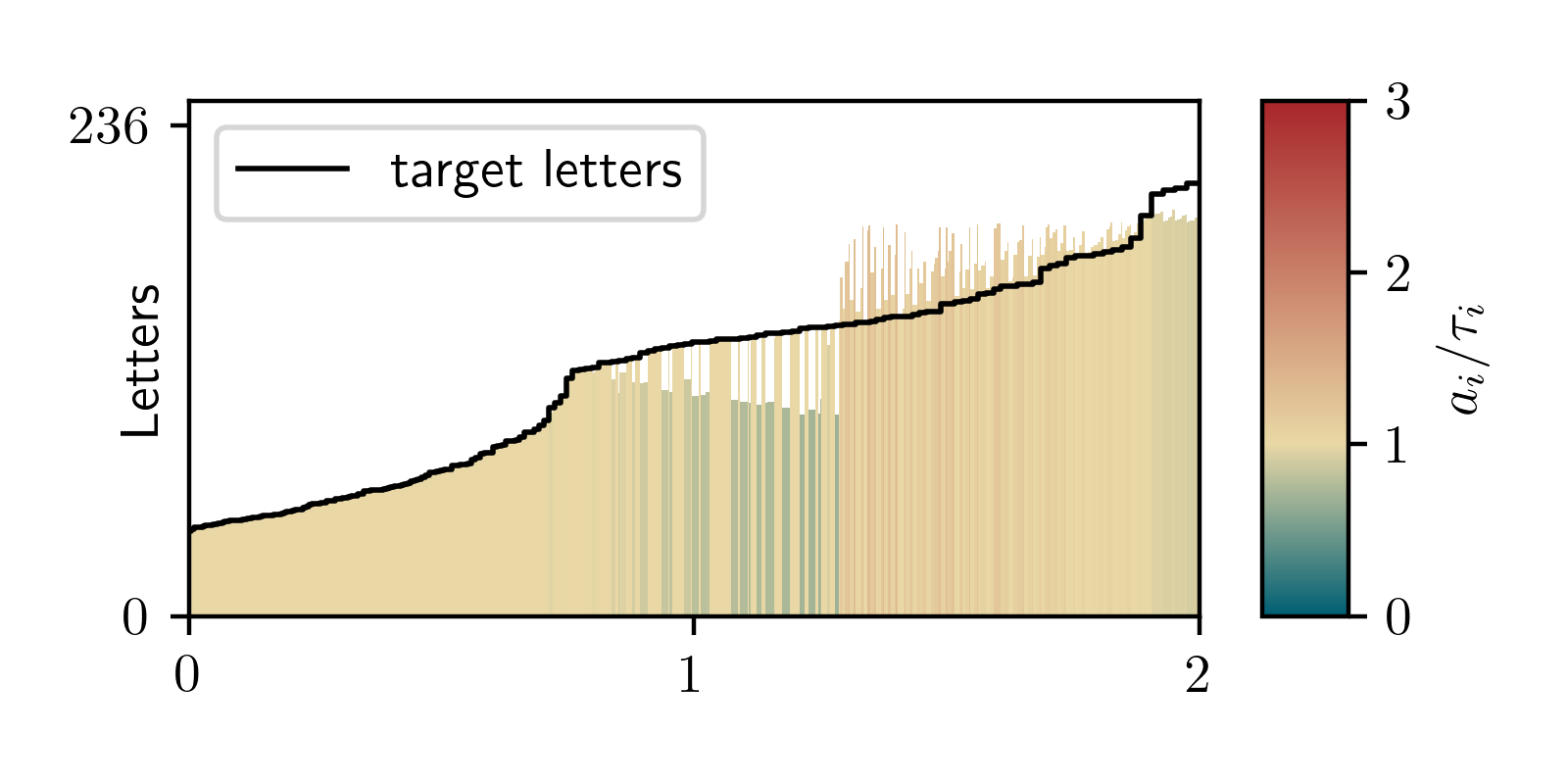}
        \caption{\colgen ($t_G\!=\!2$)}
        \label{fig:results_Sachsen-Anhalt_Small_column_generation}
    \end{subfigure}
    \begin{subfigure}{0.32\textwidth}
        \includegraphics[draft=\draft, width=\linewidth]{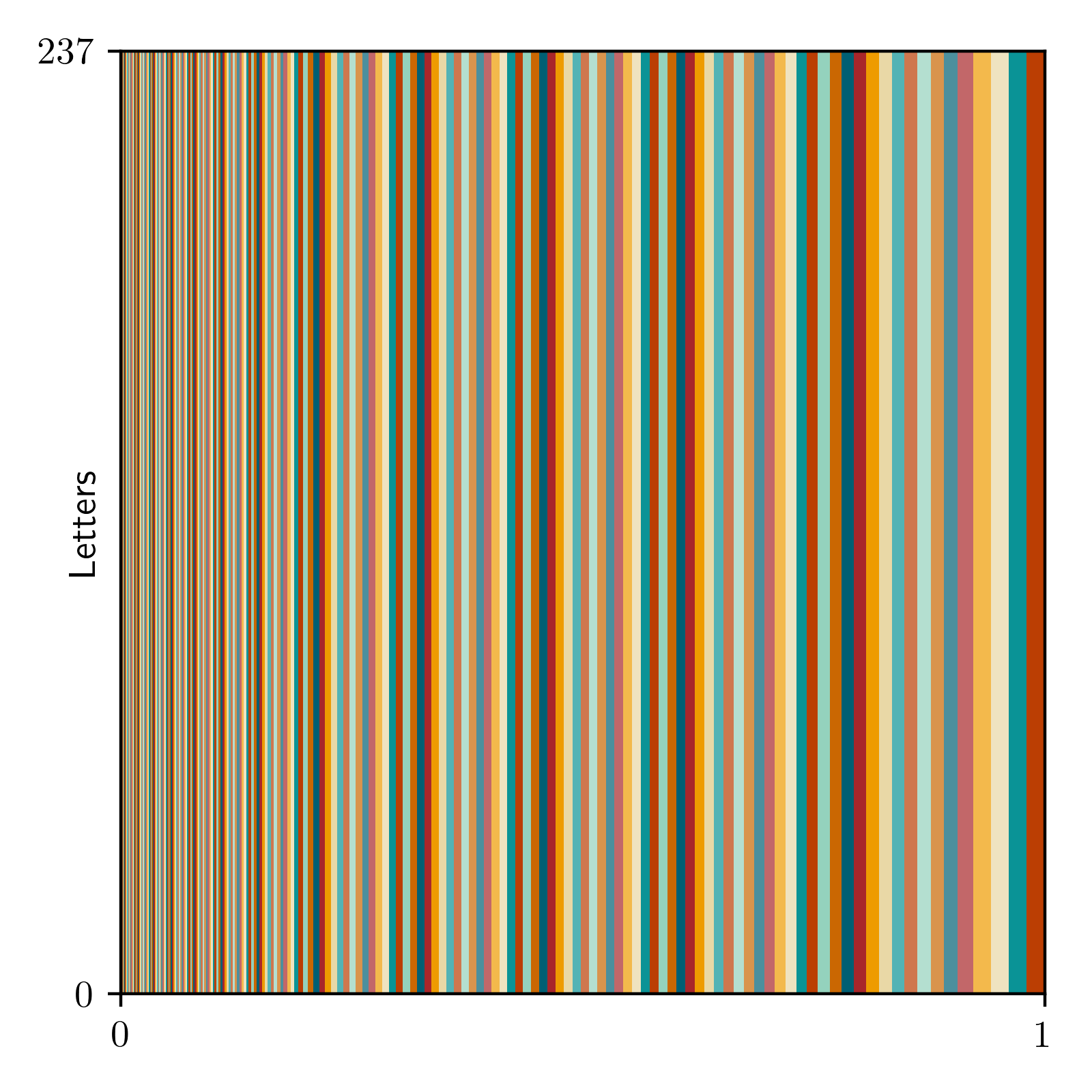}
        \includegraphics[draft=\draft, width=\linewidth]{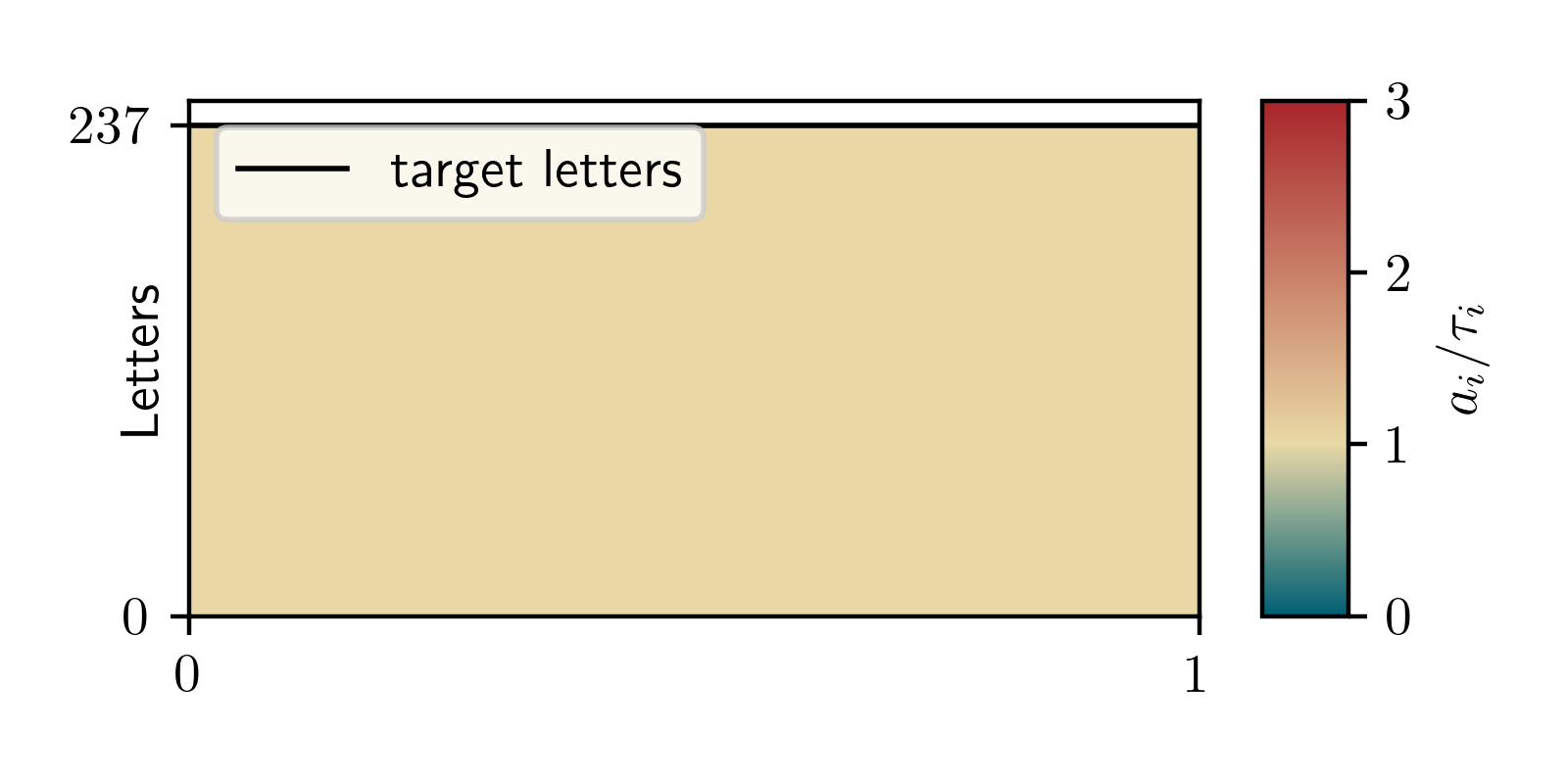}
        \caption{\buckets ($t_G = 1$)}
        \label{fig:results_Sachsen-Anhalt_Small_greedy_bucket_fill}
    \end{subfigure}
    \caption{Small municipalities of Sachsen-Anhalt ($\ell_G = 237$)}
    \label{fig:results_Sachsen-Anhalt_Small}
\end{figure} 

\begin{figure}
    \centering
    \begin{subfigure}{0.32\textwidth}
        \includegraphics[draft=\draft, width=\linewidth]{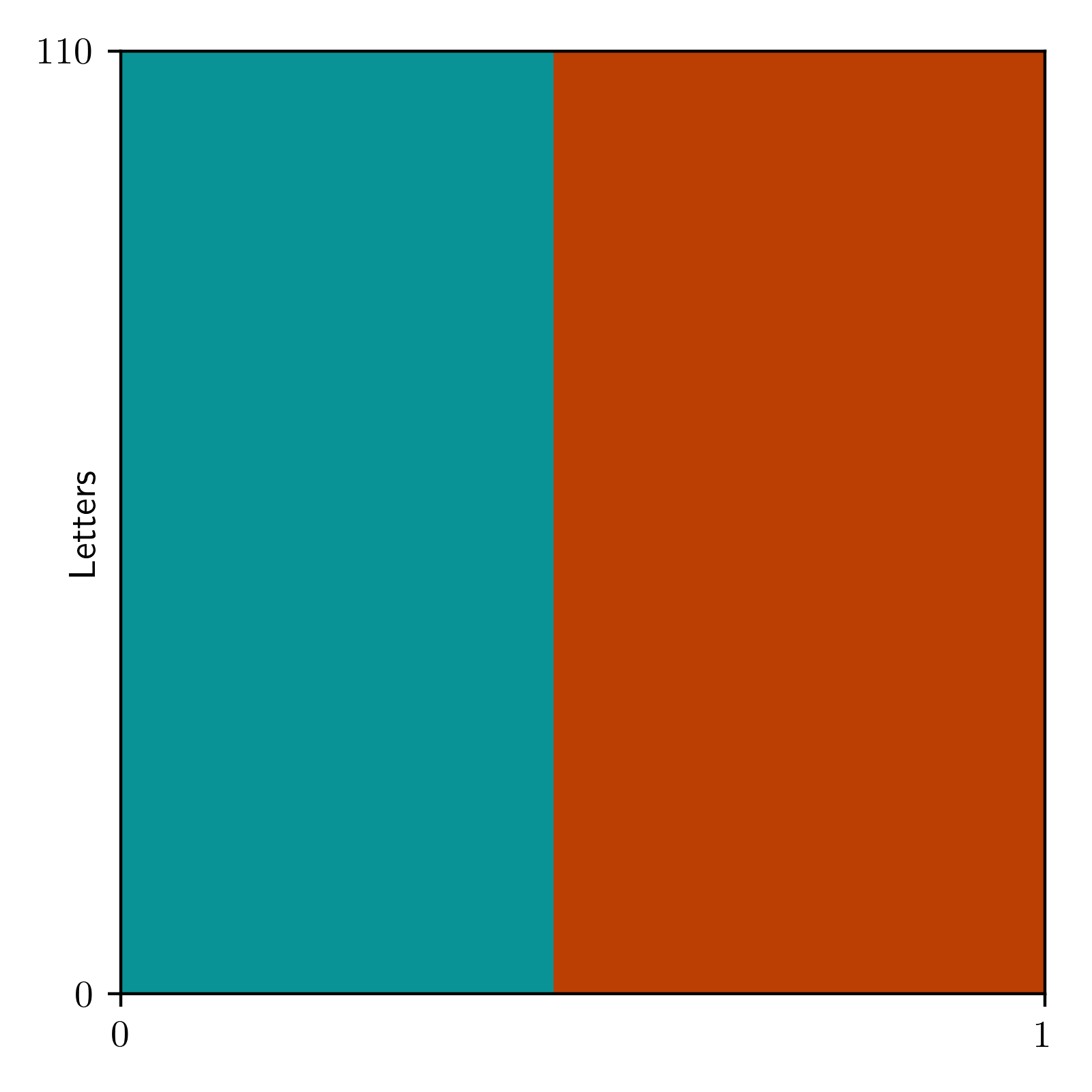}
        \includegraphics[draft=\draft, width=\linewidth]{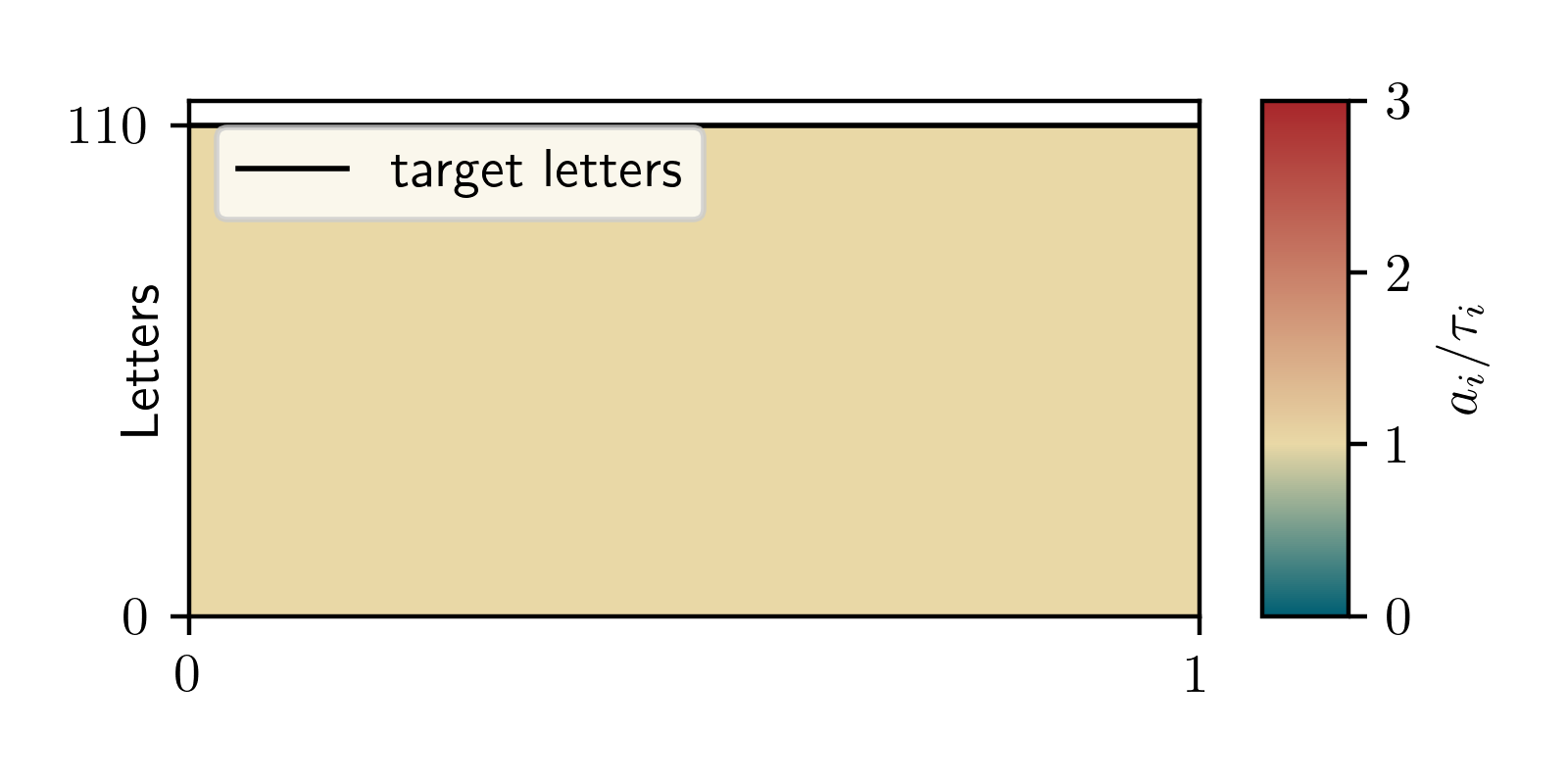}
        \caption{\greq ($t_G = 1$)}
        \label{fig:results_Schleswig-Holstein_Large_greedy_equal}
    \end{subfigure}
    \begin{subfigure}{0.32\textwidth}
        \includegraphics[draft=\draft, width=\linewidth]{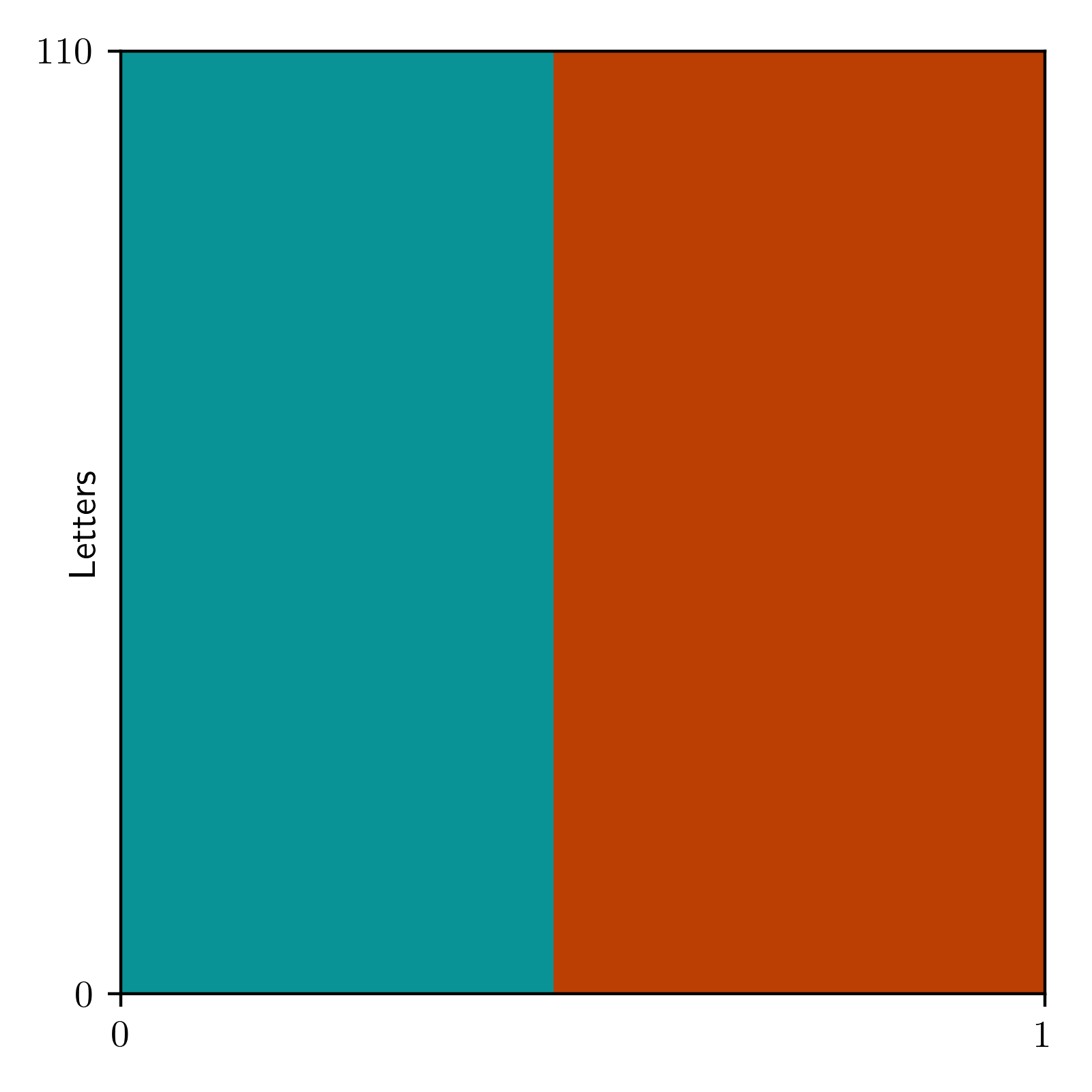}
        \includegraphics[draft=\draft, width=\linewidth]{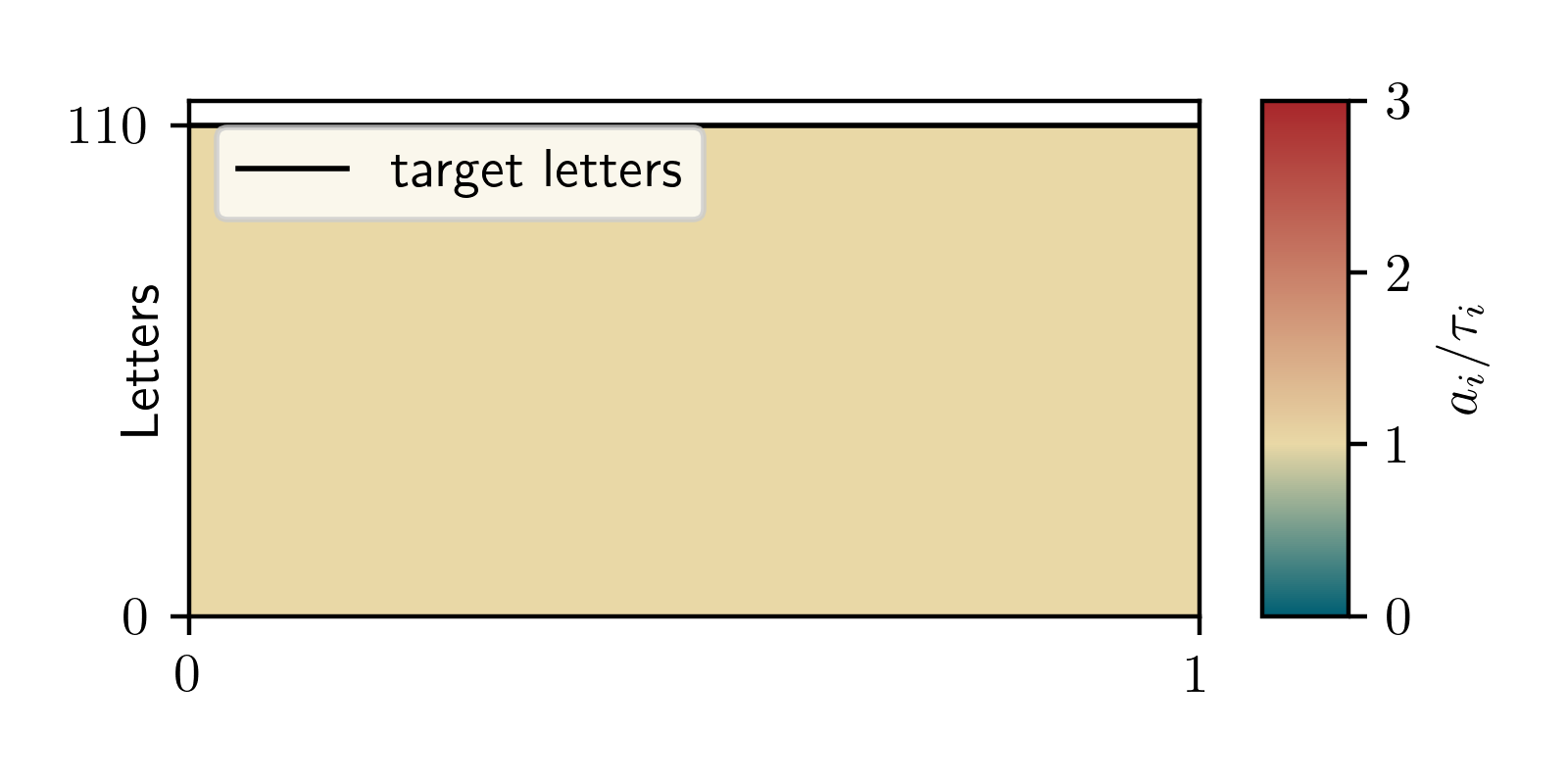}
        \caption{\colgen ($t_G\!=\!1$)}
        \label{fig:results_Schleswig-Holstein_Large_column_generation}
    \end{subfigure}
    \begin{subfigure}{0.32\textwidth}
        \includegraphics[draft=\draft, width=\linewidth]{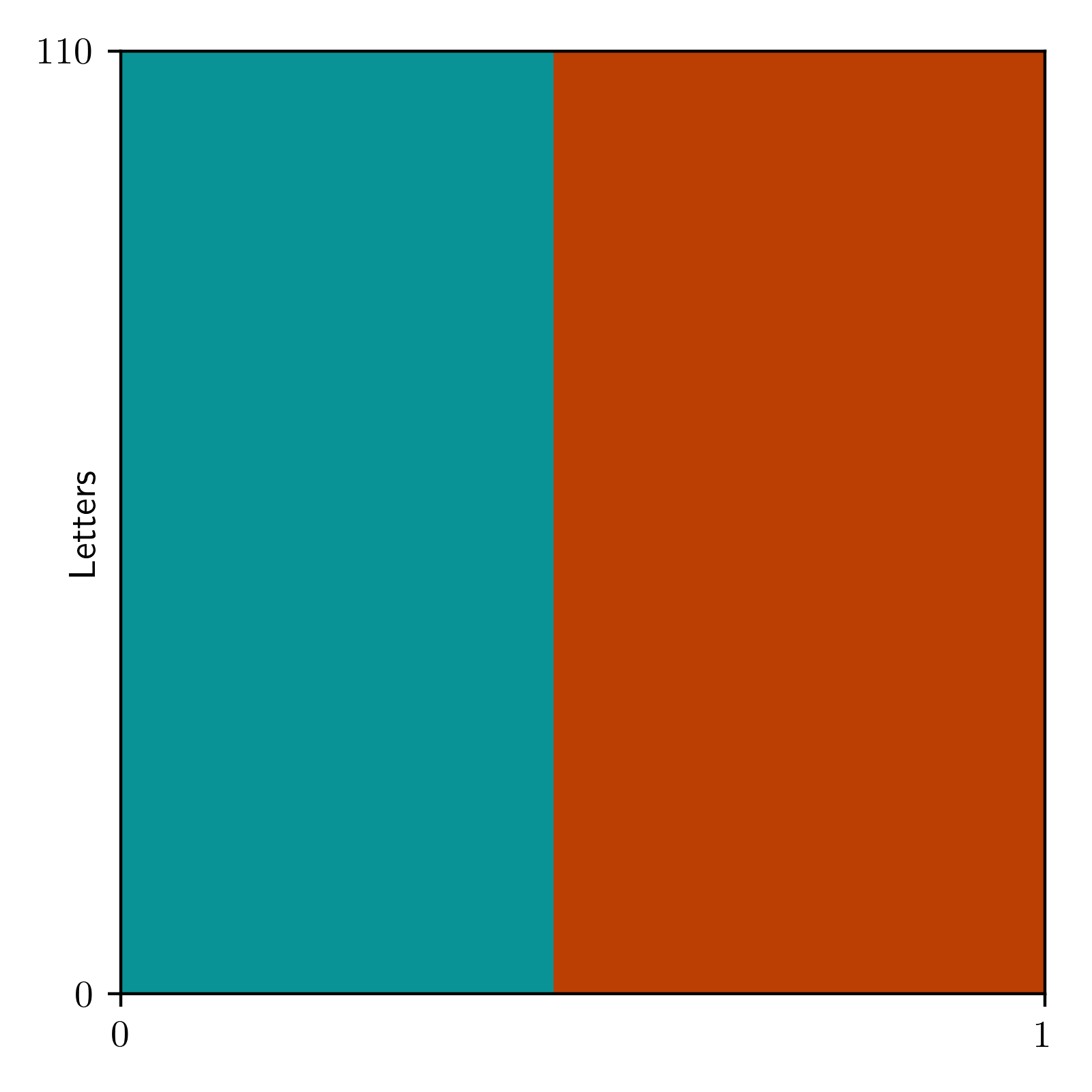}
        \includegraphics[draft=\draft, width=\linewidth]{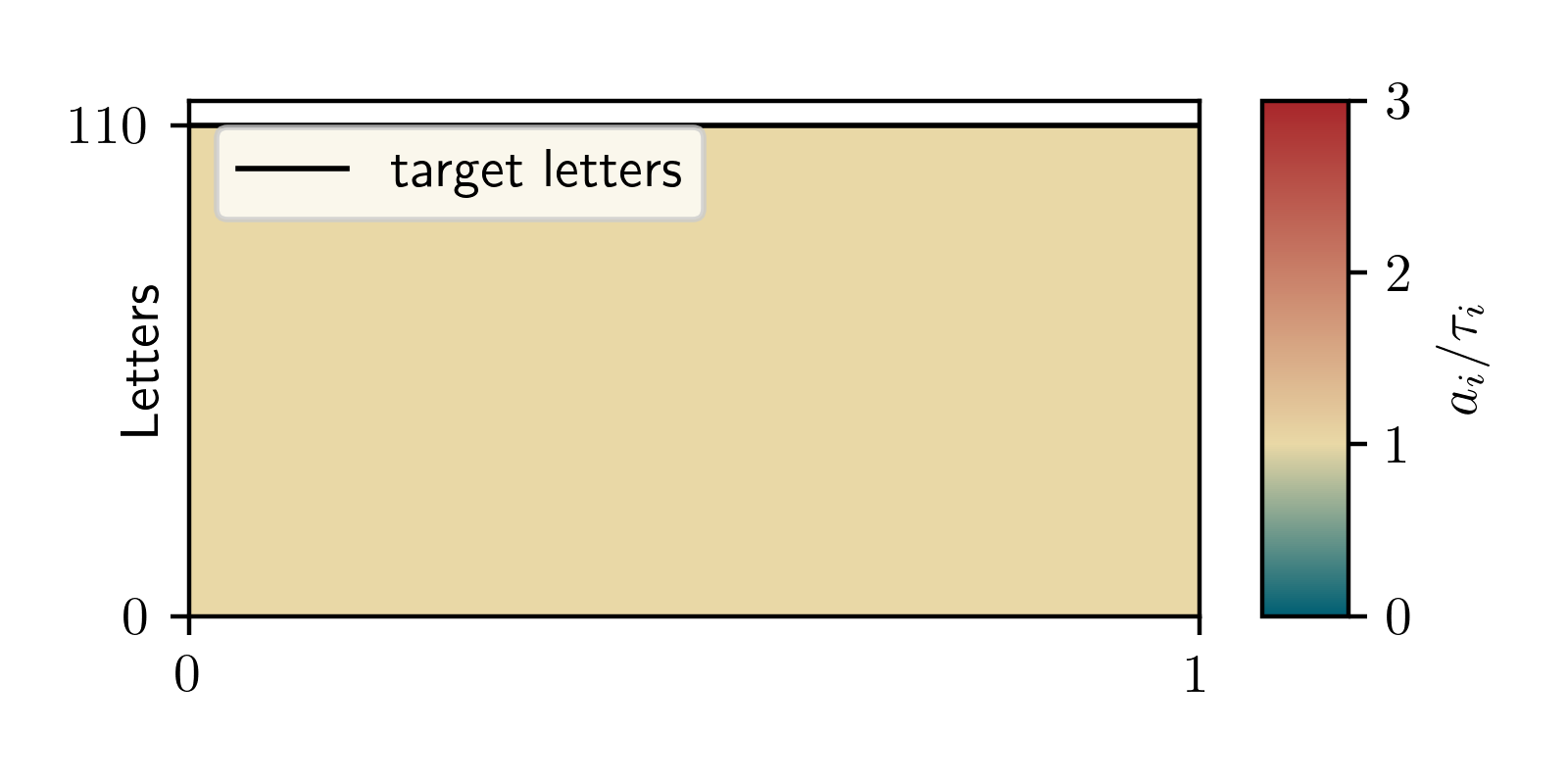}
        \caption{\buckets ($t_G = 1$)}
        \label{fig:results_Schleswig-Holstein_Large_greedy_bucket_fill}
    \end{subfigure}
    \caption{Large municipalities of Schleswig-Holstein ($\ell_G = 110$)}
    \label{fig:results_Schleswig-Holstein_Large}
\end{figure} 

\begin{figure}
    \centering
    \begin{subfigure}{0.32\textwidth}
        \includegraphics[draft=\draft, width=\linewidth]{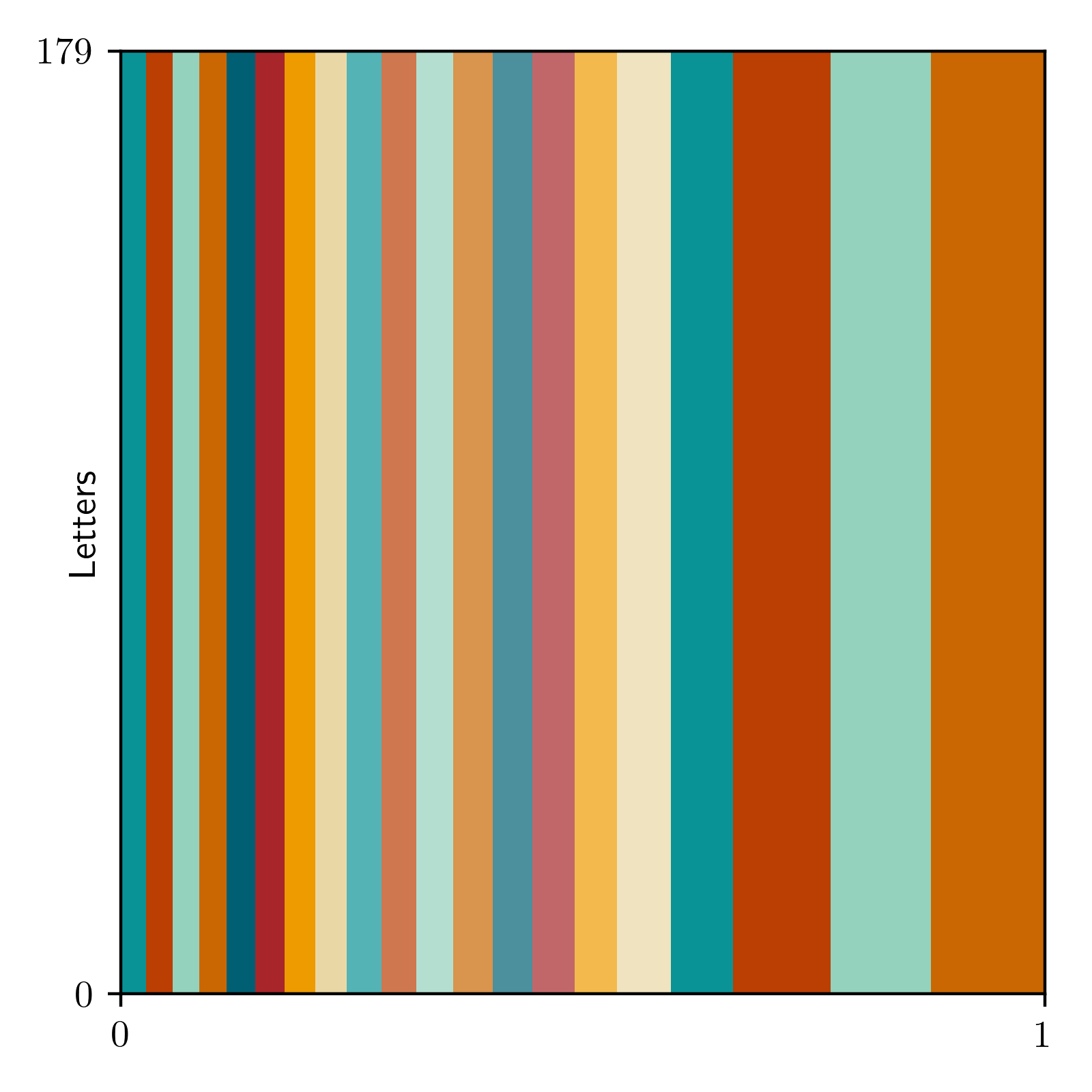}
        \includegraphics[draft=\draft, width=\linewidth]{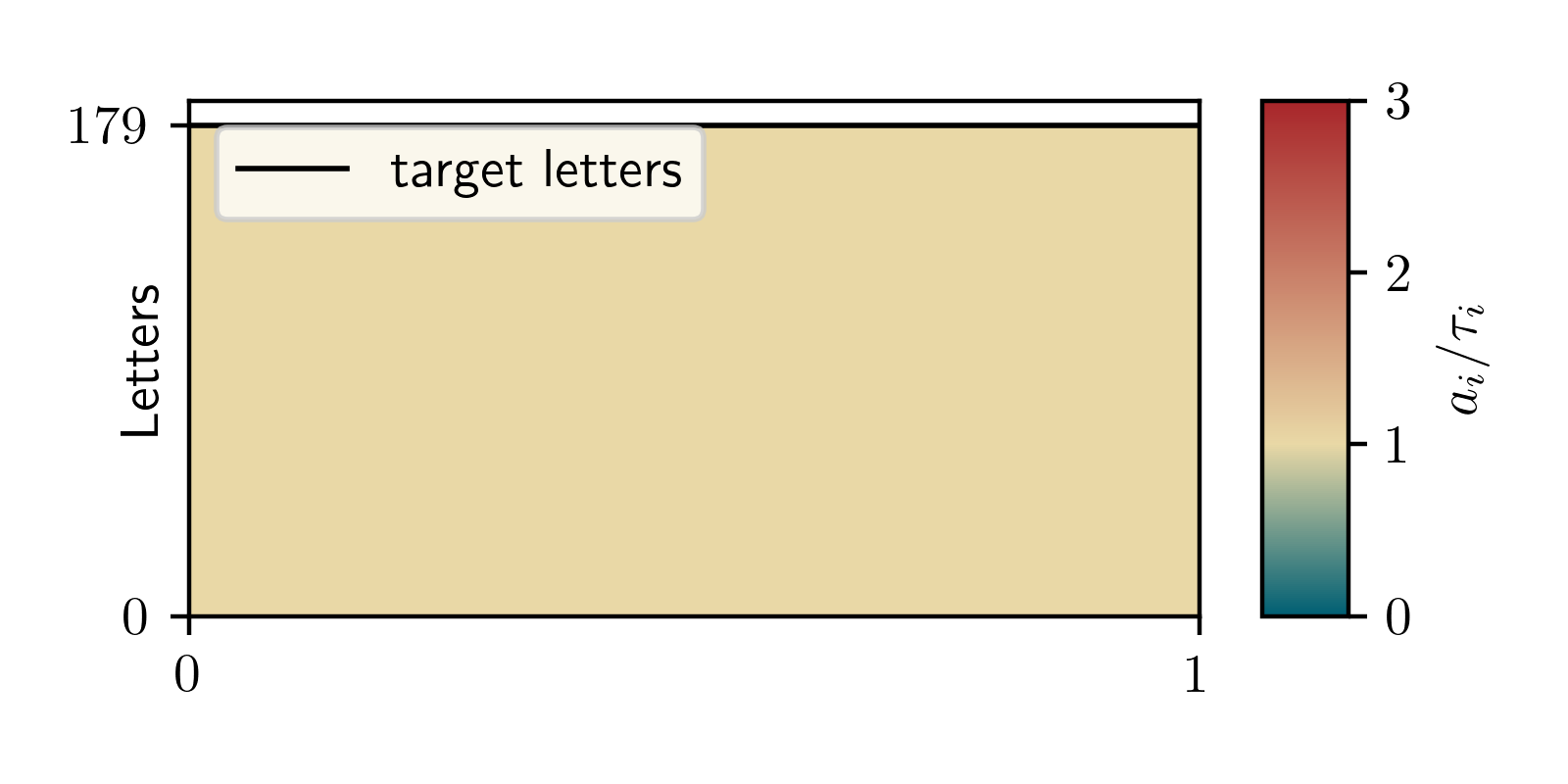}
        \caption{\greq ($t_G = 1$)}
        \label{fig:results_Schleswig-Holstein_Medium_greedy_equal}
    \end{subfigure}
    \begin{subfigure}{0.32\textwidth}
        \includegraphics[draft=\draft, width=\linewidth]{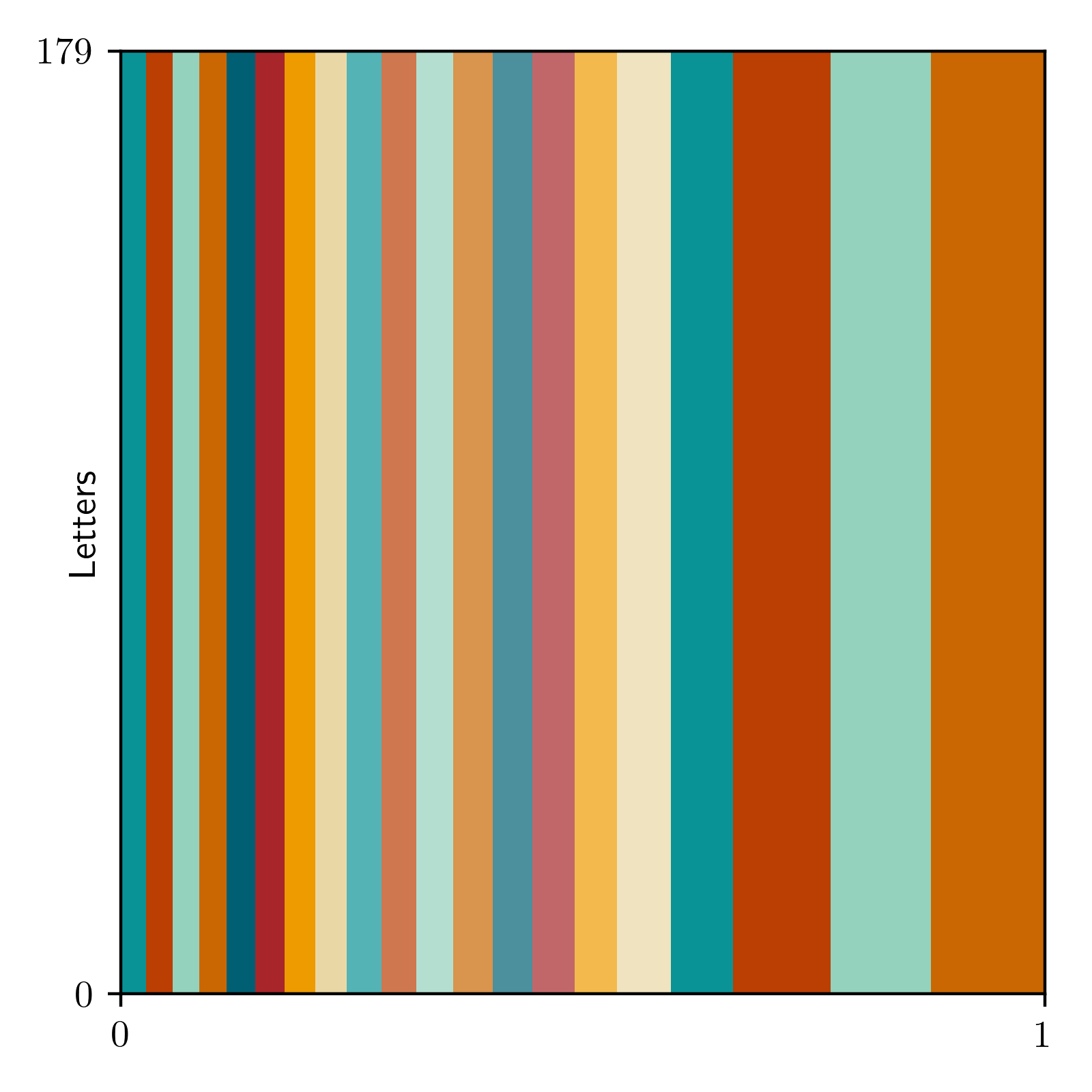}
        \includegraphics[draft=\draft, width=\linewidth]{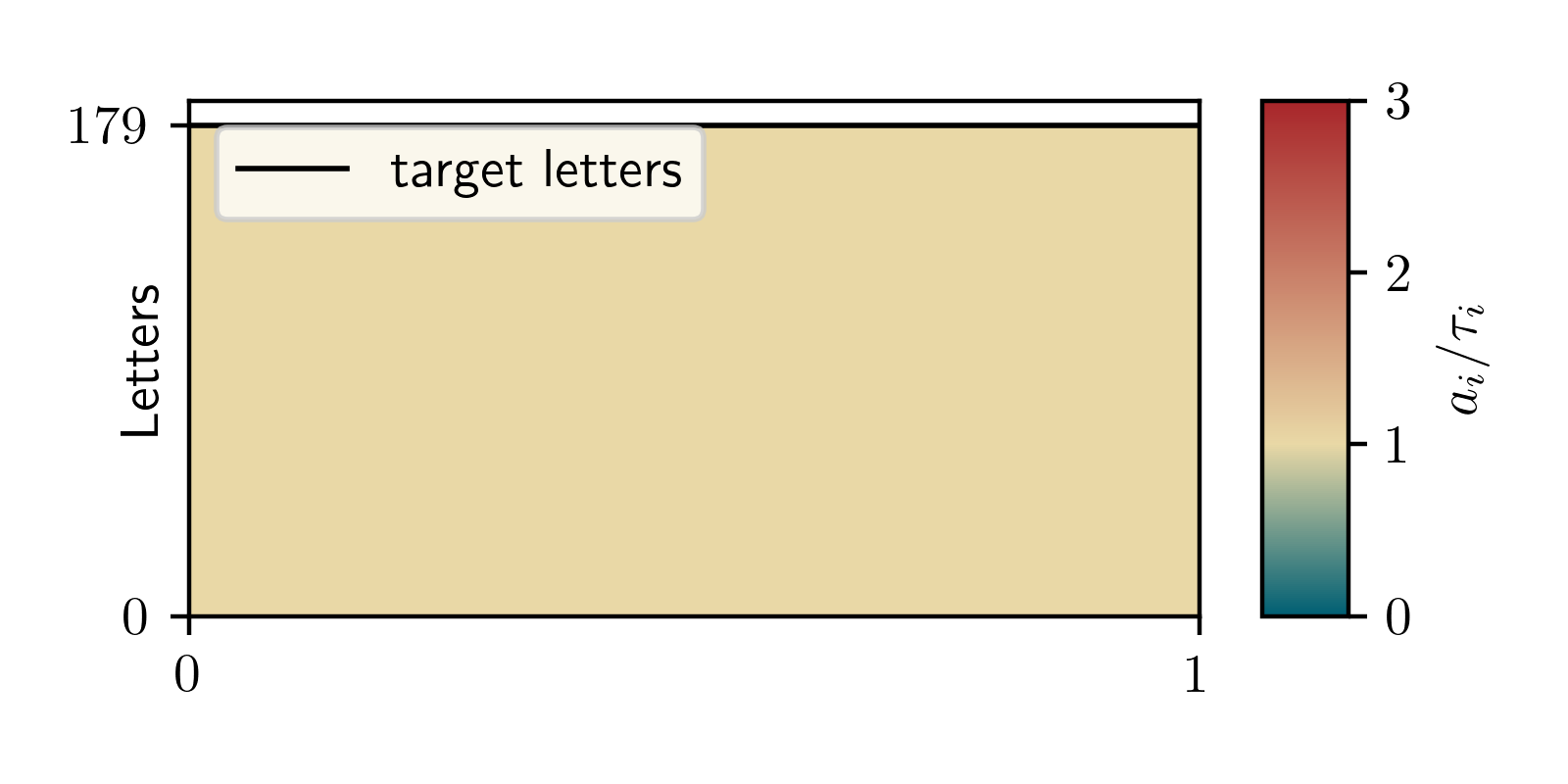}
        \caption{\colgen ($t_G\!=\!1$)}
        \label{fig:results_Schleswig-Holstein_Medium_column_generation}
    \end{subfigure}
    \begin{subfigure}{0.32\textwidth}
        \includegraphics[draft=\draft, width=\linewidth]{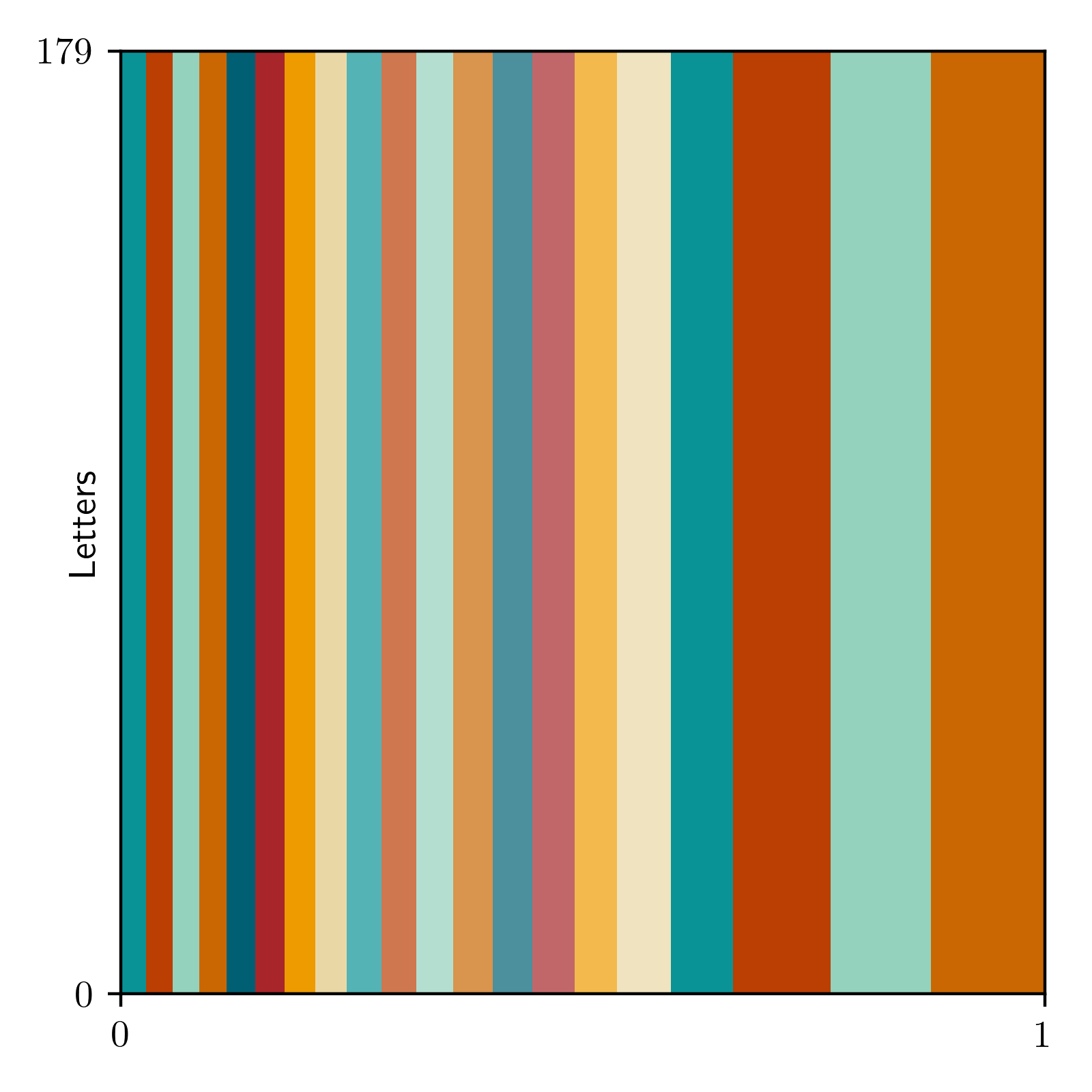}
        \includegraphics[draft=\draft, width=\linewidth]{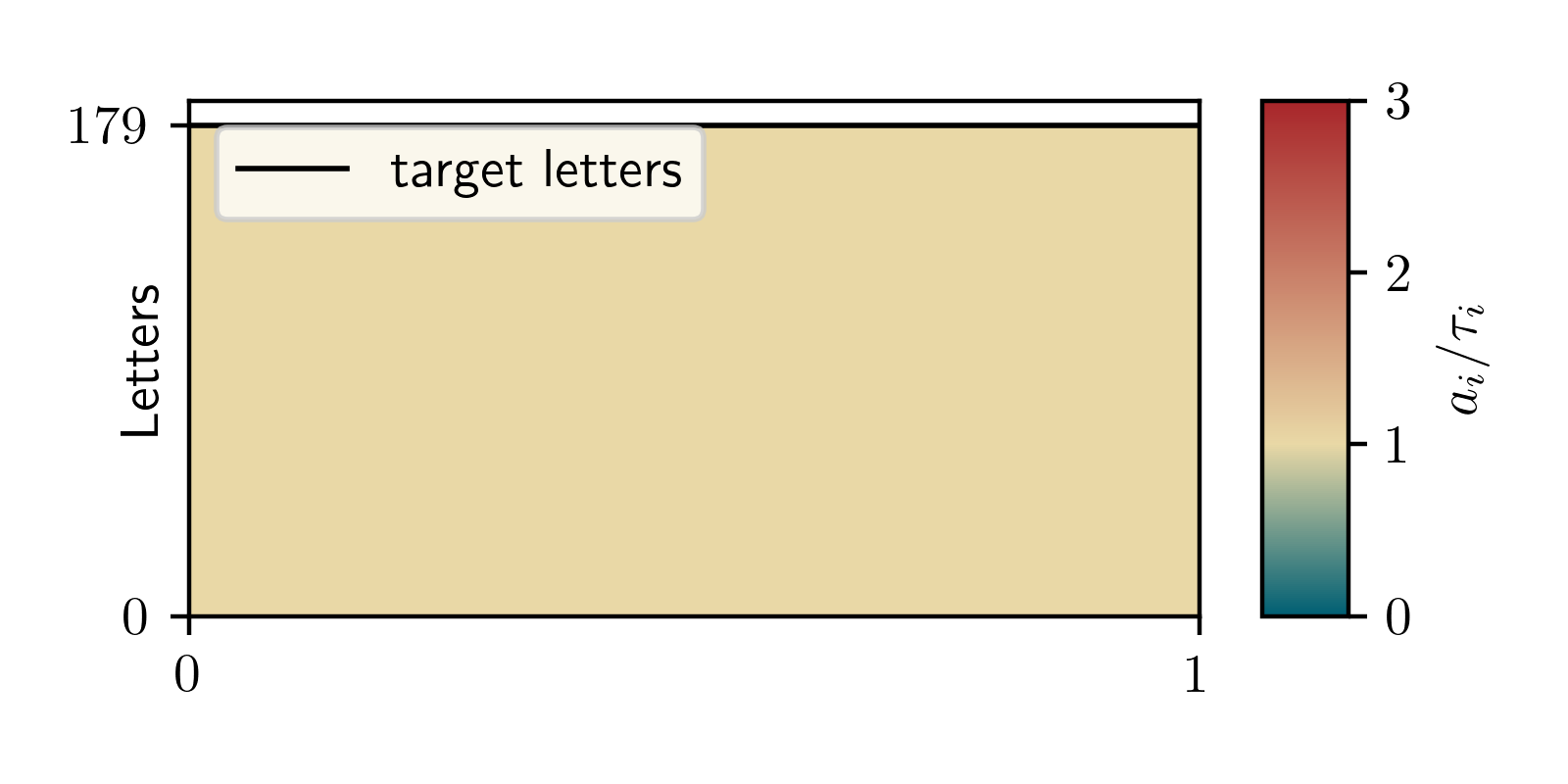}
        \caption{\buckets ($t_G = 1$)}
        \label{fig:results_Schleswig-Holstein_Medium_greedy_bucket_fill}
    \end{subfigure}
    \caption{Medium municipalities of Schleswig-Holstein ($\ell_G = 179$)}
    \label{fig:results_Schleswig-Holstein_Medium}
\end{figure} 

\begin{figure}
    \centering
    \begin{subfigure}{0.32\textwidth}
        \includegraphics[draft=\draft, width=\linewidth]{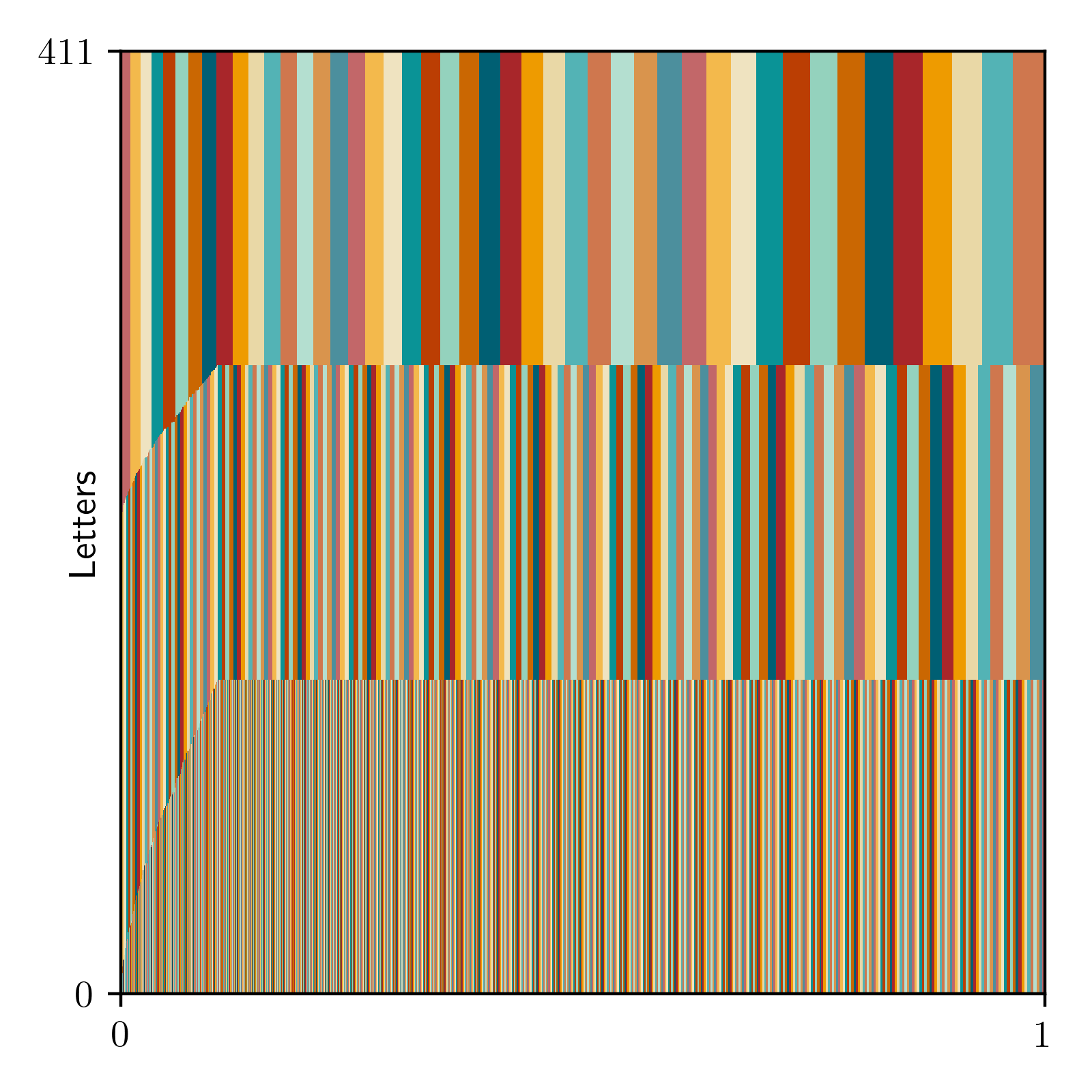}
        \includegraphics[draft=\draft, width=\linewidth]{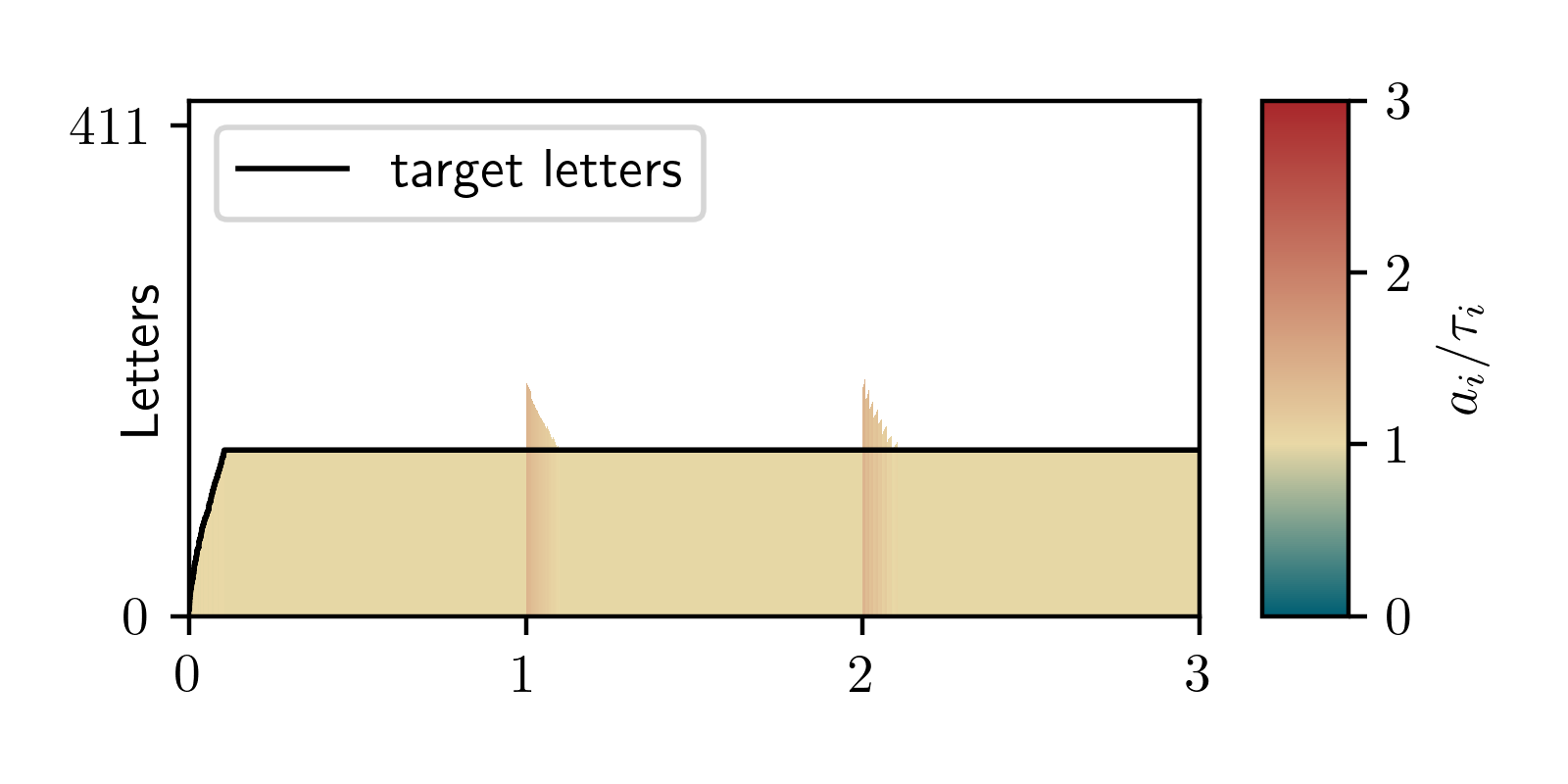}
        \caption{\greq ($t_G = 3$)}
        \label{fig:results_Schleswig-Holstein_Small_greedy_equal}
    \end{subfigure}
    \begin{subfigure}{0.32\textwidth}
        \includegraphics[draft=\draft, width=\linewidth]{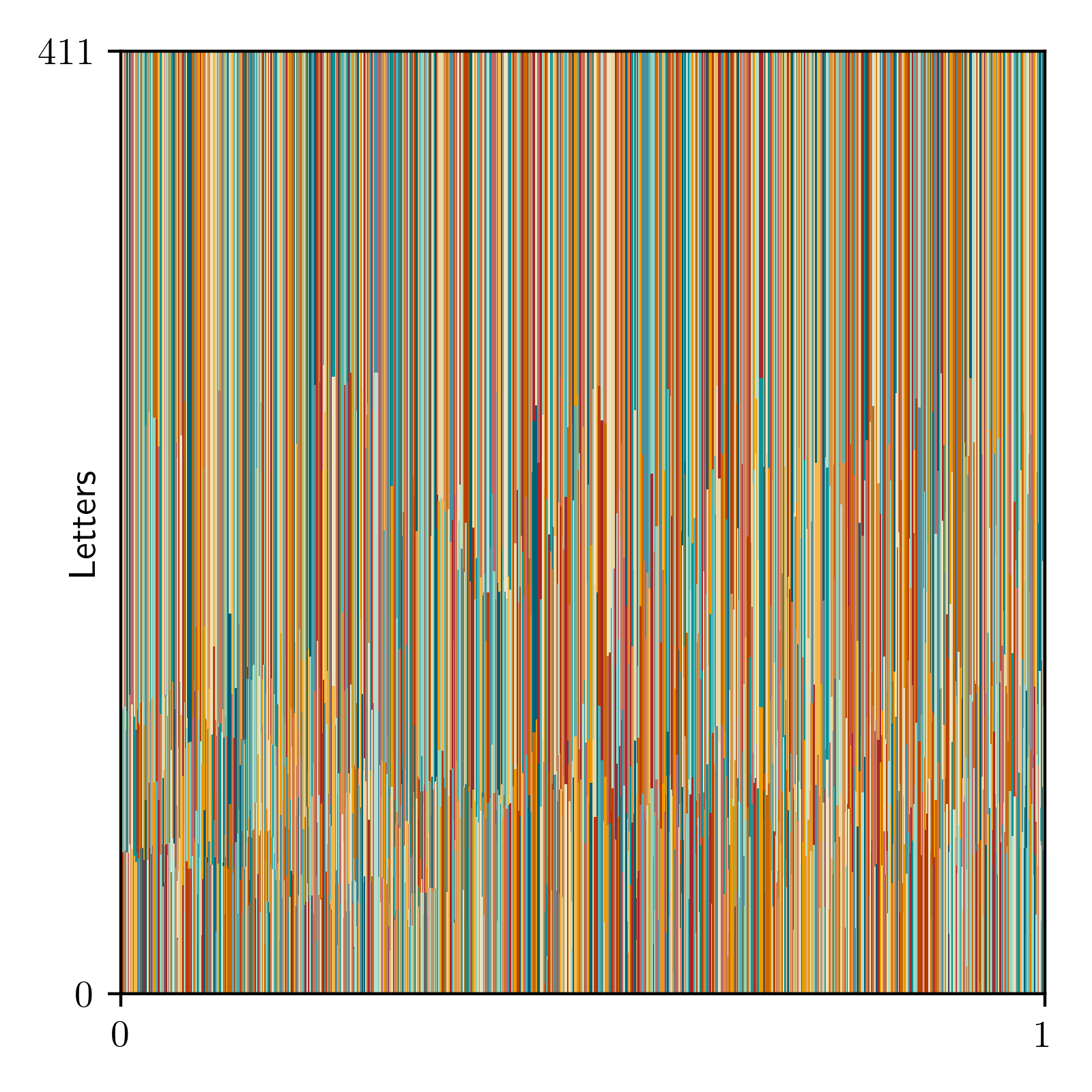}
        \includegraphics[draft=\draft, width=\linewidth]{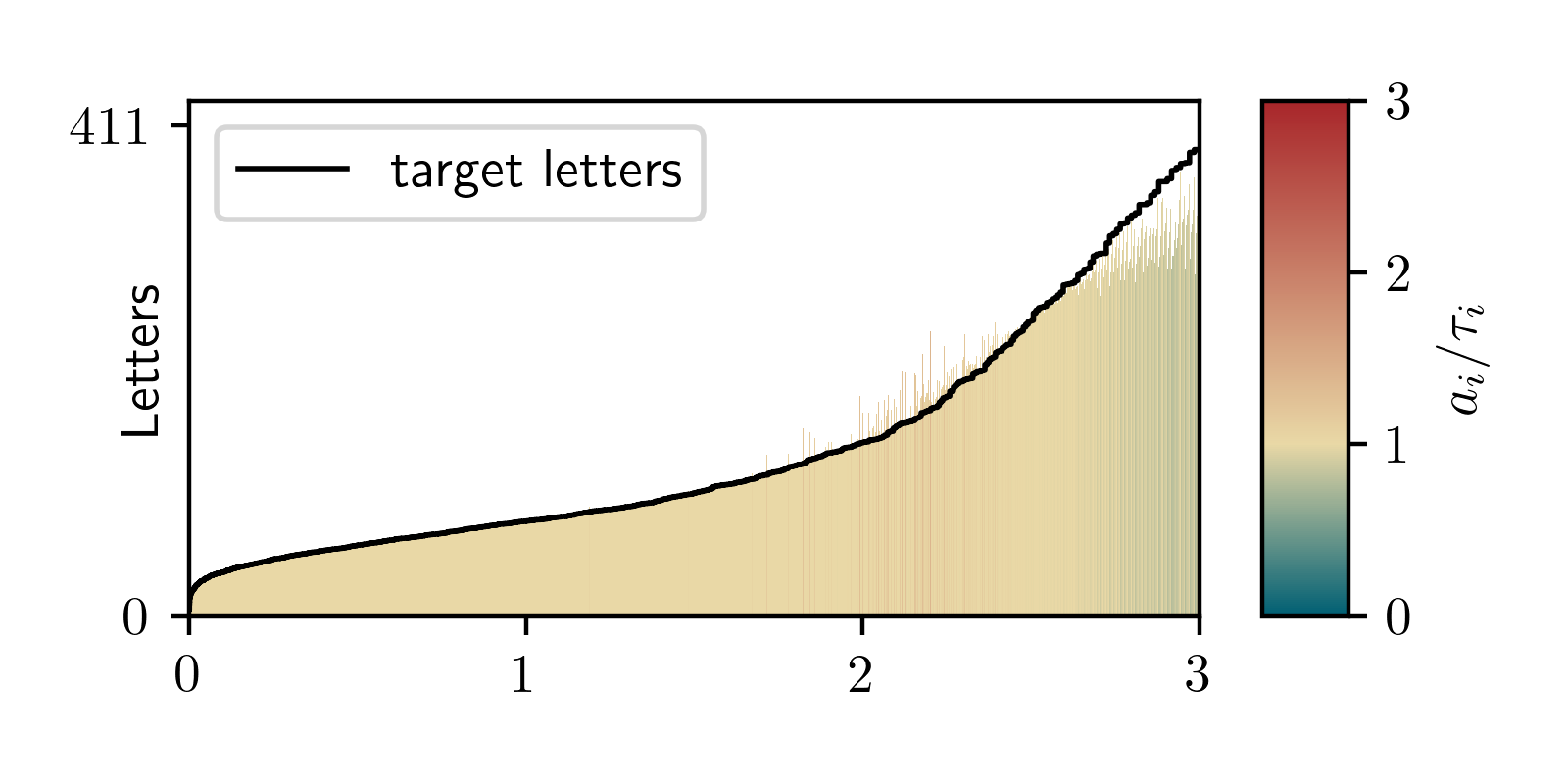}
        \caption{\colgen ($t_G\!=\!3$)}
        \label{fig:results_Schleswig-Holstein_Small_column_generation}
    \end{subfigure}
    \begin{subfigure}{0.32\textwidth}
        \includegraphics[draft=\draft, width=\linewidth]{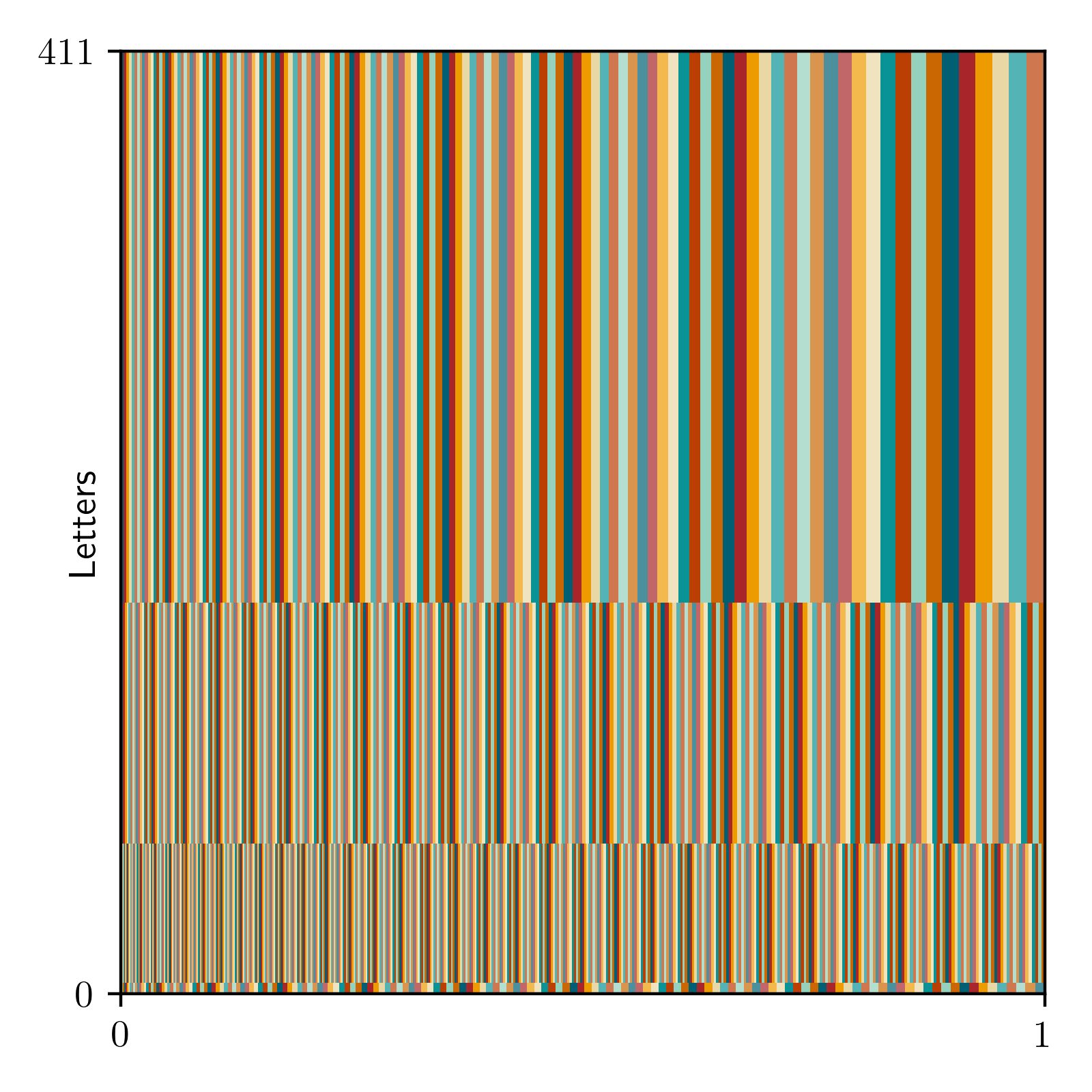}
        \includegraphics[draft=\draft, width=\linewidth]{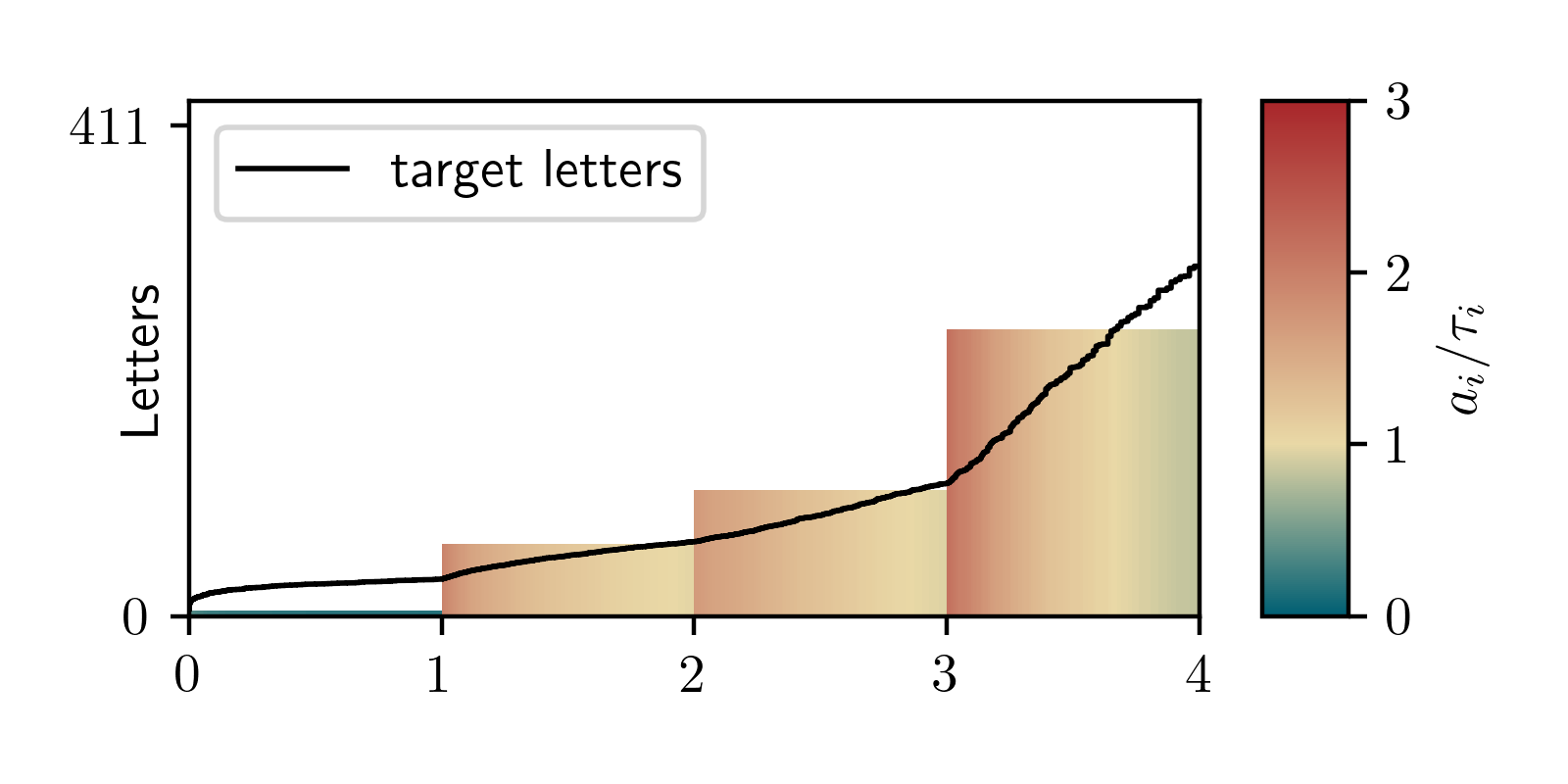}
        \caption{\buckets ($t_G = 4$)}
        \label{fig:results_Schleswig-Holstein_Small_greedy_bucket_fill}
    \end{subfigure}
    \caption{Small municipalities of Schleswig-Holstein ($\ell_G = 411$)}
    \label{fig:results_Schleswig-Holstein_Small}
\end{figure} 

\begin{figure}
    \centering
    \begin{subfigure}{0.32\textwidth}
        \includegraphics[draft=\draft, width=\linewidth]{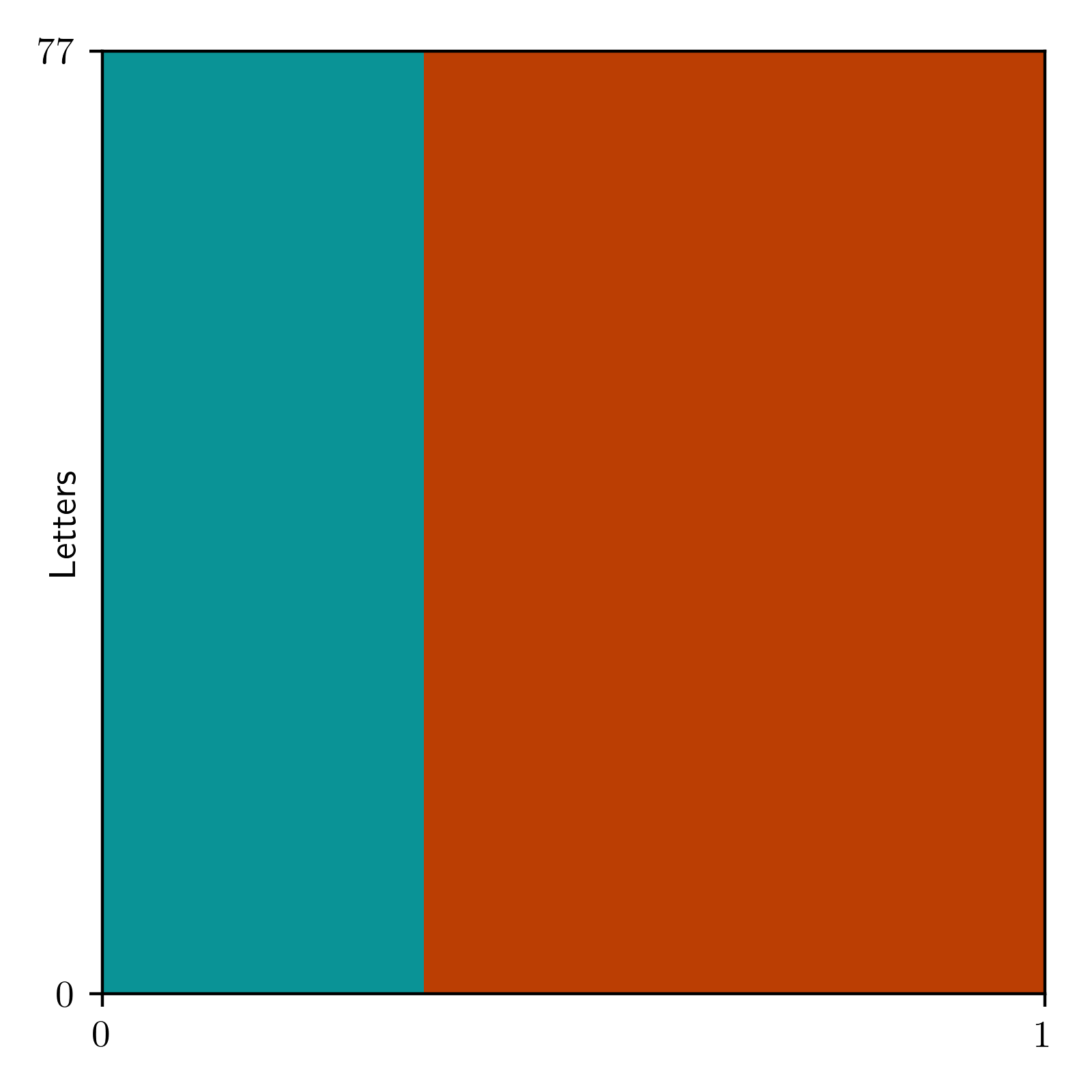}
        \includegraphics[draft=\draft, width=\linewidth]{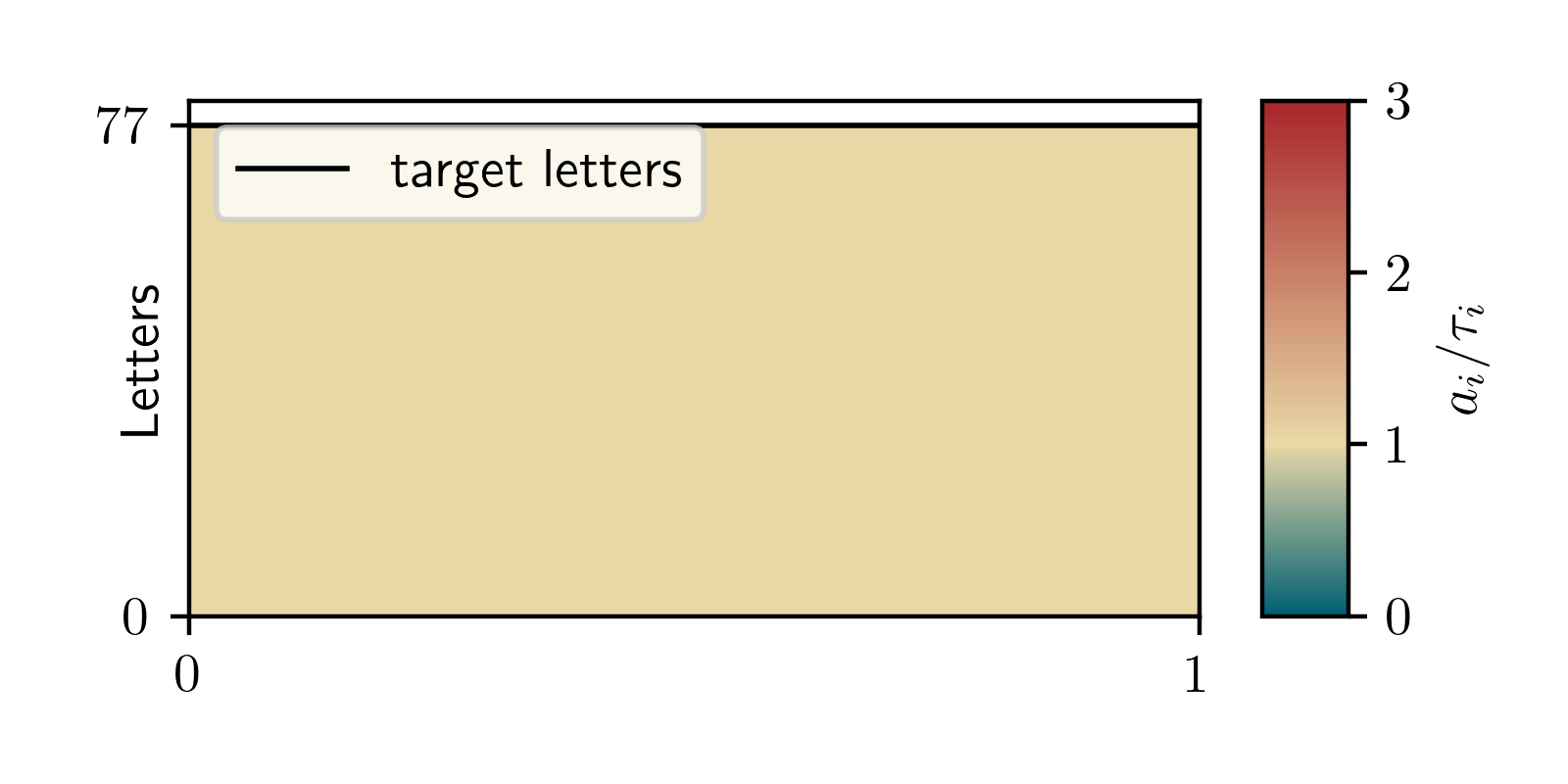}
        \caption{\greq ($t_G = 1$)}
        \label{fig:results_Thüringen_Large_greedy_equal}
    \end{subfigure}
    \begin{subfigure}{0.32\textwidth}
        \includegraphics[draft=\draft, width=\linewidth]{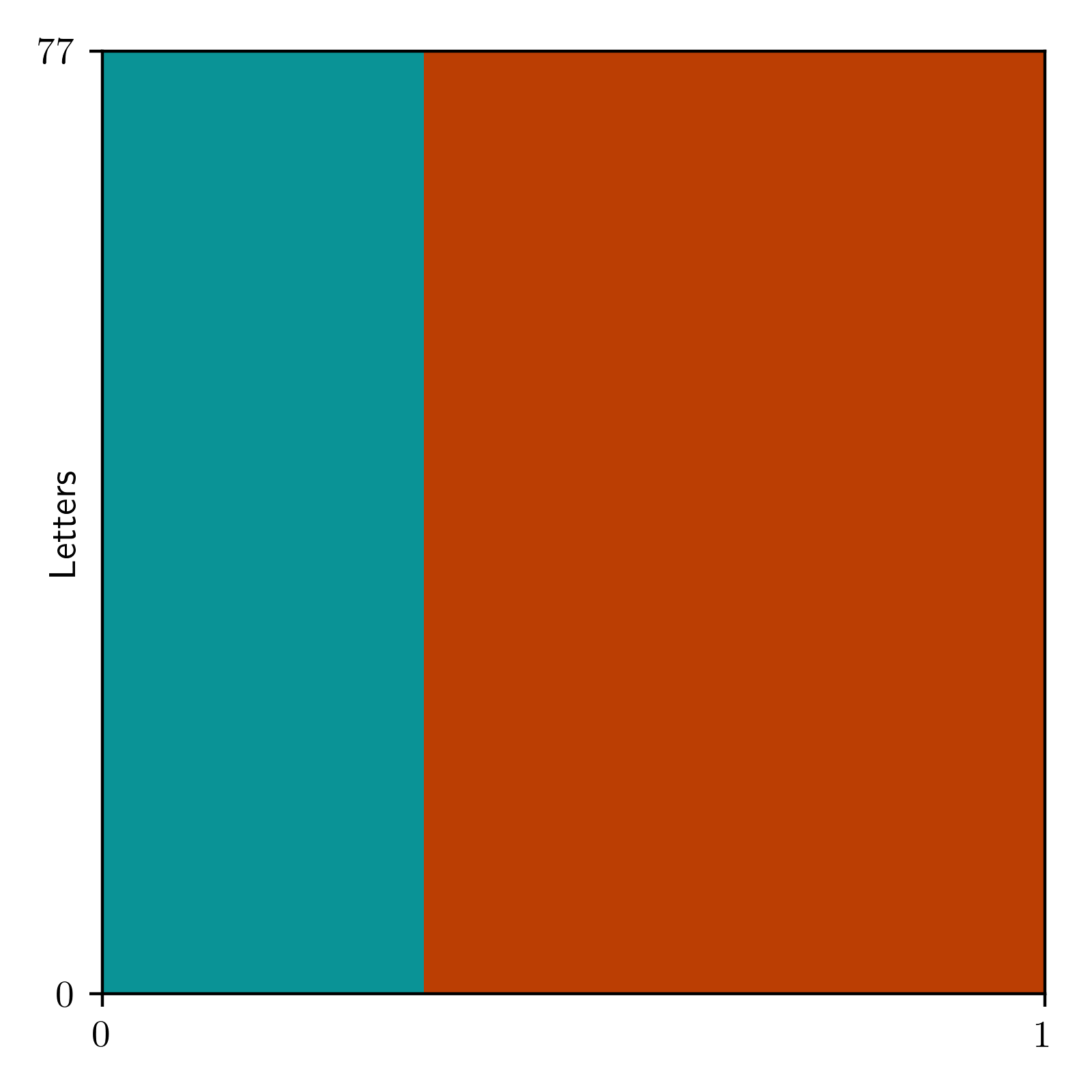}
        \includegraphics[draft=\draft, width=\linewidth]{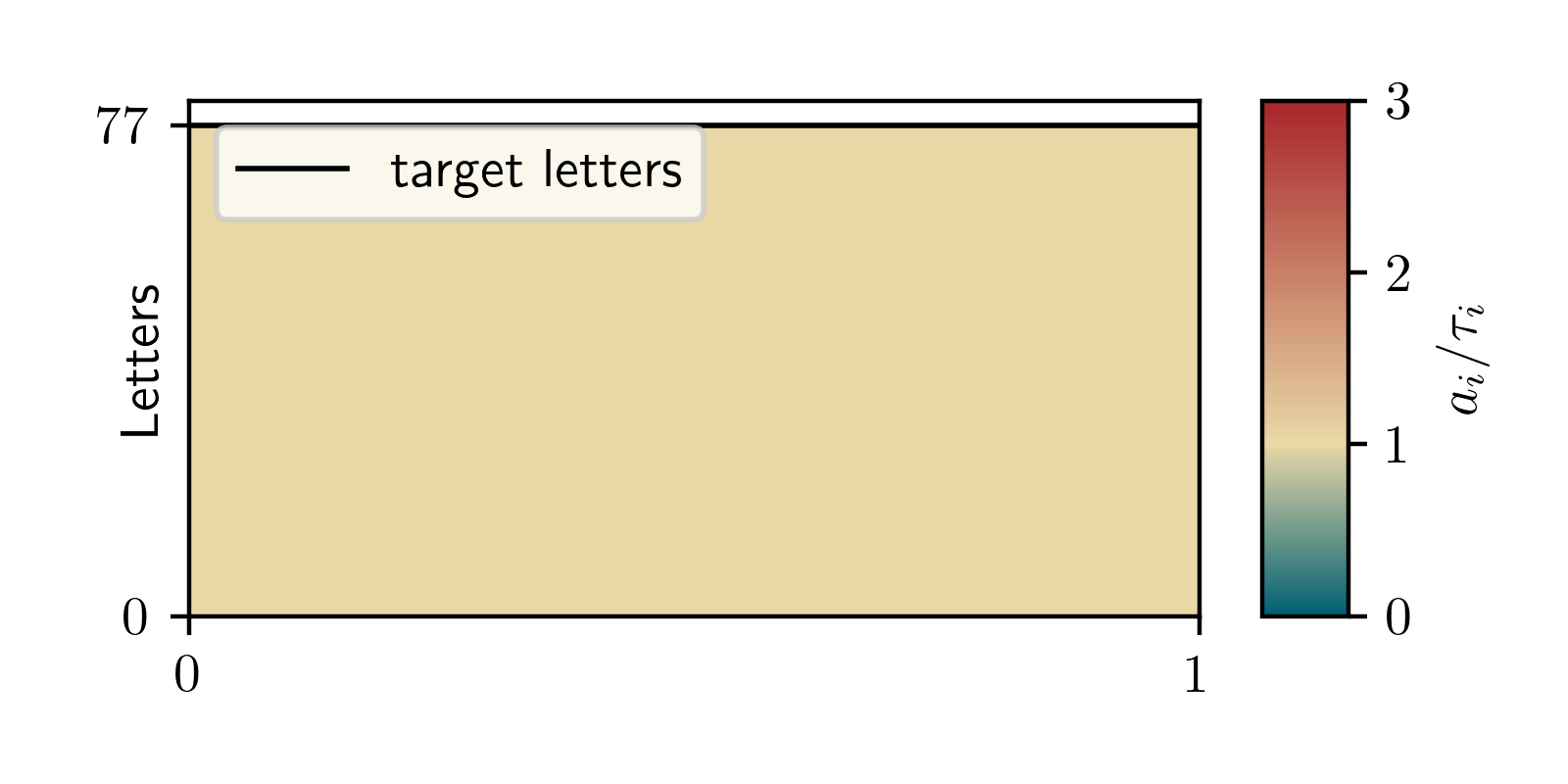}
        \caption{\colgen ($t_G\!=\!1$)}
        \label{fig:results_Thüringen_Large_column_generation}
    \end{subfigure}
    \begin{subfigure}{0.32\textwidth}
        \includegraphics[draft=\draft, width=\linewidth]{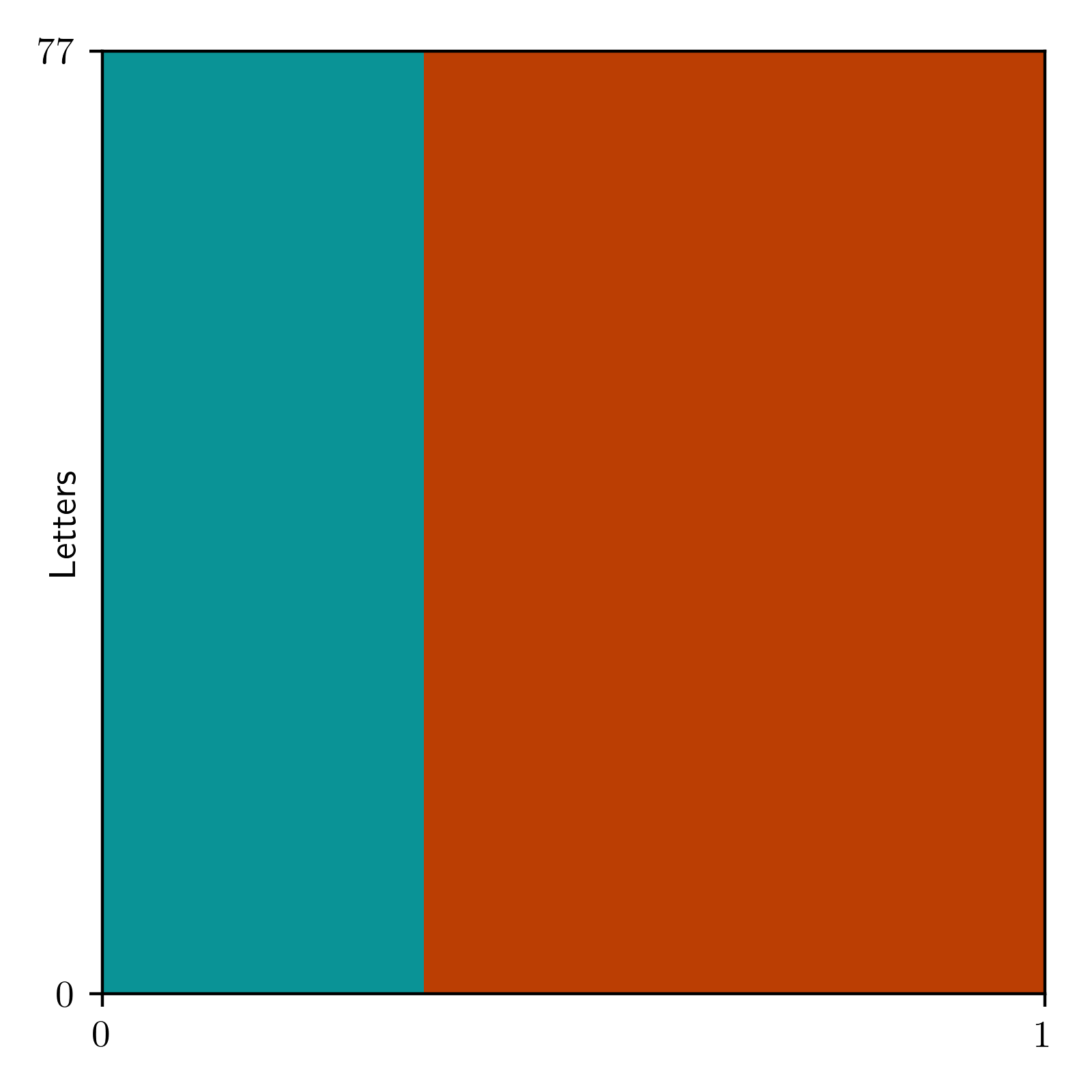}
        \includegraphics[draft=\draft, width=\linewidth]{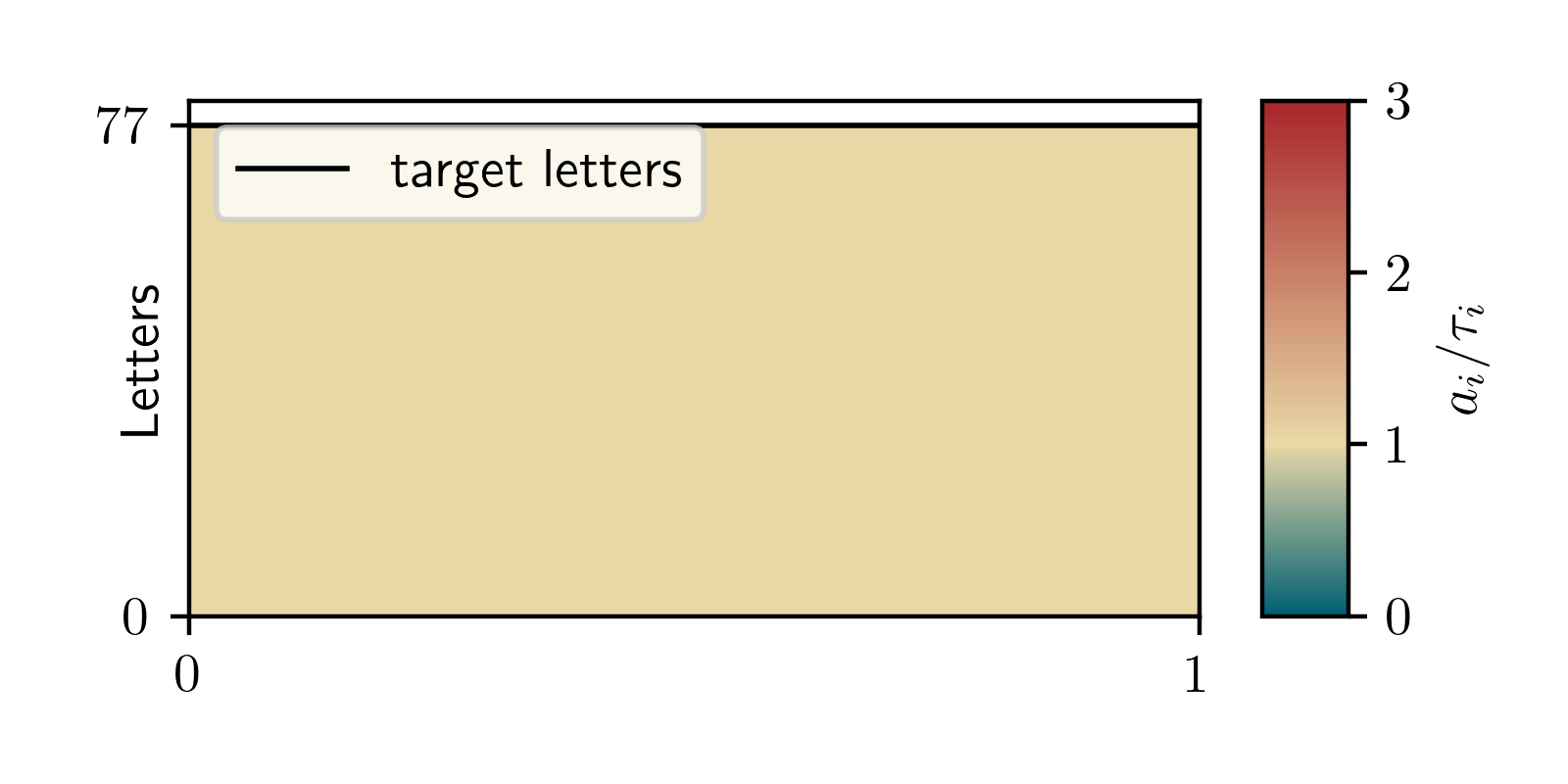}
        \caption{\buckets ($t_G = 1$)}
        \label{fig:results_Thüringen_Large_greedy_bucket_fill}
    \end{subfigure}
    \caption{Large municipalities of Thüringen ($\ell_G = 77$)}
    \label{fig:results_Thüringen_Large}
\end{figure} 

\begin{figure}
    \centering
    \begin{subfigure}{0.32\textwidth}
        \includegraphics[draft=\draft, width=\linewidth]{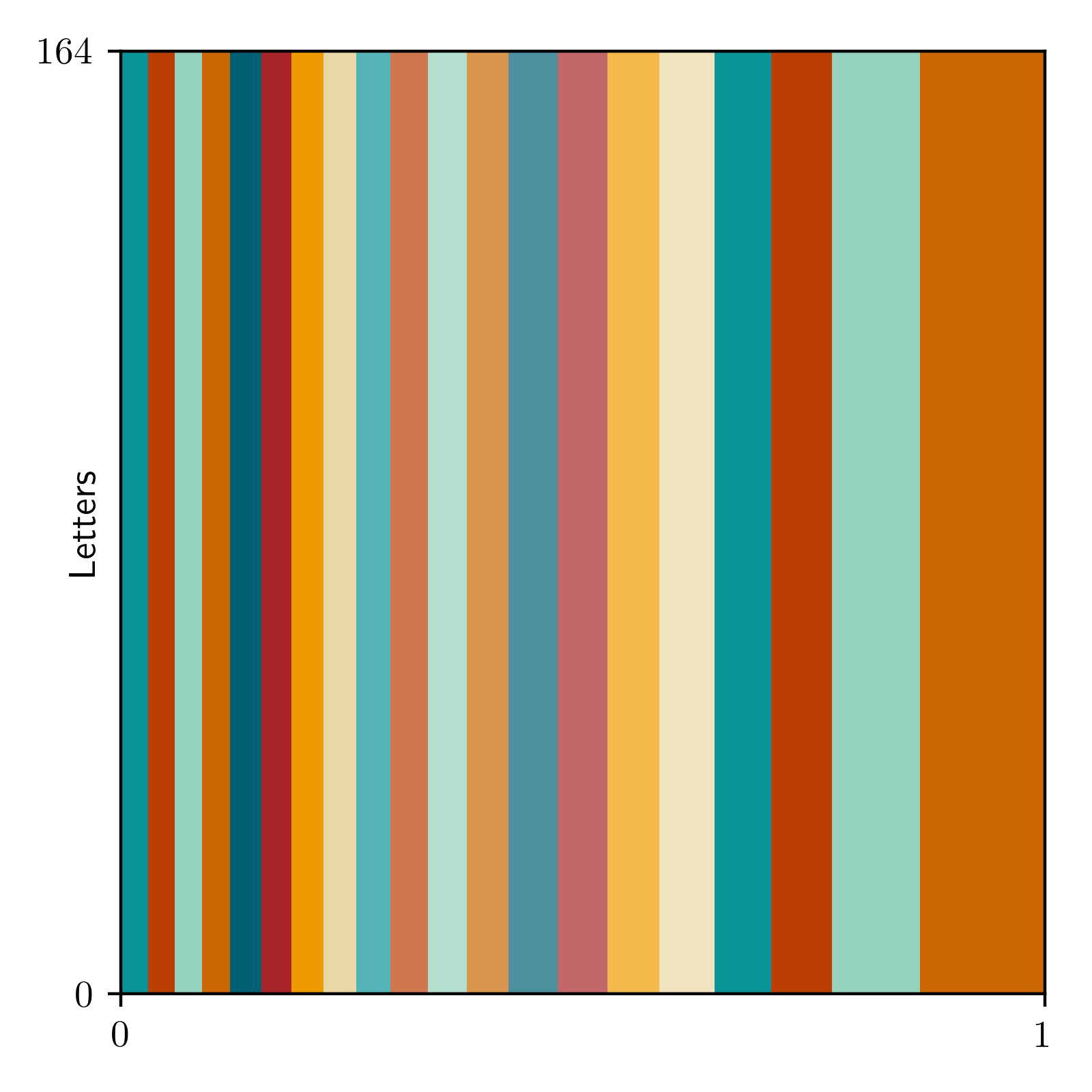}
        \includegraphics[draft=\draft, width=\linewidth]{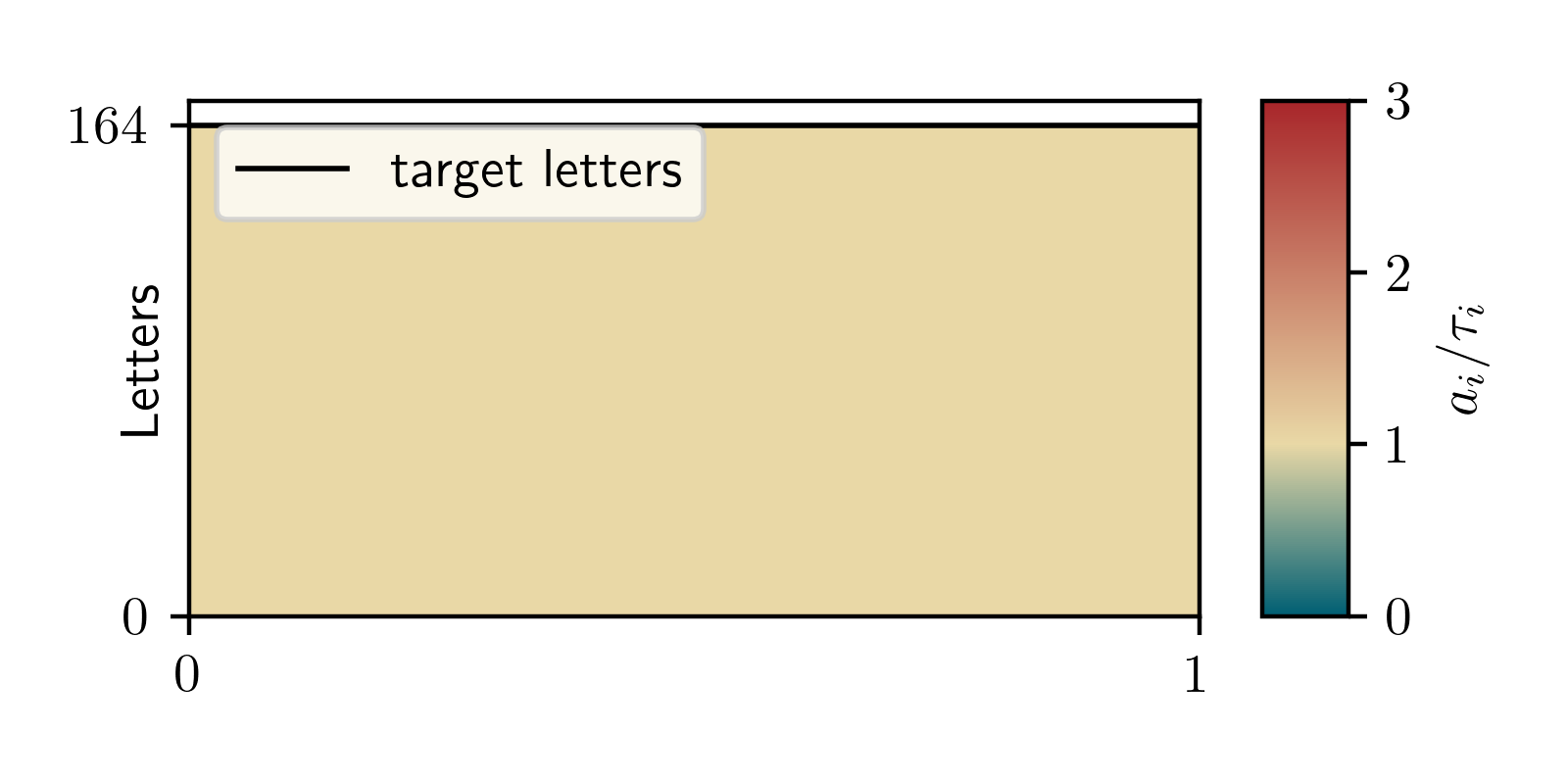}
        \caption{\greq ($t_G = 1$)}
        \label{fig:results_Thüringen_Medium_greedy_equal}
    \end{subfigure}
    \begin{subfigure}{0.32\textwidth}
        \includegraphics[draft=\draft, width=\linewidth]{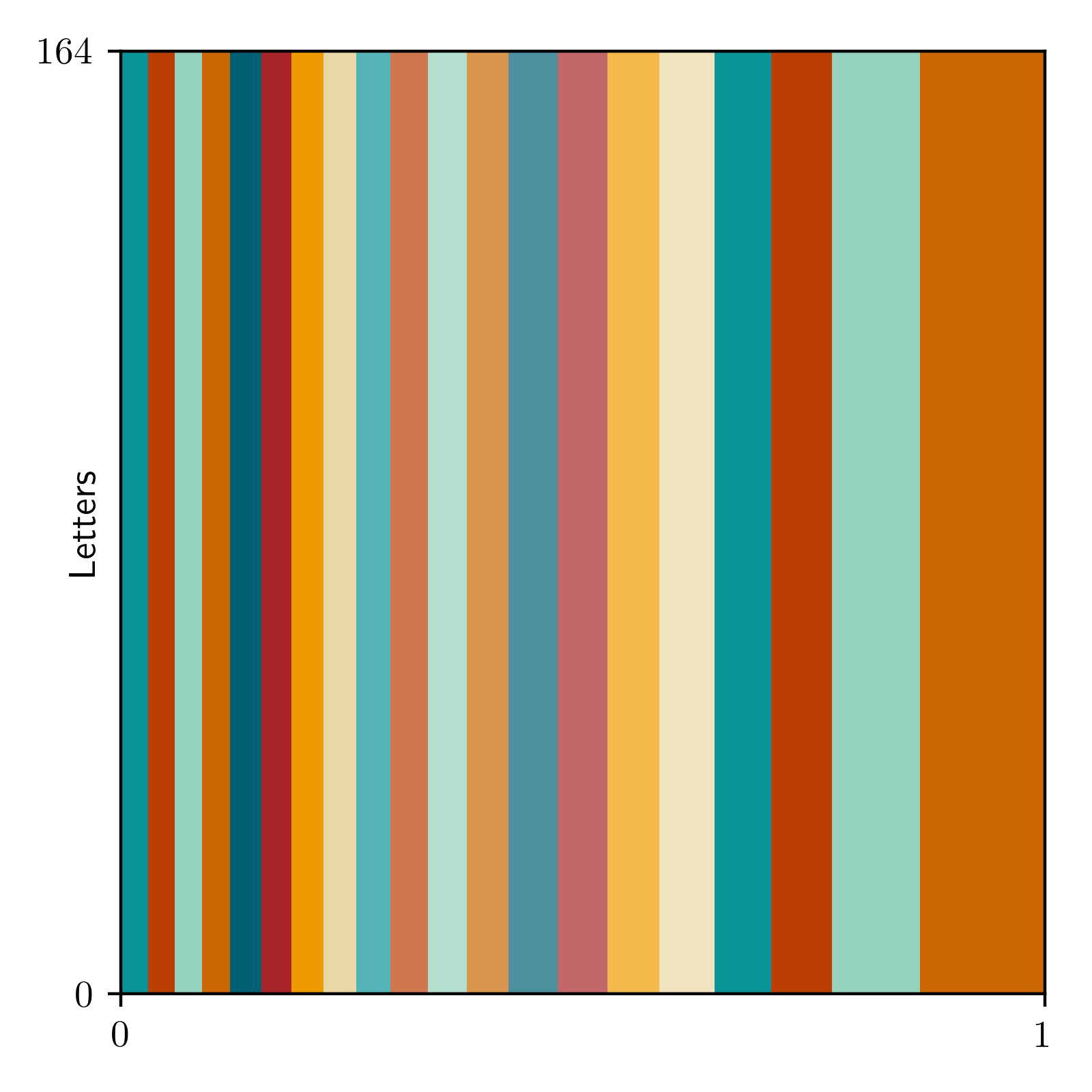}
        \includegraphics[draft=\draft, width=\linewidth]{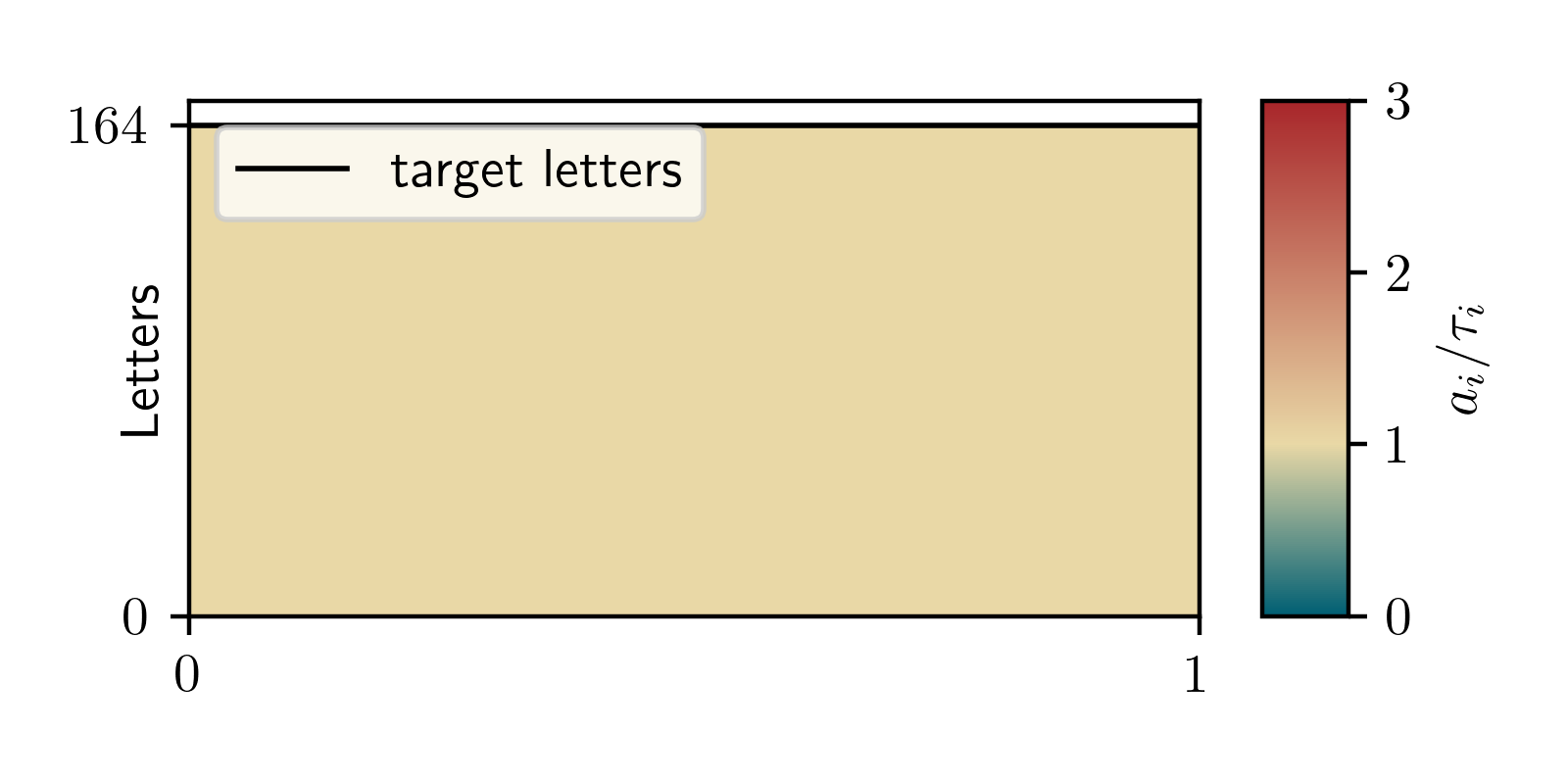}
        \caption{\colgen ($t_G\!=\!1$)}
        \label{fig:results_Thüringen_Medium_column_generation}
    \end{subfigure}
    \begin{subfigure}{0.32\textwidth}
        \includegraphics[draft=\draft, width=\linewidth]{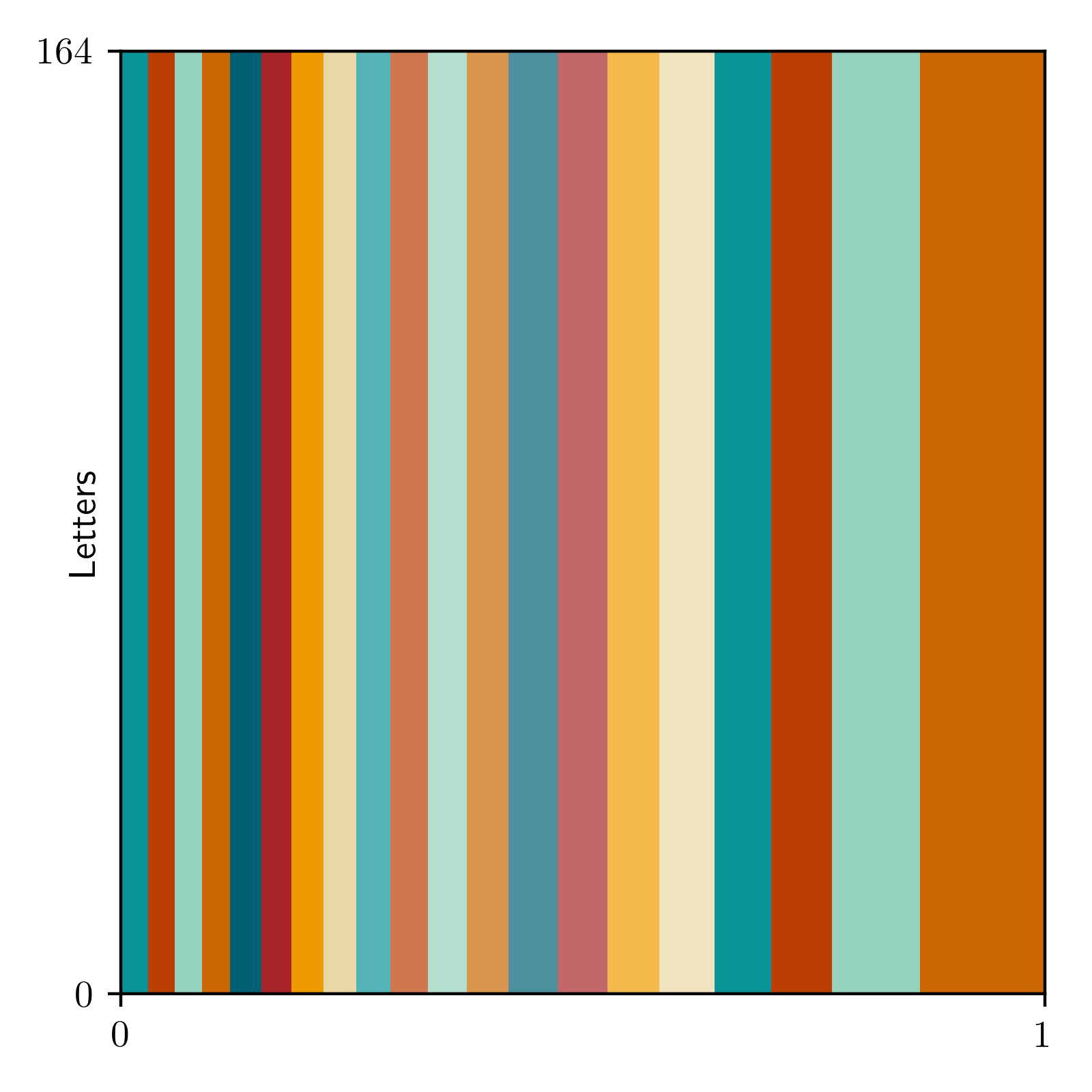}
        \includegraphics[draft=\draft, width=\linewidth]{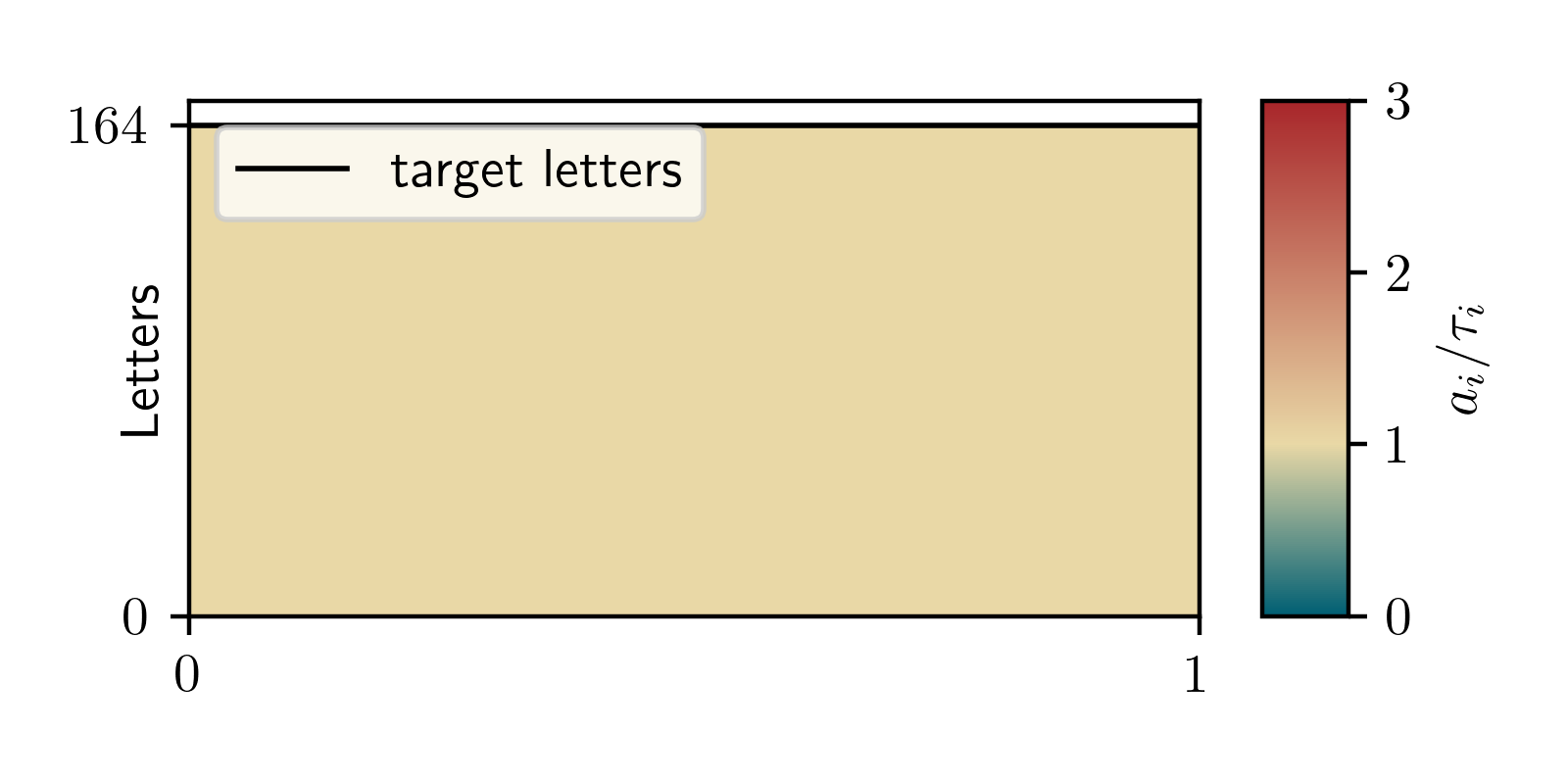}
        \caption{\buckets ($t_G = 1$)}
        \label{fig:results_Thüringen_Medium_greedy_bucket_fill}
    \end{subfigure}
    \caption{Medium municipalities of Thüringen ($\ell_G = 164$)}
    \label{fig:results_Thüringen_Medium}
\end{figure} 

\begin{figure}
    \centering
    \begin{subfigure}{0.32\textwidth}
        \includegraphics[draft=\draft, width=\linewidth]{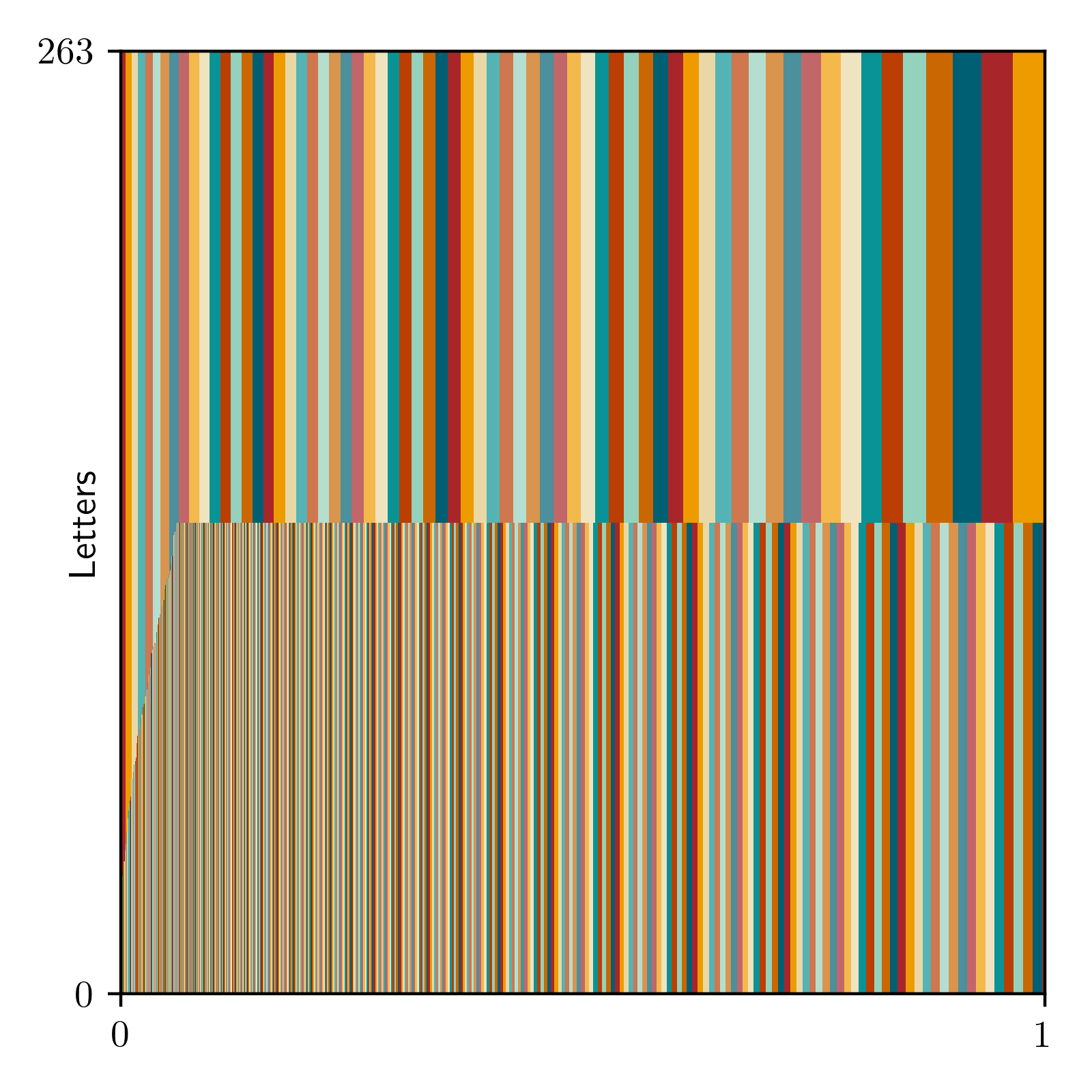}
        \includegraphics[draft=\draft, width=\linewidth]{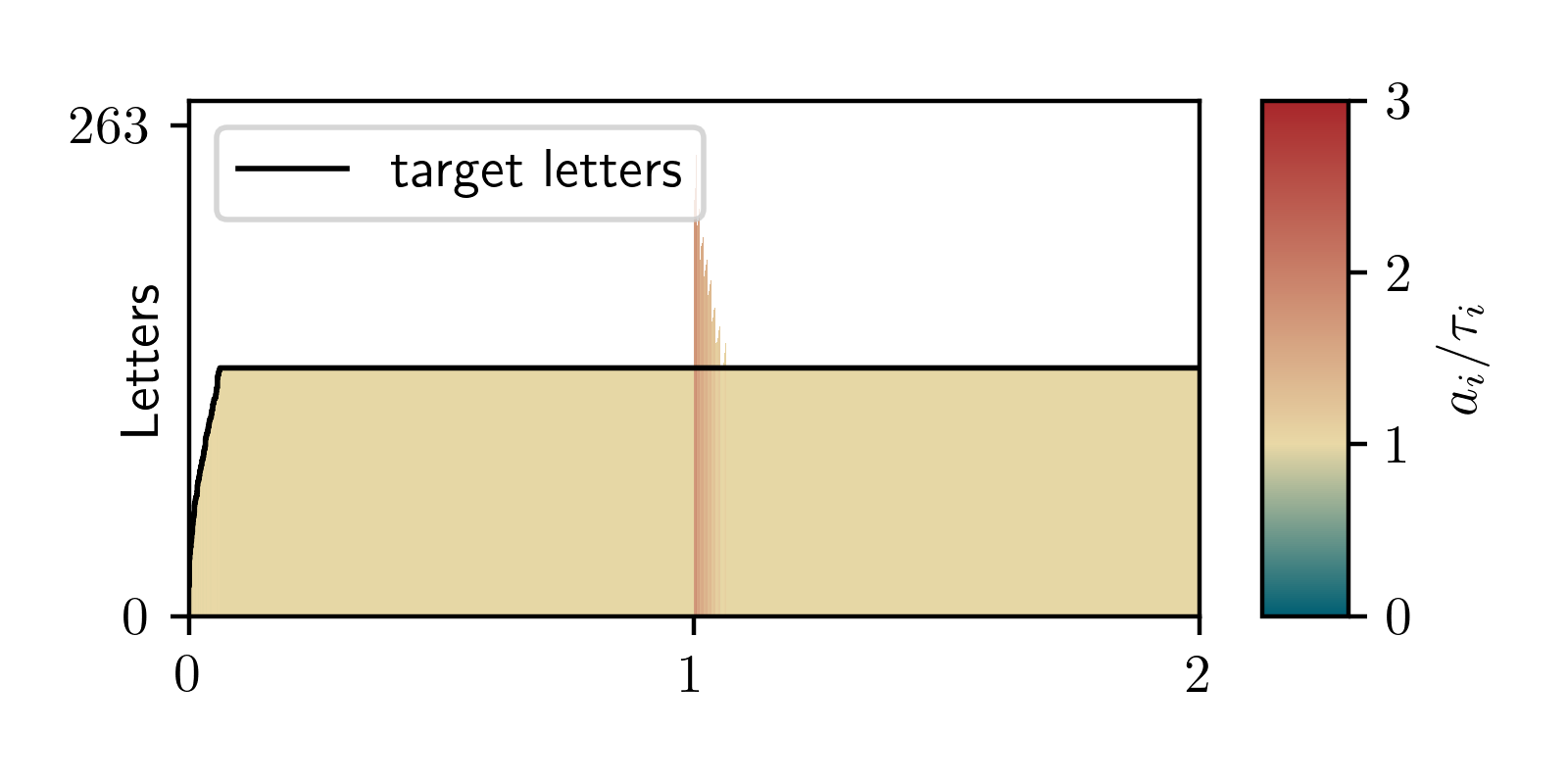}
        \caption{\greq ($t_G = 2$)}
        \label{fig:results_Thüringen_Small_greedy_equal}
    \end{subfigure}
    \begin{subfigure}{0.32\textwidth}
        \includegraphics[draft=\draft, width=\linewidth]{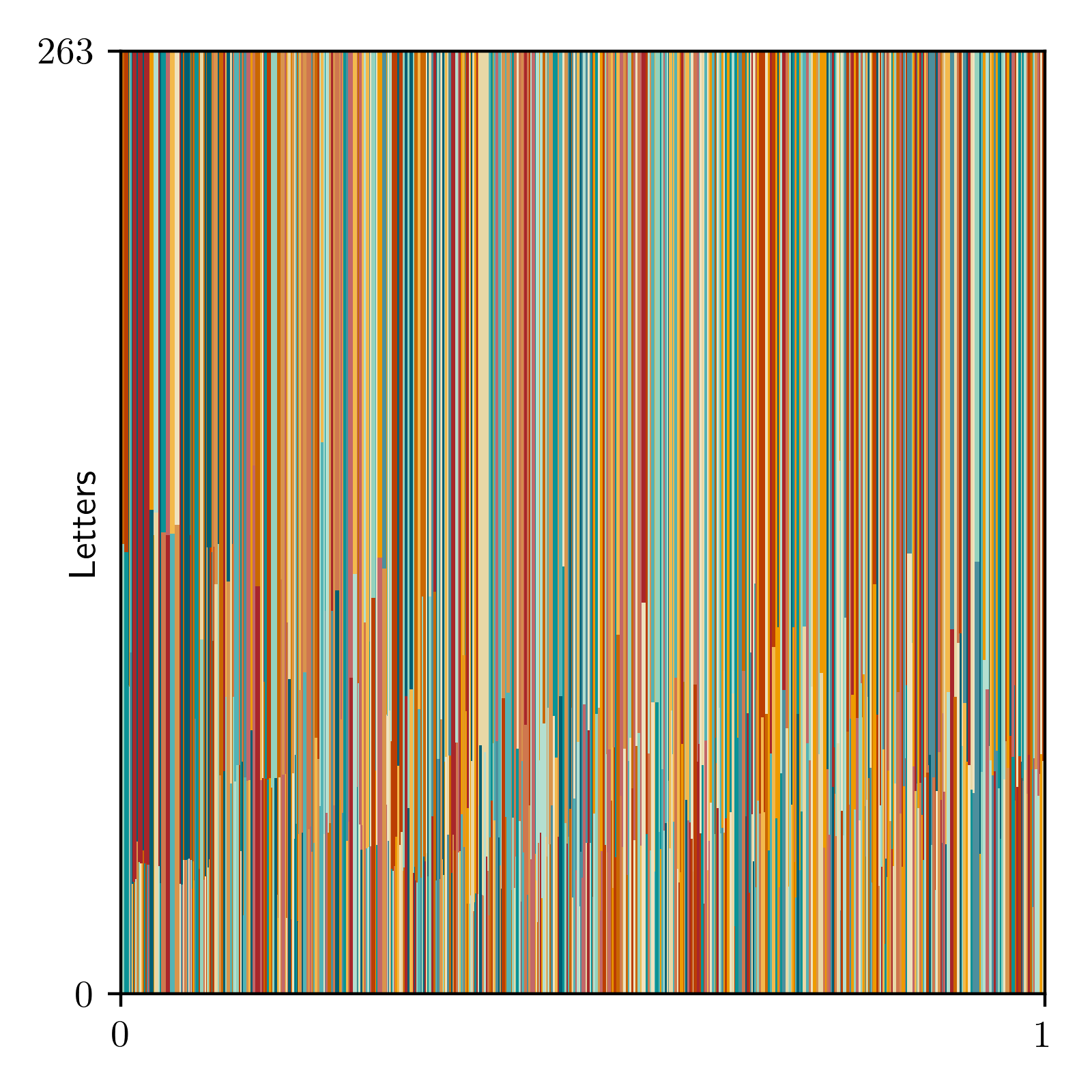}
        \includegraphics[draft=\draft, width=\linewidth]{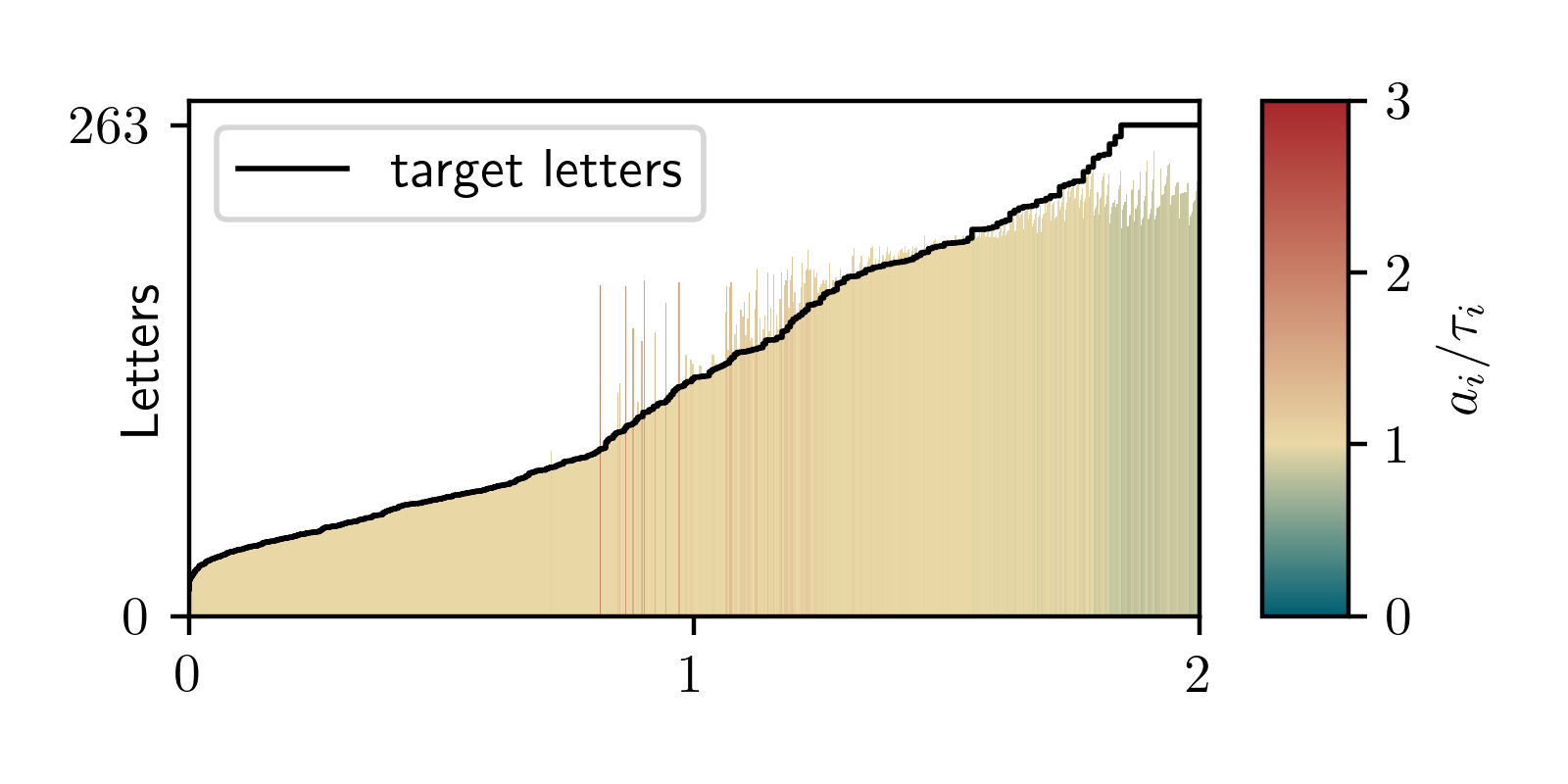}
        \caption{\colgen ($t_G\!=\!2$)}
        \label{fig:results_Thüringen_Small_column_generation}
    \end{subfigure}
    \begin{subfigure}{0.32\textwidth}
        \includegraphics[draft=\draft, width=\linewidth]{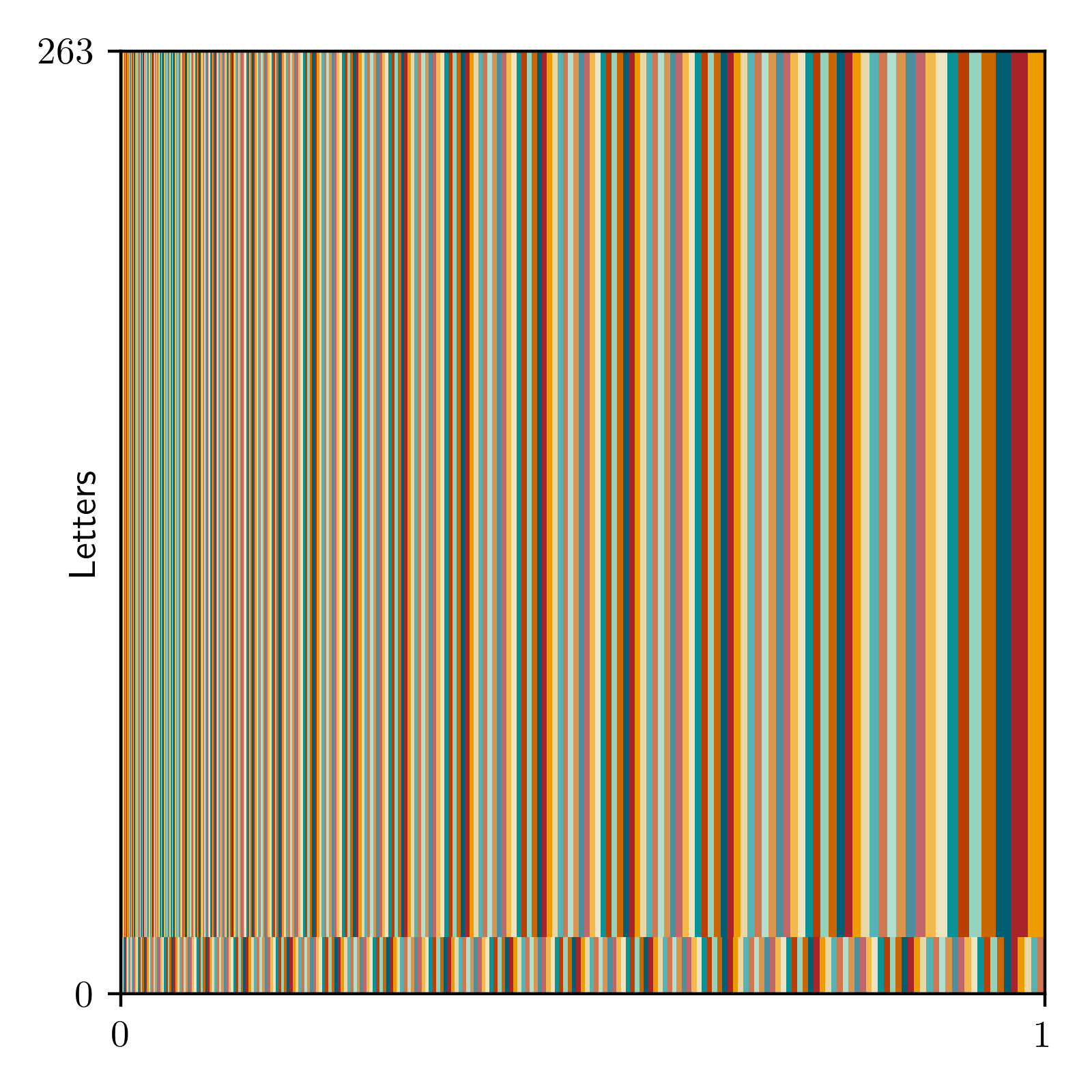}
        \includegraphics[draft=\draft, width=\linewidth]{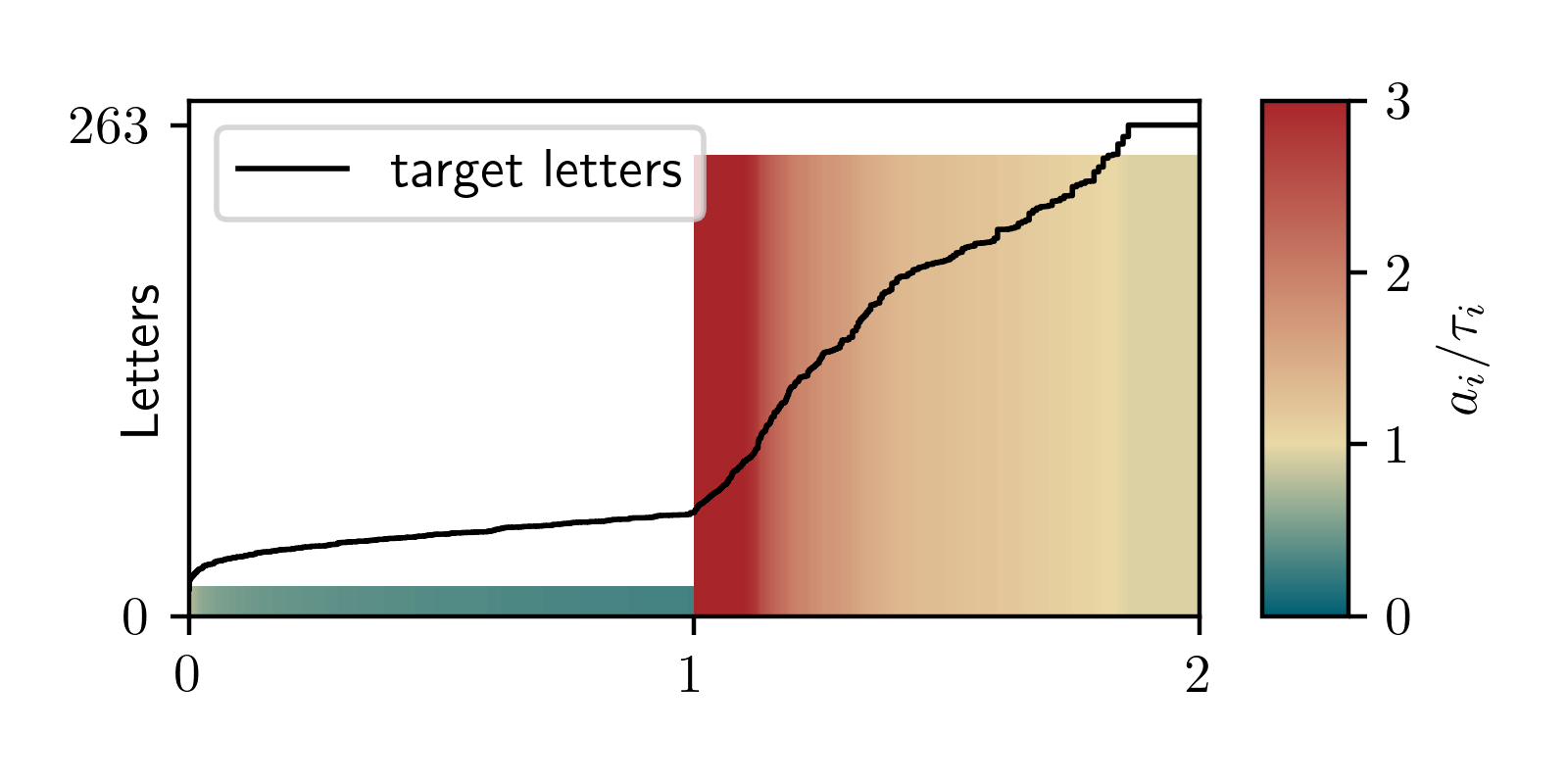}
        \caption{\buckets ($t_G = 2$)}
        \label{fig:results_Thüringen_Small_greedy_bucket_fill}
    \end{subfigure}
    \caption{Small municipalities of Thüringen ($\ell_G = 263$)}
    \label{fig:results_Thüringen_Small}
\end{figure}

\begin{landscape}

\begin{figure}
    \centering
    \includegraphics[width=\linewidth]{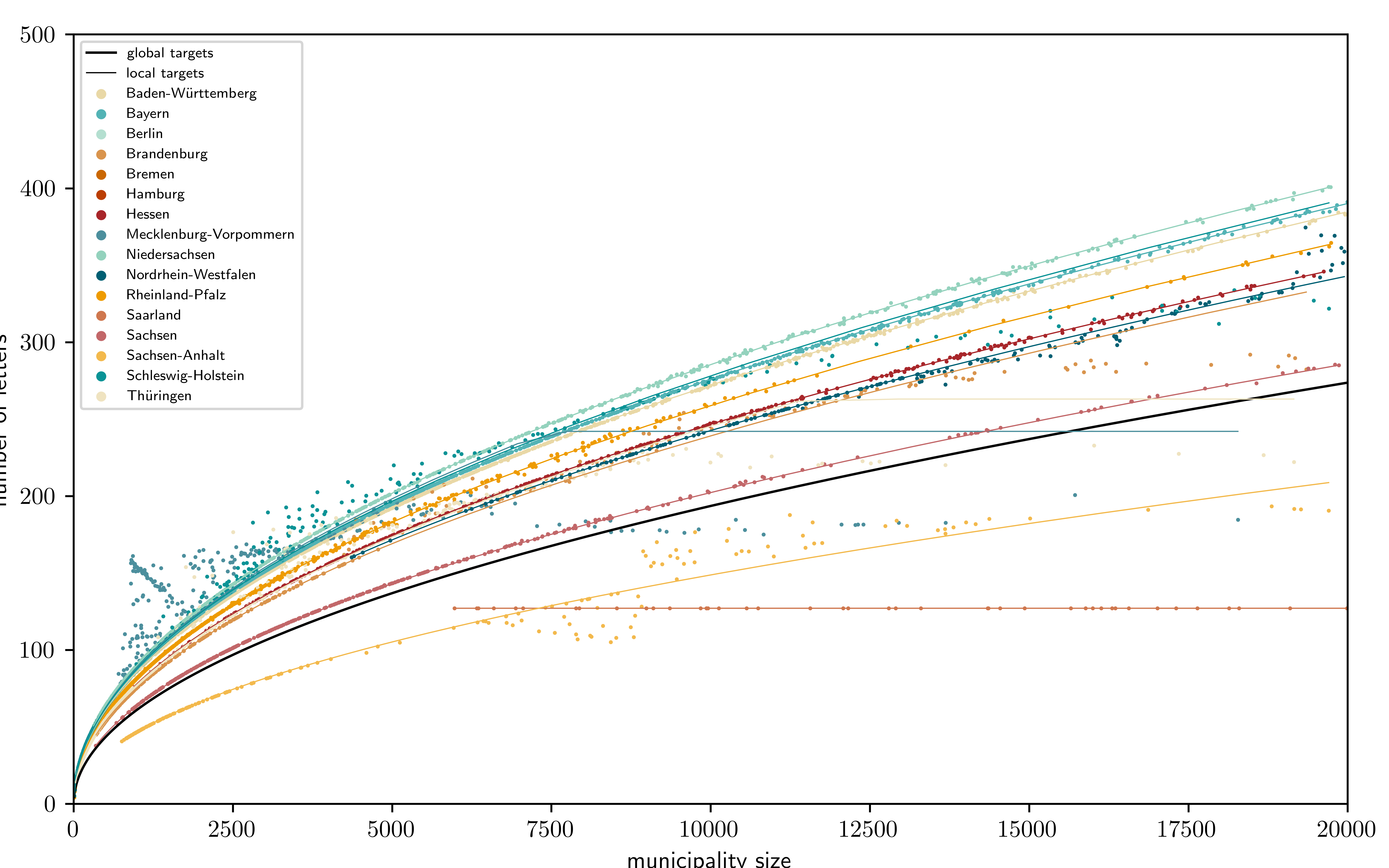}
    \caption{\colgen for groups with small cities}
    \label{fig:global_local_targets_colgen_small}
\end{figure} 

\begin{figure}
    \centering
    \includegraphics[width=\linewidth]{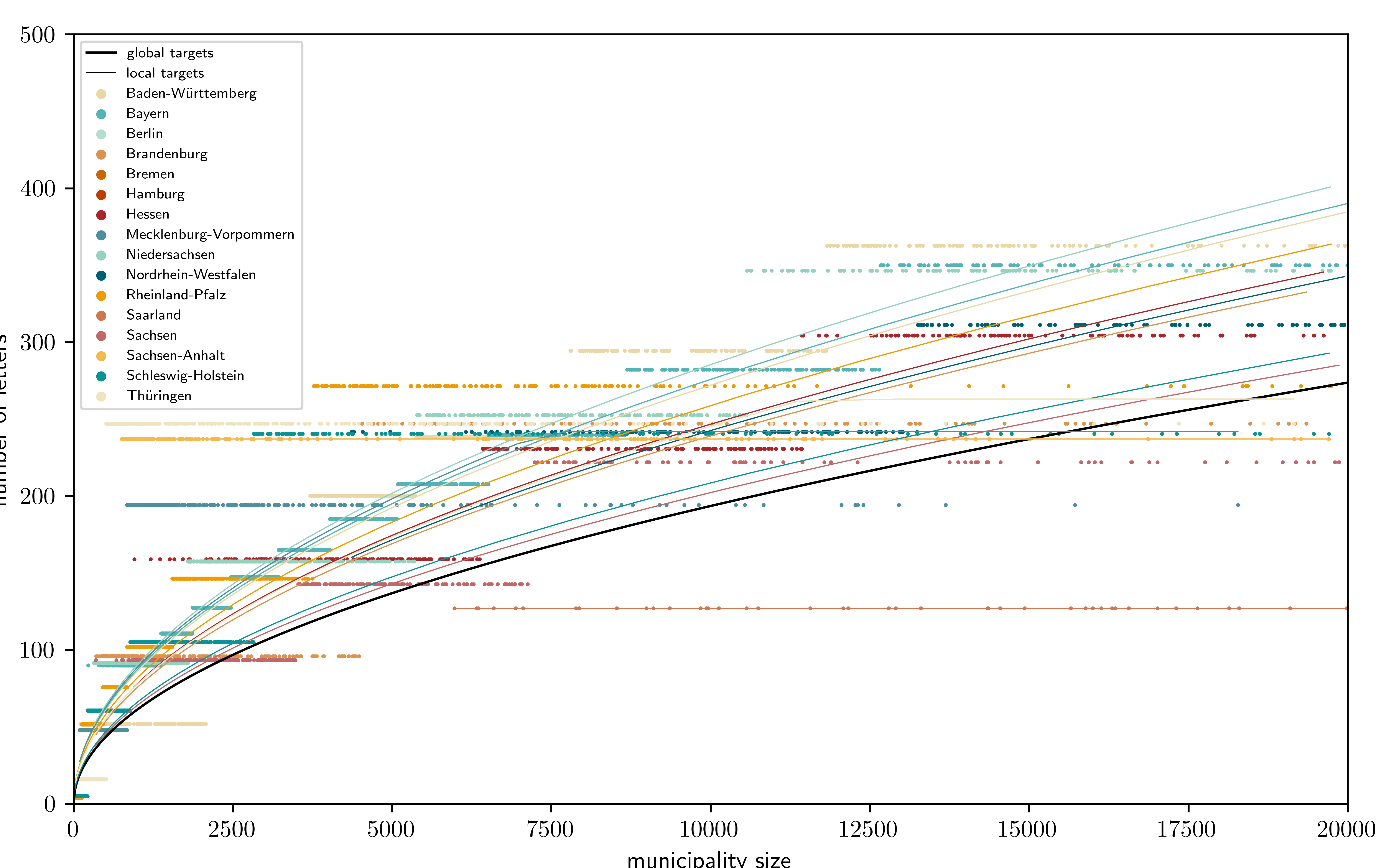}
    \caption{\buckets for groups with small cities}
    \label{fig:global_local_targets_bucket_small}
\end{figure} 

\begin{figure}
    \centering
    \includegraphics[width=\linewidth]{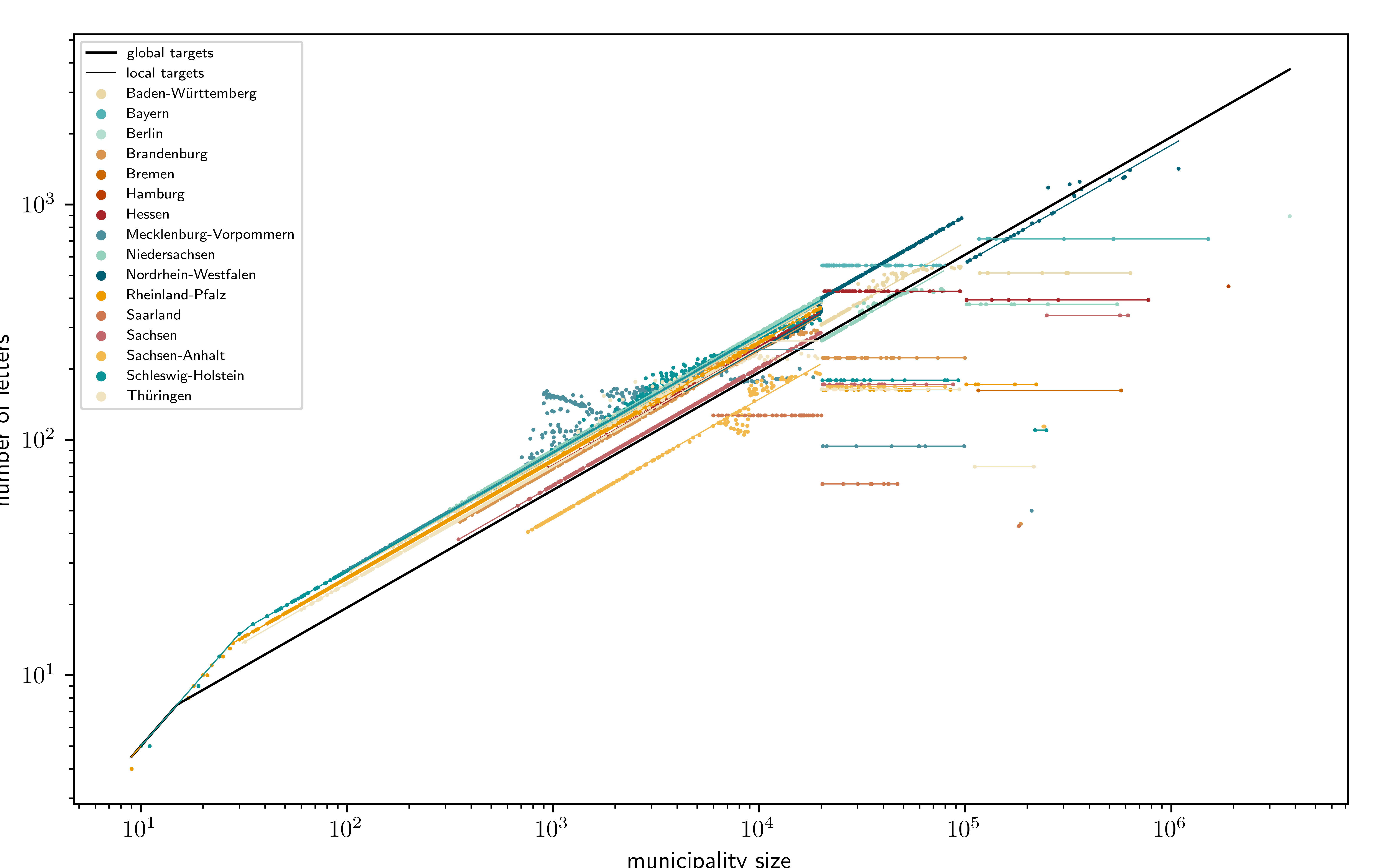}
    \caption{\colgen for all groups}
    \label{fig:global_local_targets_colgen_all}
\end{figure} 
    
\begin{figure}
    \centering
    \includegraphics[width=\linewidth]{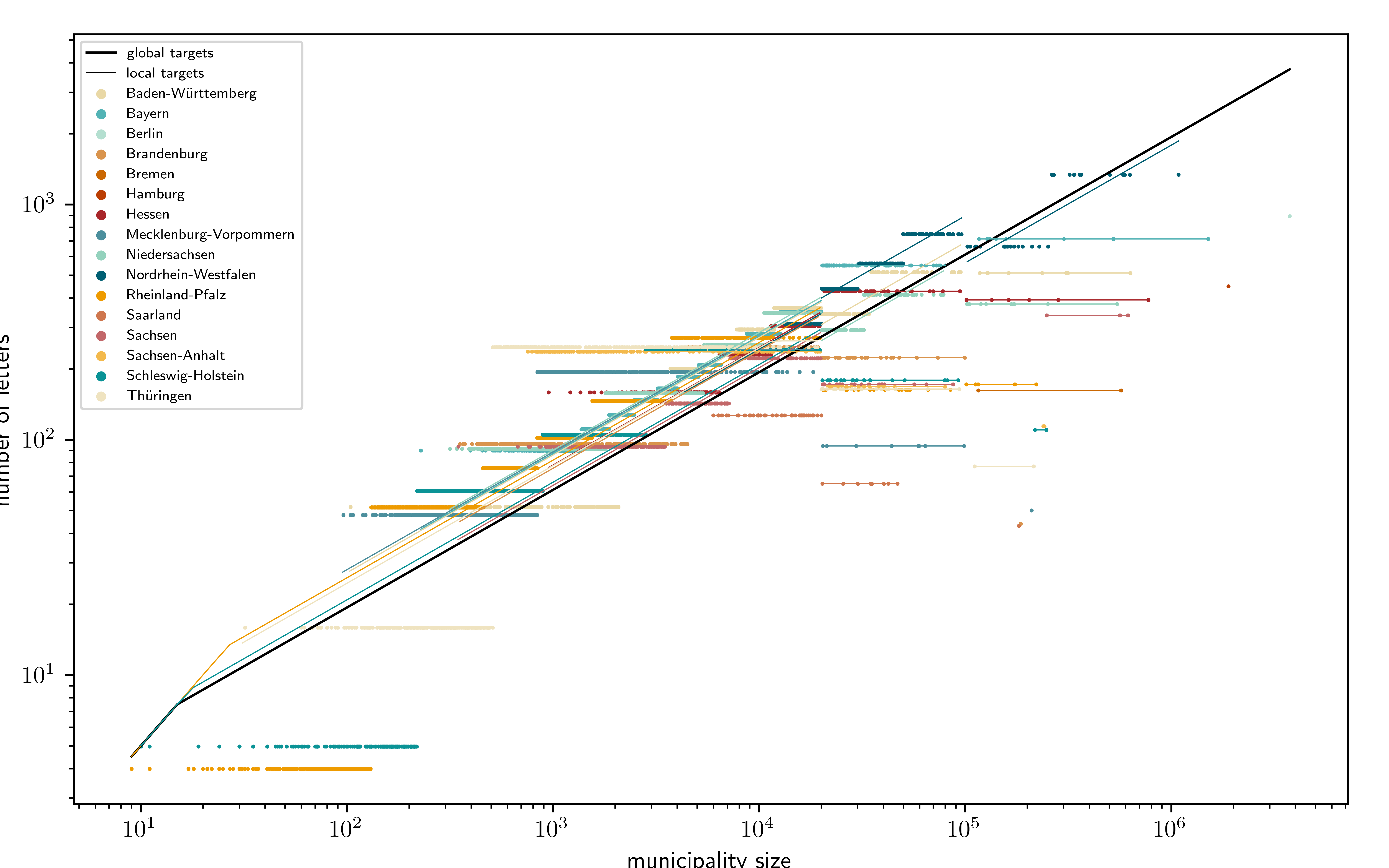}
    \caption{\buckets for all groups}
    \label{fig:global_local_targets_bucket_all}
\end{figure} 

\end{landscape}

\begin{figure}
    \centering
    \begin{subfigure}{0.32\textwidth}
        \includegraphics[draft=\draft, width=\linewidth]{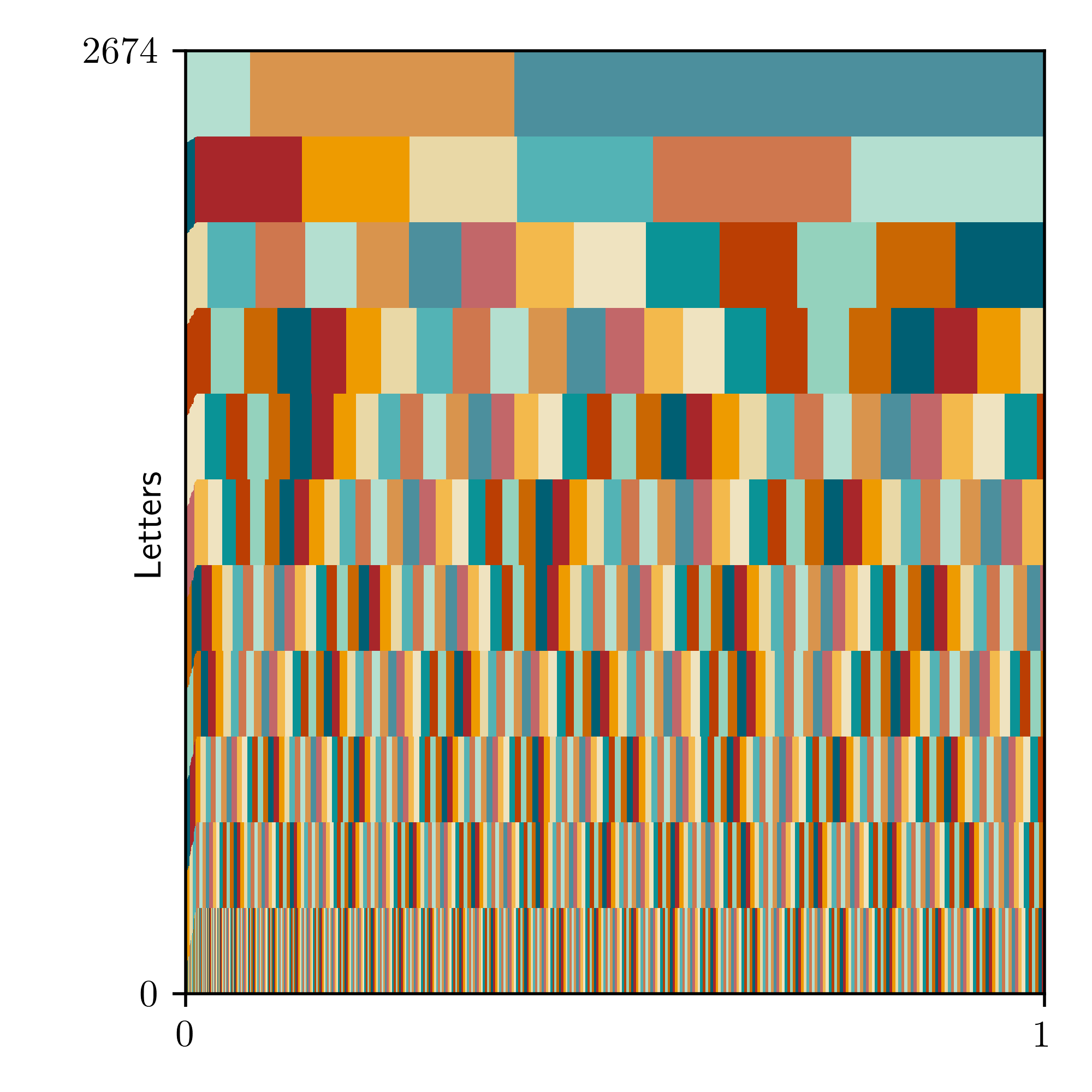}
        \includegraphics[draft=\draft, width=\linewidth]{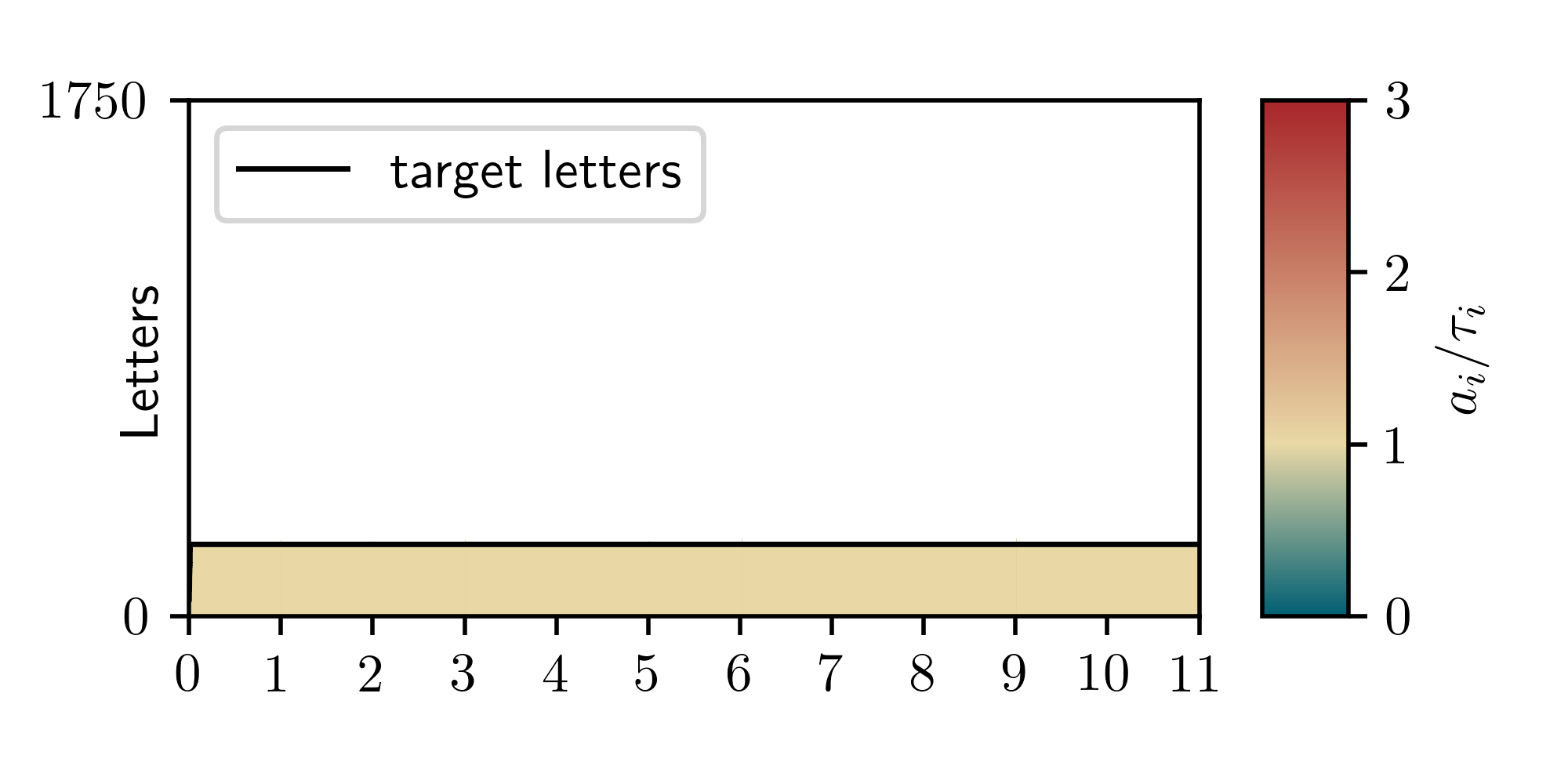}
        \caption{\greq ($t = 11$)}
        \label{fig:results_Baden-Württemberg_All_greedy_equal}
    \end{subfigure}
    \begin{subfigure}{0.32\textwidth}
        \includegraphics[draft=\draft, width=\linewidth]{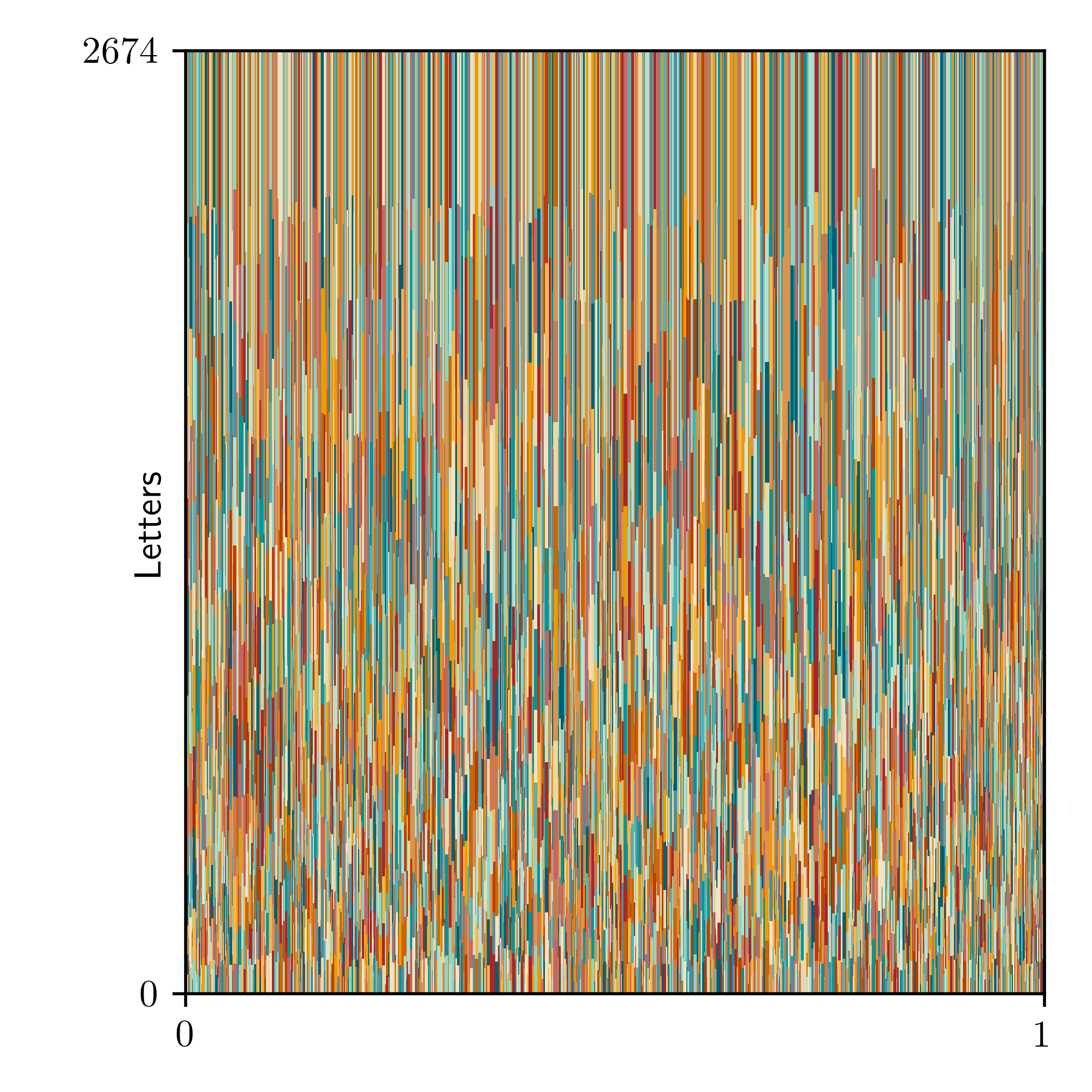}
        \includegraphics[draft=\draft, width=\linewidth]{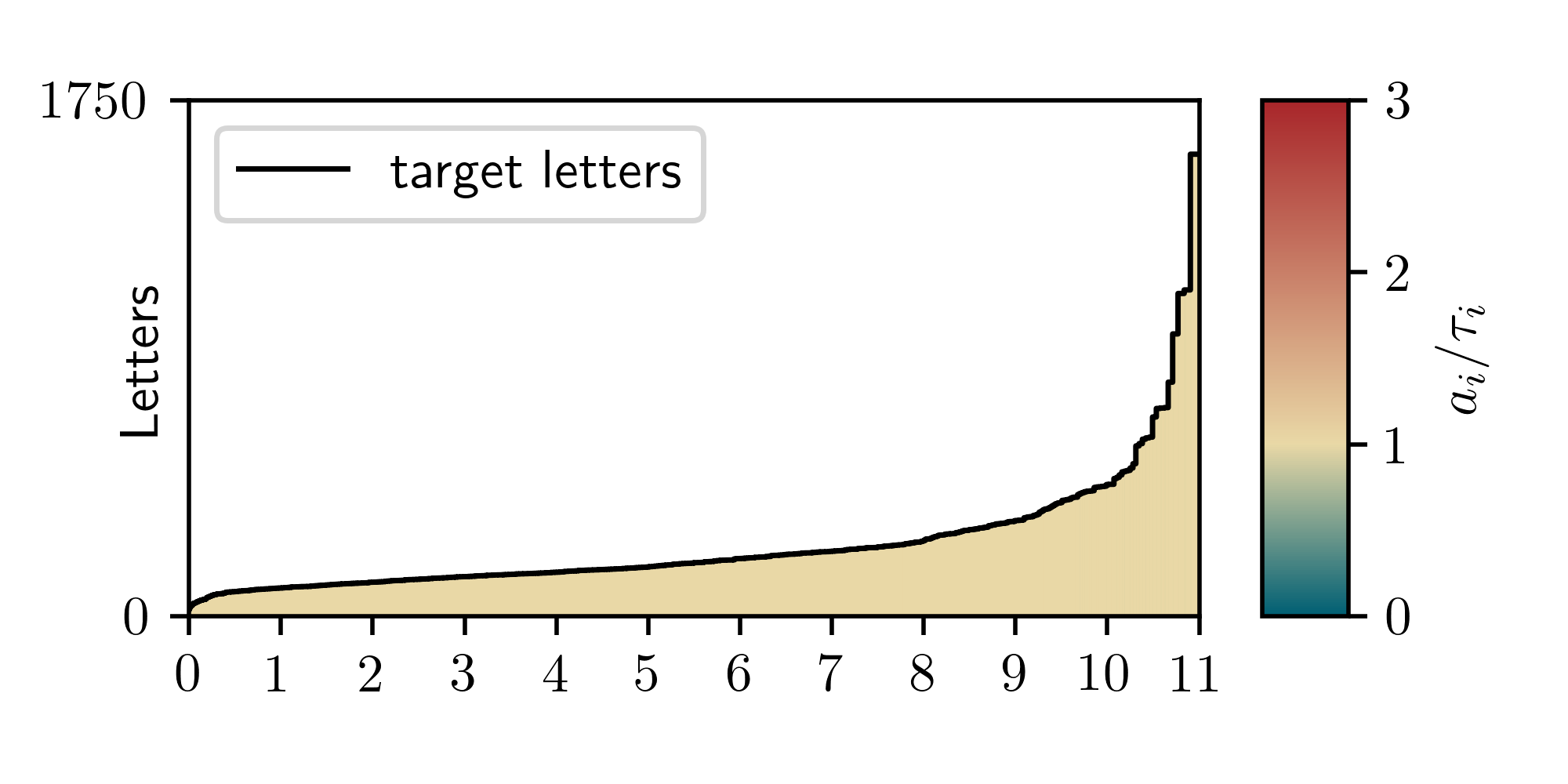}
        \caption{\colgen ($t = 11$)}
        \label{fig:results_Baden-Württemberg_All_column_generation}
    \end{subfigure}
    \begin{subfigure}{0.32\textwidth}
        \includegraphics[draft=\draft, width=\linewidth]{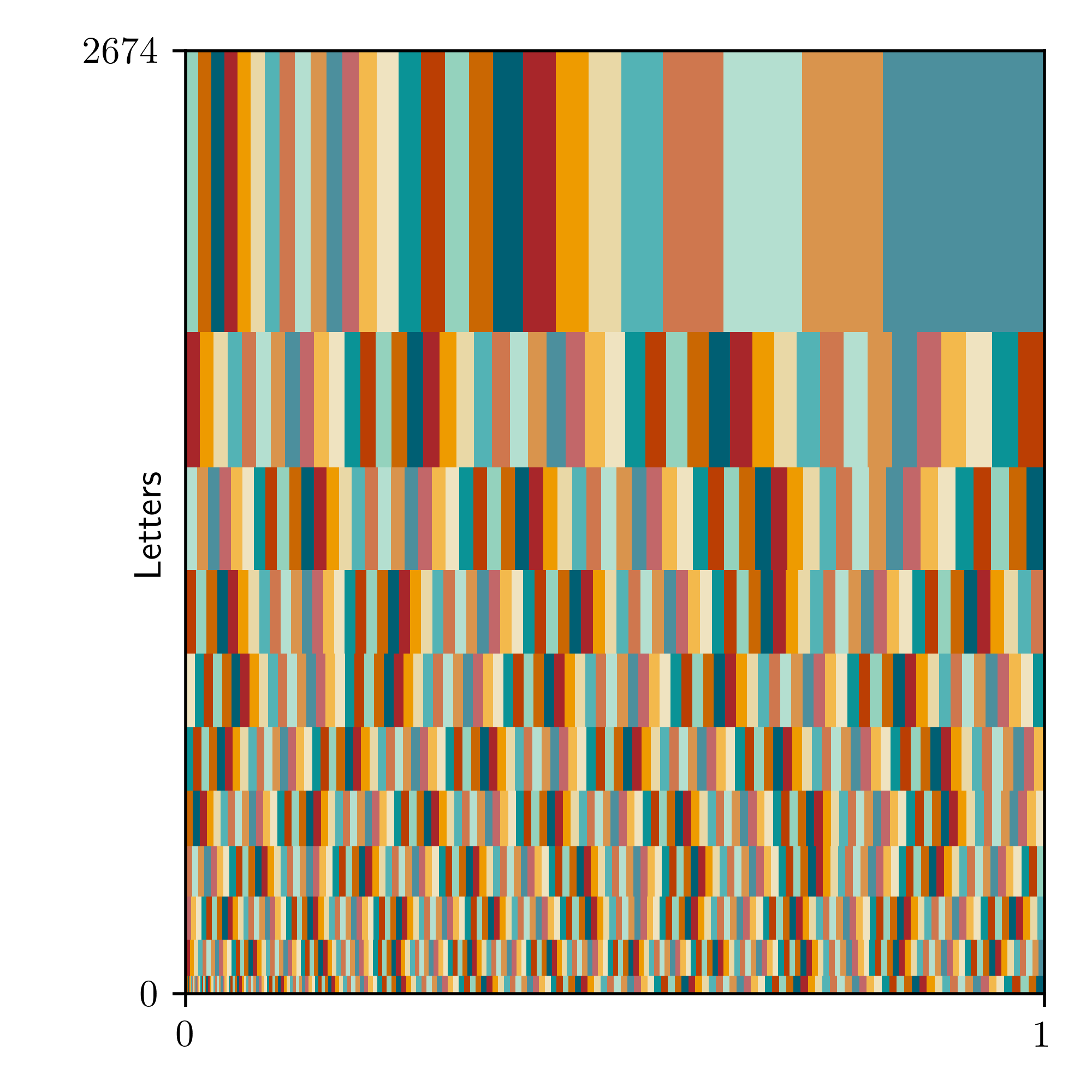}
        \includegraphics[draft=\draft, width=\linewidth]{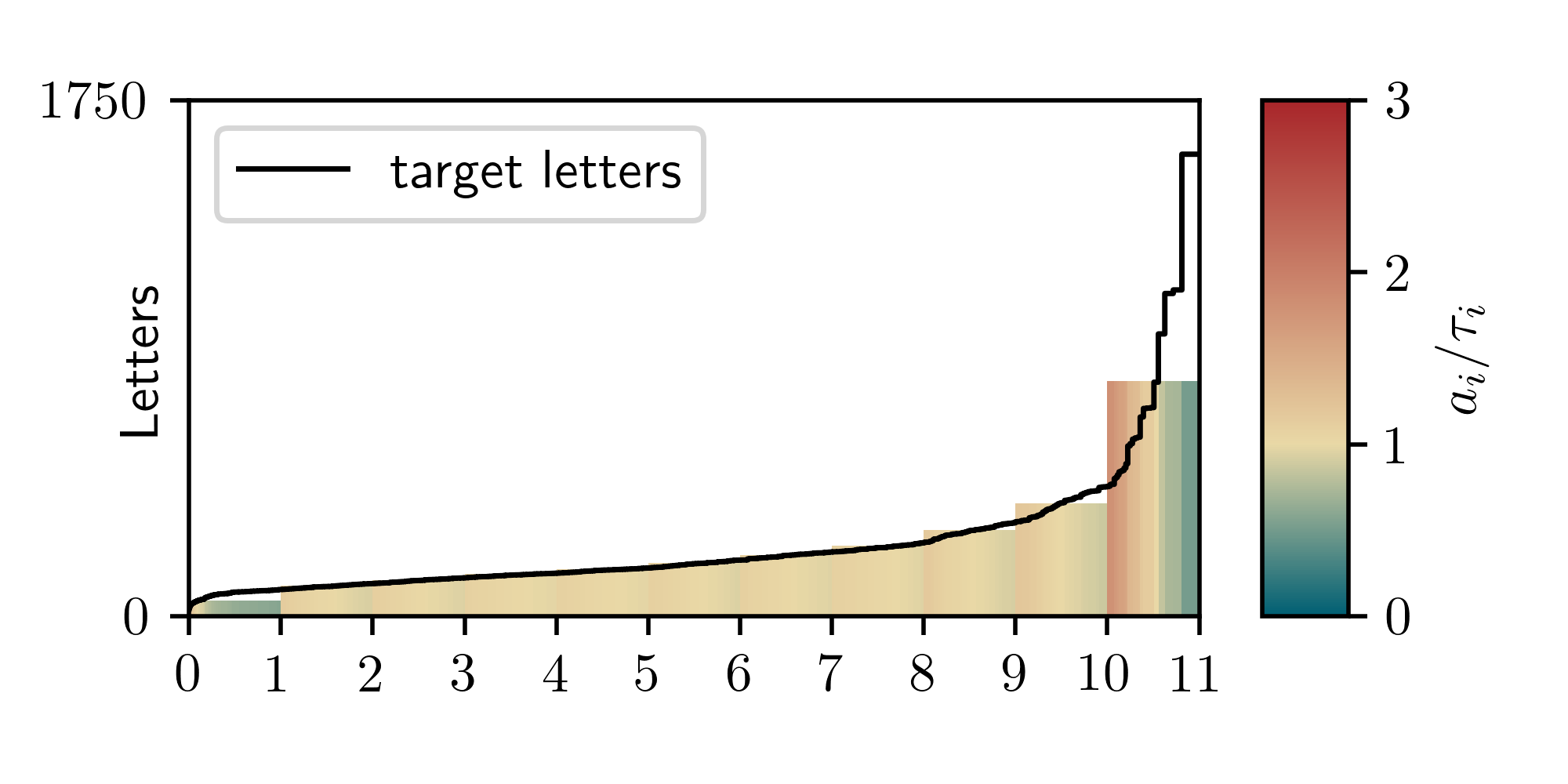}
        \caption{\buckets ($t = 11$)}
        \label{fig:results_Baden-Württemberg_All_greedy_bucket_fill}
    \end{subfigure}
    \caption{All cities of Baden-Württemberg ($\ell = 2674$)}
    \label{fig:results_Baden-Württemberg_All}
\end{figure}

\end{document}